\documentclass[11pt, a4paper]{book}
\usepackage[top=25mm,bottom=25mm,left=40mm,right=25mm,includehead,includefoot]{geometry}

\usepackage{amsthm}
\usepackage{amsmath}
\usepackage{amsfonts}
\usepackage{amssymb}
\usepackage{mathtools}
\usepackage{physics}
\usepackage{ytableau}

\usepackage{newtxtext}
\usepackage{newtxmath}
\usepackage{eucal}
\usepackage{tikz}
\usetikzlibrary{calc,decorations.markings, shapes.misc, decorations.pathmorphing}
\tikzset{cross/.style={cross out, draw, 
		minimum size=2*(#1-\pgflinewidth), 
		inner sep=0pt, outer sep=0pt}} 
\usepackage{subcaption}

\usepackage{xcolor}
\usepackage{url}
\definecolor{airforceblue}{rgb}{0.36, 0.54, 0.66}
\definecolor{firebrick}{rgb}{0.7, 0.13, 0.13}
\definecolor{tealblue}{rgb}{0.21, 0.46, 0.53}
\definecolor{darkgreen}{rgb}{0.0, 0.2, 0.13}
\definecolor{dartmouthgreen}{rgb}{0.05, 0.5, 0.06}
\usepackage{hyperref}   
\hypersetup{
	colorlinks=true,
	linktoc=page,
	citecolor=dartmouthgreen,
	linkcolor=firebrick,
	urlcolor=firebrick
}

\usepackage[numbers,sort&compress]{natbib}
\newtheorem{corollary}{Corollary}
\newtheorem{theorem}{Theorem} %% Theorem's are used when no proof is provided
\newtheorem{proposition}{Proposition} %% Use prop if we provide proof

\theoremstyle{definition}
\newtheorem{definition}{Definition}
\newtheorem{example}{Example}

\newtheorem{remark}{Remark}

\newcommand{\rotaterelation}[1]{\rotatebox[origin=c]{90}{$\mathstrut#1$}}

\newcommand{\Hilbertspace}{\mathcal{H}_{\text{inv}}}

\newcommand{\Adj}[1]{\operatorname{Ad}(#1)}

\newcommand{\Cclass}{C}
\DeclareMathOperator{\Perms}{Perms}
\DeclareMathOperator{\diag}{diag}
\newcommand{\Cl}[1]{\mathrm{Cl}(#1)}
\newcommand{\Rep}[1]{\mathrm{Rep}(#1)}
\newcommand{\N}{N}
\newcommand{\SN}{S_\N}
\newcommand{\sn}{\sigma}
\newcommand{\Pmat}[1]{(P_{#1})}
\renewcommand{\P}{P}
\newcommand{\VN}{V_\N}
\newcommand{\Paction}[1]{P_{#1}}
\newcommand{\fixp}[1]{F(#1)}
\newcommand{\YT}[1]{Y_{#1}}
\newcommand{\SYT}[1]{\mathrm{SYT}_{#1}}
\newcommand{\R}{\mathcal{R}}

\newcommand{\VNinner}[2]{(#1, #2)}

\newcommand{\DG}[1]{\mathcal{G}_{#1}} %% set of graphs with k edges

\newcommand{\vactab}{\mathfrak{v}}
\newcommand{\Zdual}{\widetilde{\mathcal{Z}}}
\newcommand{\setpart}[1]{\Pi_{#1}}
\newcommand{\X}{X}
\newcommand{\dgram}{D}
\newcommand{\dimPk}[1]{m^{#1}_{k,\N}}
\newcommand{\dimSN}[1]{\dim V_{#1}}
\newcommand{\join}{\vee}
\newcommand{\SPk}[1]{\mathfrak{#1}}

\newcommand{\PAdiagram}[3]{\;\begin{tikzpicture}[baseline={([yshift=-.5ex]current bounding box.center)}]
		\def \n {#1};
		\def \edges {#2};
		\def \arcs {#3};
		\def \sep {0.5};
		\foreach \v in {1,...,\n}
		{
			\pgfmathparse{(\v-1)*\sep};
			\coordinate (v\v) at (\pgfmathresult,0.25);
			\node[circle,fill,inner sep=1pt] at (v\v) {};
		}
		\foreach \v in {1,...,\n}
		{
			\pgfmathparse{(\v-1)*\sep};
			\coordinate (v-\v) at (\pgfmathresult,-0.25);
			\node[circle,fill,inner sep=1pt] at (v-\v) {};
		}
		\foreach \endpointOne/\endpointTwo in \edges
		{
			\draw[] (v\endpointOne) -- (v\endpointTwo);
		}
		\foreach \endpointOne/\endpointTwo in \arcs
		{
			\draw[] (v\endpointOne) to[bend left] (v\endpointTwo);
		}
	\end{tikzpicture}\;}
\newcommand{\PAdiagramLabeled}[3]{\;\begin{tikzpicture}[baseline={([yshift=-.5ex]current bounding box.center)}]
		\def \n {#1};
		\def \edges {#2};
		\def \arcs {#3};
		\def \sep {0.5};
		\foreach \v in {1,...,\n}
		{
			\pgfmathparse{(\v-1)*\sep};
			\coordinate (v\v) at (\pgfmathresult,0.25);
			\node[circle,fill,inner sep=1pt, label=above:{\tiny $\v'$}] at (v\v) {};
		}
		\foreach \v in {1,...,\n}
		{
			\pgfmathparse{(\v-1)*\sep};
			\coordinate (v-\v) at (\pgfmathresult,-0.25);
			\node[circle,fill,inner sep=1pt, label=below:{\tiny $\v$}] at (v-\v) {};
		}
		\foreach \endpointOne/\endpointTwo in \edges
		{
			\draw[] (v\endpointOne) -- (v\endpointTwo);
		}
		\foreach \endpointOne/\endpointTwo in \arcs
		{
			\draw[] (v\endpointOne) to[bend left] (v\endpointTwo);
		}
	\end{tikzpicture}\;}

\newcommand{\PAdiagramOrbit}[3]{\;\begin{tikzpicture}[baseline={([yshift=-.5ex]current bounding box.center)}]
		\def \n {#1};
		\def \edges {#2};
		\def \arcs {#3};
		\def \sep {0.5};
		\foreach \v in {1,...,\n}
		{
			\pgfmathparse{(\v-1)*\sep};
			\coordinate (v\v) at (\pgfmathresult,0.25);
			\node[circle, draw,inner sep=1pt] at (v\v) {};
		}
		\foreach \v in {1,...,\n}
		{
			\pgfmathparse{(\v-1)*\sep};
			\coordinate (v-\v) at (\pgfmathresult,-0.25);
			\node[circle, draw,inner sep=1pt] at (v-\v) {};
		}
		\foreach \endpointOne/\endpointTwo in \edges
		{
			\draw[] (v\endpointOne) -- (v\endpointTwo);
		}
		\foreach \endpointOne/\endpointTwo in \arcs
		{
			\draw[] (v\endpointOne) to[bend left] (v\endpointTwo);
		}
	\end{tikzpicture}\;}

\newcommand{\onerowdiagramunLabeled}[2]{\;\begin{tikzpicture}[baseline={([yshift=-.5ex]current bounding box.center)}]
		\def \n {#1};
		\def \arcs {#2};
		\def \sep {0.5};
		\foreach \v in {1,...,\n}
		{
			\pgfmathparse{(\v-1)*\sep};
			\coordinate (v\v) at (\pgfmathresult,0.25);
			\node[circle,fill,inner sep=1pt]  at (v\v) {};
		}
		\foreach \endpointOne/\endpointTwo in \arcs
		{
			\draw[] (v\endpointOne) to[bend right] (v\endpointTwo);
		}
	\end{tikzpicture}\;}
\newcommand{\onerowdiagram}[3]{\;\begin{tikzpicture}[baseline={([yshift=-.5ex]current bounding box.center)}]
		\def \n {#1};
		\def \edges {#2};
		\def \arcs {#3};
		\def \sep {0.5};
		\foreach \v in {1,...,\n}
		{
			\pgfmathparse{(\v-1)*\sep};
			\coordinate (v\v) at (\pgfmathresult,0.25);
			\node[circle,fill,inner sep=1pt, label=above:{\tiny $\v$}]  at (v\v) {};
		}
		\foreach \endpointOne/\endpointTwo in \edges
		{
			\draw[] (v\endpointOne) -- (v\endpointTwo);
		}
		\foreach \endpointOne/\endpointTwo in \arcs
		{
			\draw[] (v\endpointOne) to[bend right] (v\endpointTwo);
		}
	\end{tikzpicture}\;}

\DeclareMathOperator{\Span}{Span}
\DeclareMathOperator{\End}{End}
\DeclareMathOperator{\im}{im}
\DeclareMathOperator{\M}{Mat}
\DeclareMathOperator{\Sym}{Sym}
\DeclareMathOperator{\Hom}{Hom}
\newcommand{\idn}{\mathbb{I}}

\newcommand{\chr}{\chi}
\newcommand{\ptition}{\lambda}
\newcommand{\D}{D}
\newcommand{\Ind}[3]{\mathrm{Ind}^{#2}_{#1}(#3)}
\renewcommand{\Res}[3]{\mathrm{Res}^{#1}_{#2}(#3)}

\newcommand{\e}{\mathrm{e}}
\newcommand{\Z}{\mathcal{Z}}
\newcommand{\obs}{\mathcal{O}}

\usepackage[breakable]{tcolorbox}
\usepackage{parskip} % Stop auto-indenting (to mimic markdown behaviour)

\usepackage{iftex}
\usepackage{graphicx}
\usepackage{caption}
\DeclareCaptionFormat{nocaption}{}

\usepackage[Export]{adjustbox} % Used to constrain images to a maximum size
\adjustboxset{max size={0.9\linewidth}{0.9\paperheight}}
\usepackage{float}
\usepackage{xcolor} % Allow colors to be defined
\usepackage{enumerate} % Needed for markdown enumerations to work
\usepackage{geometry} % Used to adjust the document margins
\usepackage{amsmath} % Equations
\usepackage{amssymb} % Equations
\usepackage{textcomp} % defines textquotesingle
\AtBeginDocument{%
	\def\PYZsq{\textquotesingle}% Upright quotes in Pygmentized code
}
\usepackage{upquote} % Upright quotes for verbatim code
\usepackage{eurosym} % defines \euro
\usepackage[mathletters]{ucs} % Extended unicode (utf-8) support
\usepackage{fancyvrb} % verbatim replacement that allows latex
\usepackage{grffile} % extends the file name processing of package graphics 
\makeatletter % fix for grffile with XeLaTeX
\def\Gread@@xetex#1{%
	\IfFileExists{"\Gin@base".bb}%
	{\Gread@eps{\Gin@base.bb}}%
	{\Gread@@xetex@aux#1}%
}
\makeatother

\usepackage{hyperref}
\usepackage{longtable} % longtable support required by pandoc >1.10
\usepackage{booktabs}  % table support for pandoc > 1.12.2
\usepackage[inline]{enumitem} % IRkernel/repr support (it uses the enumerate* environment)
\usepackage[normalem]{ulem} % ulem is needed to support strikethroughs (\sout)
\usepackage{mathrsfs}

\definecolor{urlcolor}{rgb}{0,.145,.698}
\definecolor{linkcolor}{rgb}{.71,0.21,0.01}
\definecolor{citecolor}{rgb}{.12,.54,.11}

\definecolor{ansi-black}{HTML}{3E424D}
\definecolor{ansi-black-intense}{HTML}{282C36}
\definecolor{ansi-red}{HTML}{E75C58}
\definecolor{ansi-red-intense}{HTML}{B22B31}
\definecolor{ansi-green}{HTML}{00A250}
\definecolor{ansi-green-intense}{HTML}{007427}
\definecolor{ansi-yellow}{HTML}{DDB62B}
\definecolor{ansi-yellow-intense}{HTML}{B27D12}
\definecolor{ansi-blue}{HTML}{208FFB}
\definecolor{ansi-blue-intense}{HTML}{0065CA}
\definecolor{ansi-magenta}{HTML}{D160C4}
\definecolor{ansi-magenta-intense}{HTML}{A03196}
\definecolor{ansi-cyan}{HTML}{60C6C8}
\definecolor{ansi-cyan-intense}{HTML}{258F8F}
\definecolor{ansi-white}{HTML}{C5C1B4}
\definecolor{ansi-white-intense}{HTML}{A1A6B2}
\definecolor{ansi-default-inverse-fg}{HTML}{FFFFFF}
\definecolor{ansi-default-inverse-bg}{HTML}{000000}

\DefineVerbatimEnvironment{Highlighting}{Verbatim}{commandchars=\\\{\}}
\let\Oldtex\TeX
\let\Oldlatex\LaTeX
\renewcommand{\TeX}{\textrm{\Oldtex}}
\renewcommand{\LaTeX}{\textrm{\Oldlatex}}
\title{PIGM\_ExpVals}

\makeatletter
\def\PY@reset{\let\PY@it=\relax \let\PY@bf=\relax%
	\let\PY@ul=\relax \let\PY@tc=\relax%
	\let\PY@bc=\relax \let\PY@ff=\relax}
\def\PY@tok#1{\csname PY@tok@#1\endcsname}
\def\PY@toks#1+{\ifx\relax#1\empty\else%
	\PY@tok{#1}\expandafter\PY@toks\fi}
\def\PY@do#1{\PY@bc{\PY@tc{\PY@ul{%
				\PY@it{\PY@bf{\PY@ff{#1}}}}}}}
\def\PY#1#2{\PY@reset\PY@toks#1+\relax+\PY@do{#2}}

\expandafter\def\csname PY@tok@w\endcsname{\def\PY@tc##1{\textcolor[rgb]{0.73,0.73,0.73}{##1}}}
\expandafter\def\csname PY@tok@c\endcsname{\let\PY@it=\textit\def\PY@tc##1{\textcolor[rgb]{0.25,0.50,0.50}{##1}}}
\expandafter\def\csname PY@tok@cp\endcsname{\def\PY@tc##1{\textcolor[rgb]{0.74,0.48,0.00}{##1}}}
\expandafter\def\csname PY@tok@k\endcsname{\let\PY@bf=\textbf\def\PY@tc##1{\textcolor[rgb]{0.00,0.50,0.00}{##1}}}
\expandafter\def\csname PY@tok@kp\endcsname{\def\PY@tc##1{\textcolor[rgb]{0.00,0.50,0.00}{##1}}}
\expandafter\def\csname PY@tok@kt\endcsname{\def\PY@tc##1{\textcolor[rgb]{0.69,0.00,0.25}{##1}}}
\expandafter\def\csname PY@tok@o\endcsname{\def\PY@tc##1{\textcolor[rgb]{0.40,0.40,0.40}{##1}}}
\expandafter\def\csname PY@tok@ow\endcsname{\let\PY@bf=\textbf\def\PY@tc##1{\textcolor[rgb]{0.67,0.13,1.00}{##1}}}
\expandafter\def\csname PY@tok@nb\endcsname{\def\PY@tc##1{\textcolor[rgb]{0.00,0.50,0.00}{##1}}}
\expandafter\def\csname PY@tok@nf\endcsname{\def\PY@tc##1{\textcolor[rgb]{0.00,0.00,1.00}{##1}}}
\expandafter\def\csname PY@tok@nc\endcsname{\let\PY@bf=\textbf\def\PY@tc##1{\textcolor[rgb]{0.00,0.00,1.00}{##1}}}
\expandafter\def\csname PY@tok@nn\endcsname{\let\PY@bf=\textbf\def\PY@tc##1{\textcolor[rgb]{0.00,0.00,1.00}{##1}}}
\expandafter\def\csname PY@tok@ne\endcsname{\let\PY@bf=\textbf\def\PY@tc##1{\textcolor[rgb]{0.82,0.25,0.23}{##1}}}
\expandafter\def\csname PY@tok@nv\endcsname{\def\PY@tc##1{\textcolor[rgb]{0.10,0.09,0.49}{##1}}}
\expandafter\def\csname PY@tok@no\endcsname{\def\PY@tc##1{\textcolor[rgb]{0.53,0.00,0.00}{##1}}}
\expandafter\def\csname PY@tok@nl\endcsname{\def\PY@tc##1{\textcolor[rgb]{0.63,0.63,0.00}{##1}}}
\expandafter\def\csname PY@tok@ni\endcsname{\let\PY@bf=\textbf\def\PY@tc##1{\textcolor[rgb]{0.60,0.60,0.60}{##1}}}
\expandafter\def\csname PY@tok@na\endcsname{\def\PY@tc##1{\textcolor[rgb]{0.49,0.56,0.16}{##1}}}
\expandafter\def\csname PY@tok@nt\endcsname{\let\PY@bf=\textbf\def\PY@tc##1{\textcolor[rgb]{0.00,0.50,0.00}{##1}}}
\expandafter\def\csname PY@tok@nd\endcsname{\def\PY@tc##1{\textcolor[rgb]{0.67,0.13,1.00}{##1}}}
\expandafter\def\csname PY@tok@s\endcsname{\def\PY@tc##1{\textcolor[rgb]{0.73,0.13,0.13}{##1}}}
\expandafter\def\csname PY@tok@sd\endcsname{\let\PY@it=\textit\def\PY@tc##1{\textcolor[rgb]{0.73,0.13,0.13}{##1}}}
\expandafter\def\csname PY@tok@si\endcsname{\let\PY@bf=\textbf\def\PY@tc##1{\textcolor[rgb]{0.73,0.40,0.53}{##1}}}
\expandafter\def\csname PY@tok@se\endcsname{\let\PY@bf=\textbf\def\PY@tc##1{\textcolor[rgb]{0.73,0.40,0.13}{##1}}}
\expandafter\def\csname PY@tok@sr\endcsname{\def\PY@tc##1{\textcolor[rgb]{0.73,0.40,0.53}{##1}}}
\expandafter\def\csname PY@tok@ss\endcsname{\def\PY@tc##1{\textcolor[rgb]{0.10,0.09,0.49}{##1}}}
\expandafter\def\csname PY@tok@sx\endcsname{\def\PY@tc##1{\textcolor[rgb]{0.00,0.50,0.00}{##1}}}
\expandafter\def\csname PY@tok@m\endcsname{\def\PY@tc##1{\textcolor[rgb]{0.40,0.40,0.40}{##1}}}
\expandafter\def\csname PY@tok@gh\endcsname{\let\PY@bf=\textbf\def\PY@tc##1{\textcolor[rgb]{0.00,0.00,0.50}{##1}}}
\expandafter\def\csname PY@tok@gu\endcsname{\let\PY@bf=\textbf\def\PY@tc##1{\textcolor[rgb]{0.50,0.00,0.50}{##1}}}
\expandafter\def\csname PY@tok@gd\endcsname{\def\PY@tc##1{\textcolor[rgb]{0.63,0.00,0.00}{##1}}}
\expandafter\def\csname PY@tok@gi\endcsname{\def\PY@tc##1{\textcolor[rgb]{0.00,0.63,0.00}{##1}}}
\expandafter\def\csname PY@tok@gr\endcsname{\def\PY@tc##1{\textcolor[rgb]{1.00,0.00,0.00}{##1}}}
\expandafter\def\csname PY@tok@ge\endcsname{\let\PY@it=\textit}
\expandafter\def\csname PY@tok@gs\endcsname{\let\PY@bf=\textbf}
\expandafter\def\csname PY@tok@gp\endcsname{\let\PY@bf=\textbf\def\PY@tc##1{\textcolor[rgb]{0.00,0.00,0.50}{##1}}}
\expandafter\def\csname PY@tok@go\endcsname{\def\PY@tc##1{\textcolor[rgb]{0.53,0.53,0.53}{##1}}}
\expandafter\def\csname PY@tok@gt\endcsname{\def\PY@tc##1{\textcolor[rgb]{0.00,0.27,0.87}{##1}}}
\expandafter\def\csname PY@tok@err\endcsname{\def\PY@bc##1{\setlength{\fboxsep}{0pt}\fcolorbox[rgb]{1.00,0.00,0.00}{1,1,1}{\strut ##1}}}
\expandafter\def\csname PY@tok@kc\endcsname{\let\PY@bf=\textbf\def\PY@tc##1{\textcolor[rgb]{0.00,0.50,0.00}{##1}}}
\expandafter\def\csname PY@tok@kd\endcsname{\let\PY@bf=\textbf\def\PY@tc##1{\textcolor[rgb]{0.00,0.50,0.00}{##1}}}
\expandafter\def\csname PY@tok@kn\endcsname{\let\PY@bf=\textbf\def\PY@tc##1{\textcolor[rgb]{0.00,0.50,0.00}{##1}}}
\expandafter\def\csname PY@tok@kr\endcsname{\let\PY@bf=\textbf\def\PY@tc##1{\textcolor[rgb]{0.00,0.50,0.00}{##1}}}
\expandafter\def\csname PY@tok@bp\endcsname{\def\PY@tc##1{\textcolor[rgb]{0.00,0.50,0.00}{##1}}}
\expandafter\def\csname PY@tok@fm\endcsname{\def\PY@tc##1{\textcolor[rgb]{0.00,0.00,1.00}{##1}}}
\expandafter\def\csname PY@tok@vc\endcsname{\def\PY@tc##1{\textcolor[rgb]{0.10,0.09,0.49}{##1}}}
\expandafter\def\csname PY@tok@vg\endcsname{\def\PY@tc##1{\textcolor[rgb]{0.10,0.09,0.49}{##1}}}
\expandafter\def\csname PY@tok@vi\endcsname{\def\PY@tc##1{\textcolor[rgb]{0.10,0.09,0.49}{##1}}}
\expandafter\def\csname PY@tok@vm\endcsname{\def\PY@tc##1{\textcolor[rgb]{0.10,0.09,0.49}{##1}}}
\expandafter\def\csname PY@tok@sa\endcsname{\def\PY@tc##1{\textcolor[rgb]{0.73,0.13,0.13}{##1}}}
\expandafter\def\csname PY@tok@sb\endcsname{\def\PY@tc##1{\textcolor[rgb]{0.73,0.13,0.13}{##1}}}
\expandafter\def\csname PY@tok@sc\endcsname{\def\PY@tc##1{\textcolor[rgb]{0.73,0.13,0.13}{##1}}}
\expandafter\def\csname PY@tok@dl\endcsname{\def\PY@tc##1{\textcolor[rgb]{0.73,0.13,0.13}{##1}}}
\expandafter\def\csname PY@tok@s2\endcsname{\def\PY@tc##1{\textcolor[rgb]{0.73,0.13,0.13}{##1}}}
\expandafter\def\csname PY@tok@sh\endcsname{\def\PY@tc##1{\textcolor[rgb]{0.73,0.13,0.13}{##1}}}
\expandafter\def\csname PY@tok@s1\endcsname{\def\PY@tc##1{\textcolor[rgb]{0.73,0.13,0.13}{##1}}}
\expandafter\def\csname PY@tok@mb\endcsname{\def\PY@tc##1{\textcolor[rgb]{0.40,0.40,0.40}{##1}}}
\expandafter\def\csname PY@tok@mf\endcsname{\def\PY@tc##1{\textcolor[rgb]{0.40,0.40,0.40}{##1}}}
\expandafter\def\csname PY@tok@mh\endcsname{\def\PY@tc##1{\textcolor[rgb]{0.40,0.40,0.40}{##1}}}
\expandafter\def\csname PY@tok@mi\endcsname{\def\PY@tc##1{\textcolor[rgb]{0.40,0.40,0.40}{##1}}}
\expandafter\def\csname PY@tok@il\endcsname{\def\PY@tc##1{\textcolor[rgb]{0.40,0.40,0.40}{##1}}}
\expandafter\def\csname PY@tok@mo\endcsname{\def\PY@tc##1{\textcolor[rgb]{0.40,0.40,0.40}{##1}}}
\expandafter\def\csname PY@tok@ch\endcsname{\let\PY@it=\textit\def\PY@tc##1{\textcolor[rgb]{0.25,0.50,0.50}{##1}}}
\expandafter\def\csname PY@tok@cm\endcsname{\let\PY@it=\textit\def\PY@tc##1{\textcolor[rgb]{0.25,0.50,0.50}{##1}}}
\expandafter\def\csname PY@tok@cpf\endcsname{\let\PY@it=\textit\def\PY@tc##1{\textcolor[rgb]{0.25,0.50,0.50}{##1}}}
\expandafter\def\csname PY@tok@c1\endcsname{\let\PY@it=\textit\def\PY@tc##1{\textcolor[rgb]{0.25,0.50,0.50}{##1}}}
\expandafter\def\csname PY@tok@cs\endcsname{\let\PY@it=\textit\def\PY@tc##1{\textcolor[rgb]{0.25,0.50,0.50}{##1}}}

\def\PYZbs{\char`\\}
\def\PYZus{\char`\_}
\def\PYZob{\char`\{}
\def\PYZcb{\char`\}}
\def\PYZca{\char`\^}

\def\PYZlt{\char`\<}
\def\PYZgt{\char`\>}
\def\PYZsh{\char`\#}
\def\PYZpc{\char`\%}

\def\PYZhy{\char`\-}
\def\PYZsq{\char`\'}
\def\PYZdq{\char`\"}

\makeatother

\makeatletter
\newbox\Wrappedcontinuationbox 
\newbox\Wrappedvisiblespacebox 
\newcommand*\Wrappedvisiblespace {\textcolor{red}{\textvisiblespace}} 
\newcommand*\Wrappedcontinuationsymbol {\textcolor{red}{\llap{\tiny$\m@th\hookrightarrow$}}} 
\newcommand*\Wrappedcontinuationindent {3ex } 
\newcommand*\Wrappedafterbreak {\kern\Wrappedcontinuationindent\copy\Wrappedcontinuationbox} 
\newcommand*\Wrappedbreaksatspecials {% 
	\def\PYGZus{\discretionary{\char`\_}{\Wrappedafterbreak}{\char`\_}}% 
	\def\PYGZob{\discretionary{}{\Wrappedafterbreak\char`\{}{\char`\{}}% 
	\def\PYGZcb{\discretionary{\char`\}}{\Wrappedafterbreak}{\char`\}}}% 
	\def\PYGZca{\discretionary{\char`\^}{\Wrappedafterbreak}{\char`\^}}% 
	\def\PYGZam{\discretionary{\char`\&}{\Wrappedafterbreak}{\char`\&}}% 
	\def\PYGZlt{\discretionary{}{\Wrappedafterbreak\char`\<}{\char`\<}}% 
	\def\PYGZgt{\discretionary{\char`\>}{\Wrappedafterbreak}{\char`\>}}% 
	\def\PYGZsh{\discretionary{}{\Wrappedafterbreak\char`\#}{\char`\#}}% 
	\def\PYGZpc{\discretionary{}{\Wrappedafterbreak\char`\%}{\char`\%}}% 
	\def\PYGZdl{\discretionary{}{\Wrappedafterbreak\char`\$}{\char`\$}}% 
	\def\PYGZhy{\discretionary{\char`\-}{\Wrappedafterbreak}{\char`\-}}% 
	\def\PYGZsq{\discretionary{}{\Wrappedafterbreak\textquotesingle}{\textquotesingle}}% 
	\def\PYGZdq{\discretionary{}{\Wrappedafterbreak\char`\"}{\char`\"}}% 
	\def\PYGZti{\discretionary{\char`\~}{\Wrappedafterbreak}{\char`\~}}% 
} 
\newcommand*\Wrappedbreaksatpunct {% 
	\lccode`\~`\.\lowercase{\def~}{\discretionary{\hbox{\char`\.}}{\Wrappedafterbreak}{\hbox{\char`\.}}}% 
	\lccode`\~`\,\lowercase{\def~}{\discretionary{\hbox{\char`\,}}{\Wrappedafterbreak}{\hbox{\char`\,}}}% 
	\lccode`\~`\;\lowercase{\def~}{\discretionary{\hbox{\char`\;}}{\Wrappedafterbreak}{\hbox{\char`\;}}}% 
	\lccode`\~`\:\lowercase{\def~}{\discretionary{\hbox{\char`\:}}{\Wrappedafterbreak}{\hbox{\char`\:}}}% 
	\lccode`\~`\?\lowercase{\def~}{\discretionary{\hbox{\char`\?}}{\Wrappedafterbreak}{\hbox{\char`\?}}}% 
	\lccode`\~`\!\lowercase{\def~}{\discretionary{\hbox{\char`\!}}{\Wrappedafterbreak}{\hbox{\char`\!}}}% 
	\lccode`\~`\/\lowercase{\def~}{\discretionary{\hbox{\char`\/}}{\Wrappedafterbreak}{\hbox{\char`\/}}}% 
	\catcode`\.\active
	\catcode`\,\active 
	\catcode`\;\active
	\catcode`\:\active
	\catcode`\?\active
	\catcode`\!\active
	\catcode`\/\active 
	\lccode`\~`\~ 	
}
\makeatother

\let\OriginalVerbatim=\Verbatim
\makeatletter
\renewcommand{\Verbatim}[1][1]{%
	\sbox\Wrappedcontinuationbox {\Wrappedcontinuationsymbol}%
	\sbox\Wrappedvisiblespacebox {\FV@SetupFont\Wrappedvisiblespace}%
	\def\FancyVerbFormatLine ##1{\hsize\linewidth
		\vtop{\raggedright\hyphenpenalty\z@\exhyphenpenalty\z@
			\doublehyphendemerits\z@\finalhyphendemerits\z@
			\strut ##1\strut}%
	}%
	\def\FV@Space {%
		\nobreak\hskip\z@ plus\fontdimen3\font minus\fontdimen4\font
		\discretionary{\copy\Wrappedvisiblespacebox}{\Wrappedafterbreak}
		{\kern\fontdimen2\font}%
	}%
	
	\Wrappedbreaksatspecials
	\OriginalVerbatim[#1,codes*=\Wrappedbreaksatpunct]%
}
\makeatother

\definecolor{incolor}{HTML}{303F9F}
\definecolor{outcolor}{HTML}{D84315}
\definecolor{cellborder}{HTML}{CFCFCF}
\definecolor{cellbackground}{HTML}{F7F7F7}

\makeatletter
\newcommand{\boxspacing}{\kern\kvtcb@left@rule\kern\kvtcb@boxsep}
\makeatother
\newcommand{\prompt}[4]{
	\ttfamily\llap{{\color{#2}[#3]:\hspace{3pt}#4}}\vspace{-\baselineskip}
}

\usepackage{titlesec}
\titleformat{\subsection}[runin]{\normalfont\bfseries}{\thesubsection}{.5em}{}[]

\usepackage{setspace}
\begin{document}
	
	%%------------------------------------------------------------------------------------
	%% Beginning of Frontmatter
	%%------------------------------------------------------------------------------------
	\frontmatter
	
\newcommand{\HRule}{\rule{\linewidth}{0.5mm}}

\begin{titlepage}
\centering
	
	\begin{figure}[h]
		\centering
		\includegraphics[scale=0.8]{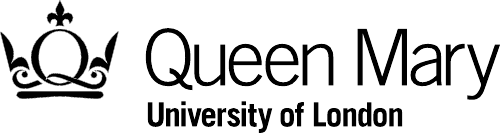}
	\end{figure}\ \\
	
	{\LARGE
		\textsc{Thesis submitted for the degree of \\[0.1cm] Doctor of Philosophy}
	}\\[1.1cm]
	
	\HRule\\[0.5cm]
	{\huge
		\bfseries{Permutation invariance, partition algebras and large $\N$ matrix models
	}}\\[0.5cm]
	\HRule\\[1.5cm]
	
	{\huge
		\textsc{Adrian Padellaro}
	}\\[2cm]
	
	{\Large
		Supervisor \\[0.3cm]
		\textsc{Dr Sanjaye Ramgoolam}\\[1.5cm]
		July 30, 2023}\\[0.5cm]
	
	{\large
		%\textsf{}
		Centre for Theoretical Physics \\[0.1cm]
		School of Physical and Chemical Sciences \\[0.1cm]
		Queen Mary University of London}

\end{titlepage} % titlepage
	
	\thispagestyle{empty}
\topskip0pt
\vspace*{3cm}
\begin{center}
	\textit{
		Dedicated to my family and partner.} \\
\end{center}

\vspace*{\fill}

	% Statement of originality
	%----------------------------------------------------------
	\chapter*{Declaration}

\vspace{-0.5cm}
I, Adrian Padellaro, confirm that the research included within this thesis is my own work or that where it has been carried out in collaboration with, or supported by others, that this is duly acknowledged below and my contribution indicated. Previously published material is also acknowledged below. \\

\noindent I attest that I have exercised reasonable care to ensure that the work is original, and does not to the best of my knowledge break any UK law, infringe any third party’s copyright or other Intellectual Property Right, or contain any confidential material. \\

\noindent I accept that the College has the right to use plagiarism detection software to check the electronic version of the thesis. \\

\noindent I confirm that this thesis has not been previously submitted for the award of a degree by this or any other university. \\

\noindent The copyright of this thesis rests with the author and no quotation from it or information derived from it may be published without the prior written consent of the author. \\

\noindent Signature:
\begin{figure}[h]
	\hspace{0.2cm}
\end{figure}

\noindent \vspace{-0.1cm}Date: \today \\

\noindent This thesis describes research carried out with my supervisor Sanjaye Ramgoolam and colleague George Barnes. It is based on and has substantial overlap with the published papers \cite{Barnes2022b, Barnes:2021tjp, Barnes:2022qli} and the paper \cite{PIGTM} which is due to appear in the near future.

	% Abstract
	%----------------------------------------------------------
	\chapter*{Abstract}

Models with degrees of freedom that naturally arrange themselves into matrices have a long history in science. Statistical models of large Hermitian matrices are believed to capture universal features of chaotic quantum systems. The various connections between matrix theory and string theory have been so prolific that one might argue that matrix models capture generic features of string theories. The first sign of this connection (gauge-string duality) was discovered by 't Hooft, where string worldsheets emerge from the combinatorics of Feynman diagrams in $U(\N \rightarrow \infty)$ Yang-Mills theory.

Many aspects of this emergence can be understood from the mathematical duality known as Schur-Weyl duality. It relates two algebraic structures: the representation theory of $U(\N)$ and the representation theory of symmetric groups $S_k$. This has implications for $U(\N)$ matrix models where observables find an eloquent description in terms of the group algebras $\mathbb{C}(S_k)$. The duality underlies the geometric construction of gauge-string duality, where string worldsheets emerge from a connection between symmetric group elements, ribbon graphs and Riemann surfaces.

In this thesis we will study matrix models with discrete gauge group $\SN$. We will put these matrix models into a generalized Schur-Weyl duality framework where dual algebras, known as partition algebras, emerge. These form generalizations of the symmetric group algebras -- they are semi-simple finite-dimensional associative algebras with a basis labelled by diagrams. We review the structure and representation theory of partition algebras. These algebras are then used to compute expectation values of $\SN$ invariant observables. This is a step towards studying the emergence of new geometric structures in their Feynman diagram expansion. Matrix models also appear in the form of quantum mechanical models of matrix oscillators. We explore the implications of the Schur-Weyl duality framework to matrix quantum mechanics with permutation symmetry.

	% Epighraph 
	%----------------------------------------------------------
%	\cleardoublepage
%	\thispagestyle{empty}
%	\epigraph{O time, thou must untangle this, not I.
%		It is too hard a knot for me to untie!}
%	{\textit{Twelfth Night}\\
%		\textsc{Shakespeare}}
%	
	% Acknowledgements and licences
	%----------------------------------------------------------
	\chapter*{Acknowledgements}
I want to thank my supervisor Sanjaye Ramgoolam for giving me the opportunity to do a PhD in theoretical physics. Without him I would not have experienced what ended up being the most fascinating four years of my life. I will be forever grateful for his patience as he guided me to my current level of mathematical maturity. I thank my family and my partner Felicia Tyllsjö for their love and patience, and their support in my pursuit of a career in academia. My friends and colleagues Rashid Alawadhi and Rajath Radhakrishnan -- who I spent countless hours discussing physics, philosophy and life with -- were a major contributing factor to my positive experience of the last four years. I also want to thank George Barnes, Sam Wikeley, Manuel Accettulli Huber, Stefano De Angelis, Shun-Qing Zhang, David Peinador Veiga, Graham Brown, Josh Gowdy, Mitchell Woolley, Lewis Sword, Tancredi Schettini, Kymani Armstrong-Williams, Mahesh Balasubramanian and Anindya Banerjee for making this journey a fun and lively experience. Lastly, I thank my examiners Denjoe O'Connor and Congkao Wen for their interesting comments, questions and suggestions for my thesis.

	%  Table of contents           
	%----------------------------------------------------------
	\tableofcontents
%	\listoffigures
%	\listoftables

	\mainmatter
	%%------------------------------------------------------------------------------------
	%% Beginning of Mainmatter
	%%------------------------------------------------------------------------------------
	\chapter{Introduction}
Models with degrees of freedom that naturally arrange themselves into matrices are ubiquitous in science \cite{Wigner1955, Dyson1962,Guhr1998, Edelman2013, Akemann:2016keq}. In physics, statistical models of large Hermitian matrices are believed to -- since the work of Wigner and Dyson -- capture universal features of chaotic quantum systems. Our understanding of the strong force as a Yang-Mills theory is another instance where matrix degrees of freedom appear. Despite a lot of effort and progress, a full non-perturbative understanding of the strong force remains elusive to this day. A beautiful framework for studying the strong force emerges by taking the gauge group $U(\N)$ to be very large. In the $\N \rightarrow \infty$ limit a new picture of the strong force emerges in terms of strings and their worldsheet geometry, as discovered by 't Hooft \cite{tHooft}. This can be viewed as the first sign of a gauge-string duality. Since then, many connections between matrix models and strings have been discovered, giving evidence to the perspective that large $\N$ matrix models capture generic features of string theories. Matrix models, and more generally tensor models, are also known to be closely related to integrable structures, hidden algebras \cite{PCA2016, Geloun2021, Mironov2022,Ramgoolam:2023vyq}, geometry and topological quantum field theories  \cite{MulaseYu, Kimura2014,Kimura2017}.

Recently, random matrix theory has been applied to tasks in computational linguistics.
The use of frequency vectors to study the meaning of words and phrases is an old idea \cite{Firth1957,Harris1968}. New frameworks extend this idea by modelling grammatical composition using tensor contractions \cite{coecke2010mathematical, Baroni2014FregeIS}. These constructions give rise to matrix data from language. In \cite{Kartsaklis2017}, matrix models were constructed to study the statistics of these matrices. A salient feature of these models is the lack of continuous symmetry. The remaining symmetry group is a discrete group $\SN \subset GL(\N)$ of permutation matrices. It was further proposed that the most important features of the matrix data can be captured by permutation invariant polynomial functions, referred to as observables. The general permutation invariant Gaussian matrix model was developed and solved in \cite{Ramgoolam2019a}. A selection of analytic formulas were provided for expectation values of low degree invariant polynomials. These models have successfully predicted features of real-world data \cite{Ramgoolam2019, Huber2022a}. Permutation invariant 2-matrix models were developed in \cite{Barnes2022b} and theoretical features of permutation invariant observables were investigated in \cite{Barnes:2021tjp}.

Many aspects of large $\N$ simplifications in matrix systems with continuous group symmetry are consequences of Schur-Weyl duality. The standard instance of Schur-Weyl duality \cite{Fulton2013} concerns tensor products $\VN^{\otimes k}$ of the fundamental representation of $U(\N)$. Well-known special cases include: $k=2$ which corresponds to the decomposition of matrices into symmetric and anti-symmetric parts
\begin{equation}
	\VN \otimes \VN \cong \Sym^2(\VN) \oplus \Lambda^2(\VN) \label{eq: intro VNVN}
\end{equation}
and $N=2$ which corresponds to the decomposition of $k$-body wave functions of spin-$\tfrac{1}{2}$ particles giving
\begin{equation}
	V_2^{\otimes k} = W_1^{\otimes k} \cong W_{k} \oplus (k-1)W_{k-1}, \label{eq: intro VNk}
\end{equation}
where $W_s$ is a spin-$\tfrac{s}{2}$ representation.
Generally, the symmetric group $S_k$ acts on $\VN^{\otimes k}$ by permuting tensor factors. This action commutes with the diagonal action of $U(\N)$ and Schur-Weyl duality amounts to the statement that the corresponding action of the group algebra $\mathbb{C}(S_k)$ contains the full algebra $\End_{U(\N)}(\VN^{\otimes k})$ of operators commuting with the action of $U(\N)$ on $\VN^{\otimes k}$. For large $\N$, they are isomorphic $\mathbb{C}(S_k) \cong \End_{U(\N)}(\VN^{\otimes k})$ and all irreducible representation of $S_k$ appear in $\VN^{\otimes k}$. Together with the double centralizer theorem, this implies that dimensions of irreducible representations of $S_k$ control the multiplicity of irreducible representations of $U(\N)$ in the decomposition of $\VN^{\otimes k}$. Standard results in representation theory of symmetric groups then imply that multiplicities are computed by enumerating standard Young tableaux. From the perspective of Schur-Weyl duality, the r.h.s. of \eqref{eq: intro VNVN} should be understood as the trivial and sign representation of $S_2$, respectively. The decomposition \eqref{eq: intro VNk} is understood in terms of irreducible representations of $S_k$. Namely the trivial and $(k-1)$-dimensional representation (known as the standard or hook representation), respectively.
The fact that \eqref{eq: intro VNk} only contains two irreducible representations of $S_k$ is a so-called finite $\N$ effect -- anti-symmetrising over more than two indices gives a vanishing result.

The generalization to arbitrary $N$ and $k$  has important implications for the classification and construction of matrix polynomial functions $f(X)$ invariant under conjugation by $U(\N)$: $f(X) = f(gXg^\dagger)$ for $g \in U(\N)$. Well-known examples of such functions are the multi-trace observables in supersymmetric $U(N)$ gauge theories. At large $N$, multi-trace observables of degree $k$ are in one-to-one correspondence with conjugacy classes of $S_k$. Let $\tau$ be an element of $S_k$ with $c_i$ cycles of length $i=1,\dots, k$, then
\begin{equation}
	O_\tau = \sum_{i_1, \dots, i_k =1}^\N X_{i_1 i_{(1)\tau}} \dots X_{i_k i_{(k)\tau}} = (\Tr X^1)^{c_1} (\Tr X^2)^{c_2} \dots (\Tr X^k)^{c_k}.
\end{equation}
Fourier transforming $\mathbb{C}(S_k)$, using irreducible characters $\chi^\lambda(\tau)$ of $S_k$, gives an alternative basis for invariant observables, labelled by irreducible representations $V_\lambda$ of $S_k$
\begin{equation}
	O_\lambda = \frac{1}{k!}\sum_{\tau \in S_k} \chi^\lambda(\tau) O_\tau.
\end{equation}
While the multi-trace basis is orthogonal (with respect to the free two-point function) up to $O(1/\N)$ corrections, the representation basis is orthogonal for all $N$.
These techniques, based on Schur-Weyl duality have lead to many important results for gauge invariant observables in matrix theories relevant to AdS/CFT. Some highlights include the identification of CFT duals \cite{Balasubramanian2002, CJR, Berenstein2004} of giant gravitons \cite{McGreevy2000, HHI2000, GMT2000} and computation of correlators \cite{CJR, Kimura2007, Brown2008, Bhattacharyya:2008rb, Bhattacharyya2008b, Kimura2008, Brown2009, QuivCalc, CDD1301, Ber1504, KRS, CLBSR, ADHSSS, LY2107}.

The defining representation $\VN$ of the symmetric group $\SN$ corresponds to the set of permutation matrices. Schur-Weyl duality for $\SN$ acting on $\VN^{\otimes k}$ was first discovered by Jones and Martin \cite{Jones1994, Martin1994, Martin1996} while studying statistical mechanics and Potts models. The dual algebras $\End_{\SN}(\VN^{\otimes k})$, controlling the multiplicities in the decomposition of $\VN^{\otimes k}$, were identified with the so-called partition algebras $P_k(\N)$. In Potts models, partition algebras appear as the algebra generated by transfer matrices. Partition algebras are so-called diagram algebras. Namely, finite-dimensional associative algebras with distinguished bases labelled by diagrams, where the product is defined through diagram concatenation. The decomposition \eqref{eq: intro VNVN} is not irreducible with respect to the diagonal action of $\SN$. This is captured by Schur-Weyl duality and instead one finds
\begin{equation}
	\VN \otimes \VN \cong 2 V_{[\N]} \oplus 3 V_{[\N{-}1,1]} \oplus V_{[\N{-}2,2]} \oplus V_{[\N{-}2,1,1]}, \label{eq: intro VNVN for SN}
\end{equation}
where the irreducible representations on the r.h.s. are labelled by integer partitions of $\N \geq 4$. In analogy to the $U(N)$ case, the multiplicities on the r.h.s. are determined by dimensions of irreducible representations of $P_2(\N)$. This is a refinement of the matrix decomposition in \eqref{eq: intro VNVN}, into components of the matrix that are invariant under conjugation by a permutation matrix. The combinatorial analogue of standard Young tableaux, in this case, are called vacillating tableaux and their enumeration corresponds to computing the multiplicities in the above decomposition.
Recently, many mathematical results about partition algebras and their representation theory have been developed. Some highlights include expressions for characters \cite{Halverson2001}, combinatorial formulas for the dimensions of irreducible representations \cite{Benkart2017}, combinatorial (Young diagram) realisations of the irreducible representations \cite{Halverson2018} and many more (see \cite{Halverson2005, Enyang_2012, Benkart2016, Doty2019} and references therein).

In this thesis we will leverage these new results to put the zero-dimensional permutation invariant matrix models into the Schur-Weyl duality framework, previously only developed for continuous symmetries. The decomposition \eqref{eq: intro VNVN for SN} plays a crucial role in the construction \cite{Ramgoolam2019a} of the general permutation invariant Gaussian matrix model. As we will review, the quadratic part of the action is constructed by coupling parts of the matrices that transform in isomorphic irreducible representations of $\SN$. The resulting coupling matrices are invariant tensors that can be understood in terms of elements of $\End_{\SN}(\VN^{\otimes 2})$, or equivalently $P_2(\N)$.
Therefore, the explicit construction of these coupling matrices is equivalent to the construction of a basis for $P_2(\N)$. In particular, we will explain how the construction of, and solution to, the most general quadratic permutation invariant model in \cite{Ramgoolam2019a} can be understood from the construction of a basis of matrix units for $P_1(\N)$ and $P_2(\N)$, also known as an Artin-Wedderburn decomposition \cite{Artin,Wedderburn}. As we will see, general observables are closely related to partition algebras, and this connection gives rise to an algebraic/combinatorial algorithm for computing their expectation values as rational functions of $\N$.

Since the discovery of simplifications of large $\N$ matrix quantum field theories by 't Hooft \cite{tHooft}, systems with matrix degrees of freedom have played a major role in the development of gauge-string dualities. Examples of gauge-string duality based on large $\N$ include: the duality between non-critical strings and matrix models \cite{Douglas1990, Brezin1990, Gross1990}, Gaussian matrix models and Belyi maps \cite{ITZYK, MelloKoch2010, Gopak2011, dMKLN}, two-dimensional Yang-Mills and Hurwitz spaces \cite{1993Gross_1,GrossTaylor, Minahan1993, SCHNITZER1993, Gross1993a, MP1993, Horava1996, Cordes1997, Kimura:2008gs} and AdS/CFT \cite{Malda, Witten1998, Gubser1998}.
Gauge-string dualities are difficult to study in full generality. Therefore, it has been useful to study them in restricted corners of parameter space, or manageable sectors. Clever choices of such restrictions have led to a rich interplay between gauge-string dualities and quantum many-body systems. This includes the connections between free fermions and large $\N$ two-dimensional Yang-Mills \cite{Douglas1993YM}, half-BPS sectors of $\mathcal{N}=4$ SYM \cite{CJR, Berenstein2004}, supersymmetric indices \cite{Murthy2022}; 3D bosons and eight-BPS states in $\mathcal{N}=4$ SYM \cite{Berenstein2006}; and spin matrix theories as quantum mechanical models of AdS/CFT \cite{Harmark2014,Baiguera2022}.

Motivated by this, we also study the implications of permutation symmetry on quantum mechanical matrix models. The mathematical techniques used in zero dimensions can be leveraged to construct exactly solvable models of matrix oscillators in a permutation invariant potential. We also describe the subspace of states invariant under the adjoint action of permutations. We give a geometric basis for this subspace that generalizes the multi-trace basis for singlet states in gauged matrix quantum mechanics. This basis exhibits a form of large $\N$ factorization -- or large $\N$ orthogonality -- generalizing the large $\N$ factorization of multi-trace states that underlies their interpretation in terms of multi-particle states \cite{Balasubramanian2002}. Large $\N$ factorization is also used in the construction of gauge-string duals in collective field theory \cite{Jevicki1980, Yaffe1982, Das1990, MelloKoch2011a}, a useful tool for studying the emergence of classical limits at large $\N$. Similarly, it enters the Master field approach to large $\N$ \cite{Witten1980}.

The outline of the thesis is as follows.
\begin{itemize}
	\item
In Chapter \ref{ch: SN} we give a brief review of the basic objects used in the theory of symmetric groups $\SN$. Namely, Young diagrams, tableaux, group algebras and centres of group algebras. This is supplemented by Appendix \ref{apx: SN units} where we give a procedure for constructing matrix units for the symmetric group algebras, and consequently irreducible representations of $\SN$. This appendix is meant to serve as a warm-up, in a more familiar setting, to the construction of matrix units for partition algebras given in chapter \ref{chapter: partition algebra}. The main results of this chapter concerns the most concrete representation of $\SN$, as permutation matrices, also known as the defining representation $\VN$. We discuss explicit decompositions of the defining representation into irreducible representations. Lastly, we consider tensor powers $\VN^{\otimes k}$ and give two ways of counting the multiplicity of irreducible representations in the decomposition. We review how multiplicities can be understood through the combinatorial objects known as vacillating tableaux. These play a major role in the following chapter.
	\item 
Chapter \ref{chapter: partition algebra} elaborates on many of these results. They are put into the context of Schur-Weyl duality where the multiplicities are understood as dimensions of irreducible representations of the partition algebras $P_k(\N)$. We give an explicit description of the partition algebras in terms of diagrams, which multiply through diagram concatenation. The rest of the chapter sets up the necessary mathematical background for constructing matrix units for $P_k(\N)$. This is very similar in spirit to the previously mentioned construction in Appendix \ref{apx: SN units}. Explicit tables of matrix units for $P_2(\N)$ are found in Appendix \ref{apx: P2N units}, these are used in the subsequent chapter to construct permutation invariant Gaussian matrix models.
	\item
Chapter \ref{chapter: 0d} is all about zero-dimensional matrix models with permutation symmetry. In particular, permutation symmetry with respect to the adjoint (conjugation) action of permutation matrices. We describe the most general permutation invariant Gaussian matrix models and derive expression for the first and second moment of matrix elements. These moments can be expressed as linear combinations of the matrix units constructed in Appendix \ref{apx: P2N units}. Observables in this matrix model are defined to be permutation invariant matrix polynomials. We give two descriptions of the space of observables, in terms of equivalence classes of set partitions, and directed graphs. The former description is most useful for computing expectation values of observables, while the latter is more useful for combinatorial counting and construction. The penultimate section in this chapter describes a combinatorial algorithm for computing expectation values of observables. In the last section we give a description of directed graphs in terms of double cosets. This is supplemented with Appendix \ref{apx: double coset} where we describe generating functions for sizes of these double cosets. We also provide code implementing the double coset counting.
	\item
In Chapter \ref{ch: 1d} we apply the mathematical framework developed in previous chapters to matrix quantum mechanics. We start the Chapter with a review of quantum matrix harmonic oscillators, that is, $\N^2$ uncoupled harmonic oscillators. The next section defines a model of harmonic oscillators in a permutation invariant quadratic potential. We explain that this Hamiltonian can be exactly diagonalized for arbitrarily large $\N$ by taking advantage of the permutation symmetry of the system.
In the following section we consider a subspace of the total Hilbert space, made out of permutation invariant states. The invariant states have an algebraic description in terms of partition algebras and therefore inherit three natural bases: the diagram basis, orbit basis and representation basis. The diagram basis satisfies a generalization of the large $\N$ factorisation well-known for multi-trace observables. In the next section we describe some algebraic Hamiltonians based on diagram algebras. We give Hamiltonians that are diagonalized in the representation basis previously mentioned. We discuss extensions of this construction and possibilities of constructing Hamiltonians for which the representation basis forms a complete set of non-degenerate eigenvectors.
\end{itemize}
	
	\chapter{Symmetric groups}\label{ch: SN}
The symmetric groups $\SN$ of permutations of $\N$ objects play a prominent role in the theory of finite groups and their representation theory. Their representation theory, including computations of characters and construction of irreducible representations, is well understood in terms of combinatorial objects such as Young diagrams and tableaux (see \cite{Sagan2013} for a dedicated mathematical treatment or \cite{Hamermesh1962} for a standard reference aimed at physicists).

In this Chapter, we will start by reviewing a subset of these objects that will be relevant for future chapters. We will focus most of our attention on the most concrete realisation of the symmetric group, as a set of permutation matrices, known as the defining representation. The main result of this chapter concerns the decomposition \eqref{eq: VNotimesk} of tensor products of defining representations and its combinatorial interpretation in Theorem \ref{thm: multi is vac tab}. This result is absolutely central to the construction of matrix units in Chapter \ref{chapter: partition algebra} which leads to the construction of matrix models in Chapter \ref{chapter: 0d} and \ref{ch: 1d}.
\section{Review, notation and conventions}
In this thesis we use the following definition of symmetric groups.
\begin{definition}[Symmetric group] \label{def: SN}
	The symmetric group $\SN$ is the set of bijective maps $\sn: \{1,\dots,\N \} \rightarrow \{1,\dots,\N\}$ with multiplication given by composition of maps.
\end{definition}
The number of elements in the set $\SN$ is $\abs{\SN} = \N!$.

We read products of group elements from left to right. That is, for $i \in \{1, \dots, \N\}$ and $\sn_1, \sn_2 \in \SN$ the product $\sn_1 \sn_2$ corresponds to the map $i \mapsto \sn_2(\sn_1(i))$ which we write $(i)\sn_1 \sn_2$.
\begin{example}\label{ex: S3 grp law}
	Consider $\N = 3$ and the two maps $\sn_1, \sn_2 \in \SN$ given by
	\begin{equation}
		(1)\sn_1=2, (2)\sn_1=3, (3)\sn_1=1, \quad (1)\sn_2=2, (2)\sn_2=1, (3)\sn_2=3,
	\end{equation}
	or in cycle notation $\sn_1 = (123), \sn_2 = (12)(3)$.
	Then $\sn_1 \sn_2$ corresponds to the map
	\begin{equation}	
		\begin{aligned}
			&(1)\sn_1 \sn_2  = (2)\sn_2 = 1,\\
			&(2)\sn_1 \sn_2 = (3)\sn_2 = 3, \\
			&(3)\sn_1 \sn_2  = (1)\sn_2 = 2.
		\end{aligned}
	\end{equation}
	and in cycle notation $\sn_1 \sn_2 = (1)(23)$.	
\end{example}
\begin{definition}[Cycle structure]	
The conjugacy classes of $\SN$ correspond to cycle structures of elements in cycle notation. For $\sn = c_1 c_2 \dots c_l$ decomposed into disjoint cycles $c_i$ of length $\abs{c_i}$ we say that it has cycle structure $c(\sn) = [\abs{c_1}, \abs{c_2}, \dots, \abs{c_l}]$ and it lies in the conjugacy class $\Cclass_{[\abs{c_1}, \abs{c_2}, \dots, \abs{c_l}]}$. Note that
\begin{equation}
	\sum_{i=1}^l \abs{c_i} = \N,
\end{equation}
and because all the cycles are disjoint they commute. Therefore, we can choose to order the elements in the decomposition $\sn = c_1 c_2 \dots c_l$ such that $\abs{c_1} \geq \abs{c_2} \geq \dots \geq \abs{c_l}$. This defines an integer partition of $\N$.
\end{definition}
\begin{example}
	The following permutations $\sn_1 = (123)(45), \, \sn_2 = (12345), \, \sn_3 = (12)(3)(4)(5)$ in $S_5$ have cycle structure $[3,2], \, [5], \, [2,1,1]$ respectively.
\end{example}
\begin{definition}[Integer partition]
	A list $\ptition = [\ptition_1, \ptition_2, \dots, \ptition_l]$ of integers with $\ptition_1 \geq \ptition_2 \geq \dots \geq \ptition_l$ satisfying $\sum_i \ptition_i = \N$ is called an integer partition of $\N$. We use the shorthand $\ptition \vdash \N$ to say that $\ptition$ is an integer partition of $\N$ and shorthand $\abs{\lambda} = \N$.
	An entry in $\ptition$ is called a part and the number of parts is denoted $l(\ptition) = l$.
\end{definition}
\subsection{Young diagrams and representations of $\SN$.}
Integer partitions are in bijection with the combinatorial objects known as Young diagrams.
\begin{definition}[Young diagram]
	A Young diagram $\YT{\ptition}$  of shape $\ptition \vdash \N$ is diagram with $l(\ptition)$ rows where each row $i$ has $\ptition_i$ boxes.
\end{definition}
\begin{example}
	The following are Young diagrams with $\N = 4$ boxes,
	\begin{equation}\ytableausetup{smalltableaux}
		\YT{[4]} = \ydiagram{4}\quad \YT{[3,1]} = \ydiagram{3,1} \quad \YT{[2,2]} = \ydiagram{2,2}.
	\end{equation}
\end{example}
The purpose of describing integer partitions in terms of boxes is to fill the boxes with numbers, or more generally elements of an ordered set.
\begin{definition}[Young tableau]
	A filling of the boxes of a Young diagram with integers taken from a subset of the positive integers is called a Young tableau.
\end{definition}
\begin{example}
	All of the below fillings are examples of Young tableaux
	\begin{equation}\ytableausetup{smalltableaux=false}
		\ytableaushort{1 2 3, 4 5} \quad  \ytableaushort{1 3 3, 4 5 5} \quad \ytableaushort{5 2 3, 1, 1 ,1}
	\end{equation}
\end{example}
The most useful Young tableaux have restrictions on their fillings. For example, standard tableaux are used to compute dimensions of irreducible representations of $\SN$.
\begin{definition}[Standard tableau]
	Let $\lambda \vdash \N$ and $\YT{\lambda}$ be the corresponding Young diagram. Consider a filling of $\YT{\lambda}$ using integers from the set $\{1, \dots \N\}$. The filling is called standard if it is increasing along every row (read left to right) and every column (read top to bottom). A Young tableau with standard filling is called standard tableau. The set of standard tableaux with shape $\lambda$ is denoted $\SYT{\lambda}$.
\end{definition}
\begin{example}
	First we give some examples of standard tableaux
	\begin{equation}
		\ytableaushort{1 2 3, 4 5}\quad \ytableaushort{1 2, 3 4} \quad \ytableaushort{1 3, 2 4} \quad \ytableaushort{1, 2, 3}
	\end{equation}
	The following tableaux are not standard
	\begin{equation}
		\ytableaushort{1 3 2, 4 5} \quad \ytableaushort{1 2, 4 3} \quad \ytableaushort{1 3, 4 2} \quad \ytableaushort{2, 1, 3}
	\end{equation}
\end{example}
The number of standard tableaux can be computed using the famous hook formula.
\begin{theorem}(Hook formula)
	Let $\lambda \vdash \N$, then
	\begin{equation}
		\abs{\SYT{\lambda}} = \frac{\N!}{\prod_{(i,j) \in Y_\lambda} \abs{h(i,j)}},
	\end{equation}
	where the product is over the cells $c$ in the Young diagram $Y_\lambda$, $h(i,j)$ is the hook of cell $(i,j)$ and $\abs{h(i,j)}$ is the number of cells in the hook.
\end{theorem}
\begin{proof}
	See \cite[Section 3.10]{Sagan2013} which also includes the history surrounding the hook formula. One method for proving this formula is to re-write the hook formula as
	\begin{equation}
	 \N ! = \abs{\SYT{\lambda}}\prod_{(i,j) \in Y_\lambda} \abs{h(i,j)},
	\end{equation}
	and realising that $\N!$ is the number of tableaux (a filling using $1,\dots,\N$ without repetition but no further restrictions) of shape $\lambda$. One then constructs a set $S$ of order $\abs{S} = \abs{\SYT{\lambda}}\prod_{(i,j) \in Y_\lambda} \abs{h(i,j)}$ and bijective maps between $S$ and the set of tableaux of shape $\lambda$.
\end{proof}

Having introduced some of the important notions used in the combinatorial representation theory of symmetric groups we will quote the following beautiful result without proof.
\begin{theorem}[Irreducible representations of $\SN$] \label{thm: SN irreps}
	There exists a set of vector spaces $\{V_{\lambda} \, \vert \, \lambda \vdash \N\}$ with $\dim V_\lambda = \abs{\SYT{\lambda}}$ that forms a complete set of non-isomorphic irreducible representations of $\SN$. That is, the vector space $V_\lambda$ has a basis labelled by standard tableaux of shape $\lambda$ and $\sn \in \SN$ acts by permuting the numbers in the filling.
\end{theorem}
\begin{proof}
	See \cite[Theorem 2.4.6 and Theorem 2.5.2]{Sagan2013}. In Appendix \ref{apx: SN units} we elaborate on this theorem and outline an algorithm for constructing all irreducible representations of $\SN$. These realisations are called the Young bases and have a combinatorial description in terms of permutations acting on standard Young tableaux. The principle used in this appendix is closely related to the construction in Chapter \ref{chapter: partition algebra} that allows us to solve permutation invariant Gaussian matrix models in Chapter \ref{chapter: 0d}.
\end{proof}

\subsection{Group algebras and centers.}
Representations of $\SN$ also give rise to representations of its group algebra, which will play an important role in this thesis.
\begin{definition}[Group algebra] \label{def: group algebra}
	Let $G$ be a finite group and $\mathbb{F} = \mathbb{R}, \mathbb{C}$ be the real or complex numbers. The group algebra $\mathbb{F}(G)$ is the vector space
	\begin{equation}
		\mathbb{F}(G) = \Span_\mathbb{F}(g \in G),
	\end{equation}
	with multiplication defined through group multiplication and linearity.
\end{definition}
\begin{example}
	Let $G$ be a finite group with multiplication defined by
	\begin{equation}
		gh = C_{gh}^{g'} g',
	\end{equation}
	for $g,h,g' \in G$. For two elements $a,b \in \mathbb{F}(G)$ with expansions
	\begin{equation}
		a = \sum_{g \in G} a_g g, \quad b = \sum_{h \in G} b_h h, 
	\end{equation}
	the product $ab$ is
	\begin{equation}
		ab = \sum_{g,h \in G} a_g b_{h} gh =  \sum_{g,h,g' \in G} a_g b_{h} C_{gh}^{g'} g'.
	\end{equation}
\end{example}
The group algebra has a subalgebra that is particularly important in representation theory.
\begin{definition}[Center of $\mathbb{F}(G)$] \label{def: center}
	The center $\mathcal{Z}[\mathbb{F}(G)]$ of a group algebra $\mathbb{F}(G)$ is the set of elements $z \in \mathbb{F}(G)$ such that
	\begin{equation}
		z g = gz, \quad \forall g \in G.
	\end{equation}
\end{definition}
The center of a group algebra has two canonical bases.
\begin{proposition} \label{prop: cc basis of ZFG}
	Let $\Cl{G}$ be the set of conjugacy classes of $G$. The set of elements
	\begin{equation}
		\{z_C = \sum_{g \in C} g \, \vert \, \forall C \in \Cl{G}\}
	\end{equation}
	form a basis for the center.
\end{proposition}
\begin{proof}
	First we prove that any element in the center can be expanded in terms of $z_C$.
	Suppose $z \in \mathcal{Z}[\mathbb{F}(G)]$ has expansion
	\begin{equation}
		z = \sum_{h \in G} a_h h.
	\end{equation}
	For $z$ to be a central element it has to satisfy
	\begin{equation}
		g^{-1}zg = z \Rightarrow a_{g^{-1}hg} = a_h.
	\end{equation}
	In other words, the coefficients of elements in the same conjugacy class are equal.
	For $h \in C$ we write $a_C = a_h$, then we can rewrite $z$ as
	\begin{equation}
		z = \sum_{C \in \Cl{G}} \sum_{h \in C} a_h h = \sum_{C \in \Cl{G}} a_C z_C.
	\end{equation}
	Lastly, we note that $g^{-1} z_C g = z_C$. Therefore, any linear combination of $z_C$ is in the center.
\end{proof}
\begin{remark}
	This shows that $\dim \mathcal{Z}[\mathbb{F}(G)] = \abs{\Cl{G}}$.
\end{remark}
The second canonical basis is labelled by irreducible representations of $G$.
\begin{proposition} \label{prop: rep basis of ZFG}
	Let $\Rep{G}$ be a labelling set for the set of non-isomorphic irreducible representations of $G$ and for $R \in \Rep{G}$, let $\chr^R$ the corresponding irreducible character. The set
	\begin{equation}
		\{p^{}_R = \frac{1}{|G|} \sum_{g \in G} \chr^R(g^{-1})g \, \vert \, \forall R \in \Rep{G}\}
	\end{equation} 
	form a basis for $\mathcal{Z}[\mathbb{F}(G)]$.
\end{proposition}
\begin{proof}
	First we show that the $p^{}_R$ are linearly independent by showing that they are orthogonal with respect to the inner product
	\begin{equation}
		(g,h) = \begin{cases}
			1, \qq{if $g=h^{-1}$}\\
			0, \qq{otherwise.}
		\end{cases}
	\end{equation}
	We have
	\begin{align}
		(p^{}_{R_1}, p^{}_{R_2}) &= \frac{1}{|G|^2}\sum_{g,h \in G} \chr^{R_1}(g^{-1})\chr^{R_2}(h^{-1}) (g,h) \\
		&=\frac{1}{|G|^2}\sum_{h \in G} \chr^{R_1}(h)\chr^{R_2}(h^{-1})= \frac{\delta^{R_1 R_2}}{|G|},
	\end{align}
	where the last step uses orthogonality of characters.
	From the equalities $\dim \mathcal{Z}[\mathbb{F}(G)]  = \abs{\Cl{G}} = \abs{\Rep{G}}$, it follows that the $p^{}_{R}$ form a spanning set.
\end{proof}

The following corollary highlights the importance of the centre to the study of representations.
\begin{corollary}
	Let $R \in \Rep{G}$ and $D^R(g)$ the corresponding matrix for $g \in G$. Schur's lemma implies that the irreducible representation of a central element $z \in \mathcal{Z}(\mathbb{F}(G))$ is proportional to the identity matrix
	\begin{equation}
		D^R(z) = \frac{\chi^R(z)}{\chi^R(1)}\idn.
	\end{equation}
\end{corollary}
This will be used many times in the following chapters.

\section{Defining representation}
The symmetric group $\SN$ is faithfully represented by the set of $\N \times \N$ permutation matrices. This representation is called the defining representation of $\SN$ and is given by the following.
\begin{definition}[Defining representation of $\SN$]
	Let $\VN$ be a $\N$-dimensional vector space with basis $\{e_1, \dots, e_\N\}$. The defining representation of $\SN$ associates to every $\sn \in \SN$ a linear map $\Paction{\sn} \in \End(\VN)$ defined by
	\begin{equation}
		\Paction{\sn}e_i = e_{(i)\sn^{-1}}. \label{eq: perm action}
	\end{equation}
\end{definition}

As we now prove, the definition in equation \ref{eq: perm action} is a homomorphism. Let $\sn_1, \sn_2 \in \SN$ and consider
\begin{equation}
	\Paction{\sn_1}\Paction{\sn_2}e_i = \Paction{\sn_1}e_{(i)\sn_2^{-1}} = e_{(i)\sn_2^{-1}\sn_1^{-1}} = e_{(i)[\sn_1 \sn_2]^{-1}} = \Paction{\sn_1 \sn_2}e_i.
\end{equation}
We will use the common abuse of language of referring to $\VN$ as the defining representation, with the homomorphism $\Paction{\sn}$ implicitly included.
The matrices corresponding to the linear operators $\Paction{\sn}$ defined in equation \ref{eq: perm action} are permutation matrices. To see this, we use equation \ref{eq: perm action}, which gives
\begin{equation}
	\Pmat{\sn}^j_i e_j = e_{(i)\sn^{-1}}.
\end{equation}
or
\begin{equation}
	\Pmat{\sn}^j_i = \delta^j_{(i)\sn^{-1}} = \delta^{(j)\sn}_{ i}. \label{eq: VN as perm matrix}
\end{equation}
where the last equality follows because $j = (i)\sn^{-1}$ implies that $(j)\sn = i$. Notably, this is a permutation of the rows of the identity matrix. That is, $\Pmat{\sn}^j_{i}$ is a permutation matrix.

\begin{example}
	Consider $\sn_1 = (123), \sn_2 = (12)(3)$ as in Example \ref{ex: S3 grp law}. Using equation \ref{eq: VN as perm matrix} we have
	\begin{equation}
		\begin{aligned}
			&\Pmat{\sn_1}^1_{i} = \delta^2_{i}, 	\Pmat{\sn_1}^2_{i} = \delta^3_{i}, 	\Pmat{\sn_1}^3_{i} = \delta^1_{i} \\
			&\Pmat{\sn_2}^1_{i} = \delta^2_{i}, 	\Pmat{\sn_2}^2_{i} = \delta^1_{i}, 	\Pmat{\sn_2}^3_{i} = \delta^3_{i},
		\end{aligned}
	\end{equation}
	or
	\begin{equation}
		\begin{aligned}
			&\Pmat{\sn_1}^j_i = \mqty(0 & 1 & 0 \\ 0 & 0 & 1 \\ 1 & 0 & 0),\\
			&\Pmat{\sn_2}^j_i = \mqty(0 & 1 & 0 \\ 1 & 0 & 0 \\ 0 & 0 & 1).
		\end{aligned}
	\end{equation}
\end{example}

\subsection{Decomposition of defining representation.}
We will now study the defining representation in some detail. In particular, we will prove that it is a reducible representation, give its decomposition into irreducible representations and give an explicit basis for each irreducible subspace.
\begin{proposition}
	The defining representation of $\SN$ is a reducible representation. It decomposes into a one-dimensional representation $V_{[\N]}$ and a $(\N{-}1)$-dimensional irreducible representation $V_{[\N{-}1,1]}$.
\end{proposition}
\begin{proof}
	Define the homomorphism $\phi: e_i \mapsto e_1 + e_2 + \dots + e_\N$ and
	\begin{equation}
		V_{[\N]} = \im \phi = \Span(e_1 + e_2 + \dots + e_{\N}), \quad V_{[\N{-}1,1]} = \ker \phi.
	\end{equation}
	Because $\phi$ is a homomorphism, the isomorphism theorem says that the above subspaces are representations of $\SN$, and we have the following decomposition of $\VN$,
	\begin{equation}
		\VN = V_{[\N]} \oplus V_{[\N{-}1,1]}. \label{eq: VN decomp}
	\end{equation}
	The representation $V_{[\N]}$ forms an invariant one-dimensional subspace of $\VN$ and is therefore irreducible. We will now prove that $V_{[\N{-}1,1]}$ is irreducible as well.
	
	Let $\chr$ be the character of $\VN$
	\begin{equation}
		\chr(\sn) =  \Pmat{\sn}^i_i.
	\end{equation}
	It is equal to the number of fixed points of $\sn$, which we call $\fixp{\sn}$
	\begin{equation}
		\chr(\sn)  = \fixp{\sn}. \label{eq: VN character}
	\end{equation}
	
	Let $\chr^{}_{[\N]}=1, \chr^{}_{[\N{-}1,1]}$ be characters of $V_{[\N]}, V_{[\N{-}1,1]}$ respectively.
	By equation \ref{eq: VN decomp} we have
	\begin{equation}
		\chr^{}_{[\N{-}1,1]} = \chr - \chr^{}_{[\N]} = \chr - 1. \label{eq: hook character}
	\end{equation}
	
	Character orthogonality implies that
	\begin{equation}
		\frac{1}{\abs{\SN}} \sum_{\sn \in \SN} \chr^{}_{[\N{-}1,1]}(\sn) \chr^{}_{[\N{-}1,1]}(\sn^{-1})  = 1
	\end{equation}
	if and only if $V_{[\N{-}1,1]}$ is irreducible. Substituting equation \ref{eq: hook character} into the above equation gives
	\begin{equation}
		\frac{1}{\abs{\SN}} \sum_{\sn \in \SN} \qty(\chr(\sn)^2 - 2 \chr(\sn) + 1) = 1.
	\end{equation}
	To prove this, we use Burnside's lemma. First, consider
	\begin{equation}
		\frac{1}{\abs{\SN}} \sum_{\sn \in \SN} \chr(\sn) = \frac{1}{\abs{\SN}} \sum_{\sn \in \SN} \fixp{\sn}.
	\end{equation}
	Burnside's lemma say that
	\begin{equation}
		\frac{1}{\abs{\SN}} \sum_{\sn \in \SN} \fixp{\sn} = \text{\# Orbits of $\SN$ acting on $\{1,\dots, \N\}$} = 1, \label{eq: orbits of X}
	\end{equation}
	where the last equality follows since there exists at least one $\sn \in \SN$ such that $(i)\sn = j$ for any pair $(i,j)$. To evaluate
	\begin{equation}
		\frac{1}{\abs{\SN}} \sum_{\sn \in \SN} \chr(\sn)^2 = \frac{1}{\abs{\SN}} \sum_{\sn \in \SN} \fixp{\sn}^2
	\end{equation}
	note that the number of fixed points of $(\sn_1, \sn_2) \in \SN \times \SN$ acting on $(i,j) \in \{1,\dots,\N\} \times \{1,\dots,\N\}$ as
	\begin{equation}
		(i,j) \mapsto ((i)\sn_1, (j)\sn_2)
	\end{equation}
	is $\fixp{\sn_1}\fixp{\sn_2}$. The relevant special case is $\sn_1 = \sn_2 = \sn$. The average number of fixed points of this diagonal subgroup is
	\begin{equation}
		\frac{1}{\abs{\SN}} \sum_{\sn \in \SN} \fixp{\sn}^2 = \text{\# Orbits of $\SN$ acting on $\{1,\dots,\N\}^{\times 2}$} = 2, \label{eq: orbits of X2}
	\end{equation}
	where the last equality follows because elements of the form $(i,i)$ and $(i,j)$ with $i\neq j$ form distinct orbits. Equation \ref{eq: orbits of X} and \ref{eq: orbits of X2} proves
	\begin{equation}
		\frac{1}{\abs{\SN}} \sum_{\sn \in \SN} \qty(\chr(\sn)^2 - 2 \chr(\sn) + 1) = 2-2+1 = 1,
	\end{equation}
	and therefore $V_{[\N{-}1,1]}$ is irreducible.
\end{proof}

\begin{proposition} \label{prop: VN decomp}
	The vectors
	\begin{align}
		E^{[\N]} &= \frac{e_1 + e_2 + \dots + e_\N}{\sqrt{\N}} \\
		E^{[\N-1,1]}_1&= \frac{e_1 - e_2}{\sqrt{2}} \\
		%		E_2 &= \frac{e_1 + e_2 - 2 e_3}{\sqrt{6}} \\
		&\,\, \vdots \\
		E^{[\N-1,1]}_a &= \frac{e_1 + e_2 + \dots + e_a - a e_{a+1}}{\sqrt{a(a+1)}} \\
		&\, \, \vdots \\
		E^{[\N-1,1]}_{\N{-}1} &= \frac{e_1 + e_2 + \dots + e_{\N{-}1} - (\N{-}1) e_{\N} }{\sqrt{\N(\N+1)}}
	\end{align}
	form an orthonormal basis for $V_{[\N]} \oplus V_{[\N{-}1,1]}$ with respect to the $\SN$-invariant inner product $\VNinner{e_i}{e_j} = \delta_{ij}$.
\end{proposition}
\begin{proof}
	Let $a,b \in \{1,\dots,\N{-}1\}$ and assume $a < b$. Then
	\begin{equation}
		\VNinner{E^{[\N-1,1]}_a}{E^{[\N-1,1]}_b} = \frac{\sum_{i = 1}^{a} \VNinner{e_i}{e_i} - a}{\sqrt{a(a+1)b(b+1)}} = \frac{a - a}{\sqrt{a(a+1)b(b+1)}} = 0,
	\end{equation}
	and by symmetry of the inner product $\VNinner{E_a}{E_b} = 0$ for $b > a$ as well. For $a=b$ we have
	\begin{equation}
		\VNinner{E^{[\N-1,1]}_a}{E^{[\N-1,1]}_a} = \frac{\sum_{i = 1}^{a} \VNinner{e_i}{e_i} + a^2}{a(a+1)} = \frac{a + a^2}{a(a+1)} = 1.
	\end{equation}
	Therefore,
	\begin{equation}
		\VNinner{E^{[\N-1,1]}_a}{E^{[\N-1,1]}_b} = \delta_{ab}.
	\end{equation}
	Similar computations give
	\begin{equation}
		\VNinner{E^{[\N]}}{E^{[\N]}} = 1, \quad \VNinner{E^{[\N]}}{E^{[\N-1,1]}_a} = 0.
	\end{equation}
\end{proof}

\section{Tensor powers of defining representation}
In the previous section we decomposed $\VN$ into irreducible representations. For applications to matrix models, we want to consider arbitrary tensor powers
\begin{equation}
	\VN^{\otimes k} \cong \underbrace{\VN \otimes \dots \otimes \VN}_{k}.
\end{equation}
This is a representation of $\SN$ with vector space
\begin{equation}
	\VN^{\otimes k} = \Span(e_{i_1} \otimes e_{i_2} \otimes \dots \otimes e_{i_k} \, \vert \, i_1,\dots,i_k =1,\dots,\N),
\end{equation}
where $\SN$ acts diagonally through permutation matrices
\begin{equation}
	\Paction{\sn}(e_{i_1} \otimes e_{i_2} \otimes \dots \otimes e_{i_k}) = \Paction{\sn}e_{i_1} \otimes \Paction{\sn}e_{i_2} \otimes \dots \otimes \Paction{\sn} e_{i_k}.
\end{equation}

The following special case is of central importance to the matrix models studied in Chapter \ref{chapter: 0d}.
\begin{proposition} \label{cor: VNVN decomp} 
	Assume $\N \geq 4$, then the decomposition of $\VN \otimes \VN$ into irreducible representations is given by
	\begin{equation}
		\VN \otimes \VN \cong 2 V_{[\N]} \oplus 3 V_{[\N{-}1,1]} \oplus V_{[\N{-}2,2]} \oplus V_{[\N{-}2,1,1]}, \label{eq: VNVN decomp}
	\end{equation}
	where $V_{[\N{-}2,2]}, V_{[\N{-}2,1,1]}$ are two irreducible representations with dimensions $\N(\N-3)/2$ and $(\N-1)(\N-2)/2$ respectively.
\end{proposition}
\begin{proof}
	From the decomposition \eqref{eq: VN decomp} we have
	\begin{equation}
		\begin{aligned}
			\VN \otimes \VN &\cong (V_{[\N]} \oplus V_{[\N-1,1]}) \otimes (V_{[\N]} \oplus V_{[\N-1,1]}) \\
			&\cong \begin{aligned}[t]
				&(V_{[\N]} \otimes V_{[\N]}) \oplus (V_{[\N]} \otimes V_{[\N-1,1]}) \oplus (V_{[\N-1,1]} \otimes V_{[\N]} ) \\ &\oplus( V_{[\N-1,1]} \otimes V_{[\N-1,1]}).   
			\end{aligned}
		\end{aligned}
	\end{equation}
	The first three factors decompose into $V_{[\N]}, V_{[\N-1,1]}, V_{[\N-1,1]}$ respectively since $V_{[\N]}$ is the trivial representation. It remains to decompose the last factor, for which we use \cite[Equation 7-167]{Hamermesh1962}
	\begin{equation}
		V_{[\N-1,1]} \otimes V_{[\N-1,1]} \cong V_{[\N]} \oplus V_{[\N-1,1]} \oplus V_{[\N-2,2]} \oplus V_{[\N-2,1,1]}.     
	\end{equation}
	This completes the proof of equation \ref{eq: VNVN decomp}.
\end{proof}

When dealing with large tensor products of vector spaces it is useful to introduce diagram notation. Diagram notation works as follows.
Let $\VN = \Span(e_1, \dots, e_\N)$ be a $\N$-dimensional vector space, and $M \in \End(\VN)$ be a linear map
\begin{equation}
	M(e_i) = \sum_j M^j_i e_j.
\end{equation}
We identify the matrix $M^j_i$ with the following diagram
\begin{equation}
	M^j_i = 	
	\vcenter{\hbox{\tikzset{every picture/.style={line width=0.75pt}} %set default line width to 0.75pt        
			\begin{tikzpicture}[x=0.75pt,y=0.75pt,yscale=-1,xscale=1]
				
				%Straight Lines [id:da3488686735880915] 
				\draw    (90,40) -- (90,50) ;
				%Shape: Rectangle [id:dp5394474967649019] 
				\draw   (80,50) -- (100,50) -- (100,70) -- (80,70) -- cycle ;
				%Straight Lines [id:da7267419297203148] 
				\draw    (90,70) -- (90,80) ;
				
				% Text Node
				\draw (82,52) node [anchor=north west][inner sep=0.75pt]    {$M$};
				% Text Node
				\draw (85,19) node [anchor=north west][inner sep=0.75pt]  [font=\tiny]  {$j$};
				% Text Node
				\draw (85,89) node [anchor=north west][inner sep=0.75pt]  [font=\tiny]  {$i$};
	\end{tikzpicture}}}
\end{equation}
Composition of maps, or index contractions corresponds to connecting edges. For example, matrix multiplication is given by
\begin{equation}
	(MN)^j_i = \sum_k M^j_k N^k_i= \vcenter{\hbox{
			\tikzset{every picture/.style={line width=0.75pt}} %set default line width to 0.75pt        
			
			\begin{tikzpicture}[x=0.75pt,y=0.75pt,yscale=-1,xscale=1]
				
				%Straight Lines [id:da523441929232636] 
				\draw    (240,30) -- (240,40) ;
				%Shape: Rectangle [id:dp12640523026900063] 
				\draw   (230,40) -- (250,40) -- (250,60) -- (230,60) -- cycle ;
				%Straight Lines [id:da6804575959367201] 
				\draw    (240,60) -- (240,70) ;
				%Straight Lines [id:da034525795674739346] 
				\draw    (240,70) -- (240,80) ;
				%Shape: Rectangle [id:dp08502393014111198] 
				\draw   (230,80) -- (250,80) -- (250,100) -- (230,100) -- cycle ;
				%Straight Lines [id:da07408703113809856] 
				\draw    (240,100) -- (240,110) ;
				
				% Text Node
				\draw (232,42) node [anchor=north west][inner sep=0.75pt]    {$M$};
				% Text Node
				\draw (235,9) node [anchor=north west][inner sep=0.75pt]  [font=\tiny]  {$j$};
				% Text Node
				\draw (233,82) node [anchor=north west][inner sep=0.75pt]    {$N$};
				% Text Node
				\draw (235,119) node [anchor=north west][inner sep=0.75pt]  [font=\tiny]  {$i$};
	\end{tikzpicture}}}
\end{equation}
This generalizes to tensors and linear maps $T \in \End(\VN^{\otimes k})$,
\begin{equation}
	T^{j_1 \dots j_k}_{i_1 \dots i_k} = \vcenter{\hbox{

			\tikzset{every picture/.style={line width=0.75pt}} %set default line width to 0.75pt        
			
			\begin{tikzpicture}[x=0.75pt,y=0.75pt,yscale=-1,xscale=1]
				
				%Straight Lines [id:da7995156359156868] 
				\draw    (120,170) -- (120,180) ;
				%Shape: Rectangle [id:dp28967716082538986] 
				\draw   (110,180) -- (150,180) -- (150,200) -- (110,200) -- cycle ;
				%Straight Lines [id:da38030980893926136] 
				\draw    (120,200) -- (120,210) ;
				%Straight Lines [id:da7828357925351779] 
				\draw    (140,170) -- (140,180) ;
				%Straight Lines [id:da6129285790496579] 
				\draw    (140,200) -- (140,210) ;
				
				% Text Node
				\draw (122,183) node [anchor=north west][inner sep=0.75pt]    {$T$};
				% Text Node
				\draw (109,158) node [anchor=north west][inner sep=0.75pt]  [font=\tiny]  {$j_{1}$};
				% Text Node
				\draw (135,158) node [anchor=north west][inner sep=0.75pt]  [font=\tiny]  {$j_{k}$};
				% Text Node
				\draw (109,209) node [anchor=north west][inner sep=0.75pt]  [font=\tiny]  {$i_{1}$};
				% Text Node
				\draw (135,209) node [anchor=north west][inner sep=0.75pt]  [font=\tiny]  {$i_{k}$};
	\end{tikzpicture}}} = \vcenter{\hbox{

\tikzset{every picture/.style={line width=0.75pt}} %set default line width to 0.75pt        

\begin{tikzpicture}[x=0.75pt,y=0.75pt,yscale=-1,xscale=1]

%Straight Lines [id:da3802626826414415] 
\draw    (230,170) -- (230,180) ;
%Shape: Rectangle [id:dp33214632197020855] 
\draw   (210,180) -- (250,180) -- (250,200) -- (210,200) -- cycle ;
%Straight Lines [id:da5257073281312961] 
\draw    (230,200) -- (230,210) ;

% Text Node
\draw (222,183) node [anchor=north west][inner sep=0.75pt]    {$T$};
% Text Node
\draw (209,159) node [anchor=north west][inner sep=0.75pt]  [font=\tiny]  {$j_{1}$};
% Text Node
\draw (235,159) node [anchor=north west][inner sep=0.75pt]  [font=\tiny]  {$j_{k}$};
% Text Node
\draw (209,210) node [anchor=north west][inner sep=0.75pt]  [font=\tiny]  {$i_{1}$};
% Text Node
\draw (235,210) node [anchor=north west][inner sep=0.75pt]  [font=\tiny]  {$i_{k}$};
\end{tikzpicture}} }
\end{equation}
The last diagram is used when keeping track of every individual index is superfluous. We will now consider higher order tensor products of $\VN$.

\subsection{Decomposition of tensor powers.}
The decomposition of general tensor powers of the defining representation can be computed combinatorially in terms of elementary objects.
\begin{theorem}[Decomposition of $\VN^{\otimes k}$]
	Let $m_{k,\N}^{\lambda}$ be the multiplicity of irreducible representations $V_\lambda$ of $\SN$ in
	\begin{equation}
		\VN^{\otimes k} \cong \bigoplus_{\lambda \vdash \N} m^\lambda_{k,\N} V_\lambda. \label{eq: VNotimesk}
	\end{equation}
	Then
	\begin{equation}
		m_{k,\N}^{\lambda}  = \sum_{t = \abs{\lambda^{\#}}}^\N  S(k,t) f^{\lambda/[\N-t]},
	\end{equation}
	where $\lambda^{\#} = [\lambda_2, \dots, \lambda_l]$ is a partition of $\N - \lambda_1$, $S(k,t)$ is the second Stirling number and $f^{\lambda/[\N-t]}$ is the number of standard tableaux with skew-shape $\lambda/[\N-t]$.
\end{theorem}
\begin{proof}
	See \cite[Section 3.1]{Benkart2017}. The Stirling numbers appear in this formula by considering a fixed basis element in $\VN^{\otimes k}$. This defines a set partition of $\{1,\dots, k\}$ by considering which tensor factor contain the same basis element of $\VN$. For example $e_1 \otimes e_2 \otimes e_1 \otimes e_3$ corresponds to the set partition $13|2|4$ with three blocks. The number of such set partitions is the Stirling number $S(4,3)$. This partition structure is preserved by the action of $\SN$ and therefore define subrepresentations of $\VN^{\otimes k}$. However, these subrepresentation turn out to be reducible and the decomposition into irreducible components is proven to be determined by $f^{\lambda / [\N-t]}$.
\end{proof}
\begin{example}
	A useful fact is that $m_{k,\N}^\lambda = 0$ if $\abs{\lambda^{\#}} > k$ because $S(k,t) = 0$ for $t > k$.
\end{example}
\begin{example}\label{ex: k=2 VNVN multi}\ytableausetup{smalltableaux}
	We can verify this theorem in the case of $k=2$ and $\N \geq 4$ using equation \ref{eq: VNVN decomp}. To compute $m^{[\N]}_{2,\N}$ note that $\abs{[\N]^{\#}} = \abs{[]} = 0$ and $S(k,t)$ vanishes for $t > k$ and $t=0$. Therefore
	\begin{align}
		m^{[\N]}_{2,\N} &= S(2,1)f^{[\N]/[\N-1]} + S(2,2)f^{[\N]/[\N-2]} \nonumber \\
		&= f^{[\N]/[\N-1]} + f^{[\N]/[\N-2]},
	\end{align}
	where the last equality follows because $S(2,1) = S(2,2) = 1$.
	The Young diagram $\YT{[\N]/[\N-1]}$ with skew-shape $[\N]/[\N-1]$ is just
	\begin{equation}
		\YT{[\N]/[\N-1]} = \ydiagram{1}
	\end{equation}
	and it has a single standard filling ${\ytableaushort{1}}$. Similarly,
	\begin{equation}
		\YT{[\N]/[\N-2]} = \ydiagram{2}
	\end{equation}
	has a single standard filling $\ytableaushort{1 2}$.
	To compute $m^{[\N-1,1]}_{2,\N}$ we use $\abs{[\N-1,1]^{\#}} = \abs{[1]}=1$ such that
	\begin{equation}
		m^{[\N-1,1]}_{2,\N} =f^{[\N-1,1]/[\N-1]}  + f^{[\N-1,1]/[\N-2]}.
	\end{equation}
	We have the two Young diagrams
	\begin{equation}
		\YT{[\N-1,1]/[\N-1]} = \ydiagram{1} \quad \YT{[\N-1,1]/[\N-2]} = \ydiagram{1+1,1},
	\end{equation}
	with standard fillings
	\begin{equation}
		\ytableaushort{1}, \quad \ytableaushort{\none 1, 2}, \quad \ytableaushort{\none 2, 1}.
	\end{equation}
	That is,
	\begin{equation}
		m^{[\N-1,1]}_{2,\N} = 1 + 2 = 3.
	\end{equation}
	We compute $m^{[\N-2,2]}_{2,\N}$ similarly. We have $\abs{[\N-2,2]^{\#}} = \abs{[2]} = 2$,
	\begin{equation}
		m^{[\N-2,2]}_{2,\N} = S(2,2)f^{[\N-2,2]/[\N-2]} = f^{[\N-2,2]/[\N-2]}.
	\end{equation}
	The relevant Young diagram is
	\begin{equation}
		\YT{[\N-2,2]/[\N-2]} = \ydiagram{2}
	\end{equation}
	with a single standard filling $\ytableaushort{1 2}$ such that $m^{[\N-2,2]}_{2,\N} = 1$.
	Lastly, for $m^{[\N-2,1,1]}_{2,\N}$ we have $\abs{[\N-2,1,1]^{\#}} = \abs{[1,1]} = 2$ and
	\begin{equation}
		m^{[\N-2,1,1]}_{2,\N} = S(2,2)f^{[\N-2,1,1]/[\N-2]},
	\end{equation}
	The Young diagram is
	\begin{equation}
		\YT{{[\N-2,1,1]/[\N-2]}} = \ydiagram{1,1}
	\end{equation}
	with a single standard filling $\ytableaushort{1, 2}$ such that $m^{[\N-2,1,1]}_{2,\N} = 1$.
	This verifies the result in equation \ref{eq: VNVN decomp}.
\end{example}
\begin{example}\label{ex: k=3 VNVNVN multi}\ytableausetup{smalltableaux}
	Another example is $k=3$ with $N \geq 6$. We have
	\begin{align}
		m^{[\N]}_{3,\N} &= S(3,1) f^{\scriptstyle \ydiagram{1}} + S(3,2) f^{\scriptstyle \ydiagram{2} } + S(3,3) f^{\scriptstyle \ydiagram{3} }\\
		m^{[\N-1,1]}_{3,\N} &= S(3,1) f^{\scriptstyle \ydiagram{1}}  + S(3,2) f^{\scriptstyle \ydiagram{1+1,1} }  + S(3,3) f^{\scriptstyle \ydiagram{1+2,1}} \\
		m^{[\N-2,2]}_{3,\N} &= S(3,2) f^{\scriptstyle \ydiagram{2}}  + S(3,3) f^{\scriptstyle \ydiagram{2+1,2}}  \\
		m^{[\N-2,1,1]}_{3,\N} &= S(3,2) f^{\scriptstyle \ydiagram{1,1}} + S(3,3) f^{\scriptstyle \ydiagram{1+1,1,1}}   \\
		m^{[\N-3,3]}_{3,\N} &= S(3,3) f^{\scriptstyle \ydiagram{3}}   \\
		m^{[\N-3,2,1]}_{3,\N} &=  S(3,3) f^{\scriptstyle \ydiagram{2,1}}\\
		m^{[\N-3,1,1,1]}_{3,\N} &=  S(3,3) f^{\scriptstyle \ydiagram{1,1,1}}
	\end{align}
	with $S(3,1) = 1, S(3,2) = 3, S(3,3) = 1$ and
	\begin{align}
		&f^{\scriptstyle \ydiagram{1}} = f^{\scriptstyle \ydiagram{2}} = f^{\scriptstyle \ydiagram{3} } = 1 \\
		&f^{\scriptstyle \ydiagram{1}}  = 1, \,  f^{\scriptstyle \ydiagram{1+1,1}}  = 2, \, f^{\scriptstyle \ydiagram{1+2,1}}  = 3 \\
		&f^{\scriptstyle \ydiagram{2}} = 1, \, f^{\scriptstyle \ydiagram{2+1,2}} = 3, \\
		&f^{\scriptstyle \ydiagram{1,1}} = 1, \, f^{\scriptstyle \ydiagram{1+1,1,1}} = 3\\
		&f^{\scriptstyle \ydiagram{3}} = 1 \\
		&f^{\scriptstyle \ydiagram{2,1}} = 2 \\ 
		&f^{\scriptstyle \ydiagram{1,1,1}} = 1.
	\end{align}
	This gives
		\begin{align}
		m^{[\N]}_{3,\N} &= 5\\
		m^{[\N-1,1]}_{3,\N} &= 10 \\
		m^{[\N-2,2]}_{3,\N} &= 6 \\
		m^{[\N-2,1,1]}_{3,\N} &=   6 \\
		m^{[\N-3,3]}_{3,\N} &= 1  \\
		m^{[\N-3,2,1]}_{3,\N} &=  2\\
		m^{[\N-3,1,1,1]}_{3,\N} &= 1,
	\end{align}
	and
	\begin{equation}
		\VN^{\otimes 3} \cong \begin{aligned}[t]
			5V_{[\N]} &\oplus10 V_{[\N-1,1]} \oplus 6V_{[\N-2,2]} \oplus 6 V_{[\N-2,1,1]} \\
			&\oplus V_{[\N-3,3]} \oplus 2V_{[\N-3,2,1]} \oplus V_{[\N-1,1,1,1]}.
		\end{aligned}\label{eq: VNVNVN decomp}
	\end{equation}
\end{example}

\subsection{Tensor powers from restriction and induction.}
The above theorem is very useful when explicit computation of the multiplicities is necessary. It will be useful to interpret these multiplicities using different combinatorial and representation theoretic structures. We will now see that multiplicities $m^\lambda_{k,\N}$ also have an interpretation in terms of restriction and induction of representations.

First, we recall what it means to restrict representations of $\SN$ to $S_{\N-1}$.
\begin{definition}
	Let $\sn \in \SN, \lambda \vdash \N$ and $\P^\lambda_{ab}(\sn)$ be an irreducible representation of $\SN$ with corresponding vector space $V_\lambda$. The representation $\Res{\SN}{S_{\N-1}}{V_\lambda}$ of $S_{\N-1}$ is given by $\P^\lambda_{ab}(\sn)$ acting on $V_{\lambda}$ for $\sn \in S_{\N-1}$. It is called the restricted representation.
\end{definition}
Induction is a procedure for going in the opposite direction.
\begin{definition}
	Let $\sn \in S_{\N-1}, \lambda \vdash \N-1$, and $\P^\lambda_{ab}(\sn)$ be an irreducible representation of $S_{\N-1}$ acting on the vector space $V_{\lambda}$ with basis $e_a$. Consider the coset $\SN/S_{\N-1}$ with a set $\qty{\sn_1, \dots, \sn_\N}$ of representatives. That is, any set of elements $\sn_1, \dots, \sn_\N \in \SN$ such that
	\begin{equation}
		\SN = \bigcup_{i=1}^{\N} \sn_i S_{\N-1}, \quad \sn_i S_{\N-1} \cap \sn_j S_{\N-1} = \emptyset \qq{if $i\neq j$.}
	\end{equation}
	We define the vector space
	\begin{equation}
		\Ind{S_{\N-1}}{\SN}{V_\lambda} =  \Span\qty(\,  \sn_i \otimes_{S_{\N-1}}  e_a \, \vert \, i=1,\dots,\N, a=1,\dots, \dimSN{\lambda}).
	\end{equation}
	The tensor product symbol $\otimes_{S_{\N-1}}$ means that elements of the group $S_{\N-1}$ can be passed through the tensor product. This new vector space, of dimension $\N \times \dimSN{\lambda}$ is turned into a representation of $\SN$ through the following action.
	Suppose $\sn \in \SN$ and $\sn_i$ are such that $\sn \sn_i = \sn_j \rho$ for $\rho \in S_{\N-1}$. Then the induced representation $\Pi^\lambda: \SN \rightarrow GL(\Ind{S_{\N-1}}{\SN}{V_\lambda})$ of $\SN$ is defined by
	\begin{equation}
		\begin{aligned}
			\Pi^{\lambda}(\sn)\qty(\sn_i \otimes_{S_{\N-1}} e_a) &= (\sn \sn_i) \otimes_{S_{\N-1}} e_a = \sn_j \rho \otimes_{S_{\N-1}} e_a \\
			&= \sn_j \otimes_{S_{\N-1}} \sum_{b} P^\lambda_{ba}(\rho)e_b.
		\end{aligned}
	\end{equation}
	For general vectors, the action is extended linearly.
\end{definition}
\begin{example}	\label{ex: defining from induction}
	The prototypical example of an induced representation comes from inducing the trivial representation $\P_{ab}^{[\N-1]}(\sn) = 1$ of $S_{\N-1}$ to $\SN$. The permutations $\sn_i = (i\N)$ for $i=1,\dots,\N-1$ and $\sn_\N = (1)\dots(\N)$ form a set of representatives of left cosets. With $V_{[\N-1]} = \Span e_0$ we have
	\begin{equation}
		\Ind{S_{\N-1}}{\SN}{V_{[\N-1]}} = \Span(\sigma_i \otimes_{S_{\N-1}}  e_0 \, \vert \, i=1,\dots,\N).
	\end{equation}
	If $\sn \in \SN$ is such that $\sn \sn_i = \sn_j \rho$ for $\rho \in S_{\N-1}$, then
	\begin{equation}
		\Pi^\lambda(\sn) \sn_i \otimes_{S_{\N-1}} e_0 = \sn_j \rho \otimes_{S_{\N-1}} e_0 = \sn_j \otimes_{S_{\N-1}} e_0,
	\end{equation}
	since $\P_{ab}^{[\N-1]}(\rho) = 1$.
	This is the permutation representation associated with the set $\SN/S_{\N-1}$ of cosets. In fact
	\begin{equation}
		\Ind{S_{\N-1}}{\SN}{V_{[\N-1]}} \cong \VN,
	\end{equation}
	where $\VN$ is the defining representation of $\SN$.
\end{example}
The following theorem also helps to build some intuition for induced representations.
\begin{theorem}[Frobenius reciprocity]
	Let $V_{\lambda}$ be an irreducible representation of $\SN$ and $V_{\rho}$ an irreducible representation of $S_{\N-1}$. Suppose we have the following two decompositions with multiplicities $I_\rho^\lambda, R_\lambda^{\rho}$
	\begin{align}
		&\Ind{S_{\N-1}}{\SN}{V_{\rho}} = \bigoplus_{\lambda \vdash \N} I_\rho^\lambda V_{\lambda}, \\
		&\Res{\SN}{S_{\N-1}}{V_\lambda} = \bigoplus_{\rho \vdash \N-1} R_\lambda^{\rho} V_{\rho}.
	\end{align}
	Frobenius reciprocity gives
	\begin{equation}
		I_\rho^\lambda = R_\lambda^\rho.
	\end{equation}
	Frobenius reciprocity gives a precise sense in which induction is the "dual" or adjoint of restriction.
\end{theorem}
\begin{proof}
	This is a standard result proven in most textbooks on representation theory of finite groups, see for example \cite[Theorem 13]{Serre1977}. We give a very rough sketch of an elementary proof of this. The above theorem corresponds to an identity between two inner products of characters of $\SN$ and $S_{\N-1}$ respectively. Schematically, it takes the form
	\begin{equation}
		\langle V_\lambda, \Ind{S_{\N-1}}{\SN}{V_{\rho}} \rangle_{\SN} = \langle V_\rho, \Res{\SN}{S_{\N-1}}{V_\lambda} \rangle_{S_{\N-1}}.
	\end{equation}
	The identity is proven by invoking the definition of the character of an induced representation (see \cite[Theorem 12]{Serre1977}).
\end{proof}

Example \ref{ex: defining from induction} is a special case of a more general relationship.
\begin{theorem} \label{thm: ind res VNotimesk}
	Tensor powers of the defining representation of $\SN$ are isomorphic to iterated induction and restriction of the trivial representation of $\SN$.
	\begin{align}
		\VN &\cong \Ind{S_{\N-1}}{\SN}{\Res{\SN}{S_{\N-1}}{V_{[\N]}}} \\
		\VN^{\otimes 2} &\cong \Ind{S_{\N-1}}{\SN}{\Res{\SN}{S_{\N-1}}{ \Ind{S_{\N-1}}{\SN}{\Res{\SN}{S_{\N-1}}{V_{[\N]}}} }} \\
		&\vdots \\
		\VN^{\otimes k}& \cong (\mathrm{Ind}\, \mathrm{Res})^k(V_{[\N]}), \label{eq: VNk is ind res}
	\end{align}
	where $(\mathrm{Ind}\,\mathrm{Res})^k$ is shorthand for iterated restriction followed by induction.
\end{theorem}
\begin{proof}
	An inductive proof of this is given in \cite[Section 2.3]{Benkart2017}. It uses the identity
	\begin{equation}
		\Ind{S_{\N-1}}{\SN}{X \otimes \Res{\SN}{S_{\N-1}}{Y}} = \Ind{S_{\N-1}}{\SN}{X} \otimes Y,
	\end{equation}
	that can be found in for example \cite[Remark (3) below Theorem 13]{Serre1977}. Choose
	\begin{equation}
		X = \VN^{\otimes k}, \quad Y = \Res{\SN}{S_{\N-1}}{V_{[\N]}},
	\end{equation}
	the identity implies
	\begin{equation}
	\begin{aligned}	
		\Ind{S_{\N-1}}{\SN}{\Res{\SN}{S_{\N-1}}{ \VN^{\otimes k} } }&=  \Ind{S_{\N-1}}{\SN}{ \Res{\SN}{S_{\N-1}}{ \VN^{\otimes k} } \otimes \Res{\SN}{S_{\N-1}}{ V_{[\N]} }  } \\
		&= \VN^{\otimes k} \otimes  \Ind{S_{\N-1}}{\SN}{\Res{\SN}{S_{\N-1}}{ V_{[\N]} }} = \VN^{\otimes k} \otimes \VN = \VN^{\otimes k+1}
	\end{aligned}
	\end{equation}	
	The theorem follows by induction.
\end{proof}

The upside of this formulation of tensor powers is that restriction and induction of symmetric groups can be computed combinatorially using Young diagrams.
\begin{theorem}
	Let $\lambda \vdash \N$, then
	\begin{equation}
		\Res{\SN}{S_{\N-1}}{V_{\lambda}} \cong \bigoplus_{\lambda' \in \lambda - {\scriptstyle \ydiagram{1}}} V_{\lambda'}, \label{eq: RES SN}
	\end{equation}
	and if $\lambda \vdash \N-1$
	\begin{equation}
		\Ind{S_{\N-1}}{\SN}{V_{\lambda}} \cong \bigoplus_{\lambda' \in \lambda + {\scriptstyle \ydiagram{1}}} V_{\lambda'} \label{eq: IND SN},
	\end{equation}
	where $\lambda - {\scriptstyle \ydiagram{1}}$ is the set of Young diagrams obtainable by removing a box from the end of a row of $\lambda$ and $\lambda + {\scriptstyle \ydiagram{1}}$ the set of Young diagrams obtainable by adding a box to the end of a row of $\lambda$.
\end{theorem}
\begin{proof}
	See \cite[Theorem 2.8.3]{Sagan2013} for a proper proof of this. Intuitively, the restriction can be understood by considering $V_\lambda$ as a vector space with basis labelled by standard Young tableaux of shape $\lambda$. The cell of a standard Young tableaux filled with $\N$ can always be removed to give a new standard Young tableaux of shape $\lambda' = \lambda - {\scriptstyle \ydiagram{1}}$. This tableau is a vector in the representation $V_{\lambda'}$. By considering all standard tableaux of shape $\lambda$ we find all possible $\lambda'$ that appear in the restriction. The induction formula follows from Frobenius reciprocity.
\end{proof}
Note that restriction(induction) of irreducible representations of $\SN$ is multiplicity free. Furthermore, it acts linearly on direct sums of irreducible representations
\begin{equation}
	\Res{\SN}{S_{\N-1}}{V_{\lambda} \oplus V_{\lambda'}} = \Res{\SN}{S_{\N-1}}{V_{\lambda}} \oplus \Res{\SN}{S_{\N-1}}{V_{\lambda'}},
\end{equation}
and similarly for induction. Consequently, specifying the irreducible representations in each sequence of restrictions and inductions in \eqref{eq: VNk is ind res} gives a complete description of the multiplicities of the final irreducible representations. This is a theorem, which we will state using the following definition.
\begin{definition}[Vacillating tableau] \label{def: vac tableu}
	Let $\lambda \vdash \N$ with $\abs{\lambda^{\#}} \leq k$. An alternating sequence
	\begin{equation}
		\vactab = (\lambda^{(0)} = [\N], \lambda^{(\frac{1}{2})} = [\N-1], \lambda^{(1)}, \lambda^{(\frac{3}{2})},\dots,\lambda^{(k)}=\lambda)
	\end{equation}
	is called a vacillating tableau of shape $\lambda$ and length $k$ if
	\begin{align}
		\lambda^{(i+\frac{1}{2})} &\in \lambda^{(i)} -{\scriptstyle \ydiagram{1}},\\
		\lambda^{(i+1)} &\in	\lambda^{(i+\frac{1}{2})} + {\scriptstyle \ydiagram{1}},
	\end{align}
	and
	\begin{align}
		\lambda^{(i)} \in \Lambda_{i, \N} &= \{\lambda \vdash \N \, \vert \, \abs*{\lambda^{\#}} \leq i\} \\
		\lambda^{(i+\frac{1}{2})} \in \Lambda_{i+\frac{1}{2}, \N} &= \{\lambda \vdash \N-1 \, \vert \, \abs*{\lambda^{\#}} \leq i\}.
	\end{align}
\end{definition}
\begin{example}
	The following is the set of all vacillating tableaux of length $k=2$.
	\begin{align}
		\lambda=[\N]: &\begin{aligned}[t]
			&\vactab_1 = ([\N], [\N-1], [\N], [\N-1], [\N]),\\
			&\vactab_2 = ([\N], [\N-1], [\N-1,1], [\N-1], [\N])
		\end{aligned} \\
		\lambda = [\N-1,1]:&\begin{aligned}[t]
			&\vactab_1 = ([\N], [\N-1], [\N], [\N-1], [\N-1,1]), \\
			&\vactab_2 = ([\N], [\N-1], [\N-1,1], [\N-2,1], [\N-1,1]), \\
			&\vactab_3 = ([\N], [\N-1], [\N-1,1], [\N-1], [\N-1,1])
		\end{aligned} \\
		\lambda = [\N-2,2]:&\begin{aligned}[t]
			&\vactab_1 = ([\N], [\N-1], [\N-1,1], [\N-2,1], [\N-2,2]),
		\end{aligned}\\
		\lambda = [\N-2,1,1]:&\begin{aligned}[t]
			&\vactab_1 = ([\N], [\N-1], [\N-1,1], [\N-2,1], [\N-2,1,1]).
		\end{aligned}
	\end{align}
\end{example}
Observe that the number of vacillating tableaux of shape $\lambda$ and length $k=2$ is equal to $m^\lambda_{2,\N}$ as computed in Example \ref{ex: k=2 VNVN multi}.
The general theorem is the following.
\begin{theorem} \label{thm: multi is vac tab}
	Let $\lambda \in \Lambda_{k,\N}$, then
	\begin{equation}
		m^\lambda_{k,\N} = \abs{\{\text{vacillating tableaux of shape $\lambda$ and length $k$}\}}.
	\end{equation}
\end{theorem}
\begin{proof}
	This is proven in for example \cite[Section 2.3]{Benkart2017}.
\end{proof}
Therefore the multiplicities $\dimPk{\lambda}$ in the decomposition of $\VN^{\otimes k}$ can be understood in terms of counting of vacillating tableaux of shape $\lambda$ and length $k$. This structure behind the tensor powers will be crucial for the constructions in the next chapter.

\section{Summary}
In this chapter we reviewed some basic aspects of symmetric groups and their representation theory. We focused on the defining representation $\VN$, where the symmetric group is represented by permutation matrices. This representation will be directly relevant to the symmetry in permutation invariant matrix models. We saw that the defining representation of $\SN$ can be constructed by inducing the trivial representation of $S_{\N-1}$. This was a special case of a more general construction, that of the tensor powers $\VN^{\otimes k}$ from repeated restriction and induction. This repeated restriction and induction could be understood from the combinatorial rules \eqref{eq: RES SN}, \eqref{eq: IND SN} of removing and adding boxes to a Young diagram. A very important theorem was the relationship between the sequence of restrictions and inductions, called vacillating tableaux, and multiplicities in the decomposition of tensor powers $\VN^{\otimes k}$ into irreducible representations. This structure will play an important role in understanding the inductive chain of partition algebras and construction of matrix units, which we review in the next chapter. In matrix models, tensor powers with $k=1,2$ will be essential. They correspond to linear and quadratic parts of the action, respectively. 

\chapter{Partition algebras}\label{chapter: partition algebra}
The group algebra $\mathbb{C}(S_k)$ fits into a class of finite associative algebras known as diagram algebras. These are algebras with distinguished bases labelled by diagrams. In the diagram basis the algebra product takes the form of diagram concatenation. For example, consider $S_3$ where
\begin{equation}
	(12) \leftrightarrow \PAdiagram{3}{-1/2,-2/1,-3/3}{}, \quad (123) \leftrightarrow \PAdiagram{3}{-1/2,-2/3,-3/1}{},
\end{equation}
and
\begin{equation}
	(12)(123) = \begin{aligned}
		&\PAdiagram{3}{-1/2,-2/3,-3/1}{}\\
		&\PAdiagram{3}{-1/2,-2/1,-3/3}{}
	\end{aligned} = \PAdiagram{3}{-1/3,-2/2,-3/1}{} = (13).
\end{equation}

As mentioned in the introduction, diagram algebras feature in the mathematical duality known as (generalized) Schur-Weyl duality. The classical instance of Schur-Weyl duality relates representation theory of $GL(\N)$ to representation theory of $\mathbb{C}(S_k)$ \cite{Fulton2013}. In particular, consider the diagonal action of $g \in GL(\N)$ on $\VN^{{\otimes k}}$
\begin{equation}
	g(v_1 \otimes v_2 \otimes \dots v_k) = gv_1 \otimes gv_2 \otimes \dots \otimes g v_k.
\end{equation}
Since the same $g \in GL(\N)$ acts on all tensor factors, it commutes with the action of $\mathbb{C}(S_k)$ permuting the order of tensor factors. Schur-Weyl duality says that this action of $\mathbb{C}(S_k)$ is the complete set of elements in $\End_{GL(\N)}(\VN^{\otimes k}) \subset \End(\VN^{\otimes k})$ that commute with $GL(\N)$. That is,
\begin{equation}
	\End_{GL(\N)}(\VN^{\otimes k}) \cong \mathbb{C}(S_k),
\end{equation}
as long as $\N > k$.
As we will review, the double centralizer theorem then implies that
\begin{equation}
	\VN^{\otimes k} \cong \bigoplus_{\substack{\lambda \vdash k \\ l(\lambda) \leq \N}} W_\lambda \otimes V_\lambda,
\end{equation}
where $W_\lambda$ are irreducible representations of $GL(\N)$ and $V_\lambda$ irreducible representations of $S_k$. Note that this decomposition is multiplicity free, or equivalently, there is a one-to-one correspondence between irreducible representations $W_\lambda$ and $V_\lambda$. This implies that the multiplicity of $W_\lambda$ in the decomposition is equal to the dimension of $V_\lambda$ -- that is $\abs{\SYT{\lambda}}$. Therefore, Schur-Weyl duality serves as a powerful tool for dealing with the multiplicities in this tensor product decomposition.

Schur-Weyl duality has been generalized in various directions, including other matrix groups such as $O(\N)$ and $Sp(\N)$ \cite{Brauer} . It has also been generalized to other actions and vector spaces. In these cases, more general diagram algebras replace $\mathbb{C}(S_k)$, such as Brauer algebras, which have played a role in large $\N$ physics \cite{Kimura2007}.

In this thesis we are concerned with large $\N$ models with permutation symmetry. As we will review in detail in chapter \ref{chapter: 0d}, such models are described by a quadratic polynomial in matrices $X = \vert \vert X_{ij} \vert \vert$. For example,
\begin{equation}
	\sum_{i,j,k,l=1}^N G^{ij;kl}X_{ij}X_{kl},
\end{equation}
where the tensor $G^{ij;kl}$ parametrises the contribution to the action. The model is permutation invariant if and only if
\begin{equation}
	G^{ij;kl} = G^{(i)\sigma (j){\sigma}; (k)\sigma (l)\sigma} \quad \forall \sigma \in \SN.
\end{equation}
Invariant tensors like these correspond to elements of $\End_{\SN}(\VN^{\otimes 2})$. This is how Schur-Weyl duality appears in the context of matrix models.

Fortunately, Schur-Weyl duality has been developed for tensor powers of defining representations of $\SN$. In this case, the dual algebras $\End_{{\SN}}(\VN^{\otimes k})$ are isomorphic to the partition algebras $P_k(\N)$. They were first discovered in the context of statistical mechanics and Potts models \cite{Martin1994, Jones1994}, where they appear as the algebra generated by transfer matrices. The study of these diagram algebras has seen a lot of recent progress \cite{Halverson2005, Benkart2017, Halverson2018, Enyang_2012} (this is by no means an exhaustive list). 

In this chapter, we will use these results to construct a basis of matrix units for $P_1(\N), P_2(\N)$. This basis gives an efficient and explicit parametrisations of the permutation invariant Gaussian matrix models, or equivalently the tensors $G^{ij;kl}$.

\section{Schur-Weyl duality} \label{sec: partition algebras}
In this section we will investigate the set of linear maps on $\VN^{\otimes k}$ that commute with $\SN$. We will give a basis for the vector space of such maps and explain its relevance to the decomposition of $\VN^{\otimes k}$ through the double centralizer theorem. We will then explain how the partition algebras enter the picture and explore its inductive structure.

Let $\End(\VN^{\otimes k})$ be the set of linear maps $T: \VN^{\otimes k} \rightarrow \VN^{\otimes k}$. It is isomorphic to
\begin{equation}
	\End(\VN^{\otimes k}) \cong \VN^{\otimes k} \otimes (\VN^*)^{\otimes k}
\end{equation}
with  $\VN^*$ the dual space
\begin{equation}
	\VN^{*} = \Span(f_i \, \vert \, i=1,\dots,\N \text{ with } f_i(e_j) = \delta_{ij} ).
\end{equation}
We have
\begin{equation}
	f_i(e_j) = \VNinner{e_i}{e_j},
\end{equation}
and the map $f_i \mapsto e_i$ is an isomorphism of vector spaces.

In fact, the two vector spaces are isomorphic as representations of $\SN$.
\begin{definition}(Dual representation)
	We define the representation $\Paction{\sn}^*$ of $\sn \in \SN$ on $\VN^*$ by
	\begin{equation}
		(\Paction{\sn}^*f_i)(e_j) = f_i(\Paction{\sn^{-1}}e_j), \label{eq: dual VN rep}
	\end{equation}
	such that
	\begin{equation}
		(\Paction{\sn_1 \sn_2}^*f_i)(e_j) = f_i(\Paction{\sn_2^{-1}} \Paction{\sn_1^{-1}}e_j) = (\Paction{\sn_2}^*f_i)( \Paction{\sn_1^{-1}}e_j) =  (\Paction{\sn_1}^*\Paction{\sn_2}^*f_i)(e_j).
	\end{equation}
	In coordinates \eqref{eq: dual VN rep} reads
	\begin{equation}
		(\Paction{\sn}^*)_i^j = (\Paction{\sn^{-1}})_j^i.
	\end{equation}
\end{definition}
\begin{corollary}
	Because the defining representation is orthogonal (real unitary) $(\Paction{\sn^{-1}})_j^i = (\Paction{\sn})^j_i$. Therefore, $\VN^* \cong \VN$ as representations of $\SN$ and by extension $(\VN^*)^{\otimes k} \cong \VN^{\otimes k}$.
\end{corollary}
This also implies that $\End(\VN^{\otimes k}) \cong \VN^{\otimes 2k}$. This correspondence is frequently used in matrix model computations.

We are interested in the subspace of $\End(\VN^{\otimes k})$ that commute with $\SN$.
\begin{definition}
The $\SN$ invariant subspace of $\End(\VN^{\otimes k})$ is defined as
\begin{equation}
	\End_{\SN}(\VN^{\otimes k}) = \{T \in \End(\VN^{\otimes k}) \text{ with } T\Paction{\sn} = \Paction{\sn}T \quad \forall \sn \in \SN\}.
\end{equation}
\end{definition}
This vector space is in fact an algebra with multiplication given by composition of maps.
\begin{remark}
It is instructive to look at the constraint $T\Paction{\sn} = \Paction{\sn}T$ in components. We re-write this as $\Paction{\sn^{-1}}T\Paction{\sn} = T$, which in components reads
\begin{equation}
	T_{(i_1)\sn (i_2)\sn {\dots} (i_k)\sn }^{(j_1)\sn (j_2)\sn {\dots} (j_k)\sn } = T_{i_1 i_2 \dots i_k}^{j_1 j_2 \dots j_k} \quad \forall \sn \in \SN. \label{eq: SN end constraint}
\end{equation}
\end{remark}
A basis for $\End_{\SN}(\VN^{\otimes k}) $ can be constructed by considering orbits of $\SN$ on a basis for $\End(\VN^{\otimes k})$. A basis for $\End(\VN^{\otimes k})$ is simply the set of rank $2k$ tensors that vanish for all components except one. For example, consider $k=3$ and focus on the tensor whose only non-zero component is
\begin{equation}
	T^{112}_{123}.
\end{equation}
We construct the $\SN$ orbit
\begin{equation}
	\{T^{(1)\sn(1)\sn(2)\sn}_{(1)\sn(2)\sn(3)\sn} \, \vert \, \forall \sn \in \SN\}.
\end{equation}
Because $\sn \in \SN$ is a bijection, it preserves relations such as $i\neq j$ or $i=j$. Therefore, the above set is equal to the set
\begin{equation}
	\{T^{iij}_{ijk} \, \vert \,i,j,k\in\{1,\dots,\N\},\, i\neq j, i\neq k, j\neq k\}. \label{eq: tensor orbits}
\end{equation}
By considering all such basis tensors, we find a set of distinct orbits.
The set of distinct orbits label a basis for $\End_{S_N}(\VN^{\otimes 3})$.

\begin{example}\label{ex: End S2 V2}
As an example, we can consider $\End_{S_2}(V_2)$. A basis for $\End(V_2)$ is given by the elementary 2-by-2 matrices $\{E_{11}, E_{12}, E_{21}, E_{22}\}$. They form two distinct orbits under the action of $S_2$,
\begin{equation}
	O_1 = \{E_{11}, E_{22}\}, \quad O_2 = \{E_{12}, E_{21}\}.
\end{equation}
The tensors
\begin{equation}
	(O_1)^{i}_j = (E_{{11}})^i_j + (E_{22})^i_j, \quad (O_2)^{i}_j = (E_{{12}})^i_j + (E_{21})^i_j,
\end{equation}
are invariant since, for example,
\begin{equation}
	(O_1)^{(i)\sigma}_{(j)\sigma} = \delta_{1}^{(i)\sigma} \delta^1_{(j)\sigma}+ \delta_{2}^{(i)\sigma} \delta^2_{(j)\sigma} = \delta_{(1)\sigma^{-1}}^i \delta^{(1)\sigma^{-1}}_j + \delta_{(2)\sigma^{-1}}^i \delta^{(2)\sigma^{-1}}_j = (O_1)^i_j
\end{equation}
for all $\sigma \in S_2$ and form a basis for $\End_{S_2}(V_2)$.
\end{example}

The orbits $O_1, O_2$ are naturally thought of as set partitions of the set $\{i,j\}$. Since the indices of the elementary matrices in $O_1$ are equal, we think of this as the set partition $\{i,j\} = \{i,j\}$. On the other hand, the indices are distinct in the orbit $O_2$ and we think of this as corresponding to the set partition $\{i\} \dot\cup \{j\} = \{i,j\}$.
We will now give a basis for $\End_{\SN}(\VN^{\otimes k})$ based on these observations. To state the theorem, we need the following definition.
\begin{definition}[Set partition]
	Let $S$ be a set of order $k$. A set $\pi = \{\pi_1, \dots, \pi_b\}$ of non-empty subsets of $S$ is called a set partition of $S$  if
	\begin{equation}
		\pi_1 \dot\cup \pi_2 \dot\cup \dots \dot\cup \pi_b = S,
	\end{equation}
	where $\dot\cup$ is meant to emphasize that the subsets in $\pi$ are disjoint. The set of all set partitions is denoted $\setpart{S}$. A subset in a set partition is called a block and we write $\abs{\pi} = b$ for the number of blocks.
\end{definition}
\begin{example}
	As an example consider $S=\{1,2,3,4\}$. There are $15$ set partitions $\pi$ of $S$, which we write down in the following notation $\pi = \pi_1 | \pi_2 | \dots | \pi_b$.
	\begin{align}
		&1234, \quad 12|34, \quad 13|24, \quad 14|23, \quad 123|4, \quad 124|3, \quad 134|2, \quad 234|1 \\
		&1|2|34, \quad 1|3|24, \quad 1|4|23, \quad 2|3|14, \quad 2|4|13, \quad 3|4|12, \quad 1|2|3|4.
	\end{align}
\end{example}

Every orbit of the type in \eqref{eq: tensor orbits} corresponds to a set partition of the $6$ indices on $T^{j_1 j_2 j_3}_{i_1 i_2 i_3}$. For example, the orbit in \eqref{eq: tensor orbits} corresponds to $i_1 j_1 j_2 | i_2 j_3 | i_3$. However, the orbits where the number of distinct indices is larger than $\N$ are empty. Generalizing these observations to general $k$ gives the following theorem
\begin{theorem}[Orbit basis]\label{thm: orbit basis}
	The space $\End_{\SN}(\VN^{\otimes k})$ of $\SN$ equivariant maps $\VN^{\otimes k} \rightarrow \VN^{\otimes k}$ has a basis labelled by set partitions $\pi \in \setpart{\{1,\dots,k,1',\dots,k'\}}$ with $\abs{\pi} \leq \N$.
	We use the short-hand
	\begin{equation}
		[k\vert k'] = \{1,\dots,k,1',\dots,k'\}
	\end{equation}
	and
	\begin{equation}
		\setpart{[k\vert k'], \N} = \qty{\pi \in \setpart{[k\vert k']}, \,\, \abs{\pi} \leq \N  }.
	\end{equation}
	For $\pi \in \setpart{[k\vert k'], \N}$ we define $X_\pi \in \End_{\SN}(\VN^{\otimes k})$ by
	\begin{equation}
		\X_\pi(e_{i_1} {\otimes} \dots \otimes e_{i_k}) = (X_\pi)^{i_{1'} \dots i_{k'}}_{i_1 \dots i_k} e_{i_{1'}} {\otimes} \dots \otimes e_{i_{k'}},
	\end{equation}
	where
	\begin{equation} \label{eq: orbit basis action}
		(\X_\pi)^{i_{1'} \dots i_{k'}}_{i_1 \dots i_k} =  \begin{cases}
			1 \qq{if $i_a = i_b$ if and only if $a$ and $b$ are in the same block of $\pi$,} \\
			0 \qq{otherwise.}
		\end{cases}
	\end{equation}
	The set of maps $\X_\pi$ for $\pi \in \setpart{[k \vert k'], \N}$ form the so called orbit basis for $\End_{\SN}(\VN^{\otimes k})$.
\end{theorem}
\begin{proof}
	See \cite[Section §1]{Jones1994} or \cite[Theorem 5.4]{Benkart2017}.
\end{proof}
A great feature of this basis is that it works for all $\N$. A downside is that the composition of maps is inconvenient to state in this basis.

For $\N \geq 2k$ there exists a basis where multiplication has a beautiful combinatorial description
\begin{theorem}[Diagram basis]
	Let $\pi \in \setpart{[k\vert k']}$ and $\N \geq 2k$. Then the following set of maps form a basis of $\End_{\SN}(\VN^{\otimes k})$
	\begin{equation}
		(\dgram_{\pi})^{i_{1'} \dots i_{k'}}_{i_1 \dots i_k} =  \begin{cases}
		1 \qq{if $i_a = i_b$ when $a$ and $b$ are in the same block of $\pi$,} \\
		0 \qq{otherwise.}
	\end{cases}
	\end{equation}
\end{theorem}
\begin{proof}
	 See \cite[Section §1]{Jones1994}. We will give formulas relating these two bases in Chapter \ref{ch: 1d}.
\end{proof}
This basis has very simple expression in terms of Kronecker deltas
\begin{example}\label{ex: kron deltas}
	Let $\pi = 123|1'|2'3'$ then
	\begin{equation}
		(\dgram_\pi)^{i_{1'} i_{2'} i_{3'}}_{i_1 i_2 i_3}  = \delta_{i_1 i_2} \delta_{i_2 i_3} \delta^{i_{2'} i_{3'}}.
	\end{equation}
\end{example}
We will give a diagrammatic description of composition in this basis shortly. First, we want to elaborate on why this algebra is important to study.

The motivation for studying the algebra $\End_{\SN}(\VN^{\otimes k})$ is the double centralizer theorem.
\begin{theorem}
	The double centralizer theorem gives a multiplicity free decomposition of the tensor power $\VN^{\otimes k}$. We write
	\begin{equation}
		\VN^{\otimes k} \cong \bigoplus_{\lambda \vdash \N}V_{\lambda} \otimes M_{\lambda},
		\label{eq: VN SW}
	\end{equation}
	where $M_{\lambda}$ are irreducible representations of $\End_{\SN}(\VN^{\otimes k})$.
\end{theorem}
\begin{proof}
	See \cite[Theorem 4.54]{Etingof09} for a proof of the double centralizer theorem. The idea is that $\End_{\SN}(\VN^{\otimes k})$ is the full centralizer of the image of $\SN$ in $\End(\VN^{\otimes k})$ and vice versa. In particular, every element in $\End(\VN^{\otimes k})$ commutes with every element in the image of $\SN$. Therefore, there should exist a basis where this is manifest. Proving the exact form of the decomposition in the theorem requires more care.
\end{proof}

In fact, from the restriction and induction description of tensor powers \eqref{eq: VNk is ind res} we know that above decomposition is more constrained,
\begin{equation}
	\VN^{\otimes k} \cong \bigoplus_{l=0}^k \bigoplus_{\lambda^{\#}} V_{[\N-l,\lambda^{\#}]} \otimes M_{[\N-l, \lambda{\#}]}.
\end{equation}
The Young diagram $\lambda = [\N-l,\lambda^\#]$, which is an integer partition of $\N$, is constructed by placing the diagram $\lambda^\#$ (having $l$ boxes) below a first row of $\N-l$ boxes. This follows from the fact that each repetition of restriction and induction can move at most one box below the first row.
Requiring $\lambda $ to be a valid Young diagram imposes a condition on the first row length of   $ r_1(  \lambda^{\#} )  \le  \N- l $. For $\N < 2k$ this is a non-trivial constraint, while it is trivially satisfied for all  $\lambda^{\#} $ having up to $k$ boxes  for $ \N \ge 2k$. The latter is called the stable limit.

\subsection{$\SN$ Schur-Weyl duality.}
We now reach the first main theorem of this section, which is the fact that $\End_{\SN}(\VN^{\otimes k}) \cong P_k(\N)$ for $\N \geq 2k$. As mentioned, this is a diagram algebra and the best way to introduce it is to introduce partition diagrams.
\begin{definition}[Partition diagram]
	A partition diagram is a graphical representation of a set partition $\pi \in \setpart{[k \vert k']}$. Consider an undirected graph $d$ with $k$ vertices labelled by $\{1,\dots,k\}$ and another set of $k$ vertices labelled by $\{1',\dots,k'\}$. We allow at most one edge connecting a pair of vertices and disallow edges connected to the same vertex (loops). Every such graph $d$ corresponds to a set partition $\pi$ through the following map:
	\begin{equation}
		d \mapsto \pi
	\end{equation}
	where $\pi$ is a set partition where $a$ is in the same subset as $b$ if they are connected by an edge in $d$. This map is surjective but not bijective -- two different graphs can correspond to the same set partition. For convenience these graphs are drawn with the first $k$ vertices ordered horizontally in a row below the second set of $k$ vertices (see example below).
\end{definition}
\begin{example}
	Many of these aspects can be seen already for $k=2$. First, consider a set of diagrams and their corresponding set partitions
	\begin{align}
		&\PAdiagramLabeled{2}{-1/1,-2/2}{} \mapsto 11'|22', \quad 	\PAdiagramLabeled{2}{-1/1}{} \mapsto 11'|2|2', \\
		&\PAdiagramLabeled{2}{-1/1,-1/-2,-2/2}{} \mapsto 11'22', \quad 	\PAdiagramLabeled{2}{-1/1,1/2,-2/2}{} \mapsto 11'22'.
	\end{align}
	The last two examples show that two different diagrams can correspond to the same set partition.
\end{example}
We think of two partition diagrams that correspond to the same set partition as equivalent. When we use the term the set of partition diagrams, we mean the set of equivalence classes under this identification. We write $d_\pi$ for any representative diagram that corresponds to the set partition $\pi$.

It is often useful to consider the reflection of a partition diagram.
\begin{definition}[Transpose diagram]\label{def: transpose diagram}
Let $d_\pi$ be a partition diagram corresponding to the set partition $\pi$. Then the transpose $d_\pi^T$ is the diagram corresponding to the set partition $\pi^T$ where all unprimed integers are turned into primed integers and all primed integers are turned into unprimed integers. Geometrically this is a horizontal reflection of the diagram $d_\pi$.
\end{definition}
\begin{definition}(Diagram basis)
Consider the vector space
\begin{equation}
	P_k = \Span(d_\pi \, \vert \, \pi \in \setpart{[k\vert k']}).
\end{equation}
We call this the diagram basis for $P_k$.
It has finite dimension $\dim P_k = B(2k)$. The Bell number $B(2k)$ is the number of possible partitions of a set with $2k$ distinct elements and can be computed using the generating function
\begin{equation}
	\sum_{k=0}^\infty \frac{B(k)}{k!}x^k = e^{e^x -1},
\end{equation}
from which one finds $B(2k) = 2, 15, 203, 4140$ for $k=1,2,3,4$.
\end{definition}
The vector space $P_k$ is turned into an algebra $P_k(\N)$ as follows.
\begin{definition}[Partition algebra]\label{def: partition algebra}
The partition algebras $P_k(N)$ for $k=1,2,\dots$ are vector spaces $P_k$ with multiplication defined through diagram concatenation (in the diagram basis). Let $d_{\pi}$ and $d_{\pi'}$ be two diagrams in $P_k$.
The composition $d_{\pi''} = d_{\pi} d_{\pi'}$ is constructed by placing $d_{\pi}$ above $d_{\pi'}$ and identifying the bottom vertices of $d_{\pi}$ with the top vertices of $d_{\pi'}$. The diagram is simplified by following the edges connecting the bottom vertices of $d_{\pi'}$ to the top vertices of $d_{\pi}$. Any connected components within  the middle rows are removed and we multiply by $\N^c$, where $c$ is the number of these complete blocks removed. Multiplication is independent of the choice of representative diagram. For linear combinations of diagrams, multiplication is defined by linear extension.
\end{definition}
\begin{example}
For example,
\begin{equation}
	\begin{aligned}
		&\PAdiagram{3}{-3/2}{-1/-2} \\
		&\PAdiagram{3}{-3/3}{2/1,-2/-3}
	\end{aligned} = \N \PAdiagram{3}{-3/2}{-2/-3}
	\qq{and} 
	\begin{aligned}
		&\PAdiagram{3}{-1/1,-2/3,-3/2}{} \\
		&\PAdiagram{3}{-1/2,-3/3}{-2/-3}
	\end{aligned} = \PAdiagram{3}{-1/3,-3/2}{-2/-3},
\end{equation}
where the factor of $\N$ in the first equation comes from removing the middle component at vertex $1$ and $2$.
\end{example}

Partition algebras allow us to study $\VN^{\otimes k}$ and $\End_{S_N}(\VN^{\otimes k})$ concretely, using Schur-Weyl duality.
\begin{theorem}[$\SN$ Schur-Weyl duality] \label{thm: VN Schur-Weyl}
	The map
	\begin{align}
		\phi_k: \, &P_k(\N) \longrightarrow \End_{\SN}(\VN^{\otimes k}) \\
		&d_\pi \mapsto (\dgram_{\pi})^{i_{1'} \dots i_{k'}}_{i_1 \dots i_k},
	\end{align}
	is a surjective algebra homomorphism.
\end{theorem}
\begin{proof}
	The Kronecker delta representation given in Example \ref{ex: kron deltas} exactly reproduces the diagram multiplication described above. However, outside the stable limit $P_k(\N)$ is larger than $\End_{S_N}(\VN^{\otimes k})$ and one will find non-trivial linear dependence among the elements in the image of $\phi_k$. We refer to \cite{Jones1994} for a detailed proof of the theorem.
\end{proof}
For most purposes in this thesis, we will consider the stable limit, where the above homomorphism becomes an isomorphism.  In other words, the map that associates with each diagram a product of Kronecker deltas for every edge is a realisation of $P_k(\N)$.
\begin{theorem}
	 In the stable limit we can write the decomposition \eqref{eq: VN SW} in terms of $\SN \times P_k(\N)$ representations
	\begin{align} \label{eq: VN SW simple}
		\VN^{\otimes k} = \bigoplus_{\lambda  \in \Lambda_{k,\N} } V_{\lambda} \otimes Z_{\lambda},
	\end{align}
	where
	\begin{equation}
		\Lambda_{k,\N} = \{\lambda \vdash \N \, \vert \, \abs*{\lambda^{\#}} \leq k\},
	\end{equation}
	is a labelling set of irreducible representations $Z_{\lambda}$ of $P_k(\N)$.
\end{theorem}
\begin{proof}
	This follows from the isomorphism discussed in the previous theorem, together with the double centralizer theorem, see \cite[Section 2.5]{Halverson2018}.
\end{proof}
This implies that $\dim Z_\lambda = \dimPk{\lambda}$.

We take the r.h.s. of \eqref{eq: VN SW simple} to have an orthonormal basis
\begin{equation}
	\{E^\lambda_a \otimes E^\lambda_\alpha \, \vert \, a=1,\dots, \dimSN{\lambda}, \alpha =1,\dots, m^\lambda_{k,\N}, \lambda \in \Lambda_{k,\N} \}. \label{eq: VN SW basis}
\end{equation}
with respect to the inner product
\begin{equation}
	(e_{i_{1}} \otimes {\dots} \otimes e_{i_{k}}, e_{i_{1'}} \otimes {\dots} \otimes e_{i_{k'}}) = \begin{cases}
		1 \qq{if $i_a = i_{a'}$ for all $a=1,\dots,k$}\\
		0 \qq{otherwise.}
	\end{cases} \label{eq: VNk inner prod}
\end{equation}
They are related to the l.h.s. by Clebsch-Gordan coefficients
\begin{equation}
	E^\lambda_a \otimes E^\lambda_\alpha = (C^{\lambda}_{a,\alpha})^{i_1 i_2 \dots i_k} e_{i_1} \otimes \dots \otimes e_{i_k}, \label{eq: VNk clebsch}
\end{equation}
or inversely
\begin{equation}
	e_{i_1}^{} \otimes \dots \otimes e_{i_k}^{} = (C_\lambda^{a, \alpha})_{i_1 \dots i_k}E_a^\lambda \otimes E_\alpha^\lambda.
\end{equation}
It will be useful to introduce diagram notation for these equations. We write the Clebsch-Gordan coefficients in \eqref{eq: VNk clebsch} as
\begin{equation}
	(C^{\lambda}_{a,\alpha})^{i_1 i_2 \dots i_k} = \vcenter{\hbox{
			\tikzset{every picture/.style={line width=0.75pt}} %set default line width to 0.75pt        
			
			\begin{tikzpicture}[x=0.75pt,y=0.75pt,yscale=-1,xscale=1]
				
				%Straight Lines [id:da3802626826414415] 
				\draw    (230,170) -- (230,180) ;
				%Shape: Rectangle [id:dp33214632197020855] 
				\draw   (210,180) -- (250,180) -- (250,200) -- (210,200) -- cycle ;
				%Straight Lines [id:da6508722800423687] 
				\draw    (220,200) -- (220,210) ;
				%Straight Lines [id:da21541716028566138] 
				\draw    (240,200) -- (240,210) ;
				
				% Text Node
				\draw (222,183) node [anchor=north west][inner sep=0.75pt]    {$C^{\lambda }$};
				% Text Node
				\draw (212,159) node [anchor=north west][inner sep=0.75pt]  [font=\tiny]  {$i_{1}$};
				% Text Node
				\draw (239,159) node [anchor=north west][inner sep=0.75pt]  [font=\tiny]  {$i_{k}$};
				% Text Node
				\draw (212,209) node [anchor=north west][inner sep=0.75pt]  [font=\tiny]  {$a$};
				% Text Node
				\draw (239,209) node [anchor=north west][inner sep=0.75pt]  [font=\tiny]  {$\alpha $};
	\end{tikzpicture}} }
\end{equation}

As matrix elements of equivariant maps, they have the property
\begin{align}
	(C^{\lambda}_{a,\alpha})^{i_1 i_2 \dots i_k} D_\pi(e_{i_1} \otimes \dots \otimes e_{i_k} )= E^\lambda_a \otimes D^\lambda_{\beta \alpha }(d_\pi)E^\lambda_\beta \label{eq: VNk clebsch PkN equivariance}\\
	(C^{\lambda}_{a,\alpha})^{i_1 i_2 \dots i_k} \P_\sn(e_{i_1} \otimes \dots \otimes e_{i_k} )= \P^\lambda_{ba}(\sn)E^\lambda_b \otimes E^\lambda_\beta, \label{eq: VNk clebsch SN equivariance}
\end{align}
where $\P^\lambda(\sn), D^\lambda(d_\pi)$ are irreducible representations of $\SN$ and $P_k(N)$ respectively.
Diagrammatically, equivariance takes the form
\begin{equation}
	\vcenter{\hbox{

			\tikzset{every picture/.style={line width=0.75pt}} %set default line width to 0.75pt        
			
			\begin{tikzpicture}[x=0.75pt,y=0.75pt,yscale=-1,xscale=1]
				%uncomment if require: \path (0,345); %set diagram left start at 0, and has height of 345
				
				%Straight Lines [id:da7363746409995258] 
				\draw    (400,170) -- (400,180) ;
				%Shape: Rectangle [id:dp8907220615362075] 
				\draw   (380,180) -- (420,180) -- (420,200) -- (380,200) -- cycle ;
				%Straight Lines [id:da4340715667243278] 
				\draw    (390,200) -- (390,210) ;
				%Straight Lines [id:da14509578712334403] 
				\draw    (410,200) -- (410,210) ;
				%Shape: Rectangle [id:dp9195098672807362] 
				\draw   (380,150) -- (420,150) -- (420,170) -- (380,170) -- cycle ;
				%Straight Lines [id:da760463546897328] 
				\draw    (400,140) -- (400,150) ;
				
				% Text Node
				\draw (389,180) node [anchor=north west][inner sep=0.75pt]    {$C^{\lambda }$};
				% Text Node
				\draw (379,209) node [anchor=north west][inner sep=0.75pt]  [font=\tiny]  {$a$};
				% Text Node
				\draw (405,209) node [anchor=north west][inner sep=0.75pt]  [font=\tiny]  {$\alpha $};
				% Text Node
				\draw (390,152) node [anchor=north west][inner sep=0.75pt]    {$D_{\pi }$};
				% Text Node
				\draw (379,129) node [anchor=north west][inner sep=0.75pt]  [font=\tiny]  {$i_{1}$};
				% Text Node
				\draw (405,129) node [anchor=north west][inner sep=0.75pt]  [font=\tiny]  {$i_{k}$};

	\end{tikzpicture}}} = \vcenter{\hbox{

\tikzset{every picture/.style={line width=0.75pt}} %set default line width to 0.75pt        

\begin{tikzpicture}[x=0.75pt,y=0.75pt,yscale=-1,xscale=1]
%uncomment if require: \path (0,345); %set diagram left start at 0, and has height of 345

%Straight Lines [id:da645586128362952] 
\draw    (470,140) -- (470,150) ;
%Shape: Rectangle [id:dp0571070553472115] 
\draw   (450,150) -- (490,150) -- (490,170) -- (450,170) -- cycle ;
%Straight Lines [id:da14271087168042462] 
\draw    (460,170) -- (460,210) ;
%Straight Lines [id:da28617565488998564] 
\draw    (480,170) -- (480,180) ;
%Shape: Rectangle [id:dp8719089494637517] 
\draw   (470,180) -- (490,180) -- (490,200) -- (470,200) -- cycle ;
%Straight Lines [id:da7131724834989828] 
\draw    (480,200) -- (480,210) ;

% Text Node
\draw (459,150) node [anchor=north west][inner sep=0.75pt]    {$C^{\lambda }$};
% Text Node
\draw (449,129) node [anchor=north west][inner sep=0.75pt]  [font=\tiny]  {$i_{1}$};
% Text Node
\draw (475,129) node [anchor=north west][inner sep=0.75pt]  [font=\tiny]  {$i_{k}$};
% Text Node
\draw (449,209) node [anchor=north west][inner sep=0.75pt]  [font=\tiny]  {$a$};
% Text Node
\draw (475,209) node [anchor=north west][inner sep=0.75pt]  [font=\tiny]  {$\alpha $};
% Text Node
\draw (470,181) node [anchor=north west][inner sep=0.75pt]    {$d_{\pi }$};
% Text Node
\draw (435,219) node [anchor=north west][inner sep=0.75pt]  [font=\tiny]  {$i_{k}$};
\end{tikzpicture}} }
\end{equation}
and
\begin{equation}
	\vcenter{\hbox{

			\tikzset{every picture/.style={line width=0.75pt}} %set default line width to 0.75pt        
			
			\begin{tikzpicture}[x=0.75pt,y=0.75pt,yscale=-1,xscale=1]
				%uncomment if require: \path (0,345); %set diagram left start at 0, and has height of 345
				
				%Straight Lines [id:da6061037576944925] 
				\draw    (350,260) -- (350,270) ;
				%Shape: Rectangle [id:dp7580199229759923] 
				\draw   (330,270) -- (370,270) -- (370,290) -- (330,290) -- cycle ;
				%Straight Lines [id:da7974394061618273] 
				\draw    (340,290) -- (340,300) ;
				%Straight Lines [id:da04754877184545547] 
				\draw    (360,290) -- (360,300) ;
				%Shape: Rectangle [id:dp9199278557367838] 
				\draw   (330,240) -- (370,240) -- (370,260) -- (330,260) -- cycle ;
				%Straight Lines [id:da7783131751124868] 
				\draw    (350,230) -- (350,240) ;
				
				% Text Node
				\draw (339,270) node [anchor=north west][inner sep=0.75pt]    {$C^{\lambda }$};
				% Text Node
				\draw (329,299) node [anchor=north west][inner sep=0.75pt]  [font=\tiny]  {$a$};
				% Text Node
				\draw (355,299) node [anchor=north west][inner sep=0.75pt]  [font=\tiny]  {$\alpha $};
				% Text Node
				\draw (337,242) node [anchor=north west][inner sep=0.75pt]    {$P_{\sigma }{}$};
				% Text Node
				\draw (329,219) node [anchor=north west][inner sep=0.75pt]  [font=\tiny]  {$i_{1}$};
				% Text Node
				\draw (355,219) node [anchor=north west][inner sep=0.75pt]  [font=\tiny]  {$i_{k}$};

	\end{tikzpicture}} } = \vcenter{\hbox{

\tikzset{every picture/.style={line width=0.75pt}} %set default line width to 0.75pt        

\begin{tikzpicture}[x=0.75pt,y=0.75pt,yscale=-1,xscale=1]
%uncomment if require: \path (0,345); %set diagram left start at 0, and has height of 345

%Straight Lines [id:da13591886438145928] 
\draw    (430,230) -- (430,240) ;
%Shape: Rectangle [id:dp4368717258300592] 
\draw   (410,240) -- (450,240) -- (450,260) -- (410,260) -- cycle ;
%Straight Lines [id:da8735358789721521] 
\draw    (420,260) -- (420,270) ;
%Straight Lines [id:da714202287904196] 
\draw    (420,290) -- (420,300) ;
%Shape: Rectangle [id:dp6225907756449183] 
\draw   (410,270) -- (430,270) -- (430,290) -- (410,290) -- cycle ;
%Straight Lines [id:da6755195284952156] 
\draw    (440,260) -- (440,300) ;

% Text Node
\draw (419,240) node [anchor=north west][inner sep=0.75pt]    {$C^{\lambda }$};
% Text Node
\draw (409,219) node [anchor=north west][inner sep=0.75pt]  [font=\tiny]  {$i_{1}$};
% Text Node
\draw (435,219) node [anchor=north west][inner sep=0.75pt]  [font=\tiny]  {$i_{k}$};
% Text Node
\draw (409,299) node [anchor=north west][inner sep=0.75pt]  [font=\tiny]  {$a$};
% Text Node
\draw (435,299) node [anchor=north west][inner sep=0.75pt]  [font=\tiny]  {$\alpha $};
% Text Node
\draw (413,275) node [anchor=north west][inner sep=0.75pt]    {$\sigma $};

\end{tikzpicture}} }
\end{equation}
Requiring orthonormality
\begin{equation}
	(E_a^\lambda \otimes E_\alpha^\lambda, E_b^{\lambda'} \otimes E_\beta^{\lambda'}) = \delta_{ab} \delta_{\alpha \beta} \delta^{\lambda \lambda'},
\end{equation}
gives
\begin{equation}
	\sum_{i_1 \dots i_k} [(C^{\lambda}_{a,\alpha})^{i_1 i_2 \dots i_k}]^* (C^{\lambda'}_{b,\beta})^{i_1 i_2 \dots i_k} = \delta_{ab} \delta_{\alpha \beta} \delta^{\lambda \lambda'}.
\end{equation}
In other words,
\begin{equation}
	[(C^{\lambda}_{a,\alpha})^{i_1 i_2 \dots i_k} ]^*= (C_\lambda^{a, \alpha})_{i_1 \dots i_k}.
\end{equation}
However, Clebsch-Gordan coefficients for $S_N$ can be chosen real \cite[Section 7.14]{Hamermesh1962}. Therefore
\begin{equation}
	(C^{\lambda}_{a,\alpha})^{i_1 i_2 \dots i_k}= (C_\lambda^{a, \alpha})_{i_1 \dots i_k}. \label{eq: clebsch orthonormality}
\end{equation}

\begin{example}
	It may be useful for the reader to consider Schur-Weyl duality in the simple case of $\VN^{\otimes 1}=\VN$. The decomposition, including Clebsch-Gordan coefficients, was given in \ref{prop: VN decomp}. From this, we may deduce the irreducible representations of $P_1(\N)$. Note that $P_1(\N)$ has two non-isomorphic irreducible representations, both of dimension one, since the decomposition of $\VN$ is multiplicity free. To find the irreducible representation $Z_{[\N]}$, we act on $E^{[\N]}$ with the non-identity element in $P_1(\N)$:
	\begin{equation}
		 \PAdiagram{1}{}{}(E^{[\N]}) = \frac{1}{\sqrt{N}} \sum_{i=1}^N \PAdiagram{1}{}{}(e_i) = \frac{1}{\sqrt{N}} \sum_{i=1}^N \sum_{j=1}^N e_j = \N E^{[\N]}.
 	\end{equation}
 	Therefore, the non-identity element of $P_1(\N)$ is represented by the number $\N$. To find the irreducible representation $Z_{[\N-1,1]}$, we may use any of the vectors $E^{[\N-1,1]}_a$. For example, we have
 	\begin{equation}
 		\PAdiagram{1}{}{}(E^{[\N-1,1]}_1) =\frac{1}{\sqrt{2}}\qty(\PAdiagram{1}{}{}(e_1) - \PAdiagram{1}{}{}(e_2))=\frac{1}{\sqrt{2}}\qty(\sum_i e_i - \sum_i e_i)=0.
 	\end{equation}
 	Therefore, in this irreducible representation, the non-identity element is represented by $0$.
\end{example}

One of the many powerful consequences of Schur-Weyl duality is that it allows us to relate central elements in the group algebra of $\SN$ to elements in the center of partition algebras.
\begin{definition}[$k$-dual center]	\label{def: kdual}
Let $z \in \mathcal{Z}[\mathbb{F}(\SN)]$ have expansion
\begin{equation}
	z = \sum_{{\sn }\in \SN} a_\sn \sn.
\end{equation}
We define a corresponding element $\P_z \in \End_{S_N}(\VN^{\otimes k})$ by
\begin{equation}
	\P_z = \sum_{\sn \in \SN} a_{\sn} \P_\sn.
\end{equation}
Since
\begin{equation}
	\P_z \in \End_{S_N}(\VN^{\otimes k}) \cong \P_k(\N)
\end{equation}
for $\N \geq 2k$, there exists
\begin{equation}
	D_z = \sum_{\pi \in \setpart{[k \vert k']}} a_\pi D_\pi
\end{equation}
such that
\begin{equation}
	\P_z(e_{i_1} \otimes {\dots} \otimes e_{i_k}) = D_z(e_{i_1} \otimes {\dots} \otimes e_{i_k})= (D_z)^{i_{1'} \dots i_{k'}}_{i_1 \dots i_k}e_{i_{1'}} \otimes {\dots} \otimes e_{i_{k'}}. \label{eq: kdual def}
\end{equation}
This fixes the coefficients $a_\pi$ and we define $d_z \in P_k(N)$ by
\begin{equation}
	d_z = \sum_{\pi \in \setpart{[k \vert k']}} a_\pi d_\pi.
\end{equation}
For fixed $k$, we call the image under the map $z \mapsto d_z$ the $k$-dual center $\Zdual[P_k(N)]$.
\end{definition}
\begin{example}\label{ex: kdual}
	A useful example to consider is the $k=1$ dual of the sum of transpositions
	\begin{equation}
		z = \sum_{i < j} (ij).
	\end{equation}
	To compute $D_z$, we will use a trick involving elementary matrices
	\begin{equation}
		(E^j_i)^l_k = \delta^{jl} \delta_{ik}.
	\end{equation}
	Introduce the matrix-valued map $E: \VN \rightarrow \End(\VN) \otimes \VN$ by
	\begin{equation}
		E(e_i) = E^j_i e_j.
	\end{equation}
	Let $D_\pi$ be a diagram basis element for $\End_{\SN}(\VN)$, we have the following property
	\begin{equation}
		\Tr_{\VN}(E D_\pi)^i_j = [E^k_l (D_\pi)_k^l]^i_j = (E^k_l)^i_j (D_\pi)^k_l = \delta^{ki}\delta_{lj}(D_\pi)^k_l = (D_\pi)^i_j.
	\end{equation}
	For example, we have
	\begin{equation}
		\Tr_{\VN}(E D_{12}) = \sum_i E^i_i = D_{12}, \quad 	\Tr_{\VN}(E D_{1|2}) = \sum_{i,j} E^i_j = D_{1|2}. \label{eq: elem matrix diagram dual}
	\end{equation}
	
	Now note that the linear operator corresponding to transpositions can be written as elementary matrices
	\begin{equation}
		P_{(ij)} = \qty(\sum_k E^k_k) + E^i_j + E^j_i - E^i_i- E^j_j.
	\end{equation}
	and $z$ can be re-written as an unrestricted sum
	\begin{equation}
		z = \frac{1}{2}\sum_{i \neq j} (ij).
	\end{equation}
	Therefore,
	\begin{equation}
		P_z = \frac{1}{2}\sum_{i \neq j} \qty[\qty(\sum_k E^k_k) + E^i_j + E^j_i - E^i_i- E^j_j].
	\end{equation}
	We re-write
	\begin{equation}
		\sum_{i \neq j} E^i_j = \sum_{i,j} E^i_j - \sum_k E^k_k, \qq{and } \sum_{i\neq j} E^i_i = (\N-1)\sum_{k} E^k_k
	\end{equation}
	to get
	\begin{align}
		P_z &=  \frac{\N(\N-1)}{2}\sum_k  E^k_k + \sum_{i,j} E^i_j  - \N \sum_{k} E^k_k
	\end{align}
	From the previous observation \eqref{eq: elem matrix diagram dual} about sums of elementary matrices we can identify
	\begin{equation}
		D_z = D_{1|1'} + \binom{\N}{2} D_{11'} - \N D_{11'},
	\end{equation}
	and
	\begin{equation}
		d_z = \PAdiagram{1}{}{} + \binom{\N}{2}\PAdiagram{1}{-1/1}{}-\N \PAdiagram{1}{-1/1}{}.
	\end{equation}
\end{example}

\begin{proposition}
	The $k$-dual center is a subalgebra of the full center of the partition algebra: $\Zdual[P_k(N)] \subseteq \mathcal{Z}[P_k(N)]$.
\end{proposition}
\begin{proof}
	Since $\P_z D_\pi = D_\pi \P_z$ for all $D_\pi \in \End_{S_N}(\VN^{\otimes k}) $  we have $D_z D_\pi = D_\pi D_z$ and consequently $d_z d_\pi = d_\pi d_z$.
\end{proof}
This will be used to construct Cartan-like elements (known as Jucys-Murphy elements in the theory of symmetric groups and partition algebras) of the partition algebra, allowing us to construct important bases for $P_k(\N)$.

\subsection{$S_{\N-1}$ Schur-Weyl duality.}
Theorem \ref{thm: ind res VNotimesk} indicates that the subgroup $S_{\N-1} \subset \SN$ plays an important role in the description of $\VN^{\otimes k}$. In this section we give a version of Schur-Weyl duality for $\VN^{\otimes k}$ when the action of $\SN$ is restricted to $S_{\N-1}$.

As a representation of $S_{\N-1}$ we have
\begin{equation}
	\VN^{\otimes k} \cong \VN^{\otimes k} \otimes e_{\N},
\end{equation}
since $S_{\N-1}$ leaves $e_{\N}$ invariant (see \cite[Section 3]{Halverson2005}). This clever trick is used to consider
\begin{equation}
	\End_{S_{\N-1}}(\VN^{\otimes k}) \cong \End_{S_{\N-1}}(\VN^{\otimes k} \otimes e_\N). \label{eq: SN-1 end}
\end{equation}
The elements on the r.h.s. are tensors satisfying
\begin{equation}
	T_{(i_1)\sn (i_2)\sn {\dots} (i_k)\sn \N }^{(j_1)\sn (j_2)\sn {\dots} (j_k)\sn \N} = T_{i_1 i_2 \dots i_k \N}^{j_1 j_2 \dots j_k \N} \quad \forall \sn \in S_{\N-1}. \label{eq: SN-1 end constraint}
\end{equation}
In analogy to before, we consider the $S_{\N-1}$ orbits of bases elements for $\End(\VN^{\otimes k} \otimes e_\N)$. These correspond to set partitions $\pi \in \setpart{[k+1 \vert (k+1)']}$ with the constraint that $k+1$ and $(k+1)'$ always lie in the same block, as to maintain the form $\VN^{\otimes k} \otimes e_\N$. We note that the block containing $k+1$ and $(k+1)'$ can also contain other numbers.

This indicates that the following special subalgebra of $P_{k+1}(\N)$ is important.
\begin{definition}
	The partition algebras labelled by half-integers are defined by
	\begin{equation}
		P_{k+\frac{1}{2}}(\N) = \Span(d_\pi \, \vert \, \pi \in \setpart{[k+1 \vert (k+1)']} \text{ where $k+1$ and $(k+1)'$ are in the same block}).
	\end{equation}
\end{definition}
\begin{example}
	For example,
	\begin{equation}
		\PAdiagram{1}{-1/1}{} \in P_{0+\frac{1}{2}}(\N), \quad \PAdiagram{1}{}{} \not\in P_{0+\frac{1}{2}}(\N), \quad \PAdiagram{2}{-2/2,-1/-2}{} \in P_{1+\frac{1}{2}}(\N).
	\end{equation}
\end{example}

The half-integer partition algebras $P_{k+\frac{1}{2}}(\N)$ act on $\VN^{\otimes k}$ in an a priori unintuitive way, inspired by the observation in \eqref{eq: SN-1 end}.
\begin{definition}\label{def: half integer action}
	Let $d_\pi \in P_{k+\frac{1}{2}}(\N)$. Since $d_\pi$ is also an element of $P_{k+1}(\N)$, it  acts on $\VN^{\otimes k+1}$ in the conventional way
	\begin{equation}
		D_\pi(e_{i_1} \otimes \dots \otimes e_{i_k} \otimes e_{i_{k+1}^{}}).
	\end{equation}
	For the action on $\VN^{\otimes k}$ we define
	\begin{equation}
		(\Delta_\pi)^{i_{1'} \dots i_{k'}}_{i_1 \dots i_k} = (D_\pi)^{i_{1'} \dots i_{k'} \N}_{i_1 \dots i_k \N},
	\end{equation}
	and $\Delta_\pi \in \End(\VN^{\otimes k})$ corresponding $d_\pi$ by
	\begin{equation}
		\Delta_\pi(e_{i_1} \otimes \dots \otimes e_{i_k}) = (\Delta_\pi)^{i_{1'} \dots i_{k'}}_{i_1 \dots i_k}e_{i_{1'}} \otimes \dots \otimes e_{i_{k'}}.
	\end{equation}	
\end{definition}

This action of $P_{k+\frac{1}{2}}(\N)$ is Schur-Weyl dual to $S_{\N-1}$ acting on $\VN^{\otimes k}$.
\begin{theorem}
	Let $d_\pi \in P_{k+\frac{1}{2}}(\N)$ and $\Delta_\pi$ the corresponding linear map on $\VN^{\otimes k}$ in Definition \ref{def: half integer action}. This correspondence is an isomorphism
	\begin{equation}
		\End_{S_{\N-1}}(\VN^{\otimes k}) \cong P_{k+\frac{1}{2}}(\N),
	\end{equation}
	of algebras	for $\N \geq 2k+1$.
\end{theorem}
\begin{proof}
	See \cite[Theorem 3.6]{Halverson2005} and note that the kernel of the map is empty for $\N \geq 2k+1$. This is the analogue of Theorem \ref{thm: VN Schur-Weyl} but for $S_{\N-1}$.
\end{proof}
As before, it follows that $\VN^{\otimes k}$ has a multiplicity free decomposition in terms of $S_{\N-1} \times P_{k + \frac{1}{2}}(N)$ representations (see \cite[Theorem 3.22]{Halverson2005})
\begin{corollary}
	As a representation of $S_{\N-1} \times P_{k + \frac{1}{2}}(\N)$
	\begin{equation}
		\VN^{\otimes k} \cong \bigoplus_{\lambda \in \Lambda_{k+\frac{1}{2}, \N}} V_\lambda \otimes Z_\lambda^{1/2},  \label{eq: SN-1 SW simple}
	\end{equation}
	where $Z_\lambda^{1/2}$ are irreducible representations of $ P_{k + \frac{1}{2}}(\N)$.
\end{corollary}
We take the r.h.s. to have an orthonormal (with respect to \eqref{eq: VNk inner prod}) basis
\begin{equation}
	\{E^\lambda_{a} \otimes E^\lambda_\alpha \, \vert \, a=1,\dots, \dimSN{\lambda}, \alpha =1,\dots, \dim Z_\lambda^{1/2},  \lambda \in \Lambda_{k+\frac{1}{2},\N} \}. \label{eq: SN-1 SW basis}
\end{equation}

%Substituting the r.h.s of the $\SN$ decomposition in \eqref{eq: VN SW simple} into the l.h.s of the above decomposition gives
%\begin{equation}
%	\bigoplus_{\lambda  \in \Lambda_{k,\N} } \Res{\SN}{S_{\N-1}}{V_\lambda} \otimes Z_{\lambda} \cong  \bigoplus_{\lambda \in \Lambda_{k+\frac{1}{2}, \N}} V_\lambda \otimes Z_\lambda^{1/2}.
%\end{equation}
%Using the combinatorial rule \eqref{eq: RES SN} for restriction of $\SN$ representations gives
%\begin{equation}
%	\bigoplus_{\lambda \in \Lambda_{k+\frac{1}{2}, \N}} V_\lambda \otimes Z_\lambda^{1/2} \cong \bigoplus_{\lambda'  \in \Lambda_{k,\N} } \bigoplus_{\lambda \in \lambda' - {\scriptstyle \ydiagram{1}}} V_{\lambda} \otimes Z_{\lambda'}.
%\end{equation}
%We take the l.h.s to have an orthonormal basis
%\begin{equation}
%	\{ E^{\lambda}_{\underline{a}} \otimes E^\lambda_{\underline{\alpha}} \, \vert \, \underline{a}=1,\dots, \dimSN{\lambda}, \underline{\alpha} = 1,\dots,\dim Z^{1/2}_\lambda, \lambda \in \Lambda_{k+\frac{1}{2}, \N}\}. \label{eq: SN-1 SW basis}
%\end{equation}
%related to the irreducible $\SN$ basis as
%\begin{equation}
%	E^{\lambda}_{\underline{a}} \otimes E^\lambda_{\underline{\alpha}} = (\R_{\lambda'}^{\lambda})_{\underline{a}}^a E_a^{\lambda'} \otimes E_\alpha^{\lambda'},
%\end{equation}
%where $ (B^{\lambda \lambda'}_{\underline{a}\, \underline{\alpha}})^{a \alpha}$ are branching coefficients for $\SN \rightarrow S_{\N-1}$.

Analogously to the $\SN$ Schur-Weyl duality case, we have a duality between the center $\mathcal{Z}[\mathbb{F}(S_{\N-1})]$ and the center $\mathcal{Z}[P_{k+\frac{1}{2}}(\N)]$.
\begin{definition}[$(k+\frac{1}{2})$-dual center] \label{def: half dual}
	Let $z \in \mathcal{Z}[\mathbb{F}(S_{\N-1})]$ and define $\Delta_z$ through the equality
	\begin{equation}
		\P_z(e_{i_1} \otimes {\dots} \otimes e_{i_k}) = \Delta_z(e_{i_1} \otimes {\dots} \otimes e_{i_k}).
	\end{equation}
	Suppose it has an expansion
	\begin{equation}
		\Delta_z = \sum_{\pi \in \setpart{[k+1\vert (k+1)']}} a_\pi \Delta_\pi.
	\end{equation}
	The $(k+\tfrac{1}{2})$-dual is the element $d_z^{}$ defined by
	\begin{equation}
		d_z^{} = \sum_{\pi \in \setpart{[k+1\vert (k+1)']}} a_\pi d_\pi.
	\end{equation}
	We denote the image under the map $z \rightarrow d_z^{}$ by $\widetilde{\mathcal{Z}}[P_{k+\frac{1}{2}}(\N)]$.
\end{definition}
\begin{example}\label{ex: khalf dual}
	A useful example to consider is the $1+\frac{1}{2}$-dual of
	\begin{equation}
		z = \sum_{1 \leq l < m \leq \N-1} (lm).
	\end{equation}
	To find the dual, note that
	\begin{equation}
		\sum_{1 \leq l < m  \leq \N-1} P_{(lm)}(e_i \otimes e_\N) = \sum_{1 \leq l < m \leq \N-1} P_{(lm)}e_i \otimes e_\N = \sum_{1 \leq l < m \leq \N-1} P_{(lm)} \otimes \idn(e_i \otimes e_\N).
	\end{equation}
	We re-write
	\begin{equation}
		\sum_{1 \leq l < m \leq \N} (lm) - \sum_{l=1}^{\N-1} (l\N),
	\end{equation}
	such that
	\begin{equation}
		\sum_{1 \leq l < m \leq \N-1} P_{(lm)} \otimes \idn(e_i \otimes e_\N) = \sum_{1 \leq l < m \leq \N} P_{(lm)} \otimes \idn(e_i \otimes e_\N) - \sum_{l=1}^{\N-1} P_{(l\N)} \otimes \idn(e_i \otimes e_\N).
	\end{equation}
	The dual of the first term immediately follows from Example \ref{ex: kdual} as
	\begin{equation}
		D_{1|1'|22'} + \binom{\N}{2} D_{11'|22'} - \N D_{11'|22'}, \label{eq: khalf dual term 1}
	\end{equation}
	since $\idn = \sum_k E^k_k$.
	It remains to find the dual of the second term. The second term can be simplified by considering the cases $i=\N, i\neq \N$
	\begin{align}
	&\sum_{l=1}^{\N-1} P_{(l\N)} \otimes \idn(e_i \otimes e_\N) = \delta_{i\N} \sum_{l=1}^{\N-1} e_l \otimes e_\N + (1-\delta_{i\N})(e_\N \otimes e_\N + (\N-2)e_i \otimes e_\N) \\
		&= \delta_{i\N} \sum_{l=1}^{\N} e_l \otimes e_\N - \delta_{i\N} e_\N \otimes e_\N + (1-\delta_{i\N})(e_\N \otimes e_\N + (\N-2)e_i \otimes e_\N),
	\end{align}
	in the second equality we added $1 = \delta_{i\N} e_\N \otimes e_\N - \delta_{i\N} e_\N \otimes e_\N$. Expand all the terms gives
	\begin{equation}
		\delta_{i\N} \sum_l e_l \otimes e_\N - 2 \delta_{i\N} e_\N \otimes e_\N + e_\N \otimes e_\N + (\N-2)e_i \otimes e_\N - (\N-2)\delta_{i\N}e_i \otimes e_\N.
	\end{equation}
	This corresponds to the action of the following tensor product of elementary matrices
	\begin{equation*}
		\sum_{k,l} E^l_k \otimes E^{k}_k - 2 \sum_k E^k_k \otimes E^{k}_k + \sum_{k,l} E^k_l \otimes E^k_k + (\N-2) \sum_{k,l} E^k_k \otimes E^l_l - (\N-2)\sum_{k}E^k_k \otimes E^k_k(e_i \otimes e_\N).
	\end{equation*}
	Using the correspondence observed in Example \ref{ex: kdual} we find that this corresponds to the following diagram basis elements
	\begin{equation}
		D_{1'|122'}-\N D_{11'22'}+D_{1|1'22'} + (\N-2) D_{11'|22'}(e_i \otimes e_\N).
	\end{equation}
	Together with the terms in \eqref{eq: khalf dual term 1} we have the dual
	\begin{equation}
		d_z = \PAdiagram{2}{-2/2}{} + \binom{\N}{2} \PAdiagram{2}{-1/1,-2/2}{} - 2\N  \PAdiagram{2}{-1/1,-2/2}{} +2  \PAdiagram{2}{-1/1,-2/2}{} - \PAdiagram{2}{-1/-2,-2/2}{}- \PAdiagram{2}{1/2,-2/2}{}+\N  \PAdiagram{2}{-1/1,-2/2,1/2,-1/-2}{}.
	\end{equation}
\end{example}

\subsection{Inductive chain.}
The partition algebras $P_{k}(\N), P_{k-\frac{1}{2}}(\N)$ form a chain of subalgebras,
\begin{equation}
	P_1(\N) \subset P_{1+\frac{1}{2}}(\N) \subset \dots \subset P_{k-\frac{1}{2}}(\N) \subset  P_{k}(\N),
\end{equation}
where $P_{k}(\N)$ is embedded into $P_{k + \frac{1}{2}}(\N)$ by adding a strand to the right of all diagram basis elements. For example
\begin{align}
	P_1(\N) \ni \PAdiagram{1}{}{} \mapsto &\PAdiagram{2}{-2/2}{} \in P_{1+\frac{1}{2}}(\N) \qq{and}\\
											&\begin{aligned}[t]
												\PAdiagram{2}{-2/2}{} \mapsto &\PAdiagram{2}{-2/2}{} \in P_{2}(\N) \qq{and} \\
												&\PAdiagram{2}{-2/2}{} \mapsto \PAdiagram{3}{-2/2, -3/3}{} \in P_{2+\frac{1}{2}}(\N).
											\end{aligned}											
\end{align}

We will now investigate the properties of restrictions of irreducible representations along the chain. We will see that there is a close relationship between the induction/restriction construction of $\VN^{\otimes k}$ and restriction of partition algebras.
\begin{proposition}
	We will prove the following statements.
	\begin{enumerate}
	\item[(a)]
	Let $Z_\lambda$ be an irreducible representation of $P_k(\N)$. For  $\lambda \in \Lambda_{k,\N}$, the restriction to $P_{k-\frac{1}{2}}(\N)$ decomposes into
	\begin{equation}
		\Res{P_k(\N)}{P_{k-\frac{1}{2}}(\N)}{Z_\lambda} \cong \bigoplus_{\lambda' \in \lambda - {\scriptstyle \ydiagram{1}}} Z_{\lambda'}^{1/2}. \label{eq: res k to k-12}
	\end{equation}
	\item[(b)]
	Let $Z_\lambda^{1/2}$ be an irreducible representation of $P_{k+\frac{1}{2}}(\N)$. For $\lambda \in \Lambda_{k+\frac{1}{2},\N}$. The restriction to $P_{k}(\N)$ decomposes into
	\begin{equation}
		\Res{P_{k+\frac{1}{2}}(\N)}{P_{k}(\N)}{Z_\lambda^{1/2}} \cong \bigoplus_{\lambda' \in \lambda + {\scriptstyle \ydiagram{1}}} Z_{\lambda}. \label{eq: res k+12 to k}
	\end{equation}
	\end{enumerate}
\end{proposition}
\begin{proof}
	To prove (a), recall that $\Ind{\SN}{S_{\N-1}}{\Res{\SN}{S_{\N-1}}{\VN^{\otimes k-1}} }\cong V^{\otimes k}$. As a representation of $\SN \times P_k(\N)$ we have
	\begin{equation}
		V^{\otimes k} \cong \bigoplus_{\lambda \in \Lambda_{k,\N}} V_\lambda \otimes Z_\lambda.
	\end{equation}
	On the other hand, as a representation of $S_{\N} \times P_{k-\frac{1}{2}}(\N)$,
	\begin{align}
		\Ind{\SN}{S_{\N-1}}{\Res{\SN}{S_{\N-1}}{\VN^{\otimes k-1}}} &\cong  \bigoplus_{\lambda' \in \Lambda_{k-1,\N}}  \Ind{S_{\N-1}}{\SN}{V_{\lambda'}}  \otimes Z_{\lambda'}^{1/2} \\
		&\cong \bigoplus_{\lambda' \in \Lambda_{k-1,\N}} \bigoplus_{\lambda \in \lambda' + {\scriptstyle \ydiagram{1}}}V_{\lambda}  \otimes Z_{\lambda'}^{1/2}.
	\end{align}
	Note that the set of $Z_{\lambda'}^{1/2}$ appearing next to $V_{\lambda}$ are such that $\lambda \in \lambda' + {\scriptstyle \ydiagram{1}}$, or equivalently the subset of $\lambda' \in \Lambda_{k-1,\N}$ satisfying $\lambda' \in \lambda - {\scriptstyle \ydiagram{1}}$.
	Restricting $Z_{\lambda}$ to $P_{k-\frac{1}{2}}(\N)$ and comparing the representations appearing next to $V_{\lambda}$ in the two decompositions gives
	\begin{equation}
		\Res{P_{k}(\N)}{P_{k-\frac{1}{2}}(\N)}{Z_\lambda} \cong \bigoplus_{\lambda' \in \lambda - {\scriptstyle \ydiagram{1}}} Z_{\lambda'}^{1/2}.
	\end{equation}

	To prove (b) we use a similar argument. As a representation of $\SN \times P_k(\N)$ we have
	\begin{equation}
		V^{\otimes k} \cong \bigoplus_{\lambda \in \Lambda_{k,\N}} V_\lambda \otimes Z_\lambda,
	\end{equation}
	and therefore as a $S_{\N-1} \times P_k(\N)$ representation
	\begin{equation}
		V^{\otimes k} \cong \bigoplus_{\lambda \in \Lambda_{k,\N}} \bigoplus_{\lambda' \in \lambda - {\scriptstyle \ydiagram{1}}} V_{\lambda'} \otimes Z_\lambda.
	\end{equation}
	On the other hand,
	\begin{equation}
		V^{\otimes k} \cong \bigoplus_{\lambda' \in \Lambda_{k+\frac{1}{2}, \N}} V_{\lambda'} \otimes Z^{1/2}_{\lambda'},
	\end{equation}
	as a representation of $S_{\N-1} \times P_{k+\frac{1}{2}}(\N)$. Restricting $Z^{1/2}_{\lambda'}$ to $P_{k}(\N)$ and comparing the two decompositions gives
	\begin{equation}
		\Res{P_{k+\frac{1}{2}}(\N)}{P_{k}(\N)}{Z^{1/2}_{\lambda'}} \cong \bigoplus_{\lambda \in \lambda' + {\scriptstyle \ydiagram{1}}} Z_{\lambda},
	\end{equation}
	where $\lambda \in \Lambda_{k,\N}$.
\end{proof}

It will be useful to introduce explicit matrices for these restrictions.
\begin{definition}
	Let the representation $Z_\lambda$ of $P_k(\N)$ have a basis $E_\alpha^\lambda$ and define an inner product where the basis is orthonormal. We take the r.h.s. of \eqref{eq: res k to k-12} to have a basis $\{E_{\beta}^{\lambda \rightarrow \lambda'} \, \vert \, \lambda' \in \lambda - {\scriptstyle \ydiagram{1}}, \beta = 1,\dots, \dim Z_{\lambda'}^{1/2}\}$. The two bases are related by the matrix $\R^{\lambda \rightarrow \lambda'}_{\alpha\beta}$
	\begin{equation}
		E^{\lambda}_\alpha = \sum_{\lambda', \beta} \R^{\lambda \rightarrow \lambda'}_{\beta \alpha} E^{\lambda \rightarrow \lambda'}_\beta
	\end{equation}
\end{definition}
We introduce a diagram for this change of basis matrix
\begin{equation}
	\R^{\lambda \rightarrow \lambda'}_{\beta \alpha} = \vcenter{\hbox{

			\tikzset{every picture/.style={line width=0.75pt}} %set default line width to 0.75pt        
			
			\begin{tikzpicture}[x=0.75pt,y=0.75pt,yscale=-1,xscale=1]
				%uncomment if require: \path (0,345); %set diagram left start at 0, and has height of 345
				
				%Straight Lines [id:da9722098475908598] 
				\draw    (460,40) -- (460,60) ;
				%Shape: Rectangle [id:dp12734598441638512] 
				\draw   (430,60) -- (490,60) -- (490,80) -- (430,80) -- cycle ;
				%Straight Lines [id:da5594631839841715] 
				\draw    (460,80) -- (460,100) ;
				
				% Text Node
				\draw (437,62) node [anchor=north west][inner sep=0.75pt]    {$\mathcal{R}^{\lambda \rightarrow \lambda '}$};
				% Text Node
				\draw (455,23) node [anchor=north west][inner sep=0.75pt]  [font=\tiny]  {$\beta $};
				% Text Node
				\draw (455,99) node [anchor=north west][inner sep=0.75pt]  [font=\tiny]  {$\alpha $};
				% Text Node
				\draw (439,85) node [anchor=north west][inner sep=0.75pt]  [font=\tiny]  {$\lambda $};
				% Text Node
				\draw (438,44) node [anchor=north west][inner sep=0.75pt]  [font=\tiny]  {$\lambda '$};
		\end{tikzpicture}} }
\end{equation}
Demanding that $E^{\lambda \rightarrow \lambda'}_\beta$ are orthonormal gives
\begin{equation}
	E_\beta^{\lambda \rightarrow \lambda'} = \sum_\alpha \langle E^{\lambda \rightarrow \lambda'}_\beta, E^{\lambda}_\alpha\rangle E^{\lambda}_\alpha = \sum_\alpha  \R^{\lambda \rightarrow \lambda'}_{\beta \alpha}E^{\lambda}_\alpha, \label{eq: res basis small}
\end{equation}
and
\begin{equation}
	\langle E^{\lambda \rightarrow \lambda'}_\alpha, E^{\lambda \rightarrow \lambda''}_\beta \rangle = \sum_{\gamma, \gamma'}  \R^{\lambda \rightarrow \lambda'}_{ \alpha \gamma}\R^{\lambda \rightarrow \lambda''}_{\beta \gamma' } \langle E^{\lambda}_{\gamma}, E^{\lambda}_{\gamma'} \rangle =\sum_{\gamma, \gamma'}  \R^{\lambda \rightarrow \lambda'}_{\alpha \gamma }\R^{\lambda \rightarrow \lambda''}_{\beta \gamma' }\delta_{\gamma \gamma'}= \delta^{\lambda' \lambda''}\delta_{\alpha \beta} \label{eq: res basis ON condition}
\end{equation}
%	\begin{equation}
%		E^\lambda_\alpha = \sum_{\lambda' \in \lambda - {\scriptstyle \ydiagram{1}}}\sum_{\beta} (\R^{\lambda}_{\lambda'})_\alpha^\beta E_\beta^{\lambda'}, \label{eq: res basis big}
%	\end{equation}
%	and
%	\begin{equation}
%		E_\beta^{\lambda'} = \sum_\alpha E^\lambda_\alpha (R^{\lambda}_{\lambda'})^\beta_\alpha \label{eq: res basis small}
%	\end{equation}
Let $d \in P_{k-\frac{1}{2}}(\N)$ and $D^\lambda(d), D^{\lambda'}(d)$ be irreducible representations of $P_k(\N), P_{k-\frac{1}{2}}(\N)$ respectively. Then the change of basis matrix satisfies
%\begin{equation}
%	\sum_{\gamma} D^{\lambda'}(d)_{\gamma \beta}E^{\lambda \rightarrow \lambda'}_\gamma =\sum_\gamma D^{\lambda'}(d)_{\gamma \beta}\sum_{\alpha} \R^{\lambda \rightarrow \lambda'}_{\gamma \alpha}E_\alpha^{\lambda} = \sum_{\alpha} \R^{\lambda \rightarrow \lambda'}_{\beta \alpha}\sum_{\gamma} D^{\lambda}(d)_{\gamma \alpha } E_\gamma^{\lambda}, \label{eq: res Pk equivariance}
%\end{equation}
%	\begin{equation}
%		\sum_\alpha D^\lambda_{\alpha \beta}(d)E^\lambda_\alpha =\sum_{\lambda' \in \lambda - {\scriptstyle \ydiagram{1}}}\sum_{\alpha, \gamma} D^\lambda_{\alpha \beta}(d)(\R^{\lambda}_{\lambda'})_\alpha^\gamma E_\gamma^{\lambda'}
%		 =\sum_{\lambda' \in \lambda - {\scriptstyle \ydiagram{1}}}\sum_{\beta, \gamma} (\R^{\lambda}_{\lambda'})_\alpha^\gamma D^{\lambda'}_{\beta \gamma}(d) E_\beta^{\lambda'}, 
%	\end{equation}
$\R^{\lambda \rightarrow \lambda'} D^\lambda(d) = D^{\lambda'}(d)\R^{\lambda \rightarrow \lambda'} $ or
\begin{equation}
	\vcenter{\hbox{

			\tikzset{every picture/.style={line width=0.75pt}} %set default line width to 0.75pt        
			
			\begin{tikzpicture}[x=0.75pt,y=0.75pt,yscale=-1,xscale=1]
				%uncomment if require: \path (0,345); %set diagram left start at 0, and has height of 345
				
				%Straight Lines [id:da8489431207210623] 
				\draw    (550,50) -- (550,70) ;
				%Shape: Rectangle [id:dp9550404065313889] 
				\draw   (520,70) -- (580,70) -- (580,90) -- (520,90) -- cycle ;
				%Straight Lines [id:da4698620240108593] 
				\draw    (550,90) -- (550,110) ;
				%Shape: Rectangle [id:dp6728490745023741] 
				\draw   (540,30) -- (560,30) -- (560,50) -- (540,50) -- cycle ;
				%Straight Lines [id:da9121719814993572] 
				\draw    (550,20) -- (550,30) ;
				
				% Text Node
				\draw (528,71) node [anchor=north west][inner sep=0.75pt]    {$\R^{\lambda \rightarrow \lambda '}$};
				% Text Node
				\draw (545,9) node [anchor=north west][inner sep=0.75pt]  [font=\tiny]  {$\beta $};
				% Text Node
				\draw (545,109) node [anchor=north west][inner sep=0.75pt]  [font=\tiny]  {$\alpha $};
				% Text Node
				\draw (529,95) node [anchor=north west][inner sep=0.75pt]  [font=\tiny]  {$\lambda $};
				% Text Node
				\draw (529,54) node [anchor=north west][inner sep=0.75pt]  [font=\tiny]  {$\lambda '$};
				% Text Node
				\draw (546,32) node [anchor=north west][inner sep=0.75pt]    {$d$};

	\end{tikzpicture}}  } = \vcenter{\hbox{

\tikzset{every picture/.style={line width=0.75pt}} %set default line width to 0.75pt        

\begin{tikzpicture}[x=0.75pt,y=0.75pt,yscale=-1,xscale=1]
%uncomment if require: \path (0,345); %set diagram left start at 0, and has height of 345

%Straight Lines [id:da7017492948271962] 
\draw    (630,50) -- (630,70) ;
%Shape: Rectangle [id:dp08074497558415894] 
\draw   (600,70) -- (660,70) -- (660,90) -- (600,90) -- cycle ;
%Straight Lines [id:da4188220259854667] 
\draw    (630,90) -- (630,110) ;
%Shape: Rectangle [id:dp05353813277125985] 
\draw   (620,110) -- (640,110) -- (640,130) -- (620,130) -- cycle ;
%Straight Lines [id:da07363013290019071] 
\draw    (630,130) -- (630,140) ;

% Text Node
\draw (608,71) node [anchor=north west][inner sep=0.75pt]    {$\R^{\lambda \rightarrow \lambda '}$};
% Text Node
\draw (625,39) node [anchor=north west][inner sep=0.75pt]  [font=\tiny]  {$\beta $};
% Text Node
\draw (625,139) node [anchor=north west][inner sep=0.75pt]  [font=\tiny]  {$\alpha $};
% Text Node
\draw (609,95) node [anchor=north west][inner sep=0.75pt]  [font=\tiny]  {$\lambda $};
% Text Node
\draw (609,54) node [anchor=north west][inner sep=0.75pt]  [font=\tiny]  {$\lambda '$};
% Text Node
\draw (625,112) node [anchor=north west][inner sep=0.75pt]    {$d$};

\end{tikzpicture}}  }
\end{equation}
We use the same notation $(\R^{\lambda \rightarrow \lambda'})$ for restriction matrices $P_{k+\frac{1}{2}}(\N) \rightarrow P_k(\N)$ and $\lambda \in \Lambda_{k+\frac{1}{2},\N}, \lambda' \in \Lambda_{k,\N}$.

%\begin{equation}
%	\sum_{\mu \in \lambda - {\scriptstyle \ydiagram{1}}} \sum_\gamma (R^{\lambda}_{\mu})^\gamma_\alpha (R^{\lambda'}_{\mu})^{\gamma}_{\beta} = \delta^{\lambda \lambda'}\delta_{\alpha \beta}, \label{eq: restriction ON condition}
%\end{equation}
%for $\lambda,\lambda' \in \Lambda_{k,\N}$. Similarly, equation \eqref{eq: res basis small} gives
%\begin{equation}
%	\sum_\alpha (R^\lambda_\mu)^{\beta}_\alpha (R^\lambda_\mu)^{\beta'}_\alpha = \delta_{\mu \mu'} \delta^{\beta \beta'}. \label{eq: res matrix ON}
%\end{equation}

%\begin{definition}[Inductive basis] 
Let $E^{\lambda}_\alpha$ a basis for $Z_\lambda$ and repeat the restriction to find
\begin{align}
	E^{\lambda^{(k)}}_\alpha &= \sum_{\lambda^{(k-\frac{1}{2})}, \beta} \R^{\lambda^{(k)} \rightarrow \lambda^{(k-\frac{1}{2})}}_{\beta \alpha} E^{\lambda^{(k)} \rightarrow \lambda^{(k-\frac{1}{2})}}_\beta \\
	&=\sum_{\lambda^{(k-\frac{1}{2})}, \beta} \R^{\lambda^{(k)} \rightarrow \lambda^{(k-\frac{1}{2})}}_{\beta \alpha} 
	\sum_{\lambda^{(k-1)}, \gamma} \R^{\lambda^{(k-\frac{1}{2})} \rightarrow \lambda^{(k-1)}}_{\gamma \beta} E^{\lambda^{(k-\frac{1}{2})} \rightarrow \lambda^{(k-1)}}_\gamma.
\end{align}
Let $\vactab = (\lambda^{(0)} = [\N], \lambda^{(\frac{1}{2})} = [\N-1], \lambda^{(1)}, \lambda^{(\frac{3}{2})},\dots,\lambda^{(k)})$ be a vacillating tableaux with shape $\lambda = \lambda^{(k)}$ and length $k$ (see Definition \ref{def: vac tableu}). Define
\begin{equation}
	E^\vactab_\alpha = (\R^{\lambda^{(1+\frac{1}{2})} \rightarrow \lambda^{(1)}} \R^{\lambda^{(2)} \rightarrow \lambda^{(1+\frac{1}{2})}}\dots \R^{\lambda^{(k-\frac{1}{2})} \rightarrow \lambda^{(k-1)}}\R^{\lambda^{(k)} \rightarrow \lambda^{(k-\frac{1}{2})}} )_{1 \alpha}E^{\lambda^{(1+\frac{1}{2})} \rightarrow \lambda^{(1)}}_1, \label{def: inductive basis}
\end{equation}
where we have used the fact that all irreducible $P_1(\N)$ representations are one-dimensional to fix the index on $E^{\lambda^{(1)}}_\beta = E^{\lambda^{(1+\frac{1}{2})} \rightarrow \lambda^{(1)}}_1$.
With this definition, we can write the basis in the following suggestive form
\begin{equation}
	E^{\lambda^{(k)}}_\alpha = \sum_{\vactab} E^\vactab_\alpha, \label{eq: vac tab expansion}
\end{equation}
where the sum is over all vacillating tableaux of shape $\lambda^{(k)}$ and length $k$.
%\end{definition}

We give a diagrammatic definition of this chain of restriction matrices
\begin{equation}
	\vcenter{\hbox{

			\tikzset{every picture/.style={line width=0.75pt}} %set default line width to 0.75pt        
			
			\begin{tikzpicture}[x=0.75pt,y=0.75pt,yscale=-1,xscale=1]
				%uncomment if require: \path (0,502); %set diagram left start at 0, and has height of 502
				
				%Straight Lines [id:da48744125314379594] 
				\draw    (120,280) -- (120,300) ;
				%Shape: Rectangle [id:dp28225366337584146] 
				\draw   (100,300) -- (140,300) -- (140,320) -- (100,320) -- cycle ;
				%Straight Lines [id:da44830397064859007] 
				\draw    (120,320) -- (120,340) ;
				
				% Text Node
				\draw (108,303) node [anchor=north west][inner sep=0.75pt]    {$\R^{\vactab }$};
				% Text Node
				\draw (115,339) node [anchor=north west][inner sep=0.75pt]  [font=\tiny]  {$\alpha $};
				% Text Node
				\draw (99,323) node [anchor=north west][inner sep=0.75pt]  [font=\tiny]  {$\lambda ^{( k)}$};
				% Text Node
				\draw (99,284) node [anchor=north west][inner sep=0.75pt]  [font=\tiny]  {$\lambda ^{( 1)}$};

	\end{tikzpicture}} } = \vcenter{\hbox{

\tikzset{every picture/.style={line width=0.75pt}} %set default line width to 0.75pt        

\begin{tikzpicture}[x=0.75pt,y=0.75pt,yscale=-1,xscale=1]
	%uncomment if require: \path (0,502); %set diagram left start at 0, and has height of 502
	
	%Straight Lines [id:da47232202429159753] 
	\draw    (210,360) -- (210,380) ;
	%Shape: Rectangle [id:dp006684033440778059] 
	\draw   (180,380) -- (260,380) -- (260,410) -- (180,410) -- cycle ;
	%Straight Lines [id:da693814371134059] 
	\draw    (210,410) -- (210,430) ;
	%Straight Lines [id:da7241966963038038] 
	\draw  [dash pattern={on 0.84pt off 2.51pt}]  (210,320) -- (210,360) ;
	%Straight Lines [id:da9643079545514646] 
	\draw    (210,250) -- (210,270) ;
	%Shape: Rectangle [id:dp6984091303019875] 
	\draw   (180,270) -- (260,270) -- (260,300) -- (180,300) -- cycle ;
	%Straight Lines [id:da227966486611328] 
	\draw    (210,300) -- (210,320) ;
	
	% Text Node
	\draw (185,380) node [anchor=north west][inner sep=0.75pt]    {$\R^{\lambda ^{( k)}\rightarrow \lambda ^{( k-\frac{1}{2})}}$};
	% Text Node
	\draw (205,429) node [anchor=north west][inner sep=0.75pt]  [font=\tiny]  {$\alpha $};
	% Text Node
	\draw (185,270) node [anchor=north west][inner sep=0.75pt]    {$\R^{\lambda ^{( 1+\frac{1}{2})}\rightarrow \lambda ^{( 1)}}$};
\end{tikzpicture}} }
\end{equation}
This will be useful for constructing matrix units, which is the topic of the following sections.
%Since each restriction is multiplicity free, this is uniquely fixed by the property \eqref{eq: res Pk equivariance}, given a set of matrix representations and choice of normalization of $E_\beta^{\lambda^{(1)}} = E^{\lambda^{(1)}} $.
%The existence of such a basis is also proven in \cite[Theorem 3.37]{Halverson2005}.

\section{Semi-simplicity of $P_k(N)$ and matrix units}
\label{sec: Semi-Simple Algebra Technology}
The partition algebras are known to be semi-simple for $\N \geq 2k$. This means that they can be understood as algebras of block matrices.
In this section we will review the theoretical background necessary to go from the diagram basis to the basis that makes this correspondence manifest. This is used in the next section where we give an explicit algorithm for constructing the change of basis matrix as a function of $\N$ for all $\N \geq 2k$.

Let $\mathcal{B} = \{b_1,\dots,b_{B(2k)}\}$ be a basis for $P_k(N)$ with structure constants
\begin{equation}
	b_i b_j = \sum_{k=1}^{B(2k)}C_{ij}^k b_k.
\end{equation}
The regular representation of $P_k(N)$ is defined by the action of $P_k(N)$ on itself.
\begin{definition}[Regular representation]
	Let $V^{\text{reg}}$ be the vector space
	\begin{equation}
		V^{\text{reg}} = \Span(\vec{b}_i),
	\end{equation}
	where we emphasise that $\vec{b}_i$ are vectors in $V^{\text{reg}}$ in contrast to elements of the algebra $P_k(N)$. The partition algebra acts on $V^{\text{reg}}$ by left multiplication
	\begin{equation}
		b_i \vec{b}_j = C_{ij}^k \vec{b}_k.
	\end{equation}
	and right multiplication
	\begin{equation}
		\vec{b}_j	b_i  = C_{ji}^k \vec{b}_k.
	\end{equation}
\end{definition}
The trace of the left action on the regular representation can be written as
\begin{equation}
	\tr(b_i) = \sum_{j=1}^{B(2k)} C_{ij}^j = \sum_{j=1}^{B(2k)} \mathrm{Coeff}(b_j, b_i b_j),
\end{equation}
where $\mathrm{Coeff}(b_j, d)$ is the coefficient of $b_j$ in the expansion of $d \in P_k(N)$ in the basis $\mathcal{B}$.

We will use the following closely related theorems.
\begin{theorem}
	For $N \geq 2k$, $P_k(N)$ is semi-simple and therefore
	\begin{equation}
		G_{ij} \equiv \tr(b_i b_j) \label{eq: gram for reg rep}
	\end{equation}
	is an invertible matrix. We say that the trace in the regular representation defines a non-degenerate bilinear form on $P_k(N)$
\end{theorem}
\begin{proof}
	See \cite[Theorem~3.27]{Halverson2005} and \cite[Equation~5.9]{Halverson2005}.
\end{proof}
It will be useful to have the following equivalent definition of non-degeneracy in what follows.
\begin{definition}
	 A bilinear form on $P_k(N)$ is non-degenerate if there exists no non-zero element $d \in P_k(N)$ such that
\begin{equation}
	\tr(b_i d) = 0 \quad \forall i=1,\dots,B(2k).
\end{equation}
\end{definition}

\subsection{Schur-Weyl duality and non-degenerate bilinear forms.}
Semi-simplicity of $P_k(N)$ also implies that the regular representation of $P_k(N)$ is completely decomposable. In particular, because the left and right action of $P_k(N)$ commutes we have the following result.
\begin{theorem}
	Let $V^{\text{reg}}$ be the regular representation of $P_k(N)$, then
	\begin{equation}
		V^{\text{reg}} = \bigoplus_{\lambda \in \Lambda_{k,\N}} Z_\lambda \otimes Z_\lambda. \label{eq: pkn reg rep decomp}
	\end{equation}
	as a representation of the left and right action of $P_k(N)$.
\end{theorem}
\begin{proof}
	See statements in proof of \cite[Proposition~5.7]{Halverson2005}.
\end{proof}
As a consequence of \eqref{eq: pkn reg rep decomp} we have
\begin{corollary}
The trace of $d\in P_k(N)$ in the regular representation can be decomposed as
\begin{equation}
	\tr(d) = \sum_{\lambda \in \Lambda_{k,\N}} \Tr(D^\lambda(d)) \Tr(D^\lambda(1)) = \sum_{\lambda \in \Lambda_{k,\N} } \chr^{{\lambda}}(d) \dimPk{\lambda} , \label{eq: reg trace decomposition}
\end{equation}
where the sum is over all irreducible representations of $P_k(N)$, $\dimPk{\lambda}$ is the dimension of the representation $Z_\lambda$ and $\chr^{{\lambda}}$ is the corresponding character.
\end{corollary}
\begin{corollary}
The characters can be extracted from the trace by means of projection operators $p_{\lambda} \in P_k(N)$,
\begin{equation}
	\tr(p_{\lambda}d) = \dimPk{\lambda} \chr^{{\lambda}}(d). \label{eq: reg trace decomp}
\end{equation}
\end{corollary}
\begin{proof}
	
This can be seen as a consequence of character orthogonality (see \cite[Theorem~3.8, Theorem~3.9]{AR90DissertCh1})
\begin{equation}
	\sum_{i,j=1}^{B(2k)} \dimPk{\lambda } \chr^{\lambda }(b_i)(G^{-1})_{ij}\chr^{\lambda'}(b_j)  = \delta^{\lambda \lambda'},
\end{equation}
and the fact that projectors can be written as
\begin{equation}
	p_{\lambda} = \sum_{i,j=1}^{B(2k)} \dimPk{\lambda} \chr^{\lambda}(b_i)(G^{-1})_{ij}b_j,
\end{equation}
where $(G^{-1})_{ij}$ is the inverse of the matrix $G_{ij}$ in \eqref{eq: gram for reg rep}. Alternatively, it follows from the decomposition \eqref{eq: pkn reg rep decomp}.
\end{proof}

%As we have reviewed in section \ref{sec: PA}, $P_k(N)~\cong~\End_{S_N}(V_N^{\otimes k})$ when $N \geq 2k$, where $\End(V_N^{\otimes k})$ is the vector space of linear maps $V_N^{\otimes k} \rightarrow V_N^{\otimes k}$ and $\End_{S_N}(V_N^{\otimes k})$ is the subspace of maps that commute with the action of $S_N$.
%Note that we use the same symbol for elements $d \in P_k(N)$ and the corresponding element in $d \in \End_{S_N}(V_N^{\otimes k})$ in what follows. It will be clear from context if $d$ is acting on $V_N^{\otimes k}$.
We now move on to the trace in $V_N^{\otimes k}$ and show that it defines a non-degenerate bilinear form. First we need the following proposition that relates the two traces.
\begin{proposition}
Let $d \in P_k(N)$ and $D \in \End_{\SN}(\VN^{\otimes k})$ be the corresponding linear map. Define
\begin{equation}
	\Omega = \sum_{\lambda \in \Lambda_{k,\N}} \frac{\dimSN{\lambda }}{\dimPk{\lambda }} p_{\lambda }. \label{eq: omega}
\end{equation}
then,
\begin{equation}
	\Tr_{V_N^{\otimes k}}(d) = \tr(\Omega d). \label{eq: traces relation omega factor}
\end{equation}
\end{proposition}
\begin{proof}
Assume $\N \geq 2k$, then Schur-Weyl duality \eqref{eq: VN SW simple} implies that the trace in $V_N^{\otimes k}$ decomposes as
\begin{equation}
	\Tr_{V_N^{\otimes k}}(D) = \sum_{\lambda \in \Lambda_{k,\N}} \dimSN{\lambda} \chr^{{\lambda}}(d), \label{eq: traces V_N otimes k decomposition}
\end{equation}
where the sum is over the irreducible representations that appear in equation \eqref{eq: VN SW simple}.
Consequently, we can relate the two traces by substituting \eqref{eq: reg trace decomp} into each summand of \eqref{eq: traces V_N otimes k decomposition}
\begin{equation}
	\Tr_{V_N^{\otimes k}}(D) = \sum_{\lambda \in \Lambda_{k,\N}} \dimSN{\lambda} \chr^{{\lambda}}(d) =  \sum_{\lambda \in \Lambda_{k,\N}} \frac{\dimSN{\lambda }}{\dimPk{\lambda }} \tr(p_{\lambda }d) = \tr(\Omega d). \label{eq: Vn trace as reg trace}
\end{equation}
\end{proof}

We can now prove non-degeneracy.
\begin{proposition}\label{prop: nondegen VNtrace}
	Let $d,d' \in P_k(N)$ and $D,D' \in \End_{S_N}(\VN^{\otimes k})$ be the corresponding linear maps. Define the bilinear form $\langle - , - \rangle$
	\begin{equation}
		\langle d, d' \rangle = \Tr_{V_N^{\otimes k}}(D (D')^T) \label{eq: VN trace form}
	\end{equation}
	on $P_k(N)$. It is non-degenerate for $N \geq 2k$.
\end{proposition}
\begin{proof}	
We give a proof by contradiction.
Let $\mathcal{B}=\{b_1, \dots, b_{B(2k)}\}$ be a basis for $P_k(N)$ and $B_i, D \in \End_{\SN}(\VN^{\otimes k})$ the corresponding linear maps.
Assume $d \in P_k(N)$ is a non-zero element (with corresponding linear map $D$) satisfying
\begin{equation}
	\langle b_i, d \rangle = 0, \quad \forall i=1,\dots,B(2k).
\end{equation}
From \eqref{eq: traces relation omega factor} this implies
\begin{equation}
	\langle b_i, d \rangle = \Tr_{V_N^{\otimes k}}(B_i D^T) = \tr(\Omega b_i d^T) = 0, \quad \forall i=1,\dots,B(2k).
\end{equation}
Consequently, the element $d' = d^T \Omega \in P_k(N)$ is such that
\begin{equation}
	\tr(b_i d') = 0, \quad \forall i=1,\dots,B(2k),
\end{equation}
which contradicts the fact that the trace in the regular representation of $P_k(N)$ defines a non-degenerate bilinear form.
\end{proof}
The point of having a non-degenerate bilinear form is to define dual elements.
\begin{proposition}\label{def: dual PkN}
	Let $\mathcal{B} = \{b_1, \dots, b_{B(2k)}\}$  be a basis for $P_k(N)$ and $B_i$ the corresponding linear maps. Define
	\begin{equation}
		g_{ij} =  \Tr_{V_N^{\otimes k}}(B_i B_j)
	\end{equation}
	and
	\begin{equation}
		b_i^* = \sum_{j=1}^{B(2k)} (g^{-1})_{ij} b_j
	\end{equation}
	The dual basis with respect to the inner product \eqref{eq: VN trace form} is given by
	\begin{equation}
		b^\dagger_i = \sum_{j=1}^{B(2k)}(g^{-1})_{ji}b_j^T = (b_i^*)^T.
	\end{equation}
\end{proposition}
\begin{proof}
	Plugging the definition of the dual into \eqref{eq: VN trace form} gives
	\begin{equation}
		\langle b_i, b_j^\dagger \rangle =  \sum_{k=1}^{B(2k)} \Tr_{V_N^{\otimes k}}(B_i (g^{-1})_{jk}(B_k^T)^T ) = \sum_{k=1}^{B(2k)} \Tr_{V_N^{\otimes k}}(B_i (g^{-1})_{jk}B_k ) = \delta_{ij}.
	\end{equation}
\end{proof}

%\begin{corollary}
%It immediately follows (use proof by contradiction again) that the bilinear form given by
%\begin{equation}
%	\langle b_i, b_j \rangle = \Tr_{V_N^{\otimes k}}(b_i b_j^{T}) \equiv g_{ij}, \label{eq: bilinear form V_N otimes k}
%\end{equation}
%is non-degenerate and $g_{ij}$ is invertible.
%\end{corollary}
%The inverse matrix is used to define elements dual to $b_i$ which we denote $b_i^*$
%\begin{equation}
%	\boxed{b_i^{*} = \sum_{j=1}^{B(2k)} (g^{-1})_{ij}b_j.} \label{eq: algebra dual}
%\end{equation}
%It follows that we can define dual elements $b_i^*$ satisfying
%\begin{equation}
%	\langle b_i^*, b_j \rangle = \delta_{ij}. \label{eq: inner prod b_i dual b_j}
%\end{equation}

The dual elements are essential for proving orthogonality of matrix elements.
The proof also uses the following property of the bilinear form
\begin{equation}
	\langle b_i, b_j b_k \rangle = \langle b_i b_k^T, b_j \rangle = \langle b_j^T b_i, b_k \rangle. \label{eq: frob associ}
\end{equation}
The first step uses $(b_j b_k)^T = b_k^T b_j^T$ and the second step uses cyclicity of the trace.
%This is a slightly modified version of the associativity of a Frobenius pairing (see 2.3.10 in \cite{Kock2004}).

\subsection{Orthogonality of matrix elements.}
% As we will now prove, the definition of dual elements given in the previous subsection is such that 
%\begin{equation}
%	\sum_{i=1}^{B(2k)} D_{\alpha \beta}^{\lambda_1}(b_i)D_{\rho \sigma}^{\lambda_1'}((b_i^*)^T) \propto \delta_{\alpha \sigma}\delta_{\beta \rho}\delta^{\lambda_1 \lambda_1'}. \label{eq: almost grand orth theorem}
%\end{equation}
The matrix elements $D_{\alpha \beta}^{\lambda_1}(b_i)$ of irreducible representations of $P_k(N)$ are orthogonal. This is a generalization of the corresponding orthogonality theorem for group algebras (see section 3.15 in \cite{Hamermesh1962}). Before proving this, we need some intermediary results.
\begin{proposition}
	Let $d \in P_k(N)$ and $D \in \End_{\SN}(\VN^{\otimes k})$ the corresponding linear map. For $\N \geq 2k$ the irreducible representations of $P_k(N)$ can be written as
	\begin{equation}
		D_{\alpha \beta}^{\lambda}(d) = \sum_{\substack{i_1,\dots,i_k \\ i_{1'},\dots, i_{k'}}}(C^\lambda_{a, \beta})^{i_1 \dots i_k} D^{i_{1'} \dots i_{k'}}_{i_1 \dots i_k} (C^\lambda_{a, \alpha})^{i_{1'} \dots i_{k'}}. \label{eq: PkN irreps from clebsch}
	\end{equation}
\end{proposition}
\begin{proof}
	From \eqref{eq: VNk clebsch PkN equivariance} we have
	\begin{equation}
		(C^{\lambda}_{a,\alpha})^{i_1 i_2 \dots i_k} D(e_{i_1} \otimes \dots \otimes e_{i_k} )= E^\lambda_a \otimes D^\lambda_{\beta \alpha }(d)E^\lambda_\beta.
	\end{equation}
	The l.h.s. can be expanded as
	\begin{equation}
		\begin{aligned}
		(C^{\lambda}_{a,\alpha})^{i_1 i_2 \dots i_k} D^{i_{1'} \dots i_{k'}}_{i_1 \dots i_k}& (C^{b \beta}_{\lambda'})_{i_{1'} \dots i_{k'}}E^{\lambda'}_b \otimes E^{\lambda'}_\beta \\
		&= (C^{\lambda}_{a,\alpha})^{i_1 i_2 \dots i_k} D^{i_{1'} \dots i_{k'}}_{i_1 \dots i_k} (C_{b \beta}^{\lambda'})^{i_{1'} \dots i_{k'}}E^{\lambda'}_b \otimes E^{\lambda'}_\beta,
		\end{aligned}
	\end{equation}
	where the equality follows from \eqref{eq: clebsch orthonormality}.
	The r.h.s. can be re-written as
	\begin{equation}
		\delta_{ab} \delta^{\lambda {\lambda'}} E_b^{\lambda'} \otimes D^{\lambda'}_{\beta \alpha}(d)E_\beta^{\lambda'}.
	\end{equation}
	Equating the coefficients on the l.h.s. and r.h.s. gives
	\begin{equation}
		(C^{\lambda}_{a,\alpha})^{i_1 i_2 \dots i_k} D^{i_{1'} \dots i_{k'}}_{i_1 \dots i_k} (C_{b \beta}^{\lambda'})^{i_{1'} \dots i_{k'}} = \delta_{ab} \delta^{\lambda {\lambda'}}  D^{\lambda'}_{\beta \alpha}(d),
	\end{equation}
	which reduces to \eqref{eq: PkN irreps from clebsch} for $a=b, \lambda=\lambda'$ and swapping $\alpha \leftrightarrow \beta$.	
\end{proof}
A consequence of this relationship between irreducible representations of $P_k(N)$ and Clebsch-Gordan coefficients is the following.
\begin{corollary}
The irreducible representations of $P_k(N)$ as defined above satisfy
\begin{equation}
	D_{\alpha \beta}^{\lambda}(d^T) = D_{\beta \alpha}^{\lambda}(d), \qq{for $d \in P_k(N)$,} \label{eq: orthogonal reps of PkN}
\end{equation}
where the transpose $d^T$ is defined in Definition \ref{def: transpose diagram}.
\end{corollary}
\begin{proof}
This follows from
\begin{align}
	D_{\alpha \beta}^{\lambda}(d^T) &= C_{a,i_{1'} \dots i_{k'}}^{\lambda, \alpha}  C_{a,i_1 \dots i_k}^{\lambda, \beta} (d^T)^{i_{1'} \dots i_{k'}}_{i_1 \dots i_k} \\
	&= C_{a,i_{1'} \dots i_{k'}}^{\lambda, \alpha}  C_{a,i_1 \dots i_k}^{\lambda, \beta} (d)_{i_{1'} \dots i_{k'}}^{i_1 \dots i_k} = D_{\beta \alpha}^{\lambda}(d). \label{eq: transpose diagram is transpose matrix element}
\end{align}
\end{proof}

Because the bilinear form \eqref{eq: VN trace form} includes a transpose, the symmetrisation theorem (proposition 2.6 in \cite{AR90DissertCh1}) is modified slightly. We have the following version.
\begin{proposition}	
Let $C$ be a $\dimPk{\lambda} \times \dimPk{\lambda'}$ rectangular matrix, and $D^{\lambda}(d), D^{\lambda'}(d)$ be two irreducible matrix representations of $P_k(N)$ with dimension $\dimPk{\lambda}, \dimPk{\lambda'}$ respectively. Define the symmetrised matrix
\begin{equation}
	[C] = \sum_{i=1}^{B(2k)} D^{\lambda}(b_i)C[D^{\lambda'}(b_i^\dagger)]^T = \sum_{i=1}^{B(2k)} D^{\lambda}(b_i)CD^{\lambda'}(b_i^*).
\end{equation}
It satisfies
\begin{equation}
	D^{\lambda}(d)[C] = [C]D^{\lambda'}(d), \label{eq: sym theorem}
\end{equation}
for all $d \in P_k(N)$.
\end{proposition}
\begin{proof}
The proof follows from a straight-forward computation
\begin{equation}
	\begin{aligned}
		D^{\lambda}(d)[C] &= \sum_{i} D^{\lambda}(db_i)CD^{\lambda'}(b_i^*) = \sum_{i} D^{\lambda}(\sum_j \langle b_j^\dagger, db_i\rangle b_j)CD^{\lambda'}(b_i^*) \\
		&=\sum_j  D^{\lambda}(b_j)CD^{\lambda'}(\sum_{i} b_i^* \langle b_j^\dagger, db_i\rangle ) \\
		&=\sum_j  D^{\lambda}(b_j)CD^{\lambda'}(\sum_{i} b_i^* \langle d^Tb_j^\dagger, b_i\rangle ) \\
		&=\sum_j  D^{\lambda}(b_j)CD^{\lambda'}(\sum_{i} b_i^T \langle d^Tb_j^\dagger, b_i^\dagger \rangle ) \\
		&=\sum_j  D^{\lambda}(b_j)CD^{\lambda'}( (d^T b_j^\dagger)^T) \\
		&=[C]D^{\lambda'}(d),
	\end{aligned}
\end{equation}
where in the second equality we expand $db_i$ in terms of $b_j$ using the inner product, the third line uses the modified Frobenius associativity \eqref{eq: frob associ}, the fourth line uses $(b_i^*)^\dagger = b_i^T$.
\end{proof}
By Schur's lemma, $[C]$ is proportional to the identity matrix if and only if $\lambda = \lambda'$ and zero otherwise. For some constant $c^{\lambda_1}$,
\begin{equation}
	[C]_{\alpha \sigma} = \delta^{\lambda_1 \lambda_1'} c^{\lambda_1} \delta_{\alpha \sigma}.
\end{equation}

This observation is key to proving orthogonality of matrix elements of irreducible representations of $P_k(N)$.
\begin{proposition} \label{prop: orthogonality of matrix elements}
	Let $D_{\alpha \beta}^{\lambda}(b_i)$ be irreducible representations of $P_k(N)$ in the basis $\mathcal{B} = \{b_1, \dots, b_{B(2k)}\}$ and $b^\dagger$ the dual element defined in Definition \ref{def: dual PkN}. Then
	\begin{equation}
		\sum_{i=1}^{B(2k)} D_{\alpha \beta}^{\lambda}(b_i)D_{\rho \sigma}^{\lambda'}(b_i^\dagger) = \frac{1}{\dimSN{\lambda}} \delta_{\alpha \rho} \delta_{\beta \sigma} \delta^{\lambda \lambda'}. \label{eq: orthogonality Pk matrix elements}
	\end{equation}
\end{proposition}
\begin{proof}
The l.h.s. of \eqref{eq: orthogonality Pk matrix elements} reads
\begin{equation}
	\sum_{i} \big( D^{\lambda}(b_i)E_{\beta \rho}D^{\lambda'}( b_i^\dagger)  \big)_{\alpha \sigma} =  \sum_{i} \big( D^{\lambda}(b_i)E_{\beta \sigma}[D^{\lambda'}( b_i^\dagger) ]^T \big)_{\alpha \rho}= [E_{\beta \sigma}]_{\alpha \rho},
\end{equation}
where $E_{\beta \rho}$ is an elementary matrix with $0$ everywhere except in row $\beta$, column $\rho$ which has a $1$. Schur's lemma gives
\begin{equation}
	[E_{\beta \sigma}]_{\alpha \rho} = C^{\lambda}_{\beta \sigma} \delta_{\alpha \rho}\delta^{\lambda \lambda'}.
\end{equation}
It remains to determine the constant $C^{\lambda}_{\beta \sigma}$.

For this, set $\lambda = \lambda'$ and note that
\begin{equation}		
\begin{array}{@{} l l l @{}}
	\sum_{\alpha = 1}^{\dimPk{\lambda}}[E_{\beta \sigma}]_{\alpha \alpha} = \sum_{i=1}^{B(2k)}	\sum_{\alpha = 1}^{\dimPk{\lambda}} D_{\alpha \beta}^{\lambda}(b_i)D_{\alpha \sigma}^{\lambda}(b_i^\dagger) &=& \sum_{i} \Tr(D^{\lambda}(b_i)E_{\beta \sigma}[D^{\lambda}( b_i^\dagger) ]^T) \\
	\, \, \rotaterelation{=} & & \rotaterelation{=} \\
		\sum_{\alpha = 1}^{\dimPk{\lambda}} C^{\lambda}_{\beta \sigma} \delta_{\alpha \alpha} = C^{\lambda}_{\beta \sigma} \dimPk{\lambda}&= &\sum_{i} \Tr(D^{\lambda}((b_i^\dagger)^Tb_i)E_{\beta \sigma}).
\end{array}\label{eq: orthogonality step 2}
\end{equation}
As we now show, the matrix $D^{\lambda}(\sum_i (b_i^\dagger)^Tb_i)$ is proportional to the identity matrix.

In particular, the element $\sum_i (b_i^\dagger)^Tb_i$ is related to the $\Omega$ factor in \eqref{eq: Vn trace as reg trace}. Since
\begin{equation}
	(b_i^\dagger)^T = b_i^* = \sum_k(g^{-1})_{ki}b_k
\end{equation}
we have
\begin{equation}
	\sum_i (b_i^\dagger)^Tb_i = \sum_{i,k}(g^{-1})_{ki}b_k b_i = \sum_i b_i(b_i^\dagger)^T,
\end{equation}
where the last step follows by $(g^{-1})_{ki}$ being symmetric.
Define
\begin{equation}
	\Omega^{-1} = \sum_{\lambda \in \Lambda_{k,\N}} \frac{\dimPk{\lambda}}{\dimSN{\lambda}} p_\lambda,
\end{equation}
and consider
\begin{equation}
	\begin{aligned}
		\tr(d\sum_i (b_i^\dagger)^T b_i ) &= \tr(d\sum_i b_i(b_i^\dagger)^T)\\
		&= \Tr_{V_N^{\otimes k}}(\Omega^{-1} d\sum_i b_i(b_i^\dagger)^T) \\
		&=\sum_i \langle \Omega^{-1} d b_i, b_i^\dagger \rangle \\
		&=\sum_i\mathrm{Coeff}(b_i, \Omega^{-1} d b_i) \\
		&=\tr(\Omega^{-1}d).
	\end{aligned}
\end{equation}
This implies that
\begin{equation}
	\tr(d(\sum_i (b_i^\dagger)^T b_i -\Omega^{-1})) = 0,
\end{equation}
for all $d \in P_k(\N)$ and by the non-degeneracy of the trace in the regular representation
\begin{equation}
	\sum_i (b_i^\dagger)^T b_i  = \Omega^{-1}.
\end{equation}

Thus, substituting this into \eqref{eq: orthogonality step 2} we find
\begin{equation}
	\begin{aligned}
		C^\lambda_{\beta \sigma} \dimPk{\lambda} &= \Tr(D^{\lambda}(\Omega^{-1})E_{\beta \sigma}) \\
		&=  \sum_{\lambda' \in \Lambda_{k,\N}}\frac{\dimPk{\lambda'}}{\dimSN{\lambda'}} \Tr(D^{\lambda}( p_{\lambda'})E_{\beta \sigma}) \\
		&=  \frac{\dimPk{\lambda}}{\dimSN{\lambda}} \delta_{\beta \sigma},
	\end{aligned}
\end{equation}
where the last step uses
\begin{equation}
	D^{\lambda}(p_{\lambda'}) = \delta^{\lambda \lambda'}, \quad \Tr(E_{\beta \sigma}) = \delta_{\beta \sigma}.
\end{equation}
Therefore,
\begin{equation}
	C^\lambda_{\beta \sigma} = \frac{1}{\dimSN{\lambda}} \delta_{\beta \sigma}.
\end{equation}
\end{proof}

\subsection{Matrix units for $P_k(\N)$.} \label{apx: subsec matrix unit PkN}
We will use the orthogonality in Proposition \ref{prop: orthogonality of matrix elements} to prove that the following linear combinations of elements form a basis of matrix units.
\begin{proposition}
	Let $\mathcal{B} = \{b_1, \dots, b_{B(2k)}\}$ be a basis for $P_k(\N)$, $b_i^\dagger$ the corresponding duals with respect to the inner product \eqref{eq: VN trace form}, $b_i^T$ the transpose element defined in Definition \ref{def: transpose diagram} and $\lambda \in \Lambda_{k,\N}$ be a labelling set of non-isomorphic irreducible representations of $P_k(\N)$. Then elements
	\begin{equation}
		Q_{\alpha \beta}^\lambda = \sum_{i=1}^{B(2k)} \dimSN{\lambda} D_{\alpha \beta}^\lambda(b_i^\dagger)b_i, \label{def: matrix unit PkN}
	\end{equation}
	form a basis for $P_k(\N)$ with structure constants given by
	\begin{equation}
		Q_{\alpha \beta}^\lambda Q_{\alpha' \beta'}^{\lambda'} = \delta^{\lambda \lambda'}\delta_{\beta \alpha'} Q_{\alpha \beta'}^\lambda.
	\end{equation}
\end{proposition}
\begin{proof}
	Plugging in the definition \eqref{def: matrix unit PkN} we have
	\begin{align}
		Q_{\alpha \beta}^\lambda Q_{\alpha' \beta'}^{\lambda'}  &= \sum_{i,j=1}^{B(2k)} \dimSN{\lambda}  \dimSN{\lambda'} D_{\alpha \beta}^\lambda(b_i^\dagger) D_{\alpha' \beta'}^{\lambda'}(b_j^\dagger)\underbrace{b_i b_j }_{\sum_k \langle b_i b_j, b_k^\dagger \rangle b_k}\\
		&= \sum_{i,j,k=1}^{B(2k)} \dimSN{\lambda}  \dimSN{\lambda'} D_{\alpha \beta}^\lambda(b_i^\dagger) D_{\alpha' \beta'}^{\lambda'}(b_j^\dagger)\underbrace{\langle b_i b_j, b_k^\dagger \rangle}_{\langle b_j, b_i^T b_k^\dagger {\rangle}} b_k \label{eq: matrix unit frob prop step} \\
		&=\sum_{i,j,k=1}^{B(2k)} \dimSN{\lambda}  \dimSN{\lambda'} D_{\alpha \beta}^\lambda(b_i^\dagger) D_{\alpha' \beta'}^{\lambda'}(b_j^\dagger \langle b_j, b_i^T b_k^\dagger \rangle) b_k \\
		&=\sum_{i,k=1}^{B(2k)} \dimSN{\lambda}  \dimSN{\lambda'} D_{\alpha \beta}^\lambda(b_i^\dagger) D_{\alpha' \beta'}^{\lambda'}(b_i^T b_k^\dagger ) b_k \\
		&=\sum_{i,k=1}^{B(2k)} \dimSN{\lambda}  \dimSN{\lambda'} D_{\alpha \beta}^\lambda(b_i^\dagger) \sum_{\rho=1}^{\dimPk{\lambda'}} D_{\alpha' \rho}^{\lambda'}(b_i^T) D^{\lambda'}_{\rho \beta'}(b_k^\dagger) b_k, \label{eq: matrix unit last step}
	\end{align}
	where \eqref{eq: matrix unit frob prop step} uses \eqref{eq: frob associ}.
	The last line is simplified by using \eqref{eq: orthogonal reps of PkN} and orthogonality \eqref{eq: orthogonality Pk matrix elements},
	\begin{equation}
		\sum_{i=1}^{B(2k)} D_{\alpha \beta}^\lambda(b_i^\dagger) D_{\alpha' \rho}^{\lambda'}(b_i^T) = \sum_{i=1}^{B(2k)} D_{\rho \alpha'}^{\lambda'}(b_i) D_{\alpha \beta}^\lambda(b_i^\dagger) = \frac{1}{\dimSN{\lambda}}\delta^{\lambda' \lambda} \delta_{\rho \alpha} \delta_{\alpha' \beta}.
	\end{equation}
	Plugging this into \eqref{eq: matrix unit last step} gives
	\begin{equation}
		Q_{\alpha \beta}^\lambda Q_{\alpha' \beta'}^{\lambda'}  = \sum_{k=1}^{B(2k)}\sum_{\rho=1}^{\dimPk{\lambda'}} \delta^{\lambda' \lambda}\delta_{\rho \alpha} \delta_{\alpha' \beta} \dimSN{\lambda'} D^{\lambda'}_{\rho \beta'}(b_k^\dagger) b_k = \delta^{\lambda' \lambda} \delta_{\alpha' \beta} Q^{\lambda}_{\alpha \beta'}.
	\end{equation}
\end{proof}

\begin{corollary}\label{cor: d on Q}
	Equipped with a matrix unit basis of $P_k(N)$ we use this to show
	\begin{align} \label{eq: d on Q left}
		d Q^{\lambda}_{\alpha \beta} = \sum_{\sigma = 1}^{\dimPk{\lambda}} D^{\lambda}_{ \sigma \alpha}(d) Q^{\lambda}_{\sigma \beta}, \quad Q^{\lambda}_{\alpha \beta} d= \sum_{\sigma = 1}^{\dimPk{\lambda}}D^{\lambda}_{  \beta \sigma}(d) Q^{\lambda}_{\alpha \sigma}.
	\end{align}
\end{corollary}
\begin{proof}
	Using the definition we have
	\begin{align}
		d Q^{\lambda}_{\alpha \beta}&=  \sum_{i=1}^{B(2k)} \dimSN{\lambda} D_{\alpha \beta}^\lambda(b_i^\dagger)db_i \\
		&=\sum_{i,j=1}^{B(2k)} \dimSN{\lambda} D_{\alpha \beta}^\lambda(b_i^\dagger) \langle db_i, b_j^\dagger \rangle b_j  \\
		&=\sum_{i,j=1}^{B(2k)} \dimSN{\lambda} D_{\alpha \beta}^\lambda(b_i^\dagger) \langle b_i, d^T b_j^\dagger \rangle b_j  \\
		&=\sum_{i,j=1}^{B(2k)} \dimSN{\lambda} D_{\alpha \beta}^\lambda(b_i^\dagger \langle b_i, d^T b_j^\dagger \rangle)  b_j \\
		&=\sum_{j=1}^{B(2k)} \dimSN{\lambda} D_{\alpha \beta}^\lambda(d^T b_j^\dagger)  b_j \\
		&=\sum_{\sigma = 1}^{\dimPk{\lambda}}\sum_{j=1}^{B(2k)} \dimSN{\lambda} D_{\alpha \sigma}^\lambda(d^T) D_{\sigma \beta}^\lambda(b_j^\dagger)  b_j \\
		&=\sum_{\sigma = 1}^{\dimPk{\lambda}}D_{\sigma \alpha }^\lambda(d) Q_{\sigma \beta}^\lambda,
	\end{align}
	where in the last line we used \eqref{eq: orthogonal reps of PkN}. The proof is identical for the right action.
\end{proof}
In summary, we have found an Artin-Wedderburn decomposition
\begin{equation}
	P_k(N)  \cong \bigoplus_{\lambda \in \Lambda_{k,\N}} \End(Z_\lambda, Z_\lambda) \cong \bigoplus_{\lambda \in \Lambda_{k,\N}} Z_\lambda \otimes Z_\lambda, \label{eq: PkN AW decomp}
\end{equation}
where $\End(Z_\lambda, Z_\lambda)$ is the algebra of $m^{\lambda}_{k,\N} \times m^{\lambda}_{k,\N}$ matrices.

\section{Construction of matrix units}\label{sec: construction of units}
In this subsection we will give a construction of matrix units of $P_k(N)$ starting from the diagram basis. The construction uses the notion of $k$-duals and $(k+\tfrac{1}{2})$-dual (Definition \ref{def: kdual}, \ref{def: half dual}) to find a complete set of commuting elements in $P_k(N)$, whose eigenvectors are the matrix units. We show that the simultaneous eigenspaces of these elements are one-dimensional and produce projection operators for each eigenspace.

\subsection{Eigenvalues of duals of central elements.}
%Consider the decomposition \eqref{eq: VN SW simple} of $\VN^{\otimes k}$. It has a basis \eqref{eq: VN SW basis} $E_a^\lambda \otimes E_\alpha^\lambda$. From Theorem \ref{thm: multi is vac tab}, the indices $\alpha = 1,\dots, \dimPk{\lambda}$ can be understood as a set of vacillating tableaux with shape $\lambda$ and length $k$. We will now introduce central elements in $\mathcal{Z}[\mathbb{F}(\SN)]$, and $\mathcal{Z}[\mathbb{F}(S_{\N-1})]$ whose eigenvalues can be used to distinguish the vacillating tableaux.
%Through the duality in Definition \ref{def: kdual} \ref{def: half dual}, between the $\mathcal{Z}[\mathbb{F}(\SN)]$ and $\Zdual[P_k(N)]$, we will be able to leverage these results to give an explicit construction of matrix units.

For $n=\N$ or $\N-1$, let $z \in \mathcal{Z}[\mathbb{F}(S_n)]$ be a central element and $\P^\lambda$ an irreducible representation of $S_n$. Schur's lemma implies that
\begin{equation}
	\P^{\lambda}_{ab}(z) = \frac{\chr^{}_\lambda(z)}{\dimSN{\lambda}}\delta_{ab},
\end{equation}
where $\chr_\lambda(z) = \sum_a \P^\lambda_{aa}(z)$.
\begin{definition}[Normalized character]	
The combination
\begin{equation}
	\hat{\chr}^{}_\lambda(z) = \frac{\chr^{}_\lambda(z)}{\dimSN{\lambda}},
\end{equation}
is known as a normalized character.
\end{definition}
A particularly important instance of this is the conjugacy class basis element
\begin{equation}
	z_n = \sum_{1 \leq i < j \leq n} (ij) \in \mathcal{Z}[\mathbb{F}(S_{n})] \label{eq: z_n}
\end{equation}
Normalized characters of $z_n$ are expressible in terms of combinatorial quantities of Young diagrams.
\begin{theorem}
	Let $\lambda \vdash n=\N, \N-1$ and $\YT{\lambda}$ the corresponding Young diagram. Then
	\begin{equation}
		\hat{\chr}_\lambda(z_n) = \sum_{(i,j) \in \YT{\lambda}} (j-i)\label{eq: norm characters of T2s},
	\end{equation}
	where $(i,j)$ corresponds to the cell in the $i$th row and $j$th column of the Young diagram (the top left box has coordinate $(1,1)$).
\end{theorem}
\begin{proof}
	See \cite[Example 7 in Section I.7]{Macdonald1998}.
\end{proof}
\begin{example}\label{ex: norm chars}
	Some relevant examples for $n=\N$ are
	\begin{align}
		\hat{\chr}_{[\N]}(z_\N) = 0 + 1 + 2 + \dots + (\N -1) = \frac{\N(\N-1)}{2}, \\
		\hat{\chr}_{[\N-1,1]}(z_\N) = 0 + 1 + 2 \dots + (\N-2) - 1 = \frac{\N(\N-3)}{2}, \\
		\hat{\chr}_{[\N-2,2]}(z_\N) = 0 + 1 + 2 \dots + (\N-3) - 1 + 0 = \frac{(\N-1)(\N-4)}{2}, \\
		\hat{\chr}_{[\N-2,1,1]}(z_\N) = 0 + 1 + 2 \dots + (\N-3) - 1 - 2= \frac{\N(\N-5)}{2}.
	\end{align}
\end{example}
The relevance of this result is that these normalized characters form eigenvalues of the following linear operator.
\begin{proposition}
	Define the following linear operator on $\VN^{\otimes k}$
	\begin{equation}
		\P_{\N} = \sum_{1 \leq i < j \leq \N} \P_{(ij)}.
	\end{equation}
	It has eigenvalues $\hat{\chr}^{}_\lambda(z_\N)$ for $\lambda \in \Lambda_{k,\N}$.
\end{proposition}
\begin{proof}
	The eigenvectors of $\P_{\N}$ are
	\begin{equation}
		E^\lambda_a \otimes E^\lambda_\alpha
	\end{equation}
	as defined in \eqref{eq: VNk clebsch SN equivariance}, since
	\begin{equation}
		\P_{\N}(E^\lambda_a \otimes E^\lambda_\alpha) = \P^\lambda_{ba}(z_\N)E_b^\lambda \otimes E^\lambda_\alpha = \hat{\chr}^{}_\lambda(z_\N)E^\lambda_a \otimes E^\lambda_\alpha.
	\end{equation}
\end{proof}

We can use this result for all $l \leq k$ to produce eigenvalues corresponding to characters of irreducible representations of $\SN$ in each tensor product factor.
\begin{proposition}
Define the operator $\P_{\N}^{(l)}$ by
\begin{equation}
	\P_{\N}^{(l)}(e_{i_1} \otimes \dots \otimes e_{i_k}) = \P_{\N}(e_{i_1} \otimes \dots \otimes e_{i_l}) \otimes e_{i_{l+1}} \otimes \dots \otimes e_{i_k}.
\end{equation}
It has eigenvalues $\hat{\chr}^{}_\lambda(z_\N)$ for $\lambda \in \Lambda_{l,\N}$.
\end{proposition}
\begin{proof}
This follows analogously to the previous proof. An eigenbasis corresponds to vectors of the form
\begin{equation}
	E^\lambda_a \otimes E^\lambda_\alpha \otimes e_{i_{l+1}} \otimes \dots \otimes e_{i_k},
\end{equation}
for $\lambda \in \Lambda_{l,\N}$.
\end{proof}

We also want to produce eigenvalues corresponding to characters of irreducible representations of $S_{\N-1}$.
\begin{proposition}
Define
\begin{equation}
	\P_{\N-1} = \sum_{1 \leq i < j \leq \N-1} \P_{(ij)}
\end{equation}
and
\begin{equation}
	\P^{(l + \frac{1}{2})}_{\N}(e_{i_1} \otimes \dots \otimes e_{i_k}) = \P_{\N-1}(e_{i_1} \otimes \dots \otimes e_{i_l}) \otimes e_{i_{l+1}} \otimes \dots \otimes e_{i_k}.
\end{equation}
The eigenvalues of $\P^{(l+\frac{1}{2})}_{\N}$ are normalized characters $\hat{\chr}^{}_{\lambda}(z^{}_{\N-1})$ for $\lambda \in \Lambda_{l+\frac{1}{2}, \N}$.
\end{proposition}
\begin{proof}
	Let $\lambda \vdash \N, \lambda' \in \lambda - {\scriptstyle \ydiagram{1}}$, and
	\begin{align}
		V_\lambda &= \Span(E^{\lambda}_a \, \vert \, a=1,\dots, \dimSN{\lambda}), \\
		V_{\lambda'} &= \Span(E^{\lambda'}_{\underline{a}} \, \vert \, \underline{a}=1,\dots, \dimSN{\lambda'}),
	\end{align}
	be irreducible representations of $\SN$ and $S_{\N-1}$, respectively.
	We assume the two bases to be orthonormal, and therefore related by the restriction matrix $R^{\lambda'}_\lambda$ as
	\begin{equation}
		E^{\lambda}_a = \sum_{\lambda' \in \lambda - {\scriptstyle \ydiagram{1}}} \sum_{\underline{a}}(R^{\lambda'}_\lambda)^{\underline{a}}_a E^{\lambda'}_{\underline{a}},
	\end{equation}
	or
	\begin{equation}
		E^{\lambda'}_{\underline{a}} = \sum_a (R^{\lambda'}_\lambda)_{\underline{a}}^a E_{a}^{\lambda}.
	\end{equation}
	Therefore, $\VN^{\otimes k}$ has a basis
	\begin{equation}
		\{E^{\lambda'}_{\underline{a}} \otimes E^\lambda_{\alpha} \otimes e_{i_{l+1}} \otimes \dots \otimes e_{i_k}\,\vert \, \lambda \in \Lambda_{l,\N}, \lambda' \in \lambda - {\scriptstyle \ydiagram{1}}\},
	\end{equation}
	with elements satisfying
	\begin{align}
		\P_{\N-1}(E^{\lambda'}_{\underline{a}} \otimes E^\lambda_{\alpha}) \otimes e_{i_{l+1}} \otimes \dots \otimes e_{i_k} &=\sum_{\underline{b}}\P^{\lambda'}_{\underline{b}\, \underline{a}}(z_{\N-1}) E^{\lambda'}_{\underline{b}} \otimes E^\lambda_{\alpha} \otimes e_{i_{l+1}} \otimes \dots \otimes e_{i_k} \\
		&= \hat{\chr}^{}_\lambda(z_{\N-1})E^{\lambda'}_{\underline{a}} \otimes E^\lambda_{\alpha} \otimes e_{i_{l+1}} \otimes \dots \otimes e_{i_k}.
	\end{align}
\end{proof}

The $l$-duals of $\P^{(l)}_{\N}$ and $(l+\tfrac{1}{2})$-duals of $\P^{(l+\frac{1}{2})}_{\N}$ are in fact known. We will need the special cases $l=1,2$ and $l+\tfrac{1}{2} = 1+\tfrac{1}{2}$.
\begin{theorem}\label{thm: murphys}
	\begin{enumerate}
	\item[]
	\item[(a)]
	Define
	\begin{equation}
		Z_1 = \PAdiagram{1}{}{},
	\end{equation}
	it acts on $\VN$ as
	\begin{equation}
		\P_{\N}^{(1)} - \binom{\N}{2} + \N,
	\end{equation}
	and is central in $P_1(\N)$.
	\item[(b)]
	Define
	\begin{align}
		Z_2 &= \PAdiagram{2}{-1/1,-2/2}{} + \PAdiagram{2}{-1/1}{} + \PAdiagram{2}{-2/2}{} - \PAdiagram{2}{1/2,-2/2}{} - \PAdiagram{2}{-1/1,1/2}{} - \PAdiagram{2}{-1/1,-1/-2}{} - \PAdiagram{2}{-1/-2,-2/2}{} +   \PAdiagram{2}{-1/2,-2/1}{} + \N \PAdiagram{2}{-1/1,-2/2, 1/2, -1/-2}{}, \\
		Z_{1\frac{1}{2}} &= \PAdiagram{2}{-1/1,-2/2}{} + \PAdiagram{2}{-2/2}{} -\PAdiagram{2}{-1/-2,-2/2}{} - \PAdiagram{2}{1/2,-2/2}{} + \N \PAdiagram{2}{-1/1,-2/2, -1/-2, 1/2}{}.
	\end{align}
	They act on $\VN^{\otimes 2}$ as
	\begin{equation}
		\P_{\N}^{(2)} - \binom{\N}{2} + 2\N, \quad \P_{\N}^{(1+\frac{1}{2})} - \binom{\N}{2} +2\N -1
	\end{equation}
	respectively. Further, $Z_2$ is central in $P_2(\N)$ and $Z_{1\frac{1}{2}}$ is central in $P_{1+\frac{1}{2}}(\N)$.
	\end{enumerate}
\end{theorem}
\begin{proof}
	See Example \ref{ex: kdual} for statement (a), Example \ref{ex: khalf dual} for $Z_{1\frac{1}{2}}$ and \cite[Theorem 3.35 and examples below Equation 3.32]{Halverson2005} for the general case and $Z_2$.
\end{proof}
This theorem leads to an important corollary.
\begin{corollary}
The inclusion of $Z_1$ into $P_{1+\frac{1}{2}}(\N)$ is
\begin{equation}
	Z_1 \otimes 1 = \PAdiagram{2}{-2/2}{}.
\end{equation}
It commutes with $Z_{1\frac{1}{2}}$, which is central in $P_{1+\frac{1}{2}}(\N)$. It also commutes with $Z_2$, which is central in $P_2(\N)$. That is $Z_1 \otimes 1, Z_{1\frac{1}{2}}, Z_2$ form a set of commuting elements in $P_2(\N)$.
\end{corollary}

For the construction of matrix units, we are interested in the eigenvalues of these dual elements in the regular representation of $P_2(\N)$. From Schur's lemma, the eigenvalues will be normalized characters of irreducible representations of $P_1(\N), P_{1+\frac{1}{2}}(\N)$ and $P_2(\N)$, respectively. We will now show that Schur-Weyl duality has implications for the normalized characters of dual elements.
\begin{proposition}	\label{prop: kduals same normalized characters}
Let $d_{\N} \in \Zdual[P_k(\N)]$ be the $k$-dual of $z_\N$, $\hat{\chr}_\lambda(z_\N)$ the normalized $\SN$ character of $z_\N$ and $\hat{\chr}^\lambda(d_{\N})$ the normalized $P_k(\N)$ character of $d_\N$ for $\lambda \in \Lambda_{k,\N}$, then
\begin{equation}
	\hat{\chr}_\lambda(z_\N) = \hat{\chr}^\lambda(d_{\N}), 
\end{equation}
\end{proposition}
\begin{proof}
	Let $d_{\N}$ have the expansion
	\begin{equation}
		d_{\N} = \sum_{\pi \in \setpart{[k \vert k']}} a_\pi d_\pi,
	\end{equation}
	and define
	\begin{equation}
		\D_{\N} = \sum_{\pi \in \setpart{[k \vert k']}} a_\pi D_\pi.
	\end{equation}
	$k$-duality demands (see Definition \ref{def: kdual})
	\begin{equation}
		\P_{\N}(e_{i_1} \otimes \dots \otimes e_{i_k}) = D_{\N}(e_{i_1} \otimes \dots \otimes e_{i_k}).
	\end{equation}
	Therefore, using \eqref{eq: VNk clebsch SN equivariance} and \eqref{eq: VNk clebsch PkN equivariance} we have
	\begin{equation}
		\begin{aligned}
			\sum_b \P_{ba}^\lambda(z_\N)E^\lambda_b \otimes E^\lambda_\alpha &= (C^\lambda_{a, \alpha})^{i_1 \dots i_k}\sum_{1 \leq i < j \leq \N} \P_{\N}(e_{i_1} \otimes \dots \otimes e_{i_k}) \\
			&= (C^\lambda_{a, \alpha})^{i_1 \dots i_k}D_\N(e_{i_1} \otimes \dots \otimes e_{i_k}) \\
			&= \sum_\beta D^\lambda_{\beta \alpha}(d_{\N})E^\lambda_a \otimes E^\lambda_\beta.
		\end{aligned}
	\end{equation}
	In other words,
	\begin{equation}
		\P_{ba}^\lambda(z_\N) \delta_{\beta \alpha} = \delta_{ab} D^\lambda_{\beta \alpha}(d_{\N}).
	\end{equation}
	Taking the trace on both sides gives
	\begin{equation}
		\chr_\lambda(z_\N) \dimPk{\lambda} = \dimSN{\lambda} \chr^\lambda(d_{\N}) \Rightarrow \hat{\chr}_\lambda(z_\N) = \hat{\chr}^\lambda(d_{\N}),
	\end{equation}
	which is the claim in the proposition.
\end{proof}
The analogous result holds for normalized characters of $(k+\tfrac{1}{2})$ dual elements.
\begin{proposition}	\label{prop: khalf dual same normalized character}
	Let $d_{\N-1} \in \Zdual[P_{k+\frac{1}{2}}(\N)]$ be the $(k+\frac{1}{2})$-dual of $z_{\N-1}$, then
	\begin{equation}
		\hat{\chr}_\lambda(z_{\N-1}) = \hat{\chr}^\lambda(d_{\N-1}), 
	\end{equation}
	for $\lambda \in \Lambda_{k+\frac{1}{2},\N}$.
\end{proposition}
\begin{proof}
	Let $d_{\N-1}$ have the expansion
	\begin{equation}
		d_{\N-1} = \sum_{\pi \in \setpart{[k+1 \vert (k+1)']}} a_\pi d_\pi,
	\end{equation}
	and define
	\begin{equation}
		\Delta_{\N-1} = \sum_{\pi \in \setpart{[k+1 \vert (k+1)']}} a_\pi \Delta_\pi.
	\end{equation}
	$(k+\frac{1}{2})$-duality demands (see Definition \ref{def: half dual})
	\begin{equation}
		\P_{\N-1}(e_{i_1} \otimes \dots \otimes e_{i_k}) = \Delta_{\N-1}(e_{i_1} \otimes \dots \otimes e_{i_k}).
	\end{equation}
	Using the $S_{\N-1} \times P_{k+\frac{1}{2}}$ irreducible basis \eqref{eq: SN-1 SW basis} for $\VN^{\otimes k}$ this gives
	\begin{equation}
		\P_{ba}^\lambda(z_{\N-1}) \delta_{\beta \alpha} = \delta_{ab} D^\lambda_{\beta \alpha}(d_{\N-1}),
	\end{equation}
	for $\lambda \in \Lambda_{k+\frac{1}{2}, \N}$.
	Taking the trace on both sides gives
	\begin{equation}
		\chr_\lambda(z_{\N-1}) \dim Z_\lambda^{1/2} = \dimSN{\lambda} \chr^\lambda(d_{\N-1}) \Rightarrow \hat{\chr}_\lambda(z_{\N-1}) = \hat{\chr}^\lambda(d_{\N-1}).
	\end{equation}
\end{proof}
\begin{example}\label{ex: norm chars of Z}
	It will be useful to consider some normalized characters of $Z_1$.
	\begin{align}
		\hat{\chr}^{[\N]}(Z_1) &= \hat{\chr}_{[\N]}(z_\N) - \binom{\N}{2} + \N = \N \\
		\hat{\chr}^{[\N-1,1]}(Z_1) &= \hat{\chr}_{[\N-1,1]}(z_\N) - \binom{\N}{2} + \N = \frac{\N(\N-3) - \N(\N-1)}{2} + \N = -\N
	\end{align}
\end{example}

\subsection{Vacillating tableaux and projectors.}
In equation \eqref{def: inductive basis} we defined a set of elements coming from repeated restriction of an irreducible representation $Z_\lambda$ of $P_k(\N)$. We will now use this definition to show that there exists a set of basis vectors that are eigenvectors of elements dual to $\P_{\N}, \P_{\N-1}$.

The following definition will be useful.
\begin{definition}
Let $\vactab = (\lambda^{(0)} = [\N], \lambda^{(\frac{1}{2})} = [\N-1], \lambda^{(1)}, \lambda^{(\frac{3}{2})},\dots,\lambda^{(k)})$ be a vacillating tableaux of shape $\lambda$ and length $k$ (see Definition \ref{def: vac tableu}). Consider the set of orthogonal projectors
\begin{equation}
	\{p^{}_{\lambda^{(l)}} \in P_{l}(\N) \subset P_k(\N)\, \vert \, l = 1,\frac{3}{2}, \dots, k\},
\end{equation}
to the corresponding irreducible representations of $P_{l}(\N)$ and define
\begin{equation}
	p_\vactab = \prod_l p^{}_{\lambda^{(l)}}.
\end{equation}
\end{definition}
Note that the projectors $p_{\lambda^{(l)}}$ commute among themselves because $p_{\lambda^{(l)}}$ is central in $P_l(\N)$ and therefore commutes with $p_{\lambda^{(m)}}$ for $m \leq l$ since $P_m(\N) \subseteq P_l(\N)$. It follows from orthogonality of projectors that $p_{\vactab} p_{\vactab'} = \delta_{\vactab \vactab'} p_\vactab$.

The image of $p_\vactab$ in $Z_\lambda$ has a particularly simple expression in terms of elements in equation \eqref{def: inductive basis}. To see this we will prove the following.
\begin{proposition} \label{prop: image of P_vactab}
	Let $\vactab = (\lambda^{(0)} = [\N], \lambda^{(\frac{1}{2})} = [\N-1], \lambda^{(1)}, \lambda^{(\frac{3}{2})},\dots,\lambda = \lambda^{(k)})$ be a vacillating tableau, $p_\vactab$ be as above and consider $E^{\lambda}_\alpha$ as a sum of elements defined in \eqref{def: inductive basis}, then
	\begin{equation}
		\sum_{\gamma} D^{\lambda}_{\gamma \alpha}(p_\vactab)E^{\lambda}_\gamma =  E^{\vactab}_\alpha
	\end{equation}
\end{proposition}
\begin{proof}
Plugging in \eqref{eq: vac tab expansion} we have
\begin{equation}
	\sum_\gamma D^{\lambda}_{\gamma \alpha}(p^{}_{\vactab})E^{\lambda}_\gamma = \sum_{\vactab'} \sum_\gamma D^{\lambda}_{\gamma \alpha}(p^{}_{\vactab})E^{\vactab'}_\gamma.
\end{equation}
Consider the contribution on the r.h.s. from a single vacillating tableaux $\vactab' = (\rho^{(0)} = [\N], \rho^{(\frac{1}{2})} = [\N-1], \rho^{(1)}, \rho^{(\frac{3}{2})},\dots,\lambda^{(k)})$ of the same shape as $\vactab$. This is equal to
\begin{equation}
\sum_\gamma D^{\lambda}_{\gamma \alpha}(p^{}_{\vactab}) (\R^{\rho^{(1+\frac{1}{2})} \rightarrow \rho^{(1)}} \R^{\rho^{(2)} \rightarrow \rho^{(1+\frac{1}{2})}}\dots \R^{\rho^{(k-\frac{1}{2})} \rightarrow \rho^{(k-1)}}\R^{\lambda^{(k)} \rightarrow \rho^{(k-\frac{1}{2})}} )_{1 \gamma}E^{\rho^{(1)}}_1.
\end{equation}
Using the equivariance of the restriction matrix this gives
 \begin{equation}
 	\left(\begin{aligned}
 	&D^{\rho^{(1)}}(p^{}_{\lambda^{(1)}})\R^{\rho^{(1+\frac{1}{2})} \rightarrow \rho^{(1)}}
	 D^{\rho^{(1+\frac{1}{2})}}(p^{}_{\lambda^{(1+\frac{1}{2})}})\R^{\rho^{(2)} \rightarrow \rho^{(1+\frac{1}{2})}}
	 \dots\\
	 &
	 D^{\rho^{(k-1)}}(p^{}_{\lambda^{(k-1)}})\R^{\rho^{(k-\frac{1}{2})} \rightarrow \rho^{(k-1)}}
	 D^{\rho^{(k-\frac{1}{2})}}(p^{}_{\lambda^{(k-\frac{1}{2})}})\R^{\lambda^{(k)} \rightarrow \rho^{(k-\frac{1}{2})}}
	 D^{\lambda^{(k)}}(p^{}_{\lambda^{(k)}})
 	\end{aligned}\right)_{1 \alpha}E^{\rho^{(1)}}_1,
 \end{equation}
or diagrammatically
\begin{equation}
	\vcenter{\hbox{

			\tikzset{every picture/.style={line width=0.75pt}} %set default line width to 0.75pt        
			
			\begin{tikzpicture}[x=0.75pt,y=0.75pt,yscale=-1,xscale=1]
				%uncomment if require: \path (0,776); %set diagram left start at 0, and has height of 776
				
				%Straight Lines [id:da5136673255830186] 
				\draw    (390,360) -- (390,380) ;
				%Shape: Rectangle [id:dp46664125173882454] 
				\draw   (370,380) -- (410,380) -- (410,400) -- (370,400) -- cycle ;
				%Straight Lines [id:da5547543140390041] 
				\draw    (390,400) -- (390,420) ;
				%Shape: Rectangle [id:dp684692874496857] 
				\draw   (380,420) -- (400,420) -- (400,440) -- (380,440) -- cycle ;
				%Straight Lines [id:da7322073768581836] 
				\draw    (390,440) -- (390,460) ;
				
				% Text Node
				\draw (376,378) node [anchor=north west][inner sep=0.75pt]    {$\R^{\vactab' }$};
				% Text Node
				\draw (385,459) node [anchor=north west][inner sep=0.75pt]  [font=\tiny]  {$\alpha $};
				% Text Node
				\draw (369,403) node [anchor=north west][inner sep=0.75pt]  [font=\tiny]  {$\lambda ^{( k)}$};
				% Text Node
				\draw (369,364) node [anchor=north west][inner sep=0.75pt]  [font=\tiny]  {$\lambda ^{( 1)}$};
				% Text Node
				\draw (380,420) node [anchor=north west][inner sep=0.75pt]    {$p_{\vactab }$};
	\end{tikzpicture} }  } = \vcenter{\hbox{

			\tikzset{every picture/.style={line width=0.75pt}} %set default line width to 0.75pt        
			
			\begin{tikzpicture}[x=0.75pt,y=0.75pt,yscale=-1,xscale=1]
				%uncomment if require: \path (0,776); %set diagram left start at 0, and has height of 776
				
				%Straight Lines [id:da585392368330554] 
				\draw    (230,620) -- (230,640) ;
				%Shape: Rectangle [id:dp10808275769451314] 
				\draw   (200,640) -- (280,640) -- (280,670) -- (200,670) -- cycle ;
				%Straight Lines [id:da6199252683511738] 
				\draw    (230,670) -- (230,690) ;
				%Straight Lines [id:da7562049946550911] 
				\draw  [dash pattern={on 0.84pt off 2.51pt}]  (230,580) -- (230,620) ;
				%Straight Lines [id:da40492000829035724] 
				\draw    (230,470) -- (230,490) ;
				%Shape: Rectangle [id:dp746504914548167] 
				\draw   (200,490) -- (280,490) -- (280,520) -- (200,520) -- cycle ;
				%Straight Lines [id:da049632648018712144] 
				\draw    (230,520) -- (230,550) ;
				%Shape: Rectangle [id:dp6381119732484091] 
				\draw   (210,690) -- (250,690) -- (250,710) -- (210,710) -- cycle ;
				%Straight Lines [id:da9886200181388269] 
				\draw    (230,710) -- (230,730) ;
				%Shape: Rectangle [id:dp5571780394432293] 
				\draw   (210,550) -- (250,550) -- (250,570) -- (210,570) -- cycle ;
				%Shape: Rectangle [id:dp647588063626011] 
				\draw   (210,450) -- (250,450) -- (250,470) -- (210,470) -- cycle ;
				%Straight Lines [id:da7675856925838997] 
				\draw    (230,430) -- (230,450) ;
				
				% Text Node
				\draw (202,640) node [anchor=north west][inner sep=0.75pt]    {$\R^{\lambda ^{( k)}\rightarrow \rho ^{( k-\frac{1}{2})}}$};
				% Text Node
				\draw (222,732) node [anchor=north west][inner sep=0.75pt]  [font=\tiny]  {$\alpha $};
				% Text Node
				\draw (200,490) node [anchor=north west][inner sep=0.75pt]    {$\R^{\rho ^{( 1+\frac{1}{2})}\rightarrow \rho ^{( 1)}}$};
				% Text Node
				\draw (215,691) node [anchor=north west][inner sep=0.75pt]    {$p_{\lambda ^{( k)}}$};
				% Text Node
				\draw (211,552) node [anchor=north west][inner sep=0.75pt]    {$p_{\lambda ^{( 1+\frac{1}{2})}}$};
				% Text Node
				\draw (215,455) node [anchor=north west][inner sep=0.75pt]    {$p_{\lambda ^{( 1)}}$};

	\end{tikzpicture}  }  }
\end{equation}
This vanishes unless $\vactab = \vactab'$ because of the projector property $D^{\lambda}(p^{}_{\lambda^{(l)}}) = \delta^{\lambda \lambda^{(l)}} \idn$, where $\idn$ is the identity matrix.
\end{proof}
In other words
\begin{equation}
	\im p_\vactab = \Span(E_\alpha^\vactab).
\end{equation}

Crucially, as we will now see, the images are one-dimensional and therefore the projectors $p_\vactab$ provide a means for constructing a nice basis for $Z_\lambda$.
\begin{proposition}
	Let $\vactab = (\lambda^{(0)} = [\N], \lambda^{(\frac{1}{2})} = [\N-1], \lambda^{(1)}, \lambda^{(\frac{3}{2})},\dots,\lambda = \lambda^{(k)})$ be a vacillating tableau of shape $\lambda$. The image of $p^{}_\vactab$ is one-dimensional as a linear map on $Z_\lambda$.
\end{proposition}
\begin{proof}
	Since $p_\vactab$ is a projector, the trace of the representation matrix gives the rank (dimension of the image). Therefore, we compute
	\begin{align}
		\sum_{\alpha} D^{\lambda}_{\alpha \alpha}(p_\vactab^{}) &=\sum_\alpha \langle E^{\lambda}_\alpha, \sum_{\gamma} D^{\lambda}_{\gamma \alpha}(p_\vactab)E^{\lambda}_\gamma \rangle \\
		&= \sum_\alpha \langle E^{\lambda}_\alpha, (\R^{\lambda^{(1+\frac{1}{2})} \rightarrow \lambda^{(1)}} \R^{\lambda^{(2)} \rightarrow \lambda^{(1+\frac{1}{2})}}\dots \R^{\lambda^{(k-\frac{1}{2})} \rightarrow \lambda^{(k-1)}}\R^{\lambda^{(k)} \rightarrow \lambda^{(k-\frac{1}{2})}} )_{1 \alpha}E^{\lambda^{(1)}}_1 \rangle \nonumber \\
		&=  \begin{aligned}[t]
			\sum_\alpha \langle \sum_{\vactab'} &(\R^{\rho^{(1+\frac{1}{2})} \rightarrow \rho^{(1)}} \R^{\rho^{(2)} \rightarrow \rho^{(1+\frac{1}{2})}}\dots \R^{\rho^{(k-\frac{1}{2})} \rightarrow \rho^{(k-1)}}\R^{\lambda^{(k)} \rightarrow \rho^{(k-\frac{1}{2})}} )_{1 \alpha}E^{\rho^{(1)}}_1,\\
			& (\R^{\lambda^{(1+\frac{1}{2})} \rightarrow \lambda^{(1)}} \R^{\lambda^{(2)} \rightarrow \lambda^{(1+\frac{1}{2})}}\dots \R^{\lambda^{(k-\frac{1}{2})} \rightarrow \lambda^{(k-1)}}\R^{\lambda^{(k)} \rightarrow \lambda^{(k-\frac{1}{2})}} )_{1 \alpha}E^{\lambda^{(1)}}_1 \rangle, \nonumber
		\end{aligned}
	\end{align}
where the second equality uses Proposition \ref{prop: image of P_vactab} and the sum is over vacillating tableaux $\vactab' = (\rho^{(0)} = [\N], \rho^{(\frac{1}{2})} = [\N-1], \rho^{(1)}, \rho^{(\frac{3}{2})},\dots,\lambda^{(k)})$ with shape $\lambda^{(k)}=\lambda$.
The last equation includes matrix multiplications of the form
\begin{equation}
	\sum_\alpha \R^{\lambda^{(k)} \rightarrow \rho^{(k-\frac{1}{2})}}_{\beta \alpha} \R^{\lambda^{(k)} \rightarrow \lambda^{(k-\frac{1}{2})}}_{\gamma \alpha} = \delta^{\rho^{(k-\frac{1}{2})} \lambda^{(k-\frac{1}{2})} } \delta_{\beta \gamma},
\end{equation}
where we used the orthonormality condition \eqref{eq: res basis ON condition}. Successively applying this identity gives
\begin{equation}
	\sum_{\alpha} D^{\lambda}_{\alpha \alpha}(p_\vactab^{}) = \sum_{\vactab'} \delta^{\vactab \vactab'} = 1,
\end{equation}
where
\begin{equation}
	 \delta^{\vactab \vactab'}  = \prod_{l} \delta^{\lambda^{(l)} \rho^{(l)}}.
\end{equation}
\end{proof}

As we promised, we have the following Corollary about dual elements acting on the image of $p_\vactab$.
\begin{corollary}
Let $\vactab$ be a vacillating tableaux of shape $\lambda \in \Lambda_{2,\N}$.
The elements $Z_1, Z_{1\frac{1}{2}}, Z_2 \in P_2(\N)$ defined in Theorem \ref{thm: murphys} act on $E^\vactab_\alpha$ (the image of $p_\vactab$ on $Z_\lambda$) by normalized characters.
\end{corollary}
\begin{proof}
This readily follows by considering for example
\begin{equation}
	\sum_{\gamma} D^{\lambda}_{\gamma \alpha}(Z_1)E^{\vactab}_{\gamma},
\end{equation}
and using the equivariance property of the restriction matrices together with
\begin{equation}
	D^{\lambda^{(1)}}(Z_1) = \hat{\chr}^{\lambda^{(1)}}(Z_1) \idn.
\end{equation}
Similarly for $ Z_{1\frac{1}{2}}, Z_2 $.
\end{proof}
Note that this allows us to pick up normalized characters of the various irreducible representations in the vacillating tableau.

From Theorem \ref{thm: multi is vac tab}, the number of vacillating tableaux is equal to the dimension of $Z_\lambda$. Therefore, we have found a complete set of orthogonal projectors with rank one, and together their images form a total space isomorphic to $Z_\lambda$. Let $E^{\lambda}_\vactab$ be a basis for the image of $p_\vactab$, then the irreducible representation of $P_k(\N)$
\begin{equation}
	Z_\lambda \cong \Span(E^{\lambda}_\vactab \, \vert \, \text{for all vacillating tableaux $\vactab$ of shape $\lambda$ and length $k$.}).
\end{equation}
The vectors $E^{\lambda}_\vactab$ are simultaneous eigenvectors of $Z_1, Z_{1\frac{1}{2}}, Z_2 \in P_2(\N)$. We give a name to the change of basis matrix
\begin{equation}
	E^\lambda_\vactab = \sum_{\alpha} V^{\lambda}_{\alpha \vactab} E^{\lambda}_\alpha.
\end{equation}
Because the restriction properties of these basis elements are manifest in the vacillating tableaux labelling them, it is called an inductive basis. Inductive bases respecting the restriction of $P_k(\N)$ are discussed in \cite[Theorem 3.37]{Halverson2005}.

So far we have only discussed the existence of an inductive basis and seen that it forms an eigenbasis of a set of commuting operators (e.g. $Z_1, Z_{1\frac{1}{2}}, Z_2$). In the next subsection we will use this in the context of the regular representation of $P_k(\N)$ to give an explicit procedure for constructing matrix units.

\subsection{All $\N$ construction of matrix units.} \label{subsec: all N matrix units}
Recall that the regular representation of $P_k(\N)$ decomposes into representations of the left and right action as (see \eqref{eq: PkN AW decomp})
\begin{equation}
	P_k(\N) \cong \bigoplus_{\lambda  \in \Lambda_{k,\N} } Z_\lambda \otimes Z_\lambda.
\end{equation}
From the previous section, each component on the r.h.s. has a basis
\begin{equation}
	Z_\lambda \otimes Z_\lambda \cong\Span(E^{\lambda}_\vactab \otimes E^{\lambda}_{\vactab'} \, \vert \, \text{for all vacillating tableaux $\vactab, \vactab'$ of shape $\lambda$ and length $k$.}).
\end{equation}

In principle, this basis is found by acting on $P_k(\N)$ from the left and right using $p_\vactab, p_{\vactab'}$ and producing a basis for the image. However, explicit forms of the projectors are not known in general. Fortunately, a Lagrange interpolation method can be used to construct them \cite{Doty2019}. We state the result for  $k=2$. The case of $k=1$ is given by removing the terms involving $Z_{1\frac{1}{2}}, Z_{2}$ (see Example \ref{ex: P1N units} below). For general $k$ it is necessary to consider additional $Z_i$.
\begin{proposition}
	Let $\vactab = (\lambda^{(0)} = [\N], \lambda^{(\frac{1}{2})} = [\N-1], \lambda^{(1)}, \lambda^{(1 +\frac{1}{2})}, \lambda^{(2)})$ be a vacillating tableau of length two, then
	\begin{equation}
		p_\vactab = \begin{aligned}[t]			
			&\prod_{\substack{\rho^{(1)} \in \Lambda_{1,\N} \\ \rho^{(1)} \neq \lambda^{(1)} } } \frac{Z_1 - \hat{\chr}^{\rho^{(1)}}(Z_1)}{\hat{\chr}^{\lambda^{(1)}}(Z_1) - \hat{\chr}^{\rho^{(1)}}(Z_1)} \times \\
			&\prod_{\substack{\rho^{(1+\frac{1}{2})} \in \Lambda_{1+\frac{1}{2},\N}  \\
					\rho^{(1+\frac{1}{2})} \neq \lambda^{(1+\frac{1}{2})} } } \frac{Z_{1\frac{1}{2}} - \hat{\chr}^{\rho^{(1+\frac{1}{2})}}(Z_{1\frac{1}{2}})}{\hat{\chr}^{\lambda^{(1+\frac{1}{2})}}(Z_{1\frac{1}{2}}) - \hat{\chr}^{\rho^{(1+\frac{1}{2})}}(Z_{1\frac{1}{2}})} \times \\
			&\prod_{\substack{\rho^{(2)} \in \Lambda_{2,\N} \\
					\rho^{(2)} \neq \lambda^{(2)} } } \frac{Z_{2} - \hat{\chr}^{\rho^{(2)}}(Z_{2})}{\hat{\chr}^{\lambda^{(2)}}(Z_{2}) - \hat{\chr}^{\rho^{(2)}}(Z_{2})}.
		\end{aligned}
	\end{equation}	\label{eq: vactab projector}
\end{proposition}
\begin{proof}
	By construction, this operator vanishes when acting on $E^{\lambda}_{\vactab'}$ unless $\vactab = \vactab'$, since $Z_1, Z_{1\frac{1}{2}}, Z_2$ act through normalized characters that are all distinct.
\end{proof}

We define new linear operators on $P_k(\N)$ labelled by two vacillating tableaux.
\begin{definition} \label{def: vactab proj}
	Let $\vactab, \vactab'$ be two vacillating tableaux of the same shape and length $k$. Define $\mathcal{P}_{\vactab, \vactab'}$ by
	\begin{equation}
		\mathcal{P}_{\vactab \vactab'}(d) = p_\vactab d p_{\vactab'}.
	\end{equation}
\end{definition}
\begin{example}\label{ex: P1N units}
	The simplest example is to consider $P_1(\N)$, where we only need to consider $Z_1 = {\scriptstyle \PAdiagram{1}{}{}}$. There are two vacillating tableaux $\vactab_1 = ([\N], [\N], [\N]), \vactab_2 = ([\N], [\N], [\N-1,1])$ and using Theorem \ref{thm: murphys} together with Example \ref{ex: norm chars} we have
	\begin{align}
		&p_{\vactab_1}^{} =  \frac{\PAdiagram{1}{}{} - \PAdiagram{1}{-1/1}{}\chi^{[\N-1,1]}\qty({\scriptstyle \PAdiagram{1}{}{}}) }{\chi^{[\N]}\qty({\scriptstyle \PAdiagram{1}{}{}})-\chi^{[\N-1,1]}\qty({\scriptstyle \PAdiagram{1}{}{}})} = \frac{\PAdiagram{1}{}{}}{\N}\\
		&p_{\vactab_2}^{} =  \frac{\PAdiagram{1}{}{} - \PAdiagram{1}{-1/1}{}\chi^{[\N]}\qty({\scriptstyle \PAdiagram{1}{}{}}) }{\chi^{[\N-1,1]}\qty({\scriptstyle \PAdiagram{1}{}{}})-\chi^{[\N]}\qty({\scriptstyle \PAdiagram{1}{}{}})} = \frac{\PAdiagram{1}{}{} - \PAdiagram{1}{-1/1}{}\N }{-\N}
	\end{align}
	Therefore
	\begin{align}
		\begin{aligned}
			\mathcal{P}_{\vactab_1 \vactab_1}(\PAdiagram{1}{}{}) &= \PAdiagram{1}{}{}\\
			\mathcal{P}_{\vactab_1 \vactab_1}(\PAdiagram{1}{-1/1}{}) &= \frac{1}{\N}\PAdiagram{1}{}{}
		\end{aligned}
		\quad
		\begin{aligned}
			\mathcal{P}_{\vactab_2 \vactab_2}(\PAdiagram{1}{}{}) &= 0\\
			\mathcal{P}_{\vactab_2 \vactab_2}(\PAdiagram{1}{-1/1}{}) &= \PAdiagram{1}{-1/1}{}-\frac{1}{\N}\PAdiagram{1}{}{}
		\end{aligned}
	\end{align}
	and
	\begin{equation}
		\im(\mathcal{P}_{\vactab_1 \vactab_1}) = \Span(Q^{[\N]} = \frac{1}{\N} \PAdiagram{1}{}{}), \quad \im(\mathcal{P}_{\vactab_2 \vactab_2}) = \Span(Q^{[\N-1,1]} =\PAdiagram{1}{-1/1}{}-\frac{1}{\N}\PAdiagram{1}{}{}).
	\end{equation}
	Note that $Q^{\lambda}Q^{\lambda'} = \delta^{\lambda \lambda'} Q^{\lambda}$.
\end{example}

The image of a matrix can be found by computing its pivot columns \cite[\textbf{2O} in Section 2.4]{strang2006linear}. The vectors formed by the pivot columns constitute a basis of the image. From the previous example, we see that the matrices corresponding to $\mathcal{P}_{\vactab \vactab'}$ have rational functions of $\N$ as entries. In general, it is non-trivial to compute the pivot columns of a matrix with rational functions. Instead, we use the following trick.

Consider the matrix associated with the operators $P_{\vactab \vactab'}$ in the diagram basis
\begin{equation}
	\mathcal{P}_{\vactab \vactab'}(d_\pi) = \sum_{\pi' \in \setpart{[k \vert k']}} (\mathcal{P}_{\vactab \vactab'})_{\pi' \pi} d_{\pi'}. \label{eq: P vactab vactab matrix}
\end{equation}
Because we know that the image is one-dimensional, this matrix has a single pivot column. We find the pivot column by substituting $\N = n$ for some integer $n \geq 2k$. Let this matrix have pivot column $\nu$, then the element
\begin{equation}
	\sum_{\pi \in \setpart{[k \vert k']}}  (\mathcal{P}_{\vactab \vactab'})_{\pi \nu} d_{\pi},
\end{equation}
is a basis for the image of $\mathcal{P}_{\vactab \vactab'}$. For two vacillating tableaux $\vactab, \vactab'$ of shape $\lambda \in \Lambda_{k,\N}$ we define the matrix units
\begin{equation}
	Q^{\lambda}_{\vactab \vactab'} = \sum_{\pi \in \setpart{[k \vert k']}}  (\mathcal{P}_{\vactab \vactab'})_{\pi \nu} d_{\pi}.
\end{equation}

The validity of this trick is argued as follows. Suppose we were able to find the row echelon form
\begin{equation}
	\widetilde{(\mathcal{P}_{\vactab \vactab'})}_{\pi' \pi}
\end{equation}
of $(\mathcal{P}_{\vactab \vactab'})_{\pi' \pi}$. The row echelon form has a single non-zero row and pivot element
\begin{equation}
	\widetilde{(\mathcal{P}_{\vactab \vactab'})}_{1 \nu} = f(\N).
\end{equation}
Since $f(\N)$ is a rational function it has a finite number of zeroes and poles. Away from these we can construct the matrix
\begin{equation}
	\frac{1}{f(\N)}\widetilde{(\mathcal{P}_{\vactab \vactab'})}_{\pi' \pi},
\end{equation}
whose pivot element is $1$, and in particular independent of $\N$. This matrix is now in reduced row echelon form. The reduced row echelon form of a matrix is unique. Thus, away from the poles and zeroes of $f(\N)$, we can argue that the pivot column is independent of $\N$.

In appendix \ref{apx: P2N units} we give the result of applying this procedure to $P_2(\N)$ and give a table of matrix units. Note that the above procedure does not fix the normalization of each matrix unit, since a basis for a one-dimensional subspace is only determined up to a scalar. We address this issue in the following chapter (see section \ref{subsection: normalization}).

\section{Summary}
In the first section of this chapter we saw that partition algebras $P_k(\N)$, $P_{k+\frac{1}{2}}(\N)$, form a family of diagram algebras that are Schur-Weyl dual to the symmetric groups $\SN$ and $S_{\N-1}$ respectively. This duality implies that the representation theory of $P_k(\N)(P_{k+\frac{1}{2}}(\N))$ can be used to study the representation theory of $\SN(S_{\N-1})$ and vice versa. The partition algebras form an "inductive chain"
\begin{equation}
	P_1(\N) \subset P_{1+\frac{1}{2}}(\N) \subset \dots \subset P_{k-\frac{1}{2}}(\N) \subset P_k(\N),
\end{equation}
and we saw that the restriction-induction construction of tensor powers, discussed in Chapter \ref{ch: SN}, determine the decomposition of representations $Z_\lambda$ of partition algebras under restriction along this chain. Particularly noteworthy is the absence of multiplicities of irreducible representations in the restriction. This gave rise to an "inductive basis" of elements labelled by vacillating tableaux. A vacillating tableau describes the set of irreducible representations a particular basis element of $Z_\lambda$ belongs to under restriction along the inductive chain. The fact that all irreducible representations of $P_1(\N)$ are one-dimensional, together with the fact that restrictions along the chain are multiplicity free, guaranteed that these elements form a basis and that each basis element is unique up to normalization. The inductive structure of $P_k(\N)$ and its relation, through Schur-Weyl duality, to induction and restriction of $\SN$ and $S_{\N-1}$ is well-known in the mathematical literature \cite{Martin1996, Halverson2005}. We have attempted to present this structure in language familiar to physicists.

In the second section, we studied the semi-simplicity of $P_k(\N)$. We found that there exists a basis of matrix units for $P_k(\N)$, where it is clear that it corresponds to an algebra of block matrices. The explicit change of basis is a generalization of the Fourier inversion formula studied in finite group theory. Semi-simplicity of partition algebras and formal expressions for matrix units are known in the literature, see \cite{Halverson2005} and \cite{AR90DissertCh1} respectively. In \cite{AR90DissertCh1}, general semi-simple algebras are discussed, and the regular representation is used to define a non-degenerate bilinear form. The formula for matrix units in this thesis uses a non-degenerate bilinear form coming from a trace in $\VN^{\otimes k}$. The proof of this formula is new and was first proven in \cite{Barnes:2022qli}.

In the third and last section, we gave a method for constructing the matrix units starting from the diagram basis. For this, we used the existence of the inductive basis. Crucially, we were able to construct projection operators $\mathcal{P}_{\vactab \vactab'}: P_k(\N) \rightarrow P_k(\N)$, with one-dimensional images corresponding to each matrix unit (up to normalization). Because the matrices corresponding to projection operators contain rational functions in $\N$, the usual procedure for finding a basis of the image does not work and/or is inefficient. Instead, we devised an all $\N$ construction of the images using a trick. The validity of this trick can be argued based on poles and zeroes of a particular entry in the row echelon form of the projection matrix. The matrix units for $P_1(\N)$ and $P_2(\N)$ underlie the construction of permutation invariant matrix models in the next chapter.

The use of $k$-duals and $(k+\tfrac{1}{2})$-duals of central elements in $\mathbb{C}(S_\N)$ to construct matrix units for $P_k(\N)$ is new. The construction of idempotent elements in algebras like the partition algebras, using central elements, has been discussed in \cite{Doty2019}. The fact that $\mathcal{P}_{\vactab \vactab'}$ can be constructed from left and right actions of such idempotents is implicit in the mathematical literature but we have not seen an explicit and concrete discussion of this in the literature. In this sense, our construction of matrix units is new and the technique used for finding the image for all $\N$ is new as well. This procedure was invented in the upcoming work \cite{PIGTM}, where it is used to construct permutation invariant Gaussian tensor models.

	\chapter{Permutation invariant matrix models}\label{chapter: 0d}
%\section{Permutation invariant Gaussian matrix models}
Matrix models can be thought of as zero-dimensional quantum field theories. They are defined by giving a probability distribution $p(X)$ on a space of matrices $X=\abs{\abs{X_{ij}}}$ for $i,j=1, \dots, \N$.
In this chapter we use the results of the previous section to give some new perspectives on the class of matrix models defined in \cite{Kartsaklis2017, Ramgoolam2019a}. Unlike the classical matrix models \cite{Wigner1955, Dyson1962}, which are invariant under a continuous symmetry group, these matrix models have discrete (permutation) symmetry. We give a new algorithm for computing expectation values of observables in these models, based on the new ideas.

The most general permutation invariant Gaussian one-matrix models were first solved in \cite{Ramgoolam2019a}: a 13-parameter model was constructed using Clebsch-Gordan coefficients for the decomposition of $\VN \otimes \VN$ into irreducible representations of $\SN$; the first and second moments were given as functions of $\N$ and general expectation values were shown to be computable using Wick's theorem. Observables in permutation invariant matrix models are permutation invariant matrix polynomials. A basis for the space of observables at large $\N$, labelled by directed graphs, was proposed in \cite{Kartsaklis2017}. Several expectation values of observables in the graph basis were computed in \cite{Ramgoolam2019a} using Wick's theorem. These results were generalized to two-matrix models in \cite{Barnes2022b} and a precise bijection between observables and directed graphs with no more than $\N$ vertices was proven. A combinatorial framework for counting and constructing observables using directed graphs and double cosets of permutation groups was explained in \cite{Barnes2022b}. A general algorithm for computing expectation values, as explicit functions of $\N$, of observables in the 2-matrix model was given in \cite{Barnes2022b}. The expected connection between permutation invariant matrix models and partition algebras through Schur-Weyl duality was mentioned in \cite{Kartsaklis2017}. This connection was explicitly used in \cite{Barnes:2021tjp} to construct observables from partition algebra elements. Various parameter limits of the permutation invariant Gaussian matrix models were explored as well, and points in the parameter space where the $\SN$ symmetry got enhanced to $O(\N)$ symmetry were discovered. Interesting factorisation results, for two-point functions of observables, were proven at the simplest of $O(\N)$ symmetric points in parameter space. The factorisation result intimately relied on the connection to partition algebras. We now give an introductory description of the model solved in \cite{Ramgoolam2019a} and elaborate on some of the above points. Along the way we will point out the new contributions in this thesis and point to the specific sections where further details are given.

The probability distribution of a permutation invariant Gaussian matrix model (PIGMM) is defined in terms of a quadratic polynomial function $V(X)$
\begin{equation}
	V(X) = \sum_{i,j=1}^{\N}J^{ij}X_{ij} + \sum_{i,j,k,l=1}^{\N}G^{ij; kl} X_{ij} X_{kl}, \label{eq: perm inv potential}
\end{equation}
where the parameters $J^{ij}, G^{ij;kl}$ are constrained by demanding
\begin{equation}
	V(\Paction{\sn}X\Paction{\sn}^T) = V(X) \quad \forall \sn \in \SN,
\end{equation}
or equivalently
\begin{equation}
	J^{(i)\sn (j)\sn} = J^{ij}, \quad G^{(i)\sn (j)\sn; (k)\sn (l)\sn} = G^{ij;kl} \quad \forall \sn \in \SN,
\end{equation}
and
\begin{equation}
	\Z = \int \dd{X} \e^{-V(X)} < \infty.
\end{equation}
Here $\dd{X} = \prod_{i,j=1}^\N \dd{X_{ij}}$. The probability to find the matrix $X$ in the interval $\dd{X}$ is given by
\begin{equation}
	p(X) = \frac{\dd{X} \e^{-V(X)}}{\Z}.
\end{equation}
In this form, the model is non-trivial to solve for large $\N$ because it involves inverting the $\N^2 \times \N^2$ matrix $G^{ij;kl}$.

Instead, the model was reformulated in \cite{Ramgoolam2019a} using a representation theoretic change of basis, from the matrix basis $X_{ij}$ to the irreducible basis
\begin{equation}
	X_{\lambda, \alpha, a} = \sum_{i,j=1}^\N C^{ij}_{\lambda \alpha a} X_{ij}.
\end{equation}
Here $C^{ij}_{\lambda \alpha a}$ are Clebsch-Gordan coefficients for the decomposition of $\VN \otimes \VN$ into irreducible representations of $\SN$. The indices $\alpha, \beta$ are multiplicity indices for this decomposition and the most general quadratic invariants were shown to be linear combinations of polynomials $q_{\lambda; \alpha \beta}(X)$ defined as
\begin{equation}
	q_{\lambda; \alpha, \beta}(X) = \sum_{a} X_{\lambda, \alpha, a} X_{\lambda, \beta, a}.
\end{equation}
As we will see, there are only $11$ independent polynomials of this form. This follows from studying the multiplicities in the decomposition of $\VN \otimes \VN$.
From Proposition \ref{cor: VNVN decomp} the multiplicity of any particular irreducible representations in $\VN \otimes \VN$ is less than four. This implies that solving the model in this formulation only requires the inversion of matrices of at most size three. As mentioned, we will elaborate on this in full detail in the next section.
By first computing the one-point and two-point functions in the representation basis, where they are simple, and then going back to the matrix basis, formulas for the expectation values $\expval{X_{ij}}$ and $\expval{X_{ij} X_{kl}}$ were computed. As we will see, the latter takes the form of linear combinations of the invariant tensors
\begin{equation}
	Q_{\lambda; \alpha\beta}^{ij;kl}=\sum_a C^{ij}_{\lambda \alpha a} C^{kl}_{\lambda \beta a}. \label{eq: pigmm intro tensors}
\end{equation}
A similar story is true for $\expval{X_{ij}}$, which has an expansion in terms of invariant tensors as well.
Since the model is quadratic, general expectation values could be exactly formulated using Wick's theorem. Therefore, an explicit solution to the model is determined by computing the above-mentioned invariant tensors. In \cite{Ramgoolam2019a} this was done through a careful analysis of the decomposition of $\VN$ as well as $\VN \otimes \VN$. The above formulation of PIGMMs is reviewed section \ref{subsec: distribution}.

%As we will see, the actions defining these matrix models include terms without pair-wise contractions of matrix indices.

%Matrix models including interactions without pair-wise contractions have been studied before (e.g. \cite{Lionni2019}). Here, even the quadratic terms have non-pair-wise contracted terms.

The expectation value of a polynomial function $f:  \M_\N(\mathbb{R}) \rightarrow \mathbb{R}$ in a PIGMM is given by
\begin{equation}
	\expval{f} =  \frac{\int \dd{X} f(X) \e^{-V(X)}}{\Z}.
\end{equation}
For our purposes, we will restrict our attention to invariant functions $\obs(X)$ satisfying
\begin{equation}
	\obs(\Paction{\sn}X\Paction{\sn}^T) = \obs(X) \quad \forall \sn \in \SN,
\end{equation}
where $\Paction{\sn}$ is a permutation matrix.
As mentioned, a bijection between directed graphs with $k$ edges and a basis for the space of invariant functions of homogeneous degree $k$ was proposed  in \cite{Kartsaklis2017, Ramgoolam2019a}. This bijection was proven in \cite{Barnes2022b}. In fact, the proof also gives a bijection that works away from the stable limit ($\N \geq 2k$). We review this in section \ref{subsec: observables}. Section \ref{subsec: observables} also gives a new description of the graph basis in terms of equivalence classes of 1-row set partition diagrams. This description was implicitly used in the algorithm in \cite{Barnes2022b}, but the connection to partition algebras and 1-row diagrams was not known at the time. The construction of observables using partition algebras, given in  \cite{Barnes:2021tjp}, naturally leads to 2-row set partition diagrams. In this thesis we will see that the 1-row partition diagram description is useful for computing expectation values. However, the invariant tensors in \eqref{eq: pigmm intro tensors} are closely related to 2-row diagrams and we will need a prescription for translating between the two.

A complete set of quadratic and a selection of cubic and quartic expectation values were computed in \cite{Ramgoolam2019a} using the graph basis and Wick's theorem. Later \cite{Barnes2022b}, an algorithm for computing general expectation values was given in terms of manipulations of so-called F-graphs. In this thesis, and in particular section \ref{subsec: exp vals}, we will give a new algorithm for computing expectation values of observables. It is based on the connection between permutation invariant tensors and partition algebra elements on one hand, and partition algebra elements and 1-row partition diagrams on the other hand. In particular, we will see that the invariant tensors in \eqref{eq: pigmm intro tensors} correspond to matrix units of $P_2(\N)$. Matrix units of $P_1(\N)$ enter into the one-point function. This observation is new and has not been exploited in the previous literature on permutation invariant matrix models. However, the  observation is largely based on the insight from upcoming work \cite{PIGTM} that uses this to construct permutation invariant tensor models. While we claim that this algorithm is new, it can roughly speaking be understood, in the matrix model case, as fully expanding the F-tensors appearing in the algorithm of \cite{Barnes2022b} in terms of products of Kronecker deltas and interpreting these products as 1-row diagrams. We will explain how the 2-row diagrams appearing in the expansion of matrix units of $P_1(\N), P_2(\N)$ are translated into 1-row diagrams. This is crucial, since it allows us to formulate degree $k$ expectation values as formal linear combinations of 1-row diagrams $\pi_i$, using Wick's theorem. In turn, the expectation value of an invariant polynomial corresponding to a 1-row diagram $\pi$ can be computed by pairing $\pi$ with all the $\pi_i$. This is explained in detail in section \ref{subsec: exp vals} and is supplemented by Appendix \ref{apx: EV code} which contains a detailed description of computer code implementing the algorithm.

The combinatorial framework for counting and constructing observables, first presented in \cite{Barnes2022b} for the case of permutation invariant observables, is reviewed in section \ref{subsec: graph counting} and Appendix \ref{apx: double coset} with accompanying code for counting double cosets of permutation groups.

\section{Distributions: Potentials and block diagonalization} \label{subsec: distribution}
We will now show that the space of PIGMM's is parametrised by
\begin{equation}
	\M_2^{+}(\mathbb{R}) \times \M_3^{+}(\mathbb{R})  \times \mathbb{R}^+ \times \mathbb{R}^+, \label{eq: moduli space of PIGMM}
\end{equation}
where $\M_i^+(\mathbb{R})$ are real positive-definite $i$-by-$i$ matrices and $\mathbb{R}^+$ positive real numbers. This is achieved using the matrix units of $P_2(\N)$ constructed in the previous chapter. As we will see, they give a basis for the space of quadratic invariant functions of $X$ that block diagonalize the matrix $G^{ij; kl}$ in equation \ref{eq: perm inv potential}. Further, the construction gives explicit formulas for the expectation values
\begin{equation}
	\expval{X_{ij}}, \quad \expval{X_{ij} X_{kl}}.
\end{equation}
This is sufficient data to determine arbitrary expectation values $\expval{\obs}$ through Wick's theorem.

To understand the connection between permutation invariant matrix models and partition algebras note that
\begin{equation}
	(\Paction{\sn}X\Paction{\sn}^T)_{ij} = X_{(i)\sn^{-1}(j)\sn^{-1}},
\end{equation}
and therefore
\begin{equation}
	\Span(X_{ij} \, \vert \, i,j=1,\dots,\N) \cong \VN\otimes \VN,
\end{equation}
as a representation of $\SN$. From corollary \ref{cor: VNVN decomp}
\begin{equation}
	\VN \otimes \VN \cong 2 V_{[\N]} \oplus 3 V_{[\N{-}1,1]} \oplus V_{[\N{-}2,2]} \oplus V_{[\N{-}2,1,1]}.
\end{equation}
The l.h.s. has a basis $X_{ij}$ while the RHS has a basis $X_{\lambda, \alpha,a}$ with labels
\begin{align}
	&\lambda = \{[\N], [\N-1,1], [\N-2,2], [\N-2,1,1]\}, \\
	&\alpha = 1,\dots, m^\lambda_{2,\N}, \\
	&a = 1, \dots, \dim V_{\lambda}.
\end{align}
The indices $\alpha,\beta$ should be thought of as vacillating tableaux $\vactab$ of shape $\lambda$ and length $2$.
The two bases are related by Clebsch-Gordan coefficients
\begin{align}
	X_{\lambda, \alpha,a} &=  \sum_{i,j=1}^\N C_{\lambda \alpha a}^{ij} X_{ij}, \label{eq: rep basis}\\
	X_{ij} &= \sum_{\lambda, \alpha, a} C^{\lambda \alpha a}_{ij}	X_{\lambda, \alpha,a}, \label{eq: matrix basis}
\end{align}
satisfying
\begin{align}
	\sum_{\lambda, \alpha, a} C^{\lambda \alpha a}_{ij}  C_{\lambda \alpha a}^{kl} &= \delta_i^k \delta_{j}^l \\
	\sum_{i,j=1}^\N C^{\lambda \alpha a}_{ij}C_{\lambda' \beta b}^{ij}  &= \delta^{\lambda}_{\lambda'}\delta^\alpha_\beta \delta^a_b.
\end{align}
That is $C^{\lambda \alpha a}_{ij}$ is the inverse of $C_{\lambda \alpha a}^{ij}$ and vice versa.

We are free to choose an orthonormal basis satisfying
\begin{equation}
	\VNinner{X_{\lambda, \alpha,a}}{X_{\lambda', \beta,b}} = \delta_{\lambda \lambda'}\delta_{\alpha \beta}\delta_{ab}, \label{eq: rep basis ON}
\end{equation}
with respect to the inner product defined by
\begin{equation}
	\VNinner{X_{ij}}{X_{kl}} = \delta_{ik} \delta_{jl}. \label{eq: matrix basis ON}
\end{equation}
This implies
\begin{equation}
	C^{\lambda \alpha a}_{ij}  = C_{\lambda \alpha a}^{ij}.
\end{equation}

Because the Clebsch-Gordan coefficients are matrix elements of an equivariant map they satisfy
\begin{equation}
	\sum_{i,j=1}^\N C_{\lambda \alpha a}^{ij} (\Paction{\sn}X\Paction{\sn}^T)_{ij} = \sum_{b} \D^{(\lambda)}(\sn)^b_a X_{\lambda \alpha b},
\end{equation}
where $\D^{\lambda}(\sn)$ are real orthogonal irreducible representations of $\SN$. The above equation corresponds to a system of linear equations over real numbers. Consequently, the solutions can be chosen real. In other words, Clebsch-Gordan coefficients for $\SN$ can be chosen to be real.

The above results can be used to state a characterisation of the set of invariant potentials $V(X)$.
\begin{proposition}The space of permutation invariant Gaussian matrix models can be understood by the following two statements.
	
	\begin{enumerate}
	\item[(a)] 
	The linear parameters in $V(X)$ can be expanded in terms of Clebsch-Gordan coefficients
	\begin{equation}
		J^{ij} = \sum_{\alpha=1}^2 J^{\alpha} C_{[\N] \alpha}^{ij},
	\end{equation}
	\item[(b)] The quadratic parameters have the form
	\begin{equation}
		G^{ij;kl} = \sum_{\lambda,\alpha,\beta,a} G^{\lambda; \alpha \beta} C_{\lambda \alpha a}^{ij} C_{\lambda \beta a}^{kl},
	\end{equation}
	where $G^{\lambda; \alpha \beta}$ is symmetric in $\alpha, \beta$. Define
	\begin{equation}
		Q_{\lambda; \alpha\beta}^{ij;kl} = \sum_a C_{\lambda \alpha a}^{ij} C_{\lambda \beta a}^{kl}
	\end{equation}
	then
	\begin{equation}
		G^{ij;kl} = \sum_{\lambda, \alpha, \beta} G^{\lambda; \alpha \beta} Q_{\lambda; \alpha\beta}^{ij;kl}
	\end{equation}
	\end{enumerate}
\end{proposition}
\begin{proof}
	First we prove statement (a).
	The invariant linear polynomials in $X_{ij}$ are vectors in the trivial subrepresentation of $\VN \otimes \VN$. This is spanned by $X_{[\N], 1}, X_{[\N], 2}$ and therefore
	\begin{equation}
		\sum_{i,j=1}^\N J^{ij}X_{ij} = \sum_\alpha J^{\alpha} X_{[\N], \alpha} = \sum_{i,j=1, \alpha}^\N J^{\alpha} C_{[\N] \alpha}^{ij} X_{ij}.
	\end{equation}

	The proof of (b) goes as follows.
	The quadratic invariant polynomials in $X_{ij}$ correspond to invariant vectors $\Sym^2(\VN \otimes \VN)$, the symmetric subspace of $(\VN \otimes \VN)^{\otimes 2}$. Let $V_\lambda, V_{\lambda'}$ be two irreducible representations of $\SN$, then
	\begin{equation}
		\Hom(V_\lambda, V_{\lambda'}) = V_{\lambda'} \otimes V_{\lambda}^*,
	\end{equation}
	where $V_{\lambda}^*$ is the dual space (linear functions on $V_{\lambda}$). If $\SN$ acts on $V_\lambda$ through the matrix $\D^\lambda(\sn)$, the action on $V_\lambda^*$ in the dual basis is through the matrix $\D^\lambda(\sn^{-1})^T$. For orthogonal representations
	\begin{equation}
		\D^\lambda(\sn^{-1})^T = \D^{\lambda}(\sn),
	\end{equation}
	and therefore
	\begin{equation}
		V_{\lambda'} \otimes V_{\lambda}^* \cong V_{\lambda'} \otimes V_{\lambda}. \label{eq: Vdual is V}
	\end{equation}
	Secondly, the vector space of $\SN$-equivariant linear maps (homomorphisms)
	\begin{equation}
		\Hom_{\SN}(V_\lambda, V_{\lambda'}),
	\end{equation}
	is isomorphic to the subspace of invariants in \eqref{eq: Vdual is V}. We write this as
	\begin{equation}
		\Hom_{\SN}(V_\lambda, V_{\lambda'}) \cong (V_{\lambda'} \otimes V_{\lambda}^*)^{\SN} \cong (V_{\lambda'} \otimes V_{\lambda})^{\SN}.
	\end{equation}
	To show this, let $E_{\lambda a}$ be a basis for $V_\lambda$. In this basis, an element in $\Hom_{\SN}(V_\lambda, V_{\lambda'}) $ corresponds to a matrix $M$ satisfying
	\begin{equation}
		MD^{\lambda}(\sn) = D^{\lambda'}(\sn)M,
	\end{equation}
	or
	\begin{equation}
		D^{\lambda'}(\sn^{-1})MD^{\lambda}(\sn) = [D^{\lambda'}(\sn)]^T MD^{\lambda}(\sn) = M.
	\end{equation}
	If $M_{a}^b$ are the matrix elements of $M$, then the above equation implies that
	\begin{equation}
		\sum_{a,b} M_{a}^b E_{\lambda a} \otimes E_{\lambda' b}
	\end{equation}
	is an invariant vector. The converse statement holds as well.
	Schur's lemma says that for any $M \in \Hom_{\SN}(V_\lambda, V_{\lambda'})$, $M_{b}^a \propto \delta_{b}^a \delta_{\lambda \lambda'}$. That is,
	\begin{equation}
		(V_{\lambda} \otimes V_{\lambda})^{\SN} = \Span(\sum_{a} E_{\lambda a} \otimes E_{\lambda a}).
	\end{equation}

	We now apply this to
	\begin{align}
		\Hom_{\SN}(\VN \otimes \VN, \VN \otimes \VN) &= \End_{S_N}(\VN^{\otimes 2})\\
		 &\cong  (\VN \otimes \VN \otimes \VN \otimes \VN)^{\SN}.
	\end{align}
	Schur's lemma together with \eqref{eq: VNVN decomp} gives,
	\begin{equation}
		\End_{S_N}(\VN^{\otimes 2})\cong \bigoplus_{\lambda} \Hom_{\SN}(m^{\lambda}_{2,\N} V_{\lambda}, m^{\lambda}_{2,\N} V_{\lambda}),
	\end{equation}
	since there are no homomorphisms between non-isomorphic irreducible representations.
	This implies that
	\begin{equation}
		(\VN \otimes \VN \otimes \VN \otimes \VN)^{\SN} = \Span(\sum_a X_{\lambda \alpha a} \otimes X_{\lambda \beta a}).
	\end{equation}
	A basis of invariant quadratic polynomials is
	\begin{equation}
		\qty{\sum_a X_{\lambda \alpha a} X_{\lambda \beta a} \, \vert \, \alpha \leq \beta},
	\end{equation}
	since $X_{\lambda \alpha a}$ are commuting variables.
	Explicitly, any quadratic invariant polynomial has the form
	\begin{equation}
		\sum_{\lambda, \alpha, \beta} G^{\lambda; \alpha \beta} \sum_a X_{\lambda \alpha a}  X_{\lambda \beta a} = \sum_{\lambda, \alpha, \beta, a} \sum_{i,j,k,l=1}^\N G^{\lambda; \alpha \beta}  C_{\lambda \alpha a}^{ij} C_{\lambda \beta a}^{kl} X_{ij} X_{kl}.
	\end{equation}
	This concludes the proof.
\end{proof}

It immediately follows that
\begin{corollary}\label{cor: inv potential}
	The most general permutation invariant Gaussian matrix potential is
	\begin{equation}
		V(X) =  -\sum_\alpha J^{\alpha} X_{[\N], \alpha}  + \frac{1}{2}\sum_{\lambda, a}\sum_{\alpha \beta} G^{\lambda; \alpha \beta} X_{\lambda \alpha a} X_{\lambda \beta a}.
	\end{equation}
\end{corollary}

From now on we will use summation convention for paired upper and lower indices, unless stated otherwise.
Having introduced the variables $X_{\lambda \alpha a}$, it will be useful to define the partition function in terms of them. In fact they give rise to an almost fully decoupled partition function. We will now prove that
\begin{proposition}\label{thm: PIGMM}
	The permutation invariant Gaussian matrix models have partition function
	\begin{equation}
		\mathcal{Z} = \int \dd{X} \e^{J^{\alpha} X_{[\N], \alpha}  - \frac{1}{2}\sum_{\lambda, a}G^{\lambda; \alpha \beta} X_{\lambda \alpha a} X_{\lambda \beta a}},
	\end{equation}
	where $\dd{X} = \prod_{\lambda,\alpha, a} \dd{X_{\lambda \alpha a}} = \prod_{i,j=1} \dd{X_{ij}}$ .
\end{proposition}
\begin{proof}
	The form of the polynomial in the exponential follows from corollary \ref{cor: inv potential}. It remains to show that
	\begin{equation}
		\prod_{\lambda,\alpha, a} \dd{X_{\lambda \alpha a}} = \prod_{i,j=1} \dd{X_{ij}}.
	\end{equation}
	This is a consequence of both $X_{\lambda \alpha a}$ and $X_{ij}$ forming orthonormal bases and therefore the Jacobian matrix
	\begin{equation}
		\pdv{X_{ij}}{X_{\lambda \alpha a}} = C^{\lambda \alpha a}_{ij},
	\end{equation}
	is an orthogonal matrix. From equation \eqref{eq: matrix basis ON}
	\begin{equation}
		\sum_{i,j}  C^{\lambda \alpha a}_{ij} C^{\lambda' \beta b}_{ij} = \delta_{\lambda \lambda'} \delta_{\alpha \beta}\delta_{ab}, \label{eq: CGC norm}
	\end{equation}
	which should be read as the matrix equation $CC^T = 1$. Consequently $\abs{\det(C)} = 1$ and
	\begin{equation}
		\prod_{i,j=1} \dd{X_{ij}} = \abs{\det(C)} 	\prod_{\lambda,\alpha, a} \dd{X_{\lambda \alpha a}}   = 	\prod_{\lambda,\alpha, a} \dd{X_{\lambda \alpha a}} .
	\end{equation}
\end{proof}

We now have all the ingredients necessary to construct the generating function of expectation values
\begin{definition}
The generating function is given by
\begin{equation}
	\Z[J] = \int \dd{X}  \e^{J^{\lambda \alpha a} X_{\lambda \alpha a}  - \frac{1}{2}\sum_{\lambda, a}G^{\lambda; \alpha \beta} X_{\lambda \alpha a} X_{\lambda \beta a}}.
\end{equation}
\end{definition}
Standard results on multi-variate Gaussian integration gives
\begin{equation}
	\Z[J]= \sqrt{ \frac{(2\pi)^{N^2}}{\prod_{\lambda} \det(G^{\lambda})} }\e^{\frac{1}{2}\sum_{\lambda, \alpha,\beta, a} J^{\lambda \alpha a} (G^{-1})_{\lambda; \alpha \beta} J^{\lambda \beta a}}. \label{eq: PIGMM ZJ}
\end{equation}
\begin{corollary}
We compute the average as
\begin{equation}
	\expval{X_{\lambda, \alpha a}} = \frac{1}{\Z}\left. \pdv{}{J^{\lambda \alpha a}}\Z[J] \right \vert_{J^\alpha \neq 0} = \delta_{\lambda [\N]} (G^{-1})_{[\N]; \alpha \beta}J^\beta \label{eq: 1pt}
\end{equation}
where the subscript $J^\alpha \neq 0$ is a reminder to set all $J^{\lambda \alpha a}$ except $J^\alpha = J^{[\N] \alpha}$ to zero.
\end{corollary}
\begin{corollary}
To compute the second moment we need
\begin{align}
		&\expval{X_{\lambda \alpha a} X_{\lambda' \beta b}} = \frac{1}{\Z}\left. \pdv{}{J^{\lambda \alpha a}}{J^{\lambda' \beta b}}\Z[J] \right \vert_{J^\alpha \neq 0} \\
		&=  \qty(\delta_{\lambda [\N]} (G^{-1})_{[\N]; \alpha \gamma}J^\gamma) \qty(\delta_{\lambda' [\N]} (G^{-1})_{[\N]; \beta \delta}J^\delta) + \delta_{\lambda \lambda'} (G^{-1})_{\lambda; \alpha \beta}\delta_{ab},
\end{align}
note that the first term is simply
\begin{equation}
	\expval{X_{\lambda \alpha a}}\expval{X_{\lambda' \beta b}}.
\end{equation}
\end{corollary}
It is useful to define
\begin{equation}
	\expval{X_{\lambda \alpha a} X_{\lambda' \beta b}}_c = \expval{X_{\lambda \alpha a} X_{\lambda' \beta b}} - \expval{X_{\lambda \alpha a }}\expval{X_{\lambda' \beta b}} = \delta_{\lambda \lambda'} (G^{-1})_{\lambda; \alpha \beta}\delta_{ab}. \label{eq: con 2pt}
\end{equation}

This sets us up for the main result in this section.
\begin{proposition} \label{prop: 1mom 2mom}
	The average in the matrix basis is given by
	\begin{equation}
		\expval{X_{ij}} = C^{[\N] \alpha}_{ij} (G^{-1})_{[\N]; \alpha \beta}J^\beta
	\end{equation}
	and
	\begin{equation}
		\expval{X_{ij}X_{kl}}_c =  (G^{-1})_{\lambda; \alpha \beta} Q_{ij;kl}^{\lambda; \alpha \beta},
	\end{equation}
	where
	\begin{equation}
		Q_{ij:kl}^{\lambda; \alpha \beta} = \sum_a C^{\lambda \alpha a}_{ij} C^{\lambda \beta a}_{kl} = \sum_a C_{\lambda \alpha a}^{ij} C_{\lambda \beta a}^{kl} = Q^{ij;kl}_{\lambda; \alpha \beta}.
	\end{equation}
\end{proposition}
\begin{proof}
	From equation \ref{eq: matrix basis} and \ref{eq: 1pt} we have
	\begin{equation}
		\begin{aligned}
		\expval{X_{ij}} &= C^{\lambda \alpha a}_{ij} \expval{	X_{\lambda, \alpha,a}} = \sum_{\lambda}C^{\lambda \alpha a}_{ij} \delta_{\lambda [\N]} (G^{-1})_{\lambda; \alpha \beta} J^\beta\\
		& = C^{[\N] \alpha}_{ij} (G^{-1})_{[\N]; \alpha \beta} J^\beta,	
		\end{aligned} \label{eq: matrix basis 1pt}
	\end{equation}
	and from \eqref{eq: con 2pt}
	\begin{equation}
		\begin{aligned}
			\expval{X_{ij}X_{kl}}_c& = C^{\lambda \alpha a}_{ij} C^{\lambda' \beta b}_{kl} \expval{X_{\lambda, \alpha,a}   X_{\lambda', \beta,b}   }_c \\
			&= C^{\lambda \alpha a}_{ij} C^{\lambda' \beta b}_{kl} \delta_{\lambda \lambda'} (G^{-1})_{\lambda; \alpha \beta} \delta_{ab} = (G^{-1})_{\lambda; \alpha \beta}Q^{\lambda; \alpha \beta}_{ij;kl}.
		\end{aligned} \label{eq: matrix basis con 2pt}
	\end{equation}
\end{proof}
\begin{remark}
	The Clebsch-Gordan coefficients $C_{ij}^{[\N] \alpha}$ correspond to elements of $\End_{\SN}(\VN)$. In particular, they are the matrix units for $P_1(\N)$ found in example \ref{ex: P1N units}.
	The invariant tensors $Q^{\lambda; \alpha\beta}_{ij;kl}$ correspond to elements of $\End_{\SN}(\VN^{\otimes 2})$. In fact, they correspond to the matrix units constructed in section \ref{sec: construction of units}. These can be found in Appendix \ref{apx: P2N units}.
\end{remark}

Before concluding this section, we comment on equation \ref{eq: moduli space of PIGMM}, the claimed parameter space of PIGMM's. The convergence of the partition function in Theorem \ref{thm: PIGMM} is guaranteed for positive definite matrices $G^{\lambda}$ and the generating function \eqref{eq: PIGMM ZJ} exists since positive definite matrices are invertible.
\section{Observables: Directed graphs and set partitions}\label{subsec: observables}
In this section we will give two bases of observables. The first one is based on 1-row partition diagrams and is particularly useful for computing expectation values. The second basis uses directed graphs and is more amenable to combinatorial counting and construction methods.

To understand the space of observables we must first understand the space of matrix polynomials. The following proposition gives a description in terms of representations of $\SN$.
\begin{proposition}
	The vector space of polynomials in $X_{ij}$ is isomorphic to
	\begin{equation}
		\bigoplus_{k=0}^\infty \Sym^{k}(\VN^{\otimes 2}).
	\end{equation}
\end{proposition}
\begin{proof}
A general matrix polynomial has the form
\begin{equation}
	\sum_{k=0}^\infty a_{(k)}^{i_1 i_{1'}, \dots, i_{k} i_{k'}} \X_{i_1 i_{1'}} \dots \X_{i_{k} j_{k'}},
\end{equation}
with a finite number of non-zero coefficients $a_{(k)}$. Therefore, the space of matrix polynomials is a graded vector space, with grading given by the degree. For fixed degree $k$ the map
\begin{equation}
	\X_{i_1 i_{1'}} \dots \X_{i_{k} j_{k'}} \mapsto \frac{1}{k!} \sum_{\tau \in S_k} (e_{i_{(1)\tau}} \otimes e_{i_{(1)\tau'}}) \otimes \dots \otimes (e_{i_{(k)\tau}} \otimes e_{i_{(k)\tau'}}), \label{eq: Symk proj}
\end{equation}
is an isomorphism from the vector space of degree $k$ matrix polynomials to $\Sym^k(\VN^{\otimes 2})$. Applying this map for every degree $k$ proves the above theorem.

\end{proof}

As we discovered in section \ref{sec: partition algebras}, $\SN$ invariant tensors are closely related to set partitions. For invariant maps, it was useful to describe the set partitions in terms of 2-row partition diagrams.
For a basis of $(\VN^{\otimes 2k})^{\SN}$ it is natural to consider 1-row partition diagrams. We will use set partitions $\pi \in \setpart{\{1,\dots,2k\}}$ and diagrams interchangeably. As an example of this correspondence we have
\begin{align}
	1\vert 2 = \onerowdiagram{2}{}{1/2} &\longleftrightarrow \sum_{i_1, i_2}\delta^{i_1 i_2} e_{i_1} \otimes e_{i_2}, \\
	13|2|4 = \onerowdiagram{4}{}{1/3} &\longleftrightarrow \sum_{i_1, i_2, i_3, i_4} \delta^{i_1 i_3} e_{i_1} \otimes e_{i_2} \otimes e_{i_3} \otimes e_{i_4}, \\
	1234 = \onerowdiagram{4}{}{1/2,2/3,3/4} &\longleftrightarrow \sum_{i_1, i_2,i_3,i_4} \delta^{i_1 i_2}\delta^{i_2 i_3} \delta^{i_3 i_4} e_{i_1} \otimes e_{i_2} \otimes e_{i_3} \otimes e_{i_4}.
\end{align}
\begin{remark}	
The projection to $\Sym^k(\VN^{\otimes 2})$ is given on the r.h.s of \eqref{eq: Symk proj}. It identifies distinct invariant vectors, for example
\begin{equation}
	\onerowdiagram{4}{}{1/2} - \onerowdiagram{4}{}{3/4} \mapsto 0.
\end{equation}
Therefore, a basis of invariant matrix polynomials is labelled by equivalence classes of 1-row diagrams.
\end{remark}
For a set partition(1-row diagram) $\pi \in \setpart{\{1,\dots,2k\}}$ we write $X(\pi)$ for the corresponding matrix polynomial and vector in $\Sym^k(\VN^{\otimes 2})$, for example
\begin{align}
	\X(12|3|4) &= \X({\onerowdiagram{4}{}{1/2}}) 
	= \sum_{i_1, i_2, i_3, i_4} \delta^{i_1 i_2} X_{i_1 i_2} X_{i_3 i_4} \\ 
	&= \X({\onerowdiagram{4}{}{3/4}}).
\end{align}

We will now prove that there is a bijection from equivalence classes of 1-row partition diagrams and invariant matrix polynomials. For this, we will use an intermediate bijection, between equivalence classes of 1-row partition diagrams and directed graphs. We will then prove that directed graphs are in bijection with invariant matrix polynomials.
\begin{proposition}
	Let $\DG{k}$ be the set of unlabelled directed graphs with $k$ edges. There exists a surjective map
	\begin{align}
		P:\, &\setpart{\{1,\dots,2k\}} \rightarrow \DG{k} \\
		&\pi \mapsto G
	\end{align}
	with the property that the inverse image
	\begin{equation}
		P^{-1}(G) = \{\pi \in \setpart{\{1,\dots,2k\}} \, \vert \, P(\pi) ={G}\}
	\end{equation}
	satisfies
	\begin{equation}
		X(\pi) = X({\pi'})
	\end{equation}
	for all pairs $\pi, \pi' \in P^{-1}(G)$.
\end{proposition}
\begin{proof}
	We describe a directed graph $G \in \DG{k}$ with $l$ vertices by a collection of $k$ pairs $(i,j) \in \{1,\dots,l\} \times \{1,\dots,l\}$. Each ordered pair corresponds to an edge from the vertex $i$ to the vertex $j$. 
	
	The graph $G=P(\pi)$ is constructed as follows.
	Given a set partition $\pi = \{\pi_1, \dots, \pi_l\} \in \setpart{\{1,\dots,2k\}}$ we define an ordered list of blocks $\pi^o=(\pi_1, \dots, \pi_l)$ (for example any total ordering on the set of subsets of $\{1,\dots,2k\}$ will work). The ordered list $\pi^o$ defines a map
	\begin{align}
		p:\, &\{1,\dots,2k\} \rightarrow \{1,\dots,l\}
		&p(i) = j \quad \text{if $i$ is in $\pi_j$.}
	\end{align}
	The collection
	\begin{equation}
		G = P(\pi) = (p(1), p(2)), \dots, (p(2k-1),p(2k))
	\end{equation}
	of pairs is a graph $G \in \DG{k}$ with $l$ vertices, corresponding to the set partition $\pi$. Observe that re-ordering the blocks in $\pi^o$ corresponds to relabelling the vertices in $G(\pi)$.
	
	For this to be a surjection there should exist at least one $\pi$ for every graph $G$. Let $G \in \DG{k}$ with $l$ vertices be described by a collection of $k$ pairs. We want to match these with the set of pairs $(1,2), \dots, (2k-1,2k)$ to define a map $p$ as above. To do this, pick an ordering of the pairs in $G$ and label the entries of the ordered pairs as
	\begin{equation}
		(i_1, i_2), \dots, (i_{2k-1}, i_{2k})
	\end{equation}
	Now define the function
	\begin{equation}
		p(r) = i_r, \quad r=1,\dots,2k.
	\end{equation}
	This corresponds to the set partition made out of blocks $\pi_j$ with the property
	\begin{equation}
		\pi_{j} = \{r  \in \{1,\dots,2k\} \, \vert \, p(r) = j\}.
	\end{equation}
	Therefore, at least one $\pi = \pi_1 \vert \dots \vert \pi_l \in \setpart{\{1,\dots,2k\}}$ exists for every $G \in \DG{k}$.
	
	Lastly, note that two set partitions $\pi, \pi'$ give rise to the same graph $G$ if there exists $\tau \in \diag(S_k) \subset S_{k} \times S_k \subset S_{2k}$ permuting the $k$ pairs such that
	\begin{align}
		&(p((1)\tau), p((2)\tau)), \dots, (p((2k-1)\tau),p((2k)\tau)) \nonumber\\
		&= (p'(1), p'(2)), \dots, (p'(2k-1),p'(2k)) = P(\pi').
	\end{align}
	This is exactly the relation that makes two set partitions satisfy $X(\pi) = X(\pi')$.
\end{proof}
\begin{example}
	To illustrate this, we consider some simple examples for $k=2$. First, let $\pi = 12|3|4$ and pick $\pi^o=(\pi_1, \pi_2,\pi_3)=(\{1,2\}, \{3\}, \{4\})$ such that
	\begin{equation}
		G = P(12|3|4) = P({\onerowdiagram{4}{}{1/2}}) = \{(p(1), p(2)), (p(3),p(4))\} = \{(1,1), (2,3)\}.
	\end{equation}
	Had we chosen $\pi^o=(\{2\},\{3\},\{1,2\})$ we would find
	\begin{equation}
		G=P(12|3|4) = \{(p(1), p(2)), (p(3),p(4))\} = \{(3,3), (1,2)\}.
	\end{equation}
	As unlabelled directed graphs these are the same.
\end{example}
\begin{example}
	We give an example of the opposite construction. Let $G$ be the graph described by the set of pairs
	\begin{equation}
		G = \{(1,1), (2,3), (2,4)\}.
	\end{equation}
	We order them as written above. This gives the map
	\begin{equation}
		p(1) = 1, p(2) = 1, p(3) =2, p(4) = 3, p(5) = 2, p(6) = 4.
	\end{equation}
	The corresponding blocks of $\pi$ are $\pi_1 = \{1,2\}, \pi_2 = \{3,5\}, \pi_3 = \{4\}, \pi_4 = \{6\}$ or
	\begin{equation}
		\pi = \onerowdiagram{6}{}{1/2,3/5}.
	\end{equation}
	Had we ordered the pairs as $(2,3), (1,1), (2,4)$ we get the map
	\begin{equation}
		p(1)=2, p(2) =3, p(3)=1,p(4)=1, p(5)=2, p(6)=4
	\end{equation}
	and set partition with blocks $\pi'_1 = \{3,4\}, \pi'_2 = \{1,5\}, \pi'_3 = \{2\}, \pi'_4=\{6\}$ or
	\begin{equation}
		\pi' = \onerowdiagram{6}{}{1/5,3/4}.
	\end{equation} 
	However, these are equivalent in the sense of $X(\pi) = X(\pi')$ since we can swap the vertices $1,2$ with the vertices $3,4$ in the 1-row diagrams.
\end{example}

The space of degree $k$ observables is the space of invariants under the action of $\SN \times S_k$ on $\VN^{\otimes 2k}$. We denote this subspace of invariants by $[\VN^{\otimes 2k}]^{\SN \times S_k}$. Define the following projectors on $\VN^{\otimes 2k}$
\begin{align}
	&\begin{aligned}[t]
		P_{[\N]}(e_{i_1} \otimes e_{i_{1'}} \otimes \dots &\otimes e_{i_k} \otimes e_{i_{k'}}) \\&= \frac{1}		{\N!}\sum_{\sn \in \SN} e_{(i_1)\sn} \otimes e_{(i_{1'})\sn} \otimes \dots \otimes e_{(i_k)\sn} \otimes e_{(i_{k'})\sn},
	\end{aligned}\\
	&\begin{aligned}[t]
		P_{[k]}(e_{i_1} \otimes e_{i_{1'}} \otimes \dots &\otimes e_{i_k} \otimes e_{i_{k'}}) \\ &= \frac{1}{k!} \sum_{\tau \in S_k}   e_{i_{(1)\tau}} \otimes e_{i_{(1)\tau'}}  \otimes \dots \otimes  e_{i_{(k)\tau}} \otimes e_{i_{(k)\tau'}}.
	\end{aligned}
\end{align}
Then
\begin{equation}
	\dim \, [\VN^{\otimes 2k}]^{\SN \times S_k} = \Tr_{{{\VN^{\otimes 2k}}}}(P_{[\N]} P_{[k]})
\end{equation}
and we will prove that
\begin{proposition}\label{prop: graphs equals trace}
	Let $\mathcal{G}_{k,\N}$ be the set of unlabelled directed graphs with $k$ edges and $\N$ vertices, then
	\begin{equation}
		\dim \, [\VN^{\otimes 2k}]^{\SN \times S_k} = \abs{\mathcal{G}_{k,\N}}.
	\end{equation}
\end{proposition}
\begin{proof}
	A directed graph with $l$ labelled vertices and $k$ labelled edges corresponds to an ordered list of $k$ pairs $(i_a,i_{a'}) \in \{1,\dots, l\}^{\times 2}$ with $a \in \{1,\dots,k\}$. Unlabelled directed graphs are in one-to-one correspondence with orbits of such sets under the action
	\begin{equation}
		(i_a,i_{a'}) \mapsto ((i_a)\sn, (i_{a'})\sn), \quad \sn \in S_l
	\end{equation}
	and
	\begin{equation}
		(i_a, i_{a'}) \mapsto (i_{(a)\tau}, i_{(a)'\tau}) \quad \tau \in S_k.
	\end{equation}
	By Burnside's lemma the number of orbits is equal to the average number of fixed points.
	The ordered list is fixed by $\sn \in S_l$, $\tau \in S_k$ if
	\begin{equation}
		(i_a,i_{a'}) = ((i_{(a)\tau})\sn, (i_{(a)'\tau})\sn), \quad \forall a \in \{1,\dots,k\},
	\end{equation}
	or equivalently
	\begin{equation}
		\prod_{a=1}^k \delta_{i_a, (i_{(a)\tau})\sn} \delta_{i_{a'}, (i_{(a)'\tau})\sn} = 1.
	\end{equation}
	Therefore, the number of orbits is equal to
	\begin{equation}
		\abs{\mathcal{G}_{k,l}}=\sum_{i_a,i_{a'} = 1}^l \frac{1}{l!} \frac{1}{k!}\sum_{\sn \in S_l} \sum_{\tau \in S_k}\prod_{a=1}^k \delta_{i_a, (i_{(a)\tau})\sn} \delta_{i_{a'}, (i_{(a)'\tau})\sn} .
	\end{equation}	
	On the other hand, this is equal to
	\begin{equation}
		\Tr_{{{V_l^{\otimes 2k}}}}(P_{[l]} P_{[k]})
	\end{equation}
	for $l=\N$ since the trace can be computed as
	\begin{align}
		&\sum_{i_a,i_{a'}=1}^\N \VNinner{e_{i_1} \otimes e_{i_{1'}} \otimes \dots \otimes e_{i_k} \otimes e_{i_{k'}}}{P_{[\N]} P_{[k]}e_{i_1} \otimes e_{i_{1'}} \otimes \dots \otimes e_{i_k} \otimes e_{i_{k'}}} \\
		&=\frac{1}{\N!} \frac{1}{k!}\sum_{\sn \in \SN} \sum_{\tau \in S_k} \sum_{i_a,i_{a'}=1}^\N  \prod_{a=1}^k \VNinner{e_{i_a}}{e_{(i_{(a)\tau})\sn}}  \VNinner{e_{i_{a'}}}{e_{(i_{(a)'\tau})\sn}} \\
		&=\frac{1}{\N!} \frac{1}{k!}\sum_{\sn \in \SN} \sum_{\tau \in S_k} \sum_{i_a,i_{a'}=1}^\N  \prod_{a=1}^k  \delta_{i_a, (i_{(a)\tau})\sn} \delta_{i_{a'}, (i_{(a)'\tau})\sn}.
	\end{align}
	To conclude we have found that
	\begin{equation}
		\abs{\mathcal{G}_{k,\N}} = \Tr_{{{\VN^{\otimes 2k}}}}(P_{[\N]} P_{[k]}) = \dim \, [\VN^{\otimes 2k}]^{\SN \times S_k}.
	\end{equation}
\end{proof}

The previous results gives an immediate corollary on the stability of the dimension of the space of observables.
\begin{corollary}
	Since the maximum number of vertices that can be occupied by $k$ edges is equal to $2k$ we have
	\begin{equation}
		\abs{\mathcal{G}_{k,2k}} = \abs{\mathcal{G}_{k,2k+p}}, 
	\end{equation}
	for all non-negative integers $p$. Consequently
	\begin{equation}
		\dim \, [\VN^{\otimes 2k}]^{\SN \times S_k},
	\end{equation}
	stabilizes at $\N = 2k$.
\end{corollary}

The trace in
\begin{equation}
	\abs{\mathcal{G}_{k,\N}} = \Tr_{{{\VN^{\otimes 2k}}}}(P_{[\N]} P_{[k]}) 
\end{equation}
was computed in \cite[Appendix B]{Kartsaklis2017}. It was found that
\begin{align} \label{Eqn: Non-coloured trivial counting} \nonumber 
	\abs{\mathcal{G}_{k,\N}} &= \sum_{p \vdash \N} \sum_{q \vdash k}  \frac{1}{\prod_{i=1}^\N i^{p_i} p_i! \prod_{i=1}^k i^{q_i} q_i!} \prod^{k}_{i=1} \big(\sum_{l|i} l p_l \big)^{2q_i}.
\end{align}
The sums are over $l | i$, the set of divisors of $i$, and partitions $p = \{p_1, p_2, \dots, p_\N \}$, $q = \{ q_1, q_2, \dots, q_k \}$ obeying $\sum_i i p_i = \N$, $\sum_i i q_i = k$ respectively. A more detailed derivation of this result is contained within the appendices of \cite{Kartsaklis2017}. In Table \ref{tab: Table of invariant dimensions} we give $\abs{\mathcal{G}_{k, \N=2k}}$ for various $k$.
\begin{table}[h!]
	\begin{center}
		\caption{Number of invariants contained within $\Sym^k(\VN^{\otimes 2})$}
		\label{tab: Table of invariant dimensions}
		\begin{tabular}{l | l}
			$k$ & $\abs{\mathcal{G}_{k, \N=2k}}$\\
			\hline
			1 &  2 \\
			2 & 11 \\
			3 & 52 \\
			4 & 296 \\
			5 & 1724 \\
			6 & 11060 \\
			\end{tabular}
	\end{center}
\end{table}

In section \ref{subsec: graph counting} we give a group theoretical algorithm for enumerating and counting directed graphs. Having introduced the equivalence classes of 1-row diagrams, we can describe a combinatorial algorithm for computing expectation values of invariant polynomials.

\section{Expectation values: Combinatorial algorithms} \label{subsec: exp vals}
In this section we will give a combinatorial algorithm for computing expectation values $\expval{X(\pi)}$, for general permutation invariant matrix polynomials labelled by $\pi \in \setpart{\{1,\dots,2k\}}$.

As previously mentioned, expectation values can be computed exactly through Wick's theorem which says
\begin{theorem}[Wick's theorem]
	Consider a Gaussian distribution with mean/1-point function/expectation value $\expval{X_{ij}}$ and covariance/two-point function/propagator $\expval{X_{ij}X_{kl}}_c$. For $\pi=\pi_1 \vert \dots \vert \pi_b$ a partition of the set $\{(1,2), \dots, (2k-1,2k)\}$ with parts of size one or two, define
	\begin{equation}
		\expval{X}_{\pi_i} = \begin{cases}
			\expval{X_{i_a i_b}} \quad &\text{for $(a,b) \in \pi_i$, if $\abs{\pi_i} = 1$,}\\
			\expval{X_{i_a i_b} X_{i_c i_d}}_c \quad &\text{for $(a,b), (c,d) \in \pi_i$ if $\abs{\pi_i} = 2$.}
		\end{cases}
	\end{equation}
	The expectation value of a product of matrix elements is equal to
	\begin{equation}
		\expval{X_{i_1 i_2} \dots X_{i_{2k-1} i_{2k}} } = \sum_{\pi} \prod_{i=1}^{\abs{\pi}} \expval{X}_{\pi_i}, \label{eq: wicks thm}
	\end{equation}
	where the sum is over all set partitions of the kind mentioned above.
\end{theorem}
\begin{example}
	The first non-trivial example is $k=3$. Note that there are $4$ set partitions of $\{(1,2),(3,4),(5,6)\}$ with parts of size one or two:
	\begin{align}
		&\pi = (1,2) \vert (3,4)  \vert (5,6), \quad \pi = (1,2) (3,4)  \vert (5,6) \\
		&\pi = (1,2) (5,6) \vert (3,4), \quad \pi = (3,4)  (5,6) \vert (1,2).
	\end{align}
	Thus,
	\begin{equation} \label{eq: wicks thm degree 3}
		\expval{X_{i_1 i_2} X_{i_3 i_4} X_{i_5 i_6}} =\begin{aligned}[t]
			 &\expval{X_{i_1 i_2}}\expval{X_{i_3 i_4}}\expval{X_{i_5 i_6}} + \expval{X_{i_1 i_2} {X_{i_3 i_4}}}_c\expval{X_{i_5 i_6}}   \\
			+ &\expval{X_{i_1 i_2} {X_{i_5 i_6}}}_c\expval{X_{i_3 i_4}} + \expval{X_{i_3 i_4} {X_{i_5 i_6}}}_c\expval{X_{i_1 i_2}}.
		\end{aligned}
	\end{equation}
\end{example}

\subsection{1-row diagrams from 2-row diagrams.} \label{subsec: 1row to 2-row}
Each term on the r.h.s. of \eqref{eq: wicks thm} can be expanded in terms of Clebsch-Gordan coefficients and invariant tensors $Q^{\lambda; \alpha \beta}_{ij:kl}$ using \eqref{eq: matrix basis 1pt}, \eqref{eq: matrix basis con 2pt}. We now connect these invariant tensors to the matrix units constructed in \ref{sec: construction of units} and give them a 1-row partition interpretation.

Clebsch-Gordan coefficients $C_{ij}^{[\N] \alpha}$ correspond to matrix units of $P_1(\N)$ since
\begin{equation}
	\Hom_{\SN}(\VN^{\otimes 2}, \mathbb{C}) \cong \End_{\SN}(\VN), \label{eq: Hom VN2 End VN iso}
\end{equation}
as vector spaces.
The matrix units are (see Example \ref{ex: P1N units})
\begin{align}
	&Q^{[\N]}=  \frac{1}{\N} \PAdiagram{1}{}{}, \\
	&Q^{[\N-1,1]} = \PAdiagram{1}{-1/1}{} - \frac{1}{\N}\PAdiagram{1}{}{}.
\end{align}
The corresponding Clebsch-Gordan coefficients are written as linear combinations of 1-row partitions. We define
\begin{equation}
	\vactab_1 = ([\N], [\N-1], [\N], [\N-1], [\N]), \quad \vactab_2 = ([\N], [\N-1], [\N-1,1], [\N-1], [\N]) \\
\end{equation}
then
\begin{align}
	C_{ij}^{[\N], \vactab_1} &= \frac{1}{\N} \onerowdiagramunLabeled{2}{}\\
	C_{ij}^{[\N], \vactab_2} &= \frac{1}{\sqrt{\N-1}} \onerowdiagramunLabeled{2}{2/1} - \frac{1}{\N \sqrt{\N-1}} \onerowdiagramunLabeled{2}{}.
\end{align}
The translation between 2-row diagrams and 1-row diagrams is defined by placing the top row to the right of the bottom row. The difference in normalization comes from demanding that the Clebsch-Gordan coefficients define vectors with norm 1. Note that the normalization constants are
\begin{align}
	&\sqrt{\Tr_{V_N}(Q^{[\N]})} = 1,\\
	&\sqrt{\Tr_{V_N}(Q^{[\N-1,1]})} = \sqrt{\dimSN{[\N-1,1]}} = \sqrt{\N-1}.
\end{align}

Similarly for the invariant tensors of degree two.
Recall that matrix units $Q^{\lambda}_{\alpha \beta}$ of $P_2(\N)$ have an expansion in the diagram basis. We write the corresponding element of $\End_{\SN}(\VN^{\otimes 2})$ as
\begin{equation}
	(Q^{\lambda}_{\alpha \beta})^{i_{3} i_{4}}_{i_1 i_2}.
\end{equation}
This is identified with the invariant tensor of degree two
\begin{equation}
	Q^{\lambda; \alpha \beta}_{i_1 i_2: i_{3} i_{4}}.
\end{equation}
The diagrammatic translation between the two is again defined by placing the top row to the right of the bottom row. For example
\begin{equation}
	\PAdiagram{2}{1/-1}{} \mapsto \onerowdiagramunLabeled{4}{3/1}, \quad \PAdiagram{2}{1/-2}{} \mapsto \onerowdiagramunLabeled{4}{4/1}
\end{equation}
With this translation, the invariant tensors are formal linear combinations of 1-row diagrams.

Wick's theorem \eqref{eq: wicks thm} includes tensor products of invariant tensors. Thus, it is natural to define a tensor product on the space of 1-row diagrams. The tensor product $d \otimes d'$ of two 1-row diagrams is the union of the 1-row diagrams. For example
\begin{equation}
	\onerowdiagramunLabeled{4}{3/1} \otimes \onerowdiagramunLabeled{4}{2/1} = \onerowdiagramunLabeled{8}{3/1,5/6}.
\end{equation}
We remark that at the level of set partitions, this procedure requires a relabelling in the second set partitions. Considering the same example but in set partition language, we have
\begin{equation}
	13\vert 2\vert 4 \otimes 12 \vert 3 \vert 4 = 13 \vert 56 \vert 2 \vert 4 \vert 7 \vert 8.
\end{equation}
This allows us to expand Wick's theorem in terms linear combinations of tensor products of 1-row diagrams (for example at degree three \eqref{eq: wicks thm degree 3}, using \eqref{eq: matrix basis 1pt}, \eqref{eq: matrix basis con 2pt}). That is, we can write
\begin{equation}
	\expval{X_{i_1 i_2} \dots X_{i_{2k-1} i_{2k}} } = \sum_{\pi \in \setpart{[2k]}} c_\pi(\N) \pi, \label{eq: exp val as 1row sum}
\end{equation}
where $\pi$ is a 1-row partition diagram with $2k$ vertices, and $c_\pi(\N)$ are coefficients that are computed through Wick's theorem. Examples of this expansion are given below.

\subsection{Join of 1-row diagrams.}
We now describe the different contributions to expectation values of observables in terms of operations on 1-row diagrams. We begin by considering a toy example. Consider the expansion of
\begin{equation}
	\expval{X_{i_1 i_2} X_{i_3 i_4} X_{i_5 i_6}}
\end{equation}
using Equation \ref{eq: exp val as 1row sum} and suppose it contained a tensor
\begin{equation}
	\delta_{i_1 i_2} \delta_{i_2 i_3} \delta_{i_4 i_5}  \longleftrightarrow 	\onerowdiagramunLabeled{6}{2/1,3/2, 5/4}.
\end{equation}
As an example, we want to compute its contribution to the expectation value
\begin{equation}
	\expval{\sum_{i,j} X_{ii} X_{ij} X_{jj}} = \expval{X(\onerowdiagramunLabeled{6}{2/1,3/2,5/4,6/5})}.
\end{equation}
Writing
\begin{equation}
	\sum_{i,j} X_{ii} X_{ij} X_{jj} = \sum_{i_1, \dots i_6} \delta^{i_1 i_2}\delta^{i_2 i_3} \delta^{i_4 i_5} \delta^{i_5 i_6}X_{i_1 i_2} X_{i_3 i_4} X_{i_5 i_6},
\end{equation}
the contribution corresponds to
\begin{equation}
	\sum_{i_1, \dots i_6} \delta^{i_1 i_2}\delta^{i_2 i_3} \delta^{i_4 i_5} \delta^{i_5 i_6} \delta_{i_1 i_2} \delta_{i_2 i_3} \delta_{i_4 i_5}  = \N^2.
\end{equation}
Our algorithm is based on the simple observation that the power of $\N$ on the r.h.s. is equal to the number of components in the join
\begin{equation}
	\onerowdiagramunLabeled{6}{2/1,3/2,5/4,6/5} \join \onerowdiagramunLabeled{6}{2/1,3/2, 5/4} = \PAdiagram{6}{-1/1,-2/2,-3/3,-4/4,-5/5/,-6/6}{-2/-1,-3/-2,-5/-4,-6/-5,2/1,3/2,5/4} = \onerowdiagramunLabeled{6}{2/1,3/2,5/4,6/5}.
\end{equation}
We give the general definition.
\begin{definition}[Join] \label{def: join}
	The join $\pi \join {\pi'}$ of two 1-row diagrams $\pi, \pi' \in \setpart{[k]}$ is constructed as follows. Place $\pi $ above $\pi'$ and connect the $i$th vertex of $\pi$ to the $i$th vertex of $\pi'$. Simplify the diagram into a 1-row diagram $\pi \join {\pi'}$ where vertex $i$ is connected to vertex $j$ if they are in the same part of the above 2-row diagram.
\end{definition}

Given this observation and \eqref{eq: exp val as 1row sum} we have the following result.
\begin{corollary}
Let $\pi \in \setpart{[2k]}$ be a 1-row diagram and $X(\pi)$ the corresponding polynomial. Then the expectation value
\begin{equation}
	\expval{X(\pi)} = \sum_{\pi' \in \setpart{[2k]}} c_{\pi'} \N^\abs{\pi \join \pi'},
\end{equation}
where $\abs{\pi \join \pi'}$ is the number of components in the join of $\pi$ and $\pi'$.
\end{corollary}

The Sage code implementing this algorithm is detailed in Appendix \ref{apx: EV code}.
Instead of explaining this implementation here we will exemplify the algorithm outlined above for the cases of $k=1,2$.

\subsection{Example 1.}
For $k=1$ the only contribution to the expectation value comes from Clebsch-Gordan coefficients
\begin{equation}
	\expval{X_{i_1 i_2}} = J^\beta\qty(
	\frac{(G^{-1})_{[\N]; \vactab_1 \beta}}{\N} \onerowdiagram{2}{}{} +
	\frac{(G^{-1})_{[\N]; \vactab_2 \beta}}{\sqrt{\N-1}} \onerowdiagram{2}{}{2/1} - \frac{(G^{-1})_{[\N];\vactab_2 \beta}}{\N\sqrt{\N-1}} \onerowdiagram{2}{}{}).
\end{equation}
Therefore, Wick's theorem gives
\begin{align}
	c_{1|2}(\N) &= \frac{J^\beta}{\N}\qty((G^{-1})_{[\N]; \vactab_1\beta} - \frac{(G^{-1})_{[\N];  \vactab_2 \beta}}{\sqrt{\N-1}}) \\
	c_{12}(\N)  &= \frac{J^\beta(G^{-1})_{[\N]; \vactab_2\beta}}{\sqrt{\N-1}}.
\end{align}
This allows us to compute the expectation value of $X(1|2)$
\begin{align}
	\expval{X(1|2)} &= c_{1|2}(\N)\N^{\abs{1|2 \join 1|2}} + c_{12}(\N)\N^{\abs{1|2 \join 12}} =  c_{1|2}(\N)\N^{2} + c_{12}(\N)\N^{1} \\
	&=  {J^\beta}\N (G^{-1})_{[\N]; \vactab_1\beta}
\end{align}
Similarly for $X(12)$
\begin{align}
		\expval{X(12)} &= c_{1|2}(\N)\N^{\abs{12 \join 1|2}} + c_{12}(\N)\N^{\abs{12 \join 12}} =  c_{1|2}(\N)\N^{1} + c_{12}(\N)\N^{1} \\
		&= J^\beta \qty((G^{-1})_{[\N]; \vactab_1 \beta} + (\N-1)\frac{(G^{-1})_{[\N]; \vactab_2 \beta}}{\sqrt{\N-1}}).
\end{align}

In what follows, it will be useful to define
\begin{equation}
	\mu_1 = J^\alpha (G^{-1})_{[\N],\vactab_1,\alpha}, \quad \mu_2 = J^\alpha (G^{-1})_{[\N],\vactab_2,\alpha},
\end{equation}
such that
\begin{equation}
	\expval{X(12)} = \mu_1 + \sqrt{\N-1}\mu_2
\end{equation}

\subsection{Example 2.}
The case of $k=2$ is more interesting because the tensor product of 1-row diagrams enters the computation. We will not compute the full contributions here, but from Wick's theorem we have a contribution
\begin{equation}
	\expval{X_{i_1 i_2}}\expval{X_{i_3 i_4}}
\end{equation}
to the degree two expectation value. We will study this contribution because it showcases the tensor product of 1-row diagrams. In the previous section we gave the diagrammatic expansion of the degree one expectation values. The tensor product we want to compute is
\begin{align}
		&\qty(\frac{\mu_1}{\N} \onerowdiagram{2}{}{} +
			\frac{\mu_2}{\sqrt{\N-1}} \onerowdiagram{2}{}{2/1} -\frac{\mu_2}{\N\sqrt{\N-1}} \onerowdiagram{2}{}{})
			\otimes	\qty(\frac{\mu_1}{\N} \onerowdiagram{2}{}{} +
			\frac{\mu_2}{\sqrt{\N-1}} \onerowdiagram{2}{}{2/1} -\frac{\mu_2}{\N\sqrt{\N-1}} \onerowdiagram{2}{}{}) \nonumber \\
	&=\left(\begin{aligned}
		&\frac{\mu_1 \mu_1}{\N^2}\onerowdiagram{4}{}{} + \frac{\mu_1 \mu_2 }{\N\sqrt{\N-1}}\onerowdiagram{4}{}{4/3}  -\frac{\mu_1\mu_2
		}{\N^2\sqrt{\N-1}} \onerowdiagram{4}{}{}\\
		& +	\frac{\mu_1 \mu_2
		}{\N \sqrt{\N-1}} \onerowdiagram{4}{}{2/1} +\frac{\mu_2\mu_2
		}{\N-1}  \onerowdiagram{4}{}{2/1,4/3} -\frac{\mu_2\mu_2
		}{\N(\N-1)}  \onerowdiagram{4}{}{2/1} \\
		&-\frac{\mu_1 \mu_2
		}{\N^2\sqrt{\N-1}}  \onerowdiagram{4}{}{} -\frac{\mu_2\mu_2
		}{\N(\N-1)}  \onerowdiagram{4}{}{4/3} +\frac{\mu_2\mu_2
		}{\N^2\N-1}  \onerowdiagram{4}{}{}
	\end{aligned}\right)
\end{align}
Therefore, we find that it contributes to $c_{12|34}(\N)$ by
\begin{equation}
	J^\alpha J^\beta \frac{\mu_2 \mu_2}{\N-1}  
\end{equation}
and $c_{1|2|3|4}(\N)$ by
\begin{equation}
	\frac{1}{\N^2}\qty(\mu_1 - \frac{\mu_2}{\sqrt{\N-1}} )^2
%	\qty((G^{-1})_{[\N]; ([\N], [\N], [\N]) \beta}+ \frac{(G^{-1})_{[\N-1,1]; ([\N], [\N], [\N-1,1])\beta}}{\sqrt{\N-1}}).
\end{equation}

Having computed some of the $c_\pi(\N)$, we can compute their contributions to expectation values of degree two observables. For example $X(1|23|4)$ receives contributions from the above given by
\begin{equation}
	c_{12|34}(\N)\N^{\abs{1|23|4 \join 12|34}} +c_{1|2|3|4}(\N) \N^{\abs{1|23|4 \join 1|2|3|4}} = c_{12|34}(\N)\N^{1} +c_{1|2|3|4}(\N) \N^{3}.
\end{equation}

\subsection{Normalization}\label{subsection: normalization}
In the next subsection we will compare the results of this algorithm to the previously know results \cite{Kartsaklis2017, Ramgoolam2019a, Barnes2022b}. Before this, we need to fix the normalization of the $P_2(\N)$ matrix units, or equivalently the constants $(G^{-1})_{\lambda; \alpha \beta}$ multiplying them.

To fix the normalization, consider the combination
\begin{equation}
	\sum_{i,j,k,l} \expval{X_{ij}X_{kl}}_c \expval{X_{kl}X_{ij}}_c = (G^{-1})_{\lambda, \alpha \beta}(G^{-1})_{\lambda'; \alpha' \beta'} Q^{\lambda; \alpha \beta}_{ij;kl}  Q^{\lambda'; \alpha' \beta'}_{kl;ij}.
\end{equation}
From the normalization condition \eqref{eq: CGC norm} we find
\begin{equation}
	\sum_{i,j,k,l} \expval{X_{ij}X_{kl}}_c \expval{X_{kl}X_{ij}}_c = \sum_{\lambda, \alpha, \beta} (G^{-1})_{\lambda, \alpha \beta}(G^{-1})_{\lambda, \alpha \beta}\dimSN{\lambda},
\end{equation}
where we have used the fact that parameter matrices are symmetric to put it into a sum of squares form.

We will now compute this using the matrix units in Appendix \ref{apx: P2N units}. Matching the factor of $\dimSN{\lambda}$ fixes the normalization up to signs. To keep equations short, we introduce the following short-hands for vacillating tableaux of length $k=2$
\begin{align}
	&\vactab_1 = ([\N], [\N-1], [\N], [\N-1], [\N]),\\
	&\vactab_2 = ([\N], [\N-1], [\N-1,1], [\N-1], [\N]) \\
	&\vactab_3 = ([\N], [\N-1], [\N], [\N-1], [\N-1,1]), \\
	&\vactab_4 = ([\N], [\N-1], [\N-1,1], [\N-1], [\N-1,1]) \\
	&\vactab_5 = ([\N], [\N-1], [\N-1,1], [\N-2,1], [\N-1,1]), \\
	&\vactab_6 = ([\N], [\N-1], [\N-1,1], [\N-2,1], [\N-2,2]), \\
	&\vactab_7 = ([\N], [\N-1], [\N-1,1], [\N-2,1], [\N-2,1,1]).
\end{align}
In this notation, computing the above combination of correlators using the matrix units in \ref{apx: P2N units} gives
\begin{equation}
	\begin{aligned}[t]
	&\frac{1}{\N^2}(G^{-1})_{[\N]; \vactab_1 \vactab_1 }^2+
	\frac{(\N-1)}{2\N^2}(G^{-1})_{[\N]; \vactab_2 \vactab_1 }^2+
	\frac{(\N-1)^2}{\N^2}(G^{-1})_{[\N]; \vactab_2 \vactab_2 }^2 +\\
	&\frac{(\N-1)}{\N^2}(G^{-1})_{[\N-1,1]; \vactab_3 \vactab_3 }^2+ \frac{(\N-1)^2}{2\N^2}(G^{-1})_{[\N-1,1]; \vactab_4 \vactab_3 }^2+ \frac{(\N-1)^3}{\N^2}(G^{-1})_{[\N-1,1]; \vactab_4 \vactab_4 }^2 +\\ &\frac{(\N-2)}{2}(G^{-1})_{[\N-1,1]; \vactab_5 \vactab_3 }^2+ 
	\frac{(\N-1)(\N-2	)}{2}(G^{-1})_{[\N-1,1]; \vactab_5 \vactab_4 }^2+ \frac{1}{(\N-1)}(G^{-1})_{[\N-1,1]; \vactab_5 \vactab_5 }^2+\\ &\frac{\N(\N-3)}{2}(G^{-1})_{[\N-2,2]; \vactab_6 \vactab_6 }^2+ \frac{(\N-1)(\N-2)}{2}(G^{-1})_{[\N-2,1,1]; \vactab_7 \vactab_7 }^2.
	\end{aligned}
\end{equation}
Therefore, the correct normalization of matrix units is given by
\begin{equation}\label{eq: matrix units normalizations}
	\begin{aligned}[t]
	&Q^{[\N]}_{\vactab_1 \vactab_1} \mapsto \N Q^{[\N]}_{\vactab_1 \vactab_1} \\
	&Q^{[\N]}_{\vactab_2 \vactab_1} \mapsto \N \sqrt{\frac{2}{\N-1}}Q^{[\N]}_{\vactab_2 \vactab_1} \\
	&Q^{[\N]}_{\vactab_2 \vactab_2} \mapsto {\frac{\N}{(\N-1)}}Q^{[\N]}_{\vactab_2 \vactab_2} \\
	&Q^{[\N-1,1]}_{\vactab_3 \vactab_3} \mapsto \N Q^{[\N-1,1]}_{\vactab_3 \vactab_3} \\
	&Q^{[\N-1,1]}_{\vactab_4 \vactab_3} \mapsto \N \sqrt{\frac{2}{\N-1}}Q^{[\N-1,1]}_{\vactab_4 \vactab_3} \\
	&Q^{[\N-1,1]}_{\vactab_4 \vactab_4} \mapsto {\frac{\N}{(\N-1)}}Q^{[\N-1,1]}_{\vactab_4 \vactab_4} \\	
	&Q^{[\N-1,1]}_{\vactab_5 \vactab_3} \mapsto \sqrt{\frac{2(\N-1)}{(\N-2)}}Q^{[\N-1,1]}_{\vactab_5 \vactab_3} \\	
	&Q^{[\N-1,1]}_{\vactab_5 \vactab_4} \mapsto \sqrt{\frac{2}{\N-2}}Q^{[\N-1,1]}_{\vactab_5 \vactab_4} \\
	&Q^{[\N-1,1]}_{\vactab_5 \vactab_5} \mapsto (N-1)Q^{[\N-1,1]}_{\vactab_5 \vactab_5}.
	\end{aligned}
\end{equation}
while the $\lambda = [N-2,2], [N-2,1,1]$ units are correctly normalized (up to signs) as stated in Appendix \ref{apx: P2N units}. Alternatively, the same transformations can be performed on the parameters $(G^{-1})_{\lambda; \alpha \beta}$.

\subsection{Matching results}
Having fixed the normalizations in the previous subsection, we can now compare the outputs of our algorithm to previous results. Since we are using a different orthonormal basis for the multiplicity space compared to \cite{Ramgoolam2019a}, we should expect the coupling constants to be related by a change of basis. We will now compute the set of degree two expectation values necessary to set up a system of linear equations for finding the change of basis by comparing the results of \cite[Section 3]{Ramgoolam2019a}.

Computing the expectation value $\expval{X(13|24)}=\sum_{i,j} \expval{X_{ij} X_{ij}}$ using the Sage code gives
\begin{align}
		&\mu_1^2 + \mu_2^2 + (G^{-1})_{[\N]; \vactab_1 \vactab_1}+(G^{-1})_{[\N]; \vactab_2 \vactab_2}
		+(\N-1)(G^{-1})_{[\N-1,1]; \vactab_3 \vactab_3} \\
		&+(\N-1)(G^{-1})_{[\N-1,1]; \vactab_4 \vactab_4}
		+(\N-1)(G^{-1})_{[\N-1,1]; \vactab_5 \vactab_5}\\
		&+\frac{\N(\N-3)}{2}(G^{-1})_{[\N-2,2]; \vactab_6 \vactab_6}
		-\frac{(\N-1)(\N-2)}{2}(G^{-1})_{[\N-2,1,1]; \vactab_7 \vactab_7}.
\end{align}

The expectation value $\expval{X(14|23)}=\sum_{i,j} \expval{X_{ij} X_{ji}}$ is equal to
\begin{align}
		&\mu_1^2 + \mu_2^2 + (G^{-1})_{[\N]; \vactab_1 \vactab_1}+(G^{-1})_{[\N]; \vactab_2 \vactab_2}
		+\sqrt{2(\N-1)}(G^{-1})_{[\N-1,1]; \vactab_4 \vactab_3} \\
		&+\sqrt{2(\N-1)(\N-2)}(G^{-1})_{[\N-1,1]; \vactab_5 \vactab_3}
		-\sqrt{2(\N-2)}(G^{-1})_{[\N-1,1]; \vactab_5 \vactab_4} \\
		&+\sqrt{2(\N-2)}(G^{-1})_{[\N-1,1]; \vactab_5 \vactab_5} \\
		&+\frac{\N(\N-3)}{2}(G^{-1})_{[\N-2,2]; \vactab_6 \vactab_6}
		+\frac{(\N-1)(\N-2)}{2}(G^{-1})_{[\N-2,1,1]; \vactab_7 \vactab_7}.
\end{align}

The expectation value $\expval{X(123|4)}=\sum_{i,j} \expval{X_{ii} X_{ij}}$ is equal to
\begin{align}
		&\mu_1^2 +
		\sqrt{\N-1}\mu_1 \mu_2+
		(G^{-1})_{[\N]; \vactab_1 \vactab_1}+
		\sqrt{\frac{\N-1}{2}}(G^{-1})_{[\N]; \vactab_2 \vactab_1}\\
		&+ \sqrt{\frac{\N-1}{2}}(\N-1)(G^{-1})_{[\N-1,1]; \vactab_4 \vactab_3}
		+(\N-1)(G^{-1})_{[\N-1,1]; \vactab_3 \vactab_3}
\end{align}

The expectation value $\expval{X(124|3)}=\sum_{i,j} \expval{X_{ii} X_{ji}}$ is equal to
\begin{align}
		&\mu_1^2 +
		\sqrt{\N-1}\mu_1 \mu_2+
		(G^{-1})_{[\N]; \vactab_1 \vactab_1}+
		\sqrt{\frac{\N-1}{2}}(G^{-1})_{[\N]; \vactab_2 \vactab_1}\\
		&+ \sqrt{\frac{\N-1}{2}}(G^{-1})_{[\N-1,1]; \vactab_4 \vactab_3} 
		+(\N-1)(G^{-1})_{[\N-1,1]; \vactab_4 \vactab_4} \\
		&+\sqrt{\frac{(\N-1)(\N-2)}{2}}(G^{-1})_{[\N-1,1]; \vactab_5 \vactab_3}
		+\sqrt{\frac{(\N-2)}{2}}(\N-1)(G^{-1})_{[\N-1,1]; \vactab_5 \vactab_4}
\end{align}

The expectation value $\expval{X(13|2|4)}=\sum_{i,j,k} \expval{X_{ij} X_{ik}}$ is equal to
\begin{align}
		&\mu_1^2 + \N(G^{-1})_{[\N]; \vactab_1 \vactab_1}+
		\sqrt{2(\N-2)}\N(G^{-1})_{[\N]; \vactab_5 \vactab_4}\\
		&+(\N-2)(G^{-1})_{[\N-1,1]; \vactab_5 \vactab_5} 
		+\N(G^{-1})_{[\N-1,1]; \vactab_4 \vactab_4}
\end{align}

The expectation value $\expval{X(24|1|3)}=\sum_{i,j,k} \expval{X_{ij} X_{kj}}$ is equal to
\begin{equation}
		\N \mu_1^2 +
		\N (G^{-1})_{[\N]; \vactab_1 \vactab_1}+
		(\N-1)(G^{-1})_{[\N]; \vactab_3 \vactab_3}
\end{equation}

The expectation value $\expval{X(23|1|4)}=\sum_{i,j} \expval{X_{ij} X_{jk}}$ is equal to
\begin{align}
		&\N \mu_1^2 +
		\N (G^{-1})_{[\N]; \vactab_1 \vactab_1}+
		\N\sqrt{\frac{\N-1}{2}}(G^{-1})_{[\N]; \vactab_4 \vactab_3} \\
		&+ \N \sqrt{\frac{(\N-1)(\N-2)}{2}}(G^{-1})_{[\N-1,1]; \vactab_5 \vactab_3}
\end{align}

The expectation value $\expval{X(1|2|3|4)}=\sum_{i,j,k,l} \expval{X_{ij} X_{kl}}$ is equal to
\begin{equation}
		\N^2 \mu_1^2 +	\N^2 (G^{-1})_{[\N]; \vactab_1 \vactab_1}
\end{equation}

The expectation value $\expval{X(1234)}=\sum_{i} \expval{X_{ii} X_{ii}}$ is equal to
\begin{align}
		&\frac{1}{\N}\mu_1^2 +
		\frac{\N-1}{\N} \mu_2^2+
		2\frac{\sqrt{\N-1}}{\N}\mu_1 \mu_2 + \frac{1}{\N}(G^{-1})_{[\N]; \vactab_1 \vactab_1}+
		\frac{\sqrt{2(\N-1)}}{\N}(G^{-1})_{[\N]; \vactab_2 \vactab_1}+\\
		&\frac{(\N-1)}{\N}(G^{-1})_{[\N]; \vactab_2 \vactab_2}
		+\frac{(\N-1)}{\N}(G^{-1})_{[\N-1,1]; \vactab_3 \vactab_3}\\
		&+\frac{\sqrt{2(\N-1)}}{\N}(\N-1)(G^{-1})_{[\N-1,1]; \vactab_4 \vactab_3}
		+\frac{(\N-1)^2}{\N}(G^{-1})_{[\N-1,1]; \vactab_4 \vactab_4}
\end{align}

The expectation value $\expval{X(12|34)}=\sum_{i,j} \expval{X_{ii} X_{jj}}$ is equal to
\begin{align}
		&\mu_1^2 +
		{\N-1} \mu_2^2+
		2\sqrt{\N-1}\mu_1 \mu_2 +
		(G^{-1})_{[\N]; \vactab_1 \vactab_1}+\\
		&\sqrt{2(\N-1)}(G^{-1})_{[\N]; \vactab_2 \vactab_1}+
		{(\N-1)}(G^{-1})_{[\N]; \vactab_2 \vactab_2}
\end{align}

The expectation value $\expval{X(12|3|4)}=\sum_{i,j,k} \expval{X_{ii} X_{jk}}$ is equal to
\begin{equation}
		\N\mu_1^2 +
		\N \sqrt{\N-1}\mu_1 \mu_2 + 
		\N (G^{-1})_{[\N]; \vactab_1 \vactab_1}+
		\sqrt{\frac{(\N-1)}{2}}(G^{-1})_{[\N]; \vactab_2 \vactab_1}
\end{equation}

As a proof of concept, we can relate the constants $(G^{-1})_{[\N], \alpha \beta}$ to the constants $(\Lambda^{-1}_{V_0})_{11}, (\Lambda^{-1}_{V_0})_{12}, (\Lambda^{-1}_{V_0})_{22}$ in \cite{Ramgoolam2019a}. Comparing the expectation values of $X(13|24), X(123|4)$ and $X(13|2|4)$ gives
\begin{equation}
	\mqty(1 & 0 & 1 \\ 1 & \sqrt{\frac{\N-1}{2}} & 0 \\ \N & 0 & 0) \mqty((G^{-1})_{[\N], \vactab_1 \vactab_1} \\ (G^{-1})_{[\N], \vactab_2 \vactab_1} \\ (G^{-1})_{[\N], \vactab_2 \vactab_2}) = \mqty(1 & 0 & 1 \\ 1 & \sqrt{\N-1} & 0 \\ \N & 0 & 0) \mqty((\Lambda^{-1}_{V_0})_{11} \\ (\Lambda^{-1}_{V_0})_{12} \\ (\Lambda^{-1}_{V_0})_{22}).
\end{equation}
Inverting the matrix on the l.h.s. gives
\begin{align}
	&(G^{-1})_{[\N], \vactab_1 \vactab_1} = (\Lambda^{-1}_{V_0})_{11} \\
	&(G^{-1})_{[\N], \vactab_2 \vactab_1} = \frac{1}{\sqrt{2}}(\Lambda^{-1}_{V_0})_{12} \\
	&(G^{-1})_{[\N], \vactab_2 \vactab_2} = (\Lambda^{-1}_{V_0})_{22}.
\end{align}
Identical procedures will relate the constants $(G^{-1})_{[\N-1,1], \alpha \beta}$ to $(\Lambda^{-1}_{V_H})$ in \cite{Ramgoolam2019a}. We do not perform this procedure here.

\section{Graph counting: Graph generating permutation diagrams}\label{subsec: graph counting}
By Proposition \ref{prop: graphs equals trace}, observables can be enumerated by directed graphs. We will now set up a scheme for counting and constructing graphs using group theory. By generalizing the double coset description of directed graphs introduced in \cite{deMelloKoch:2011uq, MelloKoch2012} we can enumerate invariants using appropriate equivalence classes of permutations, which define double cosets.

It is useful to describe the local structure (number of incoming and outgoing edges at each vertex) of a directed graph using vector partitions.
\begin{definition}[Vector partition]
	A vector partition of the vector $(K_1, K_2, \dots, K_c) \in \mathbb{Z}_+^{\times c}$ is a set of vectors $\{(k^{(i)}_1, \dots, k^{(i)}_c)\}_{i=1}^l$ satisfying
	\begin{equation}
		(K_1, \dots, K_c) = \sum_{i=1}^l (k^{(i)}_1, \dots, k^{(i)}_c).
	\end{equation}
	We call $l$ the number of parts.
\end{definition}
Let $G$ be a directed graph with $k$ edges and $l$ vertices. We record the number of outgoing and incoming edges at each vertex in terms of a set of $l$ pairs $(k^+_i, k^-_i)$, where $k^+_i(k^-_i)$ is the number of outgoing (incoming) edges at vertex $i$. We use the set of pairs to define a vector partition
\begin{equation}
	\label{eq: One Colour Graph Vector Partition}
	(k,k) = (k^+_1, k^-_1)+ \dots +(k^+_l, k^-_l), \quad 0 \leq k^\pm_i \leq k.
\end{equation}
We also define the vectors
\begin{align}
	\vec{k}^+ &= (k^+_1, \dots, k^+_l)\\
	\vec{k}^- &= (k^-_1, \dots, k^-_l),
\end{align}
and introduce the shorthand $(\vec{k}^+,\vec{k}^-)$ for the corresponding vector partition.
\begin{example}
	For example, the graph in Figure \ref{fig: One Colour Double Coset Graph} is associated with the vector partition $(\vec{k}^+,\vec{k}^-)=\{(3,0),(1,2),(1,1),(0,2)\}$.
\end{example}

The set of directed graphs with $k$ edges and local structure $(\vec{k}^+,\vec{k}^-)$ can be generated using permutations $\sn \in S_k$. This is illustrated in Figure \ref{fig: One Colour Double Coset Graph _a}. We call a diagram of this type a graph generating permutation diagram (GGPD). 
\begin{definition}[Graph Generating Permutation Diagram (GGPD)]	 \label{def: GGPD}
	A GGPD is specified by a vector partition $(\vec{k}^+,\vec{k}^-)$ with $l$ parts and the following choices
	\begin{enumerate}
		\item Label the set of vertices using $\{1,\dots,l\}$ and order them in ascending order from left to right.
		\item The $i$th vertex has $k_i^+$ outgoing edges emanating north of the vertex, and $k_i^-$ incoming edges emanating south of the vertex.
		\item Label the incoming edges using $\{1,\dots,k\}$. Pick an order for the incoming edge labels. For example, we will use $1<2<\dots<k$ from left to right as seen in Figure \ref{fig: One Colour Double Coset Graph}.
		\item Label and order the outgoing edges similarly.
	\end{enumerate}
	The edges in a GGPD are connected to vertices using the following rules.
	\begin{enumerate}
		\item Apply a permutation $ \sigma $ to the labels of the outgoing edges, which corresponds in Figure  
		\ref{fig: One Colour Double Coset Graph _a} to a re-ordering of the edges coming into the $ \sigma$-box from below before they emerge at the top. 
		\item Identify the end-points on the top line which have incoming lines to the points on the bottom line directly below them, with outgoing lines. 
	\end{enumerate}
\end{definition}
\begin{example}
	In Figure \ref{fig: One Colour Double Coset Graph _a} we take the incoming (and outgoing) edges as initially labeled $1,2,\dots, 5$ (from left to right). For $\sigma = (34)$ the third edge on the first vertex is swapped with the edge on the second vertex and we arrive at the graph in Figure \ref{fig: One Colour Double Coset Graph _b}.
	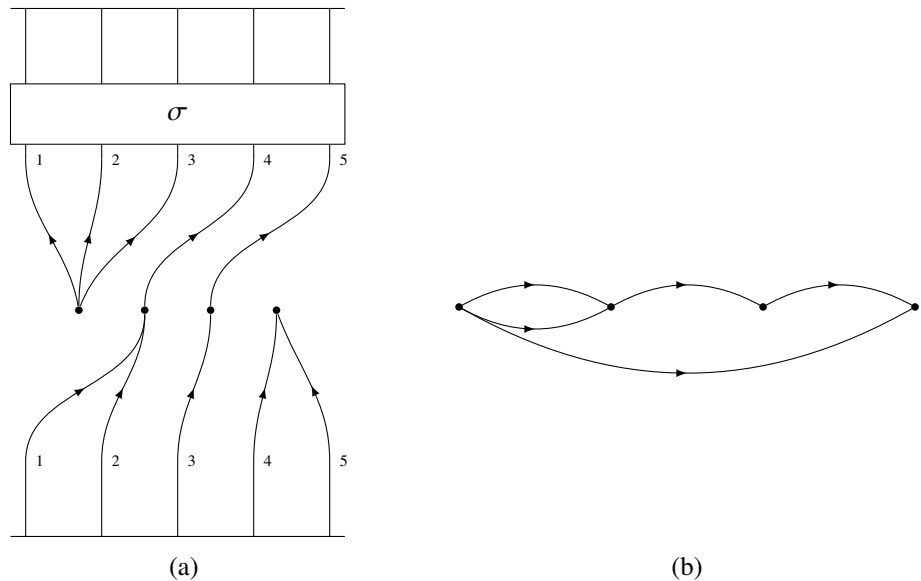
\begin{figure}[h!]
		\centering
		\subcaptionbox{\label{fig: One Colour Double Coset Graph _a}}[0.4\textwidth]
		{\begin{tikzpicture}[scale=2,baseline]
				\def \k {3}
				\def \m {4}
				\def \sep {0.5}
				\def \voffset {0.35}
				\pgfmathparse{(\sep*(\m)-2*\voffset)/\k};
				\pgfmathsetmacro{\vsep}{\pgfmathresult};
				\pgfmathint{\k-1};
				\pgfmathsetmacro{\kk}{\pgfmathresult};
				\foreach \v in {0,...,\k}
				{
					\pgfmathparse{\v*\vsep+\voffset}
					\node[circle, fill, inner sep=1pt](v\v) at (\pgfmathresult,0) {};
				}
				\foreach \eOutm in {0,...,\m}
				{
					\pgfmathint{\eOutm+1};
					\pgfmathsetmacro{\seOutm}{\pgfmathresult};
					\pgfmathparse{\eOutm*\sep};
					\coordinate (eom\seOutm) at (\pgfmathresult,2);
				}
				\foreach \eOutm in {0,...,\m}
				{
					\pgfmathint{\eOutm+1};
					\pgfmathsetmacro{\seOutm}{\pgfmathresult};
					\pgfmathparse{\eOutm*\sep};
					\coordinate (eomm\seOutm) at (\pgfmathresult,1);
					\node[right, node distance = 0pt and 0pt] at (eomm\seOutm) {\tiny \seOutm};
					\draw[] (eom\seOutm) -- (eomm\seOutm);
				}
				\foreach \eInm in {0,...,\m}
				{
					\pgfmathint{\eInm+1};
					\pgfmathsetmacro{\seInm}{\pgfmathresult};
					\pgfmathparse{\eInm*\sep};
					\coordinate (eim\seInm) at (\pgfmathresult,-1.5);
				}
				\foreach \eInm in {0,...,\m}
				{
					\pgfmathint{\eInm+1};
					\pgfmathsetmacro{\seInm}{\pgfmathresult};
					\pgfmathparse{\eInm*\sep};
					\coordinate (eimm\seInm) at (\pgfmathresult,-1);
					\node[right, node distance = 0pt and 0pt] at (eimm\seInm) {\tiny \seInm};
					\draw[] (eim\seInm) -- (eimm\seInm);
				}
				\begin{scope}[decoration={markings, mark=at position 0.5 with \arrow{latex}}]
					\draw[postaction={decorate}] (v0) to[out=100,in=-90] (eomm1);
					\draw[postaction={decorate}] (v0) to[out=90,in=-90] (eomm2);
					\draw[postaction={decorate}] (v0) to[out=70,in=-90] (eomm3);
					\draw[postaction={decorate}] (v1) to[out=90,in=-90] (eomm4);
					\draw[postaction={decorate}] (v2) to[out=90,in=-90] (eomm5);
					\draw[postaction={decorate}] (eimm1) to[out=90,in=-90] (v1);
					\draw[postaction={decorate}] (eimm2) to[out=90,in=-90] (v1);
					\draw[postaction={decorate}] (eimm3) to[out=90,in=-90] (v2);
					\draw[postaction={decorate}] (eimm4) to[out=90,in=-90] (v3);
					\draw[postaction={decorate}] (eimm5) to[out=90,in=-70] (v3);
				\end{scope}
				\draw[fill=white] ($(eomm1)+(-0.1,0.5)$) rectangle node{$\sigma$} ($(eomm5)+(0.1,0.1)$);
				\draw ($(eom1)-(0.1,0)$) -- ($(eom5)+(0.1,0)$);
				\draw ($(eim1)-(0.1,0)$) -- ($(eim5)+(0.1,0)$);
		\end{tikzpicture}}
		\subcaptionbox{\label{fig: One Colour Double Coset Graph _b}}[0.5\textwidth]
		{\begin{tikzpicture}[scale=2,baseline]
				\begin{scope}[decoration={markings, mark=at position 0.5 with \arrow{latex}}]
					\node[circle, fill, inner sep=1pt] at (0,0) {};
					\node[circle, fill, inner sep=1pt] at (1,0) {};	
					\node[circle, fill, inner sep=1pt] at (2,0) {};	
					\node[circle, fill, inner sep=1pt] at (3,0) {};
					\draw[postaction={decorate}] (0,0) to[bend left] (1,0);
					\draw[postaction={decorate}] (0,0) to[bend right] (1,0);
					\draw[postaction={decorate}] (0,0) to[bend right] (3,0);
					\draw[postaction={decorate}] (1,0) to[bend left] (2,0);
					\draw[postaction={decorate}] (2,0) to[bend left] (3,0);
				\end{scope}
				\node[circle, fill, inner sep=1pt,opacity=0] at (0,-1.5) {};			\end{tikzpicture}}
		\caption{Directed graphs of a fixed type (determined by a vector partition) correspond to a permutation $\sigma \in S_k$, where $k$ is the number of edges. (a) illustrates the correspondence with an example where the graph type is a vector partition $(5,5) = (3,0)+(1,2)+(1,1)+(0,2)$. (b) is the graph constructed from this vector partition with the permutation $\sigma = (34)$.}
		\label{fig: One Colour Double Coset Graph}
	\end{figure}
\end{example}

By scanning over all $\sigma \in S_k$ we can construct all graphs consistent with the local structure defined by the vector partition. However, some permutations lead to equivalent graphs. We use this to define an equivalence relation on permutations $ \sigma \in S_k $. The permutations within an equivalence class lead to different labellings of the same graph. We will now describe this equivalence relation in detail.

\subsection{Edge symmetry groups.}
There are two parts to the equivalence relation on permutations $\sigma \in S_k$. We will first describe the part due to edge permutation symmetry. The second part, described in the next subsection is due to vertex permutation symmetry.

Given a GGDP we can define two ordered lists of ordered lists
\begin{align}
	&[ K_1^+  , K_2^+ , \cdots , K_l^+ ], \\
	&[ K_1^-  , K_2^+ , \cdots , K_l^- ],
\end{align}
where $K_i^+(K_i^-)$ is an ordered list of the outgoing (incoming) edge labels at vertex $i$.
Explicitly for $K_i^+$ we have
\begin{align}
	K_1^+  & =    [ 1, 2, \cdots , k_1^+ ] \\
	K_2^+  & =    [ k_1^+ + 1, k_1^+ + 2, \cdots , k_1^+ + k_2^+ ] \\
	& \vdots  \nonumber  \\
	K_i^+  & =    [ k_1^+ + k_2^+ + \cdots + k_{i-1}^+ + 1 ,   k_1^+ + k_2^+ + \cdots + k_{i-1}^+ +2 , \cdots , k_1^+ + k_2^+ + \cdots + k_{ i-1}^+ +  k_{ i}^+  ] \\
	& \vdots   \nonumber\\
	K_l^+  & =   [ k_1^+ + k_2^+ + \cdots + k_{l-1}^+ + 1 ,   k_1^+ + k_2^+ + \cdots + k_{l-1}^+ +2 , \cdots , k_1^+ + k_2^+ + \cdots + k_{ l-1}^+ +  k_{ l}^+  ],
\end{align}
and similarly for $K_i^-$. Note that the concatenation  of these lists is the set of numbers $ [ 1 ,  \dots ,  k ] $.
\begin{align}\label{concatM}  
	[ K_1^+  , K_2^+ , \cdots , K_l^+ ] = [ 1, 2, \dots , k ] 
\end{align}

\begin{example}
For example, Figure \ref{fig: One Colour Double Coset Graph _a} defines
\begin{align}
	K_1^+ = [1,2,3], K_2^+ = [4], K_3^+ = [5] K_4^+ = [],\\
	K_1^- = [], K_2^- = [1,2], K_3^- = [3] K_4^- = [4,5].	
\end{align}
\end{example}

We view a permutation $ \sigma \in S_k$ as a re-arrangement of this list (similarly to one-line notation for permutations). 
The permutations within the sublists $ K_i^+$ or $K_i^-$ define subgroups of $S_k$ isomorphic to 
\begin{equation} 
	S_{\vec{k}^+} \cong S_{ k_1^+}  \times S_{ k_2^+ } \times \cdots \times S_{ k_l^+},
\end{equation}
and
\begin{equation} 
	S_{\vec{k}^-} =S_{ k_1^-}  \times S_{ k_2^- } \times \cdots \times S_{ k_l^-}.
\end{equation}

\begin{example}
For example, Figure \ref{fig: One Colour Double Coset Graph _a} gives
\begin{equation} 
	S_{\vec{k}^+} \cong S_3  \times S_1 \times S_1 \times S_0
\end{equation}
and
\begin{equation} 
	S_{\vec{k}^-} = S_0 \times S_2 \times S_1 \times S_2
\end{equation}
where $S_1$ is the trivial group and $S_0$ is the empty set.
\end{example}

It will be useful to view these groups as symmetric groups
\begin{equation}
	S_{k_i^\pm} = \Perms(\{1,\dots, k_i^\pm\}),
\end{equation}
and as subgroups of $S_k = \Perms(\{1,\dots,k\})$ simultaneously. For this, we define injective homomorphisms $\Perms(\{1,\dots, k_i^\pm\}) \rightarrow S_k$.
\begin{definition}[Edge symmetry group]\label{def: edge sym grp}
	
We describe the map for outgoing edges -- the construction is identical for incoming.
Let $\nu^+_i \in S_{ k_i^+}$ and define $\gamma_{  k_i^+ } (\nu^+_i ) \in S_k$ as the map
\begin{align}
	\gamma^+_i ( \nu_i^+ )  : 
	[ K_1^+ , K_2^+ , \cdots , K_i^+ , \cdots , K_l^+ ] \mapsto [ K_1^+ , K_2^+ , \cdots , \nu_i^+ ( K_i^+ ) , \cdots , K_l^+ ] 
\end{align}
where
\begin{align} 
	\nu_i^+ ( K_i^+ ) = [ \sum_{ j =1}^{ i-1}  k_j^+ + \nu_i^+ ( 1 ) , 
	\sum_{ j =1}^{ i-1}  k_j^+ + \nu_i^+ ( 2 ) , \cdots ,  \sum_{ j =1}^{ i-1}  k_j^+  + \nu_i^+ ( k_i^+ ) ]
\end{align}
For a general element $\nu^+ \in S_{\vec{k}^+}$ we apply $\gamma^+_i$ to each factor to get a homomorphism
\begin{equation}
	\gamma^+: S_{\vec{k}^+ } \rightarrow S_k.
\end{equation}
A similar construction exists for incoming lines collected into lists $K_i^-$ and a homomorphism
\begin{align} 
	\gamma^-  : S_{\vec{k}^-} \rightarrow S_k.
\end{align}
We call the image of $\gamma^+, \gamma^-$ the outgoing and incoming edge symmetry group of a GGPD, respectively.
\end{definition}
The naming of these subgroups is motivated by the fact that $\sigma \in S_k$ and
\begin{equation}
	\sigma' = \gamma^+ \sigma (\gamma^-)^{-1},
\end{equation}
generate the same directed graph for any $\gamma^+ \in \im \gamma^+, \gamma^- \in \im \gamma^-$.

\begin{example}
Continuing the example of Figure \ref{fig: One Colour Double Coset Graph _a} we have $S_k = S_5$ and we have the embeddings(homomorphisms)
\begin{equation} 
	\gamma^+: S_{\vec{k}^+} \rightarrow \Perms(\{1,2,3\})  \times \Perms(\{4\}) \times \Perms(\{5\}) \times \Perms(\emptyset)
\end{equation}
and
\begin{equation} 
	\gamma^-: S_{\vec{k}^-} \rightarrow \Perms(\emptyset)  \times \Perms(\{1,2\})\times \Perms(\{3\})\times \Perms(\{4,5\}),
\end{equation}
where the group of permutations of a single element set is just the identity element and the group of permutations of the empty set is the empty "group".
\end{example}

\begin{example}
Consider the graph in Figure \ref{fig: One Color One Sigma Equivalence Relation Without Vertex Symmetry} where the first vertex has three outgoing edges labeled $1,2,3$. Here two permutations $\sigma \in S_5$, which are related by a permutation $\nu^+_1 $ in $S_3$ permuting the list $ [ 1,2,3]$, lead to equivalent graphs.
From Figure \ref{fig: One Color One Sigma Equivalence Relation Without Vertex Symmetry} we see that this equivalence comes from left multiplication $\sigma \sim \gamma^+_1(\nu^+_1) \sigma $.
Similarly, for incoming edges we have equivalence under right multiplication $\sigma \sim \sigma \gamma^-_2((\nu^-_2)^{-1})\gamma^-_4((\nu^-_4)^{-1})$.

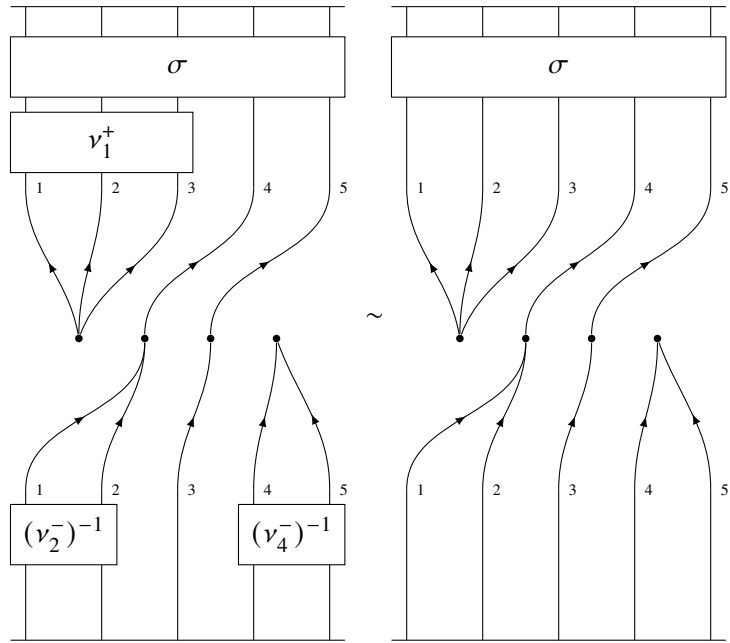
\begin{figure}[h!]
	\centering	
	\caption{For any permutation $\sigma$ in $S_5$, the two diagrams correspond to the same graph for any $\nu_1^+ \in S_3, \nu_2^- \in S_2, \nu_4^- \in S_2$.}
	\raisebox{-0.5\height}{\begin{tikzpicture}[scale=2]
			\def \k {3}
			\def \m {4}
			\def \sep {0.5}
			\def \voffset {0.35}
			\pgfmathparse{(\sep*(\m)-2*\voffset)/\k};
			\pgfmathsetmacro{\vsep}{\pgfmathresult};
			\pgfmathint{\k-1};
			\pgfmathsetmacro{\kk}{\pgfmathresult};
			\foreach \v in {0,...,\k}
			{
				\pgfmathparse{\v*\vsep+\voffset}
				\node[circle, fill, inner sep=1pt](v\v) at (\pgfmathresult,0) {};
			}
			\foreach \eOutm in {0,...,\m}
			{
				\pgfmathint{\eOutm+1};
				\pgfmathsetmacro{\seOutm}{\pgfmathresult};
				\pgfmathparse{\eOutm*\sep};
				\coordinate (eom\seOutm) at (\pgfmathresult,2.2);
			}
			\foreach \eOutm in {0,...,\m}
			{
				\pgfmathint{\eOutm+1};
				\pgfmathsetmacro{\seOutm}{\pgfmathresult};
				\pgfmathparse{\eOutm*\sep};
				\coordinate (eomm\seOutm) at (\pgfmathresult,1);
				\node[right, node distance = 0pt and 0pt] at (eomm\seOutm) {\tiny \seOutm};
				\draw[] (eom\seOutm) -- (eomm\seOutm);
			}
			\foreach \eInm in {0,...,\m}
			{
				\pgfmathint{\eInm+1};
				\pgfmathsetmacro{\seInm}{\pgfmathresult};
				\pgfmathparse{\eInm*\sep};
				\coordinate (eim\seInm) at (\pgfmathresult,-2);
			}
			\foreach \eInm in {0,...,\m}
			{
				\pgfmathint{\eInm+1};
				\pgfmathsetmacro{\seInm}{\pgfmathresult};
				\pgfmathparse{\eInm*\sep};
				\coordinate (eimm\seInm) at (\pgfmathresult,-1);
				\node[right, node distance = 0pt and 0pt] at (eimm\seInm) {\tiny \seInm};
				\draw[] (eim\seInm) -- (eimm\seInm);
			}
			\begin{scope}[decoration={markings, mark=at position 0.5 with \arrow{latex}}]
				\draw[postaction={decorate}] (v0) to[out=100,in=-90] (eomm1);
				\draw[postaction={decorate}] (v0) to[out=90,in=-90] (eomm2);
				\draw[postaction={decorate}] (v0) to[out=70,in=-90] (eomm3);
				\draw[postaction={decorate}] (v1) to[out=90,in=-90] (eomm4);
				\draw[postaction={decorate}] (v2) to[out=90,in=-90] (eomm5);
				\draw[postaction={decorate}] (eimm1) to[out=90,in=-90] (v1);
				\draw[postaction={decorate}] (eimm2) to[out=90,in=-90] (v1);
				\draw[postaction={decorate}] (eimm3) to[out=90,in=-90] (v2);
				\draw[postaction={decorate}] (eimm4) to[out=90,in=-90] (v3);
				\draw[postaction={decorate}] (eimm5) to[out=90,in=-70] (v3);
			\end{scope}
			\draw[fill=white] ($(eomm1)+(-0.1,1)$) rectangle node{$\sigma$} ($(eomm5)+(0.1,0.6)$);
			\draw[fill=white] ($(eomm1)+(-0.1,0.5)$) rectangle node{$\nu_1^+$} ($(eomm3)+(0.1,0.1)$);
			\draw[fill=white] ($(eimm1)-(0.1,0.5)$) rectangle node{$(\nu_2^-)^{-1}$} ($(eimm2)-(-0.1,0.1)$);
			\draw[fill=white] ($(eimm4)-(0.1,0.5)$) rectangle node{$(\nu_4^-)^{-1}$} ($(eimm5)-(-0.1,0.1)$);
			\draw ($(eom1)-(0.1,0)$) -- ($(eom5)+(0.1,0)$);
			\draw ($(eim1)-(0.1,0)$) -- ($(eim5)+(0.1,0)$);
	\end{tikzpicture}} $\sim$
	\raisebox{-0.5\height}{\begin{tikzpicture}[scale=2]
			\def \k {3}
			\def \m {4}
			\def \sep {0.5}
			\def \voffset {0.35}
			\pgfmathparse{(\sep*(\m)-2*\voffset)/\k};
			\pgfmathsetmacro{\vsep}{\pgfmathresult};
			\pgfmathint{\k-1};
			\pgfmathsetmacro{\kk}{\pgfmathresult};
			\foreach \v in {0,...,\k}
			{
				\pgfmathparse{\v*\vsep+\voffset}
				\node[circle, fill, inner sep=1pt](v\v) at (\pgfmathresult,0) {};
			}
			\foreach \eOutm in {0,...,\m}
			{
				\pgfmathint{\eOutm+1};
				\pgfmathsetmacro{\seOutm}{\pgfmathresult};
				\pgfmathparse{\eOutm*\sep};
				\coordinate (eom\seOutm) at (\pgfmathresult,2.2);
			}
			\foreach \eOutm in {0,...,\m}
			{
				\pgfmathint{\eOutm+1};
				\pgfmathsetmacro{\seOutm}{\pgfmathresult};
				\pgfmathparse{\eOutm*\sep};
				\coordinate (eomm\seOutm) at (\pgfmathresult,1);
				\node[right, node distance = 0pt and 0pt] at (eomm\seOutm) {\tiny \seOutm};
				\draw[] (eom\seOutm) -- (eomm\seOutm);
			}
			\foreach \eInm in {0,...,\m}
			{
				\pgfmathint{\eInm+1};
				\pgfmathsetmacro{\seInm}{\pgfmathresult};
				\pgfmathparse{\eInm*\sep};
				\coordinate (eim\seInm) at (\pgfmathresult,-2);
			}
			\foreach \eInm in {0,...,\m}
			{
				\pgfmathint{\eInm+1};
				\pgfmathsetmacro{\seInm}{\pgfmathresult};
				\pgfmathparse{\eInm*\sep};
				\coordinate (eimm\seInm) at (\pgfmathresult,-1);
				\node[right, node distance = 0pt and 0pt] at (eimm\seInm) {\tiny \seInm};
				\draw[] (eim\seInm) -- (eimm\seInm);
			}
			\begin{scope}[decoration={markings, mark=at position 0.5 with \arrow{latex}}]
				\draw[postaction={decorate}] (v0) to[out=100,in=-90] (eomm1);
				\draw[postaction={decorate}] (v0) to[out=90,in=-90] (eomm2);
				\draw[postaction={decorate}] (v0) to[out=70,in=-90] (eomm3);
				\draw[postaction={decorate}] (v1) to[out=90,in=-90] (eomm4);
				\draw[postaction={decorate}] (v2) to[out=90,in=-90] (eomm5);
				\draw[postaction={decorate}] (eimm1) to[out=90,in=-90] (v1);
				\draw[postaction={decorate}] (eimm2) to[out=90,in=-90] (v1);
				\draw[postaction={decorate}] (eimm3) to[out=90,in=-90] (v2);
				\draw[postaction={decorate}] (eimm4) to[out=90,in=-90] (v3);
				\draw[postaction={decorate}] (eimm5) to[out=90,in=-70] (v3);
			\end{scope}
			\draw[fill=white] ($(eomm1)+(-0.1,1)$) rectangle node{$\sigma$} ($(eomm5)+(0.1,0.6)$);
			%	\draw[fill=white] ($(eomm1)+(-0.1,0.5)$) rectangle node{$\nu_1^+$} ($(eomm3)+(0.1,0.1)$);
			%	\draw[fill=white] ($(eimm1)-(0.1,0.5)$) rectangle node{$(\nu_2^-)^{-1}$} ($(eimm2)-(-0.1,0.1)$);
			%	\draw[fill=white] ($(eimm4)-(0.1,0.5)$) rectangle node{$(\nu_4^-)^{-1}$} ($(eimm5)-(-0.1,0.1)$);
			\draw ($(eom1)-(0.1,0)$) -- ($(eom5)+(0.1,0)$);
			\draw ($(eim1)-(0.1,0)$) -- ($(eim5)+(0.1,0)$);
	\end{tikzpicture}}
	\label{fig: One Color One Sigma Equivalence Relation Without Vertex Symmetry}
\end{figure}
\end{example}
In general, we have combined left and right equivalence
\begin{equation}
\sigma \sim \sigma' \qq{iff} \quad \exists \nu^+ \in S_{\vec{k}^+}, \nu^- \in S_{\vec{k}^-}, \qq{st } \sigma = \gamma^+(\nu^+)  \sigma' \gamma^-((\nu^-)^{-1})
\end{equation}
The equivalence classes are in one-to-one correspondence with distinct graphs when the ordered pairs $(k^+_i, k^-_i)$ are all different.\footnote{Recall that $k^\pm_i$ can be zero and one, and we defined $S_0$ to be the empty set and $S_1$ to be the trivial group, containing just the identity element.} 

\subsection{Vertex symmetry group.}
When $(k^+_i, k^-_i)=(k^+_j, k^-_j)$ for $i\neq j$, the symmetry is enhanced and permutations which are related by permuting indistinguishable vertices give equivalent graphs. For example, the graphs in Figure \ref{fig: One Color Graph Permutation Equivalence With Vertex Symmetry} have $(k^+_1, k^-_1) = (k^+_2, k^-_2) = (3,2)$.
Indistinguishable vertices define a subgroup of $ S_{ l} = \Perms  ( [ 1, 2, \cdots , l ] ) $ which permutes labels of vertices having the same number of incoming and outgoing vertices. As mentioned, this is a symmetry of the GGPD. To describe the equivalence relation on $\sigma \in S_k$ due to this symmetry we embed this subgroup into $S_k$ as well.
\begin{definition}[Vertex symmetry group]\label{def: vertex sym grp}
The permutations $ \mu \in S_l$  of identical vertices are mapped to permutations in $ S_k $ as rearrangements of the concatenated lists. We define $\rho^+, \rho^-: S_l \rightarrow S_k$ as
\begin{align} 
	\rho^+(\mu) : [ K_1^+ , K_2^+ , \cdots , K_l^+ ] \rightarrow [ K_{ \mu ( 1) }^+ , K_{ \mu (2) }^+ , \cdots,	K_{ \mu (l) }^+ ] 
\end{align}
and
\begin{align} 
	\rho^-(\mu) : [ K_1^- , K_2^- , \cdots , K_l^- ] \rightarrow [ K_{ \mu ( 1) }^- , K_{ \mu (2) }^- , \cdots , K_{ \mu (l) }^- ].
\end{align}
The images of these maps are called outgoing and incoming vertex symmetry groups, respectively.
\end{definition}

\begin{example}
Consider the GGPD on the l.h.s. of Figure \ref{fig: One Color Graph Permutation Equivalence With Vertex Symmetry}.
A permutation $\mu \in S_2 \subset S_3$ which swaps the first two vertices
\begin{align} \nonumber
	\rho^+((1,2)) :&[ K_1^+ , K_2^+ , K_3^+ ] \rightarrow [ K_{2 }^+ , K_{1 }^+ , K_{3}^+ ] \\
	\rho^-((1,2)) :&[ K_1^- , K_2^- , K_3^- ] \rightarrow [ K_{2 }^- , K_{1 }^- , K_{3}^- ] 
\end{align}
gives back the same graph. In cycle notation we have
\begin{equation}
	\rho^+((1,2)) = (14)(25)(36), \quad \rho^-((1,2)) = (13)(24).
\end{equation}
\end{example}
More generally, if $(k^+_{i_{_1}}, k^-_{i_{_1}}) = \dots = (k^+_{i_{_r}}, k^-_{i_{_r}})$ for a set of vertex labels $\{i_1,\dots,i_r\} \subseteq \{1,\dots,l\}$, the subgroup $S_r \cong \Perms([i_1, \dots i_r]) \subseteq S_l$ will give equivalent graphs when acting on vertices. For sets of identical vertices of size $l_1, l_2, \dots$ the full group of permutations is isomorphic to
\begin{equation}
	S_{\vec{l}} = S_{l_{_1}} \times S_{l_{_2}} \times \dots.
\end{equation}

We can now state the full equivalence relation on GGPD's given by edge and vertex symmetry.
\begin{definition}[Equivalent GGPD]
	Two GGPDs are equivalent if their vector partition is the same and their respective permutations $\sigma, \sigma'$ are equivalent under the relation
	\begin{align}
		\sigma \sim \sigma' &\qq{iff} \; \exists \nu^+ \in S_{\vec{k}^+}, \nu^- \in S_{\vec{k}^-}, \mu \in S_{\vec{l}}, \\
		&\qq{st } \sigma  = \rho^+(\mu)\gamma^+(\nu^+) \sigma' (\rho^-(\mu)\gamma^-(\nu^-))^{-1}. \label{eqn: One Color Graph Permutation Equivalence}
	\end{align}
\end{definition}
Diagrammatically this equivalence corresponds to Figure \ref{fig: One Color Graph Permutation Equivalence With Vertex Symmetry}.

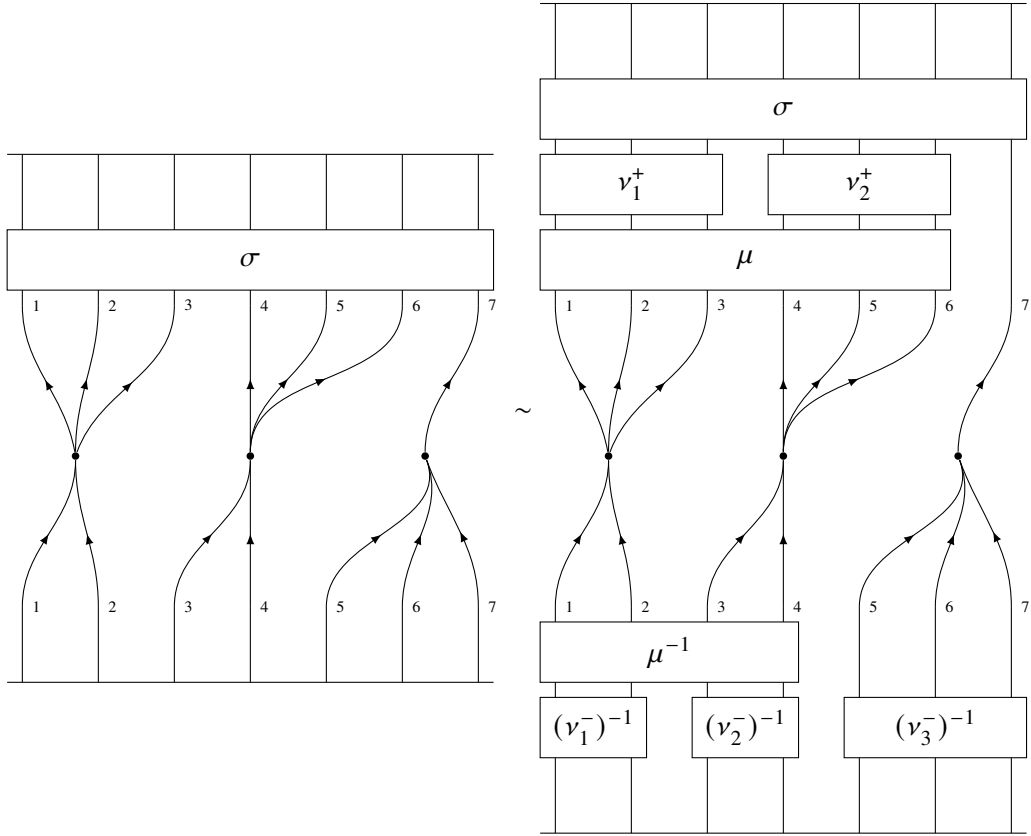
\begin{figure}[h!]
	\centering
	\raisebox{-0.5\height}{\begin{tikzpicture}[scale=2]
			\def \k {2}
			\def \m {6}
			\def \sep {0.5}
			\def \voffset {0.35}
			\pgfmathparse{(\sep*(\m)-2*\voffset)/\k};
			\pgfmathsetmacro{\vsep}{\pgfmathresult};
			\pgfmathint{\k-1};
			\pgfmathsetmacro{\kk}{\pgfmathresult};
			\foreach \v in {0,...,\k}
			{
				\pgfmathparse{\v*\vsep+\voffset}
				\node[circle, fill, inner sep=1pt](v\v) at (\pgfmathresult,0) {};
			}
			\foreach \eOutm in {0,...,\m}
			{
				\pgfmathint{\eOutm+1};
				\pgfmathsetmacro{\seOutm}{\pgfmathresult};
				\pgfmathparse{\eOutm*\sep};
				\coordinate (eom\seOutm) at (\pgfmathresult,2);
			}
			\foreach \eOutm in {0,...,\m}
			{
				\pgfmathint{\eOutm+1};
				\pgfmathsetmacro{\seOutm}{\pgfmathresult};
				\pgfmathparse{\eOutm*\sep};
				\coordinate (eomm\seOutm) at (\pgfmathresult,1);
				\node[right, node distance = 0pt and 0pt] at (eomm\seOutm) {\tiny \seOutm};
				\draw[] (eom\seOutm) -- (eomm\seOutm);
			}
			\foreach \eInm in {0,...,\m}
			{
				\pgfmathint{\eInm+1};
				\pgfmathsetmacro{\seInm}{\pgfmathresult};
				\pgfmathparse{\eInm*\sep};
				\coordinate (eim\seInm) at (\pgfmathresult,-1.5);
			}
			\foreach \eInm in {0,...,\m}
			{
				\pgfmathint{\eInm+1};
				\pgfmathsetmacro{\seInm}{\pgfmathresult};
				\pgfmathparse{\eInm*\sep};
				\coordinate (eimm\seInm) at (\pgfmathresult,-1);
				\node[right, node distance = 0pt and 0pt] at (eimm\seInm) {\tiny \seInm};
				\draw[] (eim\seInm) -- (eimm\seInm);
			}
			\begin{scope}[decoration={markings, mark=at position 0.5 with \arrow{latex}}]
				\draw[postaction={decorate}] (v0) to[out=100,in=-90] (eomm1);
				\draw[postaction={decorate}] (v0) to[out=90,in=-90] (eomm2);
				\draw[postaction={decorate}] (v0) to[out=70,in=-90] (eomm3);
				\draw[postaction={decorate}] (v1) to[out=90,in=-90] (eomm4);
				\draw[postaction={decorate}] (v1) to[out=90,in=-90] (eomm5);
				\draw[postaction={decorate}] (v1) to[out=90,in=-90] (eomm6);
				\draw[postaction={decorate}] (v2) to[out=90,in=-90] (eomm7);
				\draw[postaction={decorate}] (eimm1) to[out=90,in=-90] (v0);
				\draw[postaction={decorate}] (eimm2) to[out=90,in=-90] (v0);
				\draw[postaction={decorate}] (eimm3) to[out=90,in=-90] (v1);
				\draw[postaction={decorate}] (eimm4) to[out=90,in=-90] (v1);
				\draw[postaction={decorate}] (eimm5) to[out=90,in=-70] (v2);
				\draw[postaction={decorate}] (eimm6) to[out=90,in=-70] (v2);
				\draw[postaction={decorate}] (eimm7) to[out=90,in=-70] (v2);
			\end{scope}
			\draw[fill=white] ($(eomm1)+(-0.1,0.5)$) rectangle node{$\sigma$} ($(eomm7)+(0.1,0.1)$);
			\draw ($(eom1)-(0.1,0)$) -- ($(eom7)+(0.1,0)$);
			\draw ($(eim1)-(0.1,0)$) -- ($(eim7)+(0.1,0)$);
	\end{tikzpicture}} $\sim$
	\raisebox{-0.5\height}{\begin{tikzpicture}[scale=2]
			\def \k {2}
			\def \m {6}
			\def \sep {0.5}
			\def \voffset {0.35}
			\pgfmathparse{(\sep*(\m)-2*\voffset)/\k};
			\pgfmathsetmacro{\vsep}{\pgfmathresult};
			\pgfmathint{\k-1};
			\pgfmathsetmacro{\kk}{\pgfmathresult};
			\foreach \v in {0,...,\k}
			{
				\pgfmathparse{\v*\vsep+\voffset}
				\node[circle, fill, inner sep=1pt](v\v) at (\pgfmathresult,0) {};
			}
			\foreach \eOutm in {0,...,\m}
			{
				\pgfmathint{\eOutm+1};
				\pgfmathsetmacro{\seOutm}{\pgfmathresult};
				\pgfmathparse{\eOutm*\sep};
				\coordinate (eom\seOutm) at (\pgfmathresult,3);
			}
			\foreach \eOutm in {0,...,\m}
			{
				\pgfmathint{\eOutm+1};
				\pgfmathsetmacro{\seOutm}{\pgfmathresult};
				\pgfmathparse{\eOutm*\sep};
				\coordinate (eomm\seOutm) at (\pgfmathresult,1);
				\node[right, node distance = 0pt and 0pt] at (eomm\seOutm) {\tiny \seOutm};
				\draw[] (eom\seOutm) -- (eomm\seOutm);
			}
			\foreach \eInm in {0,...,\m}
			{
				\pgfmathint{\eInm+1};
				\pgfmathsetmacro{\seInm}{\pgfmathresult};
				\pgfmathparse{\eInm*\sep};
				\coordinate (eim\seInm) at (\pgfmathresult,-2.5);
			}
			\foreach \eInm in {0,...,\m}
			{
				\pgfmathint{\eInm+1};
				\pgfmathsetmacro{\seInm}{\pgfmathresult};
				\pgfmathparse{\eInm*\sep};
				\coordinate (eimm\seInm) at (\pgfmathresult,-1);
				\node[right, node distance = 0pt and 0pt] at (eimm\seInm) {\tiny \seInm};
				\draw[] (eim\seInm) -- (eimm\seInm);
			}
			\begin{scope}[decoration={markings, mark=at position 0.5 with \arrow{latex}}]
				\draw[postaction={decorate}] (v0) to[out=100,in=-90] (eomm1);
				\draw[postaction={decorate}] (v0) to[out=90,in=-90] (eomm2);
				\draw[postaction={decorate}] (v0) to[out=70,in=-90] (eomm3);
				\draw[postaction={decorate}] (v1) to[out=90,in=-90] (eomm4);
				\draw[postaction={decorate}] (v1) to[out=90,in=-90] (eomm5);
				\draw[postaction={decorate}] (v1) to[out=90,in=-90] (eomm6);
				\draw[postaction={decorate}] (v2) to[out=90,in=-90] (eomm7);
				\draw[postaction={decorate}] (eimm1) to[out=90,in=-90] (v0);
				\draw[postaction={decorate}] (eimm2) to[out=90,in=-90] (v0);
				\draw[postaction={decorate}] (eimm3) to[out=90,in=-90] (v1);
				\draw[postaction={decorate}] (eimm4) to[out=90,in=-90] (v1);
				\draw[postaction={decorate}] (eimm5) to[out=90,in=-70] (v2);
				\draw[postaction={decorate}] (eimm6) to[out=90,in=-70] (v2);
				\draw[postaction={decorate}] (eimm7) to[out=90,in=-70] (v2);
			\end{scope}
			\draw[fill=white] ($(eomm1)+(-0.1,1.5)$) rectangle node{$\sigma$} ($(eomm7)+(0.1,1.1)$);
			\draw[fill=white] ($(eomm1)+(-0.1,1)$) rectangle node{$\nu_1^+$} ($(eomm3)+(0.1,.6)$);
			\draw[fill=white] ($(eomm4)+(-0.1,1)$) rectangle node{$\nu_2^+$} ($(eomm6)+(0.1,.6)$);
			\draw[fill=white] ($(eimm1)-(0.1,1)$) rectangle node{$(\nu_1^-)^{-1}$} ($(eimm2)-(-0.1,.6)$);
			\draw[fill=white] ($(eimm3)-(0.1,1)$) rectangle node{$(\nu_2^-)^{-1}$} ($(eimm4)-(-0.1,.6)$);
			\draw[fill=white] ($(eimm5)-(0.1,1)$) rectangle node{$(\nu_3^-)^{-1}$} ($(eimm7)-(-0.1,.6)$);
			\draw[fill=white] ($(eomm1)+(-0.1,0.5)$) rectangle node{$\mu$} ($(eomm6)+(0.1,0.1)$);
			\draw[fill=white] ($(eimm1)-(0.1,0.5)$) rectangle node{$\mu^{-1}$} ($(eimm4)-(-0.1,0.1)$);
			\draw ($(eom1)-(0.1,0)$) -- ($(eom7)+(0.1,0)$);
			\draw ($(eim1)-(0.1,0)$) -- ($(eim7)+(0.1,0)$);
	\end{tikzpicture}}	
	\caption{This graph has two identical vertices of type $(3,2)$. Therefore any $\mu \in S_k$ which swaps all the edges of the two vertices gives rise to the same graph.}
	\label{fig: One Color Graph Permutation Equivalence With Vertex Symmetry}
\end{figure}

\subsection{Double coset description.}
The equivalence relation in \eqref{eqn: One Color Graph Permutation Equivalence} can be viewed as a partially solved ("gauge fixed" in physics jargon) version of a double coset.
\begin{definition}[Double coset]
A double coset
\begin{equation}
	H_1 \left\backslash G \right/ H_2,
\end{equation}
is the set of equivalence classes of elements $g,g' \in G$ under the identification
\begin{equation}
	g \sim g' \qq{iff} \quad \exists h_1 \in H_1, h_2 \in H_2, \quad g = h_1 g' h_2^{-1},
\end{equation}
where $H_1, H_2$ are subgroups of $G$. The equivalence classes are called double cosets.
\end{definition}

In our case, the double coset will have the form
\begin{equation}
	G(\vec{k}^+,\vec{k}^-) \left\backslash \qty( S_k^+ \times S_k^-) \right/ \diag(S_k), \label{eqn:1colordoublecoset}
\end{equation}
and we will now define the groups appearing in this quotient and prove
\begin{proposition}
The number of double cosets is equal to the number of inequivalent GGPDs
\begin{equation}
	N(\vec{k}^+,\vec{k}^-) = \abs{G(\vec{k}^+,\vec{k}^-) \left\backslash \qty( S_k^+ \times S_k^-) \right/ \diag(S_k)} = \begin{aligned}
		 \text{\# }&\text{Inequivalent GGPDs with}\\ &\text{vertex structure $(\vec{k}^+,\vec{k}^-)$.}
	\end{aligned}
\end{equation} \label{eq: double cosets}
\end{proposition}

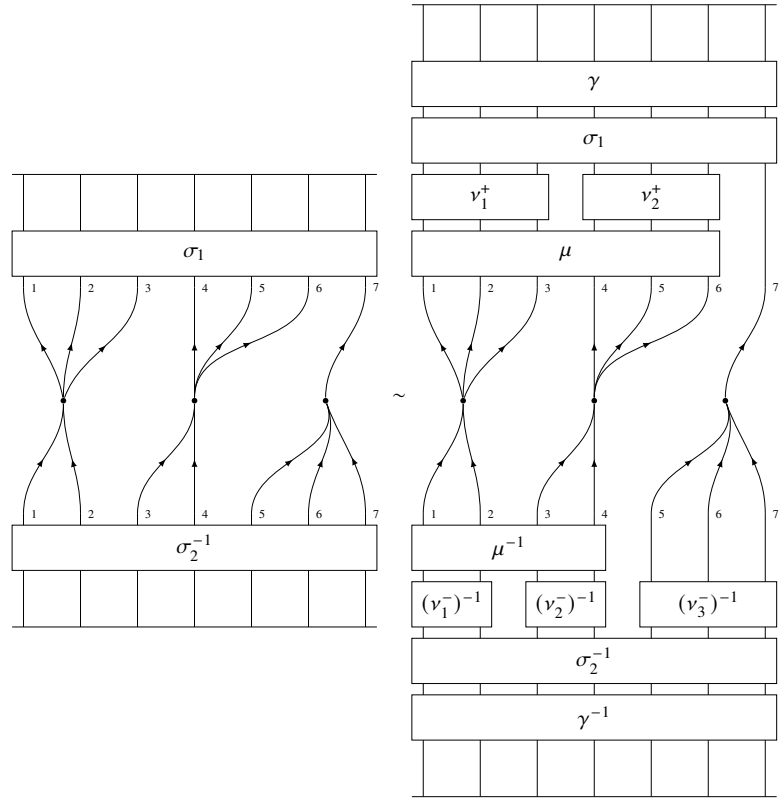
\begin{figure}
	\centering
	\scalebox{0.75}[0.75]{\raisebox{-0.5\height}{\begin{tikzpicture}[scale=2]
				\def \k {2}
				\def \m {6}
				\def \sep {0.5}
				\def \voffset {0.35}
				\pgfmathparse{(\sep*(\m)-2*\voffset)/\k};
				\pgfmathsetmacro{\vsep}{\pgfmathresult};
				\pgfmathint{\k-1};
				\pgfmathsetmacro{\kk}{\pgfmathresult};
				\foreach \v in {0,...,\k}
				{
					\pgfmathparse{\v*\vsep+\voffset}
					\node[circle, fill, inner sep=1pt](v\v) at (\pgfmathresult,0) {};
				}
				\foreach \eOutm in {0,...,\m}
				{
					\pgfmathint{\eOutm+1};
					\pgfmathsetmacro{\seOutm}{\pgfmathresult};
					\pgfmathparse{\eOutm*\sep};
					\coordinate (eom\seOutm) at (\pgfmathresult,2);
				}
				\foreach \eOutm in {0,...,\m}
				{
					\pgfmathint{\eOutm+1};
					\pgfmathsetmacro{\seOutm}{\pgfmathresult};
					\pgfmathparse{\eOutm*\sep};
					\coordinate (eomm\seOutm) at (\pgfmathresult,1);
					\node[right, node distance = 0pt and 0pt] at (eomm\seOutm) {\tiny \seOutm};
					\draw[] (eom\seOutm) -- (eomm\seOutm);
				}
				\foreach \eInm in {0,...,\m}
				{
					\pgfmathint{\eInm+1};
					\pgfmathsetmacro{\seInm}{\pgfmathresult};
					\pgfmathparse{\eInm*\sep};
					\coordinate (eim\seInm) at (\pgfmathresult,-2);
				}
				\foreach \eInm in {0,...,\m}
				{
					\pgfmathint{\eInm+1};
					\pgfmathsetmacro{\seInm}{\pgfmathresult};
					\pgfmathparse{\eInm*\sep};
					\coordinate (eimm\seInm) at (\pgfmathresult,-1);
					\node[right, node distance = 0pt and 0pt] at (eimm\seInm) {\tiny \seInm};
					\draw[] (eim\seInm) -- (eimm\seInm);
				}
				\begin{scope}[decoration={markings, mark=at position 0.5 with \arrow{latex}}]
					\draw[postaction={decorate}] (v0) to[out=100,in=-90] (eomm1);
					\draw[postaction={decorate}] (v0) to[out=90,in=-90] (eomm2);
					\draw[postaction={decorate}] (v0) to[out=70,in=-90] (eomm3);
					\draw[postaction={decorate}] (v1) to[out=90,in=-90] (eomm4);
					\draw[postaction={decorate}] (v1) to[out=90,in=-90] (eomm5);
					\draw[postaction={decorate}] (v1) to[out=90,in=-90] (eomm6);
					\draw[postaction={decorate}] (v2) to[out=90,in=-90] (eomm7);
					\draw[postaction={decorate}] (eimm1) to[out=90,in=-90] (v0);
					\draw[postaction={decorate}] (eimm2) to[out=90,in=-90] (v0);
					\draw[postaction={decorate}] (eimm3) to[out=90,in=-90] (v1);
					\draw[postaction={decorate}] (eimm4) to[out=90,in=-90] (v1);
					\draw[postaction={decorate}] (eimm5) to[out=90,in=-70] (v2);
					\draw[postaction={decorate}] (eimm6) to[out=90,in=-70] (v2);
					\draw[postaction={decorate}] (eimm7) to[out=90,in=-70] (v2);
				\end{scope}
				\draw[fill=white] ($(eomm1)+(-0.1,0.5)$) rectangle node{$\sigma_1$} ($(eomm7)+(0.1,0.1)$);
				\draw[fill=white] ($(eimm1)-(0.1,0.5)$) rectangle node{$\sigma_2^{-1}$} ($(eimm7)-(-0.1,0.1)$);
				\draw ($(eom1)-(0.1,0)$) -- ($(eom7)+(0.1,0)$);
				\draw ($(eim1)-(0.1,0)$) -- ($(eim7)+(0.1,0)$);
		\end{tikzpicture}}
		$\sim$
		\raisebox{-0.5\height}{\begin{tikzpicture}[scale=2]
				\def \k {2}
				\def \m {6}
				\def \sep {0.5}
				\def \voffset {0.35}
				\pgfmathparse{(\sep*(\m)-2*\voffset)/\k};
				\pgfmathsetmacro{\vsep}{\pgfmathresult};
				\pgfmathint{\k-1};
				\pgfmathsetmacro{\kk}{\pgfmathresult};
				\foreach \v in {0,...,\k}
				{
					\pgfmathparse{\v*\vsep+\voffset}
					\node[circle, fill, inner sep=1pt](v\v) at (\pgfmathresult,0) {};
				}
				\foreach \eOutm in {0,...,\m}
				{
					\pgfmathint{\eOutm+1};
					\pgfmathsetmacro{\seOutm}{\pgfmathresult};
					\pgfmathparse{\eOutm*\sep};
					\coordinate (eom\seOutm) at (\pgfmathresult,3.5);
				}
				\foreach \eOutm in {0,...,\m}
				{
					\pgfmathint{\eOutm+1};
					\pgfmathsetmacro{\seOutm}{\pgfmathresult};
					\pgfmathparse{\eOutm*\sep};
					\coordinate (eomm\seOutm) at (\pgfmathresult,1);
					\node[right, node distance = 0pt and 0pt] at (eomm\seOutm) {\tiny \seOutm};
					\draw[] (eom\seOutm) -- (eomm\seOutm);
				}
				\foreach \eInm in {0,...,\m}
				{
					\pgfmathint{\eInm+1};
					\pgfmathsetmacro{\seInm}{\pgfmathresult};
					\pgfmathparse{\eInm*\sep};
					\coordinate (eim\seInm) at (\pgfmathresult,-3.5);
				}
				\foreach \eInm in {0,...,\m}
				{
					\pgfmathint{\eInm+1};
					\pgfmathsetmacro{\seInm}{\pgfmathresult};
					\pgfmathparse{\eInm*\sep};
					\coordinate (eimm\seInm) at (\pgfmathresult,-1);
					\node[right, node distance = 0pt and 0pt] at (eimm\seInm) {\tiny \seInm};
					\draw[] (eim\seInm) -- (eimm\seInm);
				}
				\begin{scope}[decoration={markings, mark=at position 0.5 with \arrow{latex}}]
					\draw[postaction={decorate}] (v0) to[out=100,in=-90] (eomm1);
					\draw[postaction={decorate}] (v0) to[out=90,in=-90] (eomm2);
					\draw[postaction={decorate}] (v0) to[out=70,in=-90] (eomm3);
					\draw[postaction={decorate}] (v1) to[out=90,in=-90] (eomm4);
					\draw[postaction={decorate}] (v1) to[out=90,in=-90] (eomm5);
					\draw[postaction={decorate}] (v1) to[out=90,in=-90] (eomm6);
					\draw[postaction={decorate}] (v2) to[out=90,in=-90] (eomm7);
					\draw[postaction={decorate}] (eimm1) to[out=90,in=-90] (v0);
					\draw[postaction={decorate}] (eimm2) to[out=90,in=-90] (v0);
					\draw[postaction={decorate}] (eimm3) to[out=90,in=-90] (v1);
					\draw[postaction={decorate}] (eimm4) to[out=90,in=-90] (v1);
					\draw[postaction={decorate}] (eimm5) to[out=90,in=-70] (v2);
					\draw[postaction={decorate}] (eimm6) to[out=90,in=-70] (v2);
					\draw[postaction={decorate}] (eimm7) to[out=90,in=-70] (v2);
				\end{scope}
				\draw[fill=white] ($(eomm1)+(-0.1,2)$) rectangle node{$\gamma$} ($(eomm7)+(0.1,1.6)$);
				\draw[fill=white] ($(eimm1)-(0.1,2)$) rectangle node{$\gamma^{-1}$} ($(eimm7)-(-0.1,1.6)$);
				\draw[fill=white] ($(eomm1)+(-0.1,1.5)$) rectangle node{$\sigma_1$} ($(eomm7)+(0.1,1.1)$);
				\draw[fill=white] ($(eimm1)-(0.1,1.5)$) rectangle node{$\sigma_2^{-1}$} ($(eimm7)-(-0.1,1.1)$);		
				\draw[fill=white] ($(eomm1)+(-0.1,1)$) rectangle node{$\nu^+_1$} ($(eomm3)+(0.1,0.6)$);
				\draw[fill=white] ($(eomm4)+(-0.1,1)$) rectangle node{$\nu^+_2$} ($(eomm6)+(0.1,0.6)$);
				\draw[fill=white] ($(eomm1)+(-0.1,0.5)$) rectangle node{$\mu$} ($(eomm6)+(0.1,0.1)$);
				\draw[fill=white] ($(eimm1)-(0.1,1)$) rectangle node{$(\nu^-_1)^{-1}$} ($(eimm2)-(-0.1,0.6)$);
				\draw[fill=white] ($(eimm3)-(0.1,1)$) rectangle node{$(\nu^-_2)^{-1}$} ($(eimm4)-(-0.1,0.6)$);
				\draw[fill=white] ($(eimm5)-(0.1,1)$) rectangle node{$(\nu^-_3)^{-1}$} ($(eimm7)-(-0.1,0.6)$);
				\draw[fill=white] ($(eimm1)-(0.1,0.5)$) rectangle node{$\mu^{-1}$} ($(eimm4)-(-0.1,0.1)$);		
				\draw ($(eom1)-(0.1,0)$) -- ($(eom7)+(0.1,0)$);
				\draw ($(eim1)-(0.1,0)$) -- ($(eim7)+(0.1,0)$);
	\end{tikzpicture}}}
	\caption{Diagrammatic description of the double coset equivalence in equation \eqref{eq: Double Coset Equivalence}.}
	\label{fig: Two Sigma One Color Double Coset Graph}
\end{figure}
\begin{proof}
	
The diagrammatic equivalence to have in mind for the double coset is Figure \ref{fig: Two Sigma One Color Double Coset Graph}.
Because the incoming edges at the top line are identified with the outgoing edges of the bottom line, it is effectively only the product $\sigma_1 \sigma_2^{-1}$ which acts on the edges in this picture. We have increased the redundancy in the picture by going from a single permutation to a pair $(\sigma_1, \sigma_2) \in S_k^+ \times S_k^-$. If $(\sigma_1, \sigma_2)$ is replaced by $(\sigma_1 \gamma, \sigma_2 \gamma)$ for $\gamma \in S_k$, the combination $\sigma_1 \sigma_2^{-1} \mapsto \sigma_1 \gamma \gamma^{-1}\sigma_2^{-1} = \sigma_1\sigma_2^{-1}$ is unchanged. This is the origin of the quotient by $\diag(S_k)$, it describes the redundancy of using pairs of permutations.

The group $G(\vec{k}^+,\vec{k}^-)$ is the subgroup of $S_k^+ \times S_k^-$ with elements of the form
\begin{equation}
	(\rho^+(\mu) \gamma^+(\nu^+), \rho^-(\mu) \gamma^-(\nu^-)),
\end{equation}
for $\mu \in S_{\vec{l}}, \nu^+ \in S_{\vec{k}^+}, \nu^- \in S_{\vec{k}^-}$.
The double cosets are equivalence classes of the relation
\begin{align} \nonumber \label{eq: Double Coset Equivalence}
	(\sigma_1, \sigma_2) \sim &(\sigma_1', \sigma_2') \qq{iff} \exists \nu^+ \in S_{\vec{k}^+}, \nu^- \in S_{\vec{k}^-}, \mu \in S_{\vec{k}}, \gamma \in S_k, \\
	&\qq{st }(\sigma_1, \sigma_2) = (\rho^+(\mu)\gamma^+(\nu^+)\sigma_1'\gamma^{-1}, \rho^-(\mu)\gamma^-(\nu^-)\sigma_2'\gamma^{-1}). 
\end{align}
To see how equation \eqref{eq: Double Coset Equivalence} relates to \eqref{eqn: One Color Graph Permutation Equivalence}, we count the number of equivalence classes. We define the Kronecker delta on a group $G$ as the function that evaluates to $1$ on the identity element and vanishes otherwise,
\begin{equation}
	\delta(g) = \begin{cases}
		1, \qq{if} g=e \\
		0, \qq{otherwise}.
	\end{cases}
\end{equation}
By Burnside's lemma, the number of double cosets is
\begin{align} 
	N(\vec{k}^+,\vec{k}^-) &= \frac{1}{|G(\vec{k}^+,\vec{k}^-)||S_k|}\times \nonumber \\ 
	&\sum_{\substack{\mu \in S_{\vec{l}}, \nu^+ \in S_{\vec{k}^+} \\ \nu^- \in S_{\vec{k}^-},\gamma \in \diag(S_k)}} \sum_{\sigma_1, \sigma_2 \in S_k}
	\begin{aligned}[t]
		&\delta(\sigma_1^{-1}\rho^+(\mu)\gamma^+(\nu^+)\sigma_1\gamma^{-1}) \\
		&\delta(\sigma_2^{-1}\rho^-(\mu)\gamma^-(\nu^-)\sigma_2\gamma^{-1})
	\end{aligned} \label{eqn: Burnsides lemma double coset one color} \\ 
	&=\frac{1}{|G(\vec{k}^+,\vec{k}^-)||S_k|}\times \\ \nonumber
	&\sum_{\substack{\mu \in S_{\vec{l}}, \nu^+ \in S_{\vec{k}^+} \\ \nu^- \in S_{\vec{k}^-}}} \sum_{\sigma_1, \sigma_2 \in S_k}
	\begin{aligned}[t]
		&\delta(\sigma_1^{-1}\rho^+(\mu)\gamma^+(\nu^+)\sigma_1\sigma_2^{-1}\gamma^-((\nu^-)^{-1})\rho^-(\mu^{-1})\sigma_2)
	\end{aligned} \\ 
	&=\frac{1}{|G(\vec{k}^+,\vec{k}^-)|}	\sum_{\substack{\mu \in S_{\vec{l}} \\ \nu^+ \in S_{\vec{k}^+} \\ \nu^- \in S_{\vec{k}^-}}} \sum_{\sigma \in S_k}
	\begin{aligned}[t]
		&\delta(\sigma^{-1}\rho^+(\mu)\gamma^+(\nu^+)\sigma\gamma^-((\nu^-)^{-1})\rho^-(\mu^{-1})).
	\end{aligned}
\end{align}
In the second equality, we carried out the sum over $\gamma$ to impose the second delta function. In the third equality we renamed $\sigma_1\sigma_2^{-1} \equiv \sigma$, this makes the summand independent of $\sigma_2$. Consequently the sum over $\sigma_2$ just gives a factor of $|S_k|$. From Burnside's lemma, we recognize the last line as the counting of equivalence classes of \eqref{eqn: One Color Graph Permutation Equivalence}. This shows the correspondence between the double coset and the counting of graphs under edge and vertex symmetry.

\end{proof}
In appendix \ref{apx: double coset} we give a procedure for explicitly computing the number of double cosets in Proposition \ref{eq: double cosets} using generating functions known as cycle indices.

\section{Summary}
In this chapter we have described a class of matrix models, based on discrete permutation symmetry. We gave a description of the most general Gaussian permutation invariant distribution on real matrices. The matrix units constructed in chapter \ref{chapter: partition algebra} enabled us to explicitly solve for the mean/one-point function/expectation values and covariance/two-point function/propagator of matrix elements for general $\N$. This was first done in \cite{Kartsaklis2017, Ramgoolam2019a} without the use of the partition algebra technology presented here. The partition algebra construction is new and naturally gives rise to the interpretation of correlators in terms of linear combinations of 1-row diagrams. We presented this new perspective here because it generalizes beyond matrix models. In particular, this interpretation is used for permutation invariant tensor models in the upcoming work \cite{PIGTM}.

Observables in these models are defined to be permutation invariant polynomial functions in the matrix elements of general degree, following \cite{Kartsaklis2017, Ramgoolam2019a}. We gave two useful descriptions of the vector space of observables. The first one was a basis labelled by equivalence classes of 1-row partition diagrams. This basis was used in an algebraic combinatorial algorithm for computing expectation values of observables. The algorithm was based on the observation that expectation values can be computed in terms of a pairing on the vector space of 1-row partitions, where the pairing computes the number of components in the join of two 1-row partitions. The algorithm outputs an exact function of $\N$ and all the coupling constants/parameters in the model. This algorithm is new, and in fact different from the algorithm presented in \cite{Barnes2022b} based on so-called F-graphs.

The second basis of observables was labelled by directed graphs. This description was already known in \cite{Kartsaklis2017, Ramgoolam2019a}. We saw that this description was useful for constructing generating functions that count observables. In particular, we saw that directed graphs can be understood through permutations acting on GGPDs defined by a vector partition. These diagrams had two types of symmetry, edge symmetry and vertex symmetry, and the symmetries gave rise to equivalence relations on the set of permutations. We described the symmetry groups as permutation groups acting on the GGPDs. Finally, we saw that these equivalence classes could be understood as double cosets of permutation groups. We refer to Appendix \ref{apx: double coset} for explicit descriptions of the permutation subgroups entering the double coset, and constructions of generating functions counting directed graphs/double cosets. The use of GGPD's to enumerate graphs corresponding to gauge invariant quantities in gauge-string duality is not new (see for example \cite{deMelloKoch:2011uq, MelloKoch2012}). Our construction is a generalization of this to permutation invariant observables. This is new and was first presented in \cite{Barnes2022b}.
	
	\chapter{Matrix quantum mechanics}\label{ch: 1d}
In this chapter we apply the mathematical techniques developed in the previous sections to matrix quantum mechanics. Quantum mechanics corresponds to quantum field theory in one dimension and matrix models can be considered quantum field theories in zero dimensions. Therefore, matrix quantum mechanics is a natural avenue to extend the techniques to. This chapter is based on \cite{Barnes:2022qli}, where permutation invariant matrix quantum mechanics was first invented.

In section \ref{sec: MQM} we introduce systems of quantum matrix harmonic oscillators. We use the simplest model, that of $\N^2$ decoupled harmonic oscillators, to set up the basic language used in later sections. A useful description of the Hilbert space of the matrix harmonic oscillator is in terms of a Fock space spanned by states constructed from creation operators $a^\dagger_{ij}$ acting on the vacuum. As expected, the Hamiltonian measures the number of oscillators -- or degree as a polynomial in $a^\dagger_{ij}$ -- in the state. We then review the model of harmonic oscillators in a permutation invariant quadratic potential constructed and solved in \cite{Barnes:2022qli}.

In the subsequent sections we take inspiration from quantum mechanical systems with singlet constraints, such as gauged matrix models and spin matrix theory \cite{Harmark2014, Baiguera2022}, by considering the physics and algebraic structure of the permutation invariant subspace of the matrix harmonic oscillator. Section \ref{sec: Perm subspace} reviews the observations in \cite{Barnes:2022qli} concerning this subspace and its connection to partition algebras. We present three bases: the diagram basis, the orbit basis and the representation basis for the subspace of invariants. The three bases have several distinguishing properties. The diagram basis is the most geometrical of the three and forms an orthogonal basis for $\N \rightarrow \infty$. The orbit basis is exactly orthogonal for all $\N$ and is useful for describing finite $\N$ effects. As we will see, the representation basis is closely related to matrix units for $P_k(\N)$ and forms an eigenbasis of the algebraic Hamiltonians that we describe in section \ref{sec: algebraic hammy}.

Algebraic Hamiltonians, based on partition algebras, that act on the subspace of invariants through diagram multiplication were constructed in \cite{Barnes:2022qli}. For particular choices, these have exactly solvable spectra. Their description in terms of creation and annihilation operators used projectors onto fixed degree $k$ states. As we will see in section \ref{sec: algebraic hammy}, these algebraic Hamiltonians have a nice description -- without projection operators -- when the diagrams are restricted to permutation diagrams. This is a new observation in this particular context, but Hamiltonians based on symmetric groups have been considered in spin matrix theory. In this thesis, we present the Hamiltonians based on partition algebras as generalizations of those coming from permutation diagrams. Therefore, it is natural to forgo the inclusion of projectors. We illustrate the challenges involved in solving these Hamiltonians when projectors are excluded but do not commit to a particular solution.

In the last section we consider vacuum expectation values inspired by extremal correlators in $\mathcal{N}=4$ SYM. We show that they have a nice description in terms of an outer product in the diagram basis. The representation basis is used to prove representation theoretic selection rules for the extremal correlators. This was first proved in \cite{Barnes:2022qli}.

\section{Matrix harmonic oscillator} \label{sec: MQM}
The simplest matrix harmonic oscillator has a Lagrangian
\begin{equation} \label{eq: free lagrangian}
	L_0 = \frac{1}{2}\Bigg( \sum_{i,j=1}^\N \partial_t X_{ij} \partial_t X_{ij} - X_{ij} X_{ij} \Bigg).
\end{equation}
It describes a system of $\N^2$ decoupled oscillators. The conjugate momenta are 
\begin{equation}
	\Pi_{ij} = \pdv{L_0}{(\partial_t X_{ij})}= \pdv{t} X_{ij}. 
\end{equation}
The Hamiltonian corresponding to $L_0$ is
\begin{equation} \label{eq: free H}
	H_0 = \frac{1}{2} \Bigg( \sum_{i,j=1}^N \Pi_{ij}\Pi_{ij} + X_{ij}X_{ij} \Bigg).
\end{equation}
The canonical commutation relations are
\begin{equation}
	\comm{X_{ij}}{\Pi_{kl}} = i\delta_{ik}\delta_{jl}.
\end{equation}

The Hamiltonian given in \eqref{eq: free H} is diagonalized in the usual way - introducing oscillators $a^{\dagger}_{ij}, a_{ij}$ defined by
\begin{equation}
	\begin{aligned}
		X_{ij} &= \sqrt{\frac{1}{2}}\qty(a^\dagger_{ij} + a_{ij}), \\
		\Pi_{ij} &= i\sqrt{\frac{1}{2}}\qty(a^\dagger_{ij} - a_{ij}),
	\end{aligned} \label{eq: free oscillators}
\end{equation}
with commutation relations
\begin{equation}
	\comm{a_{ij}}{a^\dagger_{kl}} = \delta_{ik}\delta_{jl}. \label{eq: simplest oscillators}
\end{equation}
Normal ordering $H_0$ gives
\begin{equation}
	H_0 = \sum_{i,j=1}^N a_{ij}^{\dagger} a_{ij}, \label{eq: simplest hamiltonian}
\end{equation}
which is just a number operator. We now show that $H_0$ is invariant under a $U(N^2)$ symmetry that acts on oscillators as
\begin{align}
	a_{ij} &\rightarrow \sum_{k,l = 1}^N U_{ij; kl} a_{kl}, \\
	a^{\dagger}_{ij} &\rightarrow \sum_{k,l = 1}^N U^{\dagger}_{kl ; ij} a^{\dagger}_{kl},
\end{align}
with $U_{ij; kl}$ an $N^2 \times N^2$ unitary matrix satisfying
\begin{align}
	\sum_{k,l = 1}^N U_{ij; kl} U^{\dagger}_{kl; mn} = \delta_{im} \delta_{jn}.
\end{align}
Under the $U(N^2)$ transformation $H_0$ is invariant,
\begin{align} \nonumber
	H_0 \rightarrow &\sum_{i,j,k,l,m,n} U^{\dagger}_{kl ; ij} U_{ij ; mn} a^{\dagger}_{kl} a_{mn} \\ \nonumber
	= &\sum_{k,l,m,n} \delta_{km} \delta_{ln} a^{\dagger}_{kl} a_{mn} \\
	= &\sum_{k,l} a^{\dagger}_{kl} a_{kl}.
\end{align}

The oscillator states
\begin{equation}
	\prod_{i,j}\frac{(a^\dagger_{ij})^{k_{ij}}}{\sqrt{{k_{ij}!}}} \ket{0} \label{eq: H0 eigenbasis}
\end{equation}
labelled by non-negative integers $k_{ij}$ with $i,j=1,\dots,N$ are energy eigenstates of $H_0$.
The total Hilbert (Fock) space $\mathcal{H}$ decomposes into subspaces $\mathcal{H}^{(k)}$ with fixed number of oscillators (degree) $k$,
\begin{equation}
	\mathcal{H} \cong \bigoplus_{k=0}^{\infty} \mathcal{H}^{(k)}.
\end{equation}
The subset of states with $k=\sum_{i,j} k_{ij}$ form an eigenbasis for the subspace $\mathcal{H}^{(k)}$ and have energy $k$.
In general the spectrum is highly degenerate. The number of states with energy $k$ is
\begin{equation}
	\dim \mathcal{H}^{(k)} = \binom{N^2+k-1}{k} = \frac{N^2(N^2 +1 ) \dots (N^2+k-1)}{k!}.
\end{equation}
This is the number of ways to choose $k$ elements from a set of $N^2$ when repetition is allowed.
It is also the dimension of the symmetric part of a $k$-fold tensor product of a vector space with dimension $N^2$.
Equivalently, it is the dimension of the vector space of states composed of $k$ bosonic oscillators $a^\dagger_{ij}$. For fixed $k$ and $N \gg  2k$ the dimension grows as $N^{2k}$.

\subsection{Permutation invariant quadratic potentials.}
\label{subsec: PIMQM}
The construction in chapter \ref{chapter: 0d} is closely related to solving a model of matrix oscillators in a permutation invariant quadratic potential $V(X)$. We show how the techniques developed in previous chapters can be used to exactly derive the energy spectrum of this system.

A system of $\N^2$ particles in a potential is described by the Lagrangian
\begin{equation}
	L = \frac{1}{2}\sum_{i,j=1}^N \partial_t X_{ij} \partial_t X_{ij} - \frac{1}{2}V(X). \label{eq: lagrangian}
\end{equation}
We take the potential to be a general quadratic $\SN$ invariant potential
\begin{equation}
	V(X_{ij}) = V(X_{(i)\sn (j)\sn}), \quad \forall \sn \in \SN. \label{eq: potential invariance}
\end{equation}
The action of $\SN$ on $X_{ij}$ defined in \eqref{eq: potential invariance} corresponds to the diagonal action on the tensor product $\VN \otimes \VN$.

In chapter \ref{chapter: 0d} we parametrised the general quadratic $\SN$ invariant potential using representation variables (see \eqref{eq: rep basis})
\begin{equation}
	X_{\lambda, \alpha,a} =  C_{\lambda \alpha a}^{ij} X_{ij},
\end{equation}
where $C_{\lambda \alpha a}^{ij}$ are Clebsch-Gordan coefficients for the decomposition of $\VN \otimes \VN$.
The full Lagrangian in the representation basis is
\begin{equation}
	L = \sum_{\lambda, \alpha,\beta, a} \qty(\delta^{\alpha\beta}\partial_t X_{\lambda, \alpha,a} \partial_t X_{\lambda, \beta,a}-X_{\lambda, \alpha,a} G^{\lambda; \alpha \beta}X_{\lambda, \beta,a}).
\end{equation}
It describes a set of coupled harmonic oscillators.
Writing this Lagrangian in decoupled form only requires the diagonalization of a set of small parameter matrices $G^{\lambda; \alpha \beta}$ (a real symmetric $3 \times 3$ matrix and another real symmetric $2\times 2$ matrix), despite having a potentially large number of harmonic oscillators ($\N^2$).

Let
\begin{equation}
	\Omega^{\lambda; \alpha \beta}= (\omega^{\lambda; \alpha})^2 \delta^{\alpha \beta}
\end{equation}
be the diagonal matrix\footnote{We assume the eigenvalues are positive such that the spectrum of the Hamiltonian is bounded from below. Therefore, we may write the eigenvalues as squares without loss of generality.} such that
\begin{equation} \label{eq: full Hamiltonian metric}
	G^{\lambda; \alpha \beta} = \sum_{\gamma, \delta} (U^{\lambda})^\alpha_\gamma \Omega^{\lambda; \gamma \delta}(U^{\lambda})^\beta_\delta,
\end{equation}
where $U^{\lambda}$ are orthogonal change of basis matrices. In the decoupled basis
\begin{equation}
	S_{\lambda, \alpha,a} = \sum_\beta X_{\lambda, \beta,a}(U^{\lambda})^{\beta}_{\alpha}, \label{eq: decoupled basis}
\end{equation}
we have
\begin{equation}
	L =\frac{1}{2}\sum_{\lambda, \alpha, a}  \qty( \partial_t S_{\lambda, \alpha,a} \partial_t S_{\lambda, \alpha,a} -(\omega^{\lambda;\alpha})^2 S_{\lambda, \alpha,a} S_{\lambda, \alpha,a}).
\end{equation}
The canonical momenta are given by
\begin{equation}
	\Sigma_{\lambda, \alpha,a} = \partial_t S_{\lambda, \alpha,a}.
\end{equation}
The new canonical coordinates satisfy
\begin{equation}
	\comm{\Sigma_{\lambda, \alpha,a}}{S_{\lambda', \beta,b}} = i\delta_{{\lambda \lambda'}}\delta_{\alpha \beta}\delta_{ab},
\end{equation}
since $U^{\lambda}$ are orthogonal matrices.

The corresponding Hamiltonian,
\begin{equation} \label{eq: perm inv Hamiltonian}
	H = \frac{1}{2}\sum_{\lambda, \alpha, a} \qty( \Sigma_{\lambda, \alpha,a} \Sigma_{\lambda, \alpha,a} + (\omega^{\lambda; \alpha})^2 S_{\lambda, \alpha,a} S_{\lambda, \alpha,a}),
\end{equation}
is diagonalized by introducing oscillators
\begin{equation}
	\begin{aligned}
		&S_{\lambda, \alpha,a} = \sqrt{\frac{1}{2\omega^{\lambda; \alpha}}}(A^\dagger_{\lambda, \alpha,a} + A_{\lambda, \alpha,a}), \\
		&\Sigma_{\lambda, \alpha,a} = i\sqrt{\frac{\omega^{\lambda; \alpha}}{2}}(A^\dagger_{\lambda, \alpha,a} - A_{\lambda, \alpha,a}),
	\end{aligned}
\end{equation}
which satisfy
\begin{equation}
	\comm{A_{\lambda, \alpha,a}}{A^\dagger_{\lambda', \alpha',a'}} = \delta_{\lambda \lambda'} \delta_{\alpha \alpha'} \delta_{aa'}.
\end{equation}
In the oscillator basis, the normal ordered Hamiltonian has the form
\begin{equation}\label{eq: hamiltonian}
	H = \sum_{\lambda, \alpha, a}\omega^{\lambda; \alpha}A^\dagger_{\lambda, \alpha,a} A_{\lambda, \alpha,a}.
\end{equation} 

\section{Permutation invariant sectors for quantum matrix systems} \label{sec: Perm subspace}
We will now consider the algebraic structure behind the subspace of permutation invariant states. We will see that these subspaces are closely connected to partition algebras. To discuss the connection between invariant states and partition algebras, it will be useful to introduce the following matrices of oscillators
\begin{equation}
	(a^\dagger)^i_j = a^\dagger_{ji}, \quad a^i_j = a_{ij},
\end{equation}
which satisfy
\begin{equation}
	\comm{a^i_j}{(a^\dagger)^l_k} = \delta^i_k \delta^l_j.
\end{equation}
We think of these as the matrix elements of operator-valued maps in $\End(\VN)$
\begin{equation}
	a^\dagger(e_i) = \sum_{j=1}^\N (a^\dagger)^j_i e_j \, \text{and}\, a(e_i) = \sum_{j=1}^\N a^j_i e_j,
\end{equation}
and $\End(\VN^{\otimes k})$ more generally by
\begin{equation}
	(a^\dagger)^{\otimes k}(e_{i_1} \otimes \dots \otimes e_{i_k}) = a^\dagger(e_{i_1}) \otimes \dots \otimes a^\dagger(e_{i_k}),
\end{equation}
and
\begin{equation}
	a^{\otimes k}(e_{i_1} \otimes \dots \otimes e_{i_k}) = a(e_{i_1}) \otimes \dots \otimes a(e_{i_k}).
\end{equation}
With this notation at hand, we can define general states in $\mathcal{H}^{(k)}$ with a beautiful formula.
\begin{definition}\label{def: tensor state}
	Let $T \in \End(\VN^{\otimes k})$. We define a state corresponding to $T$ by
	\begin{equation}
		\ket{T} = \Tr_{V_N^{\otimes k}}(T (a^\dagger)^{\otimes k}) \ket{0}.
	\end{equation}
\end{definition}
We also define $\bra{T}$ as
\begin{equation}
	\bra{T}=\bra{0} \Tr_{V_N^{\otimes k}}(T (a^\dagger)^{\otimes k})^\dagger = \bra{0} \Tr_{V_N^{\otimes k}}(T^\dagger a^{\otimes k}),
\end{equation}
where $T^\dagger$ is the complex conjugate and transpose of $T$. We have used the fact that $[(a^\dagger)^i_j]^\dagger = [a_{ji}^\dagger]^\dagger = a_{ji} = a^j_i$ in the last equality.

The pairing of two states in this notation has a nice expression in terms of traces.
\begin{proposition}
	Let $\ket{T}, \ket{T'} \in \mathcal{H}^{(k)}$ be two vectors defined by tensors $T,T'$. The pairing
	\begin{equation}
		\bra{T}\ket{T'} = \sum_{\gamma \in S_k} \Tr_{\VN^{\otimes k}}(T^\dagger D_\gamma T'  D_{\gamma^{-1}_{}}),
	\end{equation}
	where for $\gamma \in S_k$, $D_{\gamma}$ is the corresponding linear operator that permutes tensors factors of $\VN^{\otimes k}$.
\end{proposition}
\begin{proof}
	This pairing is computed by summing over all the ways of contracting creation operators into annihilation operators. Using the fact that
	\begin{equation}
		\bra{0}a^{i_1}_{j_1} \dots a^{i_r}_{j_r} (a^\dagger)^{l_1}_{k_1} \dots (a^\dagger)^{l_r}_{k_r}\ket{0} = \sum_{\gamma \in S_r} [a^{i_{(1)\gamma}}_{j_{(1)\gamma}}, (a^\dagger)^{l_1}_{k_1}] \dots [a^{i_{(r)\gamma}}_{j_{(r)\gamma}}, (a^\dagger)^{l_r}_{k_r}],
	\end{equation}
	we have
	\begin{equation}
		\bra{T}\ket{T'} =  \sum_{\gamma \in S_k} (T^\dagger)_{j_1 \dots j_k}^{i_1 \dots i_k} (T')^{i_{(1)\gamma} \dots i_{(k)\gamma}}_{j_{(1)\gamma} \dots j_{(k)\gamma}} = \sum_{\gamma \in S_k} \Tr_{\VN^{\otimes k}}(T^\dagger D_\gamma T'  D_{\gamma^{-1}_{}}).
	\end{equation}
\end{proof}

In the next section we define invariant states and relate them to partition algebras. We will then describe three bases for the subspace $\Hilbertspace$ of invariant states, and some of their distinguished properties.

\subsection{Invariant states and partition algebras.}
The action of $\sn \in \SN$ on $X_{ij}$ translates to an action on the oscillators
\begin{equation}
	a_{ij} \mapsto a_{(i)\sigma^{-1} (j)\sigma^{-1}}, \quad a^\dagger_{ij} \mapsto a^\dagger_{(i)\sigma^{-1} (j)\sigma^{-1}}.
\end{equation}
We extend this to an action on the Hilbert space $\mathcal{H}$,
\begin{definition}[Adjoint action]
	Let $\sn \in \SN$ and define $\Adj{\sn}: \mathcal{H} \rightarrow \mathcal{H}$ by
	\begin{equation}
		\Adj{\sn}a^\dagger_{i_1 j_1} \dots a^\dagger_{i_k j_k} \ket{0} = a^\dagger_{(i_1)\sn (j_1)\sn} \dots a^\dagger_{(i_k)\sn (j_k)\sn} \ket{0}.
	\end{equation}
	We call this the adjoint action of $\sn$ on $\mathcal{H}$.
\end{definition}
The goal of this section is to describe and construct the subspace of $\Adj{\sn}$ invariant states of $\mathcal{H}$, denoted
\begin{equation}
	\Hilbertspace = \{\ket{T} \in \mathcal{H} \, \vert \, \Adj{\sn} \ket{T} = \ket{T} \, \forall\sn \in \SN \}.
\end{equation}

In the notation of Definition \ref{def: tensor state}, the adjoint action of $\sn \in \SN$ takes the form
\begin{equation}
	\Adj{\sn}\ket{T} = \Tr_{V_N^{\otimes k}}(P_\sn T P_{\sn^{-1}_{}}(a^\dagger)^{\otimes k}) \ket{0}.
\end{equation}
Consequently, the adjoint invariant vectors in $\mathcal{H}$ correspond to $\SN$ invariant tensors,
\begin{equation}
	\Adj{\sn}\ket{T} = \ket{T} \,  \Leftrightarrow \, T^{(j_1)\sn \dots (j_k)\sn}_{(i_1)\sn \dots (i_k)\sn}.  \label{eq: inv space is Pkn}
\end{equation}
For $\N \geq 2k$ we may use the isomorphism with partition algebras. That is,
\begin{equation}
	\Hilbertspace^{(k)} \subset \End_{\SN}(\VN^{\otimes k}) \cong P_k(\N),
\end{equation}
where $\Hilbertspace^{(k)}$ is the subspace of $\Hilbertspace$ of degree $k$ states.

The bosonic(commutative) nature of the creation operators leads to redundancy in the correspondence $T \leftrightarrow \ket{T}$. In particular, let $\tau \in S_k$, $D_\tau \in \End_{\SN}(\VN^{\otimes k})$ the corresponding operator and $T \in \End(\VN^{\otimes k})$ then
\begin{equation}
	\ket{T} = \ket{D_\tau T D_{\tau^{-1}_{}}}.
\end{equation}
This observation leads to the following corollaries
\begin{corollary}
	Define the vector space
	\begin{equation}
		\End_{S_k}(\VN^{\otimes k}) = \Span_\mathbb{C}\{ T \in \End(\VN^{\otimes k}) \, \vert \, D_\tau T D_{\tau^{-1}_{}} = T \, \forall \tau \in S_k\},
	\end{equation}
	for linear maps on $\VN^{\otimes k}$ that commute with $S_k$. As vector spaces we have
	\begin{equation}
		\mathcal{H}^{(k)} \cong \End_{S_k}(\VN^{\otimes k}).
	\end{equation}
\end{corollary}

The action of $S_k$ commutes with the action of $\SN$ on $\VN^{\otimes k}$. Therefore, we can consider the set of elements in $\End(\VN^{\otimes k})$ that commute with both. We have the following corollary relating this vector space to the Hilbert space of invariant states.
\begin{corollary}
Combing the above corollary with equation \eqref{eq: inv space is Pkn} we have
\begin{equation}
	\Hilbertspace^{(k)} \cong \End_{\SN \times S_k}(\VN^{\otimes k}) \subset P_k(\N).
\end{equation}
We define the symmetrized subalgebra of $P_k(\N)$
\begin{equation}
	SP_k(\N) = \frac{P_k(\N)}{S_k} \cong \End_{\SN \times S_k}(\VN^{\otimes k}), \label{eq: def SPk}
\end{equation}
where
\begin{equation}
	\frac{P_k(\N)}{S_k} = \{d \in P_k(\N) \, \vert \, \tau d \tau^{-1} = d \quad \forall \tau \in S_k \subset P_k(\N)\}.
\end{equation}
Therefore, assuming $\N \geq 2k$, we have
\begin{equation}
	\Hilbertspace^{(k)} \cong SP_k(\N).
\end{equation}
\end{corollary}

To summarise the above steps in words, we are investigating the adjoint action of permutations in $\SN$ on $\N \times \N$ quantum mechanical matrix variables $X_{ij}$. The corresponding
oscillators inherit the adjoint $\SN$ action. Oscillator states with $k$ oscillators correspond to
tensors $T$ with $k$ upper and lower indices, subject to an $S_k$ symmetry permuting the $k$
upper-lower index pairs along the tensor. This $S_k$ symmetry arises from the bosonic nature
of the oscillators. The $\SN$ invariant $k$-oscillator states correspond to tensors having $k$ upper
and $k$ lower indices, subject to an $\SN \times S_k$ invariance. This subspace of tensors can be
described as a symmetrized sub-algebra $SP_k(\N)$ of the partition algebra $P_k(\N)$.
This will be used in the following subsections to construct bases for $\Hilbertspace$.

\subsection{Diagram basis.} \label{sec: diagram basis}
In Definition \ref{def: partition algebra} we introduced the partition algebra using the diagram basis. As we now describe, the symmetrized partition algebra inherits a diagram basis from $P_k(\N)$. Consider a diagram basis element $d_\pi \in P_k(\N)$ and define the orbit
\begin{equation}
	[d_\pi] = \{\tau d_\pi \tau^{-1}, \, \forall \tau \in S_k\}.
\end{equation}
The set of all such orbits of $P_k(\N)$ define a diagram basis for $SP_k(\N)$.
\begin{proposition}
	Let $[d_\pi] $ be the $S_k$ orbit of $d_\pi \in P_k(\N)$. Define the corresponding element $\SPk{d}_\pi  \in SP_k(\N)$ by
	\begin{equation}
		\SPk{d}_\pi = \frac{1}{\abs{[d_\pi]}} \sum_{d_{\pi'} \in [d_\pi]} d_{\pi'} = \frac{1}{k!}\sum_{\tau \in S_k} \tau d_\pi \tau^{-1}. \label{def: eq SPk element}
	\end{equation}
	The set of elements $\SPk{d}_\pi$ corresponding to distinct $S_k$ orbits form a basis for $SP_k(\N)$.
\end{proposition}
\begin{proof}
	Let $\SPk{d} \in SP_k(\N)$ have an expansion
	\begin{equation}
		\SPk{d} = \sum_{d_\pi} a(d_\pi) d_\pi,
	\end{equation}
	with coefficients $a(d_\pi)$.
	Invariance under $S_k$ implies
	\begin{equation}
		\sum_{d_\pi} a(d_\pi) d_\pi = \sum_{d_\pi} a(d_\pi) \tau d_\pi \tau^{-1}, \quad \forall \tau \in S_k.
	\end{equation}
	Conjugation by $S_k$ on the diagram basis is a bijective map. Therefore, relabelling the sum gives
	\begin{equation}
		a(d_\pi) = a(\tau d_\pi \tau^{-1}), \quad \forall \tau \in S_k.
	\end{equation}	
	In other words, the coefficients are constant on the orbits $[d_\pi]$ and $\SPk{d}$ can be expanded in terms of their sum $\SPk{d}_\pi$.
\end{proof}

Given this basis for $SP_k(\N)$ we have a corresponding basis for $\Hilbertspace^{(k)}$.
\begin{definition}\label{def: diagram state def}
	Let $\SPk{d} \in SP_k(\N)$ and $\SPk{D} \in \End_{\SN \times S_k}(\VN^{\otimes k})$ the corresponding linear map. We define the state $\ket{\SPk{d}} \in \Hilbertspace^{(k)}$ by
	\begin{equation}
		\ket{\SPk{d}} = \Tr_{V_N^{\otimes k}}(\SPk{D} (a^\dagger)^{\otimes k})\ket{0}.
	\end{equation}
\end{definition}

\begin{example}
	At $k=1$ there are two invariant states
	\begin{equation}
		\ket{\PAdiagram{1}{}{}} = \sum_{i,j} (a^\dagger)^i_j \ket{0}, \ket{\PAdiagram{1}{-1/1}{}} = \sum_{i} (a^\dagger)^i \ket{0}, 
	\end{equation}
	At $k=2$ there are $11$ invariant states, some examples are
	\begin{align}
		&\ket{\PAdiagram{2}{-1/1,-2/2}{}}= \sum_{i,j} (a^\dagger)^i_i (a^\dagger)^j_j\ket{0}, \\
		&\ket{\PAdiagram{2}{-1/2,-2/1}{}}= \sum_{i,j} (a^\dagger)^j_i (a^\dagger)^i_j\ket{0}, \\
		&\ket{\PAdiagram{2}{-1/1}{}}=\ket{\PAdiagram{2}{-2/2}{}}= \sum_{i,j,k} (a^\dagger)^i_i (a^\dagger)^j_k\ket{0}, \\
		&\ket{\PAdiagram{2}{}{}}= \sum_{i,j,k,l} (a^\dagger)^j_i (a^\dagger)^l_k\ket{0}.
	\end{align}
\end{example}

We now state the main property of the diagram basis and will spend the rest of this subsection proving it.
\begin{proposition}(Large $\N$ factorisation) \label{prop: large N factorisation}
	Consider two vectors $\ket{\SPk{d}_\pi}, \ket{\SPk{d}_{\pi'}} \in \Hilbertspace^{(k)}$ corresponding to diagram basis elements of $SP_k(\N)$. Define the normalized states
	\begin{equation}
		\ket*{\widehat{\SPk{d}}_\pi} = \frac{\ket{\SPk{d}_\pi}}{\sqrt{\bra{{\SPk{d}_\pi}}\ket{\SPk{d}_\pi}}}, \quad \ket*{\widehat{\SPk{d}}_{\pi'}} = \frac{\ket{\SPk{d}_{\pi'}}}{\sqrt{\bra{{\SPk{d}_{\pi'}}}\ket{\SPk{d}_{\pi'}}}}.
	\end{equation}
	They are orthonormal at large $\N$,
	\begin{equation}
		\bra*{\widehat{\SPk{d}}_\pi}\ket*{\widehat{\SPk{d}}_{\pi'}} = \begin{cases}
			1 + O(1/\sqrt{\N})\qq{if $[d_\pi] = [d_\pi']$,} \\
			0 + O(1/\sqrt{\N})\qq{if $[d_\pi] \neq [d_\pi']$.}
		\end{cases}
	\end{equation}
	We call this large $\N$ factorisation.
\end{proposition}
To prove this, we will study the powers of $\N$ appearing in
\begin{equation}
	\Tr_{\VN^{\otimes k}}(\SPk{D}_\pi D_\gamma \SPk{D}_{\pi'}^T D_{\gamma^{-1}}) = \Tr_{\VN^{\otimes k}}({D}_\pi D_\gamma {D}_{\pi'}^T D_{\gamma^{-1}}).
\end{equation}

For this purpose, it is useful to consider a simpler case.
\begin{proposition}\label{prop: trace in diagram basis}
Let $\pi, \pi' \in \setpart{[k \vert k']}$ and $d_\pi \join d_{\pi'}$ be the join of the set partition diagrams and $\pi \join \pi'$ the corresponding set partition. The trace has a formula in terms of components of the join,
\begin{equation}
	\Tr_{\VN^{\otimes k}}({D}_\pi {D}_{\pi'}^T) = \N^{\abs{\pi \join {\pi'}}},
\end{equation}
where $\abs{\pi \join \pi'}=\abs{d_\pi \join d_{\pi'}}$ is the number of components in the join $d_\pi \join d_{\pi'}$, or equivalently the number of blocks in $\pi \join \pi'$.
\end{proposition}
\begin{proof}
	Recall that the join operation (see Definition \ref{def: join}) adds the edges of the diagrams $d_\pi, d_{\pi'}$ together. Since every edge corresponds to a Kronecker delta, we have
	\begin{equation}
		(D_{\pi \join \pi'})^{i_{1'} \dots i_{k'}}_{i_1 \dots i_k} = (D_{\pi})^{i_{1'} \dots i_{k'}}_{i_1 \dots i_k} (D_{\pi'})^{i_{1'} \dots i_{k'}}_{i_1 \dots i_k} \quad \text{(no sum)}.
	\end{equation}
	It follows that
	\begin{equation}
		\Tr_{\VN^{\otimes k}}({D}_\pi {D}_{\pi'}^T) = \sum_{\substack{i_1, \dots, i_k \\ i_{1'}, \dots i_{k'}}} (D_{\pi})^{i_{1'} \dots i_{k'}}_{i_1 \dots i_k} (D_{\pi})_{i_{1'} \dots i_{k'}}^{i_1 \dots i_k} = \sum_{\substack{i_1, \dots, i_k \\ i_{1'}, \dots i_{k'}}} (D_{\pi \join \pi'})^{i_{1'} \dots i_{k'}}_{i_1 \dots i_k}.
	\end{equation}
	It remains to show that this equals $\N^\abs{{\pi} \join {\pi'}}$.
	For this, let $\rho_1, \dots, \rho_b$ be the blocks of $\pi \join \pi'$, then
	\begin{equation}
		\sum_{\substack{i_1, \dots, i_k \\ i_{1'}, \dots i_{k'}}} (D_{\pi \join \pi'})^{i_{1'} \dots i_{k'}}_{i_1 \dots i_k} = (\sum_{\rho_1} 1) \dots (\sum_{\rho_b} 1) = \N^\abs{\pi \join {\pi'}},
	\end{equation}
	where the sums over parts correspond to sums where indices in each part are set equal.
\end{proof}

%\subsubsection{Factorisation for trace form on $P_k(\N)$.}\label{sec:facPa}
The proof of the following factorisation result contains most of the essential ingredients necessary for the main factorisation result and will serve as a useful warm-up exercise.
\begin{proposition}\label{prop: simple factorization case}
Let $\pi, \pi' \in \setpart{[k \vert k']}$, $d_\pi, d_{\pi'} \in P_k(\N)$ and $D_\pi, D_{\pi'}$ the corresponding linear operators on $\VN^{\otimes k}$, then
\begin{equation}
	\frac{\Tr_{\VN^{\otimes k}}(D_\pi D_{\pi'}^T)}{\sqrt{\Tr_{\VN^{\otimes k}}(D_\pi D_\pi^T) \Tr_{\VN^{\otimes k}}(D_{\pi'} D_{\pi'}^T)}} = \begin{cases}
		1 + O(1/\sqrt{\N}) \qq{if $\pi = \pi'$,}\\
		0 + O(1/\sqrt{\N}) \qq{if $\pi \neq \pi'$,}
	\end{cases} \label{eq: simple factorization case}
\end{equation} 
\end{proposition}
This equation \eqref{eq: simple factorization case} is related to the properties of the distance function defined in proposition 3.1 of \cite{Gabriel2015a}.

To prove this proposition it will be useful to introduce the following partial ordering on set partitions.
\begin{definition}\label{def: partial ordering set partitions}
	Let $\pi, \pi' \in \setpart{[k \vert k']}$. We introduce a partial ordering $\leq$ on $\setpart{[k \vert k']}$ by
	\begin{equation}
		\pi \leq \pi' \qquad \text{if every block of } \pi \text{ is contained within a block of } \pi'.
	\end{equation}
	We say that $\pi$ is a refinement of $\pi'$ or equivalently that $\pi'$ is a coarsening of $\pi$. Equivalently, $\pi \leq \pi'$ if $d_\pi$ only contains edges that are also contained in $d_{\pi'}$. We use this partial ordering as a partial ordering on diagrams $d_\pi$ and set partitions $\pi$  interchangeably.
\end{definition}
\begin{example}
	For example,
		\begin{equation}
			\PAdiagram{1}{}{} < \PAdiagram{1}{1/-1}{} \qq{and} \PAdiagram{2}{1/2}{} < \PAdiagram{2}{-1/-2,1/2}{}.
		\end{equation}
\end{example}

The factorization in Proposition \ref{prop: simple factorization case}  is a consequence of the following proposition.
\begin{proposition}	\label{prop: non-symmetric inequality}
Let $\pi, \pi' \in \setpart{[k \vert k']}$, then
\begin{equation}
	\begin{aligned}
		&2\abs{\pi \join \pi'} = \abs{\pi \join \pi} + \abs{\pi' \join \pi'} = \abs{\pi} + \abs{\pi'} \qq{if} \pi = \pi',\\
		&2\abs{\pi \join \pi'} < \abs{\pi \join \pi} + \abs{\pi' \join \pi'} = \abs{\pi} + \abs{\pi'} \qq{if} \pi \neq \pi',		
	\end{aligned} \label{eq: non-symmetric inequality}
\end{equation}
where we have used {${\abs{\pi \join \pi} + \abs{\pi' \join \pi'} = \abs{\pi} + \abs{\pi'}}$} since $d_\pi \join d_\pi = d_\pi$.
\end{proposition}
\begin{proof}
We will prove Proposition \ref{prop: non-symmetric inequality} by separating the general pairs $d_\pi , d_{\pi'}$  into three distinct cases:
\begin{enumerate}
	\item[1. $<$] If $d_{\pi}$ only contains edges that are also contained in $d_{\pi'}$, but $d_{\pi }\neq d_{\pi'}$ we have ${\pi} < {\pi'}$.
	In this case, $d_{\pi }\join d_{\pi'} = d_{\pi'}$ and it follows that,
	\begin{equation}\label{d1ind2conn} 
		\abs{\pi \join \pi'}= \abs{\pi'}.
	\end{equation}
	Note that $d_\pi < d_{\pi'}$ implies $\abs{\pi} > \abs{\pi'}$ since every addition of a new edge decreases the number of components by at least one. This because at least one new vertex is put into an already existing block. For example, in going from
	\begin{equation}
		\PAdiagram{2}{-1/-2}{} \rightarrow \PAdiagram{2}{-1/-2,-2/1}{},
	\end{equation}
	we decrease the number of components by exactly one.	
	Therefore,
	\begin{equation}
		2\abs{\pi \join \pi'} = \abs{\pi'} + \abs{\pi'} < \abs{\pi} + \abs{\pi'}.
	\end{equation}
	Since the LHS and RHS are symmetric under exchanging $\pi \leftrightarrow \pi'$, the inequality {${2\abs{\pi \join \pi'} <  \abs{\pi} + \abs{\pi'}}$} holds for $d_{\pi' }< d_{\pi}$ as well.
	\item[2. $\not\lesseqgtr$] Suppose $d_\pi \neq d_{\pi'}$ and  that there is no set of edges that can be added to $d_{\pi}$ to turn it into $d_{\pi'}$, nor is there a set of edges that can be added to $d_{\pi'}$ to obtain $d_{\pi}$. Then, we say that $d_{\pi}$ and $d_{\pi'}$ are incomparable. We denote this by $d_{\pi }\not\lesseqqgtr d_{ \pi'}$. The following diagrams are examples of incomparable diagrams
	\begin{equation}
		\PAdiagram{2}{-1/-2,1/2}{} \not\lesseqqgtr \PAdiagram{2}{1/-2,2/-1}{},  \qq{and} \PAdiagram{2}{-1/-2}{} \not\lesseqqgtr \PAdiagram{2}{1/2,2/-2}{}. \label{eq: incomparable diagrams}
	\end{equation}
	In this incomparable case, we have 
	\begin{equation}\label{incompconn}  
		\abs{\pi \join \pi'}  < \abs{\pi}, \qq{and} \abs{\pi \join \pi'} < \abs{ \pi' } 
	\end{equation}
	since the forming of  the join involves adding to  $d_{\pi} $, additional edges creating connections which did not exist in  $d_{\pi}$, or  alternatively adding to $d_{\pi'}$ additional edges that did not exist in $d_{\pi'}$.  
	Consequently we have the inequality 
	\begin{equation}\label{IncompIneq} 
		2\abs{\pi \join \pi'} < \abs{\pi} + \abs{\pi'}.
	\end{equation}
	\item[3. $=$] If $d_\pi = d_{\pi'}$ we have
	\begin{equation}
		\abs{\pi \join \pi'} = \abs{\pi \join \pi} = \abs{\pi} = \abs{\pi'},
	\end{equation}
	and therefore,
	\begin{equation}
		2\abs{\pi \join \pi'}  = \abs{\pi} + \abs{\pi'}.
	\end{equation}
\end{enumerate}
\end{proof}
To summarize, $2\abs{\pi \join \pi'} \leq \abs{\pi} + \abs{\pi'}$ with equality if and only if $\pi=\pi'$.

With this proposition we can prove the simplified factorization case.
\begin{proof}[Proof of Proposition \ref{prop: simple factorization case}]
	From Proposition \ref{prop: trace in diagram basis}, we have
	\begin{equation}
		\qty(\frac{\Tr_{\VN^{\otimes k}}(D_\pi D_{\pi'}^T)}{\sqrt{\Tr_{\VN^{\otimes k}}(D_\pi D_\pi^T) \Tr_{\VN^{\otimes k}}(D_{\pi'} D_{\pi'}^T)}})^2 = \frac{\N_{}^{2\abs{\pi \join \pi'}}}{\N_{}^{\abs{\pi}+\abs{\pi'}}}.
	\end{equation}
	It immediately follows from Proposition \ref{prop: non-symmetric inequality}
	\begin{equation}
		\frac{\N_{}^{2\abs{\pi \join \pi'}}}{\N_{}^{\abs{\pi}+\abs{\pi'}}} = \begin{cases}
			1 &\qq{if $\pi = \pi'$} \\
			0 + O(1/\N) &\qq{if $\pi \neq \pi'$}.
		\end{cases}
	\end{equation}
	Taking the square-root gives Proposition \ref{prop: simple factorization case}.
\end{proof}

The above discussion leads to the following corollary, which will be useful in proving the main proposition.
\begin{corollary}\label{Corr1}
	Consider a fixed diagram $d_{\pi'}$ and family of diagrams $d_{\pi_1}^{},d_{\pi_2}^{},\dots$ with fixed $\abs{\pi_1}=\abs{\pi_2}=\dots$ such that $\abs{\pi'} > \abs{\pi_i}$. It follows from \eqref{d1ind2conn} and \eqref{incompconn} that for every $\pi_i$,
	\begin{align}
		&\abs{\pi \join \pi_i} < \abs{\pi_i} \qq{if $d_\pi \not\lesseqqgtr d_{\pi_i}^{}$} \\
		&\abs{\pi \join \pi_i} = \abs{\pi_i} \qq{if $d_\pi < d_{\pi_i}^{}$}.
	\end{align}
	
\end{corollary}

The inner products in Proposition \ref{prop: large N factorisation} include sums over $S_k$,
\begin{equation}
	\bra*{\widehat{\SPk{d}}_\pi}\ket*{\widehat{\SPk{d}}_{\pi'}} = \frac{\sum_{\gamma_1 \in S_k} \N^{\abs{\pi \join \gamma_1 \pi' \gamma_1^{-1}} }}{\sqrt{\sum_{\gamma_2 \in S_k} \N^{\abs{\pi \join \gamma_2 \pi \gamma_2^{-1}}}\sum_{\gamma_3 \in S_k} \N^{\abs{\pi' \join \gamma_3 \pi' \gamma_3^{-1}} }}},
\end{equation}
where $\gamma \pi \gamma^{-1}$ is the set partition corresponding to the diagram obtained by computing $\gamma d_\pi \gamma^{-1}$.
Proposition \ref{prop: large N factorisation} follows from the following result.
\begin{proposition} \label{prop: large N factorisation inequality}
\begin{equation}
	\begin{aligned}
		2\max_{\gamma_1} \abs{\pi \join \gamma_1 \pi'  \gamma_1^{-1}} = \max_{\gamma_2} \abs{\pi \join \gamma_2 \pi  \gamma_2^{-1}} + \max_{\gamma_3}  \abs{\pi' \join \gamma_3 \pi'  \gamma_3^{-1}} \qq{if} [d_\pi] = [d_{\pi'}],\\
		2\max_{\gamma_1} \abs{\pi \join \gamma_1 \pi'  \gamma_1^{-1}} < \max_{\gamma_2} \abs{\pi \join \gamma_2 \pi  \gamma_2^{-1}} + \max_{\gamma_3}  \abs{\pi' \join \gamma_3 \pi'  \gamma_3^{-1}} \qq{if} [d_\pi] \neq [d_{\pi'}]
	\end{aligned} \label{eq: Large N factorisation inequality}
\end{equation}
\end{proposition}
\begin{proof}
The first step in proving this is to simplify the terms on the r.h.s. The inequalities in Proposition \ref{prop: non-symmetric inequality} imply that $\abs{\pi \join \gamma \pi \gamma^{-1}}$ is maximised when $\pi = \gamma \pi \gamma^{-1}$. Since the identity permutation $\gamma=1$ always satisfies this equality we have
\begin{equation}
	\max_{\gamma} \abs{\pi \join \gamma \pi \gamma^{-1}} = \abs{\pi}.
\end{equation}
We are left with proving
\begin{equation}
	\begin{aligned}
			2\max_{\gamma } \abs{\pi \join \gamma  \pi'  \gamma^{-1}} = \abs{\pi} + \abs{\pi'} \qq{if} [d_{\pi}] = [d_{\pi'}],\\
			2\max_{\gamma } \abs{\pi \join \gamma  \pi'  \gamma^{-1}} < \abs{\pi} + \abs{\pi'} \qq{if} [d_\pi] \neq [d_{\pi'}].
		\end{aligned} \label{eq: Large N factorisation inequality 2}
\end{equation}
We employ the same strategy as before, and consider three distinct cases.
\begin{enumerate}
	\item Suppose $\abs{\pi} > \abs{\pi'}$ and consider the set partition $ \gamma \pi' \gamma^{-1} $ for $ \gamma \in S_k$. We have $ \abs{\pi} > \abs{\gamma \pi' \gamma^{-1}}  = \abs{\pi'} $, because conjugation does not change the number of blocks(components). Therefore, either $d_{\pi} < \gamma d_{\pi'}\gamma^{-1}$ or $d_{\pi} \not\lesseqqgtr \gamma d_{\pi'}\gamma^{-1}$. Assume $\pi , \pi' $ are such that there exists some $ \gamma^* $ with	$d_\pi < \gamma^* d_{\pi' }(\gamma^*)^{-1}$. For any such $ \gamma^*$, Corollary \ref{Corr1}
	implies that 
	\begin{equation}
		2 \abs{\pi \join \gamma^* \pi' (\gamma^*)^{-1}} = 2\abs{\pi'} < \abs{\pi} + \abs{\pi'}.
	\end{equation}
	Any $\gamma $ not satisfying this condition leads to 
	$ d_\pi \not\lesseqqgtr \gamma d_{\pi' }\gamma^{-1} $, and the inequality in Corollary \ref{Corr1} implies that 
	\begin{equation}
		2 \abs{\pi \join \gamma \pi' \gamma^{-1}} <  2\abs{\pi'}.
	\end{equation}
	This implies that
	\begin{equation}
			2\max_{\gamma} \abs{\pi \join \gamma \pi' \gamma^{-1}} = 2 \abs{\pi \join \gamma^* \pi' (\gamma^*)^{-1}} = 2\abs{\pi'} < \abs{\pi} + \abs{\pi'}. \label{eq: tau max case 1}
		\end{equation}
	The pair
	\begin{equation}
			d_\pi = \PAdiagram{2}{-1/1}{}, \quad d_{\pi'} = \PAdiagram{2}{1/2,-2/2}{},
	\end{equation}
	is an example of this case since
	\begin{equation}
			\PAdiagram{2}{-1/1}{} < \PAdiagram{2}{1/2,-1/1}{} = \begin{aligned}
					&\PAdiagram{2}{-1/2,-2/1}{} \\
					&\PAdiagram{2}{1/2,-2/2}{}\\
					&\PAdiagram{2}{-1/2,-2/1}{}
				\end{aligned}.
		\end{equation}
	The argument is identical for the case where $\abs{\pi} < \abs{\pi'}$, and there exists some 
	$\gamma^*  \in S_k$ such that $d_{\pi'} < \gamma^* d_{\pi} (\gamma^*)^{-1}$. In this case, by renaming $ \pi \leftrightarrow \pi' $ in \eqref{eq: tau max case 1}, we have 
	\begin{equation}
		2\max_{\gamma} \abs{\pi' \join \gamma \pi \gamma^{-1}} = 2 \abs{\pi' \join \gamma^* \pi (\gamma^*)^{-1}} = 2\abs{\pi} < \abs{\pi} + \abs{\pi'}. 
	\end{equation}
	Using the symmetry of the inner product it follows 
	\begin{equation}
		2\max_{\gamma} \abs{\pi \join \gamma \pi' \gamma^{-1}} < \abs{\pi} + \abs{\pi'}. 
	\end{equation}	  
	\item Secondly, consider the case of incomparability,
	\begin{equation}
			d_\pi \not\lesseqqgtr \gamma d_{\pi'} \gamma^{-1} \quad \forall \gamma \in S_k.
	\end{equation}
	Recall that for incomparable diagrams we have \eqref{IncompIneq} which says that
	\begin{equation}
			2\abs{\pi \join \gamma \pi' \gamma^{-1}} < \abs{\pi} + \abs{\gamma \pi' \gamma^{-1}} = \abs{\pi}+ \abs{\pi'},
	\end{equation}
	where the last equality follows because conjugation by a permutation does not change the number of connected components.
	Therefore
	\begin{equation}
			2\max_{\gamma} \abs{\pi \join \gamma \pi' \gamma^{-1}} < \abs{\pi} + \abs{\pi'},		
	\end{equation}
	in this case as well.
	\item When $\pi = \gamma \pi' \gamma^{-1}$ for some $\gamma \in S_k$, the bound is saturated and \begin{equation}
			2\max_{\gamma} \abs{\pi \join \gamma \pi' \gamma^{-1}} = 2\abs{\pi}.
	\end{equation}
	The condition $\pi = \gamma \pi' \gamma^{-1}$ implies $[d_\pi] = [d_{\pi'}]$.
\end{enumerate}
\end{proof}
Having proven the inequalities in equation \eqref{eq: Large N factorisation inequality}, arguments similar to those used to prove Proposition \ref{prop: simple factorization case} lead to a proof of Proposition \ref{prop: large N factorisation}. To summarise, the diagram basis for $\Hilbertspace^{(k)}$ forms an orthonormal basis at large $\N$.

\subsection{Orbit basis.}  \label{sec: orbit basis}
In Theorem \ref{thm: orbit basis}, we described the orbit basis for $\End_{\SN}(\VN^{\otimes k})$. In this subsection we will describe the corresponding orbit basis for $P_k(\N)$ and show that it corresponds to an orthogonal basis for $\Hilbertspace^{(k)}$.

The orbit basis can be expressed in terms of the diagram basis using the partial ordering in Definition \ref{def: partial ordering set partitions}.
\begin{theorem}
	The orbit basis elements $x_{\pi'}$ for $\pi' \in \setpart{[k \vert k']}$ defined by
	\begin{equation}
		d_\pi = \sum_{\pi \leq \pi'} x_{\pi'}, \label{eq: diagram to orbit basis}
	\end{equation}
	form a basis for $P_k(\N)$. The change of basis matrix determined by \eqref{eq: diagram to orbit basis}, denoted $\zeta_{2k}$, is called the zeta matrix of the partially ordered set $\setpart{[k \vert k']}$. It is upper triangular, with ones on the diagonal and hence invertible. The inverse of $\zeta_{2k}$ is the matrix $\mu_{2k}$ 
	\begin{equation} \label{eq: orbit basis expansion}
		x_{\pi} = \sum_{\pi \leq \pi'} \mu_{2k}(\pi, \pi') d_{\pi'},
	\end{equation}
	where if $\pi \leq \pi'$ and $\pi'$ consists of $b$ blocks such that the $i$th block of $\pi'$ is the union of $b_i$ blocks of $\pi$ then
	\begin{equation}
		\mu_{2k}(\pi, \pi') = \prod_{i=1}^b (-1)^{b_i -1} (b_i - 1)!
	\end{equation}
\end{theorem}
\begin{proof}
	See \cite[Section 4.2, 4.3]{Benkart2017}.
\end{proof}
\begin{example}
The diagram basis element $d_{\pi}$ is a sum of all orbit basis elements labelled by set partitions equal to or coarser than $\pi$, for example
\begin{equation}
	\PAdiagram{2}{1/-1}{} = \PAdiagramOrbit{2}{1/-1}{} +  \PAdiagramOrbit{2}{1/-1, 2/-2}{} + \PAdiagramOrbit{2}{1/-1}{2/1}  + \PAdiagramOrbit{2}{1/-1}{-1/-2} + \PAdiagramOrbit{2}{1/-1, 2/-2}{2/1}.
\end{equation}
\end{example}
We will continue to distinguish the diagram and orbit bases by drawing diagram basis elements with black vertices and labelling them with the letter $d$, and drawing orbit basis elements with white vertices and labelling them with the letter $x$.
\begin{example}
We have the following expansion of the orbit basis element labelled by {$\pi = 1|2|1'|2'$}
\begin{align} \nonumber \label{eq: o1 in diagram basis}
	\PAdiagramOrbit{2}{}{} = &\PAdiagram{2}{}{} - \PAdiagram{2}{1/-1}{} - \PAdiagram{2}{2/-2}{} - \PAdiagram{2}{2/1}{} - \PAdiagram{2}{}{-1/-2} - \PAdiagram{2}{1/-2}{} - \PAdiagram{2}{-1/2}{} + \PAdiagram{2}{1/-1, 2/-2}{} + \PAdiagram{2}{1/-2, -1/2}{} +\\
	&\PAdiagram{2}{}{2/1, -1/-2} + 2 \PAdiagram{2}{1/-1}{2/1} + 2 \PAdiagram{2}{1/-1}{-1/-2} + 2 \PAdiagram{2}{2/-2}{2/1} + 2 \PAdiagram{2}{2/-2}{-1/-2} - 6 \PAdiagram{2}{1/-1, 2/-2}{2/1}.
\end{align}
\end{example}

As with the diagram basis, we can construct an orbit basis for $SP_k(\N)$ by considering $S_k$ orbits.
\begin{definition}
	Let $x_\pi$ be an orbit basis element of $P_k(\N)$ and consider the orbit
	\begin{equation}
		[x_\pi] = \{\tau x_\pi \tau^{-1}, \, \forall \tau \in S_k\}.
	\end{equation}
	Define the corresponding element of $SP_k(\N)$ by
	\begin{equation}
		\SPk{x}_\pi = \frac{1}{\abs{[x_\pi]}} \sum_{x_\pi' \in [x_\pi]} x_{\pi'} = \frac{1}{k!} \sum_{\tau \in S_k} \tau x_{\pi} \tau^{-1}.
	\end{equation}
\end{definition}
The set of distinct orbits define a basis for $SP_k(\N)$. The proof is analogous to the diagram basis case.
Given this orbit basis for $SP_k(\N)$, we define the corresponding orbit basis for $\Hilbertspace^{(k)}$
\begin{definition}
	Let $\SPk{x} \in SP_k(\N)$ and $\SPk{X} \in \End_{\SN \times S_k}(\VN^{\otimes k})$ the corresponding linear map. We define the state $\ket{\SPk{x}} \in \Hilbertspace^{(k)}$ by
	\begin{equation}
		\ket{\SPk{x}} = \Tr_{\VN^{\otimes k}}(\SPk{X}(a^\dagger)^{\otimes k})\ket{0}.
	\end{equation}
\end{definition}

We state the main property of the orbit basis and will spend the rest of this subsection proving it.
\begin{proposition}\label{prop: orbit orthogonality}
Let $\SPk{x}_{\pi}, \SPk{x}_{\pi'} \in SP_k(\N)$ be orbit basis elements, then
\begin{equation}
	\bra{\SPk{x}_\pi}\ket{\SPk{x}_{\pi'}} =  \begin{cases}
			\abs{G_{\pi}}\N_{(\abs{\pi})} \qq{if $[x_{\pi'}] = [x_\pi]$,}\\
			0 \qq{otherwise.}
		\end{cases}
\end{equation}
where $\N_{(l)} = \N(\N-1) \dots (\N-l+1)$ is the falling factorial and $ \abs{G_{\pi}}$ is the order of the subgroup of $S_k$ that leaves $x_\pi$ invariant.
\end{proposition}

Before proving this, consider the simpler proposition
\begin{proposition}\label{prop: simple orbit trace}
Let $X_\pi, X_{\pi'}$ be the orbit basis elements of $\End_{\SN}(\VN^{\otimes k})$ (see Theorem \ref{thm: orbit basis}) then
\begin{equation} 
		\Tr_{\VN^{\otimes k}}(X_{\pi} X_{\pi'}^T) = \N_{(\abs{\pi})}\delta_{\pi \pi'}.
\end{equation}
\end{proposition}
\begin{proof}
The trace is equal to
\begin{equation}
	\Tr_{\VN^{\otimes k}}(X_{\pi} X_{\pi'}^T) = \sum_{\substack{i_1 \dots i_k \\ i_{1'} \dots i_{k'}}} (X_{\pi})^{i_{1'} \dots i_{k'}}_{i_{1} \dots i_{k}} (X_{\pi'})^{i_{1'} \dots i_{k'}}_{i_{1} \dots i_{k}}.
\end{equation}
As we now explain, Equation \eqref{eq: orbit basis action} implies
\begin{equation}
	(X_{\pi})^{i_{1'} \dots i_{k'}}_{i_{1} \dots i_{k}} (X_{\pi'})^{i_{1'} \dots i_{k'}}_{i_{1} \dots i_{k}} = \begin{cases}
		1 \qq{\parbox[t]{8cm}{if $i_a = i_b$ if and only if $a$ and $b$ are in the same block of $\pi$ and the same block of $\pi'$,}} \\
		0 \qq{otherwise.}
	\end{cases}\label{eq: orbit basis matrix elements product}
\end{equation}
If $\pi \neq \pi'$ two situations exist. Consider the set of all pairs $(a,b)$ for $a,b=1,\dots,k,1',\dots,k'$ such that $a$ and $b$ are in the same block of $\pi$. Since $\pi \neq \pi'$ at least one of these pairs are such that $a$ and $b$ are in different blocks of $\pi'$. The second case is the reverse. Consider the set of all $(a,b)$ such that $a$ and $b$ are in the same block of $\pi'$. Then $\pi' \neq \pi$ implies that there exists at least one pair such that $a$ and $b$ are not in the same block of $\pi$.
In that case, there are no choices of $i_a, i_b$ which satisfy the first criteria in \eqref{eq: orbit basis matrix elements product}.
For example, take $a,b$ to be in the same block of $\pi$ but different blocks of $\pi'$. The matrix elements $(X_{\pi})^{i_{1'} \dots i_{k'}}_{i_{1} \dots i_{k}}$ vanish if $i_a \neq i_b$ while the matrix elements $(X_{\pi'})^{i_{1'} \dots i_{k'}}_{i_{1} \dots i_{k}}$ vanish unless $i_a \neq i_b$. Therefore, the product identically vanishes,
\begin{equation}
	(X_{\pi})^{i_{1'} \dots i_{k'}}_{i_{1} \dots i_{k}} (X_{\pi'})^{i_{1'} \dots i_{k'}}_{i_{1} \dots i_{k}} = \delta_{\pi \pi'} (X_{\pi})^{i_{1'} \dots i_{k'}}_{i_{1} \dots i_{k}}
\end{equation}
and
\begin{equation}
	\Tr_{\VN^{\otimes k}}(X_{\pi} X_{\pi'}^T)  = \sum_{\substack{i_1 \dots i_k \\ i_{1'} \dots i_{k'}}} (X_{\pi})^{i_{1'} \dots i_{k'}}_{i_{1} \dots i_{k}}  \delta_{\pi \pi'} = \delta_{\pi \pi'} \N(\N-1) \dots (\N-\abs{\pi}+1).
\end{equation}
The last equality is a consequence of \eqref{eq: orbit basis action}. For example, consider the set partition $12|1'2'$. The trace of $X_{12|1'2'}$ is
\begin{equation}
	\Tr_{\VN^{\otimes 2}}(X_{12|1'2'}) = \sum_{i_1 = i_2 \neq i_3, i_3 = i_4} = \N(\N-1),
\end{equation}
since we have $\N$ choices of indices for $i_1$ and $(\N-1)$ choices for $i_3$ (for every choice of $i_1$). The general case is analogous,
\begin{equation}
	\Tr_{\VN^{\otimes k}}(X_{\pi}) = \N_{(\abs{\pi})}.
\end{equation}
We have $\N$ choices for the indices of the first block of $\pi$, $\N-1$ choices for the indices of the second block and so on.
\end{proof}

We now have the tools to prove the main proposition in this subsection.
\begin{proof}[Proof of Proposition \ref{prop: orbit orthogonality}]
Note that the inner product of two orbit basis elements of is given by
\begin{equation}
	\bra{\SPk{x}_\pi}\ket{\SPk{x}_{\pi'}} = \sum_{\gamma \in S_k} \Tr_{\VN^{\otimes k}}(D_\gamma \SPk{X}_{\pi} D_{\gamma^{-1}} \SPk{X}_{\pi'}^T).
\end{equation}
Note that
\begin{equation}
	\sum_{\gamma \in S_k} D_\gamma \SPk{X}_{\pi} D_{\gamma^{-1}} = \frac{1}{k!}\sum_{\gamma,\mu \in S_k} D_\gamma D_{\mu}X_{\pi}D_{\mu^{-1}} D_{\gamma^{-1}} = \sum_{\gamma \in S_k} D_{\gamma}X_{\pi}D_{\gamma^{-1}},
\end{equation}
and therefore
\begin{align}
	\sum_{\gamma \in S_k} \Tr_{\VN^{\otimes k}}(D_\gamma \SPk{X}_{\pi} D_{\gamma^{-1}} \SPk{X}_{\pi'}^T) &= \frac{1}{(k!)^2}\sum_{\gamma, \mu,\nu \in S_k} \Tr_{\VN^{\otimes k}}(D_\nu D_\gamma D_\mu X_{\pi} D_{\mu^{-1}}D_{\gamma^{-1}}D_{\nu^{-1}} X_{\pi'}^T) \\
	&= \sum_{\gamma \in S_k} \Tr_{\VN^{\otimes k}}(D_\gamma X_{\pi} D_{\gamma^{-1}}X_{\pi'}^T).
\end{align}
We re-write
\begin{equation}
	\sum_{\gamma \in S_k} D_\gamma X_{\pi} D_{\gamma^{-1}}  = \abs{G_{\pi}}\sum_{x_\lambda \in [x_\pi]} X_{\lambda},
\end{equation}
where the sum on the r.h.s. is over the distinct elements in the $S_k$ orbit of $x_{\pi}$. Substituting this into the trace gives
\begin{align}
	\bra{\SPk{x}_\pi}\ket{\SPk{x}_{\pi'}} = \abs{G_{\pi}}\sum_{x_\lambda \in [x_\pi]} \Tr_{\VN^{\otimes k}}(X_{\lambda}X_{\pi'}^T) &= \abs{G_{\pi}}\sum_{x_\lambda \in [x_\pi]} \N_{(\abs{\pi})} \delta_{\lambda \pi'} \\
	&= \begin{cases}
		\abs{G_{\pi}}\N_{(\abs{\pi})} \qq{if $[x_{\pi'}] = [x_\pi]$,} \\
		0 \qq{otherwise,}
	\end{cases}
\end{align}
where we used Proposition \ref{prop: simple orbit trace} in the second equality.
\end{proof}
To summarise we have found that the orbit basis is orthogonal for all $\N$. Note that the norm of $\ket{\SPk{x}_\pi}$ vanishes if $\N \leq \abs{\pi}$ since then $\N_{(\abs{\pi})} = 0$. We leave it for future work to investigate consequences of this.
\subsection{Representation basis.} \label{sec: ON basis}
In section \ref{sec: Semi-Simple Algebra Technology} we saw that $P_k(\N)$ has a basis of matrix units for $\N \geq 2k$. In this section we will see that $SP_k(\N)$ also has a basis of matrix units. The corresponding basis for $\Hilbertspace^{(k)}$ is called the representation basis. It is orthogonal and diagonalizes a set of commuting charges.

The matrix units for $SP_k(N)$ are constructed from the matrix units $Q_{\alpha \beta}^{\lambda}$ in \eqref{def: matrix unit PkN} using Branching coefficients. Branching coefficients are understood as follows.
The partition algebra $P_k(N)$ has a subalgebra (isomorphic to) $\mathbb{C}[S_k]$. For any given irreducible representation $Z_\lambda$ of $P_k(\N)$ there exists a basis where the action of $\mathbb{C}[S_k] \subset P_k(N)$ is manifest and irreducible. That is, as a representation of $S_k$
\begin{equation}
	Z_{{\lambda}} \cong \bigoplus_{\gamma \vdash k} V_{{\gamma}} \otimes M_{\lambda \rightarrow \gamma}, \label{eq: branching Pk to Sk}
\end{equation}
where $M_{\lambda \rightarrow \gamma}$ is a multiplicity space.
We take the l.h.s. to have an orthonormal basis
\begin{equation}
	E^{\lambda}_{\alpha}, \quad \alpha \in \{1, \dots \dimPk{\lambda}\}.
\end{equation}
%where the representation of $d \in P_k(N)$ is irreducible,
%\begin{equation}
%	d(E^{\lambda}_{\alpha}) = \sum_\beta D_{\beta \alpha}^{\lambda}(d)E^{\lambda}_{\beta}. \label{eq: PkN irrep2}
%\end{equation}
The r.h.s. has a basis
\begin{equation}
	\begin{aligned}
			E^{\lambda, \gamma}_{p \mu}, \quad &p \in \{1, \dots, \dimSN{\gamma} \}, \\
			&\mu \in \{1,\dots, \dim(M_{\lambda \rightarrow \gamma})\}, 
		\end{aligned}
\end{equation}
where $\mu$ is a multiplicity label for the irreducible representation $V_{\gamma}$ of $S_k$ in the decomposition.
We demand that the representation of $\tau \in \mathbb{C}[S_k]$ is irreducible in this basis and that the basis is orthonormal with respect to the inner product defined by $E^{\lambda}_{\alpha}$.
%\begin{equation}
%	\tau(E^{\lambda_1, \mu}_{\lambda_2, p}) = \sum_q D_{qp}^{\lambda_2}(\tau)E^{\lambda_1, \mu}_{\lambda_2, q},\label{eq: Sk branching irrep2}
%\end{equation}
%where $D_{qp}^{\lambda_2}(\tau)$ is an irreducible representation of $\tau \in \mathbb{C}[S_k]$.
\begin{definition}
The change of basis coefficients in \eqref{eq: branching Pk to Sk} are called branching coefficients. They are defined by
\begin{equation}
	E^{\lambda, \gamma}_{p \mu} = \sum_{\alpha}(B^{\lambda \rightarrow \gamma})^\alpha_{p \mu} E^{\lambda}_{\alpha}.
\end{equation}
\end{definition}

Given this definition we can define matrix units for $SP_k(\N)$
\begin{proposition}
The elements
\begin{equation}
	Q^{\lambda, \gamma}_{\mu \nu} = \sum_{\alpha, \beta, p} Q^{\lambda}_{\alpha \beta}(B^{\lambda \rightarrow \gamma})^\alpha_{p \mu}(B^{\lambda \rightarrow \gamma})^\beta_{p \nu}. \label{eq: spkn units}
\end{equation}
form matrix units for $SP_k(\N)$. The sum over $p$ in \eqref{eq: spkn units} implements the projection to $S_k$ invariants.
\end{proposition}
\begin{proof}
The matrix unit property
\begin{equation}
	Q^{\lambda, \gamma}_{\mu \nu}Q^{\lambda', \gamma'}_{\mu' \nu'} = \delta^{\lambda \lambda'}\delta^{\gamma \gamma'}\delta_{\nu \mu'} Q^{\lambda, \gamma}_{\mu \nu'},
\end{equation}
of the $SP_k(N)$ basis follows from that of the $P_k(N)$ units together with orthogonality of $E^{\lambda, \gamma}_{p \mu}$,
\begin{equation}
	\begin{aligned}
			Q^{\lambda, \gamma}_{\mu \nu}Q^{\lambda', \gamma'}_{\mu' \nu'}&=  \sum_{\substack{\alpha, \beta, p \\ \alpha', \beta', p'}}(B^{\lambda \rightarrow \gamma})^\alpha_{p \mu}(B^{\lambda \rightarrow \gamma})^\beta_{p \nu}(B^{\lambda' \rightarrow \gamma'})^{\alpha'}_{p' \mu'}(B^{\lambda' \rightarrow \gamma'})^{\beta'}_{p' \nu'} Q^{\lambda}_{\alpha \beta} Q^{\lambda'}_{\alpha' \beta'} \\
			&=\sum_{\substack{\alpha, \beta, p \\ \alpha', \beta', p'}}(B^{\lambda \rightarrow \gamma})^\alpha_{p \mu}(B^{\lambda \rightarrow \gamma})^\beta_{p \nu} (B^{\lambda' \rightarrow \gamma'})^{\alpha'}_{p' \mu'}(B^{\lambda' \rightarrow \gamma'})^{\beta'}_{p' \nu'} \delta^{\lambda \lambda'} \delta_{\beta \alpha'} Q_{\alpha \beta'}^{\lambda} \\
			&=\sum_{\substack{\alpha, p \\ \beta', p'}}(B^{\lambda \rightarrow \gamma})^\alpha_{p \mu}(B^{\lambda' \rightarrow \gamma'})^{\beta'}_{p' \nu'}\delta^{\lambda \lambda'}\delta^{\gamma \gamma'}\delta_{pp'}\delta_{\nu \mu'}Q_{\alpha \beta'}^{\lambda}  \\
			&=\sum_{\substack{\alpha,\beta',p}}(B^{\lambda \rightarrow \gamma})^\alpha_{p \mu}(B^{\lambda' \rightarrow \gamma'})^{\beta'}_{p \nu'}\delta^{\lambda \lambda'}\delta^{\gamma \gamma'}\delta_{\nu \mu'}Q_{\alpha \beta'}^{\lambda} \\
			&=\delta^{\lambda \lambda'}\delta^{\gamma \gamma'}\delta_{\nu \mu'} Q^{\lambda, \gamma}_{\mu \nu'}.
		\end{aligned} \label{eq: SPkN matrix unit product}
\end{equation}
Going from the first line to the second we used the matrix unit property of $Q^{\lambda_1}_{\alpha \beta}$. Going from the second line to the third line uses orthogonality
\begin{equation}
	\sum_{\alpha} (B^{\lambda \rightarrow \gamma})^\alpha_{p \mu} (B^{\lambda \rightarrow \gamma'})^\alpha_{q \nu} = \delta^{\gamma \gamma'} \delta_{pq} \delta_{\mu \nu}.
\end{equation}
\end{proof}

This proposition should be understood as the Artin-Wedderburn decomposition of $SP_k(\N)$. Recall that
\begin{equation}
	P_k(\N)  \cong \bigoplus_{\lambda \in \Lambda_{k,\N}} Z_\lambda \otimes Z_{\lambda},
\end{equation}
as a left and right representation of $P_k(\N)$. As a left and right representation of $S_k$ we have (using \eqref{eq: branching Pk to Sk})
\begin{equation}
	\bigoplus_{\lambda \in \Lambda_{k,\N}} \qty(\bigoplus_{\gamma \vdash k} V_\gamma \otimes M_{\lambda \rightarrow \gamma}) \otimes \qty(\bigoplus_{\gamma' \vdash k} V_{\gamma'} \otimes M_{\lambda \rightarrow \gamma'}).
\end{equation}
Projecting to the $S_k$ invariant part gives
\begin{equation}
	SP_k(\N) \cong \bigoplus_{\lambda \in \Lambda_{k,\N}} \bigoplus_{\gamma \vdash k} M_{\lambda \rightarrow \gamma} \otimes M_{\lambda \rightarrow \gamma}.
\end{equation}
These steps are precisely mimicked in \eqref{eq: spkn units} -- we branch to $S_k$ on both sides and project to the $S_k$ invariants by contracting the $S_k$ representation indices. The above equation is an Artin-Wedderburn decomposition of
\begin{equation}
	 SP_k(\N) \cong \End_{S_N \times S_k}(\VN^{\otimes k}).
\end{equation}

The matrix units define states in $\Hilbertspace^{(k)}$
\begin{definition}
Let $\N \geq 2k$ and define the representation states
\begin{equation}
	\ket*{Q^{\lambda, \gamma}_{\mu \nu}} = \Tr_{\VN^{\otimes k}}(Q^{\lambda, \gamma}_{\mu \nu} (a^{\dagger})^{\otimes k})\ket{0}. \label{eq: rep basis states}
\end{equation}
They form a basis for $\Hilbertspace^{(k)}$.
\end{definition}
\begin{proposition}
	The representation basis is orthogonal (for $\N \geq 2k$)
	\begin{equation}
		\bra*{Q^{\lambda, \gamma}_{\mu \nu}}\ket*{Q^{\lambda', \gamma'}_{\mu' \nu'}} \propto \delta^{\lambda \lambda'}\delta^{\gamma \gamma'}\delta_{\mu\mu'}\delta_{\nu \nu'}.
	\end{equation}
\end{proposition}
\begin{proof}
The proof goes as follows
\begin{equation}
	\begin{aligned}
		\bra*{Q^{\lambda, \gamma}_{\mu \nu}}\ket*{Q^{\lambda', \gamma'}_{\mu' \nu'}} &= \sum_{\tau \in S_k} \Tr_{\VN^{\otimes k}}(Q^{\lambda, \gamma}_{\mu \nu}\tau \qty(Q^{\lambda', \gamma'}_{\mu' \nu'})^T \tau^{-1} ) \\
		&= \sum_{\tau \in S_k} \Tr_{\VN^{\otimes k}}(Q^{\lambda, \gamma}_{\mu \nu}\tau \qty(Q^{\lambda', \gamma'}_{\mu' \nu'})^T \tau^{-1}) \\ 
		&= k! \Tr_{V_N^{\otimes k}}(Q^{\lambda, \gamma}_{\mu \nu} \qty(Q^{\lambda', \gamma'}_{\mu' \nu'})^T ) \\
		&= k! \delta^{\lambda \lambda'} \delta_{\gamma \gamma'} \delta_{{\nu} {\nu'}} \Tr_{\VN^{\otimes k}}(Q^{\lambda, \gamma}_{\mu \mu'} ).
	\end{aligned}
\end{equation}
In the second equality we used $\qty(Q^{\lambda', \gamma'}_{\mu' \nu'})^T = Q^{\lambda', \gamma'}_{\nu' \mu'}$ which follows from equation \eqref{eq: transpose diagram is transpose matrix element}. Note that
\begin{equation}
	\begin{aligned}
		\Tr_{\VN^{\otimes k}}(Q^{\lambda', \gamma'}_{\mu \mu'})
		&= \Tr_{\VN^{\otimes k}}(Q^{\lambda', \gamma'}_{\mu 1}Q^{\lambda', \gamma'}_{1 \mu'}) \\
		&=  \Tr_{\VN^{\otimes k}}(Q^{\lambda', \gamma'}_{1 \mu'} Q^{\lambda', \gamma'}_{\mu 1}) \\
		&= \delta_{\mu \mu'}\Tr_{\VN^{\otimes k}}(Q^{\lambda', \gamma'}_{1 1})
		%&= \delta_{\mu \mu'}\mathcal{N}_{\Lambda_1 \Lambda_2},
	\end{aligned}
\end{equation}
such that the normalization only depends on irreducible representations $\lambda, \gamma$, which proves orthogonality.
\end{proof}
The normalization constant is readily computed as follows. From Proposition \ref{prop: orthogonality of matrix elements}
\begin{equation}
	\chr^{{\lambda'}}(Q_{\alpha \beta}^{\lambda}) =\delta^{\lambda \lambda'} \delta_{\alpha \beta}.
\end{equation}
We use this fact together with Schur-Weyl duality to compute $\Tr_{\VN^{\otimes k}}(Q^{\lambda}_{\alpha \beta})$
\begin{equation}
	\Tr_{\VN^{\otimes k}}(Q^{\lambda}_{\alpha \beta}) = \sum_{\lambda' \in \Lambda_{k,\N}} \dimSN{\lambda'}\chr^{\lambda'}(Q^{\lambda}_{\alpha \beta}) = \sum_{\lambda' \in \Lambda_{k,\N}} \dimSN{\lambda'} \delta^{\lambda \lambda'}\delta_{\alpha \beta} = \dimSN{\lambda} \delta_{\alpha \beta}. \label{eq: trace of matrix units}
\end{equation}
Consequently,
\begin{equation}
	\begin{aligned}
			\Tr_{\VN^{\otimes k}}(Q^{\lambda, \gamma}_{\mu \nu}) 
			&= \sum_{\alpha, \beta, p} (B^{\lambda \rightarrow \gamma})^\alpha_{p \mu}(B^{\lambda \rightarrow \gamma})^\beta_{p \nu}\delta_{\alpha \beta}\dimSN{\lambda} \\
			&=\sum_{p} \delta_{pp}\delta_{\mu \nu}\dimSN{\lambda} \\
			&= \delta_{\mu \nu}\dimSN{\lambda}\dimSN{\gamma},
	\end{aligned} \label{eq: normalization constant as dimensions}
\end{equation}
where the last two equalities hold if and only if the branching coefficients are non-zero.

In \cite{Barnes:2022qli} we gave explicit constructions of matrix units for $SP_1(\N), SP_2(\N)$ and $SP_3(\N)$ using sets of commuting elements, inspired by the method used to construct matrix units for $P_k(\N)$ in Section \ref{sec: construction of units}. This required careful choices of commuting elements for each $k=1,2,3$ to ensure that all indices are distinguished by the eigenvalues of the corresponding linear operators. Here, we will describe a procedure that works for general $k$ given the matrix units for $P_k(\N)$. The downside is that this procedure only works for a subset of matrix units with $\lambda$ such that $\abs{\lambda^\#} = k$.

The procedure relies on the following result.
\begin{theorem}
Let  $Z_\lambda$ be an irreducible representation of $P_k(\N)$ with $\abs{\lambda^\#} = k$. Then the restriction to $\mathbb{C}S_k$ is multiplicity free and
\begin{equation}
	\Res{P_k(\N)}{\mathbb{C}S_k}{Z_\lambda} = V_{\lambda^{\#}}.
\end{equation}
\end{theorem}
\begin{proof}
	See \cite[Remark 4.18]{Halverson2018}.
\end{proof}
This implies that there exists a basis for $Z_\lambda$ (given in \cite{Halverson2018} in terms of set partition tableaux) where the branching coefficients are trivial
\begin{equation}
	(B^{\lambda \rightarrow \lambda^{\#}})^\alpha_{p} = \delta^\alpha_p.
\end{equation}
Consequently
\begin{equation}
	Q^{\lambda, \lambda^{\#}} = \sum_{\alpha=1}^{\dimPk{\lambda}}Q_{\alpha \alpha}^{\lambda}.
\end{equation}
This allows us to leverage the construction of $P_k(\N)$ matrix units in Section \ref{sec: construction of units} to immediately find the multiplicity free matrix units of $SP_k(\N)$.

\section{Algebraic Hamiltonians on algebraic states}\label{sec: algebraic hammy}
In this section we will leverage the results in the previous section to construct Hamiltonians that act algebraically through diagram concatenation on the invariant states. These operators will be labelled by elements of the partition algebra, and in the subspace $\Hilbertspace \subset \mathcal{H}$, the matrix units form exact eigenvectors. We discuss some challenges in constructing Hamiltonians that distinguish all the labels of $SP_k(\N)$ matrix units.

To motivate the next definition, consider the following operator
\begin{equation}
	\mathcal{O} = \frac{1}{2}(a^\dagger)^{i^{}_{1'} }_{i_2^{}}(a^\dagger)^{i^{}_{2'}}_{i_1^{}}a^{i_1}_{i^{}_{1'}}a^{i_2}_{i^{}_{2'}}.
\end{equation}
The following proposition gives the action of this operator on $\mathcal{H}$.
\begin{proposition}\label{prop: action of O}
	Let $\mathcal{O}$ be the operator defined in the previous equation and $\ket{T} \in \mathcal{H}^{(r)}$ a degree $r$ state as in Definition \ref{def: tensor state}, then
	\begin{equation}
		\mathcal{O}\ket{T} = \begin{cases}
			\sum_{1 \leq i < j \leq r} \ket{D_{(ij)} T} &\qq{for $r \geq 2$}, \\
			0 &\qq{otherwise,}
		\end{cases}
	\end{equation}
	where $D_{(ij)}$ acts on $\VN^{\otimes r}$ by permuting tensor factors $i$ and $j$.
\end{proposition}
\begin{proof}
	It is clear that $\mathcal{O}$ vanishes on states of degree $r < 2$. Thus, we now assume $r \geq 2$, for which the proposition follows from a straight-forward computation. First note that
	\begin{align}
		&a^{i_1}_{i_{1'}^{}}a^{i_2}_{i_{2'}^{}}(a^\dagger)^{l_1}_{k_1} \dots (a^\dagger)^{l_r}_{k_r}\ket{0} \\
		& = \frac{1}{(r-2)!}\sum_{\gamma \in S_r} \comm{a^{i_1}_{i_{1'}^{}}}{(a^\dagger)^{l_{(1)\gamma}}_{k_{(1)\gamma}}}\comm{a^{i_2}_{i_{2'}^{}}}{(a^\dagger)^{l_{(2)\gamma}}_{k_{(2)\gamma}}}(a^\dagger)^{l_{(3)\gamma}}_{k_{(3)\gamma}} {\dots} (a^\dagger)^{l_{(r)\gamma}}_{k_{(r)\gamma}}\ket{0}.
	\end{align}
	Therefore
	\begin{equation}
		\mathcal{O}(a^\dagger)^{l_1}_{k_1} \dots (a^\dagger)^{l_r}_{k_r}\ket{0} = \frac{1}{2(r-2)!}\sum_{\gamma \in S_r}(a^\dagger)^{l_{(1)\gamma}}_{k_{(2)\gamma}} (a^\dagger)^{l_{(2)\gamma}}_{k_{(1)\gamma}}(a^\dagger)^{l_{(3)\gamma}}_{k_{(3)\gamma}} {\dots} (a^\dagger)^{l_{(r)\gamma}}_{k_{(r)\gamma}}\ket{0}.
	\end{equation}
	Define $\tau = (12) \in S_r$ and re-write the above equation as
	\begin{equation}
		\mathcal{O}(a^\dagger)^{l_1}_{k_1} \dots (a^\dagger)^{l_r}_{k_r}\ket{0} = \frac{1}{2(r-2)!}\sum_{\gamma \in S_r}(a^\dagger)^{l_{(1)\gamma}}_{k_{(1)\tau\gamma}} (a^\dagger)^{l_{(2)\gamma}}_{k_{(2)\tau\gamma}}(a^\dagger)^{l_{(3)\gamma}}_{k_{(3)\gamma}} {\dots} (a^\dagger)^{l_{(r)\gamma}}_{k_{(r)\gamma}}\ket{0}.
	\end{equation}
	Therefore
	\begin{align}
		\mathcal{O}\ket{T} &= \sum_{\substack{l_1 \dots l_r \\ k_1 \dots k_r}} T_{l_1 \dots l_r}^{k_1 \dots k_r} \frac{1}{2(r-2)!}\sum_{\gamma \in S_r}(a^\dagger)^{l_{(1)\gamma}}_{k_{(1)\tau\gamma}} (a^\dagger)^{l_{(2)\gamma}}_{k_{(2)\tau\gamma}}(a^\dagger)^{l_{(3)\gamma}}_{k_{(3)\gamma}} {\dots} (a^\dagger)^{l_{(r)\gamma}}_{k_{(r)\gamma}}\ket{0} \\
		&= \frac{1}{2(r-2)!}\sum_{\gamma \in S_r}\sum_{\substack{l_1 \dots l_r \\ k_1 \dots k_r}} T_{l_{(1)\gamma^{-1}} \dots l_{(r)\gamma^{-1}}}^{k_{(1)(\tau\gamma)^{-1}}k_{(2)(\tau\gamma)^{-1}}k_{(3)\gamma^{-1}}\dots k_{(r)\gamma^{-1}}} (a^\dagger)^{l_{1}}_{k_{1}} (a^\dagger)^{l_{2}}_{k_{2}}(a^\dagger)^{l_{3}}_{k_{3}} {\dots} (a^\dagger)^{l_{r}}_{k_{r}}\ket{0} \\
		&= \frac{1}{2(r-2)!}\sum_{\gamma \in S_r}\sum_{\substack{l_1 \dots l_r \\ k_1 \dots k_r}} T_{l_{(1)\gamma} \dots l_{(r)\gamma}}^{k_{(1)\gamma\tau}k_{(2)\gamma\tau}k_{(3)\gamma}\dots k_{(r)\gamma}} (a^\dagger)^{l_{1}}_{k_{1}} (a^\dagger)^{l_{2}}_{k_{2}}(a^\dagger)^{l_{3}}_{k_{3}} {\dots} (a^\dagger)^{l_{r}}_{k_{r}}\ket{0},
	\end{align}
	where in the last line we relabelled the sum over $\gamma$ to a sum over $\gamma^{-1}$ and used $\tau^{-1}=(12)^{{-1}}=(12)=\tau$.
	We re-write this as a trace of operators on $\VN^{\otimes r}$ and find
	\begin{equation}
		\mathcal{O}\ket{T} = \frac{1}{2(r-2)!}\sum_{\gamma \in S_r} \ket{D_{\tau}D_{\gamma^{-1}}TD_{\gamma}} = \sum_{1 \leq i < j \leq r} \ket{D_{(ij)} T}.
	\end{equation}
	The last equality follows from $D_{\gamma}(a^\dagger)^{\otimes k}=(a^\dagger)^{\otimes r} D_{\gamma}$ and the fact that $(ij)$ is fixed by elements in $S_2 \times S_{r-2} \subset S_r$ under conjugation.
\end{proof}

The operator $\mathcal{O}$ can be written in the very suggestive form
\begin{equation}
	\mathcal{O} = \frac{1}{2}\Tr_{\VN^{\otimes 2}}((a^\dagger)^{\otimes 2} D_\tau a^{\otimes 2}),
\end{equation}
where $\tau = {(12)}$.
It has a natural generalization
\begin{definition}\label{def: Ok for Sr}
	Let $\tau = (12 \dots k) \in S_k$ and define the operator
	\begin{equation}		
		\mathcal{O}_k = \frac{1}{k}\Tr_{\VN^{\otimes k}}((a^\dagger)^{\otimes k} D_\tau a^{\otimes k}).
	\end{equation}
\end{definition}
Note that $\tau$ is invariant under conjugation by permutations in $S_{r-k}$ and cyclic permutations of $\{1,2,\dots,k\}$.
Identical steps as for the $k=2$ case gives the following corollary.
\begin{corollary}
	Let $C_k$ be the conjugacy class of $S_r$ with elements having the same cycle structure as $(12\dots k)(k+1) \dots (r)$ and $\ket{T} \in \mathcal{H}^{(r)}$, then
	\begin{equation}
		\mathcal{O}_k\ket{T} = \begin{cases}
			\sum_{\tau \in C_k \subset S_r} \ket{D_{\tau} T} &\qq{for $r \geq k$,}\\
			0 &\qq{otherwise.}
		\end{cases}
	\end{equation}
\end{corollary}
We now use these operators to construct Hamiltonians with solvable spectra.

%
%\subsection{Partition algebra elements as operators}
%To define these algebraic Hamiltonians we use the following definition
%\begin{definition}
%	Let $d_\pi \in P_k(\N)$ be a diagram basis element and define the corresponding operator on $\mathcal{H}$ by
%	\begin{equation}
%		\mathcal{O}_\pi = \Tr_{\VN^{\otimes k}}((a^\dagger)^{\otimes k}d_\pi a^{\otimes k}).
%	\end{equation}
%\end{definition}

\subsection{Spectra of algebraic Hamiltonians.}
The operators $\mathcal{O}_k$ are in fact Hermitian since
\begin{align}
	\bra{T'}\mathcal{O}_k\ket{T} &= \sum_{\tau \in C_k \subset S_r}\bra{T} \ket{D_\tau T} \\
	& = \sum_{\gamma \in S_r}\sum_{\tau \in C_k \subset S_r}  \Tr_{\VN^{\otimes r}}((T')^\dagger  D_{\gamma} D_\tau T D_{\gamma^{-1}}) \\
	& = \sum_{\gamma \in S_r}\sum_{\tau \in C_k \subset S_r}  \Tr_{\VN^{\otimes r}}((T')^\dagger D_\tau  D_{\gamma}  T D_{\gamma^{-1}}) \label{eq: Ok hermitian 3} \\
	& = \sum_{\gamma \in S_r}\sum_{\tau \in C_k \subset S_r}  \Tr_{\VN^{\otimes r}}(( D_\tau^\dagger T')^\dagger D_{\gamma}  T D_{\gamma^{-1}}) \\
	& = \sum_{\gamma \in S_r}\sum_{\tau \in C_k \subset S_r}  \Tr_{\VN^{\otimes r}}(( D_\tau T')^\dagger D_{\gamma}  T D_{\gamma^{-1}}) \label{eq: Ok hermitian 5}\\
	& = \bra{T'}\mathcal{O}_k^\dagger\ket{T}
\end{align}
where in going to \eqref{eq: Ok hermitian 3} we used the fact that $\sum_{\tau \in C_k \subset S_r} D_{\gamma} D_\tau = \sum_{\tau \in C_k \subset S_r} D_\tau D_{\gamma}$ and \eqref{eq: Ok hermitian 5} uses
\begin{equation}
	\sum_{\tau \in C_k \subset S_r} D_\tau^\dagger = \sum_{\tau \in C_k \subset S_r} D_{\tau^{-1}} = \sum_{\tau \in C_k \subset S_r} D_{\tau},
\end{equation}
which holds because $\tau$ and $\tau^{-1}$ are in the same conjugacy class for symmetric groups.

Note that the following element of $\mathbb{C}S_r$
\begin{equation}
	\sum_{\tau \in C_k \subset S_r} \tau,
\end{equation}
is central. Therefore, the corresponding linear operator $\mathcal{O}_k$ acts on irreducible representations of $S_r$ by a scalar, and in particular a normalized character. The normalized characters are known for general $r=k$ \cite[Theorem 4]{Lassalle} (see \eqref{eq: norm characters of T2s} for $k=2$). For $k=3$ we have
\begin{theorem}
	Let $\gamma \vdash r$ and $\YT{\gamma}$ the corresponding Young diagram. Then
	\begin{equation}
		\sum_{\tau \in C_k \subset S_r} \hat{\chr}_\gamma(\tau) = \sum_{(i,j) \in \YT{\gamma}}  (j-i)^2 - \binom{k}{2}\label{eq: norm characters of T3},
	\end{equation}
	where $(i,j)$ corresponds to the cell in the $i$th row and $j$th column of the Young diagram (the top left box has coordinate $(1,1)$).
\end{theorem}
\begin{proof}
	See \cite[Theorem 4]{Lassalle}.
\end{proof}

This leads us to define the following Hamiltonians
\begin{definition}
	Let $\mathcal{O}_k$ be as in Definition \ref{def: Ok for Sr} and define the Hermitian operator
	\begin{equation}
		H = \frac{\widehat{K}(\widehat{K}+1)}{2} + \mathcal{O}_3 + g\mathcal{O}_2, \label{eq: Sk Hammy}
	\end{equation}
	where
	\begin{equation}
		\widehat{K} = \sum_{i,j=1}^\N a^\dagger_{ij} a_{ij}.
	\end{equation}
\end{definition}

From \eqref{eq: norm characters of T2s}, \eqref{eq: norm characters of T3} and the fact that $\widehat{K}$ has eigenvalue $r$ on $\mathcal{H}^{(r)}$, we find that the eigenvalues of $H$ are
\begin{align}
	r &\qq{for $r = 0,1$,}\\
	r+g\sum_{(i,j) \in \YT{\gamma}} (j-i)&\qq{for $r=2$,}\\
	 r + \sum_{(i,j) \in \YT{\gamma}}  (j-i)^2 + g(j-i) &\qq{for $r \geq 3$.}
\end{align}
where $\gamma \vdash r$. The first two terms are always positive, while the third term is negative for Young diagrams with many rows and few columns. In fact, for fixed $r$ and $g > 0$ the third term is minimized for the anti-symmetric Young diagram $[1^r]$. For $g < 0$ it is minimized for the symmetric diagram $[r]$. We leave it as a future direction to investigate when the spectrum of this Hamiltonian is bounded from below, and what the ground states are.

\subsection{Algebraic eigenvectors.}
It is difficult to exactly diagonalize the Hamiltonian \eqref{eq: Sk Hammy} in $\mathcal{H}^{(k)}$ for large $\N$. However, in the smaller subspace $\Hilbertspace^{{(k)}}$ the Hamiltonian can be understood as acting on $SP_k(\N) \subset P_k(\N)$ through left multiplication. The dimension of $P_k(\N)$ is independent of $\N$ and there is hope in finding exact eigenvectors of $H$ inside $\Hilbertspace$ for all $\N$. As we will now see, the matrix units of $SP_k(\N)$ are exact eigenvectors of $H$.

The following corollary gives the action of $\tau$ on matrix units $Q^{\lambda, \gamma}_{\mu {\nu}}$.
\begin{corollary}
	Let $\tau \in S_k$ and $Q^{\lambda, \gamma}_{\mu {\nu}}$ a matrix unit as defined in \eqref{eq: spkn units} then
	\begin{equation}
		D_\tau Q^{\lambda, \gamma}_{\mu {\nu}} = \sum_{\alpha, \beta, q,p} Q_{\alpha \beta}^\lambda D^\gamma_{qp}(\tau) (B^{\lambda \rightarrow \gamma})^\alpha_{q\mu}(B^{\lambda \rightarrow \gamma})^\beta_{p\nu}.
	\end{equation}
\end{corollary}
This follows from Corollary \ref{cor: d on Q} and the equivariance property of the branching coefficients.
An immediate consequence of this is that
\begin{equation}
	\sum_{\tau \in C_k \subset S_k} D_\tau Q^{\lambda, \gamma}_{\mu {\nu}} = \sum_{\tau \in C_k \subset S_k} \hat{\chr}^{}_{\gamma}(\tau)Q^{\lambda, \gamma}_{\mu {\nu}},
\end{equation}
because $\sum_{\tau \in C_k \subset S_k} D_\tau$ is central in $\mathbb{C}S_k$. As we show below, it follows that the representation basis elements form eigenvectors of $\mathcal{O}_k$,
\begin{equation}
	\mathcal{O}_k\ket{Q^{\lambda, \gamma}_{\mu {\nu}}} = \sum_{\tau \in C_k \subset S_k} \ket{D_\tau Q^{\lambda, \gamma}_{\mu {\nu}} }= \sum_{\tau \in C_k \subset S_k} \hat{\chr}^{}_{\gamma}(\tau)\ket{Q^{\lambda, \gamma}_{\mu {\nu}}}. \label{eq: eigenvectors of Ok}
\end{equation}

Hamiltonians very similar to those constructed in this section were discussed in \cite[Section 5]{Geloun2021}. Here we have given a construction of these Hamiltonians in a matrix oscillator setting. In the next subsection we discuss some extensions beyond symmetric group Hamiltonians

\subsection{Extensions beyond symmetric group Hamiltonians.}
In Equation \ref{eq: eigenvectors of Ok} we saw that the irreducible representation $\gamma$ of $S_k$ is closely connected to the eigenvalues of the Hamiltonian defined in \eqref{eq: Sk Hammy}. Naturally, we ask if there exists analogues of the Hermitian operators $\mathcal{O}_k$, whose eigenvalues are closely connected to the irreducible representation $\lambda$ of $\SN$.

There is no issue in generalizing Definition \ref{def: Ok for Sr} to arbitrary elements $\SPk{d} \in SP_k(\N)$.
\begin{definition}
	Let $\SPk{d} \in SP_k(\N)$ and define the operator
	\begin{equation}
		\mathcal{O}_\SPk{d} = \frac{1}{k!}\Tr_{\VN^{\otimes k}}((a^\dagger)^{\otimes k} \SPk{d} a^{\otimes k}).
	\end{equation}
\end{definition}
Unfortunately, the action of $\mathcal{O}_\SPk{d}$ on $\mathcal{H}^{(r)}$ is more complicated than the action of $\mathcal{O}_k$. In particular, using the same computation as for proving Proposition \ref{prop: action of O}, one finds
\begin{equation}
	\mathcal{O}_{\SPk{d}}\ket{T} = \begin{cases}
		0 &\qq{for $r < k$,} \\
		\ket{\SPk{d} T} &\qq{for $r = k$,} \\
		\frac{1}{k!}\sum_{\gamma \in S_r} \ket{D_\gamma(\SPk{d} \otimes \idn^{r-k})D_{\gamma^{-1}} T} &\qq{for $r > k$.}
	\end{cases}
\end{equation}
Therefore, one might pick $\SPk{d}$ to be an element of the center of $P_k(\N)$, in which case
\begin{equation}
	\mathcal{O}_\SPk{d} \ket{Q^{\lambda, \gamma}_{\mu {\nu}}} = \hat{\chr}_\lambda(\SPk{d}) \ket{Q^{\lambda, \gamma}_{\mu {\nu}}},
\end{equation}
for $Q^{\lambda, \gamma}_{\mu {\nu}} \in SP_k(\N)$. However, for representation basis elements corresponding to $SP_r(\N)$ with $r > k$ we have
\begin{equation}
	\mathcal{O}_\SPk{d} \ket{Q^{\lambda, \gamma}_{\mu {\nu}}}  = \frac{1}{k!}\sum_{\gamma \in S_r} \ket{D_\gamma(\SPk{d} \otimes \idn^{r-k})D_{\gamma^{-1}} Q^{\lambda, \gamma}_{\mu {\nu}} }.
\end{equation}
It is not clear that
\begin{equation}
	\sum_{\gamma \in S_r} \gamma(\SPk{d} \otimes \idn^{r-k}){\gamma^{-1}},
\end{equation}
is a central element of $SP_r(\N)$ or $P_r(\N)$. Therefore, we are no longer able prove that this acts by a normalized character of the representation $\lambda$.

In \cite{Barnes:2022qli} we proposed a solution to this by defining operators that vanish unless $r = k$. This definition allowed us to construct Hamiltonians that completely distinguish all the labels on matrix units using a complete set of commuting operators. So far, we do not have nice expressions for these operators in terms of finite sums of oscillators, as we do for $\mathcal{O}_k$. Finding such expressions is a very interesting problem for the future.

%Regardless, it is interesting to study its ground states for various values of $g$. An analytic answer to this question is beyond the scope of this thesis. We study the question numerically by looking at the smallest eigenvalue among a subset of Young diagrams labelled by $\lambda \vdash k$ for $0 \leq k \leq 30$.  In Figure \ref{fig: H groundstates} we plot the lowest energy Young diagram $\YT{\lambda}$ within this subset, for various values of $g$. We remark that for large enough $g$, the ground state may be outside the range $0 \leq k \leq 30$. 

\section{Extremal correlators} \label{sec: extremal correlators}
Extremal correlators in $\mathcal{N}=4$ SYM form interesting sectors having non-renormalization properties \cite{Eden2000}. They are closely connected to representation theoretic quantities such as Littlewood-Richardsson coefficients, and form a crucial set of examples for checking the AdS/CFT correspondence. In the quantum mechanical model presented in this paper, vacuum expectation values similar to extremal correlators can be computed exactly. They form generalizations of the two-point functions(inner products) previously considered and obey representation theoretic selection rules that we derive in the representation basis.

Definition \ref{def: diagram state def} can be interpreted as a quantum mechanical operator-state correspondence for $S_N$ invariant states
\begin{equation}
	\ket{\SPk{d}_\pi} \longleftrightarrow \mathcal{O}_{\pi} = \Tr_{V_N^{\otimes k}}(\SPk{D}_\pi (a^\dagger)^{\otimes k}).
\end{equation}
From equation \eqref{def: tensor state} we have
\begin{equation}
	\mathcal{O}_{\pi}^\dagger = \Tr_{V_N^{\otimes k}}( \SPk{D}_\pi^\dagger a^{\otimes k}).
\end{equation}
The time-dependent operators are given by
\begin{equation}
	\mathcal{O}_{\pi}(t) = \mathrm{e}^{-iH_0t}\mathcal{O}_{\pi}\mathrm{e}^{iH_0t} = \mathrm{e}^{-ikt} \mathcal{O}_{\pi},
\end{equation}
where$H_0$ is the free Hamiltonian, defined in equation \eqref{eq: simplest hamiltonian}.

With this notation, proposition \ref{prop: large N factorisation} takes the form
\begin{equation}
	\frac{\expval{\mathcal{O}_{\pi}^\dagger \mathcal{O}_{\pi'}}}{\sqrt{\expval{\mathcal{O}_{\pi}^\dagger \mathcal{O}_{\pi}}\expval{\mathcal{O}_{\pi'}^\dagger \mathcal{O}_{\pi'}} }} = \delta_{\pi \sim \pi'} + O(1/\sqrt{N}),
\end{equation}
where
\begin{equation}
	\delta_{\pi \sim \pi'}  = \begin{cases}
		1 \qq{if $\pi, \pi'$ differ by a permutation,} \\
		0 \qq{otherwise.}
	\end{cases}
\end{equation}

\subsection{Three-point correlators}
We now study the following generalization of the above expectation value
\begin{definition}[Extremal correlator]
	Let $\SPk{d}_{\pi_1} \in SP_{k_1}(N), \SPk{d}_{\pi_2} \in SP_{k_2}(N), \SPk{d}_\pi \in SP_k(N)$ such that $k = k_1 + k_2$. Extremal three-point correlators (of degree k) are expectation values of the form
	\begin{equation}
		\bra{0} \mathcal{O}_{\pi_1}^\dagger(t_1) \mathcal{O}^{\dagger}_{\pi_2}(t_2) \mathcal{O}_{\pi}(t) \ket{0}.\label{eq: extremal correlator}
	\end{equation}
\end{definition}
In what follows, we will ignore the trivial time-dependence that factorizes as
\begin{equation}
	\mathrm{e}^{i k_1t_1 + i k_2t_2 - i  k t } \bra{0} \mathcal{O}_{\pi_1}^\dagger \mathcal{O}^{\dagger}_{\pi_2} \mathcal{O}_{\pi}\ket{0}.
\end{equation}

As we now show, extremal correlators are simple in the diagram basis.
\begin{proposition}
	Let
	\begin{equation}
		\bra{0} \mathcal{O}_{\pi_1}^\dagger \mathcal{O}^{\dagger}_{\pi_2} \mathcal{O}_{\pi} \ket{0},
	\end{equation}
	be an extremal correlator of degree $k$ (ignoring the time-dependence). It is equal to
	\begin{equation}
		\sum_{\gamma \in S_k} \N^{\abs{\pi_1 \otimes \pi_2 \join \gamma \pi \gamma^{-1}}},
	\end{equation}
	where the tensor product $\pi_1 \otimes \pi_2$ is the diagram obtained by horizontal concatenation. For example,
	\begin{equation}
		\PAdiagram{2}{}{2/1, -1/-2} \otimes \PAdiagram{3}{1/-2}{3/1} = \PAdiagram{5}{3/-4}{2/1, -1/-2, 5/3}.
	\end{equation}
\end{proposition}
\begin{proof}
	We compute the extremal correlator using Wick contractions. By construction, there are no contractions between $\mathcal{O}_{\pi_1}^\dagger$ and $\mathcal{O}_{\pi_2}^\dagger$ that give a non-zero result. Therefore,
	\begin{equation}
		\bra{0} \mathcal{O}_{\pi_1}^\dagger \mathcal{O}^{\dagger}_{\pi_2} \mathcal{O}_{\pi} \ket{0} =  \sum_{\gamma \in S_k} \Tr_{\VN^{\otimes k}}((D_{\pi_1}^T \otimes D_{\pi_2}^T) D_{\gamma }D_{\pi} D_{\gamma^{-1}}),
	\end{equation}
	where the contractions are encoded using the sum over $S_k$ and $D_{\pi_1}, D_{\pi_2}, D_\pi$ are linear maps corresponding representative diagrams in the $S_k$ orbits $[d_{\pi_1}], [d_{\pi_2}], [d_\pi]$. The proof follows since the trace is $\N$ raised to the number of connected components.
\end{proof}

We will now derive representation theoretic selection rules for the extremal correlators.
\begin{proposition}
	Consider the operator-state correspondence in the representation basis
	\begin{equation}
		\ket*{Q^{\lambda, \gamma}_{\mu \nu}}  \rightarrow \mathcal{O}^{\lambda, \gamma}_{\mu \nu} = \Tr_{\VN^{\otimes k}}(Q^{\lambda, \gamma}_{\mu \nu} (a^{\dagger})^{\otimes k}).
	\end{equation}
	Extremal correlators in the representation basis satisfies
	\begin{equation}
		\bra{0} (\mathcal{O}^{\lambda', \gamma'}_{\mu' \nu'})^{\dagger} (\mathcal{O}^{\lambda'', \gamma''}_{\mu'' \nu''})^\dagger \mathcal{O}^{\lambda, \gamma}_{\mu \nu} \ket{0} = 0 \quad \text{if $C(\lambda, \lambda', \lambda'') = 0$,}
	\end{equation}
	where $C(\lambda, \lambda', \lambda'')$ is the Kronecker coefficient, or multiplicities of irreducible representations in the decomposition of tensor products of $\SN$ representations.
\end{proposition}
\begin{proof}
	As before, this correlator (modulo time-dependence) is proportional to a trace
	\begin{equation}
		\Tr_{\VN^{\otimes k}}(Q^{\lambda', \gamma'}_{\nu' \mu'} \otimes Q^{\lambda'', \gamma''}_{\nu'' \mu''} Q^{\lambda, \gamma}_{\mu \nu}),
	\end{equation}
	where we used the fact that transposition exchanges the order of the multiplicity indices and $ Q^{\lambda, \gamma}_{\mu \nu}$ being invariant under conjugation by $\gamma \in S_k$.
	
	To prove this proposition, it is sufficient to consider a trace of $P_k(\N)$ matrix units,
	\begin{equation}
		\Tr_{\VN^{\otimes k}}(Q^{\lambda'}_{\beta' \alpha'} \otimes Q^{\lambda''}_{\beta'' \alpha''} Q^{\lambda}_{\alpha \beta}).
	\end{equation}
	Schur-Weyl duality \eqref{eq: VN SW simple} gives a decomposition of the trace into $\SN \times P_k(\N)$ representations
	\begin{equation}
		\Tr_{\VN^{\otimes k}}(Q^{\lambda'}_{\beta' \alpha'}\otimes Q^{\lambda''}_{\beta'' \alpha''} Q^{\lambda}_{\alpha \beta}) =  \sum_{\tilde{\lambda} \in \Lambda_{k,\N}} \dimSN{\tilde{\lambda}} \Tr_{Z_{\tilde{\lambda}}}(Q^{\lambda'}_{\beta' \alpha'}\otimes Q^{\lambda''}_{\beta'' \alpha''} Q^{\lambda}_{\alpha \beta}).
	\end{equation}
	From orthogonality of matrix elements \eqref{eq: orthogonality Pk matrix elements} this vanishes unless $\tilde{\lambda} = \lambda$,
	\begin{equation}
		\Tr_{\VN^{\otimes k}}(Q^{\lambda'}_{\beta' \alpha'}\otimes Q^{\lambda''}_{\beta'' \alpha''} Q^{\lambda}_{\alpha \beta})  = \dimSN{\lambda}\Tr_{Z_\lambda}(Q^{\lambda'}_{\beta' \alpha'} \otimes Q^{\lambda''}_{\beta'' \alpha''} Q^{\lambda}_{\alpha \beta}).
	\end{equation}
	Orthogonality also gives
	\begin{equation}
		\Tr_{Z_\lambda}(Q^{\lambda'}_{\beta' \alpha'} \otimes Q^{\lambda''}_{\beta'' \alpha''} Q^{\lambda}_{\alpha \beta}) = D^{\lambda}_{\beta \alpha}(Q^{\lambda'}_{\beta' \alpha'} \otimes Q^{\lambda''}_{\beta'' \alpha''}).
	\end{equation}

	To find the selection rule, we want to consider $Z_\lambda$ as a representation of $P_{k_1}(\N) \otimes P_{k_2}(\N)$. Suppose $Z_\lambda$ decomposes as follows under this restriction
	\begin{equation}
		Z_\lambda \cong \bigoplus_{\substack{\widetilde{\lambda' }\in \Lambda_{k_1,\N}  \\ \widetilde{\lambda''} \in \Lambda_{k_2,\N}} } Z_{\widetilde{\lambda'}} \otimes Z_{\widetilde{\lambda''}} \otimes M_\lambda^{\widetilde{\lambda'}, \widetilde{\lambda''}},
	\end{equation}
	with multiplicities given by $\dim M_\lambda^{\widetilde{\lambda'}, \widetilde{\lambda''}}$. Again, orthogonality implies that the representation $\widetilde{\lambda'}$ of a matrix unit $Q^{\lambda'}_{\beta' \alpha'}$ vanishes unless $\widetilde{\lambda'} = \lambda'$. Therefore, we find that
	\begin{align}
		\bra{0} (\mathcal{O}^{\lambda', \gamma'}_{\mu' \nu'})^{\dagger} (\mathcal{O}^{\lambda'', \gamma''}_{\mu'' \nu''})^\dagger \mathcal{O}^{\lambda, \gamma}_{\mu \nu} \ket{0} &\propto \dimSN{\lambda}D^{\lambda}_{\beta \alpha}(Q^{\lambda'}_{\beta' \alpha'} \otimes Q^{\lambda''}_{\beta'' \alpha''}) \\
		&= 0 \qq{if $\dim M_\lambda^{{\lambda'}, {\lambda''}}$ = 0}.
	\end{align}
	
	The branching multiplicities for partition algebras are related to the multiplicities $ C(\lambda, {\lambda'}, \lambda'') $, known as Kronecker coefficients, of irreducible 
	representations $ \lambda$ in tensor products of $\SN$ representations $ {\lambda'} \otimes {\lambda''}$ (see eq. (3.1.3) of \cite{bowman2013partition})  
	\begin{equation}
		\dim M_\lambda^{{\lambda'}, {\lambda''}} = C(\lambda, {\lambda'}, \lambda'') .
	\end{equation}
	For simplicity we are assuming $ N \ge (2 k_1 + 2 k_2)$.
	
	Matrix units of $SP_k(\N)$ are linear combinations of matrix units for $P_k(\N)$ with fixed $\lambda$. Therefore, the same selection rule applies, and the proposition follows.
\end{proof}
\begin{remark}
	Analogous selection rules for extremal correlators in general quiver gauge theories are described in \cite{QuivCalc}.	
	For comparison, in Schur-Weyl duality between $U(N)$ and $\mathbb{C}[S_k]$, Littlewood-Richardson coefficients are branching multiplicities for $S_{k_1 + k_2} \rightarrow S_{k_1} \times S_{k_2}$ but correspond to decomposition of tensor products of $U(N)$ representations.
\end{remark}

\section{Summary}
This chapter extends the mathematical techniques developed in chapter \ref{chapter: partition algebra} to one-dimensional (quantum mechanical) matrix models. We reviewed the basics of matrix harmonic oscillators, their quantization and diagonalization of the free Hamiltonian. Following the review, we studied harmonic oscillators in a permutation invariant quadratic potential. The most general such potential can be constructed using similar techniques to those used to construct the distributions in chapter \ref{chapter: 0d}. This gives a Hamiltonian with exactly solvable eigenvalues. In particular, the very difficult problem of diagonalizing a $\N^2 \times \N^2$ matrix at large $\N$ is reduced to diagonalizing two real symmetric matrices of size two and three. This result was first presented in \cite{Barnes:2022qli}.

In the next section we studied the subspace $\Hilbertspace \subset \mathcal{H}$ of $\SN$ invariant states. This was also first done in \cite{Barnes:2022qli} where it was realised that the subspace is closely related to partition algebras. A basis for degree $k$ states in $\Hilbertspace$ is in one-to-one correspondence with a basis for the subalgebra $SP_k(\N) \subset P_k(\N)$, called the symmetrised partition algebra. These algebras were first defined in \cite{Barnes:2021tjp}, but are special cases of permutation centralizer algebras \cite{PCA2016} with $A=P_k(\N)$ and $B=\mathbb{C}(S_k)$.

We described three bases for $SP_k(\N)$: the diagram basis, the orbit basis and the representation basis. The orbit and diagram bases come from symmetrisation of the corresponding bases of $P_k(\N)$. The diagram and orbit bases of $P_k(\N)$ were known already by Jones and Martin \cite{Jones1994, Martin1994}. The existence of representation bases is also known in the literature since it follows from semi-simplicity of $P_k(\N)$. The diagram basis is the most geometrical of the three. It forms an orthogonal basis for $\N \rightarrow \infty$. The orbit basis is exactly orthogonal, for all $\N$. It is particularly useful for describing finite $\N$ effects, where the Hilbert space may develop states with zero norm. The first result was first proven in \cite{Barnes:2021tjp}, but in the setting of zero-dimensional matrix models. The second result was proven in \cite{Barnes:2022qli}. The representation basis is based on matrix units for $SP_k(\N)$ and form eigenvectors of algebraic Hamiltonians described in the subsequent section. Furthermore, they are useful for deriving selection rules of extremal correlators, as described in the last section of this chapter.

In the penultimate section we used diagram algebras to construct Hamiltonian operators that act on the invariant states through diagram multiplication. We described a family of such operators, based on the diagram algebra $\mathbb{C}(S_k)$. As mentioned, we found that the representation basis is an eigenbasis for these operators. However, these algebraic operators have a highly degenerate spectrum. This naturally raised the question of constructing generalizations of the algebraic Hamiltonians where all labels on the representation basis elements are distinguished by their eigenvalue. We discussed generalizations of these algebraic Hamiltonians based on partition algebras and some of the challenges of solving them. In particular, we saw that algebraic Hamiltonians coming from elements of $SP_k(\N)$ act nicely on states with degree $r=k$. On states with $r > k$ the action is significantly more complicated due to what essentially amounts to a non-trivial embedding of $SP_k(\N) \rightarrow SP_r(\N)$. At the moment we do not have a nice description of this embedding. Hamiltonians based on sums of symmetric group elements have been considered in Spin Matrix Theory \cite{Harmark2014} in connection to AdS/CFT and $\mathcal{N}=4$ SYM. We have considered a special case of exactly solvable combinations. The generalizations of algebraic Hamiltonians coming from $SP_k(\N)$ were first constructed in \cite{Barnes:2022qli}. Here, operators projecting to fixed degree $r=k$ states were used to overcome the previously mentioned challenges.

In the last section we studied expectation values inspired by extremal correlators in AdS/CFT. Schur-Weyl duality has been very successful in studying extremal correlators in gauge theories with continuous group symmetries \cite{CJR}. This work generalizes many of these results to permutation invariant observables. We found that the extremal correlators in the diagram basis have a simple description in terms of geometric quantities, e.g. number of connected components in the join of tensor products of diagrams. Extremal correlators in the representation basis gave rise to selection rules based on Kronecker coefficients of $\SN$ irreducible representations. This was proven in detail in \cite{Barnes:2022qli}. In this thesis, we took a shortcut that gives the selection rule but does not reveal the full detailed structure of the extremal correlators.
	
	\chapter{Conclusion}
In this thesis we have built on the Schur-Weyl duality framework for studying matrix models and invariant observables. Previous applications have focused on matrix models with continuous symmetry, where the Schur-Weyl dual algebras correspond to the symmetric group algebras $\mathbb{C}(S_k)$ or Brauer algebras. Here, we considered models with significantly less symmetry, that of discrete permutation symmetry. For models with permutation symmetry, the dual algebras are called partition algebras $P_k(\N)$. They generalize symmetric group algebras, and in fact $\mathbb{C}(S_k) \subset P_k(\N)$.

The first chapter of this thesis contained a short review of the essential facts about symmetric groups and their representation theory. The combinatorial objects known as vacillating tableaux played a particularly important role in determining the multiplicities in the decomposition of $\VN^{\otimes k}$. These facts were used in the subsequent chapter, where we defined and studied the structure of partition algebras. In particular, we saw that partition algebras are semi-simple algebras and form a so-called inductive chain $P_1(\N) \subset \dots \subset P_k(\N)$. Semi-simplicity implies that there exists a basis of matrix units for $P_k(\N)$, where multiplication mimics multiplication of block matrices. The inductive chain was used to construct this basis as eigenvectors of a complete set of commuting operators. This culminated in an all $\N$ construction of the matrix units for $P_k(\N)$. A table of matrix units for $P_2(\N)$, up to normalization, was given in appendix \ref{apx: P2N units}.

In the next chapter, the matrix units were used to define the most general permutation invariant Gaussian/quadratic matrix model, in block diagonal form. The block diagonal form (matrix units form) was essential for describing the first and second moment of the matrix model in closed form. Because the matrix model is quadratic, expectation values of observables are exactly computable using Wick's theorem and the expressions for first and second moments. We defined observables as general permutation invariant matrix polynomials and found that the space of observables has two useful bases: the 1-row partition basis and the basis of directed graphs. Using the 1-row partition basis, we gave an algebraic combinatorial algorithm for computing expectation values of observables. Code implementing this algorithm in Python/Sage is given in appendix \ref{apx: EV code} with detailed comments. The directed graph basis is useful for combinatorial counting and construction. We introduced Graph Generating Permutation Diagrams for describing directed graphs. This ultimately led to a double coset description of the space of observables described by graphs with fixed local structure. That is, graphs with fixed number of in/out-going edges at each vertex. A detailed discussion of how to compute the order of these double cosets using generating functions/cycle indices and associated computer code is given in appendix \ref{apx: double coset}.

In the last chapter, we applied our mathematical techniques to matrix models in one dimension. That is, the physics of quantum mechanical matrix oscillators. We constructed exactly solvable models of harmonic oscillators in a permutation invariant quadratic potential. Following that, we focused on the subspace of states invariant under the adjoint action of permutations. We presented three interesting bases for this subspace, based on the symmetrized partition algebra $SP_k(\N) \subset P_k(\N)$: the diagram basis, the orbit basis and the representation basis. We then considered algebraic interacting Hamiltonians based on partition algebras. We found that the representation basis forms an eigenbasis for the Hamiltonians based on symmetric group algebras $\mathbb{C}(S_k) \subset P_k(\N)$ and discussed some of the challenges in diagonalizing the more general Hamiltonians based on partition algebras.
The representation basis was also useful for deriving selection rules of extremal correlators inspired by AdS/CFT. Extremal correlators play a crucial role in AdS/CFT where they provide non-trivial checks on the duality.

The main results in the thesis are
\begin{enumerate}
	\item The construction of projection operators $P_{\vactab \vactab'}: P_k(\N) \rightarrow P_k(\N)$ in Definition \ref{def: vactab proj}, labelled by pairs of vacillating tableaux of shape $\lambda$, whose one-dimensional image is spanned by matrix units $Q^\lambda_{\vactab \vactab'} \in P_k(\N)$ for $k=1,2$.
	\item The all $\N$ algorithm in section \ref{subsec: all N matrix units} for finding a basis for the above-mentioned one-dimensional subspaces, or equivalently a set of matrix units for $P_k(\N)$ as closed form functions of $\N$ for $k=1,2$.
	\item The connection between first and second moments of PIGMMs given in Proposition \ref{prop: 1mom 2mom} and matrix units for $P_1(\N), P_2(\N)$, as explained in section \ref{subsec: 1row to 2-row}.
	\item The diagrammatic interpretation of permutation invariant observables as 1-row partition diagrams as explained in the beginning of section \ref{subsec: observables}.
	\item Proposition \ref{prop: graphs equals trace}, which proves that permutation invariant observables of degree $k$ are in one-to-one correspondence with unlabelled directed graphs with $k$ edges and $\N$ vertices.
	\item The combinatorial algorithm for computing expectation values of observables as functions of $\N$ using the diagrammatics of 1-row diagrams and partition algebras, together with Wick's theorem. This is given in section \ref{subsec: exp vals}.
	\item The group theoretical framework for enumerating directed graphs and permutation invariant observables, given in section \ref{subsec: graph counting}.
	\item Proposition \ref{prop: large N factorisation}, or Large $\N$ factorisation of permutation invariant states in matrix quantum mechanics.
\end{enumerate}

Zero-dimensional matrix models with continuous symmetry broken to a discrete symmetry were considered in \cite{Lionni2019}, where the introduction of higher-order terms in the action breaks the continuous symmetry. This is distinct from the models considered here, where the continuous symmetry acting on matrices is broken already at the quadratic level. Permutation invariant random matrix models have also been considered within mathematical statistics \cite{Gabriel2015a, gabriel2016b, Gabriel2015}. It would be interesting to explore applications of the techniques developed here to these models. More generally, adding higher-order terms into the models is a good problem to tackle in the future. This would allow us to explore the phase structure of permutation invariant models as has been done for $U(\N)$ invariant models \cite{osti1980, Wadia1980N, Skag1984, Douglas1993, Sundborg2000, Aharony2004, FHY2007, Dutta2008, ramgoolam2019quiver, 2021Ali}.

The techniques presented in this thesis, for constructing permutation invariant matrix models, generalizes to tensor models \cite{PIGTM}. This was one of the main motivations behind the authors choice in presenting the permutation invariant 1-matrix model of \cite{Ramgoolam2019a} in the manner laid out in this thesis. Furthermore, Schur-Weyl duality itself has been generalized to many settings. Therefore, the techniques described in this thesis should have applications to matrix/tensor models with other symmetries given by families of groups $G_N$. This would require the study of Jucys-Murphy type elements in other dual algebras. The framework should be a valuable tool for studying large $N$ models beyond the conventional settings and beyond eigenvalue distributions. Following this direction will form important bridges between classical random matrix theory and the work this thesis is based on.

A natural question to ask is whether there exists a gauge-string dual interpretation of permutation invariant observables in matrix models. In models with $U(\N)$ symmetry, non-singlet sectors have been considered in models of low-dimensional black holes \cite{Kazakov:2000pm, SDTD2003, Maldacena2005}. Permutation invariant sectors go beyond the $U(\N)$ singlet sector and it would be interesting to explore implications of the results in this thesis for space-time duals of permutation invariant states. The large $\N$ factorization of diagram basis elements is a particularly promising result in this direction. It indicates that classicality emerges as $\N \rightarrow \infty$, in the space of diagram basis states. This could be studied, for example, using the coherent state method \cite{Yaffe1982}.

Related to this, is the observation that the $1/\N$ expansion of correlators in $U(\N)$ invariant matrix models have a geometric interpretation in terms of counting branched coverings \cite{ITZYK, MelloKoch2010} (see also \cite{Gopak2011, dMKLN} for connections to topological strings). Given that the partition functions of matrix models with higher order $U(N)$ invariant potentials can be interpreted as generating functions for branched covers, it is natural to ask if permutation invariant potentials have a similar geometric interpretation. A good starting point would be to consider cubic deformations, for example,
\begin{equation}
	\delta V(X) = \sum_{i,j,k,l,m,n} X_{ij}X_{kl}X_{mn},
\end{equation}
of the $U(N)$ invariant quadratic potential, explicitly breaking the symmetry to $\SN$. Computing
\begin{equation}
	\expval{\mathrm{e}^{\lambda \delta V}},
\end{equation}
order-by-order in $\lambda$ may give some hints towards a geometric interpretation.

Hamiltonians based on symmetric groups have been considered in Spin Matrix Theory \cite{Harmark2014, Baiguera2022} in connection to AdS/CFT and $\mathcal{N}=4$ SYM. The algebraic Hamiltonians based on symmetric groups constructed in this thesis form exactly solvable versions, corresponding to elements in the centres of the group algebras. The extensions to Hamiltonians based on partition algebras are a natural generalization of the ones in Spin Matrix Theory. They deserve a more detailed investigation than the one presented here. For example, Spin Matrix Theory is connected to spin chains at large $N$ and low temperature. I find it plausible that a similar connection can be extended to the permutation invariant case. The Hamiltonians constructed in this thesis also have a natural interpretation in terms of a $N$-by-$N$ lattice of oscillators $a_{ij}^\dagger$ labelled by sites $(i,j)$. This was explored in \cite{Barnes:2022qli}. In the lattice interpretation, permutation invariant operators seem non-local at first glance, but may have a physical interpretation if a more sophisticated perspective is developed. For some applications it would be useful to study invariant operators that are polynomials in the position operators $X_{ij}$. This is related to understanding the action of invariant operators that do not have equal numbers of annihilation and creation operators. I expect that Schur-Weyl duality, in combination with diagrammatic techniques similar to the ones presented in this thesis, can provide a powerful tool for tackling this problem. This is particularly true for multi-matrix models, since Schur-Weyl duality techniques readily generalize to such models, and would be interesting even for the $U(N)$ case.

	%%------------------------------------------------------------------------------------
	%% Appendices 
	%%------------------------------------------------------------------------------------
	\appendix % resets chapter numbering, uses letters for appendix chapter numbers
	
	\chapter{Matrix units for $\mathbb{C}(\SN)$} \label{apx: SN units}
Theorem \ref{thm: SN irreps} tells us that the set of non-isomorphic irreducible representations of $\SN$ are vector spaces $V_\lambda$ of dimension $\abs{\SYT{\lambda}}$ for $\lambda \vdash \N$. In this appendix we briefly review the inductive approach to representation theory of symmetric groups presented in \cite{VershikOkounkov}. We will see that the basis of standard Young tableaux can be understood through so called Jucys-Murphy elements. In particular, diagonalizing the Jucys-Murphy elements acting on the group algebra gives the basis of matrix units and therefore the irreducible representations of $\SN$.

The famous formula
\begin{equation}
	\abs{G} = \sum_{R \in \Rep{G}} d_R^2,
\end{equation}
where $d_R = \chi^R(1)$ is the dimension of the irreducible representation $R$ is direct consequence of the Artin-Wedderburn decomposition \cite{Artin, Wedderburn} of $\mathbb{C}(G)$. Let $V_R$ be an irreducible representation of $G$, the decomposition says that
\begin{equation}
	\mathbb{C}(G) = \bigoplus_{R \in \Rep{G}} \End(V_R), \label{eq: AW decomp CG}
\end{equation}
where $\End(V_R)$ is the set of linear maps from $V_R$ to $V_R$. Given a basis for $V_R$, $\End(V_R)$ is the algebra of $d_R \times d_R$ matrices. In other words, \eqref{eq: AW decomp CG} says that there exists a change of basis on the group algebra that makes it manifestly an algebra of block matrices. The blocks are labelled by $R$ and each block is of dimension $d_R$. A formula for this change of basis is known.
\begin{theorem}(Fourier inversion formula)
	Let $\mathbb{C}(G)$ be a group algebra and $D_{ab}^R(g)$ be irreducible matrix representations of $G$ for all $R \in \Rep{G}$. The following elements in $\mathbb{C}(G)$
	\begin{equation}
		Q^R_{ab} = \frac{1}{\abs{G}}\sum_{g \in G} d_R D_{ab}^R(g^{-1})g,
	\end{equation}
	is an isomorphism of the type in \eqref{eq: AW decomp CG}. In other words,
	\begin{equation}
		Q^{{R}}_{ab} Q^{R'}_{a'b'} = \delta^{RR'} \delta_{ba'} Q^R_{ab'}.
	\end{equation}
	These elements, called matrix units, form a basis for the group algebra.
\end{theorem}
\begin{proof}
	This is a standard result in group theory and abstract algebra (see for example \cite[Proposition 11]{Serre1977}). We proved a version of this for partition algebras in Section \ref{sec: Semi-Simple Algebra Technology}.
\end{proof}
We will use the fact that matrix units have the following property
\begin{align}
	hQ^R_{ab} &= \frac{1}{\abs{G}}\sum_{g \in G} d_R D_{ab}^R(g^{-1})hg \\
	&=\frac{1}{\abs{G}}\sum_{g \in G} d_R D_{ab}^R((h^{-1}g)^{-1})g \\
	&=\frac{1}{\abs{G}}\sum_{g \in G} d_R D_{ab}^R(g^{-1}h)g \\
	&=Q^R_{ac}D_{cb}^R(h), \label{eq: g on Q}
\end{align}
and similarly for $Q^R_{ab}h$.

\section{Jucys-Murphy elements}
In the special case of $G = \SN$, Theorem \ref{thm: SN irreps} says that the indices $a,b$ on $D^\lambda_{ab}(\sigma)$ can be understood as standard Young tableaux of shape $\lambda$.
We will define the following elements of $\mathbb{C}(\SN)$.
\begin{definition}(Jucys-Murphy elements)
	Let $X_1 = 0$ and
	\begin{equation}
		X_i = (1i) + (2i) + \dots + (i-1 i), \quad  i=2,\dots,\N.
	\end{equation}
\end{definition}
\begin{remark}
	Note that
	\begin{equation}
		X_i = z_i - z_{i-1},
	\end{equation}
	where $z_i$ are the sums over transpositions in $S_i$ as defined in \eqref{eq: z_n}.
\end{remark}
They span a maximal commutative subalgebra of $\mathbb{C}(\SN)$ \cite{VershikOkounkov}. Remarkably, they act diagonally on the basis labelled by standard Young tableaux.
\begin{theorem}
	Let $\lambda \vdash \N$, $a \in \SYT{\lambda}$ a standard Young tableaux of shape $\lambda$ and $v_a$ the corresponding basis element in $V_\lambda$. Let $(i,j)$ be the coordinate of the box in $a$ filled with the number $k$, then
	\begin{equation}
		X_k v_a = (j-i)v_a.
	\end{equation}
	The number $j-i$ is known as the content of the box containing $k$. We denote this by $c_k(a)$.
\end{theorem}
\begin{proof}
	See for example \cite[Equation 12]{Jucys74}.
\end{proof}
In \cite{VershikOkounkov} they proved that the vector $v_a$, or equivalently the standard Young tableaux $a$ is uniquely determined by the tuple of eigenvalues of $X_i$. That is, there is a bijection between ordered lists
\begin{equation}
	(c_1(a), c_2(a), \dots, c_\N(a)) \label{eq: content list}
\end{equation}
and standard Young tableaux $a$.

Representation theoretically, this bijection can be understood from the following property of the Young tableaux basis. Suppose $a$ is a standard Young tableaux of shape $\lambda^{(\N)} \vdash \N$. Removing the box labelled by $\N$ gives a new standard Young tableaux of shape $\lambda^{(\N-1)} \vdash \N-1$. Iterating this procedure gives a sequence of standard Young tableaux of shape $\lambda^{(1)}, \dots, \lambda^{{(\N)}}$. For example
\begin{align}
	&{\tiny \ytableaushort{1 2 3, 4 5, 6}} \rightarrow && {\tiny \ytableaushort{1 2 3, 4 5}} \rightarrow && {\tiny \ytableaushort{1 2 3, 4} }\rightarrow &&{\tiny \ytableaushort{1 2 3}} \rightarrow &&{\tiny \ytableaushort{1 2}} \rightarrow &&{\tiny \ytableaushort{1}} \\
	&\lambda^{{6}} \rightarrow &&\lambda^{{5}} \rightarrow &&\lambda^{{4}} \rightarrow &&\lambda^{{3}} \rightarrow &&\lambda^{{2}} \rightarrow &&\lambda^{{1}}.
\end{align}
In particular, the sequence of Young diagrams themselves (not including fillings) completely fix the positions of the numbers $\{1,\dots,\N\}$ in the final tableaux. This is another way of saying that the branching $S_{n} \rightarrow S_{n-1}$ is multiplicity free, and therefore the sequence of irreducible representations in the branching$\SN \rightarrow S_{\N-1} \rightarrow \dots \rightarrow S_1$ label a unique (up to scalar) vector in an irreducible representation. Essentially, the ordered list \eqref{eq: content list} specifies the position of each number in the Young tableaux. We have not proven this here. For more information on this construction see \cite{VershikOkounkov}.

It follows from \eqref{eq: g on Q} that
\begin{equation}
	X_k Q^R_{ab} = c_k(b)Q^R_{ab}, \quad Q^R_{ab}X_k = c_k(a)Q^R_{ab}.
\end{equation}
Therefore, the simultaneous eigenbasis of the operators corresponding to left action and right action of $X_i$ on the group algebra $\mathbb{C}(\SN)$ is the matrix unit basis.

\subsection{Example 1.}
The simplest case to consider is $\N = 2$ or $S_2$. In this case we only need to consider
\begin{equation}
	X_2 = (12),
\end{equation}
since $X_1 = 0$. We have
\begin{equation}
	X_2 (12) = (12) X_2 = (1)(2), \quad X_2 (1)(2) = (12).
\end{equation}
That is, the matrices of the left and right action of $X_2$ on $\mathbb{C}(S_2)$ are
\begin{equation}
	X_2^L = X_2^R = \mqty(0 & 1 \\ 1 & 0).
\end{equation}
The vectors
\begin{equation}
	Q^{[2]}_{\tiny \ytableaushort{1 2}, \ytableaushort{1 2}} = (1)(2) + (12), \quad Q^{[1,1]}_{\tiny \ytableaushort{1,2}, \ytableaushort{1,2}} = (1)(2) - (12),
\end{equation}
form an eigenbasis of $X_2^L$ and $X_2^R$ since
\begin{equation}
	X_2Q^{[2]}_{\tiny \ytableaushort{1 2}, \ytableaushort{1 2}} = Q^{[2]}_{\tiny \ytableaushort{1 2}, \ytableaushort{1 2}}X_2 = Q^{[2]}_{\tiny \ytableaushort{1 2}, \ytableaushort{1 2}}, \quad X_2Q^{[1,1]}_{\tiny \ytableaushort{1,2}, \ytableaushort{1,2}} = Q^{[1,1]}_{\tiny \ytableaushort{1,2}, \ytableaushort{1,2}}X_2 = - Q^{[1,1]}_{\tiny \ytableaushort{1,2}, \ytableaushort{1,2}}.
\end{equation}
Indeed, these eigenvalues are the contents of corresponding to the standard tableaux
\begin{equation}
	c_2({\tiny \ytableaushort{1 2}}) = 1-0 = 1,\quad c_2({\tiny \ytableaushort{1,2}}) = 0-1= -1.
\end{equation}

\subsection{Example 2.} For $S_3$ we will not present the explicit diagonalization. Instead, we will give a table that shows that the ordered list of contents (eigenvalues of $X_i$) distinguishes between all standard Young tableaux for $S_3$. Each row in the table is a standard Young tableaux $a$ and each column corresponds to the contents $c_2(a), c_3(a)$.
\begin{equation}
	\begin{array}{l | l l}
		\text{SYT $a$}& c_2(a) & c_3(a)\\
		{\tiny \ytableaushort{1 2 3}} & 1  & 2 \\
		{\tiny \ytableaushort{1, 2, 3}} & -1 & -2 \\
	{	\tiny \ytableaushort{1 2, 3}} & 1 & -1 \\
		{\tiny \ytableaushort{1 3, 2}} & -1 & 1
	\end{array}
\end{equation}
Note that each row of numbers is distinct.

\chapter{Table of matrix units for $P_2(\N)$}\label{apx: P2N units}
In section \ref{sec: construction of units} we gave an all $\N$ construct of matrix units for $P_k(\N)$. Here we list the result of the Sage code implementing this procedure for $k=2$. The full sage code is explained in detail in Appendix \ref{apx: EV code}.
To make the equations more readable we use the following short hands for the relevant vacillating tableaux
\begin{align}
	&\vactab_1 = ([\N], [\N-1], [\N], [\N-1], [\N]),\\
	&\vactab_2 = ([\N], [\N-1], [\N-1,1], [\N-1], [\N]) \\
	&\vactab_3 = ([\N], [\N-1], [\N], [\N-1], [\N-1,1]), \\
	&\vactab_4 = ([\N], [\N-1], [\N-1,1], [\N-1], [\N-1,1]) \\
	&\vactab_5 = ([\N], [\N-1], [\N-1,1], [\N-2,1], [\N-1,1]), \\
	&\vactab_6 = ([\N], [\N-1], [\N-1,1], [\N-2,1], [\N-2,2]), \\
	&\vactab_7 = ([\N], [\N-1], [\N-1,1], [\N-2,1], [\N-2,1,1]).
\end{align}
For $\lambda = [\N]$ we have the following matrix units
\begin{align}
	&Q^{[\N]}_{\vactab_1 \vactab_1} = \frac{1}{N^{3}}\begin{tikzpicture}[scale = 0.25,baseline={(0,-1ex/2)}] 
		\tikzstyle{vertex} = [shape = circle, fill, minimum size = 1pt, inner sep = 1pt] 
		\node[vertex] (G--2) at (1.5, -1) [shape = circle, fill, draw] {}; 
		\node[vertex] (G--1) at (0.0, -1) [shape = circle, fill, draw] {}; 
		\node[vertex] (G-1) at (0.0, 1) [shape = circle, fill, draw] {}; 
		\node[vertex] (G-2) at (1.5, 1) [shape = circle, fill, draw] {}; 
	\end{tikzpicture}\\
	&Q^{[\N]}_{\vactab_1 \vactab_2} = -\frac{1}{N^{3}}\begin{tikzpicture}[scale = 0.25, baseline={(0,-1ex/2)}] 
		\tikzstyle{vertex} = [shape = circle, fill, minimum size = 1pt, inner sep = 1pt] 
		\node[vertex] (G--2) at (1.5, -1) [shape = circle, fill, draw] {}; 
		\node[vertex] (G--1) at (0.0, -1) [shape = circle, fill, draw] {}; 
		\node[vertex] (G-1) at (0.0, 1) [shape = circle, fill, draw] {}; 
		\node[vertex] (G-2) at (1.5, 1) [shape = circle, fill, draw] {}; 
	\end{tikzpicture} + \frac{1}{N^{2}}\begin{tikzpicture}[scale = 0.25, baseline={(0,-1ex/2)}] 
		\tikzstyle{vertex} = [shape = circle, fill, minimum size = 1pt, inner sep = 1pt] 
		\node[vertex] (G--2) at (1.5, -1) [shape = circle, fill, draw] {}; 
		\node[vertex] (G--1) at (0.0, -1) [shape = circle, fill, draw] {}; 
		\node[vertex] (G-1) at (0.0, 1) [shape = circle, fill, draw] {}; 
		\node[vertex] (G-2) at (1.5, 1) [shape = circle, fill, draw] {}; 
		\draw[] (G--2) .. controls +(-0.5, 0.5) and +(0.5, 0.5) .. (G--1); 
	\end{tikzpicture}\\
	&Q^{[\N]}_{\vactab_2 \vactab_1} = -\frac{1}{N^{3}}\begin{tikzpicture}[scale = 0.25, baseline={(0,-1ex/2)}] 
		\tikzstyle{vertex} = [shape = circle, fill, minimum size = 1pt, inner sep = 1pt] 
		\node[vertex] (G--2) at (1.5, -1) [shape = circle, draw] {}; 
		\node[vertex] (G--1) at (0.0, -1) [shape = circle, draw] {}; 
		\node[vertex] (G-1) at (0.0, 1) [shape = circle, draw] {}; 
		\node[vertex] (G-2) at (1.5, 1) [shape = circle, draw] {}; 
	\end{tikzpicture} + \frac{1}{N^{2}}\begin{tikzpicture}[scale = 0.25, baseline={(0,-1ex/2)}] 
		\tikzstyle{vertex} = [shape = circle, fill, minimum size = 1pt, inner sep = 1pt] 
		\node[vertex] (G--2) at (1.5, -1) [shape = circle, draw] {}; 
		\node[vertex] (G--1) at (0.0, -1) [shape = circle, draw] {}; 
		\node[vertex] (G-1) at (0.0, 1) [shape = circle, draw] {}; 
		\node[vertex] (G-2) at (1.5, 1) [shape = circle, draw] {}; 
		\draw[] (G-1) .. controls +(0.5, -0.5) and +(-0.5, -0.5) .. (G-2); 
	\end{tikzpicture} \\
	&Q^{[\N]}_{\vactab_2 \vactab_2} =
	\frac{1}{N^{3}}\begin{tikzpicture}[scale = 0.25  , baseline={(0,-1ex/2)}] 
		\tikzstyle{vertex} = [shape = circle, minimum size = 1pt, inner sep = 1pt] 
		\node[vertex] (G--2) at (1.5, -1) [shape = circle, fill, draw] {}; 
		\node[vertex] (G--1) at (0.0, -1) [shape = circle, fill, draw] {}; 
		\node[vertex] (G-1) at (0.0, 1) [shape = circle, fill, draw] {}; 
		\node[vertex] (G-2) at (1.5, 1) [shape = circle, fill, draw] {}; 
	\end{tikzpicture} - \frac{1}{N^{2}}\begin{tikzpicture}[scale = 0.25  , baseline={(0,-1ex/2)}] 
		\tikzstyle{vertex} = [shape = circle, fill, minimum size = 1pt, inner sep = 1pt] 
		\node[vertex] (G--2) at (1.5, -1) [shape = circle, fill, draw] {}; 
		\node[vertex] (G--1) at (0.0, -1) [shape = circle, fill, draw] {}; 
		\node[vertex] (G-1) at (0.0, 1) [shape = circle, fill, draw] {}; 
		\node[vertex] (G-2) at (1.5, 1) [shape = circle, fill, draw] {}; 
		\draw[] (G-1) .. controls +(0.5, -0.5) and +(-0.5, -0.5) .. (G-2); 
	\end{tikzpicture} - \frac{1}{N^{2}}\begin{tikzpicture}[scale = 0.25  , baseline={(0,-1ex/2)}] 
		\tikzstyle{vertex} = [shape = circle, fill, minimum size = 1pt, inner sep = 1pt] 
		\node[vertex] (G--2) at (1.5, -1) [shape = circle, fill, draw] {}; 
		\node[vertex] (G--1) at (0.0, -1) [shape = circle, fill, draw] {}; 
		\node[vertex] (G-1) at (0.0, 1) [shape = circle, fill, draw] {}; 
		\node[vertex] (G-2) at (1.5, 1) [shape = circle, fill, draw] {}; 
		\draw[] (G--2) .. controls +(-0.5, 0.5) and +(0.5, 0.5) .. (G--1); 
	\end{tikzpicture} + \frac{1}{N}\begin{tikzpicture}[scale = 0.25  , baseline={(0,-1ex/2)}] 
		\tikzstyle{vertex} = [shape = circle, fill, minimum size = 1pt, inner sep = 1pt] 
		\node[vertex] (G--2) at (1.5, -1) [shape = circle, fill, draw] {}; 
		\node[vertex] (G--1) at (0.0, -1) [shape = circle, fill, draw] {}; 
		\node[vertex] (G-1) at (0.0, 1) [shape = circle, fill, draw] {}; 
		\node[vertex] (G-2) at (1.5, 1) [shape = circle, fill, draw] {}; 
		\draw[] (G--2) .. controls +(-0.5, 0.5) and +(0.5, 0.5) .. (G--1); 
		\draw[] (G-1) .. controls +(0.5, -0.5) and +(-0.5, -0.5) .. (G-2); 
	\end{tikzpicture}
\end{align}

For $\lambda = [\N-1,1]$ we have nine matrix units. Since $(Q^\gamma_{\alpha \beta})^T = Q^\lambda_{\beta \alpha}$ we only give those that are not related by diagram transposition.
\begin{align}
	&Q^{[\N-1,1]}_{\vactab_3 \vactab_3} =
	-\frac{1}{N^{3}}\begin{tikzpicture}[scale = 0.25, , baseline={(0,-1ex/2)}] 
		\tikzstyle{vertex} = [shape = circle, fill, inner sep = 1pt] 
		\node[vertex] (G--2) at (1.5, -1) [shape = circle, fill, draw] {}; 
		\node[vertex] (G--1) at (0.0, -1) [shape = circle, fill, draw] {}; 
		\node[vertex] (G-1) at (0.0, 1) [shape = circle, fill, draw] {}; 
		\node[vertex] (G-2) at (1.5, 1) [shape = circle, fill, draw] {}; 
	\end{tikzpicture} + \frac{1}{N^{2}}\begin{tikzpicture}[scale = 0.25, , baseline={(0,-1ex/2)}] 
		\tikzstyle{vertex} = [shape = circle, fill, inner sep = 1pt] 
		\node[vertex] (G--2) at (1.5, -1) [shape = circle, fill, draw] {}; 
		\node[vertex] (G-2) at (1.5, 1) [shape = circle, fill, draw] {}; 
		\node[vertex] (G--1) at (0.0, -1) [shape = circle, fill, draw] {}; 
		\node[vertex] (G-1) at (0.0, 1) [shape = circle, fill, draw] {}; 
		\draw[] (G-2) .. controls +(0, -1) and +(0, 1) .. (G--2); 
	\end{tikzpicture} \\
	&Q^{[\N-1,1]}_{\vactab_5 \vactab_3} =
	-\frac{1}{{\left(N - 1\right)} N}\begin{tikzpicture}[scale = 0.25, , baseline={(0,-1ex/2)}] 
		\tikzstyle{vertex} = [shape = circle, fill, inner sep = 1pt] 
		\node[vertex] (G--2) at (1.5, -1) [shape = circle, fill, draw] {}; 
		\node[vertex] (G--1) at (0.0, -1) [shape = circle, fill, draw] {}; 
		\node[vertex] (G-1) at (0.0, 1) [shape = circle, fill, draw] {}; 
		\node[vertex] (G-2) at (1.5, 1) [shape = circle, fill, draw] {}; 
	\end{tikzpicture} + \frac{1}{{\left(N - 1\right)} N}\begin{tikzpicture}[scale = 0.25, , baseline={(0,-1ex/2)}] 
		\tikzstyle{vertex} = [shape = circle, fill, inner sep = 1pt] 
		\node[vertex] (G--2) at (1.5, -1) [shape = circle, fill, draw] {}; 
		\node[vertex] (G--1) at (0.0, -1) [shape = circle, fill, draw] {}; 
		\node[vertex] (G-1) at (0.0, 1) [shape = circle, fill, draw] {}; 
		\node[vertex] (G-2) at (1.5, 1) [shape = circle, fill, draw] {}; 
		\draw[] (G-1) .. controls +(0.5, -0.5) and +(-0.5, -0.5) .. (G-2); 
	\end{tikzpicture} + \frac{1}{N}\begin{tikzpicture}[scale = 0.25, , baseline={(0,-1ex/2)}] 
		\tikzstyle{vertex} = [shape = circle, fill, inner sep = 1pt] 
		\node[vertex] (G--2) at (1.5, -1) [shape = circle, fill, draw] {}; 
		\node[vertex] (G-1) at (0.0, 1) [shape = circle, fill, draw] {}; 
		\node[vertex] (G--1) at (0.0, -1) [shape = circle, fill, draw] {}; 
		\node[vertex] (G-2) at (1.5, 1) [shape = circle, fill, draw] {}; 
		\draw[] (G-1) .. controls +(0.75, -1) and +(-0.75, 1) .. (G--2); 
	\end{tikzpicture} + -\frac{1}{N - 1}\begin{tikzpicture}[scale = 0.25, , baseline={(0,-1ex/2)}] 
		\tikzstyle{vertex} = [shape = circle, fill, inner sep = 1pt] 
		\node[vertex] (G--2) at (1.5, -1) [shape = circle, fill, draw] {}; 
		\node[vertex] (G-1) at (0.0, 1) [shape = circle, fill, draw] {}; 
		\node[vertex] (G-2) at (1.5, 1) [shape = circle, fill, draw] {}; 
		\node[vertex] (G--1) at (0.0, -1) [shape = circle, fill, draw] {}; 
		\draw[] (G-1) .. controls +(0.5, -0.5) and +(-0.5, -0.5) .. (G-2); 
		\draw[] (G-2) .. controls +(0, -1) and +(0, 1) .. (G--2); 
		\draw[] (G--2) .. controls +(-0.75, 1) and +(0.75, -1) .. (G-1); 
	\end{tikzpicture} + \frac{1}{{\left(N - 1\right)} N}\begin{tikzpicture}[scale = 0.25, , baseline={(0,-1ex/2)}] 
		\tikzstyle{vertex} = [shape = circle, fill, inner sep = 1pt] 
		\node[vertex] (G--2) at (1.5, -1) [shape = circle, fill, draw] {}; 
		\node[vertex] (G-2) at (1.5, 1) [shape = circle, fill, draw] {}; 
		\node[vertex] (G--1) at (0.0, -1) [shape = circle, fill, draw] {}; 
		\node[vertex] (G-1) at (0.0, 1) [shape = circle, fill, draw] {}; 
		\draw[] (G-2) .. controls +(0, -1) and +(0, 1) .. (G--2); 
	\end{tikzpicture}\\
	&\begin{aligned}Q^{[\N-1,1]}_{\vactab_5 \vactab_5} &=
	-\frac{1}{{\left(N - 1\right)}^{2} {\left(N - 2\right)}}\begin{tikzpicture}[scale = 0.25, , baseline={(0,-1ex/2)}] 
		\tikzstyle{vertex} = [shape = circle, fill, inner sep = 1pt] 
		\node[vertex] (G--2) at (1.5, -1) [shape = circle, fill, draw] {}; 
		\node[vertex] (G--1) at (0.0, -1) [shape = circle, fill, draw] {}; 
		\node[vertex] (G-1) at (0.0, 1) [shape = circle, fill, draw] {}; 
		\node[vertex] (G-2) at (1.5, 1) [shape = circle, fill, draw] {}; 
	\end{tikzpicture} + \frac{1}{{\left(N - 1\right)}^{2} {\left(N - 2\right)}}\begin{tikzpicture}[scale = 0.25, , baseline={(0,-1ex/2)}] 
		\tikzstyle{vertex} = [shape = circle, fill, inner sep = 1pt] 
		\node[vertex] (G--2) at (1.5, -1) [shape = circle, fill, draw] {}; 
		\node[vertex] (G--1) at (0.0, -1) [shape = circle, fill, draw] {}; 
		\node[vertex] (G-1) at (0.0, 1) [shape = circle, fill, draw] {}; 
		\node[vertex] (G-2) at (1.5, 1) [shape = circle, fill, draw] {}; 
		\draw[] (G-1) .. controls +(0.5, -0.5) and +(-0.5, -0.5) .. (G-2); 
	\end{tikzpicture} + \frac{1}{{\left(N - 2\right)} N}\begin{tikzpicture}[scale = 0.25, , baseline={(0,-1ex/2)}] 
		\tikzstyle{vertex} = [shape = circle, fill, inner sep = 1pt] 
		\node[vertex] (G--2) at (1.5, -1) [shape = circle, fill, draw] {}; 
		\node[vertex] (G--1) at (0.0, -1) [shape = circle, fill, draw] {}; 
		\node[vertex] (G-1) at (0.0, 1) [shape = circle, fill, draw] {}; 
		\node[vertex] (G-2) at (1.5, 1) [shape = circle, fill, draw] {}; 
		\draw[] (G-1) .. controls +(0, -1) and +(0, 1) .. (G--1); 
	\end{tikzpicture} \\
&-\frac{1}{{\left(N - 1\right)} {\left(N - 2\right)}}\begin{tikzpicture}[scale = 0.25, , baseline={(0,-1ex/2)}] 
		\tikzstyle{vertex} = [shape = circle, fill, inner sep = 1pt] 
		\node[vertex] (G--2) at (1.5, -1) [shape = circle, fill, draw] {}; 
		\node[vertex] (G--1) at (0.0, -1) [shape = circle, fill, draw] {}; 
		\node[vertex] (G-1) at (0.0, 1) [shape = circle, fill, draw] {}; 
		\node[vertex] (G-2) at (1.5, 1) [shape = circle, fill, draw] {}; 
		\draw[] (G-1) .. controls +(0.5, -0.5) and +(-0.5, -0.5) .. (G-2); 
		\draw[] (G-2) .. controls +(-0.75, -1) and +(0.75, 1) .. (G--1); 
		\draw[] (G--1) .. controls +(0, 1) and +(0, -1) .. (G-1); 
	\end{tikzpicture} + \frac{1}{{\left(N - 1\right)} {\left(N - 2\right)} N}\begin{tikzpicture}[scale = 0.25, , baseline={(0,-1ex/2)}] 
		\tikzstyle{vertex} = [shape = circle, fill, inner sep = 1pt] 
		\node[vertex] (G--2) at (1.5, -1) [shape = circle, fill, draw] {}; 
		\node[vertex] (G--1) at (0.0, -1) [shape = circle, fill, draw] {}; 
		\node[vertex] (G-2) at (1.5, 1) [shape = circle, fill, draw] {}; 
		\node[vertex] (G-1) at (0.0, 1) [shape = circle, fill, draw] {}; 
		\draw[] (G-2) .. controls +(-0.75, -1) and +(0.75, 1) .. (G--1); 
	\end{tikzpicture} + \frac{1}{{\left(N - 1\right)}^{2} {\left(N - 2\right)}}\begin{tikzpicture}[scale = 0.25, , baseline={(0,-1ex/2)}] 
		\tikzstyle{vertex} = [shape = circle, fill, inner sep = 1pt] 
		\node[vertex] (G--2) at (1.5, -1) [shape = circle, fill, draw] {}; 
		\node[vertex] (G--1) at (0.0, -1) [shape = circle, fill, draw] {}; 
		\node[vertex] (G-1) at (0.0, 1) [shape = circle, fill, draw] {}; 
		\node[vertex] (G-2) at (1.5, 1) [shape = circle, fill, draw] {}; 
		\draw[] (G--2) .. controls +(-0.5, 0.5) and +(0.5, 0.5) .. (G--1); 
	\end{tikzpicture} \\
&-\frac{1}{{\left(N - 1\right)}^{2} {\left(N - 2\right)}}\begin{tikzpicture}[scale = 0.25, , baseline={(0,-1ex/2)}] 
		\tikzstyle{vertex} = [shape = circle, fill, inner sep = 1pt] 
		\node[vertex] (G--2) at (1.5, -1) [shape = circle, fill, draw] {}; 
		\node[vertex] (G--1) at (0.0, -1) [shape = circle, fill, draw] {}; 
		\node[vertex] (G-1) at (0.0, 1) [shape = circle, fill, draw] {}; 
		\node[vertex] (G-2) at (1.5, 1) [shape = circle, fill, draw] {}; 
		\draw[] (G--2) .. controls +(-0.5, 0.5) and +(0.5, 0.5) .. (G--1); 
		\draw[] (G-1) .. controls +(0.5, -0.5) and +(-0.5, -0.5) .. (G-2); 
	\end{tikzpicture} -\frac{1}{{\left(N - 1\right)} {\left(N - 2\right)}}\begin{tikzpicture}[scale = 0.25, , baseline={(0,-1ex/2)}] 
		\tikzstyle{vertex} = [shape = circle, fill, inner sep = 1pt] 
		\node[vertex] (G--2) at (1.5, -1) [shape = circle, fill, draw] {}; 
		\node[vertex] (G--1) at (0.0, -1) [shape = circle, fill, draw] {}; 
		\node[vertex] (G-1) at (0.0, 1) [shape = circle, fill, draw] {}; 
		\node[vertex] (G-2) at (1.5, 1) [shape = circle, fill, draw] {}; 
		\draw[] (G-1) .. controls +(0.75, -1) and +(-0.75, 1) .. (G--2); 
		\draw[] (G--2) .. controls +(-0.5, 0.5) and +(0.5, 0.5) .. (G--1); 
		\draw[] (G--1) .. controls +(0, 1) and +(0, -1) .. (G-1); 
	\end{tikzpicture} + \frac{N}{{\left(N - 1\right)}^{2} {\left(N - 2\right)}}\begin{tikzpicture}[scale = 0.25, , baseline={(0,-1ex/2)}] 
		\tikzstyle{vertex} = [shape = circle, fill, inner sep = 1pt] 
		\node[vertex] (G--2) at (1.5, -1) [shape = circle, fill, draw] {}; 
		\node[vertex] (G--1) at (0.0, -1) [shape = circle, fill, draw] {}; 
		\node[vertex] (G-1) at (0.0, 1) [shape = circle, fill, draw] {}; 
		\node[vertex] (G-2) at (1.5, 1) [shape = circle, fill, draw] {}; 
		\draw[] (G-1) .. controls +(0.5, -0.5) and +(-0.5, -0.5) .. (G-2); 
		\draw[] (G-2) .. controls +(0, -1) and +(0, 1) .. (G--2); 
		\draw[] (G--2) .. controls +(-0.5, 0.5) and +(0.5, 0.5) .. (G--1); 
		\draw[] (G--1) .. controls +(0, 1) and +(0, -1) .. (G-1); 
	\end{tikzpicture} \\
&-\frac{1}{{\left(N - 1\right)}^{2} {\left(N - 2\right)}}\begin{tikzpicture}[scale = 0.25, , baseline={(0,-1ex/2)}] 
		\tikzstyle{vertex} = [shape = circle, fill, inner sep = 1pt] 
		\node[vertex] (G--2) at (1.5, -1) [shape = circle, fill, draw] {}; 
		\node[vertex] (G--1) at (0.0, -1) [shape = circle, fill, draw] {}; 
		\node[vertex] (G-2) at (1.5, 1) [shape = circle, fill, draw] {}; 
		\node[vertex] (G-1) at (0.0, 1) [shape = circle, fill, draw] {}; 
		\draw[] (G-2) .. controls +(0, -1) and +(0, 1) .. (G--2); 
		\draw[] (G--2) .. controls +(-0.5, 0.5) and +(0.5, 0.5) .. (G--1); 
		\draw[] (G--1) .. controls +(0.75, 1) and +(-0.75, -1) .. (G-2); 
	\end{tikzpicture} + \frac{1}{{\left(N - 1\right)} {\left(N - 2\right)} N}\begin{tikzpicture}[scale = 0.25, , baseline={(0,-1ex/2)}] 
		\tikzstyle{vertex} = [shape = circle, fill, inner sep = 1pt] 
		\node[vertex] (G--2) at (1.5, -1) [shape = circle, fill, draw] {}; 
		\node[vertex] (G-1) at (0.0, 1) [shape = circle, fill, draw] {}; 
		\node[vertex] (G--1) at (0.0, -1) [shape = circle, fill, draw] {}; 
		\node[vertex] (G-2) at (1.5, 1) [shape = circle, fill, draw] {}; 
		\draw[] (G-1) .. controls +(0.75, -1) and +(-0.75, 1) .. (G--2); 
	\end{tikzpicture} -\frac{1}{{\left(N - 1\right)}^{2} {\left(N - 2\right)}}\begin{tikzpicture}[scale = 0.25, , baseline={(0,-1ex/2)}] 
		\tikzstyle{vertex} = [shape = circle, fill, inner sep = 1pt] 
		\node[vertex] (G--2) at (1.5, -1) [shape = circle, fill, draw] {}; 
		\node[vertex] (G-1) at (0.0, 1) [shape = circle, fill, draw] {}; 
		\node[vertex] (G-2) at (1.5, 1) [shape = circle, fill, draw] {}; 
		\node[vertex] (G--1) at (0.0, -1) [shape = circle, fill, draw] {}; 
		\draw[] (G-1) .. controls +(0.5, -0.5) and +(-0.5, -0.5) .. (G-2); 
		\draw[] (G-2) .. controls +(0, -1) and +(0, 1) .. (G--2); 
		\draw[] (G--2) .. controls +(-0.75, 1) and +(0.75, -1) .. (G-1); 
	\end{tikzpicture} \\
&+ \frac{1}{{\left(N - 1\right)}^{2} {\left(N - 2\right)} N}\begin{tikzpicture}[scale = 0.25, , baseline={(0,-1ex/2)}] 
		\tikzstyle{vertex} = [shape = circle, fill, inner sep = 1pt] 
		\node[vertex] (G--2) at (1.5, -1) [shape = circle, fill, draw] {}; 
		\node[vertex] (G-2) at (1.5, 1) [shape = circle, fill, draw] {}; 
		\node[vertex] (G--1) at (0.0, -1) [shape = circle, fill, draw] {}; 
		\node[vertex] (G-1) at (0.0, 1) [shape = circle, fill, draw] {}; 
		\draw[] (G-2) .. controls +(0, -1) and +(0, 1) .. (G--2); 
	\end{tikzpicture}
\end{aligned} \\
	&Q^{[\N-1,1]}_{\vactab_4 \vactab_3} =
	\frac{1}{N^{3}}\begin{tikzpicture}[scale = 0.25, , baseline={(0,-1ex/2)}] 
		\tikzstyle{vertex} = [shape = circle, fill, inner sep = 1pt] 
		\node[vertex] (G--2) at (1.5, -1) [shape = circle, fill, draw] {}; 
		\node[vertex] (G--1) at (0.0, -1) [shape = circle, fill, draw] {}; 
		\node[vertex] (G-1) at (0.0, 1) [shape = circle, fill, draw] {}; 
		\node[vertex] (G-2) at (1.5, 1) [shape = circle, fill, draw] {}; 
	\end{tikzpicture} - \frac{1}{N^{2}}\begin{tikzpicture}[scale = 0.25, , baseline={(0,-1ex/2)}] 
		\tikzstyle{vertex} = [shape = circle, fill, inner sep = 1pt] 
		\node[vertex] (G--2) at (1.5, -1) [shape = circle, fill, draw] {}; 
		\node[vertex] (G--1) at (0.0, -1) [shape = circle, fill, draw] {}; 
		\node[vertex] (G-1) at (0.0, 1) [shape = circle, fill, draw] {}; 
		\node[vertex] (G-2) at (1.5, 1) [shape = circle, fill, draw] {}; 
		\draw[] (G-1) .. controls +(0.5, -0.5) and +(-0.5, -0.5) .. (G-2); 
	\end{tikzpicture} + \frac{1}{N}\begin{tikzpicture}[scale = 0.25, , baseline={(0,-1ex/2)}] 
		\tikzstyle{vertex} = [shape = circle, fill, inner sep = 1pt] 
		\node[vertex] (G--2) at (1.5, -1) [shape = circle, fill, draw] {}; 
		\node[vertex] (G-1) at (0.0, 1) [shape = circle, fill, draw] {}; 
		\node[vertex] (G-2) at (1.5, 1) [shape = circle, fill, draw] {}; 
		\node[vertex] (G--1) at (0.0, -1) [shape = circle, fill, draw] {}; 
		\draw[] (G-1) .. controls +(0.5, -0.5) and +(-0.5, -0.5) .. (G-2); 
		\draw[] (G-2) .. controls +(0, -1) and +(0, 1) .. (G--2); 
		\draw[] (G--2) .. controls +(-0.75, 1) and +(0.75, -1) .. (G-1); 
	\end{tikzpicture} - \frac{1}{N^{2}}\begin{tikzpicture}[scale = 0.25, , baseline={(0,-1ex/2)}] 
		\tikzstyle{vertex} = [shape = circle, fill, inner sep = 1pt] 
		\node[vertex] (G--2) at (1.5, -1) [shape = circle, fill, draw] {}; 
		\node[vertex] (G-2) at (1.5, 1) [shape = circle, fill, draw] {}; 
		\node[vertex] (G--1) at (0.0, -1) [shape = circle, fill, draw] {}; 
		\node[vertex] (G-1) at (0.0, 1) [shape = circle, fill, draw] {}; 
		\draw[] (G-2) .. controls +(0, -1) and +(0, 1) .. (G--2); 
	\end{tikzpicture}\\
	&\begin{aligned}Q^{[\N-1,1]}_{\vactab_4 \vactab_5} &=
	\frac{1}{{\left(N - 1\right)} N}\begin{tikzpicture}[scale = 0.25, , baseline={(0,-1ex/2)}] 
		\tikzstyle{vertex} = [shape = circle, fill, inner sep = 1pt] 
		\node[vertex] (G--2) at (1.5, -1) [shape = circle, fill, draw] {}; 
		\node[vertex] (G--1) at (0.0, -1) [shape = circle, fill, draw] {}; 
		\node[vertex] (G-1) at (0.0, 1) [shape = circle, fill, draw] {}; 
		\node[vertex] (G-2) at (1.5, 1) [shape = circle, fill, draw] {}; 
	\end{tikzpicture} -\frac{1}{N - 1}\begin{tikzpicture}[scale = 0.25, , baseline={(0,-1ex/2)}] 
		\tikzstyle{vertex} = [shape = circle, fill, inner sep = 1pt] 
		\node[vertex] (G--2) at (1.5, -1) [shape = circle, fill, draw] {}; 
		\node[vertex] (G--1) at (0.0, -1) [shape = circle, fill, draw] {}; 
		\node[vertex] (G-1) at (0.0, 1) [shape = circle, fill, draw] {}; 
		\node[vertex] (G-2) at (1.5, 1) [shape = circle, fill, draw] {}; 
		\draw[] (G-1) .. controls +(0.5, -0.5) and +(-0.5, -0.5) .. (G-2); 
	\end{tikzpicture} + \begin{tikzpicture}[scale = 0.25, , baseline={(0,-1ex/2)}] 
		\tikzstyle{vertex} = [shape = circle, fill, inner sep = 1pt] 
		\node[vertex] (G--2) at (1.5, -1) [shape = circle, fill, draw] {}; 
		\node[vertex] (G--1) at (0.0, -1) [shape = circle, fill, draw] {}; 
		\node[vertex] (G-1) at (0.0, 1) [shape = circle, fill, draw] {}; 
		\node[vertex] (G-2) at (1.5, 1) [shape = circle, fill, draw] {}; 
		\draw[] (G-1) .. controls +(0.5, -0.5) and +(-0.5, -0.5) .. (G-2); 
		\draw[] (G-2) .. controls +(-0.75, -1) and +(0.75, 1) .. (G--1); 
		\draw[] (G--1) .. controls +(0, 1) and +(0, -1) .. (G-1); 
	\end{tikzpicture} \\
&-\frac{1}{N}\begin{tikzpicture}[scale = 0.25, , baseline={(0,-1ex/2)}] 
		\tikzstyle{vertex} = [shape = circle, fill, inner sep = 1pt] 
		\node[vertex] (G--2) at (1.5, -1) [shape = circle, fill, draw] {}; 
		\node[vertex] (G--1) at (0.0, -1) [shape = circle, fill, draw] {}; 
		\node[vertex] (G-2) at (1.5, 1) [shape = circle, fill, draw] {}; 
		\node[vertex] (G-1) at (0.0, 1) [shape = circle, fill, draw] {}; 
		\draw[] (G-2) .. controls +(-0.75, -1) and +(0.75, 1) .. (G--1); 
	\end{tikzpicture} -\frac{1}{{\left(N - 1\right)} N}\begin{tikzpicture}[scale = 0.25, , baseline={(0,-1ex/2)}] 
		\tikzstyle{vertex} = [shape = circle, fill, inner sep = 1pt] 
		\node[vertex] (G--2) at (1.5, -1) [shape = circle, fill, draw] {}; 
		\node[vertex] (G--1) at (0.0, -1) [shape = circle, fill, draw] {}; 
		\node[vertex] (G-1) at (0.0, 1) [shape = circle, fill, draw] {}; 
		\node[vertex] (G-2) at (1.5, 1) [shape = circle, fill, draw] {}; 
		\draw[] (G--2) .. controls +(-0.5, 0.5) and +(0.5, 0.5) .. (G--1); 
	\end{tikzpicture} + \frac{1}{N - 1}\begin{tikzpicture}[scale = 0.25, , baseline={(0,-1ex/2)}] 
		\tikzstyle{vertex} = [shape = circle, fill, inner sep = 1pt] 
		\node[vertex] (G--2) at (1.5, -1) [shape = circle, fill, draw] {}; 
		\node[vertex] (G--1) at (0.0, -1) [shape = circle, fill, draw] {}; 
		\node[vertex] (G-1) at (0.0, 1) [shape = circle, fill, draw] {}; 
		\node[vertex] (G-2) at (1.5, 1) [shape = circle, fill, draw] {}; 
		\draw[] (G--2) .. controls +(-0.5, 0.5) and +(0.5, 0.5) .. (G--1); 
		\draw[] (G-1) .. controls +(0.5, -0.5) and +(-0.5, -0.5) .. (G-2); 
	\end{tikzpicture} \\
&-\frac{N}{N - 1}\begin{tikzpicture}[scale = 0.25, , baseline={(0,-1ex/2)}] 
		\tikzstyle{vertex} = [shape = circle, fill, inner sep = 1pt] 
		\node[vertex] (G--2) at (1.5, -1) [shape = circle, fill, draw] {}; 
		\node[vertex] (G--1) at (0.0, -1) [shape = circle, fill, draw] {}; 
		\node[vertex] (G-1) at (0.0, 1) [shape = circle, fill, draw] {}; 
		\node[vertex] (G-2) at (1.5, 1) [shape = circle, fill, draw] {}; 
		\draw[] (G-1) .. controls +(0.5, -0.5) and +(-0.5, -0.5) .. (G-2); 
		\draw[] (G-2) .. controls +(0, -1) and +(0, 1) .. (G--2); 
		\draw[] (G--2) .. controls +(-0.5, 0.5) and +(0.5, 0.5) .. (G--1); 
		\draw[] (G--1) .. controls +(0, 1) and +(0, -1) .. (G-1); 
	\end{tikzpicture} + \frac{1}{N - 1}\begin{tikzpicture}[scale = 0.25, , baseline={(0,-1ex/2)}] 
		\tikzstyle{vertex} = [shape = circle, fill, inner sep = 1pt] 
		\node[vertex] (G--2) at (1.5, -1) [shape = circle, fill, draw] {}; 
		\node[vertex] (G--1) at (0.0, -1) [shape = circle, fill, draw] {}; 
		\node[vertex] (G-2) at (1.5, 1) [shape = circle, fill, draw] {}; 
		\node[vertex] (G-1) at (0.0, 1) [shape = circle, fill, draw] {}; 
		\draw[] (G-2) .. controls +(0, -1) and +(0, 1) .. (G--2); 
		\draw[] (G--2) .. controls +(-0.5, 0.5) and +(0.5, 0.5) .. (G--1); 
		\draw[] (G--1) .. controls +(0.75, 1) and +(-0.75, -1) .. (G-2); 
	\end{tikzpicture} + \frac{1}{N - 1}\begin{tikzpicture}[scale = 0.25, , baseline={(0,-1ex/2)}] 
		\tikzstyle{vertex} = [shape = circle, fill, inner sep = 1pt] 
		\node[vertex] (G--2) at (1.5, -1) [shape = circle, fill, draw] {}; 
		\node[vertex] (G-1) at (0.0, 1) [shape = circle, fill, draw] {}; 
		\node[vertex] (G-2) at (1.5, 1) [shape = circle, fill, draw] {}; 
		\node[vertex] (G--1) at (0.0, -1) [shape = circle, fill, draw] {}; 
		\draw[] (G-1) .. controls +(0.5, -0.5) and +(-0.5, -0.5) .. (G-2); 
		\draw[] (G-2) .. controls +(0, -1) and +(0, 1) .. (G--2); 
		\draw[] (G--2) .. controls +(-0.75, 1) and +(0.75, -1) .. (G-1); 
	\end{tikzpicture} \\
&-\frac{1}{{\left(N - 1\right)} N}\begin{tikzpicture}[scale = 0.25, , baseline={(0,-1ex/2)}] 
		\tikzstyle{vertex} = [shape = circle, fill, inner sep = 1pt] 
		\node[vertex] (G--2) at (1.5, -1) [shape = circle, fill, draw] {}; 
		\node[vertex] (G-2) at (1.5, 1) [shape = circle, fill, draw] {}; 
		\node[vertex] (G--1) at (0.0, -1) [shape = circle, fill, draw] {}; 
		\node[vertex] (G-1) at (0.0, 1) [shape = circle, fill, draw] {}; 
		\draw[] (G-2) .. controls +(0, -1) and +(0, 1) .. (G--2); 
	\end{tikzpicture}
	\end{aligned}\\
	&\begin{aligned}Q^{[\N-1,1]}_{\vactab_4 \vactab_4} &=
	-\frac{1}{N^{3}}\begin{tikzpicture}[scale = 0.25, , baseline={(0,-1ex/2)}] 
		\tikzstyle{vertex} = [shape = circle, fill, inner sep = 1pt] 
		\node[vertex] (G--2) at (1.5, -1) [shape = circle, fill, draw] {}; 
		\node[vertex] (G--1) at (0.0, -1) [shape = circle, fill, draw] {}; 
		\node[vertex] (G-1) at (0.0, 1) [shape = circle, fill, draw] {}; 
		\node[vertex] (G-2) at (1.5, 1) [shape = circle, fill, draw] {}; 
	\end{tikzpicture} + \frac{1}{N^{2}}\begin{tikzpicture}[scale = 0.25, , baseline={(0,-1ex/2)}] 
		\tikzstyle{vertex} = [shape = circle, fill, inner sep = 1pt] 
		\node[vertex] (G--2) at (1.5, -1) [shape = circle, fill, draw] {}; 
		\node[vertex] (G--1) at (0.0, -1) [shape = circle, fill, draw] {}; 
		\node[vertex] (G-1) at (0.0, 1) [shape = circle, fill, draw] {}; 
		\node[vertex] (G-2) at (1.5, 1) [shape = circle, fill, draw] {}; 
		\draw[] (G-1) .. controls +(0.5, -0.5) and +(-0.5, -0.5) .. (G-2); 
	\end{tikzpicture} + \frac{1}{N^{2}}\begin{tikzpicture}[scale = 0.25,  baseline={(0,-1ex/2)}] 
		\tikzstyle{vertex} = [shape = circle, fill, inner sep = 1pt] 
		\node[vertex] (G--2) at (1.5, -1) [shape = circle, fill, draw] {}; 
		\node[vertex] (G--1) at (0.0, -1) [shape = circle, fill, draw] {}; 
		\node[vertex] (G-1) at (0.0, 1) [shape = circle, fill, draw] {}; 
		\node[vertex] (G-2) at (1.5, 1) [shape = circle, fill, draw] {}; 
		\draw[] (G--2) .. controls +(-0.5, 0.5) and +(0.5, 0.5) .. (G--1); 
	\end{tikzpicture} -\frac{1}{N}\begin{tikzpicture}[scale = 0.25,  baseline={(0,-1ex/2)}] 
		\tikzstyle{vertex} = [shape = circle, fill, inner sep = 1pt] 
		\node[vertex] (G--2) at (1.5, -1) [shape = circle, fill, draw] {}; 
		\node[vertex] (G--1) at (0.0, -1) [shape = circle, fill, draw] {}; 
		\node[vertex] (G-1) at (0.0, 1) [shape = circle, fill, draw] {}; 
		\node[vertex] (G-2) at (1.5, 1) [shape = circle, fill, draw] {}; 
		\draw[] (G--2) .. controls +(-0.5, 0.5) and +(0.5, 0.5) .. (G--1); 
		\draw[] (G-1) .. controls +(0.5, -0.5) and +(-0.5, -0.5) .. (G-2); 
	\end{tikzpicture} \\
&+ \begin{tikzpicture}[scale = 0.25,  baseline={(0,-1ex/2)}] 
		\tikzstyle{vertex} = [shape = circle, fill, inner sep = 1pt] 
		\node[vertex] (G--2) at (1.5, -1) [shape = circle, fill, draw] {}; 
		\node[vertex] (G--1) at (0.0, -1) [shape = circle, fill, draw] {}; 
		\node[vertex] (G-1) at (0.0, 1) [shape = circle, fill, draw] {}; 
		\node[vertex] (G-2) at (1.5, 1) [shape = circle, fill, draw] {}; 
		\draw[] (G-1) .. controls +(0.5, -0.5) and +(-0.5, -0.5) .. (G-2); 
		\draw[] (G-2) .. controls +(0, -1) and +(0, 1) .. (G--2); 
		\draw[] (G--2) .. controls +(-0.5, 0.5) and +(0.5, 0.5) .. (G--1); 
		\draw[] (G--1) .. controls +(0, 1) and +(0, -1) .. (G-1); 
	\end{tikzpicture} -\frac{1}{N}\begin{tikzpicture}[scale = 0.25,  baseline={(0,-1ex/2)}] 
		\tikzstyle{vertex} = [shape = circle, fill, inner sep = 1pt] 
		\node[vertex] (G--2) at (1.5, -1) [shape = circle, fill, draw] {}; 
		\node[vertex] (G--1) at (0.0, -1) [shape = circle, fill, draw] {}; 
		\node[vertex] (G-2) at (1.5, 1) [shape = circle, fill, draw] {}; 
		\node[vertex] (G-1) at (0.0, 1) [shape = circle, fill, draw] {}; 
		\draw[] (G-2) .. controls +(0, -1) and +(0, 1) .. (G--2); 
		\draw[] (G--2) .. controls +(-0.5, 0.5) and +(0.5, 0.5) .. (G--1); 
		\draw[] (G--1) .. controls +(0.75, 1) and +(-0.75, -1) .. (G-2); 
	\end{tikzpicture} -\frac{1}{N}\begin{tikzpicture}[scale = 0.25,  baseline={(0,-1ex/2)}] 
		\tikzstyle{vertex} = [shape = circle, fill, inner sep = 1pt] 
		\node[vertex] (G--2) at (1.5, -1) [shape = circle, fill, draw] {}; 
		\node[vertex] (G-1) at (0.0, 1) [shape = circle, fill, draw] {}; 
		\node[vertex] (G-2) at (1.5, 1) [shape = circle, fill, draw] {}; 
		\node[vertex] (G--1) at (0.0, -1) [shape = circle, fill, draw] {}; 
		\draw[] (G-1) .. controls +(0.5, -0.5) and +(-0.5, -0.5) .. (G-2); 
		\draw[] (G-2) .. controls +(0, -1) and +(0, 1) .. (G--2); 
		\draw[] (G--2) .. controls +(-0.75, 1) and +(0.75, -1) .. (G-1); 
	\end{tikzpicture} + \frac{1}{N^{2}}\begin{tikzpicture}[scale = 0.25,  baseline={(0,-1ex/2)}] 
		\tikzstyle{vertex} = [shape = circle, fill, inner sep = 1pt] 
		\node[vertex] (G--2) at (1.5, -1) [shape = circle, fill, draw] {}; 
		\node[vertex] (G-2) at (1.5, 1) [shape = circle, fill, draw] {}; 
		\node[vertex] (G--1) at (0.0, -1) [shape = circle, fill, draw] {}; 
		\node[vertex] (G-1) at (0.0, 1) [shape = circle, fill, draw] {}; 
		\draw[] (G-2) .. controls +(0, -1) and +(0, 1) .. (G--2); 
	\end{tikzpicture}
	\end{aligned}
\end{align}
For $\lambda = [\N-2,2]$ there is only one matrix unit.
\begin{align}
	Q^{[\N-2,2]}_{\vactab_6 \vactab_6} &=  \frac{1}{{ (N - 1) } {( N - 2) }} \begin{tikzpicture}[scale = 0.25,  baseline={(0,-1ex/2)}] 
		\tikzstyle{vertex} = [shape = circle, fill,   inner sep = 1pt] 
		\node[vertex] (G--2) at (1.5, -1) [shape = circle, fill, draw] {}; 
		\node[vertex] (G--1) at (0.0, -1) [shape = circle, fill, draw] {}; 
		\node[vertex] (G-1) at (0.0, 1) [shape = circle, fill, draw] {}; 
		\node[vertex] (G-2) at (1.5, 1) [shape = circle, fill, draw] {}; 
	\end{tikzpicture}  -\frac{1}{{ (N - 1) } { (N - 2 )}} \begin{tikzpicture}[scale = 0.25,  baseline={(0,-1ex/2)}] 
		\tikzstyle{vertex} = [shape = circle, fill,   inner sep = 1pt] 
		\node[vertex] (G--2) at (1.5, -1) [shape = circle, fill, draw] {}; 
		\node[vertex] (G--1) at (0.0, -1) [shape = circle, fill, draw] {}; 
		\node[vertex] (G-1) at (0.0, 1) [shape = circle, fill, draw] {}; 
		\node[vertex] (G-2) at (1.5, 1) [shape = circle, fill, draw] {}; 
		\draw[] (G-1) .. controls +(0.5, -0.5) and +(-0.5, -0.5) .. (G-2); 
	\end{tikzpicture} -\frac{1}{2 \, { (N - 2 )}} \begin{tikzpicture}[scale = 0.25,  baseline={(0,-1ex/2)}] 
		\tikzstyle{vertex} = [shape = circle, fill,   inner sep = 1pt] 
		\node[vertex] (G--2) at (1.5, -1) [shape = circle, fill, draw] {}; 
		\node[vertex] (G--1) at (0.0, -1) [shape = circle, fill, draw] {}; 
		\node[vertex] (G-1) at (0.0, 1) [shape = circle, fill, draw] {}; 
		\node[vertex] (G-2) at (1.5, 1) [shape = circle, fill, draw] {}; 
		\draw[] (G-1) .. controls +(0, -1) and +(0, 1) .. (G--1); 
	\end{tikzpicture} \\
&+  \frac{1}{N - 2} \begin{tikzpicture}[scale = 0.25,  baseline={(0,-1ex/2)}] 
		\tikzstyle{vertex} = [shape = circle, fill,   inner sep = 1pt] 
		\node[vertex] (G--2) at (1.5, -1) [shape = circle, fill, draw] {}; 
		\node[vertex] (G--1) at (0.0, -1) [shape = circle, fill, draw] {}; 
		\node[vertex] (G-1) at (0.0, 1) [shape = circle, fill, draw] {}; 
		\node[vertex] (G-2) at (1.5, 1) [shape = circle, fill, draw] {}; 
		\draw[] (G-1) .. controls +(0.5, -0.5) and +(-0.5, -0.5) .. (G-2); 
		\draw[] (G-2) .. controls +(-0.75, -1) and +(0.75, 1) .. (G--1); 
		\draw[] (G--1) .. controls +(0, 1) and +(0, -1) .. (G-1); 
	\end{tikzpicture} -\frac{1}{2 \, {( N - 2) }} \begin{tikzpicture}[scale = 0.25,  baseline={(0,-1ex/2)}] 
		\tikzstyle{vertex} = [shape = circle, fill,   inner sep = 1pt] 
		\node[vertex] (G--2) at (1.5, -1) [shape = circle, fill, draw] {}; 
		\node[vertex] (G--1) at (0.0, -1) [shape = circle, fill, draw] {}; 
		\node[vertex] (G-2) at (1.5, 1) [shape = circle, fill, draw] {}; 
		\node[vertex] (G-1) at (0.0, 1) [shape = circle, fill, draw] {}; 
		\draw[] (G-2) .. controls +(-0.75, -1) and +(0.75, 1) .. (G--1); 
	\end{tikzpicture} -\frac{1}{{ (N - 1) } {( N - 2) }} \begin{tikzpicture}[scale = 0.25,  baseline={(0,-1ex/2)}] 
		\tikzstyle{vertex} = [shape = circle, fill,   inner sep = 1pt] 
		\node[vertex] (G--2) at (1.5, -1) [shape = circle, fill, draw] {}; 
		\node[vertex] (G--1) at (0.0, -1) [shape = circle, fill, draw] {}; 
		\node[vertex] (G-1) at (0.0, 1) [shape = circle, fill, draw] {}; 
		\node[vertex] (G-2) at (1.5, 1) [shape = circle, fill, draw] {}; 
		\draw[] (G--2) .. controls +(-0.5, 0.5) and +(0.5, 0.5) .. (G--1); 
	\end{tikzpicture} \\
&+  \frac{1}{{( N - 1 )} { (N - 2) }} \begin{tikzpicture}[scale = 0.25,  baseline={(0,-1ex/2)}] 
		\tikzstyle{vertex} = [shape = circle, fill,   inner sep = 1pt] 
		\node[vertex] (G--2) at (1.5, -1) [shape = circle, fill, draw] {}; 
		\node[vertex] (G--1) at (0.0, -1) [shape = circle, fill, draw] {}; 
		\node[vertex] (G-1) at (0.0, 1) [shape = circle, fill, draw] {}; 
		\node[vertex] (G-2) at (1.5, 1) [shape = circle, fill, draw] {}; 
		\draw[] (G--2) .. controls +(-0.5, 0.5) and +(0.5, 0.5) .. (G--1); 
		\draw[] (G-1) .. controls +(0.5, -0.5) and +(-0.5, -0.5) .. (G-2); 
	\end{tikzpicture} +  \frac{1}{N - 2} \begin{tikzpicture}[scale = 0.25,  baseline={(0,-1ex/2)}] 
		\tikzstyle{vertex} = [shape = circle, fill,   inner sep = 1pt] 
		\node[vertex] (G--2) at (1.5, -1) [shape = circle, fill, draw] {}; 
		\node[vertex] (G--1) at (0.0, -1) [shape = circle, fill, draw] {}; 
		\node[vertex] (G-1) at (0.0, 1) [shape = circle, fill, draw] {}; 
		\node[vertex] (G-2) at (1.5, 1) [shape = circle, fill, draw] {}; 
		\draw[] (G-1) .. controls +(0.75, -1) and +(-0.75, 1) .. (G--2); 
		\draw[] (G--2) .. controls +(-0.5, 0.5) and +(0.5, 0.5) .. (G--1); 
		\draw[] (G--1) .. controls +(0, 1) and +(0, -1) .. (G-1); 
	\end{tikzpicture} -\frac{N}{N - 2} \begin{tikzpicture}[scale = 0.25,  baseline={(0,-1ex/2)}] 
		\tikzstyle{vertex} = [shape = circle, fill,   inner sep = 1pt] 
		\node[vertex] (G--2) at (1.5, -1) [shape = circle, fill, draw] {}; 
		\node[vertex] (G--1) at (0.0, -1) [shape = circle, fill, draw] {}; 
		\node[vertex] (G-1) at (0.0, 1) [shape = circle, fill, draw] {}; 
		\node[vertex] (G-2) at (1.5, 1) [shape = circle, fill, draw] {}; 
		\draw[] (G-1) .. controls +(0.5, -0.5) and +(-0.5, -0.5) .. (G-2); 
		\draw[] (G-2) .. controls +(0, -1) and +(0, 1) .. (G--2); 
		\draw[] (G--2) .. controls +(-0.5, 0.5) and +(0.5, 0.5) .. (G--1); 
		\draw[] (G--1) .. controls +(0, 1) and +(0, -1) .. (G-1); 
	\end{tikzpicture} \\
&+  \frac{1}{N - 2} \begin{tikzpicture}[scale = 0.25,  baseline={(0,-1ex/2)}] 
		\tikzstyle{vertex} = [shape = circle, fill,   inner sep = 1pt] 
		\node[vertex] (G--2) at (1.5, -1) [shape = circle, fill, draw] {}; 
		\node[vertex] (G--1) at (0.0, -1) [shape = circle, fill, draw] {}; 
		\node[vertex] (G-2) at (1.5, 1) [shape = circle, fill, draw] {}; 
		\node[vertex] (G-1) at (0.0, 1) [shape = circle, fill, draw] {}; 
		\draw[] (G-2) .. controls +(0, -1) and +(0, 1) .. (G--2); 
		\draw[] (G--2) .. controls +(-0.5, 0.5) and +(0.5, 0.5) .. (G--1); 
		\draw[] (G--1) .. controls +(0.75, 1) and +(-0.75, -1) .. (G-2); 
	\end{tikzpicture} -\frac{1}{2 \, { (N - 2) }} \begin{tikzpicture}[scale = 0.25,  baseline={(0,-1ex/2)}] 
		\tikzstyle{vertex} = [shape = circle, fill,   inner sep = 1pt] 
		\node[vertex] (G--2) at (1.5, -1) [shape = circle, fill, draw] {}; 
		\node[vertex] (G-1) at (0.0, 1) [shape = circle, fill, draw] {}; 
		\node[vertex] (G--1) at (0.0, -1) [shape = circle, fill, draw] {}; 
		\node[vertex] (G-2) at (1.5, 1) [shape = circle, fill, draw] {}; 
		\draw[] (G-1) .. controls +(0.75, -1) and +(-0.75, 1) .. (G--2); 
	\end{tikzpicture} + \frac{1}{2}\begin{tikzpicture}[scale = 0.25,  baseline={(0,-1ex/2)}] 
		\tikzstyle{vertex} = [shape = circle, fill,   inner sep = 1pt] 
		\node[vertex] (G--2) at (1.5, -1) [shape = circle, fill, draw] {}; 
		\node[vertex] (G-1) at (0.0, 1) [shape = circle, fill, draw] {}; 
		\node[vertex] (G--1) at (0.0, -1) [shape = circle, fill, draw] {}; 
		\node[vertex] (G-2) at (1.5, 1) [shape = circle, fill, draw] {}; 
		\draw[] (G-1) .. controls +(0.75, -1) and +(-0.75, 1) .. (G--2); 
		\draw[] (G-2) .. controls +(-0.75, -1) and +(0.75, 1) .. (G--1); 
	\end{tikzpicture} \\
&+  \frac{1}{N - 2} \begin{tikzpicture}[scale = 0.25,  baseline={(0,-1ex/2)}] 
		\tikzstyle{vertex} = [shape = circle, fill,   inner sep = 1pt] 
		\node[vertex] (G--2) at (1.5, -1) [shape = circle, fill, draw] {}; 
		\node[vertex] (G-1) at (0.0, 1) [shape = circle, fill, draw] {}; 
		\node[vertex] (G-2) at (1.5, 1) [shape = circle, fill, draw] {}; 
		\node[vertex] (G--1) at (0.0, -1) [shape = circle, fill, draw] {}; 
		\draw[] (G-1) .. controls +(0.5, -0.5) and +(-0.5, -0.5) .. (G-2); 
		\draw[] (G-2) .. controls +(0, -1) and +(0, 1) .. (G--2); 
		\draw[] (G--2) .. controls +(-0.75, 1) and +(0.75, -1) .. (G-1); 
	\end{tikzpicture} -\frac{1}{2 \, { (N - 2) }} \begin{tikzpicture}[scale = 0.25,  baseline={(0,-1ex/2)}] 
		\tikzstyle{vertex} = [shape = circle, fill,   inner sep = 1pt] 
		\node[vertex] (G--2) at (1.5, -1) [shape = circle, fill, draw] {}; 
		\node[vertex] (G-2) at (1.5, 1) [shape = circle, fill, draw] {}; 
		\node[vertex] (G--1) at (0.0, -1) [shape = circle, fill, draw] {}; 
		\node[vertex] (G-1) at (0.0, 1) [shape = circle, fill, draw] {}; 
		\draw[] (G-2) .. controls +(0, -1) and +(0, 1) .. (G--2); 
	\end{tikzpicture} + \frac{1}{2}\begin{tikzpicture}[scale = 0.25,  baseline={(0,-1ex/2)}] 
		\tikzstyle{vertex} = [shape = circle, fill,   inner sep = 1pt] 
		\node[vertex] (G--2) at (1.5, -1) [shape = circle, fill, draw] {}; 
		\node[vertex] (G-2) at (1.5, 1) [shape = circle, fill, draw] {}; 
		\node[vertex] (G--1) at (0.0, -1) [shape = circle, fill, draw] {}; 
		\node[vertex] (G-1) at (0.0, 1) [shape = circle, fill, draw] {}; 
		\draw[] (G-2) .. controls +(0, -1) and +(0, 1) .. (G--2); 
		\draw[] (G-1) .. controls +(0, -1) and +(0, 1) .. (G--1); 
	\end{tikzpicture}
\end{align}
For $\lambda = [\N-2,1,1]$ there is just one matrix unit as well
\begin{equation}
	Q^{[\N-2,1,1]}_{\vactab_7 \vactab_7} = 
	\frac{1}{2\N}\begin{tikzpicture}[scale = 0.25,  baseline={(0,-1ex/2)}] 
		\tikzstyle{vertex} = [shape = circle, fill,inner sep = 1pt] 
		\node[vertex] (G--2) at (1.5, -1) [shape = circle, fill, draw] {}; 
		\node[vertex] (G--1) at (0.0, -1) [shape = circle, fill, draw] {}; 
		\node[vertex] (G-1) at (0.0, 1) [shape = circle, fill, draw] {}; 
		\node[vertex] (G-2) at (1.5, 1) [shape = circle, fill, draw] {}; 
		\draw[] (G-1) .. controls +(0, -1) and +(0, 1) .. (G--1); 
	\end{tikzpicture} - \frac{1}{2\N}\begin{tikzpicture}[scale = 0.25,  baseline={(0,-1ex/2)}] 
		\tikzstyle{vertex} = [shape = circle, fill,inner sep = 1pt] 
		\node[vertex] (G--2) at (1.5, -1) [shape = circle, fill, draw] {}; 
		\node[vertex] (G--1) at (0.0, -1) [shape = circle, fill, draw] {}; 
		\node[vertex] (G-2) at (1.5, 1) [shape = circle, fill, draw] {}; 
		\node[vertex] (G-1) at (0.0, 1) [shape = circle, fill, draw] {}; 
		\draw[] (G-2) .. controls +(-0.75, -1) and +(0.75, 1) .. (G--1); 
	\end{tikzpicture} - \frac{1}{2\N}\begin{tikzpicture}[scale = 0.25,  baseline={(0,-1ex/2)}] 
		\tikzstyle{vertex} = [shape = circle, fill,inner sep = 1pt] 
		\node[vertex] (G--2) at (1.5, -1) [shape = circle, fill, draw] {}; 
		\node[vertex] (G-1) at (0.0, 1) [shape = circle, fill, draw] {}; 
		\node[vertex] (G--1) at (0.0, -1) [shape = circle, fill, draw] {}; 
		\node[vertex] (G-2) at (1.5, 1) [shape = circle, fill, draw] {}; 
		\draw[] (G-1) .. controls +(0.75, -1) and +(-0.75, 1) .. (G--2); 
	\end{tikzpicture} + \frac{1}{2}\begin{tikzpicture}[scale = 0.25,  baseline={(0,-1ex/2)}] 
		\tikzstyle{vertex} = [shape = circle, fill,inner sep = 1pt] 
		\node[vertex] (G--2) at (1.5, -1) [shape = circle, fill, draw] {}; 
		\node[vertex] (G-1) at (0.0, 1) [shape = circle, fill, draw] {}; 
		\node[vertex] (G--1) at (0.0, -1) [shape = circle, fill, draw] {}; 
		\node[vertex] (G-2) at (1.5, 1) [shape = circle, fill, draw] {}; 
		\draw[] (G-1) .. controls +(0.75, -1) and +(-0.75, 1) .. (G--2); 
		\draw[] (G-2) .. controls +(-0.75, -1) and +(0.75, 1) .. (G--1); 
	\end{tikzpicture} + \frac{1}{2\N}\begin{tikzpicture}[scale = 0.25,  baseline={(0,-1ex/2)}] 
		\tikzstyle{vertex} = [shape = circle, fill,inner sep = 1pt] 
		\node[vertex] (G--2) at (1.5, -1) [shape = circle, fill, draw] {}; 
		\node[vertex] (G-2) at (1.5, 1) [shape = circle, fill, draw] {}; 
		\node[vertex] (G--1) at (0.0, -1) [shape = circle, fill, draw] {}; 
		\node[vertex] (G-1) at (0.0, 1) [shape = circle, fill, draw] {}; 
		\draw[] (G-2) .. controls +(0, -1) and +(0, 1) .. (G--2); 
	\end{tikzpicture} - \frac{1}{2}\begin{tikzpicture}[scale = 0.25,  baseline={(0,-1ex/2)}] 
		\tikzstyle{vertex} = [shape = circle, fill,inner sep = 1pt] 
		\node[vertex] (G--2) at (1.5, -1) [shape = circle, fill, draw] {}; 
		\node[vertex] (G-2) at (1.5, 1) [shape = circle, fill, draw] {}; 
		\node[vertex] (G--1) at (0.0, -1) [shape = circle, fill, draw] {}; 
		\node[vertex] (G-1) at (0.0, 1) [shape = circle, fill, draw] {}; 
		\draw[] (G-2) .. controls +(0, -1) and +(0, 1) .. (G--2); 
		\draw[] (G-1) .. controls +(0, -1) and +(0, 1) .. (G--1); 
	\end{tikzpicture}
\end{equation}
Note that we have not included the normalization constants that we computed in section \ref{subsection: normalization}. Rather, we have chosen to tabulate the exact output of the construction algorithm for ease of comparison.

\chapter{Expectation values: Combinatorial algorithm (code)}\label{apx: EV code}
In this section we will describe the SageMath \cite{sagemath} code implementing the combinatorial algorithm discussed in section \ref{subsec: exp vals}. Roughly, the code has two parts. The first part uses the matrix units construction in \ref{sec: construction of units} to compute the 1-point and 2-point function of matrix elements. The second part uses the matrix units to compute expectation values of observables labelled by 1-row partition diagrams. The code can be found at \href{https://github.com/adrianpadellaro/PhD-Thesis}{Link to GitHub Repository}.

\section{Propagators}
Throughout this section and the code, we use the following alternative description of partitions $\lambda \vdash \N$, which is useful when describing large $\N$ partitions. Let $\lambda^{\#} = [\lambda_2, \dots, \lambda_l] \vdash k$. This corresponds to a partition $\lambda = [\lambda_1, \lambda_2, \dots, \lambda_l] \vdash \N$ where $\lambda_1 = \N - k$. For example, the empty partition $\lambda^{\#} = []$ corresponds to $[\N]$; the partition $[1]$ corresponds to $[\N-1,1]$; $[1,1]$ corresponds to $[\N-2,1,1]$ and $[2]$ corresponds to $[\N-2,2]$.

\sloppy
\captionsetup{format=nocaption,aboveskip=0pt,belowskip=0pt}
The first cell in the notebook defines a polynomial ring $\mathbb{Q}[\N]$ with variable $\N$. This ring is necessary to define the partition algebras $P_1(\N), P_2(\N)$.
\begin{tcolorbox}[breakable, size=fbox, boxrule=1pt, pad at break*=1mm,colback=cellbackground, colframe=cellborder]
	\prompt{In}{incolor}{1}{\boxspacing}
	\begin{Verbatim}[commandchars=\\\{\},fontsize=\small]
		\PY{c+c1}{\PYZsh{}\PYZsh{} Define Polynomial ring QQ[N], Partition algebras P\PYZus{}1(N) and P\PYZus{}2(N)}
		\PY{n}{R}\PY{o}{.}\PY{o}{\PYZlt{}}\PY{n}{N}\PY{o}{\PYZgt{}} \PY{o}{=} \PY{n}{QQ}\PY{p}{[}\PY{p}{]}
		\PY{n}{P1N} \PY{o}{=} \PY{n}{PartitionAlgebra}\PY{p}{(}\PY{l+m+mi}{1}\PY{p}{,}\PY{n}{N}\PY{p}{)}
		\PY{n}{P2N} \PY{o}{=} \PY{n}{PartitionAlgebra}\PY{p}{(}\PY{l+m+mi}{2}\PY{p}{,}\PY{n}{N}\PY{p}{)}
	\end{Verbatim}
\end{tcolorbox}

In the second cell, we define three functions that take in partitions (respectively irreducible representations of $\SN$) R1, R2, R3 and return the three factors defined in \eqref{eq: vactab projector}. We remark that R1, R2 should only take values from $[], [1]$; R3 can take values from $[],[1], [2], [1,1]$. The variables Z1, Z32, Z2 corresponds to the commuting elements defined in Theorem \ref{thm: murphys}, where we have removed the constant factor such that eigenvalues correspond to those of $P^{(1)}_\N,P^{(2)}_\N,P^{(1+\frac{1}{2})}_\N$. We construct the numerator and denominator of \eqref{eq: vactab projector} separately and return both. This will be useful in what follows.
\begin{tcolorbox}[breakable, size=fbox, boxrule=1pt, pad at break*=1mm,colback=cellbackground, colframe=cellborder]
	\prompt{In}{incolor}{2}{\boxspacing}
	\begin{Verbatim}[commandchars=\\\{\},fontsize=\small]
		\PY{c+c1}{\PYZsh{}\PYZsh{} We define the projector on the first slot in the vacillating tableau}
		\PY{k}{def} \PY{n+nf}{P\PYZus{}R1}\PY{p}{(}\PY{n}{R1}\PY{p}{)}\PY{p}{:}
		\PY{n}{R}\PY{o}{.}\PY{o}{\PYZlt{}}\PY{n}{N}\PY{o}{\PYZgt{}} \PY{o}{=} \PY{n}{QQ}\PY{p}{[}\PY{p}{]}
		\PY{c+c1}{\PYZsh{}\PYZsh{} Construct T2 \PYZbs{}otimes \PYZbs{}idn \PYZbs{}otimes \PYZbs{}idn defined in equation (3.42)}
		\PY{n}{Z1} \PY{o}{=} \PY{n}{P2N}\PY{p}{(}\PY{n+nb}{sum}\PY{p}{(}\PY{n}{P1N}\PY{o}{.}\PY{n}{jucys\PYZus{}murphy\PYZus{}element}\PY{p}{(}\PY{n}{i}\PY{o}{/}\PY{l+m+mi}{2}\PY{p}{)} \PY{k}{for} \PY{n}{i} \PY{o+ow}{in} \PY{p}{[}\PY{l+m+mf}{1.}\PY{o}{.}\PY{l+m+mi}{2}\PY{p}{]}\PY{p}{)}\PY{o}{+}\PY{p}{(}\PY{n}{N}\PY{o}{*}\PY{p}{(}\PY{n}{N}\PY{o}{\PYZhy{}}\PY{l+m+mi}{1}\PY{p}{)}\PY{o}{/}\PY{l+m+mi}{2}\PY{o}{\PYZhy{}}\PY{n}{N}\PY{p}{)}\PY{o}{*}\PY{n}{P1N}\PY{o}{.}\PY{n}{one}\PY{p}{(}\PY{p}{)}\PY{p}{)}
		\PY{c+c1}{\PYZsh{}\PYZsh{} Define the set of irreps and corresponding normalized characters to take a product over in equation (3.52)}
		\PY{n}{IrrepsEigenvaluesDictionary} \PY{o}{=} \PY{p}{\PYZob{}}\PY{n}{Partition}\PY{p}{(}\PY{p}{[}\PY{p}{]}\PY{p}{)}\PY{p}{:} \PY{n}{R}\PY{p}{(}\PY{l+m+mi}{1}\PY{o}{/}\PY{l+m+mi}{2}\PY{o}{*}\PY{p}{(}\PY{n}{N}\PY{o}{\PYZhy{}}\PY{l+m+mi}{1}\PY{p}{)}\PY{o}{*}\PY{n}{N}\PY{p}{)}\PY{p}{,} \PY{n}{Partition}\PY{p}{(}\PY{p}{[}\PY{l+m+mi}{1}\PY{p}{]}\PY{p}{)}\PY{p}{:} \PY{n}{R}\PY{p}{(}\PY{p}{(}\PY{n}{N}\PY{o}{\PYZhy{}}\PY{l+m+mi}{3}\PY{p}{)}\PY{o}{*}\PY{n}{N}\PY{o}{/}\PY{l+m+mi}{2}\PY{p}{)}\PY{p}{\PYZcb{}}
		\PY{n}{numerator} \PY{o}{=} \PY{n}{prod}\PY{p}{(}\PY{p}{(}\PY{n}{Z1}\PY{o}{\PYZhy{}}\PY{n}{ev2}\PY{o}{*}\PY{n}{P2N}\PY{o}{.}\PY{n}{one}\PY{p}{(}\PY{p}{)}\PY{p}{)} \PY{k}{for} \PY{p}{(}\PY{n}{rep}\PY{p}{,} \PY{n}{ev2}\PY{p}{)} \PY{o+ow}{in} \PY{n}{IrrepsEigenvaluesDictionary}\PY{o}{.}\PY{n}{items}\PY{p}{(}\PY{p}{)} \PY{k}{if} \PY{n}{rep} \PY{o}{!=} \PY{n}{R1}\PY{p}{)}
		\PY{n}{denom} \PY{o}{=} \PY{n}{prod}\PY{p}{(}\PY{n}{IrrepsEigenvaluesDictionary}\PY{p}{[}\PY{n}{R1}\PY{p}{]}\PY{o}{\PYZhy{}}\PY{n}{ev2} \PY{k}{for} \PY{p}{(}\PY{n}{rep}\PY{p}{,}\PY{n}{ev2}\PY{p}{)} \PY{o+ow}{in} \PY{n}{IrrepsEigenvaluesDictionary}\PY{o}{.}\PY{n}{items}\PY{p}{(}\PY{p}{)} \PY{k}{if} \PY{n}{rep} \PY{o}{!=} \PY{n}{R1}\PY{p}{)}
		\PY{k}{return} \PY{n}{numerator}\PY{p}{,} \PY{n}{denom}
		\PY{c+c1}{\PYZsh{}\PYZsh{} We define the projector on the second slot in the vacillating tableau}
		\PY{k}{def} \PY{n+nf}{P\PYZus{}R2}\PY{p}{(}\PY{n}{R2}\PY{p}{)}\PY{p}{:}
		\PY{n}{R}\PY{o}{.}\PY{o}{\PYZlt{}}\PY{n}{N}\PY{o}{\PYZgt{}} \PY{o}{=} \PY{n}{QQ}\PY{p}{[}\PY{p}{]}
		\PY{c+c1}{\PYZsh{}\PYZsh{} Construct T2 \PYZbs{}otimes \PYZbs{}idn \PYZbs{}otimes \PYZbs{}idn defined in equation (3.42)}
		\PY{n}{Z32} \PY{o}{=} \PY{n}{P2N}\PY{p}{(}\PY{n+nb}{sum}\PY{p}{(}\PY{n}{P2N}\PY{o}{.}\PY{n}{jucys\PYZus{}murphy\PYZus{}element}\PY{p}{(}\PY{n}{i}\PY{o}{/}\PY{l+m+mi}{2}\PY{p}{)} \PY{k}{for} \PY{n}{i} \PY{o+ow}{in} \PY{p}{[}\PY{l+m+mf}{1.}\PY{o}{.}\PY{l+m+mi}{3}\PY{p}{]}\PY{p}{)}\PY{o}{+}\PY{p}{(}\PY{n}{N}\PY{o}{*}\PY{p}{(}\PY{n}{N}\PY{o}{\PYZhy{}}\PY{l+m+mi}{1}\PY{p}{)}\PY{o}{/}\PY{l+m+mi}{2}\PY{o}{\PYZhy{}}\PY{l+m+mi}{2}\PY{o}{*}\PY{n}{N}\PY{o}{+}\PY{l+m+mi}{1}\PY{p}{)}\PY{o}{*}\PY{n}{P2N}\PY{o}{.}\PY{n}{one}\PY{p}{(}\PY{p}{)}\PY{p}{)}
		\PY{c+c1}{\PYZsh{}\PYZsh{} Define the set of irreps and corresponding normalized characters to take a product over in equation (3.52)}
		\PY{n}{IrrepsEigenvaluesDictionary} \PY{o}{=} \PY{p}{\PYZob{}}\PY{n}{Partition}\PY{p}{(}\PY{p}{[}\PY{p}{]}\PY{p}{)}\PY{p}{:} \PY{n}{R}\PY{p}{(}\PY{l+m+mi}{1}\PY{o}{/}\PY{l+m+mi}{2}\PY{o}{*}\PY{p}{(}\PY{n}{N}\PY{o}{\PYZhy{}}\PY{l+m+mi}{2}\PY{p}{)}\PY{o}{*}\PY{p}{(}\PY{n}{N}\PY{o}{\PYZhy{}}\PY{l+m+mi}{1}\PY{p}{)}\PY{p}{)}\PY{p}{,} \PY{n}{Partition}\PY{p}{(}\PY{p}{[}\PY{l+m+mi}{1}\PY{p}{]}\PY{p}{)}\PY{p}{:} \PY{n}{R}\PY{p}{(}\PY{p}{(}\PY{n}{N}\PY{o}{\PYZhy{}}\PY{l+m+mi}{4}\PY{p}{)}\PY{o}{*}\PY{p}{(}\PY{n}{N}\PY{o}{\PYZhy{}}\PY{l+m+mi}{1}\PY{p}{)}\PY{o}{/}\PY{l+m+mi}{2}\PY{p}{)}\PY{p}{\PYZcb{}}
		\PY{n}{numerator} \PY{o}{=} \PY{n}{prod}\PY{p}{(}\PY{p}{(}\PY{n}{Z32}\PY{o}{\PYZhy{}}\PY{n}{ev2}\PY{o}{*}\PY{n}{P2N}\PY{o}{.}\PY{n}{one}\PY{p}{(}\PY{p}{)}\PY{p}{)} \PY{k}{for} \PY{p}{(}\PY{n}{rep}\PY{p}{,} \PY{n}{ev2}\PY{p}{)} \PY{o+ow}{in} \PY{n}{IrrepsEigenvaluesDictionary}\PY{o}{.}\PY{n}{items}\PY{p}{(}\PY{p}{)} \PY{k}{if} \PY{n}{rep} \PY{o}{!=} \PY{n}{R2}\PY{p}{)}
		\PY{n}{denom} \PY{o}{=} \PY{n}{prod}\PY{p}{(}\PY{n}{IrrepsEigenvaluesDictionary}\PY{p}{[}\PY{n}{R2}\PY{p}{]}\PY{o}{\PYZhy{}}\PY{n}{ev2} \PY{k}{for} \PY{p}{(}\PY{n}{rep}\PY{p}{,}\PY{n}{ev2}\PY{p}{)} \PY{o+ow}{in} \PY{n}{IrrepsEigenvaluesDictionary}\PY{o}{.}\PY{n}{items}\PY{p}{(}\PY{p}{)} \PY{k}{if} \PY{n}{rep} \PY{o}{!=} \PY{n}{R2}\PY{p}{)}
		\PY{k}{return} \PY{n}{numerator}\PY{p}{,} \PY{n}{denom}
		\PY{c+c1}{\PYZsh{}\PYZsh{} We define the projector on the third slot in the vacillating tableau}
		\PY{k}{def} \PY{n+nf}{P\PYZus{}R3}\PY{p}{(}\PY{n}{R3}\PY{p}{)}\PY{p}{:}
		\PY{n}{R}\PY{o}{.}\PY{o}{\PYZlt{}}\PY{n}{N}\PY{o}{\PYZgt{}} \PY{o}{=} \PY{n}{QQ}\PY{p}{[}\PY{p}{]}
		\PY{c+c1}{\PYZsh{}\PYZsh{} Construct T2 \PYZbs{}otimes \PYZbs{}idn \PYZbs{}otimes \PYZbs{}idn defined in equation (3.42)}
		\PY{n}{Z2} \PY{o}{=} \PY{n}{P2N}\PY{p}{(}\PY{n+nb}{sum}\PY{p}{(}\PY{n}{P2N}\PY{o}{.}\PY{n}{jucys\PYZus{}murphy\PYZus{}element}\PY{p}{(}\PY{n}{i}\PY{o}{/}\PY{l+m+mi}{2}\PY{p}{)} \PY{k}{for} \PY{n}{i} \PY{o+ow}{in} \PY{p}{[}\PY{l+m+mf}{1.}\PY{o}{.}\PY{l+m+mi}{4}\PY{p}{]}\PY{p}{)}\PY{o}{+}\PY{p}{(}\PY{n}{N}\PY{o}{*}\PY{p}{(}\PY{n}{N}\PY{o}{\PYZhy{}}\PY{l+m+mi}{1}\PY{p}{)}\PY{o}{/}\PY{l+m+mi}{2}\PY{o}{\PYZhy{}}\PY{l+m+mi}{2}\PY{o}{*}\PY{n}{N}\PY{p}{)}\PY{o}{*}\PY{n}{P2N}\PY{o}{.}\PY{n}{one}\PY{p}{(}\PY{p}{)}\PY{p}{)}
		\PY{c+c1}{\PYZsh{}\PYZsh{} Define the set of irreps and corresponding normalized characters to take a product over in equation (3.52)}
		\PY{n}{IrrepsEigenvaluesDictionary} \PY{o}{=} \PY{p}{\PYZob{}}\PY{n}{Partition}\PY{p}{(}\PY{p}{[}\PY{p}{]}\PY{p}{)}\PY{p}{:} \PY{n}{R}\PY{p}{(}\PY{l+m+mi}{1}\PY{o}{/}\PY{l+m+mi}{2}\PY{o}{*}\PY{p}{(}\PY{n}{N}\PY{o}{\PYZhy{}}\PY{l+m+mi}{1}\PY{p}{)}\PY{o}{*}\PY{n}{N}\PY{p}{)}\PY{p}{,} \PY{n}{Partition}\PY{p}{(}\PY{p}{[}\PY{l+m+mi}{1}\PY{p}{]}\PY{p}{)}\PY{p}{:} \PY{n}{R}\PY{p}{(}\PY{p}{(}\PY{n}{N}\PY{o}{\PYZhy{}}\PY{l+m+mi}{3}\PY{p}{)}\PY{o}{*}\PY{n}{N}\PY{o}{/}\PY{l+m+mi}{2}\PY{p}{)}\PY{p}{,} \PY{n}{Partition}\PY{p}{(}\PY{p}{[}\PY{l+m+mi}{2}\PY{p}{]}\PY{p}{)}\PY{p}{:} \PY{n}{R}\PY{p}{(}\PY{l+m+mi}{1}\PY{o}{/}\PY{l+m+mi}{2}\PY{o}{*}\PY{p}{(}\PY{n}{N} \PY{o}{\PYZhy{}} \PY{l+m+mi}{1}\PY{p}{)}\PY{o}{*}\PY{p}{(}\PY{n}{N} \PY{o}{\PYZhy{}} \PY{l+m+mi}{4}\PY{p}{)}\PY{p}{)}\PY{p}{,} \PY{n}{Partition}\PY{p}{(}\PY{p}{[}\PY{l+m+mi}{1}\PY{p}{,}\PY{l+m+mi}{1}\PY{p}{]}\PY{p}{)}\PY{p}{:} \PY{n}{R}\PY{p}{(}\PY{l+m+mi}{1}\PY{o}{/}\PY{l+m+mi}{2}\PY{o}{*}\PY{p}{(}\PY{n}{N} \PY{o}{\PYZhy{}} \PY{l+m+mi}{5}\PY{p}{)}\PY{o}{*}\PY{n}{N}\PY{p}{)}\PY{p}{\PYZcb{}}
		\PY{n}{numerator} \PY{o}{=} \PY{n}{prod}\PY{p}{(}\PY{p}{(}\PY{n}{Z2}\PY{o}{\PYZhy{}}\PY{n}{ev2}\PY{o}{*}\PY{n}{P2N}\PY{o}{.}\PY{n}{one}\PY{p}{(}\PY{p}{)}\PY{p}{)} \PY{k}{for} \PY{p}{(}\PY{n}{rep}\PY{p}{,} \PY{n}{ev2}\PY{p}{)} \PY{o+ow}{in} \PY{n}{IrrepsEigenvaluesDictionary}\PY{o}{.}\PY{n}{items}\PY{p}{(}\PY{p}{)} \PY{k}{if} \PY{n}{rep} \PY{o}{!=} \PY{n}{R3}\PY{p}{)}
		\PY{n}{denom} \PY{o}{=} \PY{n}{prod}\PY{p}{(}\PY{n}{IrrepsEigenvaluesDictionary}\PY{p}{[}\PY{n}{R3}\PY{p}{]}\PY{o}{\PYZhy{}}\PY{n}{ev2} \PY{k}{for} \PY{p}{(}\PY{n}{rep}\PY{p}{,}\PY{n}{ev2}\PY{p}{)} \PY{o+ow}{in} \PY{n}{IrrepsEigenvaluesDictionary}\PY{o}{.}\PY{n}{items}\PY{p}{(}\PY{p}{)} \PY{k}{if} \PY{n}{rep} \PY{o}{!=} \PY{n}{R3}\PY{p}{)}
		\PY{k}{return} \PY{n}{numerator}\PY{p}{,} \PY{n}{denom}
	\end{Verbatim}
\end{tcolorbox}

Now we construct the projector $P_{\vactab \vactab'}$ and the corresponding matrix \eqref{eq: P vactab vactab matrix}. For this it is useful to enumerate all the allowed vacillating tableaux $\vactab$. This is done in the third cell.
\begin{tcolorbox}[breakable, size=fbox, boxrule=1pt, pad at break*=1mm,colback=cellbackground, colframe=cellborder]
	\prompt{In}{incolor}{3}{\boxspacing}
	\begin{Verbatim}[commandchars=\\\{\},fontsize=\small]
		\PY{c+c1}{\PYZsh{}\PYZsh{} The set of all vacillating tableaux at k=2}
		\PY{n}{VacTabs} \PY{o}{=} \PY{p}{[} \PY{p}{(}\PY{n}{Partition}\PY{p}{(}\PY{p}{[}\PY{p}{]}\PY{p}{)}\PY{p}{,}\PY{n}{Partition}\PY{p}{(}\PY{p}{[}\PY{p}{]}\PY{p}{)}\PY{p}{,}\PY{n}{Partition}\PY{p}{(}\PY{p}{[}\PY{p}{]}\PY{p}{)}\PY{p}{)}\PY{p}{,} \PYZbs{}
		\PY{p}{(}\PY{n}{Partition}\PY{p}{(}\PY{p}{[}\PY{p}{]}\PY{p}{)}\PY{p}{,}\PY{n}{Partition}\PY{p}{(}\PY{p}{[}\PY{p}{]}\PY{p}{)}\PY{p}{,}\PY{n}{Partition}\PY{p}{(}\PY{p}{[}\PY{l+m+mi}{1}\PY{p}{]}\PY{p}{)}\PY{p}{)}\PY{p}{,} \PYZbs{}
		\PY{p}{(}\PY{n}{Partition}\PY{p}{(}\PY{p}{[}\PY{l+m+mi}{1}\PY{p}{]}\PY{p}{)}\PY{p}{,}\PY{n}{Partition}\PY{p}{(}\PY{p}{[}\PY{l+m+mi}{1}\PY{p}{]}\PY{p}{)}\PY{p}{,}\PY{n}{Partition}\PY{p}{(}\PY{p}{[}\PY{l+m+mi}{1}\PY{p}{]}\PY{p}{)}\PY{p}{)}\PY{p}{,} \PYZbs{}
		\PY{p}{(}\PY{n}{Partition}\PY{p}{(}\PY{p}{[}\PY{l+m+mi}{1}\PY{p}{]}\PY{p}{)}\PY{p}{,}\PY{n}{Partition}\PY{p}{(}\PY{p}{[}\PY{l+m+mi}{1}\PY{p}{]}\PY{p}{)}\PY{p}{,}\PY{n}{Partition}\PY{p}{(}\PY{p}{[}\PY{l+m+mi}{2}\PY{p}{]}\PY{p}{)}\PY{p}{)}\PY{p}{,} \PYZbs{}
		\PY{p}{(}\PY{n}{Partition}\PY{p}{(}\PY{p}{[}\PY{l+m+mi}{1}\PY{p}{]}\PY{p}{)}\PY{p}{,}\PY{n}{Partition}\PY{p}{(}\PY{p}{[}\PY{l+m+mi}{1}\PY{p}{]}\PY{p}{)}\PY{p}{,}\PY{n}{Partition}\PY{p}{(}\PY{p}{[}\PY{l+m+mi}{1}\PY{p}{,}\PY{l+m+mi}{1}\PY{p}{]}\PY{p}{)}\PY{p}{)}\PY{p}{,} \PYZbs{}
		\PY{p}{(}\PY{n}{Partition}\PY{p}{(}\PY{p}{[}\PY{l+m+mi}{1}\PY{p}{]}\PY{p}{)}\PY{p}{,}\PY{n}{Partition}\PY{p}{(}\PY{p}{[}\PY{p}{]}\PY{p}{)}\PY{p}{,}\PY{n}{Partition}\PY{p}{(}\PY{p}{[}\PY{p}{]}\PY{p}{)}\PY{p}{)}\PY{p}{,} \PYZbs{}
		\PY{p}{(}\PY{n}{Partition}\PY{p}{(}\PY{p}{[}\PY{l+m+mi}{1}\PY{p}{]}\PY{p}{)}\PY{p}{,}\PY{n}{Partition}\PY{p}{(}\PY{p}{[}\PY{p}{]}\PY{p}{)}\PY{p}{,}\PY{n}{Partition}\PY{p}{(}\PY{p}{[}\PY{l+m+mi}{1}\PY{p}{]}\PY{p}{)}\PY{p}{)}\PY{p}{]}
	\end{Verbatim}
\end{tcolorbox}

For each pair of vacillating tableaux, there is a corresponding coupling constant. These will be symbolic variables. We define these and collect them into a dictionary in the fourth cell, for future convenience.
\begin{tcolorbox}[breakable, size=fbox, boxrule=1pt, pad at break*=1mm,colback=cellbackground, colframe=cellborder]
	\prompt{In}{incolor}{4}{\boxspacing}
	\begin{Verbatim}[commandchars=\\\{\},fontsize=\small]
		\PY{c+c1}{\PYZsh{}\PYZsh{} Some initializations}
		\PY{n}{g} \PY{o}{=} \PY{p}{\PYZob{}}\PY{p}{\PYZcb{}}
		\PY{n}{k} \PY{o}{=} \PY{n+nb}{len}\PY{p}{(}\PY{n}{VacTabs}\PY{p}{)}
		\PY{c+c1}{\PYZsh{}\PYZsh{} Now produce the diagram element corresponding to the connected 2\PYZhy{}point function}
		\PY{c+c1}{\PYZsh{}\PYZsh{} For this we need symbolic variables corresponding to coupling constants.}
		\PY{c+c1}{\PYZsh{}\PYZsh{} We encode these in a dictionary}
		\PY{k}{for} \PY{n}{v1} \PY{o+ow}{in} \PY{n+nb}{range}\PY{p}{(}\PY{n}{k}\PY{p}{)}\PY{p}{:}
		\PY{k}{for} \PY{n}{v2} \PY{o+ow}{in} \PY{n+nb}{range}\PY{p}{(}\PY{n}{k}\PY{p}{)}\PY{p}{:}
		\PY{k}{if} \PY{n}{v2} \PY{o}{\PYZlt{}}\PY{o}{=} \PY{n}{v1}\PY{p}{:}
		\PY{n}{g}\PY{p}{[}\PY{p}{(}\PY{n}{VacTabs}\PY{p}{[}\PY{n}{v1}\PY{p}{]}\PY{p}{,}\PY{n}{VacTabs}\PY{p}{[}\PY{n}{v2}\PY{p}{]}\PY{p}{)}\PY{p}{]} \PY{o}{=} \PY{n}{var}\PY{p}{(}\PY{l+s+s1}{\PYZsq{}}\PY{l+s+s1}{g\PYZus{}}\PY{l+s+si}{\PYZob{}0\PYZcb{}}\PY{l+s+si}{\PYZob{}1\PYZcb{}}\PY{l+s+s1}{\PYZsq{}}\PY{o}{.}\PY{n}{format}\PY{p}{(}\PY{n}{v1}\PY{p}{,}\PY{n}{v2}\PY{p}{)}\PY{p}{)}	
	\end{Verbatim}
\end{tcolorbox}
We record the normalization constants  in \eqref{eq: matrix units normalizations}, associated with pairs of vacillating tableaux.
\begin{tcolorbox}[breakable, size=fbox, boxrule=1pt, pad at break*=1mm,colback=cellbackground, colframe=cellborder]
	\prompt{In}{incolor}{ }{\boxspacing}
	\begin{Verbatim}[commandchars=\\\{\}]
		\PY{c+c1}{\PYZsh{}\PYZsh{} We also collect the set of normalization constants computed in section 4.3.5}
		\PY{n}{n} \PY{o}{=} \PY{p}{\PYZob{}}\PY{p}{\PYZcb{}}
		\PY{n}{n}\PY{p}{[}\PY{l+m+mi}{0}\PY{p}{,}\PY{l+m+mi}{0}\PY{p}{]} \PY{o}{=} \PY{n}{N}
		\PY{n}{n}\PY{p}{[}\PY{l+m+mi}{1}\PY{p}{,}\PY{l+m+mi}{0}\PY{p}{]} \PY{o}{=} \PY{n}{N}\PY{o}{*}\PY{n}{sqrt}\PY{p}{(}\PY{l+m+mi}{2}\PY{o}{/}\PY{p}{(}\PY{n}{N}\PY{o}{\PYZhy{}}\PY{l+m+mi}{1}\PY{p}{)}\PY{p}{)}
		\PY{n}{n}\PY{p}{[}\PY{l+m+mi}{1}\PY{p}{,}\PY{l+m+mi}{1}\PY{p}{]} \PY{o}{=} \PY{n}{N}\PY{o}{/}\PY{p}{(}\PY{n}{N}\PY{o}{\PYZhy{}}\PY{l+m+mi}{1}\PY{p}{)}
		\PY{n}{n}\PY{p}{[}\PY{l+m+mi}{2}\PY{p}{,}\PY{l+m+mi}{2}\PY{p}{]} \PY{o}{=} \PY{n}{N}
		\PY{n}{n}\PY{p}{[}\PY{l+m+mi}{3}\PY{p}{,}\PY{l+m+mi}{2}\PY{p}{]} \PY{o}{=} \PY{n}{N}\PY{o}{*}\PY{n}{sqrt}\PY{p}{(}\PY{l+m+mi}{2}\PY{o}{/}\PY{p}{(}\PY{n}{N}\PY{o}{\PYZhy{}}\PY{l+m+mi}{1}\PY{p}{)}\PY{p}{)}
		\PY{n}{n}\PY{p}{[}\PY{l+m+mi}{3}\PY{p}{,}\PY{l+m+mi}{3}\PY{p}{]} \PY{o}{=} \PY{n}{N}\PY{o}{/}\PY{p}{(}\PY{n}{N}\PY{o}{\PYZhy{}}\PY{l+m+mi}{1}\PY{p}{)}
		\PY{n}{n}\PY{p}{[}\PY{l+m+mi}{4}\PY{p}{,}\PY{l+m+mi}{2}\PY{p}{]} \PY{o}{=} \PY{n}{sqrt}\PY{p}{(}\PY{l+m+mi}{2}\PY{o}{*}\PY{p}{(}\PY{n}{N}\PY{o}{\PYZhy{}}\PY{l+m+mi}{1}\PY{p}{)}\PY{o}{/}\PY{p}{(}\PY{n}{N}\PY{o}{\PYZhy{}}\PY{l+m+mi}{2}\PY{p}{)}\PY{p}{)}
		\PY{n}{n}\PY{p}{[}\PY{l+m+mi}{4}\PY{p}{,}\PY{l+m+mi}{3}\PY{p}{]} \PY{o}{=} \PY{n}{sqrt}\PY{p}{(}\PY{l+m+mi}{2}\PY{o}{/}\PY{p}{(}\PY{n}{N}\PY{o}{\PYZhy{}}\PY{l+m+mi}{2}\PY{p}{)}\PY{p}{)}
		\PY{n}{n}\PY{p}{[}\PY{l+m+mi}{4}\PY{p}{,}\PY{l+m+mi}{4}\PY{p}{]} \PY{o}{=} \PY{p}{(}\PY{n}{N}\PY{o}{\PYZhy{}}\PY{l+m+mi}{1}\PY{p}{)}
		\PY{n}{n}\PY{p}{[}\PY{l+m+mi}{5}\PY{p}{,}\PY{l+m+mi}{5}\PY{p}{]} \PY{o}{=} \PY{l+m+mi}{1}
		\PY{n}{n}\PY{p}{[}\PY{l+m+mi}{6}\PY{p}{,}\PY{l+m+mi}{6}\PY{p}{]} \PY{o}{=} \PY{l+m+mi}{1}
	\end{Verbatim}
\end{tcolorbox}

To construct the element of $P_2(\N)$ corresponding to a propagator we define a partition algebra over symbolic rings (P2SR in the code), this is necessary for taking linear combinations of elements weighted by coupling constants. For this we iterate over all pairs of vacillating tableaux (v1, v2). At each iteration we construct the tuple LeftZip associated with the vacillating tableau v1. In the first slot LeftZip[0], there is a list of numerators for the three projectors previously defined. The second slot, LeftZip[1] is a list of numerators. LeftProj is a product of the numerators in LeftZip[0]. We also construct RightZip, RightProj corresponding to the vacillating tableau v2.

Given the two elements LeftProj, RightProj in $P_2(\N)$ we construct the matrix corresponding to left multiplication by LeftProj and right multiplication by RightProj on $P_2(\N)$. This matrix is called ProjMatrix in the cell. We want to compute the pivot column of this matrix in accordance with the construction in section \ref{sec: construction of units}. In particular we compute it for $\N = 10$ and store the column index in the variable pivot. Given the pivot column index we can extract the pivot column as ProjMatrix.column(pivot[0]). This is stored in the variable Q. Note that we have ignored the denominators up to now. We restore these to get the correct normalization for Q. The last step is to construct the transposition symmetrized element of Q, as these are the ones relevant for the propagator, and weight it by the coupling constant and normalization corresponding to the pair (v1,v2) of vacillating tableaux. This is added to the variable SRQ, which after running this cell, captures the full 1-matrix model propagator.
\begin{tcolorbox}[breakable, size=fbox, boxrule=1pt, pad at break*=1mm,colback=cellbackground, colframe=cellborder]
	\prompt{In}{incolor}{5}{\boxspacing}
	\begin{Verbatim}[commandchars=\\\{\},fontsize=\small]
		\PY{c+c1}{\PYZsh{}\PYZsh{} We now generate all the matrix units and multiply by the corresponding coupling constant (which are symbolic)}
		\PY{c+c1}{\PYZsh{}\PYZsh{} For this we need to use the partition algebra over a symbolic ring}
		\PY{n}{N} \PY{o}{=} \PY{n}{var}\PY{p}{(}\PY{l+s+s1}{\PYZsq{}}\PY{l+s+s1}{N}\PY{l+s+s1}{\PYZsq{}}\PY{p}{)}
		\PY{n}{P2SR} \PY{o}{=} \PY{n}{PartitionAlgebra}\PY{p}{(}\PY{l+m+mi}{2}\PY{p}{,}\PY{n}{N}\PY{p}{,}\PY{n}{SR}\PY{p}{)}
		\PY{c+c1}{\PYZsh{}\PYZsh{} Some initializations}
		\PY{n}{Q} \PY{o}{=} \PY{l+m+mi}{0}
		\PY{n}{SRQ} \PY{o}{=} \PY{l+m+mi}{0}
		\PY{c+c1}{\PYZsh{}\PYZsh{} We now iterate over all the pairs of vacillating tableaux}
		\PY{k}{for} \PY{n}{v1} \PY{o+ow}{in} \PY{n+nb}{range}\PY{p}{(}\PY{n}{k}\PY{p}{)}\PY{p}{:}
		\PY{n}{LeftZip} \PY{o}{=} \PY{n+nb}{list}\PY{p}{(}\PY{n+nb}{zip}\PY{p}{(}\PY{n}{P\PYZus{}R3}\PY{p}{(}\PY{n}{VacTabs}\PY{p}{[}\PY{n}{v1}\PY{p}{]}\PY{p}{[}\PY{l+m+mi}{2}\PY{p}{]}\PY{p}{)}\PY{p}{,}\PY{n}{P\PYZus{}R2}\PY{p}{(}\PY{n}{VacTabs}\PY{p}{[}\PY{n}{v1}\PY{p}{]}\PY{p}{[}\PY{l+m+mi}{1}\PY{p}{]}\PY{p}{)}\PY{p}{,}\PY{n}{P\PYZus{}R1}\PY{p}{(}\PY{n}{VacTabs}\PY{p}{[}\PY{n}{v1}\PY{p}{]}\PY{p}{[}\PY{l+m+mi}{0}\PY{p}{]}\PY{p}{)}\PY{p}{)}\PY{p}{)}
		\PY{c+c1}{\PYZsh{}\PYZsh{} LeftZip[0] is a list of the numerators of the three projectors for the first vacillating tableaux}
		\PY{n}{LeftProj} \PY{o}{=}\PY{n}{prod}\PY{p}{(}\PY{n}{LeftZip}\PY{p}{[}\PY{l+m+mi}{0}\PY{p}{]}\PY{p}{)}
		\PY{c+c1}{\PYZsh{}\PYZsh{} We take their product to get the projector associated with the first vacillating tableaux}
		\PY{k}{for} \PY{n}{v2} \PY{o+ow}{in} \PY{n+nb}{range}\PY{p}{(}\PY{n}{k}\PY{p}{)}\PY{p}{:}
		\PY{k}{if} \PY{n}{v2} \PY{o}{\PYZlt{}}\PY{o}{=} \PY{n}{v1}\PY{p}{:}
		\PY{n}{RightZip} \PY{o}{=} \PY{n+nb}{list}\PY{p}{(}\PY{n+nb}{zip}\PY{p}{(}\PY{n}{P\PYZus{}R3}\PY{p}{(}\PY{n}{VacTabs}\PY{p}{[}\PY{n}{v2}\PY{p}{]}\PY{p}{[}\PY{l+m+mi}{2}\PY{p}{]}\PY{p}{)}\PY{p}{,}\PY{n}{P\PYZus{}R2}\PY{p}{(}\PY{n}{VacTabs}\PY{p}{[}\PY{n}{v2}\PY{p}{]}\PY{p}{[}\PY{l+m+mi}{1}\PY{p}{]}\PY{p}{)}\PY{p}{,}\PY{n}{P\PYZus{}R1}\PY{p}{(}\PY{n}{VacTabs}\PY{p}{[}\PY{n}{v2}\PY{p}{]}\PY{p}{[}\PY{l+m+mi}{0}\PY{p}{]}\PY{p}{)}\PY{p}{)}\PY{p}{)}
		\PY{c+c1}{\PYZsh{}\PYZsh{} RightZip[0] is a list of the numerator of the three projectors for the second vacillating tableau}
		\PY{n}{RightProj} \PY{o}{=}\PY{n}{prod}\PY{p}{(}\PY{n}{RightZip}\PY{p}{[}\PY{l+m+mi}{0}\PY{p}{]}\PY{p}{)}
		\PY{c+c1}{\PYZsh{}\PYZsh{} We take their product to get the projector associated with the second vacillating tableaux}
		\PY{n}{ProjMatrix} \PY{o}{=} \PY{n}{LeftProj}\PY{o}{.}\PY{n}{to\PYZus{}matrix}\PY{p}{(}\PY{n}{side}\PY{o}{=}\PY{l+s+s1}{\PYZsq{}}\PY{l+s+s1}{left}\PY{l+s+s1}{\PYZsq{}}\PY{p}{)}\PY{o}{*}\PY{n}{RightProj}\PY{o}{.}\PY{n}{to\PYZus{}matrix}\PY{p}{(}\PY{n}{side}\PY{o}{=}\PY{l+s+s1}{\PYZsq{}}\PY{l+s+s1}{right}\PY{l+s+s1}{\PYZsq{}}\PY{p}{)}
		\PY{c+c1}{\PYZsh{}\PYZsh{} ProjMatrix is the matrix representing the simultaneous left action of LeftProj and right action of RightProj}
		\PY{n}{pivot} \PY{o}{=} \PY{n}{ProjMatrix}\PY{o}{.}\PY{n}{subs}\PY{p}{(}\PY{n}{N}\PY{o}{=}\PY{l+m+mi}{10}\PY{p}{)}\PY{o}{.}\PY{n}{pivots}\PY{p}{(}\PY{p}{)}
		\PY{c+c1}{\PYZsh{}\PYZsh{} Compute the pivot columns of ProjMatrix for N=10}
		\PY{k}{if} \PY{n+nb}{len}\PY{p}{(}\PY{n}{pivot}\PY{p}{)} \PY{o}{\PYZgt{}} \PY{l+m+mi}{0}\PY{p}{:}
		\PY{n}{Q} \PY{o}{=} \PY{n}{ProjMatrix}\PY{o}{.}\PY{n}{column}\PY{p}{(}\PY{n}{pivot}\PY{p}{[}\PY{l+m+mi}{0}\PY{p}{]}\PY{p}{)}\PY{o}{/}\PY{n}{prod}\PY{p}{(}\PY{n}{LeftZip}\PY{p}{[}\PY{l+m+mi}{1}\PY{p}{]}\PY{p}{)}\PY{o}{/}\PY{n}{prod}\PY{p}{(}\PY{n}{RightZip}\PY{p}{[}\PY{l+m+mi}{1}\PY{p}{]}\PY{p}{)}
		\PY{c+c1}{\PYZsh{}\PYZsh{} As long as ProjMatrix has a pivot, extract the pivot column and divide by the numerators of the projectors}
		\PY{n}{SRQ} \PY{o}{+}\PY{o}{=} \PY{n}{n}\PY{p}{[}\PY{n}{v1}\PY{p}{,}\PY{n}{v2}\PY{p}{]}\PY{o}{*}\PY{n}{g}\PY{p}{[}\PY{p}{(}\PY{n}{VacTabs}\PY{p}{[}\PY{n}{v1}\PY{p}{]}\PY{p}{,}\PY{n}{VacTabs}\PY{p}{[}\PY{n}{v2}\PY{p}{]}\PY{p}{)}\PY{p}{]}\PY{o}{*}\PY{p}{(}\PY{n}{P2SR}\PY{o}{.}\PY{n}{from\PYZus{}vector}\PY{p}{(}\PY{n}{Q}\PY{p}{)}\PY{o}{/}\PY{l+m+mi}{2}\PY{o}{+} \PY{n}{P2SR}\PY{o}{.}\PY{n}{from\PYZus{}vector}\PY{p}{(}\PY{n}{Q}\PY{p}{)}\PY{o}{.}\PY{n}{dual}\PY{p}{(}\PY{p}{)}\PY{o}{/}\PY{l+m+mi}{2}\PY{p}{)}
		\PY{c+c1}{\PYZsh{}\PYZsh{} We take the average of the pivot column and its transpose,}
		\PY{c+c1}{\PYZsh{}\PYZsh{} weight it by the coupling constant associated with the pair of vacillating tableaux and add it to SRQ,}
		\PY{c+c1}{\PYZsh{}\PYZsh{} SRQ will correspond to the connected two\PYZhy{}point function (propagator) in our matrix model}
	\end{Verbatim}
\end{tcolorbox}
This ends the first part of the code. We now implement the 1-row partition combinatorics.

\section{One-row partitions}
To compute expectation values we need to: implement a vector space with basis labelled by set partitions/1-row diagrams; translate the element SRQ into a linear combination of 1-row diagrams; construct the linear combination of 1-row diagrams corresponding to the 1-point function; implement tensor products of 1-row diagrams and an inner product of two 1-row diagrams that returns $N$ to the number of components in the join of the two diagrams. This is done in the sixth cell of the notebook.

The CombinatorialFreeModule together with AlgebrasWithBasis in Sage provide the skeletons for constructing this. All that we need to provide is the labelling set for the vector space -- this is the set of all set partitions -- which is returned by calling $SetPartitions()$. Secondly, we have to define the product of two set partitions as product\_on\_basis(). This should mimic the tensor product of two diagrams, which corresponds to concatenation of the set partitions. Note that our current implementation assumes that the two set partitions are set partitions of distinct sets. This will be sufficient for our application. A careful implementation should allow these two sets to overlap (it is then necessary to relabel and reorder the elements in the set partitions). Thirdly, we define the identity element with respect to the tensor product -- this is just the empty partition $[]$. To define the inner product, or pairing of two set partitions, we first define the inner product of two basis elements as inner\_product\_on\_basis(). It takes two set partitions of the same set, constructs the join and counts the number of blocks. We return $\N$ to the power of the number of blocks. For a pair of general elements in this vector space we define inner\_prod by linear extension. There are two help functions in this class: from\_partition\_algebra() takes in an element of $P_2(\N)$ -- for example SRQ -- and returns the corresponding linear combination of vectors in this space. It will also be useful to have a function relabel\_element() that relabels the labelling set of the set partitions in an element. For example -- if an element $a$ contains set partitions of $\{1,2,3,4\}$, relabel\_element(a, [5,6,7,8]) returns the same element as a linear combination of set partitions of $\{5,6,7,8\}$ by replacing $1 \rightarrow 5, \dots, 4 \rightarrow 8$.
\begin{tcolorbox}[breakable, size=fbox, boxrule=1pt, pad at break*=1mm,colback=cellbackground, colframe=cellborder]
	\prompt{In}{incolor}{6}{\boxspacing}
	\begin{Verbatim}[commandchars=\\\{\},,fontsize=\small]
		\PY{c+c1}{\PYZsh{}\PYZsh{} We define a vector space labeled by set partitions}
		\PY{c+c1}{\PYZsh{}\PYZsh{} It has a \PYZdq{}product\PYZdq{} that combines set partitions by concatenation of lists: see product\PYZus{}on\PYZus{}basis()}
		\PY{c+c1}{\PYZsh{}\PYZsh{} This product implements tensor products of diagrams}
		\PY{c+c1}{\PYZsh{}\PYZsh{} It has some help functions, such as}
		\PY{c+c1}{\PYZsh{}\PYZsh{} (1) a function that converts a partition algebra element into an element of this vector space: see from\PYZus{}partition\PYZus{}algebra()}
		\PY{c+c1}{\PYZsh{}\PYZsh{} (2) a function that lets you change the label set of an element: see relabel\PYZus{}element()}
		\PY{c+c1}{\PYZsh{}\PYZsh{} (3) an inner product function that computes the number of parts in the join of two set partitions (linearly extended for general elements): see inner\PYZus{}prod()}
		\PY{k}{class} \PY{n+nc}{SetPartitionNA}\PY{p}{(}\PY{n}{CombinatorialFreeModule}\PY{p}{)}\PY{p}{:}
		
		\PY{k}{def} \PY{n+nf}{\PYZus{}\PYZus{}init\PYZus{}\PYZus{}}\PY{p}{(}\PY{n+nb+bp}{self}\PY{p}{,} \PY{o}{*}\PY{o}{*}\PY{n}{keywords}\PY{p}{)}\PY{p}{:}
		\PY{n+nb+bp}{self}\PY{o}{.}\PY{n}{\PYZus{}baseset} \PY{o}{=} \PY{n}{SetPartitions}\PY{p}{(}\PY{p}{)}
		\PY{n}{CombinatorialFreeModule}\PY{o}{.}\PY{n+nf+fm}{\PYZus{}\PYZus{}init\PYZus{}\PYZus{}}\PY{p}{(}\PY{n+nb+bp}{self}\PY{p}{,} \PY{n}{SR}\PY{p}{,} \PY{n+nb+bp}{self}\PY{o}{.}\PY{n}{\PYZus{}baseset}\PY{p}{,}
		\PY{n}{category}\PY{o}{=}\PY{n}{AlgebrasWithBasis}\PY{p}{(}\PY{n}{SR}\PY{p}{)}\PY{p}{,} \PY{o}{*}\PY{o}{*}\PY{n}{keywords}\PY{p}{)}
		
		\PY{k}{def} \PY{n+nf}{product\PYZus{}on\PYZus{}basis}\PY{p}{(}\PY{n+nb+bp}{self}\PY{p}{,} \PY{n}{left}\PY{p}{,} \PY{n}{right}\PY{p}{)}\PY{p}{:}
		\PY{n}{l} \PY{o}{=} \PY{n+nb}{list}\PY{p}{(}\PY{n}{left}\PY{p}{)}
		\PY{n}{r} \PY{o}{=} \PY{n+nb}{list}\PY{p}{(}\PY{n}{right}\PY{p}{)}
		\PY{k}{return} \PY{n+nb+bp}{self}\PY{o}{.}\PY{n}{monomial}\PY{p}{(}\PY{n}{SetPartition}\PY{p}{(}\PY{n}{l}\PY{o}{+}\PY{n}{r}\PY{p}{)}\PY{p}{)}
		
		\PY{k}{def} \PY{n+nf}{one\PYZus{}basis}\PY{p}{(}\PY{n+nb+bp}{self}\PY{p}{)}\PY{p}{:}
		\PY{k}{return} \PY{n}{SetPartition}\PY{p}{(}\PY{p}{[}\PY{p}{]}\PY{p}{)}
		
		\PY{k}{def} \PY{n+nf}{algebra\PYZus{}generators}\PY{p}{(}\PY{n+nb+bp}{self}\PY{p}{)}\PY{p}{:}
		\PY{k}{return} \PY{n}{SetPartitions}\PY{p}{(}\PY{p}{)}
		
		\PY{k}{def} \PY{n+nf}{\PYZus{}repr\PYZus{}}\PY{p}{(}\PY{n+nb+bp}{self}\PY{p}{)}\PY{p}{:}
		\PY{k}{return} \PY{l+s+s2}{\PYZdq{}}\PY{l+s+s2}{Algebra of set partitions over }\PY{l+s+si}{\PYZpc{}s}\PY{l+s+s2}{ with multiplication given by concatenation}\PY{l+s+s2}{\PYZdq{}}\PY{o}{\PYZpc{}}\PY{p}{(}\PY{n}{SR}\PY{p}{)}
		
		\PY{k}{def} \PY{n+nf}{inner\PYZus{}product\PYZus{}on\PYZus{}basis}\PY{p}{(}\PY{n+nb+bp}{self}\PY{p}{,} \PY{n}{left}\PY{p}{,} \PY{n}{right}\PY{p}{)}\PY{p}{:}
		\PY{k}{if} \PY{n}{left}\PY{o}{.}\PY{n}{base\PYZus{}set}\PY{p}{(}\PY{p}{)} \PY{o}{!=} \PY{n}{right}\PY{o}{.}\PY{n}{base\PYZus{}set}\PY{p}{(}\PY{p}{)}\PY{p}{:}
		\PY{k}{return} \PY{l+m+mi}{0}
		\PY{k}{else}\PY{p}{:}
		\PY{n}{join} \PY{o}{=} \PY{n}{left}\PY{o}{.}\PY{n}{sup}\PY{p}{(}\PY{n}{right}\PY{p}{)}
		\PY{k}{return} \PY{n}{SR}\PY{p}{(}\PY{n}{N}\PY{o}{\PYZca{}}\PY{n+nb}{len}\PY{p}{(}\PY{n}{join}\PY{p}{)}\PY{p}{)}
		
		\PY{k}{def} \PY{n+nf}{inner\PYZus{}prod}\PY{p}{(}\PY{n+nb+bp}{self}\PY{p}{,} \PY{n}{left}\PY{p}{,} \PY{n}{right}\PY{p}{)}\PY{p}{:}
		\PY{n}{innerprod} \PY{o}{=} \PY{l+m+mi}{0}
		\PY{k}{for} \PY{n}{l} \PY{o+ow}{in} \PY{n}{left}\PY{p}{:}
		\PY{k}{for} \PY{n}{r} \PY{o+ow}{in} \PY{n}{right}\PY{p}{:}
		\PY{n}{innerprod} \PY{o}{+}\PY{o}{=} \PY{n}{l}\PY{p}{[}\PY{l+m+mi}{1}\PY{p}{]}\PY{o}{*}\PY{n}{r}\PY{p}{[}\PY{l+m+mi}{1}\PY{p}{]}\PY{o}{*}\PY{n+nb+bp}{self}\PY{o}{.}\PY{n}{inner\PYZus{}product\PYZus{}on\PYZus{}basis}\PY{p}{(}\PY{n}{l}\PY{p}{[}\PY{l+m+mi}{0}\PY{p}{]}\PY{p}{,}\PY{n}{r}\PY{p}{[}\PY{l+m+mi}{0}\PY{p}{]}\PY{p}{)}
		\PY{k}{return} \PY{n}{innerprod}
		
		\PY{k}{def} \PY{n+nf}{from\PYZus{}partition\PYZus{}algebra}\PY{p}{(}\PY{n+nb+bp}{self}\PY{p}{,} \PY{n}{d}\PY{p}{,} \PY{n}{baseset\PYZus{}left}\PY{p}{,} \PY{n}{baseset\PYZus{}right}\PY{p}{)}\PY{p}{:}
		\PY{n}{B} \PY{o}{=} \PY{n+nb+bp}{self}\PY{o}{.}\PY{n}{basis}\PY{p}{(}\PY{p}{)}
		\PY{n}{element} \PY{o}{=} \PY{l+m+mi}{0}
		\PY{k}{for} \PY{p}{(}\PY{n}{setpart}\PY{p}{,} \PY{n}{coeff}\PY{p}{)} \PY{o+ow}{in} \PY{n}{d}\PY{p}{:}
		\PY{n}{setpartnew} \PY{o}{=} \PY{p}{[}\PY{p}{]}
		\PY{k}{for} \PY{n}{part} \PY{o+ow}{in} \PY{n}{setpart}\PY{o}{.}\PY{n}{set\PYZus{}partition}\PY{p}{(}\PY{p}{)}\PY{p}{:}
		\PY{n}{partnew} \PY{o}{=} \PY{p}{[}\PY{p}{]}
		\PY{k}{for} \PY{n}{p} \PY{o+ow}{in} \PY{n}{part}\PY{p}{:}
		\PY{k}{if} \PY{n}{p} \PY{o}{\PYZgt{}} \PY{l+m+mi}{0}\PY{p}{:}
		\PY{n}{partnew} \PY{o}{+}\PY{o}{=} \PY{p}{[}\PY{n}{baseset\PYZus{}left}\PY{p}{[}\PY{n}{p}\PY{o}{\PYZhy{}}\PY{l+m+mi}{1}\PY{p}{]}\PY{p}{]}
		\PY{k}{elif} \PY{n}{p} \PY{o}{\PYZlt{}} \PY{l+m+mi}{0}\PY{p}{:}
		\PY{n}{partnew} \PY{o}{+}\PY{o}{=} \PY{p}{[}\PY{n}{baseset\PYZus{}right}\PY{p}{[}\PY{o}{\PYZhy{}}\PY{n}{p}\PY{o}{\PYZhy{}}\PY{l+m+mi}{1}\PY{p}{]}\PY{p}{]}
		\PY{n}{setpartnew} \PY{o}{+}\PY{o}{=} \PY{p}{[}\PY{n}{partnew}\PY{p}{]}
		\PY{n}{element} \PY{o}{+}\PY{o}{=} \PY{n}{coeff}\PY{o}{*}\PY{n}{B}\PY{p}{[}\PY{n}{SetPartition}\PY{p}{(}\PY{n}{setpartnew}\PY{p}{)}\PY{p}{]}
		\PY{k}{return} \PY{n}{element}
		
		\PY{k}{def} \PY{n+nf}{relabel\PYZus{}element}\PY{p}{(}\PY{n+nb+bp}{self}\PY{p}{,} \PY{n}{elem}\PY{p}{,} \PY{n}{labelset}\PY{p}{)}\PY{p}{:}
		\PY{n}{B} \PY{o}{=} \PY{n+nb+bp}{self}\PY{o}{.}\PY{n}{basis}\PY{p}{(}\PY{p}{)}
		\PY{n}{element} \PY{o}{=} \PY{l+m+mi}{0}
		\PY{k}{for} \PY{p}{(}\PY{n}{setpart}\PY{p}{,} \PY{n}{coeff}\PY{p}{)} \PY{o+ow}{in} \PY{n}{elem}\PY{p}{:}
		\PY{n}{setpartnew} \PY{o}{=} \PY{p}{[}\PY{p}{]}
		\PY{k}{for} \PY{n}{part} \PY{o+ow}{in} \PY{n}{setpart}\PY{p}{:}
		\PY{n}{partnew} \PY{o}{=} \PY{p}{[}\PY{p}{]}
		\PY{k}{for} \PY{n}{p} \PY{o+ow}{in} \PY{n}{part}\PY{p}{:}
		\PY{n}{partnew} \PY{o}{+}\PY{o}{=} \PY{p}{[}\PY{n}{labelset}\PY{p}{[}\PY{n}{p}\PY{o}{\PYZhy{}}\PY{l+m+mi}{1}\PY{p}{]}\PY{p}{]}
		\PY{n}{setpartnew} \PY{o}{+}\PY{o}{=} \PY{p}{[}\PY{n}{partnew}\PY{p}{]}
		\PY{n}{element} \PY{o}{+}\PY{o}{=} \PY{n}{coeff}\PY{o}{*}\PY{n}{B}\PY{p}{[}\PY{n}{SetPartition}\PY{p}{(}\PY{n}{setpartnew}\PY{p}{)}\PY{p}{]}
		\PY{k}{return} \PY{n}{element}
	\end{Verbatim}
\end{tcolorbox}

In cell seven we initialize this vector space with tensor product as $A$ and give the basis a name $B$.
\begin{tcolorbox}[breakable, size=fbox, boxrule=1pt, pad at break*=1mm,colback=cellbackground, colframe=cellborder]
	\prompt{In}{incolor}{7}{\boxspacing}
	\begin{Verbatim}[commandchars=\\\{\}, fontsize=\small]
		\PY{c+c1}{\PYZsh{}\PYZsh{} Initialize this vector space}
		\PY{n}{A} \PY{o}{=} \PY{n}{SetPartitionNA}\PY{p}{(}\PY{p}{)}
		\PY{n}{B} \PY{o}{=} \PY{n}{A}\PY{o}{.}\PY{n}{basis}\PY{p}{(}\PY{p}{)}
	\end{Verbatim}
\end{tcolorbox}

We are now ready to define the linear combination corresponding to the 1-point function \eqref{eq: matrix basis 1pt}. We define two variables gJ1, gJ2 capturing the coupling constants $(G^{-1})_{[\N];1\beta}J^{\beta}$ and $(G^{-1})_{[\N];2\beta}J^{\beta}$ respectively. C00 and CHH are the two distinct contributions to the 1-point function and EXP\_VAL is just their weighted sum.
\begin{tcolorbox}[breakable, size=fbox, boxrule=1pt, pad at break*=1mm,colback=cellbackground, colframe=cellborder]
	\prompt{In}{incolor}{8}{\boxspacing}
	\begin{Verbatim}[commandchars=\\\{\},fontsize=\small]
		\PY{c+c1}{\PYZsh{}\PYZsh{} Define Clebsches as elements of this vector space}
		\PY{n}{gJ1}\PY{p}{,} \PY{n}{gJ2} \PY{o}{=} \PY{n}{var}\PY{p}{(}\PY{l+s+s1}{\PYZsq{}}\PY{l+s+s1}{gJ1, gJ2}\PY{l+s+s1}{\PYZsq{}}\PY{p}{)}
		\PY{n}{C00} \PY{o}{=} \PY{l+m+mi}{1}\PY{o}{/}\PY{n}{N}\PY{o}{*}\PY{n}{B}\PY{p}{[}\PY{n}{SetPartition}\PY{p}{(}\PY{p}{[}\PY{p}{[}\PY{l+m+mi}{1}\PY{p}{]}\PY{p}{,}\PY{p}{[}\PY{l+m+mi}{2}\PY{p}{]}\PY{p}{]}\PY{p}{)}\PY{p}{]}
		\PY{n}{CHH} \PY{o}{=} \PY{l+m+mi}{1}\PY{o}{/}\PY{n}{sqrt}\PY{p}{(}\PY{p}{(}\PY{n}{N}\PY{o}{\PYZhy{}}\PY{l+m+mi}{1}\PY{p}{)}\PY{p}{)}\PY{o}{*}\PY{n}{B}\PY{p}{[}\PY{n}{SetPartition}\PY{p}{(}\PY{p}{[}\PY{p}{[}\PY{l+m+mi}{1}\PY{p}{,}\PY{l+m+mi}{2}\PY{p}{]}\PY{p}{]}\PY{p}{)}\PY{p}{]} \PY{o}{\PYZhy{}} \PY{l+m+mi}{1}\PY{o}{/}\PY{n}{N}\PY{o}{*}\PY{l+m+mi}{1}\PY{o}{/}\PY{n}{sqrt}\PY{p}{(}\PY{p}{(}\PY{n}{N}\PY{o}{\PYZhy{}}\PY{l+m+mi}{1}\PY{p}{)}\PY{p}{)}\PY{o}{*}\PY{n}{B}\PY{p}{[}\PY{n}{SetPartition}\PY{p}{(}\PY{p}{[}\PY{p}{[}\PY{l+m+mi}{1}\PY{p}{]}\PY{p}{,}\PY{p}{[}\PY{l+m+mi}{2}\PY{p}{]}\PY{p}{]}\PY{p}{)}\PY{p}{]}
		\PY{c+c1}{\PYZsh{}\PYZsh{} The one point function is a linear combination of these}
		\PY{n}{EXP\PYZus{}VAL} \PY{o}{=} \PY{n}{gJ1}\PY{o}{*}\PY{n}{C00}\PY{o}{+}\PY{n}{gJ2}\PY{o}{*}\PY{n}{CHH}
	\end{Verbatim}
\end{tcolorbox}
For the propagator we use the partition algebra element SRQ and the help function from\_partition\_algebra(SRQ).
\begin{tcolorbox}[breakable, size=fbox, boxrule=1pt, pad at break*=1mm,colback=cellbackground, colframe=cellborder]
	\prompt{In}{incolor}{9}{\boxspacing}
	\begin{Verbatim}[commandchars=\\\{\}, fontsize=\small]
		\PY{c+c1}{\PYZsh{}\PYZsh{} The connected two point function (propagator) is given by the partition algebra element SRQ computed earlier}
		\PY{n}{PROPAGATOR} \PY{o}{=} \PY{n}{A}\PY{o}{.}\PY{n}{from\PYZus{}partition\PYZus{}algebra}\PY{p}{(}\PY{n}{SRQ}\PY{p}{,} \PY{p}{[}\PY{l+m+mi}{1}\PY{p}{,}\PY{l+m+mi}{2}\PY{p}{]}\PY{p}{,} \PY{p}{[}\PY{l+m+mi}{3}\PY{p}{,}\PY{l+m+mi}{4}\PY{p}{]}\PY{p}{)}
	\end{Verbatim}
\end{tcolorbox}

We now have all the technology necessary to compute expectation values of observables. Observables are specified by a set partition. For degree $k$ observables they are set partitions of $\{1, \dots ,2k\}$. We give some examples in cell ten.
\begin{tcolorbox}[breakable, size=fbox, boxrule=1pt, pad at break*=1mm,colback=cellbackground, colframe=cellborder]
	\prompt{In}{incolor}{10}{\boxspacing}
	\begin{Verbatim}[commandchars=\\\{\}, fontsize=\small]
		\PY{c+c1}{\PYZsh{}\PYZsh{} Observables are specified by a set partiton. E.g.}
		\PY{n}{observable\PYZus{}as\PYZus{}set\PYZus{}partition} \PY{o}{=} \PY{n}{B}\PY{p}{[}\PY{n}{SetPartition}\PY{p}{(}\PY{p}{[}\PY{p}{[}\PY{l+m+mi}{1}\PY{p}{,}\PY{l+m+mi}{2}\PY{p}{]}\PY{p}{]}\PY{p}{)}\PY{p}{]}
		\PY{c+c1}{\PYZsh{}\PYZsh{} For a degree 1 observable}
		\PY{n}{observable\PYZus{}as\PYZus{}set\PYZus{}partition} \PY{o}{=} \PY{n}{B}\PY{p}{[}\PY{n}{SetPartition}\PY{p}{(}\PY{p}{[}\PY{p}{[}\PY{l+m+mi}{1}\PY{p}{,}\PY{l+m+mi}{2}\PY{p}{]}\PY{p}{,}\PY{p}{[}\PY{l+m+mi}{3}\PY{p}{,}\PY{l+m+mi}{4}\PY{p}{]}\PY{p}{]}\PY{p}{)}\PY{p}{]}
		\PY{c+c1}{\PYZsh{}\PYZsh{} For a degree 2 observable}
		\PY{n}{observable\PYZus{}as\PYZus{}set\PYZus{}partition} \PY{o}{=} \PY{n}{B}\PY{p}{[}\PY{n}{SetPartition}\PY{p}{(}\PY{p}{[}\PY{p}{[}\PY{l+m+mi}{1}\PY{p}{]}\PY{p}{,}\PY{p}{[}\PY{l+m+mi}{2}\PY{p}{]}\PY{p}{,}\PY{p}{[}\PY{l+m+mi}{3}\PY{p}{]}\PY{p}{,}\PY{p}{[}\PY{l+m+mi}{4}\PY{p}{]}\PY{p}{,}\PY{p}{[}\PY{l+m+mi}{5}\PY{p}{]}\PY{p}{,}\PY{p}{[}\PY{l+m+mi}{6}\PY{p}{]}\PY{p}{]}\PY{p}{)}\PY{p}{]}
		\PY{c+c1}{\PYZsh{}\PYZsh{} For a degree 3 observable}
	\end{Verbatim}
\end{tcolorbox}

To compute a degree 1 expectation values we specify an observable. The expectation value only gets contributions from the one-point function. To compute the expectation value we take the inner product of the observable and the onepoint function variable.
\begin{tcolorbox}[breakable, size=fbox, boxrule=1pt, pad at break*=1mm,colback=cellbackground, colframe=cellborder]
	\prompt{In}{incolor}{11}{\boxspacing}
	\begin{Verbatim}[commandchars=\\\{\},fontsize=\small]
		\PY{c+c1}{\PYZsh{}\PYZsh{} Degree 1 expvals are computed as}
		\PY{n}{observable\PYZus{}as\PYZus{}set\PYZus{}partition} \PY{o}{=} \PY{n}{B}\PY{p}{[}\PY{n}{SetPartition}\PY{p}{(}\PY{p}{[}\PY{p}{[}\PY{l+m+mi}{1}\PY{p}{,}\PY{l+m+mi}{2}\PY{p}{]}\PY{p}{]}\PY{p}{)}\PY{p}{]}
		\PY{n}{onepoint} \PY{o}{=} \PY{n}{EXP\PYZus{}VAL}
		\PY{n}{A}\PY{o}{.}\PY{n}{inner\PYZus{}prod}\PY{p}{(}\PY{n}{observable\PYZus{}as\PYZus{}set\PYZus{}partition}\PY{p}{,} \PY{n}{onepoint}\PY{p}{)}  
	\end{Verbatim}
\end{tcolorbox}
This returns the expectation value as a function of $\N$ and the parameters gJ1, gJ2.
\begin{tcolorbox}[breakable, size=fbox, boxrule=.5pt, pad at break*=1mm, opacityfill=0]
	\prompt{Out}{outcolor}{11}{\boxspacing}
	\begin{Verbatim}[commandchars=\\\{\}, fontsize=\small]
		N*(gJ1/N - gJ2/(sqrt(N - 1)*N)) + N*gJ2/sqrt(N - 1)
	\end{Verbatim}
\end{tcolorbox}
For degree two expectation values we receive contributions from the product of two one-point functions and a propagator -- this is stored in the variable twopoint. To get the correct form of the product of one-point functions we need to relabel EXP\_VAL. Recall that it initially was defined in terms of set partitions of $\{1,2\}$. The second EXP\_VAL should correspond to set partitions of $\{3,4\}$. Therefore, we relabel it and then take a tensor product. To find the expectation value we simply compute the inner product of twopoint function and the observable.
\begin{tcolorbox}[breakable, size=fbox, boxrule=1pt, pad at break*=1mm,colback=cellbackground, colframe=cellborder]
	\prompt{In}{incolor}{12}{\boxspacing}
	\begin{Verbatim}[commandchars=\\\{\},fontsize=\small]
		\PY{c+c1}{\PYZsh{}\PYZsh{} Degree 2 expvals are computed as}
		\PY{n}{observable\PYZus{}as\PYZus{}set\PYZus{}partition} \PY{o}{=} \PY{n}{B}\PY{p}{[}\PY{n}{SetPartition}\PY{p}{(}\PY{p}{[}\PY{p}{[}\PY{l+m+mi}{1}\PY{p}{,}\PY{l+m+mi}{2}\PY{p}{]}\PY{p}{,}\PY{p}{[}\PY{l+m+mi}{3}\PY{p}{,}\PY{l+m+mi}{4}\PY{p}{]}\PY{p}{]}\PY{p}{)}\PY{p}{]}
		\PY{n}{twopoint} \PY{o}{=} \PY{n}{PROPAGATOR} \PY{o}{+} \PY{n}{EXP\PYZus{}VAL}\PY{o}{*}\PY{n}{A}\PY{o}{.}\PY{n}{relabel\PYZus{}element}\PY{p}{(}\PY{n}{EXP\PYZus{}VAL}\PY{p}{,} \PY{p}{[}\PY{l+m+mi}{3}\PY{p}{,}\PY{l+m+mi}{4}\PY{p}{]}\PY{p}{)}
		\PY{n}{A}\PY{o}{.}\PY{n}{inner\PYZus{}prod}\PY{p}{(}\PY{n}{observable\PYZus{}as\PYZus{}set\PYZus{}partition}\PY{p}{,} \PY{n}{twopoint}\PY{p}{)}
	\end{Verbatim}
\end{tcolorbox}
This gives the following function of coupling constants and $\N$
\begin{tcolorbox}[breakable, size=fbox, boxrule=.5pt, pad at break*=1mm, opacityfill=0]
	\prompt{Out}{outcolor}{12}{\boxspacing}
	\begin{Verbatim}[commandchars=\\\{\},fontsize=\small]
		((gJ1/N - gJ2/(sqrt(N - 1)*N))\^{}2 - g\_21/(N\^{}2 - N) - g\_22/(N\^{}3 - 4*N\^{}2 + 5*N - 2)
		+ g\_33/(N\^{}2 - 3*N + 2) + g\_62/(N\^{}2 - N) + g\_00/N\^{}3 - g\_11/N\^{}3 - g\_50/N\^{}3 +
		g\_55/N\^{}3 + g\_61/N\^{}3 - g\_66/N\^{}3)*N\^{}2 - (g\_62*(1/(N\^{}2 - N) + 1/(N - 1)) -
		2*gJ2*(gJ1/N - gJ2/(sqrt(N - 1)*N))/sqrt(N - 1) - g\_21/(N\^{}2 - N) - 2*g\_22/(N\^{}3 -
		4*N\^{}2 + 5*N - 2) + 2*g\_33/(N\^{}2 - 3*N + 2) - g\_50/N\^{}2 + 2*g\_55/N\^{}2 + g\_61/N\^{}2 -
		2*g\_66/N\^{}2)*N\^{}2 + N\^{}2*(gJ2\^{}2/(N - 1) - g\_22/(N\^{}3 - 4*N\^{}2 + 5*N - 2) + g\_33/(N\^{}2
		- 3*N + 2) + g\_55/N + g\_62/(N - 1) - g\_66/N) + N*(N*g\_22/(N\^{}3 - 4*N\^{}2 + 5*N - 2)
		- N*g\_33/(N - 2) - N*g\_62/(N - 1) + g\_66) + 1/2*N*(g\_33 + g\_44) + 1/2*N*(g\_33 -
		g\_44) + N*(g\_62 - 2*g\_22/(N\^{}2 - 3*N + 2) + 2*g\_33/(N - 2)) + 1/2*N*(2*g\_21/(N\^{}2
		- N) + 2*g\_22/(N\^{}4 - 4*N\^{}3 + 5*N\^{}2 - 2*N) - g\_33/(N - 2) + g\_44/N - 2*g\_62/(N\^{}2
		- N) + 2*g\_11/N\^{}2 - 2*g\_61/N\^{}2 + 2*g\_66/N\^{}2) - N*(g\_21/(N - 1) + 2*g\_22/(N\^{}3 -
		4*N\^{}2 + 5*N - 2) - 2*g\_33/(N - 2) - g\_61/N - 2*g\_62/(N - 1) + 2*g\_66/N) +
		N*(g\_21/N + 2*g\_22/(N\^{}3 - 3*N\^{}2 + 2*N) - g\_33/(N - 2) - g\_44/N - g\_62/N) +
		1/2*N*(2*g\_22/(N\^{}2 - 2*N) - g\_33/(N - 2) + g\_44/N)
	\end{Verbatim}
\end{tcolorbox}
Similar considerations give expectation values of degree three observables. By Wick's theorem we get contributions from four different terms -- these are store in the threepoint variable. We use the relabeling function to get the correct set partitions.
\begin{tcolorbox}[breakable, size=fbox, boxrule=1pt, pad at break*=1mm,colback=cellbackground, colframe=cellborder]
	\prompt{In}{incolor}{13}{\boxspacing}
	\begin{Verbatim}[commandchars=\\\{\}, fontsize=\small]
		\PY{c+c1}{\PYZsh{}\PYZsh{} Degree 3 expvals are computed as}
		\PY{n}{observable\PYZus{}as\PYZus{}set\PYZus{}partition} \PY{o}{=} \PY{n}{B}\PY{p}{[}\PY{n}{SetPartition}\PY{p}{(}\PY{p}{[}\PY{p}{[}\PY{l+m+mi}{1}\PY{p}{]}\PY{p}{,}\PY{p}{[}\PY{l+m+mi}{2}\PY{p}{]}\PY{p}{,}\PY{p}{[}\PY{l+m+mi}{3}\PY{p}{]}\PY{p}{,}\PY{p}{[}\PY{l+m+mi}{4}\PY{p}{]}\PY{p}{,}\PY{p}{[}\PY{l+m+mi}{5}\PY{p}{]}\PY{p}{,}\PY{p}{[}\PY{l+m+mi}{6}\PY{p}{]}\PY{p}{]}\PY{p}{)}\PY{p}{]}
		\PY{n}{threepoint} \PY{o}{=} \PY{n}{PROPAGATOR}\PY{o}{*}\PY{n}{A}\PY{o}{.}\PY{n}{relabel\PYZus{}element}\PY{p}{(}\PY{n}{EXP\PYZus{}VAL}\PY{p}{,} \PY{p}{[}\PY{l+m+mi}{5}\PY{p}{,}\PY{l+m+mi}{6}\PY{p}{]}\PY{p}{)}
		\PY{n}{threepoint} \PY{o}{+}\PY{o}{=} \PY{n}{A}\PY{o}{.}\PY{n}{relabel\PYZus{}element}\PY{p}{(}\PY{n}{PROPAGATOR}\PY{p}{,} \PY{p}{[}\PY{l+m+mi}{1}\PY{p}{,}\PY{l+m+mi}{2}\PY{p}{,}\PY{l+m+mi}{5}\PY{p}{,}\PY{l+m+mi}{6}\PY{p}{]}\PY{p}{)}\PY{o}{*}\PY{n}{A}\PY{o}{.}\PY{n}{relabel\PYZus{}element}\PY{p}{(}\PY{n}{EXP\PYZus{}VAL}\PY{p}{,} \PY{p}{[}\PY{l+m+mi}{3}\PY{p}{,}\PY{l+m+mi}{4}\PY{p}{]}\PY{p}{)}
		\PY{n}{threepoint} \PY{o}{+}\PY{o}{=} \PY{n}{A}\PY{o}{.}\PY{n}{relabel\PYZus{}element}\PY{p}{(}\PY{n}{PROPAGATOR}\PY{p}{,} \PY{p}{[}\PY{l+m+mi}{3}\PY{p}{,}\PY{l+m+mi}{4}\PY{p}{,}\PY{l+m+mi}{5}\PY{p}{,}\PY{l+m+mi}{6}\PY{p}{]}\PY{p}{)}\PY{o}{*}\PY{n}{A}\PY{o}{.}\PY{n}{relabel\PYZus{}element}\PY{p}{(}\PY{n}{EXP\PYZus{}VAL}\PY{p}{,} \PY{p}{[}\PY{l+m+mi}{1}\PY{p}{,}\PY{l+m+mi}{2}\PY{p}{]}\PY{p}{)}
		\PY{n}{threepoint} \PY{o}{+}\PY{o}{=} \PY{n}{A}\PY{o}{.}\PY{n}{relabel\PYZus{}element}\PY{p}{(}\PY{n}{EXP\PYZus{}VAL}\PY{p}{,} \PY{p}{[}\PY{l+m+mi}{1}\PY{p}{,}\PY{l+m+mi}{2}\PY{p}{]}\PY{p}{)}\PY{o}{*}\PY{n}{A}\PY{o}{.}\PY{n}{relabel\PYZus{}element}\PY{p}{(}\PY{n}{EXP\PYZus{}VAL}\PY{p}{,} \PY{p}{[}\PY{l+m+mi}{3}\PY{p}{,}\PY{l+m+mi}{4}\PY{p}{]}\PY{p}{)}\PY{o}{*}\PY{n}{A}\PY{o}{.}\PY{n}{relabel\PYZus{}element}\PY{p}{(}\PY{n}{EXP\PYZus{}VAL}\PY{p}{,} \PY{p}{[}\PY{l+m+mi}{5}\PY{p}{,}\PY{l+m+mi}{6}\PY{p}{]}\PY{p}{)}
		\PY{n}{A}\PY{o}{.}\PY{n}{inner\PYZus{}prod}\PY{p}{(}\PY{n}{observable\PYZus{}as\PYZus{}set\PYZus{}partition}\PY{p}{,}\PY{n}{threepoint}\PY{p}{)}  
	\end{Verbatim}
\end{tcolorbox}
We do not print the result of this computation due its length.

For higher-degree expectation values we automate the computation of Wick contractions. The function kpoint(k) returns the appropriate element encoding Wick's theorem for degree $k$ expectation values.
\begin{tcolorbox}[breakable, size=fbox, boxrule=1pt, pad at break*=1mm,colback=cellbackground, colframe=cellborder]
	\prompt{In}{incolor}{14}{\boxspacing}
	\begin{Verbatim}[commandchars=\\\{\},fontsize=\small]
		\PY{c+c1}{\PYZsh{}\PYZsh{} For higher degree expvals it is useful to automate the Wick contractions}
		\PY{k}{def} \PY{n+nf}{kpoint}\PY{p}{(}\PY{n}{k}\PY{p}{)}\PY{p}{:}
		\PY{c+c1}{\PYZsh{}\PYZsh{} Fix a degree k}
		\PY{n}{pairs} \PY{o}{=} \PY{p}{[}\PY{p}{(}\PY{l+m+mi}{2}\PY{o}{*}\PY{n}{i}\PY{o}{\PYZhy{}}\PY{l+m+mi}{1}\PY{p}{,}\PY{l+m+mi}{2}\PY{o}{*}\PY{n}{i}\PY{p}{)} \PY{k}{for} \PY{n}{i} \PY{o+ow}{in} \PY{p}{[}\PY{l+m+mf}{1.}\PY{o}{.}\PY{n}{k}\PY{p}{]}\PY{p}{]}
		\PY{c+c1}{\PYZsh{}\PYZsh{} Generate the set of pairs [(1,2),....,(2k\PYZhy{}1,2k)]}
		\PY{n}{kpoint}\PY{o}{=}\PY{l+m+mi}{0}
		\PY{k}{for} \PY{n}{i} \PY{o+ow}{in} \PY{p}{[}\PY{l+m+mf}{0.}\PY{o}{.}\PY{n}{floor}\PY{p}{(}\PY{n}{k}\PY{o}{/}\PY{l+m+mi}{2}\PY{p}{)}\PY{p}{]}\PY{p}{:}
		\PY{c+c1}{\PYZsh{}\PYZsh{} Iterate over all set partitions of the pairs with blocks of size 1 or 2}
		\PY{k}{for} \PY{n}{term} \PY{o+ow}{in} \PY{n}{SetPartitions}\PY{p}{(}\PY{n}{pairs}\PY{p}{,} \PY{n+nb}{sorted}\PY{p}{(}\PY{p}{(}\PY{n}{k}\PY{o}{\PYZhy{}}\PY{l+m+mi}{2}\PY{o}{*}\PY{n}{i}\PY{p}{)}\PY{o}{*}\PY{p}{[}\PY{l+m+mi}{1}\PY{p}{]}\PY{o}{+}\PY{n}{i}\PY{o}{*}\PY{p}{[}\PY{l+m+mi}{2}\PY{p}{]}\PY{p}{,} \PY{n}{reverse}\PY{o}{=}\PY{k+kc}{True}\PY{p}{)}\PY{p}{)}\PY{p}{:}
		\PY{n}{kpoint} \PY{o}{+}\PY{o}{=} \PY{n}{prod}\PY{p}{(}\PY{n}{A}\PY{o}{.}\PY{n}{relabel\PYZus{}element}\PY{p}{(}\PY{n}{EXP\PYZus{}VAL}\PY{p}{,} \PY{n+nb}{list}\PY{p}{(}\PY{n}{i} \PY{k}{for} \PY{n}{c} \PY{o+ow}{in} \PY{n}{contraction} \PY{k}{for} \PY{n}{i} \PY{o+ow}{in} \PY{n}{c}\PY{p}{)}\PY{p}{)} \PY{k}{if} \PY{n+nb}{len}\PY{p}{(}\PY{n}{contraction}\PY{p}{)}\PY{o}{==}\PY{l+m+mi}{1} \PY{k}{else} \PY{n}{A}\PY{o}{.}\PY{n}{relabel\PYZus{}element}\PY{p}{(}\PY{n}{PROPAGATOR}\PY{p}{,} \PY{n+nb}{list}\PY{p}{(}\PY{n}{i} \PY{k}{for} \PY{n}{c} \PY{o+ow}{in} \PY{n}{contraction} \PY{k}{for} \PY{n}{i} \PY{o+ow}{in} \PY{n}{c}\PY{p}{)}\PY{p}{)}  \PY{k}{for} \PY{n}{contraction} \PY{o+ow}{in} \PY{n}{term} \PY{p}{)}
		\PY{k}{return} \PY{n}{kpoint}
	\end{Verbatim}
\end{tcolorbox}
This function can be checked against the known examples at degree $1,2,3$.
\begin{tcolorbox}[breakable, size=fbox, boxrule=1pt, pad at break*=1mm,colback=cellbackground, colframe=cellborder]
	\prompt{In}{incolor}{15}{\boxspacing}
	\begin{Verbatim}[commandchars=\\\{\},fontsize=\small]
		\PY{c+c1}{\PYZsh{}\PYZsh{} We can check that this gives the right answer for k=1,2,3}
		\PY{n}{kpoint}\PY{p}{(}\PY{l+m+mi}{1}\PY{p}{)} \PY{o}{==} \PY{n}{onepoint} \PY{o+ow}{and} \PY{n}{kpoint}\PY{p}{(}\PY{l+m+mi}{2}\PY{p}{)} \PY{o}{==} \PY{n}{twopoint} \PY{o+ow}{and} \PY{n}{kpoint}\PY{p}{(}\PY{l+m+mi}{3}\PY{p}{)} \PY{o}{==} \PY{n}{threepoint}
	\end{Verbatim}
\end{tcolorbox}
\begin{tcolorbox}[breakable, size=fbox, boxrule=.5pt, pad at break*=1mm, opacityfill=0]
	\prompt{Out}{outcolor}{15}{\boxspacing}
	\begin{Verbatim}[commandchars=\\\{\},fontsize=\small]
		True
	\end{Verbatim}
\end{tcolorbox}

To use this, for example at degree $4$ we fix an observable and compute
\begin{tcolorbox}[breakable, size=fbox, boxrule=1pt, pad at break*=1mm,colback=cellbackground, colframe=cellborder]
	\prompt{In}{incolor}{40}{\boxspacing}
	\begin{Verbatim}[commandchars=\\\{\},fontsize=\small]
		\PY{c+c1}{\PYZsh{}\PYZsh{} Now we can compute degree 4 expectation values}
		\PY{n}{observable\PYZus{}as\PYZus{}set\PYZus{}partition} \PY{o}{=} \PY{n}{B}\PY{p}{[}\PY{n}{SetPartition}\PY{p}{(}\PY{p}{[}\PY{p}{[}\PY{l+m+mi}{1}\PY{p}{]}\PY{p}{,}\PY{p}{[}\PY{l+m+mi}{2}\PY{p}{]}\PY{p}{,}\PY{p}{[}\PY{l+m+mi}{3}\PY{p}{]}\PY{p}{,}\PY{p}{[}\PY{l+m+mi}{4}\PY{p}{]}\PY{p}{,}\PY{p}{[}\PY{l+m+mi}{5}\PY{p}{]}\PY{p}{,}\PY{p}{[}\PY{l+m+mi}{6}\PY{p}{]}\PY{p}{,}\PY{p}{[}\PY{l+m+mi}{7}\PY{p}{]}\PY{p}{,}\PY{p}{[}\PY{l+m+mi}{8}\PY{p}{]}\PY{p}{]}\PY{p}{)}\PY{p}{]}
		\PY{n}{fourpoint}\PY{o}{=}\PY{n}{kpoint}\PY{p}{(}\PY{l+m+mi}{4}\PY{p}{)}
		\PY{n}{A}\PY{o}{.}\PY{n}{inner\PYZus{}prod}\PY{p}{(}\PY{n}{observable\PYZus{}as\PYZus{}set\PYZus{}partition}\PY{p}{,} \PY{n}{fourpoint}\PY{p}{)}
	\end{Verbatim}
\end{tcolorbox}
We do not print the result here either, but it can be found by running the code found at \href{https://github.com/adrianpadellaro/PhD-Thesis}{Link to GitHub Repository}.

\chapter{Observables: Double coset counting} \label{apx: double coset}
This appendix is devoted to a procedure for explicitly computing the number of double cosets of the type defined in Proposition \ref{eq: double cosets}.

In general, the number of double cosets in $H_1 \left\backslash G \right/ H_2$ can be written \cite{BenGeloun2014}
\begin{equation}
	|H_1 \left\backslash G \right/ H_2| = \frac{1}{|H_1||H_2|}\sum_{C} Z_C^{H_1\rightarrow G}Z_C^{H_2\rightarrow G}\Sym(C),
\end{equation}
where the sum is over conjugacy classes of $G$. The symbols $Z_C^{G\rightarrow H_1}, Z_C^{G\rightarrow H_2}$ denote the number of elements of $H_1$ and $H_2$ in the conjugacy class $C$ of $G$, respectively. $\Sym(C)$ is the number of elements in $G$ which commute with an element in $C$. 
\begin{proposition}
	For a permutation subgroup $H \subset G_1 \times G_2$, let $Z_{p,q}^{H}$ be the number of permutations $(h_1, h_2) \in H$ with cycle structure $p\vdash k$ in the first slot and $q \vdash k$ in the second slot.	
	For the double coset in \eqref{eqn:1colordoublecoset} we have
	\begin{align} \nonumber \label{eq: N(m+,m-) Cycle Index Formula}
		N(\vec{k}^+,\vec{k}^-) &= \frac{1}{|G(\vec{k}^+,\vec{k}^-)||S_k|} \sum_{p \vdash k} Z_{p,p}^{G(\vec{k}^+,\vec{k}^-)}Z_{p,p}^{\diag(S_k)} \Sym(p)^2 \\
		&=\frac{1}{|G(\vec{k}^+,\vec{k}^-)|} \sum_{p \vdash k} Z_{p,p}^{G(\vec{k}^+,\vec{k}^-)} \Sym(p).
	\end{align} 
	with
	\begin{equation}
		\Sym(p) = \prod_{i=1}^m p_i!i^{p_i}, \quad \sum_i ip_i = m.
	\end{equation}
\end{proposition}
The last equality follows from
\begin{equation}
	Z_{p,p}^{\diag(S_k)}=Z_{p}^{S_k} = \frac{|S_k|}{\Sym(p)}.
\end{equation}
We now prove the first equality (the equivalence with the counting due to Burnside's lemma \eqref{eqn: Burnsides lemma double coset one color}).
\begin{proof}
	Organize the sum $\sum_{\gamma \in \diag(S_k)}$ in \eqref{eqn: Burnsides lemma double coset one color} into a sum over conjugacy classes of $S_k^+ \times S_k^-$, and a sum over elements in the conjugacy class,
	\begin{equation}
		\sum_{\gamma \in \diag(S_m)} = \sum_{p \vdash m} \hspace{4pt} \sum_{\gamma \in C_p}.
	\end{equation}
	The Kronecker delta $\delta(\sigma_1^{-1}\rho^+(\mu)\gamma^+(\nu^+)\sigma_1 \gamma^{-1})$ vanishes unless $\rho^+(\mu)\gamma^+(\nu^+)$ is in the same conjugacy class as $\gamma^{-1} \in C_p$. Similarly for the second Kronecker delta. The number of elements $(\rho^+(\mu)\gamma^+(\nu^+),\rho^-(\mu)\gamma^-(\nu^-))$ in the conjugacy class $C_p \times C_p$ of $S_k^+ \times S_k^-$ is the definition of the coefficients $Z_{p,p}^{G(\vec{k}^+,\vec{k}^-)}$. The number of elements $(\gamma,\gamma)$ in the conjugacy class $C_p \times C_p$ is $Z_{p,p}^{\diag(S_k)}$. Given an element in $G(\vec{k}^+,\vec{k}^-)$ and an element in $\diag(S_k)$ in the same conjugacy class, there exists at least one element $(\sigma_1,\sigma_2)$ which relates the two by conjugation. Therefore, the Kronecker delta is non-zero at least $ Z_{p,p}^{G(\vec{k}^+,\vec{k}^-)}Z_{p,p}^{\diag(S_k)}$ times for each conjugacy class $C_p$. In equations we have
	\begin{align} \nonumber
		&\sum_{\substack{\mu \in S_{\vec{l}}, \nu^+ \in S_{\vec{k}^+} \\ \nu^- \in S_{\vec{k}^-},\gamma \in \diag(S_k)}} \sum_{\sigma_1, \sigma_2 \in S_k}
		\begin{aligned}[t]
			&\delta(\sigma_1^{-1}\rho^+(\mu)\gamma^+(\nu^+)\sigma_1\gamma^{-1})
			\delta(\sigma_2^{-1}\rho^-(\mu)\gamma^-(\nu^-)\sigma_2\gamma^{-1})
		\end{aligned}\\ \nonumber
		&= \sum_{p \vdash k} \sum_{\sigma_1, \sigma_2 \in S_k}
		\delta(\sigma_1^{-1}G_p^+\sigma_1\gamma^{-1}_p)
		\delta(\sigma_2^{-1}G_p^-\sigma_2\gamma^{-1}_p)
		Z_{p,p}^{G(\vec{k}^+,\vec{k}^-)}Z_{p,p}^{\diag(S_k)} \\
		&= \sum_{p \vdash k} Z_{p,p}^{G(\vec{k}^+,\vec{k}^-)}Z_{p,p}^{\diag(S_k)} \Sym(p)^2.
	\end{align}
	where $(G_p^+,G_p^-)$ is an arbitrary element of $G(\vec{k}^+,\vec{k}^-)$ in the conjugacy class $C_p \times C_p$ and similarly for $\gamma^{-1}_p$ in $\diag(S_k)$. To understand the last equality, consider the case where $(\sigma_1^{-1}G_p^+\sigma_1, \sigma_2^{-1}G_p^-\sigma_2) = (\gamma_p, \gamma_p)$. If $\sigma_1',\sigma_2'$ commute with $G_p^+$ and $G_p^-$ respectively, then
	\begin{equation}
		((\sigma_1'\sigma_1)^{-1}G_p^+\sigma_1'\sigma_1,(\sigma_2'\sigma_2)^{-1}G_p^-\sigma_2'\sigma_2 ) = (\sigma_1^{-1}G_p^+\sigma_1, \sigma_2^{-1}G_p^-\sigma_2) = (\gamma_p, \gamma_p).
	\end{equation}
	The function $\Sym(p)$ is the number of elements in $S_k$ which commute with $G_p$. This only depends on the conjugacy class $C_p$, or equivalently, the partition $p$.
\end{proof}

The functions $Z^H_{p}$, which count the number of elements in the conjugacy class labelled by $p$, are of central importance in equation \ref{eq: N(m+,m-) Cycle Index Formula}.
\begin{definition}[Cycle index]
	For a partition $p \vdash l$, let $\mathbf{x}^p$ be the degree $l$ monomial $x_1^{p_1}x_2^{p_2}\dots$, where $\sum_j jp_j = l$. We construct generating functions, called cycle indices
	\begin{equation}
		Z^H(\mathbf{x}) = Z^H(x_1,x_2,\dots) = \frac{1}{|H|} \sum_p Z_p^H \mathbf{x}^p, 
	\end{equation}
	such that
	\begin{equation}
		\frac{1}{|H|}Z^H_p = \text{Coefficient}(Z^H(x_1,x_2,\dots), \mathbf{x}^p).
	\end{equation}
\end{definition}

We are interested in the cycle indices $Z^{G(\vec{k}^+,\vec{k}^-)}$. To efficiently describe them we define the following compact notation.
\begin{definition}[Exponential notation for vector partition]
	A vector partition
	\begin{equation}
		(\vec{k}^+,\vec{k}^-) = (k_1^+, k_1^-) + \dots +  (k_l^+, k_l^-).
	\end{equation}
	can equivalently be described using a generalization of exponential notation for partitions,
	\begin{equation}
		(\vec{k}^+,\vec{k}^-) = p_{01}(0,1) + p_{10}(1,0) + \dots = \sum_{v^{(2)}} p_{v^{(2)}} v^{(2)},
	\end{equation}
	where the sum is over ordered lists of two integers $v^{(2)}$ with at least one non-zero entry and $p_{v^{(2)}}$ is the number of times it appears in the vector partition.
\end{definition}
\begin{definition}[Wreath product]
	A general wreath product $S_l[S_v]$ is a semi-direct product 
	\begin{equation}
		S_l \ltimes \underbrace{(S_v \times \dots \times S_v)}_{\text{l factors}},
	\end{equation}
	which is naturally viewed as a subgroup of $S_{lv}$.
	For example, elements of $S_4[S_2]$ correspond to diagrams
	\begin{equation}
		\vcenter{\hbox{\begin{tikzpicture}[scale=1]
					\def \k {3}
					\def \m {7}
					\def \sep {0.5}
					\def \voffset {0.25}
					%		\pgfmathparse{(\sep*(\m)-2*\voffset)/\k};
					\pgfmathparse{2*\sep};
					\pgfmathsetmacro{\vsep}{\pgfmathresult};
					\pgfmathint{\k-1};
					\pgfmathsetmacro{\kk}{\pgfmathresult};
					\foreach \v in {0,...,\k}
					{
						\pgfmathparse{\v*\vsep+\voffset}
						\coordinate (v\v) at (\pgfmathresult,.5) {};
					}
					\foreach \v in {0,...,\k}
					{
						\pgfmathparse{\v*\vsep+\voffset}
						\coordinate (w\v) at (\pgfmathresult,-.5) {};
						\draw[] (w\v) -- (v\v);
					}
					\foreach \eOutm in {0,...,\m}
					{
						\pgfmathint{\eOutm+1};
						\pgfmathsetmacro{\seOutm}{\pgfmathresult};
						\pgfmathparse{\eOutm*\sep};
						\coordinate (eom\seOutm) at (\pgfmathresult,1.2);
					}
					\foreach \eOutm in {0,...,\m}
					{
						\pgfmathint{\eOutm+1};
						\pgfmathsetmacro{\seOutm}{\pgfmathresult};
						\pgfmathparse{\eOutm*\sep};
						\coordinate (eomm\seOutm) at (\pgfmathresult,1);
						%			\node[right, node distance = 0pt and 0pt] at (eomm\seOutm) {\tiny \seOutm};
						%			\node[right, node distance = 0pt and 0pt] at (eomm\seOutm) {};
						\draw[] (eom\seOutm) -- (eomm\seOutm);
					}
					\foreach \eInm in {0,...,\m}
					{
						\pgfmathint{\eInm+1};
						\pgfmathsetmacro{\seInm}{\pgfmathresult};
						\pgfmathparse{\eInm*\sep};
						\coordinate (eim\seInm) at (\pgfmathresult,-1.8);
					}
					\foreach \eInm in {0,...,\m}
					{
						\pgfmathint{\eInm+1};
						\pgfmathsetmacro{\seInm}{\pgfmathresult};
						\pgfmathparse{\eInm*\sep};
						\coordinate (eimm\seInm) at (\pgfmathresult,-1);
						%			\node[right, node distance = 0pt and 0pt] at (eimm\seInm) {\tiny \seInm};
						%			\node[right, node distance = 0pt and 0pt] at (eimm\seInm) {};
						\draw[] (eim\seInm) -- (eimm\seInm);
					}
					\begin{scope}[decoration={markings}]
						\draw[postaction={decorate}] (v0) -- (eomm1);
						\draw[postaction={decorate}] (v0) -- (eomm2);
						\draw[postaction={decorate}] (v1) -- (eomm3);
						\draw[postaction={decorate}] (v1) -- (eomm4);
						\draw[postaction={decorate}] (v2) -- (eomm5);
						\draw[postaction={decorate}] (v2) -- (eomm6);
						\draw[postaction={decorate}] (v3) -- (eomm7);
						\draw[postaction={decorate}] (v3) -- (eomm8);
						\draw[postaction={decorate}] (eimm1)  -- (w0);
						\draw[postaction={decorate}] (eimm2)  -- (w0);
						\draw[postaction={decorate}] (eimm3)  -- (w1);
						\draw[postaction={decorate}] (eimm4)  -- (w1);
						\draw[postaction={decorate}] (eimm5)  -- (w2);
						\draw[postaction={decorate}] (eimm6)  -- (w2);
						\draw[postaction={decorate}] (eimm7)  -- (w3);
						\draw[postaction={decorate}] (eimm8)  -- (w3);
					\end{scope}
					\draw[fill=white] ($(v0)+(-0.5,-0.2)$) rectangle node{$\mu$} ($(w3)+(0.5,0.2)$);
					\draw[fill=white] ($(eimm1)-(0.1,.5)$) rectangle node{$\nu_1$} ($(eimm2)-(-0.1,0.1)$);
					\draw[fill=white] ($(eimm3)-(0.1,.5)$) rectangle node{$\nu_2$} ($(eimm4)-(-0.1,0.1)$);
					\draw[fill=white] ($(eimm5)-(0.1,.5)$) rectangle node{$\nu_3$} ($(eimm6)-(-0.1,0.1)$);
					\draw[fill=white] ($(eimm7)-(0.1,.5)$) rectangle node{$\nu_4$} ($(eimm8)-(-0.1,0.1)$);
		\end{tikzpicture}}}\label{eqn: wreath_diagram}
	\end{equation}
	with $\nu_i \in S_2, \mu \in S_4$. The vertices are concatenations of edges and $\mu$ permutes the resulting collections
	\begin{equation}
		\vcenter{\hbox{\scalebox{2}[2]{\begin{tikzpicture}[scale=1]
						\draw (-0.25,-0.3) -- (0,0) -- (-0.0,0.3);
						\draw (-0.15,-0.3) -- (0,0) -- (-0.0,0.3);
						\draw (0.15,-0.3) -- (0,0) -- (0.0,0.3);
						\draw (0.25,-0.3) -- (0,0) -- (0.0,0.3);
		\end{tikzpicture}}}}
		\longleftrightarrow
		\vcenter{\hbox{\scalebox{2}[2]{\begin{tikzpicture}[scale=1]
						\draw (-0.25,-0.3) -- (-0.05,0) -- (-0.05,0.3);
						\draw (-0.15,-0.3) -- (-0.02,0) -- (-0.02,0.3);
						\draw (0.15,-0.3) -- (0.02,0) -- (0.02,0.3);
						\draw (0.25,-0.3) -- (0.05,0) -- (0.05,0.3);
		\end{tikzpicture}}}}
	\end{equation}
\end{definition}
\begin{example}
	The illustration of $S_4[S_2]$ in \eqref{eqn: wreath_diagram} can be made more explicit by embedding it into $S_8$. Labelling the edges at the bottom from left to right by $1,\dots,8$ we have the subgroups
	\begin{equation}
		\nu_1 \in \Perms({1,2}), \nu_2 \in \Perms({3,4}), \nu_3 \in \Perms({5,6}), \nu_4 \in \Perms({7,8}),
	\end{equation}
	and $\mu$ is in the subgroup generated by the swaps
	\begin{equation}
		S_4 \cong \langle (13)(24), (35)(46), (57)(68) \rangle.
	\end{equation}
\end{example}
With these definitions, we can write the symmetry group as
\begin{equation}
	G(\vec{k}^+,\vec{k}^-) = \bigtimes_{v^{(2)}} S_{p_{v^{(2)}}}[S_{v^{(2)}}], \quad v^{(2)} \in \mathbb{N}^{\times 2}\left\backslash \{0,0\}\right. \label{eqn: sym grp Gkk}
\end{equation}
where
\begin{equation}
	S_{v^{(2)}} = S_{v^{(2)}_1} \times S_{v^{(2)}_2},
\end{equation}
and $S_{p_{v^{(2)}}}[S_{v^{(2)}}]$ is the wreath product.

The form \eqref{eqn: sym grp Gkk} of the symmetry group is particularly useful for simplifying the computation of its cycle index.
\begin{proposition} \label{prop: cycle index of prod}
	Let $G_1 \times G_2$ be a subgroup of $S_k = \Perms(\{1,\dots,k\})$ for some $k$. Then the cycle index of the direct product is a product of cycle indices
	\begin{equation}
		Z^{G_1 \times G_2}(\mathbf{x}) = Z^{{G_1}}(\mathbf{x})Z^{{G_2}}(\mathbf{x}). \label{eqn: Cycle Index Product Group}
	\end{equation}
\end{proposition}
\begin{proof}
	Suppose $g_1 \in G_1$ has cycle structure $p$ where $p_i$ is the number of $i$-cycles, and $g_2 \in G_2$ has cycle structure $q$. The cycle structure of the product $g_1 g_2$ is $r=(p_1+q_1, p_2+q_2, \dots)$. Therefore, the contribution of $(g_1, g_2) \in G_1 \times G_2$ to the cycle index is
	\begin{equation}
		\mathbf{x}^p \mathbf{x}^q.
	\end{equation}
	Summing up the contributions from every element gives
	\begin{equation}
		Z^{G_1 \times G_2}(\mathbf{x}) = Z^{{G_1}}(\mathbf{x})Z^{{G_2}}(\mathbf{x}).
	\end{equation}
\end{proof}
It is convenient to formally think of $\mathbf{x} = (x_1,x_2,\dots)$ as a countably infinite number of variables. In practice it truncates at $x_c$, where $c$ is the size of the largest cycle in ${G_1 \times G_2}$.

Proposition \ref{prop: cycle index of prod} implies that we can compute the cycle indices of each wreath group separately. Cycle indices of wreath products can be computed as follows.
\begin{theorem}\label{thm: cycle index of wreath}	
	The cycle index of a wreath product $S_l[S_v]$ is
	\begin{equation}
		Z^{S_l[S_v]}(x_1,\dots,x_{lv}) = Z^{S_l}(Z_1^{S_v}(\mathbf{x}), \dots, Z_{l}^{S_v}(\mathbf{x})),
	\end{equation}
	where
	\begin{equation}
		Z_i^{S_v}(\mathbf{x}) = Z^{S_v}(x_{1\cdot i}, x_{2 \cdot i}, \dots,x_{v\cdot i}),
	\end{equation}
	is given by multiplying the labels on $x_1,x_2, \dots$ by $i$ in the cycle index.
\end{theorem}
\begin{proof}
	This result originally proved by Pólya in \cite{Polya} says: for a permutation \eqref{eqn: wreath_diagram} with $\mu$ fixed to have cycle structure $p \vdash l$, the contribution to the cycle index as we sum over $S_v^{\times l}$ is \cite{Constantine}
	\begin{equation}
		\frac{1}{\abs{S_l}}Z_1^{S_v}(\mathbf{x})^{p_1} \dots Z_l^{S_v}(\mathbf{x})^{p_l}.
	\end{equation}
\end{proof}
We are interested in counting cycles of wreath products of the form $S_l[S_{v^+} \times S_{v^-}]$. This wreath product is most naturally thought of as a subgroup of $S_{lv^+ + lv^-}$. However, elements in $S_l[S_{v^+} \times S_{v^-}]$ are determined by $\mu \in S_l, \nu^+_i \in S_{v^+}, \nu^-_i \in S_{v^-}$ according to the diagram (in the case of $S_4[S_2 \times S_2]$)
\begin{equation}
	\vcenter{\hbox{\begin{tikzpicture}[scale=2]
				\def \k {3}
				\def \m {15}
				\def \sep {0.25}
				\def \voffset {0.375}
				\pgfmathparse{(\sep*(\m)-2*\voffset)/\k};
				\pgfmathsetmacro{\vsep}{\pgfmathresult};
				\pgfmathint{\k-1};
				\pgfmathsetmacro{\kk}{\pgfmathresult};
				\foreach \v in {0,...,\k}
				{
					\pgfmathparse{\v*\vsep+\voffset}
					\coordinate (v\v) at (\pgfmathresult,.5) {};
				}
				\foreach \v in {0,...,\k}
				{
					\pgfmathparse{\v*\vsep+\voffset}
					\coordinate (w\v) at (\pgfmathresult,-.5) {};
					\draw[] (w\v) -- (v\v);
				}
				\foreach \eOutm in {0,...,\m}
				{
					\pgfmathint{\eOutm+1};
					\pgfmathsetmacro{\seOutm}{\pgfmathresult};
					\pgfmathparse{\eOutm*\sep};
					\coordinate (eom\seOutm) at (\pgfmathresult,1.2);
				}
				\foreach \eOutm in {0,...,\m}
				{
					\pgfmathint{\eOutm+1};
					\pgfmathsetmacro{\seOutm}{\pgfmathresult};
					\pgfmathparse{\eOutm*\sep};
					\coordinate (eomm\seOutm) at (\pgfmathresult,1);
					%			\node[right, node distance = 0pt and 0pt] at (eomm\seOutm) {\tiny \seOutm};
					%			\node[right, node distance = 0pt and 0pt] at (eomm\seOutm) {};
					\draw[] (eom\seOutm) -- (eomm\seOutm);
				}
				\foreach \eInm in {0,...,\m}
				{
					\pgfmathint{\eInm+1};
					\pgfmathsetmacro{\seInm}{\pgfmathresult};
					\pgfmathparse{\eInm*\sep};
					\coordinate (eim\seInm) at (\pgfmathresult,-1.8);
				}
				\foreach \eInm in {0,...,\m}
				{
					\pgfmathint{\eInm+1};
					\pgfmathsetmacro{\seInm}{\pgfmathresult};
					\pgfmathparse{\eInm*\sep};
					\coordinate (eimm\seInm) at (\pgfmathresult,-1);
					%			\node[right, node distance = 0pt and 0pt] at (eimm\seInm) {\tiny \seInm};
					%			\node[right, node distance = 0pt and 0pt] at (eimm\seInm) {};
					\draw[] (eim\seInm) -- (eimm\seInm);
				}
				\begin{scope}[decoration={markings}]
					\draw[postaction={decorate}] (v0) -- (eomm1);
					\draw[postaction={decorate}] (v0) -- (eomm2);
					\draw[postaction={decorate}] (v0) -- (eomm3);
					\draw[postaction={decorate}] (v0) -- (eomm4);
					\draw[postaction={decorate}] (v1) -- (eomm5);
					\draw[postaction={decorate}] (v1) -- (eomm6);
					\draw[postaction={decorate}] (v1) -- (eomm7);
					\draw[postaction={decorate}] (v1) -- (eomm8);
					\draw[postaction={decorate}] (v2) -- (eomm9);
					\draw[postaction={decorate}] (v2) -- (eomm10);
					\draw[postaction={decorate}] (v2) -- (eomm11);
					\draw[postaction={decorate}] (v2) -- (eomm12);
					\draw[postaction={decorate}] (v3) -- (eomm13);
					\draw[postaction={decorate}] (v3) -- (eomm14);
					\draw[postaction={decorate}] (v3) -- (eomm15);
					\draw[postaction={decorate}] (v3) -- (eomm16);
					\draw[postaction={decorate}] (eimm1)  -- (w0);
					\draw[postaction={decorate}] (eimm2)  -- (w0);
					\draw[postaction={decorate}] (eimm3)  -- (w0);
					\draw[postaction={decorate}] (eimm4)  -- (w0);
					\draw[postaction={decorate}] (eimm5)  -- (w1);
					\draw[postaction={decorate}] (eimm6)  -- (w1);
					\draw[postaction={decorate}] (eimm7)  -- (w1);
					\draw[postaction={decorate}] (eimm8)  -- (w1);
					\draw[postaction={decorate}] (eimm9)  -- (w2);
					\draw[postaction={decorate}] (eimm10)  -- (w2);
					\draw[postaction={decorate}] (eimm11)  -- (w2);
					\draw[postaction={decorate}] (eimm12)  -- (w2);
					\draw[postaction={decorate}] (eimm13)  -- (w3);
					\draw[postaction={decorate}] (eimm14)  -- (w3);
					\draw[postaction={decorate}] (eimm15)  -- (w3);
					\draw[postaction={decorate}] (eimm16)  -- (w3);
				\end{scope}
				\draw[fill=white] ($(v0)+(-0.5,-0.2)$) rectangle node{$\mu$} ($(w3)+(0.5,0.2)$);
				\draw[fill=white] ($(eimm1)-(0.1,.5)$) rectangle node{$\nu_1^+$} ($(eimm2)-(-0.1,0.1)$);
				\draw[fill=white] ($(eimm3)-(0.1,.5)$) rectangle node{$\nu_1^-$} ($(eimm4)-(-0.1,0.1)$);
				\draw[fill=white] ($(eimm5)-(0.1,.5)$) rectangle node{$\nu_2^+$} ($(eimm6)-(-0.1,0.1)$);
				\draw[fill=white] ($(eimm7)-(0.1,.5)$) rectangle node{$\nu_2^-$} ($(eimm8)-(-0.1,0.1)$);
				\draw[fill=white] ($(eimm9)-(0.1,.5)$) rectangle node{$\nu_3^+$} ($(eimm10)-(-0.1,0.1)$);
				\draw[fill=white] ($(eimm11)-(0.1,.5)$) rectangle node{$\nu_3^-$} ($(eimm12)-(-0.1,0.1)$);
				\draw[fill=white] ($(eimm13)-(0.1,.5)$) rectangle node{$\nu_4^+$} ($(eimm14)-(-0.1,0.1)$);
				\draw[fill=white] ($(eimm15)-(0.1,.5)$) rectangle node{$\nu_4^-$} ($(eimm16)-(-0.1,0.1)$);
	\end{tikzpicture}}}\label{eqn: wreath_diagram2}
\end{equation}
which can be factorized as
\begin{equation}
	\vcenter{\hbox{{\begin{tikzpicture}[scale=1]
					\def \k {3}
					\def \m {7}
					\def \sep {0.5}
					\def \voffset {0.25}
					\pgfmathparse{(\sep*(\m)-2*\voffset)/\k};
					\pgfmathsetmacro{\vsep}{\pgfmathresult};
					\pgfmathint{\k-1};
					\pgfmathsetmacro{\kk}{\pgfmathresult};
					\foreach \v in {0,...,\k}
					{
						\pgfmathparse{\v*\vsep+\voffset}
						\coordinate (v\v) at (\pgfmathresult,.5) {};
					}
					\foreach \v in {0,...,\k}
					{
						\pgfmathparse{\v*\vsep+\voffset}
						\coordinate (w\v) at (\pgfmathresult,-.5) {};
						\draw[] (w\v) -- (v\v);
					}
					\foreach \eOutm in {0,...,\m}
					{
						\pgfmathint{\eOutm+1};
						\pgfmathsetmacro{\seOutm}{\pgfmathresult};
						\pgfmathparse{\eOutm*\sep};
						\coordinate (eom\seOutm) at (\pgfmathresult,1.2);
					}
					\foreach \eOutm in {0,...,\m}
					{
						\pgfmathint{\eOutm+1};
						\pgfmathsetmacro{\seOutm}{\pgfmathresult};
						\pgfmathparse{\eOutm*\sep};
						\coordinate (eomm\seOutm) at (\pgfmathresult,1);
						%			\node[right, node distance = 0pt and 0pt] at (eomm\seOutm) {\tiny \seOutm};
						%			\node[right, node distance = 0pt and 0pt] at (eomm\seOutm) {};
						\draw[] (eom\seOutm) -- (eomm\seOutm);
					}
					\foreach \eInm in {0,...,\m}
					{
						\pgfmathint{\eInm+1};
						\pgfmathsetmacro{\seInm}{\pgfmathresult};
						\pgfmathparse{\eInm*\sep};
						\coordinate (eim\seInm) at (\pgfmathresult,-1.8);
					}
					\foreach \eInm in {0,...,\m}
					{
						\pgfmathint{\eInm+1};
						\pgfmathsetmacro{\seInm}{\pgfmathresult};
						\pgfmathparse{\eInm*\sep};
						\coordinate (eimm\seInm) at (\pgfmathresult,-1);
						%			\node[right, node distance = 0pt and 0pt] at (eimm\seInm) {\tiny \seInm};
						%			\node[right, node distance = 0pt and 0pt] at (eimm\seInm) {};
						\draw[] (eim\seInm) -- (eimm\seInm);
					}
					\begin{scope}[decoration={markings}]
						\draw[postaction={decorate}] (v0) -- (eomm1);
						\draw[postaction={decorate}] (v0) -- (eomm2);
						\draw[postaction={decorate}] (v1) -- (eomm3);
						\draw[postaction={decorate}] (v1) -- (eomm4);
						\draw[postaction={decorate}] (v2) -- (eomm5);
						\draw[postaction={decorate}] (v2) -- (eomm6);
						\draw[postaction={decorate}] (v3) -- (eomm7);
						\draw[postaction={decorate}] (v3) -- (eomm8);
						\draw[postaction={decorate}] (eimm1)  -- (w0);
						\draw[postaction={decorate}] (eimm2)  -- (w0);
						\draw[postaction={decorate}] (eimm3)  -- (w1);
						\draw[postaction={decorate}] (eimm4)  -- (w1);
						\draw[postaction={decorate}] (eimm5)  -- (w2);
						\draw[postaction={decorate}] (eimm6)  -- (w2);
						\draw[postaction={decorate}] (eimm7)  -- (w3);
						\draw[postaction={decorate}] (eimm8)  -- (w3);
					\end{scope}
					\draw[fill=white] ($(v0)+(-0.5,-0.2)$) rectangle node{$\mu$} ($(w3)+(0.5,0.2)$);
					\draw[fill=white] ($(eimm1)-(0.1,.6)$) rectangle node{$\nu_1^+$} ($(eimm2)-(-0.1,0.1)$);
					\draw[fill=white] ($(eimm3)-(0.1,.6)$) rectangle node{$\nu_2^+$} ($(eimm4)-(-0.1,0.1)$);
					\draw[fill=white] ($(eimm5)-(0.1,.6)$) rectangle node{$\nu_3^+$} ($(eimm6)-(-0.1,0.1)$);
					\draw[fill=white] ($(eimm7)-(0.1,.6)$) rectangle node{$\nu_4^+$} ($(eimm8)-(-0.1,0.1)$);
			\end{tikzpicture}}~{\begin{tikzpicture}[scale=1]
					\def \k {3}
					\def \m {7}
					\def \sep {0.5}
					\def \voffset {0.25}
					\pgfmathparse{(\sep*(\m)-2*\voffset)/\k};
					\pgfmathsetmacro{\vsep}{\pgfmathresult};
					\pgfmathint{\k-1};
					\pgfmathsetmacro{\kk}{\pgfmathresult};
					\foreach \v in {0,...,\k}
					{
						\pgfmathparse{\v*\vsep+\voffset}
						\coordinate (v\v) at (\pgfmathresult,.5) {};
					}
					\foreach \v in {0,...,\k}
					{
						\pgfmathparse{\v*\vsep+\voffset}
						\coordinate (w\v) at (\pgfmathresult,-.5) {};
						\draw[] (w\v) -- (v\v);
					}
					\foreach \eOutm in {0,...,\m}
					{
						\pgfmathint{\eOutm+1};
						\pgfmathsetmacro{\seOutm}{\pgfmathresult};
						\pgfmathparse{\eOutm*\sep};
						\coordinate (eom\seOutm) at (\pgfmathresult,1.2);
					}
					\foreach \eOutm in {0,...,\m}
					{
						\pgfmathint{\eOutm+1};
						\pgfmathsetmacro{\seOutm}{\pgfmathresult};
						\pgfmathparse{\eOutm*\sep};
						\coordinate (eomm\seOutm) at (\pgfmathresult,1);
						%			\node[right, node distance = 0pt and 0pt] at (eomm\seOutm) {\tiny \seOutm};
						%			\node[right, node distance = 0pt and 0pt] at (eomm\seOutm) {};
						\draw[] (eom\seOutm) -- (eomm\seOutm);
					}
					\foreach \eInm in {0,...,\m}
					{
						\pgfmathint{\eInm+1};
						\pgfmathsetmacro{\seInm}{\pgfmathresult};
						\pgfmathparse{\eInm*\sep};
						\coordinate (eim\seInm) at (\pgfmathresult,-1.8);
					}
					\foreach \eInm in {0,...,\m}
					{
						\pgfmathint{\eInm+1};
						\pgfmathsetmacro{\seInm}{\pgfmathresult};
						\pgfmathparse{\eInm*\sep};
						\coordinate (eimm\seInm) at (\pgfmathresult,-1);
						%			\node[right, node distance = 0pt and 0pt] at (eimm\seInm) {\tiny \seInm};
						%			\node[right, node distance = 0pt and 0pt] at (eimm\seInm) {};
						\draw[] (eim\seInm) -- (eimm\seInm);
					}
					\begin{scope}[decoration={markings}]
						\draw[postaction={decorate}] (v0) -- (eomm1);
						\draw[postaction={decorate}] (v0) -- (eomm2);
						\draw[postaction={decorate}] (v1) -- (eomm3);
						\draw[postaction={decorate}] (v1) -- (eomm4);
						\draw[postaction={decorate}] (v2) -- (eomm5);
						\draw[postaction={decorate}] (v2) -- (eomm6);
						\draw[postaction={decorate}] (v3) -- (eomm7);
						\draw[postaction={decorate}] (v3) -- (eomm8);
						\draw[postaction={decorate}] (eimm1)  -- (w0);
						\draw[postaction={decorate}] (eimm2)  -- (w0);
						\draw[postaction={decorate}] (eimm3)  -- (w1);
						\draw[postaction={decorate}] (eimm4)  -- (w1);
						\draw[postaction={decorate}] (eimm5)  -- (w2);
						\draw[postaction={decorate}] (eimm6)  -- (w2);
						\draw[postaction={decorate}] (eimm7)  -- (w3);
						\draw[postaction={decorate}] (eimm8)  -- (w3);
					\end{scope}
					\draw[fill=white] ($(v0)+(-0.5,-0.2)$) rectangle node{$\mu$} ($(w3)+(0.5,0.2)$);
					\draw[fill=white] ($(eimm1)-(0.1,.6)$) rectangle node{$\nu_1^-$} ($(eimm2)-(-0.1,0.1)$);
					\draw[fill=white] ($(eimm3)-(0.1,.6)$) rectangle node{$\nu_2^-$} ($(eimm4)-(-0.1,0.1)$);
					\draw[fill=white] ($(eimm5)-(0.1,.6)$) rectangle node{$\nu_3^-$} ($(eimm6)-(-0.1,0.1)$);
					\draw[fill=white] ($(eimm7)-(0.1,.6)$) rectangle node{$\nu_4^-$} ($(eimm8)-(-0.1,0.1)$);
	\end{tikzpicture}}}}\label{eqn: wreath_diagram3}
\end{equation}
This amounts to embedding $S_l[S_{v^+} \times S_{v^-}]$ as a subgroup of $S_{lv^+} \times S_{lv^-}$.
From the double coset \eqref{eqn:1colordoublecoset} we can see that this is the type of embedding we are interested in. By using Theorem \ref{thm: cycle index of wreath} for the cycle index of a wreath product, we can separately keep track of the cycle structure of the left and right diagram in \eqref{eqn: wreath_diagram3}. For $\mu$ with fixed cycle structure $p \vdash l$, the contribution of the cycle index for $S_l[S_{v^+} \times S_{v^-}]$ as embedded into $S_{lv^+} \times S_{lv^-}$ is simply the product of the contribution from each. That is
\begin{equation}
	\frac{1}{\abs{S_l}}\qty[Z_1^{S_{v^+}}(\mathbf{x})^{p_1}\dots Z_l^{S_{v^+}}(\mathbf{x})^{p_l}]\qty[Z_1^{S_{v^-}}(\mathbf{y})^{p_1}\dots Z_l^{S_{v^-}}(\mathbf{y})^{p_l}]. \label{eq: xy cycle index 1}
\end{equation}
If we sum over all $\mu \in S_k$ we get the generating function
\begin{equation}
	Z^{S_l[S_{v^+} \times S_{v^-}]}(x_1,\dots,x_{lv^+ + lv^-}) = Z^{S_l}(Z_1^{S_{v^+}}(\mathbf{x})Z_1^{S_{v^-}}(\mathbf{y}), \dots, Z_{l}^{S_{v^+}}(\mathbf{x})Z_l^{S_{v^-}}(\mathbf{y})).
\end{equation}

Returning to the case at hand, $S_{p_{v^{(2)}}}\qty[S_{v^{(2)}}]$ is considered as a subgroup of $S_k^+\times S_k^-$. Our goal is to count the number of elements $(\sigma_k^+, \sigma_k^-) \in S_{p_{v^{(2)}}}\qty[S_{v^{(2)}}]$ with cycle structure $p_1,p_2$, respectively. To that end, we construct the refined version as
\begin{align}
	Z^{S_{p_{v^{(2)}}}\qty[S_{v^{(2)}}]}&(\mathbf{x},\mathbf{y})\\ &=Z^{S_{p_{v^{(2)}}}}\Big(Z^{S_{v^{(2)}_1}}_1(\mathbf{x})Z^{S_{v^{(2)}_2}}_1(\mathbf{y}), \dots, Z^{S_{v^{(2)}_1}}_{p_{v^{(2)}}}(\mathbf{x})Z^{S_{v^{(2)}_2}}_{p_{v^{(2)}}}(\mathbf{y})\Big)
\end{align}
Then the number of elements in $S_{p_{v^{(4)}}}\qty[S_{v^{(2)}}]$ with cycle structure $p_1,p_2$ is
\begin{equation}
	\frac{1}{|S_{p_{v^{(2)}}}[S_{v^{(2)}}]|} Z^{S_{p_{v^{(2)}}}\qty[S_{v^{(2)}}]}_{p_1,p_2 } = \text{Coefficient}(Z^{S_{p_{v^{(2)}}}\qty[S_{v^{(2)}}]}(\mathbf{x},\mathbf{y} ), \mathbf{x}^{p_1} \mathbf{y}^{p_2}).
\end{equation}
For products of wreath products we can use the factorization property \eqref{eqn: Cycle Index Product Group}. Consequently, the full cycle index of $G(\vec{k}^+,\vec{k}^-)$ is given by a product
\begin{equation}
	Z^{G(\vec{k}^+,\vec{k}^-)}(\mathbf{x},\mathbf{y}) = \prod_{v^{(2)}} Z^{S_{p_{v^{(2)}}}\qty[S_{v^{(2)}}]}(\mathbf{x},\mathbf{y}). \label{eq: ZG prod cycle indices}
\end{equation}

\begin{example}
\label{ex: cycle index}	
	It is instructive to calculate $N(\vec{k}^+,\vec{k}^-)$ for Figure \ref{fig: Two Sigma One Color Double Coset Graph}, where $(\vec{k}^+,\vec{k}^-) = (3,2) + (3,2) + (1,3)= 2(3,2)+(1,3)$. The first step is to write down $G(\vec{k}^+,\vec{k}^-)$ as a product of wreath products,
	\begin{equation}
		G(\vec{k}^+,\vec{k}^-) = S_2[S_3 \times S_2] \times S_1 \times S_3.
	\end{equation}
	Using the factorization property \eqref{eqn: Cycle Index Product Group} for cycle indices we have
	\begin{equation}
		Z^{G(\vec{k}^+,\vec{k}^-)}(\mathbf{x},\mathbf{y}) = Z^{S_2[S_3 \times S_2]}(\mathbf{x},\mathbf{y})Z^{S_1}(\mathbf{x})Z^{S_3}(\mathbf{y}),
	\end{equation}
	where
	\begin{equation}
		Z^{S_2[S_3 \times S_2]}(\mathbf{x},\mathbf{y}) = Z^{S_2}(Z_1^{S_3}(\mathbf{x})Z_1^{S_2}(\mathbf{y}), Z_2^{S_3}(\mathbf{x})Z_2^{S_2}(\mathbf{y})).
	\end{equation}
	The four relevant cycle indices are
	\begin{equation}
		Z^{S_0}(\mathbf{x})=1,\quad Z^{S_1}(\mathbf{x}) = x_1, \quad Z^{S_2}(\mathbf{x}) = \frac{1}{2}(x_1^2 + x_2), \quad Z^{S_3}(\mathbf{x}) = \frac{1}{6}(x_1^3 + 3x_2x_1+2x_3).
	\end{equation}
	Explicitly, the cycle index for the wreath product is
	\begin{equation}
		Z^{S_2[S_3 \times S_2]}(\mathbf{x},\mathbf{y}) = \frac{1}{2}\qty(\frac{(x_1^3 + 3x_2x_1+2x_3)}{6})^2\qty(\frac{(y_1^2 + y_2)}{2})^2 + \frac{1}{2}\frac{(x_2^3 + 3x_4x_2+2x_6)}{6}\frac{(y_2^2 + y_4)}{2}.
	\end{equation}
	To perform the sum in equation \eqref{eq: N(m+,m-) Cycle Index Formula}, we need to pick out the coefficients
	\begin{equation}
		\text{Coefficient}(Z^{G(\vec{k}^+,\vec{k}^-)}, \mathbf{x}^p\mathbf{y}^p), \quad p \vdash 7.
	\end{equation}
	There are seven non-zero coefficients of this form,
	\begin{align} \nonumber
		\text{Coefficient}(Z^{G(\vec{k}^+,\vec{k}^-)}, x_1^7y_1^7) &= \frac{1}{1728}, \\ \nonumber
		\text{Coefficient}(Z^{G(\vec{k}^+,\vec{k}^-)}, x_1^5x_2y_1^5 y_2) &=  \frac{5}{288}, \\ \nonumber
		\text{Coefficient}(Z^{G(\vec{k}^+,\vec{k}^-)}, x_1^4x_3y_1^4 y_3) &=  \frac{1}{216}, \\ \nonumber
		\text{Coefficient}(Z^{G(\vec{k}^+,\vec{k}^-)}, x_1^3x_2^2y_1^3 y_2^2) &=  \frac{7}{192}, \\ \nonumber
		\text{Coefficient}(Z^{G(\vec{k}^+,\vec{k}^-)}, x_1^2x_2x_3y_1^2 y_2y_3) &=  \frac{1}{36}, \\ \nonumber
		\text{Coefficient}(Z^{G(\vec{k}^+,\vec{k}^-)}, x_1^1x_2^3y_1^1 y_2^3) &=  \frac{1}{48}, \\
		\text{Coefficient}(Z^{G(\vec{k}^+,\vec{k}^-)}, x_1^1x_2^1x_4^1y_1^1y_2^1y_4) &=  \frac{1}{16}.
	\end{align}
	We find
	\begin{equation}
		N(\vec{k}^+,\vec{k}^-) = \frac{1}{1728}7!+\frac{5}{288}5!2+\frac{1}{216}4!3+\frac{7}{192}3!2!2^2+\frac{1}{36}2!2\cdot 3+\frac{1}{48}3!2^3+\frac{1}{16}2\cdot 4 = 11.
	\end{equation}
	
\end{example}

In this section we have discussed three ways of counting observables, with increasing level of refinement. Because the double coset counting is the most granular of the three, we expect appropriate sums over $N(\vec{k}^+,\vec{k}^-)$ to reproduce previous counting formulae. For example, by Proposition \ref{prop: graphs equals trace} the number of graphs with $k$ edges and exactly $l$ vertices is given by the difference
\begin{equation}
	\abs{\mathcal{G}_{k,l}} - \abs{\mathcal{G}_{k,l-1}} = \dim \, [V_l^{\otimes 2k}]^{S_{l} \times S_k} - \dim \, [V_{l-1}^{\otimes 2k}]^{S_{l-1} \times S_k}
\end{equation}
It is also given by a sum over those vector partitions which have exactly $l$ parts,
\begin{align} \nonumber
	\abs{\mathcal{G}_{k,l}} - \abs{\mathcal{G}_{k,l-1}} &= \sum_{\substack{(\vec{k}^+,\vec{k}^-)\vdash (k,k) \\ \text{with}~ l ~\text{parts.}}} N(\vec{k}^+,\vec{k}^-)\\ \nonumber
	&= \sum_{\substack{(\vec{k}^+,\vec{k}^-)\vdash (k,k) \\ \text{with}~ l ~\text{parts.}}}\frac{1}{|G(\vec{k}^+,\vec{k}^-)|} \sum_{p \vdash k} Z_{p,p}^{G(\vec{k}^+,\vec{k}^-)} \Sym(p).
\end{align}
This is a refinement of the counting in Section \ref{subsec: observables} (Table \ref{tab: Table of invariant dimensions}), as can be seen from Table \ref{tab: Table of invariant dimensions refined}.
\begin{table}[t!]
	\begin{center}
		\caption{Number of linearly independent invariants contained in $\Sym^k(\VN^{\otimes 2})$ with refinement on the number of vertices $l$ of the corresponding graphs.}
		\label{tab: Table of invariant dimensions refined}
		\begin{tabular}{c c c}
			\textbf{k} & \textbf{\# graphs} &\textbf{$l=1,2,\dots,2k$}\\
			\hline
			1 &  2 & 1,1 \\
			2 &  11 & 1,5,4,1\\
			3 &  52 & 1, 9, 21, 16, 4, 1\\
			4 &  296 & 1, 18, 71, 108, 71, 22, 4, 1\\
			5 &  1724&1, 27, 194, 491, 557, 326, 101, 22, 4, 1 \\
			6 &  11060&1, 43, 476, 1903, 3353, 3062, 1587, 497, 111, 22, 4, 1 \\
		\end{tabular}
	\end{center}
\end{table}

\section{Cycle index code}
SageMath \cite{sagemath} has several tools for computing cycle indices. The code below can be found at \href{https://github.com/adrianpadellaro/PhD-Thesis}{Link to GitHub Repository}. For example, the following cell computes the cycle index for $S_2$.
    \begin{tcolorbox}[breakable, size=fbox, boxrule=1pt, pad at break*=1mm,colback=cellbackground, colframe=cellborder]
\prompt{In}{incolor}{1}{\boxspacing}
\begin{Verbatim}[commandchars=\\\{\},fontsize=\small]
\PY{n}{k} \PY{o}{=} \PY{l+m+mi}{2}
\PY{n}{Z\PYZus{}k} \PY{o}{=} \PY{n}{SymmetricGroup}\PY{p}{(}\PY{n}{k}\PY{p}{)}\PY{o}{.}\PY{n}{cycle\PYZus{}index}\PY{p}{(}\PY{p}{)}
\PY{n}{Z\PYZus{}k}
\end{Verbatim}
\end{tcolorbox}
            \begin{tcolorbox}[breakable, size=fbox, boxrule=.5pt, pad at break*=1mm, opacityfill=0]
\prompt{Out}{outcolor}{1}{\boxspacing}
\begin{Verbatim}[commandchars=\\\{\},fontsize=\small]
1/2*p[1, 1] + 1/2*p[2]
\end{Verbatim}
\end{tcolorbox}
Note that SageMath computes cycle indices in terms of partitions/symmetric functions. To convert it into a polynomial we define a polynomial ring $\mathbb{Q}[x_1,x_2]$ (note that by default Sage includes $x_0$ in the polynomial ring)
    \begin{tcolorbox}[breakable, size=fbox, boxrule=1pt, pad at break*=1mm,colback=cellbackground, colframe=cellborder]
\prompt{In}{incolor}{2}{\boxspacing}
\begin{Verbatim}[commandchars=\\\{\},fontsize=\small]
\PY{n}{QX} \PY{o}{=} \PY{n}{PolynomialRing}\PY{p}{(}\PY{n}{QQ}\PY{p}{,} \PY{n}{k}\PY{o}{+}\PY{l+m+mi}{1}\PY{p}{,} \PY{l+s+s1}{\PYZsq{}}\PY{l+s+s1}{x}\PY{l+s+s1}{\PYZsq{}}\PY{p}{)}
\end{Verbatim}
\end{tcolorbox}

    \begin{tcolorbox}[breakable, size=fbox, boxrule=1pt, pad at break*=1mm,colback=cellbackground, colframe=cellborder]
\prompt{In}{incolor}{3}{\boxspacing}
\begin{Verbatim}[commandchars=\\\{\},fontsize=\small]
\PY{n}{QX}\PY{o}{.}\PY{n}{gens}\PY{p}{(}\PY{p}{)}
\end{Verbatim}
\end{tcolorbox}

            \begin{tcolorbox}[breakable, size=fbox, boxrule=.5pt, pad at break*=1mm, opacityfill=0]
\prompt{Out}{outcolor}{3}{\boxspacing}
\begin{Verbatim}[commandchars=\\\{\},fontsize=\small]
(x0, x1, x2)
\end{Verbatim}
\end{tcolorbox}
To compute the polynomial corresponding to the cycle index Z\_k we simply replace every instance of $1$ in a partition by $x_1$, every $2$ with $x_2$ and so on. 
    \begin{tcolorbox}[breakable, size=fbox, boxrule=1pt, pad at break*=1mm,colback=cellbackground, colframe=cellborder]
\prompt{In}{incolor}{4}{\boxspacing}
\begin{Verbatim}[commandchars=\\\{\},fontsize=\small]
\PY{n+nb}{sum}\PY{p}{(}\PY{n}{z}\PY{p}{[}\PY{l+m+mi}{1}\PY{p}{]}\PY{o}{*}\PY{n}{prod}\PY{p}{(}\PY{n}{QX}\PY{o}{.}\PY{n}{gens}\PY{p}{(}\PY{p}{)}\PY{p}{[}\PY{n}{i}\PY{p}{]} \PY{k}{for} \PY{n}{i} \PY{o+ow}{in} \PY{n}{z}\PY{p}{[}\PY{l+m+mi}{0}\PY{p}{]}\PY{p}{)} \PY{k}{for} \PY{n}{z} \PY{o+ow}{in} \PY{n}{Z\PYZus{}k}\PY{p}{)}
\end{Verbatim}
\end{tcolorbox}

            \begin{tcolorbox}[breakable, size=fbox, boxrule=.5pt, pad at break*=1mm, opacityfill=0]
\prompt{Out}{outcolor}{4}{\boxspacing}
\begin{Verbatim}[commandchars=\\\{\},fontsize=\small]
1/2*x1\^{}2 + 1/2*x2
\end{Verbatim}
\end{tcolorbox}
We turn this procedure into a function that takes a polynomial ring and a cycle index.        
    \begin{tcolorbox}[breakable, size=fbox, boxrule=1pt, pad at break*=1mm,colback=cellbackground, colframe=cellborder]
\prompt{In}{incolor}{5}{\boxspacing}
\begin{Verbatim}[commandchars=\\\{\},fontsize=\small]
\PY{k}{def} \PY{n+nf}{CycleIndexPolynomial}\PY{p}{(}\PY{n}{PolynomialRing}\PY{p}{,} \PY{n}{CycleIndex}\PY{p}{)}\PY{p}{:}
    \PY{k}{return} \PY{n+nb}{sum}\PY{p}{(}\PY{n}{z}\PY{p}{[}\PY{l+m+mi}{1}\PY{p}{]}\PY{o}{*}\PY{n}{prod}\PY{p}{(}\PY{n}{PolynomialRing}\PY{o}{.}\PY{n}{gens}\PY{p}{(}\PY{p}{)}\PY{p}{[}\PY{n}{i}\PY{p}{]} \PY{k}{for} \PY{n}{i} \PY{o+ow}{in} \PY{n}{z}\PY{p}{[}\PY{l+m+mi}{0}\PY{p}{]}\PY{p}{)} \PY{k}{for} \PY{n}{z} \PY{o+ow}{in} \PY{n}{CycleIndex}\PY{p}{)}
\end{Verbatim}
\end{tcolorbox}
We can check that this returns the expected polynomial for $S_2$
    \begin{tcolorbox}[breakable, size=fbox, boxrule=1pt, pad at break*=1mm,colback=cellbackground, colframe=cellborder]
\prompt{In}{incolor}{6}{\boxspacing}
\begin{Verbatim}[commandchars=\\\{\},fontsize=\small]
\PY{n}{CycleIndexPolynomial}\PY{p}{(}\PY{n}{QX}\PY{p}{,}\PY{n}{Z\PYZus{}k}\PY{p}{)}
\end{Verbatim}
\end{tcolorbox}
            \begin{tcolorbox}[breakable, size=fbox, boxrule=.5pt, pad at break*=1mm, opacityfill=0]
\prompt{Out}{outcolor}{6}{\boxspacing}
\begin{Verbatim}[commandchars=\\\{\},fontsize=\small]
1/2*x1\^{}2 + 1/2*x2
\end{Verbatim}
\end{tcolorbox}
The next step is to compute the shifted cycle indices $Z_i^{S_k}$. For this we need a polynomial ring QXX that includes $x_{1\cdot i}, \dots, x_{k\cdot i}$. We get the shifted cycle index polynomial by replacing every $1$ in the partitions of Z\_k  by $x_i$, every $2$ by $x_{2i}$ and so on.
    \begin{tcolorbox}[breakable, size=fbox, boxrule=1pt, pad at break*=1mm,colback=cellbackground, colframe=cellborder]
\prompt{In}{incolor}{7}{\boxspacing}
\begin{Verbatim}[commandchars=\\\{\},fontsize=\small]
\PY{n}{i} \PY{o}{=} \PY{l+m+mi}{3}
\PY{n}{QXX} \PY{o}{=} \PY{n}{PolynomialRing}\PY{p}{(}\PY{n}{QQ}\PY{p}{,} \PY{n}{i}\PY{o}{*}\PY{n}{k}\PY{o}{+}\PY{l+m+mi}{1}\PY{p}{,} \PY{l+s+s1}{\PYZsq{}}\PY{l+s+s1}{x}\PY{l+s+s1}{\PYZsq{}}\PY{p}{)}
\PY{n+nb}{sum}\PY{p}{(}\PY{n}{z}\PY{p}{[}\PY{l+m+mi}{1}\PY{p}{]}\PY{o}{*}\PY{n}{prod}\PY{p}{(}\PY{n}{QXX}\PY{o}{.}\PY{n}{gens}\PY{p}{(}\PY{p}{)}\PY{p}{[}\PY{n}{p}\PY{o}{*}\PY{n}{i}\PY{p}{]} \PY{k}{for} \PY{n}{p} \PY{o+ow}{in} \PY{n}{z}\PY{p}{[}\PY{l+m+mi}{0}\PY{p}{]}\PY{p}{)} \PY{k}{for} \PY{n}{z} \PY{o+ow}{in} \PY{n}{Z\PYZus{}k}\PY{p}{)}
\end{Verbatim}
\end{tcolorbox}

            \begin{tcolorbox}[breakable, size=fbox, boxrule=.5pt, pad at break*=1mm, opacityfill=0]
\prompt{Out}{outcolor}{7}{\boxspacing}
\begin{Verbatim}[commandchars=\\\{\},fontsize=\small]
1/2*x3\^{}2 + 1/2*x6
\end{Verbatim}
\end{tcolorbox}
Again, we turn this procedure into a function that takes a polynomial ring, a cycle index to be shifted and an integer (Shift) to shift by.
    \begin{tcolorbox}[breakable, size=fbox, boxrule=1pt, pad at break*=1mm,colback=cellbackground, colframe=cellborder]
\prompt{In}{incolor}{8}{\boxspacing}
\begin{Verbatim}[commandchars=\\\{\},fontsize=\small]
\PY{k}{def} \PY{n+nf}{ShiftedCycleIndexPolynomial}\PY{p}{(}\PY{n}{PolynomialRing}\PY{p}{,} \PY{n}{CycleIndex}\PY{p}{,} \PY{n}{Shift}\PY{p}{)}\PY{p}{:}
    \PY{k}{return} \PY{n+nb}{sum}\PY{p}{(}\PY{n}{z}\PY{p}{[}\PY{l+m+mi}{1}\PY{p}{]}\PY{o}{*}\PY{n}{prod}\PY{p}{(}\PY{n}{PolynomialRing}\PY{o}{.}\PY{n}{gens}\PY{p}{(}\PY{p}{)}\PY{p}{[}\PY{n}{p}\PY{o}{*}\PY{n}{Shift}\PY{p}{]} \PY{k}{for} \PY{n}{p} \PY{o+ow}{in} \PY{n}{z}\PY{p}{[}\PY{l+m+mi}{0}\PY{p}{]}\PY{p}{)} \PY{k}{for} \PY{n}{z} \PY{o+ow}{in} \PY{n}{CycleIndex}\PY{p}{)}
\end{Verbatim}
\end{tcolorbox}
We can confirm that this gives the expected result.
    \begin{tcolorbox}[breakable, size=fbox, boxrule=1pt, pad at break*=1mm,colback=cellbackground, colframe=cellborder]
\prompt{In}{incolor}{9}{\boxspacing}
\begin{Verbatim}[commandchars=\\\{\},fontsize=\small]
\PY{n}{ShiftedCycleIndexPolynomial}\PY{p}{(}\PY{n}{QXX}\PY{p}{,} \PY{n}{Z\PYZus{}k}\PY{p}{,} \PY{l+m+mi}{3}\PY{p}{)}
\end{Verbatim}
\end{tcolorbox}

            \begin{tcolorbox}[breakable, size=fbox, boxrule=.5pt, pad at break*=1mm, opacityfill=0]
\prompt{Out}{outcolor}{9}{\boxspacing}
\begin{Verbatim}[commandchars=\\\{\},fontsize=\small]
1/2*x3\^{}2 + 1/2*x6
\end{Verbatim}
\end{tcolorbox}
The next step is to compute the cycle index $Z^{S_l[S_k]}(\mathbf{x})$ (we will generalize in to $Z^{S_l[S_{k_1} \times S_{k_2}]}(\mathbf{x},\mathbf{y})$ in a minute). For this we compute the cycle index Z\_l and replace every instance of $1$ with the shifted Z\_k cycle index with shift $1$, every $2$ with  a shift by $2$ and so on. The largest unknown that can appear in this cycle index is $x_{l\cdot k}$, therefore we need a polynomial ring QXX containing this variable.
    \begin{tcolorbox}[breakable, size=fbox, boxrule=1pt, pad at break*=1mm,colback=cellbackground, colframe=cellborder]
\prompt{In}{incolor}{10}{\boxspacing}
\begin{Verbatim}[commandchars=\\\{\},fontsize=\small]
\PY{n}{l} \PY{o}{=} \PY{l+m+mi}{3}
\PY{n}{Z\PYZus{}l} \PY{o}{=} \PY{n}{SymmetricGroup}\PY{p}{(}\PY{n}{l}\PY{p}{)}\PY{o}{.}\PY{n}{cycle\PYZus{}index}\PY{p}{(}\PY{p}{)}
\PY{n}{QXX} \PY{o}{=} \PY{n}{PolynomialRing}\PY{p}{(}\PY{n}{QQ}\PY{p}{,} \PY{n}{l}\PY{o}{*}\PY{n}{k}\PY{o}{+}\PY{l+m+mi}{1}\PY{p}{,} \PY{l+s+s1}{\PYZsq{}}\PY{l+s+s1}{x}\PY{l+s+s1}{\PYZsq{}}\PY{p}{)}
\PY{n+nb}{sum}\PY{p}{(}\PY{n}{z}\PY{p}{[}\PY{l+m+mi}{1}\PY{p}{]}\PY{o}{*}\PY{n}{prod}\PY{p}{(}\PY{n}{ShiftedCycleIndexPolynomial}\PY{p}{(}\PY{n}{QXX}\PY{p}{,} \PY{n}{Z\PYZus{}k}\PY{p}{,} \PY{n}{i}\PY{p}{)} \PY{k}{for} \PY{n}{i} \PY{o+ow}{in} \PY{n}{z}\PY{p}{[}\PY{l+m+mi}{0}\PY{p}{]}\PY{p}{)} \PY{k}{for} \PY{n}{z} \PY{o+ow}{in} \PY{n}{Z\PYZus{}l}\PY{p}{)}
\end{Verbatim}
\end{tcolorbox}

            \begin{tcolorbox}[breakable, size=fbox, boxrule=.5pt, pad at break*=1mm, opacityfill=0]
\prompt{Out}{outcolor}{10}{\boxspacing}
\begin{Verbatim}[commandchars=\\\{\},fontsize=\small]
1/48*x1\^{}6 + 1/16*x1\^{}4*x2 + 3/16*x1\^{}2*x2\^{}2 + 7/48*x2\^{}3 + 1/8*x1\^{}2*x4 + 1/6*x3\^{}2 +
1/8*x2*x4 + 1/6*x6
\end{Verbatim}
\end{tcolorbox}
Below we generalize to $Z^{S_l[S_{k_1} \times S_{k_2}]}(\mathbf{x},\mathbf{y})$ with $k_1 = k_2$. For this we need a polynomial ring QYY containing $y_1, \dots, y_{l\cdot k}$. We will also need a polynomial ring QXY containing all these unknowns. To compute the just mentioned cycle index, we replace every instance of $1$ in the partitions of Z\_l by the product of shifted cycle indices, as in \eqref{eq: xy cycle index 1}.
    \begin{tcolorbox}[breakable, size=fbox, boxrule=1pt, pad at break*=1mm,colback=cellbackground, colframe=cellborder]
\prompt{In}{incolor}{11}{\boxspacing}
\begin{Verbatim}[commandchars=\\\{\},fontsize=\small]
\PY{n}{QYY} \PY{o}{=} \PY{n}{PolynomialRing}\PY{p}{(}\PY{n}{QQ}\PY{p}{,} \PY{n}{l}\PY{o}{*}\PY{n}{k}\PY{o}{+}\PY{l+m+mi}{1}\PY{p}{,} \PY{l+s+s1}{\PYZsq{}}\PY{l+s+s1}{y}\PY{l+s+s1}{\PYZsq{}}\PY{p}{)}
\PY{n}{QXY} \PY{o}{=} \PY{n}{PolynomialRing}\PY{p}{(}\PY{n}{QQ}\PY{p}{,} \PY{n}{l}\PY{o}{*}\PY{n}{k}\PY{o}{+}\PY{l+m+mi}{1}\PY{p}{,} \PY{n}{var\PYZus{}array}\PY{o}{=}\PY{p}{[}\PY{l+s+s1}{\PYZsq{}}\PY{l+s+s1}{x}\PY{l+s+s1}{\PYZsq{}}\PY{p}{,}\PY{l+s+s1}{\PYZsq{}}\PY{l+s+s1}{y}\PY{l+s+s1}{\PYZsq{}}\PY{p}{]}\PY{p}{)}
\PY{n+nb}{sum}\PY{p}{(}\PY{n}{z}\PY{p}{[}\PY{l+m+mi}{1}\PY{p}{]}\PY{o}{*}\PY{n}{prod}\PY{p}{(}\PY{n}{QXY}\PY{o}{.}\PY{n}{coerce}\PY{p}{(}\PY{n}{ShiftedCycleIndexPolynomial}\PY{p}{(}\PY{n}{QXX}\PY{p}{,} \PY{n}{Z\PYZus{}k}\PY{p}{,} \PY{n}{i}\PY{p}{)}\PY{p}{)}\PY{o}{*}\PY{n}{QXY}\PY{o}{.}\PY{n}{coerce}\PY{p}{(}\PY{n}{ShiftedCycleIndexPolynomial}\PY{p}{(}\PY{n}{QYY}\PY{p}{,} \PY{n}{Z\PYZus{}k}\PY{p}{,} \PY{n}{i}\PY{p}{)}\PY{p}{)}  \PY{k}{for} \PY{n}{i} \PY{o+ow}{in} \PY{n}{z}\PY{p}{[}\PY{l+m+mi}{0}\PY{p}{]}\PY{p}{)}  \PY{k}{for} \PY{n}{z} \PY{o+ow}{in} \PY{n}{Z\PYZus{}l}\PY{p}{)}
\end{Verbatim}
\end{tcolorbox}
This produces the following polynomial.
            \begin{tcolorbox}[breakable, size=fbox, boxrule=.5pt, pad at break*=1mm, opacityfill=0]
\prompt{Out}{outcolor}{11}{\boxspacing}
\begin{Verbatim}[commandchars=\\\{\},fontsize=\small]
1/384*x1\^{}6*y1\^{}6 + 1/128*x1\^{}4*y1\^{}6*x2 + 1/128*x1\^{}6*y1\^{}4*y2 + 1/128*x1\^{}2*y1\^{}6*x2\^{}2
+ 3/128*x1\^{}4*y1\^{}4*x2*y2 + 1/128*x1\^{}6*y1\^{}2*y2\^{}2 + 1/384*y1\^{}6*x2\^{}3 +
3/128*x1\^{}2*y1\^{}4*x2\^{}2*y2 + 3/128*x1\^{}4*y1\^{}2*x2*y2\^{}2 + 1/384*x1\^{}6*y2\^{}3 +
1/128*y1\^{}4*x2\^{}3*y2 + 7/128*x1\^{}2*y1\^{}2*x2\^{}2*y2\^{}2 + 1/128*x1\^{}4*x2*y2\^{}3 +
5/128*y1\^{}2*x2\^{}3*y2\^{}2 + 5/128*x1\^{}2*x2\^{}2*y2\^{}3 + 1/32*x1\^{}2*y1\^{}2*y2\^{}2*x4 +
1/32*x1\^{}2*y1\^{}2*x2\^{}2*y4 + 13/384*x2\^{}3*y2\^{}3 + 1/32*y1\^{}2*x2*y2\^{}2*x4 +
1/32*x1\^{}2*y2\^{}3*x4 + 1/32*y1\^{}2*x2\^{}3*y4 + 1/32*x1\^{}2*x2\^{}2*y2*y4 +
1/32*x1\^{}2*y1\^{}2*x4*y4 + 1/32*x2*y2\^{}3*x4 + 1/32*x2\^{}3*y2*y4 + 1/32*y1\^{}2*x2*x4*y4 +
1/32*x1\^{}2*y2*x4*y4 + 1/12*x3\^{}2*y3\^{}2 + 1/32*x2*y2*x4*y4 + 1/12*y3\^{}2*x6 +
1/12*x3\^{}2*y6 + 1/12*x6*y6
\end{Verbatim}
\end{tcolorbox}
We put this procedure into a function and generalize to $k_1 \neq k_2$ (in the code these correspond to k\_plus and k\_minus). The function takes in three integers and gives the polynomial $Z^{S_l[S_{k_1} \times S_{k_2}]}(\mathbf{x},\mathbf{y})$.
    \begin{tcolorbox}[breakable, size=fbox, boxrule=1pt, pad at break*=1mm,colback=cellbackground, colframe=cellborder]
\prompt{In}{incolor}{12}{\boxspacing}
\begin{Verbatim}[commandchars=\\\{\},fontsize=\small]
\PY{k}{def} \PY{n+nf}{WreathProductCycleIndexPolynomial}\PY{p}{(}\PY{n}{l}\PY{p}{,}\PY{n}{k\PYZus{}plus}\PY{p}{,} \PY{n}{k\PYZus{}minus}\PY{p}{)}\PY{p}{:}
    \PY{n}{QXX} \PY{o}{=} \PY{n}{PolynomialRing}\PY{p}{(}\PY{n}{QQ}\PY{p}{,} \PY{n}{l}\PY{o}{*}\PY{n}{k\PYZus{}plus}\PY{o}{+}\PY{l+m+mi}{1}\PY{p}{,} \PY{l+s+s1}{\PYZsq{}}\PY{l+s+s1}{x}\PY{l+s+s1}{\PYZsq{}}\PY{p}{)}
    \PY{n}{QYY} \PY{o}{=} \PY{n}{PolynomialRing}\PY{p}{(}\PY{n}{QQ}\PY{p}{,} \PY{n}{l}\PY{o}{*}\PY{n}{k\PYZus{}minus}\PY{o}{+}\PY{l+m+mi}{1}\PY{p}{,} \PY{l+s+s1}{\PYZsq{}}\PY{l+s+s1}{y}\PY{l+s+s1}{\PYZsq{}}\PY{p}{)}
    \PY{n}{QXY} \PY{o}{=} \PY{n}{PolynomialRing}\PY{p}{(}\PY{n}{QQ}\PY{p}{,} \PY{n}{l}\PY{o}{*}\PY{n+nb}{max}\PY{p}{(}\PY{n}{k\PYZus{}plus}\PY{p}{,}\PY{n}{k\PYZus{}minus}\PY{p}{)}\PY{o}{+}\PY{l+m+mi}{1}\PY{p}{,} \PY{n}{var\PYZus{}array}\PY{o}{=}\PY{p}{[}\PY{l+s+s1}{\PYZsq{}}\PY{l+s+s1}{x}\PY{l+s+s1}{\PYZsq{}}\PY{p}{,}\PY{l+s+s1}{\PYZsq{}}\PY{l+s+s1}{y}\PY{l+s+s1}{\PYZsq{}}\PY{p}{]}\PY{p}{)}
    \PY{n}{Z\PYZus{}l} \PY{o}{=} \PY{n}{SymmetricGroup}\PY{p}{(}\PY{n}{l}\PY{p}{)}\PY{o}{.}\PY{n}{cycle\PYZus{}index}\PY{p}{(}\PY{p}{)}
    \PY{n}{Z\PYZus{}kplus} \PY{o}{=} \PY{n}{SymmetricGroup}\PY{p}{(}\PY{n}{k\PYZus{}plus}\PY{p}{)}\PY{o}{.}\PY{n}{cycle\PYZus{}index}\PY{p}{(}\PY{p}{)}
    \PY{n}{Z\PYZus{}kminus} \PY{o}{=} \PY{n}{SymmetricGroup}\PY{p}{(}\PY{n}{k\PYZus{}minus}\PY{p}{)}\PY{o}{.}\PY{n}{cycle\PYZus{}index}\PY{p}{(}\PY{p}{)}
    \PY{k}{return} \PY{n+nb}{sum}\PY{p}{(}\PY{n}{z}\PY{p}{[}\PY{l+m+mi}{1}\PY{p}{]}\PY{o}{*}\PY{n}{prod}\PY{p}{(}\PY{n}{QXY}\PY{o}{.}\PY{n}{coerce}\PY{p}{(}\PY{n}{ShiftedCycleIndexPolynomial}\PY{p}{(}\PY{n}{QXX}\PY{p}{,} \PY{n}{Z\PYZus{}kplus}\PY{p}{,} \PY{n}{i}\PY{p}{)}\PY{p}{)}\PY{o}{*}\PY{n}{QXY}\PY{o}{.}\PY{n}{coerce}\PY{p}{(}\PY{n}{ShiftedCycleIndexPolynomial}\PY{p}{(}\PY{n}{QYY}\PY{p}{,} \PY{n}{Z\PYZus{}kminus}\PY{p}{,} \PY{n}{i}\PY{p}{)}\PY{p}{)}  \PY{k}{for} \PY{n}{i} \PY{o+ow}{in} \PY{n}{z}\PY{p}{[}\PY{l+m+mi}{0}\PY{p}{]}\PY{p}{)}  \PY{k}{for} \PY{n}{z} \PY{o+ow}{in} \PY{n}{Z\PYZus{}l}\PY{p}{)}
\end{Verbatim}
\end{tcolorbox}
We can check that this gives the correct answer for $l=3, k_1 = k_2 = 2$.
    \begin{tcolorbox}[breakable, size=fbox, boxrule=1pt, pad at break*=1mm,colback=cellbackground, colframe=cellborder]
\prompt{In}{incolor}{13}{\boxspacing}
\begin{Verbatim}[commandchars=\\\{\},fontsize=\small]
\PY{n}{WreathProductCycleIndexPolynomial}\PY{p}{(}\PY{l+m+mi}{3}\PY{p}{,}\PY{l+m+mi}{2}\PY{p}{,}\PY{l+m+mi}{1}\PY{p}{)}
\end{Verbatim}
\end{tcolorbox}

            \begin{tcolorbox}[breakable, size=fbox, boxrule=.5pt, pad at break*=1mm, opacityfill=0]
\prompt{Out}{outcolor}{13}{\boxspacing}
\begin{Verbatim}[commandchars=\\\{\},fontsize=\small]
1/384*x1\^{}6*y1\^{}6 + 1/128*x1\^{}4*y1\^{}6*x2 + 1/128*x1\^{}6*y1\^{}4*y2 + 1/128*x1\^{}2*y1\^{}6*x2\^{}2
+ 3/128*x1\^{}4*y1\^{}4*x2*y2 + 1/128*x1\^{}6*y1\^{}2*y2\^{}2 + 1/384*y1\^{}6*x2\^{}3 +
3/128*x1\^{}2*y1\^{}4*x2\^{}2*y2 + 3/128*x1\^{}4*y1\^{}2*x2*y2\^{}2 + 1/384*x1\^{}6*y2\^{}3 +
1/128*y1\^{}4*x2\^{}3*y2 + 7/128*x1\^{}2*y1\^{}2*x2\^{}2*y2\^{}2 + 1/128*x1\^{}4*x2*y2\^{}3 +
5/128*y1\^{}2*x2\^{}3*y2\^{}2 + 5/128*x1\^{}2*x2\^{}2*y2\^{}3 + 1/32*x1\^{}2*y1\^{}2*y2\^{}2*x4 +
1/32*x1\^{}2*y1\^{}2*x2\^{}2*y4 + 13/384*x2\^{}3*y2\^{}3 + 1/32*y1\^{}2*x2*y2\^{}2*x4 +
1/32*x1\^{}2*y2\^{}3*x4 + 1/32*y1\^{}2*x2\^{}3*y4 + 1/32*x1\^{}2*x2\^{}2*y2*y4 +
1/32*x1\^{}2*y1\^{}2*x4*y4 + 1/32*x2*y2\^{}3*x4 + 1/32*x2\^{}3*y2*y4 + 1/32*y1\^{}2*x2*x4*y4 +
1/32*x1\^{}2*y2*x4*y4 + 1/12*x3\^{}2*y3\^{}2 + 1/32*x2*y2*x4*y4 + 1/12*y3\^{}2*x6 +
1/12*x3\^{}2*y6 + 1/12*x6*y6
\end{Verbatim}
\end{tcolorbox}
We want to take products of such cycle indices, as dictated by a vector partition. For this we define a function that takes a vector partition and returns it in exponential form.        
    \begin{tcolorbox}[breakable, size=fbox, boxrule=1pt, pad at break*=1mm,colback=cellbackground, colframe=cellborder]
\prompt{In}{incolor}{14}{\boxspacing}
\begin{Verbatim}[commandchars=\\\{\},fontsize=\small]
\PY{k}{def} \PY{n+nf}{vectorpartition\PYZus{}exponential}\PY{p}{(}\PY{n}{vector\PYZus{}partition}\PY{p}{)}\PY{p}{:}
    \PY{n}{j} \PY{o}{=} \PY{l+m+mi}{0}
    \PY{n}{exp\PYZus{}vec\PYZus{}partition} \PY{o}{=} \PY{p}{[}\PY{p}{]}
    \PY{n}{exp\PYZus{}vec\PYZus{}partition} \PY{o}{=} \PY{n}{exp\PYZus{}vec\PYZus{}partition} \PY{o}{+} \PY{p}{[}\PY{p}{[}\PY{l+m+mi}{1}\PY{p}{,} \PY{n}{vector\PYZus{}partition}\PY{p}{[}\PY{n}{j}\PY{p}{]}\PY{p}{]}\PY{p}{]}
    \PY{k}{for} \PY{n}{i} \PY{o+ow}{in} \PY{n+nb}{range}\PY{p}{(}\PY{l+m+mi}{1}\PY{p}{,}\PY{n+nb}{len}\PY{p}{(}\PY{n}{vector\PYZus{}partition}\PY{p}{)}\PY{p}{)}\PY{p}{:}
        \PY{k}{if} \PY{n}{vector\PYZus{}partition}\PY{p}{[}\PY{n}{i}\PY{o}{\PYZhy{}}\PY{l+m+mi}{1}\PY{p}{]} \PY{o}{==} \PY{n}{vector\PYZus{}partition}\PY{p}{[}\PY{n}{i}\PY{p}{]}\PY{p}{:}
            \PY{n}{exp\PYZus{}vec\PYZus{}partition}\PY{p}{[}\PY{n}{j}\PY{p}{]}\PY{p}{[}\PY{l+m+mi}{0}\PY{p}{]} \PY{o}{=} \PY{n}{exp\PYZus{}vec\PYZus{}partition}\PY{p}{[}\PY{n}{j}\PY{p}{]}\PY{p}{[}\PY{l+m+mi}{0}\PY{p}{]}\PY{o}{+}\PY{l+m+mi}{1}
        \PY{k}{else}\PY{p}{:}
            \PY{n}{j} \PY{o}{=} \PY{n}{j}\PY{o}{+}\PY{l+m+mi}{1}
            \PY{n}{exp\PYZus{}vec\PYZus{}partition} \PY{o}{=} \PY{n}{exp\PYZus{}vec\PYZus{}partition} \PY{o}{+} \PY{p}{[}\PY{p}{[}\PY{l+m+mi}{1}\PY{p}{,} \PY{n}{vector\PYZus{}partition}\PY{p}{[}\PY{n}{i}\PY{p}{]}\PY{p}{]}\PY{p}{]}
    \PY{k}{return} \PY{n}{exp\PYZus{}vec\PYZus{}partition}
\end{Verbatim}
\end{tcolorbox}
As an example, we get
    \begin{tcolorbox}[breakable, size=fbox, boxrule=1pt, pad at break*=1mm,colback=cellbackground, colframe=cellborder]
\prompt{In}{incolor}{15}{\boxspacing}
\begin{Verbatim}[commandchars=\\\{\},fontsize=\small]
\PY{n}{vectorpartition\PYZus{}exp} \PY{o}{=} \PY{n}{vectorpartition\PYZus{}exponential}\PY{p}{(}\PY{n}{VectorPartition}\PY{p}{(}\PY{p}{[}\PY{p}{[}\PY{l+m+mi}{3}\PY{p}{,}\PY{l+m+mi}{2}\PY{p}{]}\PY{p}{,}\PY{p}{[}\PY{l+m+mi}{3}\PY{p}{,}\PY{l+m+mi}{2}\PY{p}{]}\PY{p}{,}\PY{p}{[}\PY{l+m+mi}{1}\PY{p}{,}\PY{l+m+mi}{3}\PY{p}{]}\PY{p}{]}\PY{p}{)}\PY{p}{)}
\PY{n}{vectorpartition\PYZus{}exp}
\end{Verbatim}
\end{tcolorbox}

            \begin{tcolorbox}[breakable, size=fbox, boxrule=.5pt, pad at break*=1mm, opacityfill=0]
\prompt{Out}{outcolor}{15}{\boxspacing}
\begin{Verbatim}[commandchars=\\\{\},fontsize=\small]
[[1, [1, 3]], [2, [3, 2]]]
\end{Verbatim}
\end{tcolorbox}
To compute \eqref{eq: ZG prod cycle indices} we take a product over the vector partition.        
    \begin{tcolorbox}[breakable, size=fbox, boxrule=1pt, pad at break*=1mm,colback=cellbackground, colframe=cellborder]
\prompt{In}{incolor}{16}{\boxspacing}
\begin{Verbatim}[commandchars=\\\{\},fontsize=\small]
\PY{n}{Z\PYZus{}G} \PY{o}{=} \PY{n}{prod}\PY{p}{(}\PY{n}{WreathProductCycleIndexPolynomial}\PY{p}{(}\PY{n}{v}\PY{p}{[}\PY{l+m+mi}{0}\PY{p}{]}\PY{p}{,}\PY{n}{v}\PY{p}{[}\PY{l+m+mi}{1}\PY{p}{]}\PY{p}{[}\PY{l+m+mi}{0}\PY{p}{]}\PY{p}{,} \PY{n}{v}\PY{p}{[}\PY{l+m+mi}{1}\PY{p}{]}\PY{p}{[}\PY{l+m+mi}{1}\PY{p}{]}\PY{p}{)} \PY{k}{for} \PY{n}{v} \PY{o+ow}{in} \PY{n}{vectorpartition\PYZus{}exp}\PY{p}{)}
\PY{n}{Z\PYZus{}G}
\end{Verbatim}
\end{tcolorbox}

            \begin{tcolorbox}[breakable, size=fbox, boxrule=.5pt, pad at break*=1mm, opacityfill=0]
\prompt{Out}{outcolor}{16}{\boxspacing}
\begin{Verbatim}[commandchars=\\\{\},fontsize=\small]
1/1728*x1\^{}7*y1\^{}7 + 1/288*x1\^{}5*y1\^{}7*x2 + 5/1728*x1\^{}7*y1\^{}5*y2 +
1/192*x1\^{}3*y1\^{}7*x2\^{}2 + 5/288*x1\^{}5*y1\^{}5*x2*y2 + 7/1728*x1\^{}7*y1\^{}3*y2\^{}2 +
1/432*x1\^{}4*y1\^{}7*x3 + 1/864*x1\^{}7*y1\^{}4*y3 + 5/192*x1\^{}3*y1\^{}5*x2\^{}2*y2 +
7/288*x1\^{}5*y1\^{}3*x2*y2\^{}2 + 1/576*x1\^{}7*y1*y2\^{}3 + 1/144*x1\^{}2*y1\^{}7*x2*x3 +
5/432*x1\^{}4*y1\^{}5*y2*x3 + 1/144*x1\^{}5*y1\^{}4*x2*y3 + 1/432*x1\^{}7*y1\^{}2*y2*y3 +
7/192*x1\^{}3*y1\^{}3*x2\^{}2*y2\^{}2 + 1/96*x1\^{}5*y1*x2*y2\^{}3 + 5/144*x1\^{}2*y1\^{}5*x2*y2*x3 +
7/432*x1\^{}4*y1\^{}3*y2\^{}2*x3 + 1/432*x1*y1\^{}7*x3\^{}2 + 1/96*x1\^{}3*y1\^{}4*x2\^{}2*y3 +
1/72*x1\^{}5*y1\^{}2*x2*y2*y3 + 1/864*x1\^{}7*y2\^{}2*y3 + 1/216*x1\^{}4*y1\^{}4*x3*y3 +
1/144*x1*y1\^{}3*x2\^{}3*y2\^{}2 + 1/64*x1\^{}3*y1*x2\^{}2*y2\^{}3 + 7/144*x1\^{}2*y1\^{}3*x2*y2\^{}2*x3 +
1/144*x1\^{}4*y1*y2\^{}3*x3 + 5/432*x1*y1\^{}5*y2*x3\^{}2 + 1/48*x1\^{}3*y1\^{}2*x2\^{}2*y2*y3 +
1/144*x1\^{}5*x2*y2\^{}2*y3 + 1/72*x1\^{}2*y1\^{}4*x2*x3*y3 + 1/108*x1\^{}4*y1\^{}2*y2*x3*y3 +
1/48*x1*y1*x2\^{}3*y2\^{}3 + 1/48*x1\^{}2*y1*x2*y2\^{}3*x3 + 7/432*x1*y1\^{}3*y2\^{}2*x3\^{}2 +
1/96*x1\^{}3*x2\^{}2*y2\^{}2*y3 + 1/36*x1\^{}2*y1\^{}2*x2*y2*x3*y3 + 1/216*x1\^{}4*y2\^{}2*x3*y3 +
1/216*x1*y1\^{}4*x3\^{}2*y3 + 1/48*x1*y1\^{}3*x2*y2\^{}2*x4 + 1/144*x1*y1\^{}3*x2\^{}3*y4 +
1/144*x1*y1*y2\^{}3*x3\^{}2 + 1/72*x1*x2\^{}3*y2\^{}2*y3 + 1/72*x1\^{}2*x2*y2\^{}2*x3*y3 +
1/108*x1*y1\^{}2*y2*x3\^{}2*y3 + 1/16*x1*y1*x2*y2\^{}3*x4 + 1/48*x1*y1*x2\^{}3*y2*y4 +
1/48*x1*y1\^{}3*x2*x4*y4 + 1/72*x1*y1\^{}3*y2\^{}2*x6 + 1/216*x1*y2\^{}2*x3\^{}2*y3 +
1/24*x1*x2*y2\^{}2*y3*x4 + 1/72*x1*x2\^{}3*y3*y4 + 1/16*x1*y1*x2*y2*x4*y4 +
1/24*x1*y1*y2\^{}3*x6 + 1/72*x1*y1\^{}3*y4*x6 + 1/24*x1*x2*y3*x4*y4 +
1/36*x1*y2\^{}2*y3*x6 + 1/24*x1*y1*y2*y4*x6 + 1/36*x1*y3*y4*x6
\end{Verbatim}
\end{tcolorbox}
All that remains is to pick out the coefficients of monomials $\mathbf{x}^{p}\mathbf{y}^{p}$ for some partition $p \vdash k$. Below, we have a function that take a vector partition, construct the corresponding cycle index and extracts the coefficient of the monomial $\mathbf{x}^{p}\mathbf{y}^{p}$ for a partition $p$.  
    \begin{tcolorbox}[breakable, size=fbox, boxrule=1pt, pad at break*=1mm,colback=cellbackground, colframe=cellborder]
\prompt{In}{incolor}{17}{\boxspacing}
\begin{Verbatim}[commandchars=\\\{\},fontsize=\small]
\PY{k}{def} \PY{n+nf}{monomial\PYZus{}coefficient\PYZus{}of\PYZus{}ZG}\PY{p}{(}\PY{n}{vectorpartition}\PY{p}{,} \PY{n}{partition}\PY{p}{)}\PY{p}{:}
    \PY{n}{vectorpartition\PYZus{}exp} \PY{o}{=} \PY{n}{vectorpartition\PYZus{}exponential}\PY{p}{(}\PY{n}{vectorpartition}\PY{p}{)}
    \PY{n}{k\PYZus{}plus} \PY{o}{=} \PY{n}{vectorpartition}\PY{o}{.}\PY{n}{sum}\PY{p}{(}\PY{p}{)}\PY{p}{[}\PY{l+m+mi}{0}\PY{p}{]}
    \PY{n}{k\PYZus{}minus} \PY{o}{=} \PY{n}{vectorpartition}\PY{o}{.}\PY{n}{sum}\PY{p}{(}\PY{p}{)}\PY{p}{[}\PY{l+m+mi}{1}\PY{p}{]}
    \PY{n}{QXX} \PY{o}{=} \PY{n}{PolynomialRing}\PY{p}{(}\PY{n}{QQ}\PY{p}{,} \PY{l+m+mi}{2}\PY{o}{*}\PY{n}{k\PYZus{}plus}\PY{o}{+}\PY{l+m+mi}{1}\PY{p}{,} \PY{l+s+s1}{\PYZsq{}}\PY{l+s+s1}{x}\PY{l+s+s1}{\PYZsq{}}\PY{p}{)}
    \PY{n}{QYY} \PY{o}{=} \PY{n}{PolynomialRing}\PY{p}{(}\PY{n}{QQ}\PY{p}{,} \PY{l+m+mi}{2}\PY{o}{*}\PY{n}{k\PYZus{}minus}\PY{o}{+}\PY{l+m+mi}{1}\PY{p}{,} \PY{l+s+s1}{\PYZsq{}}\PY{l+s+s1}{y}\PY{l+s+s1}{\PYZsq{}}\PY{p}{)}
    \PY{n}{QXY} \PY{o}{=} \PY{n}{PolynomialRing}\PY{p}{(}\PY{n}{QQ}\PY{p}{,} \PY{l+m+mi}{2}\PY{o}{*}\PY{n+nb}{max}\PY{p}{(}\PY{n}{k\PYZus{}plus}\PY{p}{,}\PY{n}{k\PYZus{}minus}\PY{p}{)}\PY{o}{+}\PY{l+m+mi}{1}\PY{p}{,} \PY{n}{var\PYZus{}array}\PY{o}{=}\PY{p}{[}\PY{l+s+s1}{\PYZsq{}}\PY{l+s+s1}{x}\PY{l+s+s1}{\PYZsq{}}\PY{p}{,}\PY{l+s+s1}{\PYZsq{}}\PY{l+s+s1}{y}\PY{l+s+s1}{\PYZsq{}}\PY{p}{]}\PY{p}{)}
    \PY{n}{X\PYZus{}monomial} \PY{o}{=} \PY{n}{prod}\PY{p}{(}\PY{n}{QXX}\PY{o}{.}\PY{n}{gens}\PY{p}{(}\PY{p}{)}\PY{p}{[}\PY{n}{i}\PY{p}{]} \PY{k}{for} \PY{n}{i} \PY{o+ow}{in} \PY{n}{partition}\PY{p}{)}
    \PY{n}{Y\PYZus{}monomial} \PY{o}{=} \PY{n}{prod}\PY{p}{(}\PY{n}{QYY}\PY{o}{.}\PY{n}{gens}\PY{p}{(}\PY{p}{)}\PY{p}{[}\PY{n}{i}\PY{p}{]} \PY{k}{for} \PY{n}{i} \PY{o+ow}{in} \PY{n}{partition}\PY{p}{)}
    \PY{n}{XY\PYZus{}monomial} \PY{o}{=} \PY{n}{QXY}\PY{o}{.}\PY{n}{coerce}\PY{p}{(}\PY{n}{X\PYZus{}monomial}\PY{p}{)}\PY{o}{*}\PY{n}{QXY}\PY{o}{.}\PY{n}{coerce}\PY{p}{(}\PY{n}{Y\PYZus{}monomial}\PY{p}{)}
    \PY{n}{Z\PYZus{}G} \PY{o}{=} \PY{n}{QXY}\PY{o}{.}\PY{n}{coerce}\PY{p}{(}\PY{n}{prod}\PY{p}{(}\PY{n}{WreathProductCycleIndexPolynomial}\PY{p}{(}\PY{n}{v}\PY{p}{[}\PY{l+m+mi}{0}\PY{p}{]}\PY{p}{,}\PY{n}{v}\PY{p}{[}\PY{l+m+mi}{1}\PY{p}{]}\PY{p}{[}\PY{l+m+mi}{0}\PY{p}{]}\PY{p}{,} \PY{n}{v}\PY{p}{[}\PY{l+m+mi}{1}\PY{p}{]}\PY{p}{[}\PY{l+m+mi}{1}\PY{p}{]}\PY{p}{)} \PY{k}{for} \PY{n}{v} \PY{o+ow}{in} \PY{n}{vectorpartition\PYZus{}exp}\PY{p}{)}\PY{p}{)}
    \PY{k}{return} \PY{n}{Z\PYZus{}G}\PY{o}{.}\PY{n}{monomial\PYZus{}coefficient}\PY{p}{(}\PY{n}{XY\PYZus{}monomial}\PY{p}{)}
\end{Verbatim}
\end{tcolorbox}
We can compare our code to Example \ref{ex: cycle index}.
    \begin{tcolorbox}[breakable, size=fbox, boxrule=1pt, pad at break*=1mm,colback=cellbackground, colframe=cellborder]
\prompt{In}{incolor}{18}{\boxspacing}
\begin{Verbatim}[commandchars=\\\{\},fontsize=\small]
\PY{n}{monomial\PYZus{}coefficient\PYZus{}of\PYZus{}ZG}\PY{p}{(}\PY{n}{VectorPartition}\PY{p}{(}\PY{p}{[}\PY{p}{[}\PY{l+m+mi}{3}\PY{p}{,}\PY{l+m+mi}{2}\PY{p}{]}\PY{p}{,}\PY{p}{[}\PY{l+m+mi}{3}\PY{p}{,}\PY{l+m+mi}{2}\PY{p}{]}\PY{p}{,}\PY{p}{[}\PY{l+m+mi}{1}\PY{p}{,}\PY{l+m+mi}{3}\PY{p}{]}\PY{p}{]}\PY{p}{)}\PY{p}{,} \PY{n}{Partition}\PY{p}{(}\PY{p}{[}\PY{l+m+mi}{1}\PY{p}{,}\PY{l+m+mi}{1}\PY{p}{,}\PY{l+m+mi}{1}\PY{p}{,}\PY{l+m+mi}{1}\PY{p}{,}\PY{l+m+mi}{1}\PY{p}{,}\PY{l+m+mi}{1}\PY{p}{,}\PY{l+m+mi}{1}\PY{p}{]}\PY{p}{)}\PY{p}{)}
\end{Verbatim}
\end{tcolorbox}

            \begin{tcolorbox}[breakable, size=fbox, boxrule=.5pt, pad at break*=1mm, opacityfill=0]
\prompt{Out}{outcolor}{18}{\boxspacing}
\begin{Verbatim}[commandchars=\\\{\},fontsize=\small]
1/1728
\end{Verbatim}
\end{tcolorbox}
        
    \begin{tcolorbox}[breakable, size=fbox, boxrule=1pt, pad at break*=1mm,colback=cellbackground, colframe=cellborder]
\prompt{In}{incolor}{19}{\boxspacing}
\begin{Verbatim}[commandchars=\\\{\},fontsize=\small]
\PY{n}{monomial\PYZus{}coefficient\PYZus{}of\PYZus{}ZG}\PY{p}{(}\PY{n}{VectorPartition}\PY{p}{(}\PY{p}{[}\PY{p}{[}\PY{l+m+mi}{3}\PY{p}{,}\PY{l+m+mi}{2}\PY{p}{]}\PY{p}{,}\PY{p}{[}\PY{l+m+mi}{3}\PY{p}{,}\PY{l+m+mi}{2}\PY{p}{]}\PY{p}{,}\PY{p}{[}\PY{l+m+mi}{1}\PY{p}{,}\PY{l+m+mi}{3}\PY{p}{]}\PY{p}{]}\PY{p}{)}\PY{p}{,} \PY{n}{Partition}\PY{p}{(}\PY{p}{[}\PY{l+m+mi}{2}\PY{p}{,}\PY{l+m+mi}{1}\PY{p}{,}\PY{l+m+mi}{1}\PY{p}{,}\PY{l+m+mi}{1}\PY{p}{,}\PY{l+m+mi}{1}\PY{p}{,}\PY{l+m+mi}{1}\PY{p}{]}\PY{p}{)}\PY{p}{)}
\end{Verbatim}
\end{tcolorbox}

            \begin{tcolorbox}[breakable, size=fbox, boxrule=.5pt, pad at break*=1mm, opacityfill=0]
\prompt{Out}{outcolor}{19}{\boxspacing}
\begin{Verbatim}[commandchars=\\\{\},fontsize=\small]
5/288
\end{Verbatim}
\end{tcolorbox}
We compute the total number of double cosets by summing over all partitions $p \vdash 7$, extracting the coefficient of the corresponding monomial and multiplying by $\Sym(p)$ (this is called p.aut() in the code).        
    \begin{tcolorbox}[breakable, size=fbox, boxrule=1pt, pad at break*=1mm,colback=cellbackground, colframe=cellborder]
\prompt{In}{incolor}{20}{\boxspacing}
\begin{Verbatim}[commandchars=\\\{\},fontsize=\small]
\PY{n+nb}{sum}\PY{p}{(}\PY{n}{partition}\PY{o}{.}\PY{n}{aut}\PY{p}{(}\PY{p}{)}\PY{o}{*}\PY{n}{monomial\PYZus{}coefficient\PYZus{}of\PYZus{}ZG}\PY{p}{(}\PY{n}{VectorPartition}\PY{p}{(}\PY{p}{[}\PY{p}{[}\PY{l+m+mi}{3}\PY{p}{,}\PY{l+m+mi}{2}\PY{p}{]}\PY{p}{,}\PY{p}{[}\PY{l+m+mi}{3}\PY{p}{,}\PY{l+m+mi}{2}\PY{p}{]}\PY{p}{,}\PY{p}{[}\PY{l+m+mi}{1}\PY{p}{,}\PY{l+m+mi}{3}\PY{p}{]}\PY{p}{]}\PY{p}{)}\PY{p}{,} \PY{n}{partition}\PY{p}{)} \PY{k}{for} \PY{n}{partition} \PY{o+ow}{in} \PY{n}{Partitions}\PY{p}{(}\PY{l+m+mi}{7}\PY{p}{)}\PY{p}{)}
\end{Verbatim}
\end{tcolorbox}

            \begin{tcolorbox}[breakable, size=fbox, boxrule=.5pt, pad at break*=1mm, opacityfill=0]
\prompt{Out}{outcolor}{20}{\boxspacing}
\begin{Verbatim}[commandchars=\\\{\},fontsize=\small]
11
\end{Verbatim}
\end{tcolorbox}
We collect this procedure as a function -- it takes a vector partition and computes the number of double cosets $N(\vec{k}^+,\vec{k}^-)$.        
    \begin{tcolorbox}[breakable, size=fbox, boxrule=1pt, pad at break*=1mm,colback=cellbackground, colframe=cellborder]
\prompt{In}{incolor}{21}{\boxspacing}
\begin{Verbatim}[commandchars=\\\{\},fontsize=\small]
\PY{k}{def} \PY{n+nf}{number\PYZus{}of\PYZus{}double\PYZus{}cosets}\PY{p}{(}\PY{n}{vectorpartition}\PY{p}{,} \PY{n}{k}\PY{p}{)}\PY{p}{:}
    \PY{k}{return} \PY{n+nb}{sum}\PY{p}{(}\PY{n}{partition}\PY{o}{.}\PY{n}{aut}\PY{p}{(}\PY{p}{)}\PY{o}{*}\PY{n}{monomial\PYZus{}coefficient\PYZus{}of\PYZus{}ZG}\PY{p}{(}\PY{n}{vectorpartition}\PY{p}{,} \PY{n}{partition}\PY{p}{)} \PY{k}{for} \PY{n}{partition} \PY{o+ow}{in} \PY{n}{Partitions}\PY{p}{(}\PY{n}{k}\PY{p}{)}\PY{p}{)}
\end{Verbatim}
\end{tcolorbox}

Lastly, we can compute the refinement by number of vertices as in Table \ref{tab: Table of invariant dimensions refined} for $k=4$ edges.
    \begin{tcolorbox}[breakable, size=fbox, boxrule=1pt, pad at break*=1mm,colback=cellbackground, colframe=cellborder]
\prompt{In}{incolor}{22}{\boxspacing}
\begin{Verbatim}[commandchars=\\\{\},fontsize=\small]
\PY{n}{k} \PY{o}{=} \PY{l+m+mi}{4}
\PY{n+nb}{list}\PY{p}{(}\PY{n+nb}{sum}\PY{p}{(}\PY{n}{number\PYZus{}of\PYZus{}double\PYZus{}cosets}\PY{p}{(}\PY{n}{v}\PY{p}{)}\PY{k}{for} \PY{n}{v} \PY{o+ow}{in} \PY{n}{VectorPartitions}\PY{p}{(}\PY{p}{[}\PY{n}{k}\PY{p}{,}\PY{n}{k}\PY{p}{]}\PY{p}{)}  \PY{k}{if} \PY{n+nb}{len}\PY{p}{(}\PY{n}{v}\PY{p}{)} \PY{o}{==} \PY{n}{l}\PY{p}{)} \PY{k}{for} \PY{n}{l} \PY{o+ow}{in} \PY{p}{[}\PY{l+m+mf}{1.}\PY{o}{.}\PY{l+m+mi}{2}\PY{o}{*}\PY{n}{k}\PY{p}{]}\PY{p}{)}
\end{Verbatim}
\end{tcolorbox}

            \begin{tcolorbox}[breakable, size=fbox, boxrule=.5pt, pad at break*=1mm, opacityfill=0]
\prompt{Out}{outcolor}{22}{\boxspacing}
\begin{Verbatim}[commandchars=\\\{\},fontsize=\small]
[1, 18, 71, 108, 71, 22, 4, 1]
\end{Verbatim}
\end{tcolorbox}

This code/notebook can be found at \href{https://github.com/adrianpadellaro/PhD-Thesis}{Link to GitHub Repository}.

	%%------------------------------------------------------------------------------------
	%% Beginning of Backmatter
	%%------------------------------------------------------------------------------------
	\backmatter % turns off chapter numbering 

	%Bibliografia
	%------------------------------------------------------------------
	
	\bibliographystyle{JHEP.bst}
	\bibliography{PhD.bib}{}

\providecommand{\href}[2]{#2}\begingroup\raggedright\begin{thebibliography}{100}

\bibitem{Barnes2022b}
G.~Barnes, A.~Padellaro and S.~Ramgoolam, \emph{{Permutation invariant Gaussian
  two-matrix models}}, \href{https://doi.org/10.1088/1751-8121/ac4de1}{\emph{J.
  Phys. A} {\bfseries 55} (2022) 145202}
  [\href{https://arxiv.org/abs/2104.03707}{{\ttfamily 2104.03707}}].

\bibitem{Barnes:2021tjp}
G.~Barnes, A.~Padellaro and S.~Ramgoolam, \emph{{Hidden symmetries and large N
  factorisation for permutation invariant matrix observables}},
  \href{https://doi.org/10.1007/JHEP08(2022)090}{\emph{JHEP} {\bfseries 08}
  (2022) 090} [\href{https://arxiv.org/abs/2112.00498}{{\ttfamily
  2112.00498}}].

\bibitem{Barnes:2022qli}
G.~Barnes, A.~Padellaro and S.~Ramgoolam, \emph{{Permutation symmetry in
  large-N matrix quantum mechanics and partition algebras}},
  \href{https://doi.org/10.1103/PhysRevD.106.106020}{\emph{Phys. Rev. D}
  {\bfseries 106} (2022) 106020}
  [\href{https://arxiv.org/abs/2207.02166}{{\ttfamily 2207.02166}}].

\bibitem{PIGTM}
G.~Barnes, A.~Padellaro and S.~Ramgoolam, ``Permutation invariant {G}aussian
  tensor models.''

\bibitem{Wigner1955}
E.P.~Wigner, \emph{Characteristic {Vectors} of {Bordered} {Matrices} {With}
  {Infinite} {Dimensions}}, \href{https://doi.org/10.2307/1970079}{\emph{Annals
  of Mathematics} {\bfseries 62} (1955) 548}.

\bibitem{Dyson1962}
F.J.~Dyson, \emph{{A Brownian-Motion Model for the Eigenvalues of a Random
  Matrix}}, \href{https://doi.org/10.1063/1.1703862}{\emph{J. Math. Phys.}
  {\bfseries 3} (1962) 1191}.

\bibitem{Guhr1998}
T.~Guhr, A.~Muller-Groeling and H.A.~Weidenmuller, \emph{{Random matrix
  theories in quantum physics: Common concepts}},
  \href{https://doi.org/10.1016/S0370-1573(97)00088-4}{\emph{Phys. Rept.}
  {\bfseries 299} (1998) 189}
  [\href{https://arxiv.org/abs/cond-mat/9707301}{{\ttfamily
  cond-mat/9707301}}].

\bibitem{Edelman2013}
A.~Edelman and Y.~Wang, \emph{Random matrix theory and its innovative
  applications},  in \emph{Advances in Applied Mathematics, Modeling, and
  Computational Science}, R.~Melnik and I.S.~Kotsireas, eds., (Boston, MA),
  pp.~91--116, Springer US (2013),
  \href{https://doi.org/10.1007/978-1-4614-5389-5_5}{DOI}.

\bibitem{Akemann:2016keq}
G.~Akemann, \emph{{Random Matrix Theory and Quantum Chromodynamics}},  3, 2016,
  \href{https://doi.org/10.1093/oso/9780198797319.003.0005}{DOI}
  [\href{https://arxiv.org/abs/1603.06011}{{\ttfamily 1603.06011}}].

\bibitem{PCA2016}
P.~Mattioli and S.~Ramgoolam, \emph{Permutation centralizer algebras and
  multimatrix invariants},
  \href{https://doi.org/10.1103/physrevd.93.065040}{\emph{Physical Review D}
  {\bfseries 93} (2016) }.

\bibitem{Geloun2021}
J.B.~Geloun and S.~Ramgoolam, \emph{Quantum mechanics of bipartite ribbon
  graphs: {Integrality}, {Lattices} and {Kronecker} coefficients},
  {\emph{arXiv:2010.04054 [hep-th, physics:quant-ph]} (2021) }.

\bibitem{Mironov2022}
A.~Mironov and A.~Morozov, \emph{Superintegrability summary},
  \href{https://doi.org/10.1016/j.physletb.2022.137573}{\emph{Physics Letters
  B} {\bfseries 835} (2022) 137573}.

\bibitem{Ramgoolam:2023vyq}
S.~Ramgoolam and L.~Sword, \emph{{Matrix and tensor witnesses of hidden
  symmetry algebras}},
  \href{https://doi.org/10.1007/JHEP03(2023)056}{\emph{JHEP} {\bfseries 03}
  (2023) 056} [\href{https://arxiv.org/abs/2302.01206}{{\ttfamily
  2302.01206}}].

\bibitem{MulaseYu}
M.~Mulase and J.T.~Yu, \emph{A generating function of the number of
  homomorphisms from a surface group into a finite group},  2002.
\newblock 10.48550/ARXIV.MATH/0209008.

\bibitem{Kimura2014}
Y.~Kimura, \emph{{Multi-matrix models and Noncommutative Frobenius algebras
  obtained from symmetric groups and Brauer algebras}},
  \href{https://doi.org/10.1007/s00220-014-2231-6}{\emph{Commun. Math. Phys.}
  {\bfseries 337} (2015) 1} [\href{https://arxiv.org/abs/1403.6572}{{\ttfamily
  1403.6572}}].

\bibitem{Kimura2017}
Y.~Kimura, \emph{Noncommutative frobenius algebras and open-closed duality},
  \href{https://arxiv.org/abs/1701.08382v1}{{\ttfamily 1701.08382v1}}.

\bibitem{Firth1957}
J.~Firth, \emph{A synopsis of linguistic theory, 1930-1955}, {\emph{Studies in
  linguistic analysis} (1957) 10}.

\bibitem{Harris1968}
Z.S.~Harris, \emph{Mathematical structures of language},  in \emph{Interscience
  tracts in pure and applied mathematics}, 1968.

\bibitem{coecke2010mathematical}
B.~Coecke, M.~Sadrzadeh and S.~Clark, \emph{Mathematical foundations for a
  compositional distributional model of meaning},  2010.

\bibitem{Baroni2014FregeIS}
M.~Baroni, R.~Bernardi and R.~Zamparelli, \emph{Frege in space: A program of
  compositional distributional semantics}, {\emph{Linguistic Issues in Language
  Technology} {\bfseries 9} (2014) }.

\bibitem{Kartsaklis2017}
D.~Kartsaklis, S.~Ramgoolam and M.~Sadrzadeh, \emph{Linguistic {Matrix}
  {Theory}}, {\emph{arXiv:1703.10252 [hep-th]} (2017) }.

\bibitem{Ramgoolam2019a}
S.~Ramgoolam, \emph{Permutation {Invariant} {Gaussian} {Matrix} {Models}},
  \href{https://doi.org/10.1016/j.nuclphysb.2019.114682}{\emph{Nuclear Physics
  B} {\bfseries 945} (2019) 114682}.

\bibitem{Ramgoolam2019}
S.~Ramgoolam, M.~Sadrzadeh and L.~Sword, \emph{Gaussianity and typicality in
  matrix distributional semantics}, {\emph{arXiv:1912.10839 [hep-th,
  physics:math-ph]} (2019) }.

\bibitem{Huber2022a}
M.A.~Huber, A.~Correia, S.~Ramgoolam and M.~Sadrzadeh, \emph{{Permutation
  invariant matrix statistics and computational language tasks}},
  \href{https://arxiv.org/abs/2202.06829}{{\ttfamily 2202.06829}}.

\bibitem{Fulton2013}
W.~Fulton and J.~Harris, \emph{Representation theory: a first course},
  vol.~129, Springer Science \& Business Media.

\bibitem{Balasubramanian2002}
V.~Balasubramanian, M.~Berkooz, A.~Naqvi and M.J.~Strassler, \emph{Giant
  gravitons in conformal field theory},
  \href{https://doi.org/10.1088/1126-6708/2002/04/034}{\emph{JHEP} {\bfseries
  04} (2002) 034}.

\bibitem{CJR}
S.~Corley, A.~Jevicki and S.~Ramgoolam, \emph{{Exact correlators of giant
  gravitons from dual N=4 SYM theory}},
  \href{https://doi.org/10.4310/ATMP.2001.v5.n4.a6}{\emph{Adv. Theor. Math.
  Phys.} {\bfseries 5} (2002) 809}
  [\href{https://arxiv.org/abs/hep-th/0111222}{{\ttfamily hep-th/0111222}}].

\bibitem{Berenstein2004}
D.~Berenstein, \emph{A {Toy} model for the {AdS} / {CFT} correspondence},
  \href{https://doi.org/10.1088/1126-6708/2004/07/018}{\emph{JHEP} {\bfseries
  07} (2004) 018}.

\bibitem{McGreevy2000}
J.~McGreevy, L.~Susskind and N.~Toumbas, \emph{Invasion of the giant gravitons
  from {Anti}-de {Sitter} space},
  \href{https://doi.org/10.1088/1126-6708/2000/06/008}{\emph{JHEP} {\bfseries
  06} (2000) 008}.

\bibitem{HHI2000}
A.~Hashimoto, S.~Hirano and N.~Itzhaki, \emph{{Large branes in AdS and their
  field theory dual}},
  \href{https://doi.org/10.1088/1126-6708/2000/08/051}{\emph{JHEP} {\bfseries
  08} (2000) 051} [\href{https://arxiv.org/abs/hep-th/0008016}{{\ttfamily
  hep-th/0008016}}].

\bibitem{GMT2000}
M.T.~Grisaru, R.C.~Myers and O.~Tafjord, \emph{{SUSY and goliath}},
  \href{https://doi.org/10.1088/1126-6708/2000/08/040}{\emph{JHEP} {\bfseries
  08} (2000) 040} [\href{https://arxiv.org/abs/hep-th/0008015}{{\ttfamily
  hep-th/0008015}}].

\bibitem{Kimura2007}
Y.~Kimura and S.~Ramgoolam, \emph{Branes, {Anti}-{Branes} and {Brauer}
  {Algebras} in {Gauge}-{Gravity} duality},
  \href{https://doi.org/10.1088/1126-6708/2007/11/078}{\emph{Journal of High
  Energy Physics} {\bfseries 2007} (2007) 078}.

\bibitem{Brown2008}
T.W.~Brown, P.J.~Heslop and S.~Ramgoolam, \emph{Diagonal multi-matrix
  correlators and {BPS} operators in {N}=4 {SYM}},
  \href{https://doi.org/10.1088/1126-6708/2008/02/030}{\emph{Journal of High
  Energy Physics} {\bfseries 2008} (2008) 030}.

\bibitem{Bhattacharyya:2008rb}
R.~Bhattacharyya, S.~Collins and R.~de~Mello~Koch, \emph{{Exact Multi-Matrix
  Correlators}},
  \href{https://doi.org/10.1088/1126-6708/2008/03/044}{\emph{JHEP} {\bfseries
  03} (2008) 044} [\href{https://arxiv.org/abs/0801.2061}{{\ttfamily
  0801.2061}}].

\bibitem{Bhattacharyya2008b}
R.~Bhattacharyya, R.~de~Mello~Koch and M.~Stephanou, \emph{Exact
  {Multi}-{Restricted} {Schur} {Polynomial} {Correlators}},
  \href{https://doi.org/10.1088/1126-6708/2008/06/101}{\emph{JHEP} {\bfseries
  06} (2008) 101}.

\bibitem{Kimura2008}
Y.~Kimura and S.~Ramgoolam, \emph{Enhanced symmetries of gauge theory and
  resolving the spectrum of local operators},
  \href{https://doi.org/10.1103/PhysRevD.78.126003}{\emph{Phys. Rev. D}
  {\bfseries 78} (2008) 126003}.

\bibitem{Brown2009}
T.W.~Brown, P.J.~Heslop and S.~Ramgoolam, \emph{Diagonal free field matrix
  correlators, global symmetries and giant gravitons},
  \href{https://doi.org/10.1088/1126-6708/2009/04/089}{\emph{JHEP} {\bfseries
  04} (2009) 089}.

\bibitem{QuivCalc}
J.~Pasukonis and S.~Ramgoolam, \emph{{Quivers as Calculators: Counting,
  Correlators and Riemann Surfaces}},
  \href{https://doi.org/10.1007/JHEP04(2013)094}{\emph{JHEP} {\bfseries 04}
  (2013) 094} [\href{https://arxiv.org/abs/1301.1980}{{\ttfamily 1301.1980}}].

\bibitem{CDD1301}
P.~Caputa, R.~de~Mello~Koch and P.~Diaz, \emph{{A basis for large operators in
  N=4 SYM with orthogonal gauge group}},
  \href{https://doi.org/10.1007/JHEP03(2013)041}{\emph{JHEP} {\bfseries 03}
  (2013) 041} [\href{https://arxiv.org/abs/1301.1560}{{\ttfamily 1301.1560}}].

\bibitem{Ber1504}
D.~Berenstein, \emph{{Extremal chiral ring states in the AdS/CFT correspondence
  are described by free fermions for a generalized oscillator algebra}},
  \href{https://doi.org/10.1103/PhysRevD.92.046006}{\emph{Phys. Rev. D}
  {\bfseries 92} (2015) 046006}
  [\href{https://arxiv.org/abs/1504.05389}{{\ttfamily 1504.05389}}].

\bibitem{KRS}
Y.~Kimura, S.~Ramgoolam and R.~Suzuki, \emph{{Flavour singlets in gauge theory
  as Permutations}}, \href{https://doi.org/10.1007/JHEP12(2016)142}{\emph{JHEP}
  {\bfseries 12} (2016) 142}
  [\href{https://arxiv.org/abs/1608.03188}{{\ttfamily 1608.03188}}].

\bibitem{CLBSR}
C.~Lewis-Brown and S.~Ramgoolam, \emph{{BPS operators in $\mathcal{N}=4$
  $SO(N)$ super Yang-Mills theory: plethysms, dominoes and words}},
  \href{https://doi.org/10.1007/JHEP11(2018)035}{\emph{JHEP} {\bfseries 11}
  (2018) 035} [\href{https://arxiv.org/abs/1804.11090}{{\ttfamily
  1804.11090}}].

\bibitem{ADHSSS}
F.~Aprile, J.M.~Drummond, P.~Heslop, H.~Paul, F.~Sanfilippo, M.~Santagata
  et~al., \emph{{Single particle operators and their correlators in free $
  \mathcal{N} $ = 4 SYM}},
  \href{https://doi.org/10.1007/JHEP11(2020)072}{\emph{JHEP} {\bfseries 11}
  (2020) 072} [\href{https://arxiv.org/abs/2007.09395}{{\ttfamily
  2007.09395}}].

\bibitem{LY2107}
H.~Lin and Y.~Zhu, \emph{{Entanglement and mixed states of Young tableau states
  in gauge/gravity correspondence}},
  \href{https://doi.org/10.1016/j.nuclphysb.2021.115572}{\emph{Nucl. Phys. B}
  {\bfseries 972} (2021) 115572}
  [\href{https://arxiv.org/abs/2107.14219}{{\ttfamily 2107.14219}}].

\bibitem{Jones1994}
V.~Jones, \emph{{The Potts model and the symmetric group}}, {\emph{Subfactors:
  Proceedings of the Taniguchi Symposium on Operator Algebras} (1994) 259?267}.

\bibitem{Martin1994}
P.~Martin, \emph{Temperley-lieb algebras for non-planar statistical mechanics
  — the partition algebra construction},
  \href{https://doi.org/10.1142/S0218216594000071}{\emph{Journal of Knot Theory
  and Its Ramifications} {\bfseries 03} (1994) 51}.

\bibitem{Martin1996}
P.~Martin, \emph{The {Structure} of the {Partition} {Algebras}},
  \href{https://doi.org/10.1006/jabr.1996.0223}{\emph{Journal of Algebra}
  {\bfseries 183} (1996) 319}.

\bibitem{Halverson2001}
T.~Halverson, \emph{Characters of the partition algebras}, {\emph{Journal of
  Algebra} {\bfseries 238} (2001) 502}.

\bibitem{Benkart2017}
G.~Benkart and T.~Halverson, \emph{Partition {Algebras} and the {Invariant}
  {Theory} of the {Symmetric} {Group}}, {\emph{arXiv:1709.07751 [math]} (2017)
  }.

\bibitem{Halverson2018}
T.~Halverson and T.N.~Jacobson, \emph{Set-partition tableaux and
  representations of diagram algebras},
  \href{https://arxiv.org/abs/1808.08118v2}{{\ttfamily 1808.08118v2}}.

\bibitem{Halverson2005}
T.~Halverson and A.~Ram, \emph{Partition algebras},
  \href{https://doi.org/10.1016/j.ejc.2004.06.005}{\emph{European Journal of
  Combinatorics} {\bfseries 26} (2005) 869}
  [\href{https://arxiv.org/abs/math/0401314v2}{{\ttfamily math/0401314v2}}].

\bibitem{Enyang_2012}
J.~Enyang, \emph{Jucys{\textendash}murphy elements and a presentation for
  partition algebras},
  \href{https://doi.org/10.1007/s10801-012-0370-4}{\emph{Journal of Algebraic
  Combinatorics} {\bfseries 37} (2012) 401}.

\bibitem{Benkart2016}
G.~Benkart, T.~Halverson and N.~Harman, \emph{Dimensions of irreducible modules
  for partition algebras and tensor power multiplicities for symmetric and
  alternating groups}, {\emph{arXiv:1605.06543 [math]} (2016) }.

\bibitem{Doty2019}
S.~Doty, A.~Lauve and G.~Seelinger, \emph{Canonical idempotents of
  multiplicity-free families of algebras},
  \href{https://doi.org/10.4171/lem/64-1/2-2}{\emph{L'Enseignement
  Math{\'{e}}matique} {\bfseries 64} (2019) 23}.

\bibitem{Artin}
E.~Artin, \emph{Zur theorie der hyperkomplexen zahlen},  in \emph{Abhandlungen
  aus dem Mathematischen Seminar der Universit{\"a}t Hamburg}, vol.~5,
  pp.~251--260, Springer, 1927.

\bibitem{Wedderburn}
J.~Wedderburn, \emph{On hypercomplex numbers}, {\emph{Proceedings of the London
  Mathematical Society} {\bfseries 2} (1908) 77}.

\bibitem{tHooft}
G.~'t~Hooft, \emph{{A Planar Diagram Theory for Strong Interactions}},
  \href{https://doi.org/10.1016/0550-3213(74)90154-0}{\emph{Nucl. Phys. B}
  {\bfseries 72} (1974) 461}.

\bibitem{Douglas1990}
M.R.~Douglas and S.H.~Shenker, \emph{{Strings in Less Than One-Dimension}},
  \href{https://doi.org/10.1016/0550-3213(90)90522-F}{\emph{Nucl. Phys. B}
  {\bfseries 335} (1990) 635}.

\bibitem{Brezin1990}
E.~Brezin and V.A.~Kazakov, \emph{Exactly {Solvable} {Field} {Theories} of
  {Closed} {Strings}},
  \href{https://doi.org/10.1016/0370-2693(90)90818-Q}{\emph{Phys. Lett. B}
  {\bfseries 236} (1990) 144}.

\bibitem{Gross1990}
D.J.~Gross and A.A.~Migdal, \emph{Nonperturbative {Two}-{Dimensional} {Quantum}
  {Gravity}}, \href{https://doi.org/10.1103/PhysRevLett.64.127}{\emph{Phys.
  Rev. Lett.} {\bfseries 64} (1990) 127}.

\bibitem{ITZYK}
M.~Bauer and C.~Itzykson, \emph{Triangulations}, {\emph{Recherche Coopérative
  sur Programme n25} {\bfseries 44} (1993) 39}.

\bibitem{MelloKoch2010}
R.~de~Mello~Koch and S.~Ramgoolam, \emph{From {Matrix} {Models} and {Quantum}
  {Fields} to {Hurwitz} {Space} and the absolute {Galois} {Group}}, .

\bibitem{Gopak2011}
R.~Gopakumar, \emph{{What is the Simplest Gauge-String Duality?}},
  \href{https://arxiv.org/abs/1104.2386}{{\ttfamily 1104.2386}}.

\bibitem{dMKLN}
R.~de~Mello~Koch and L.~Nkumane, \emph{{Topological String Correlators from
  Matrix Models}}, \href{https://doi.org/10.1007/JHEP03(2015)004}{\emph{JHEP}
  {\bfseries 03} (2015) 004} [\href{https://arxiv.org/abs/1411.5226}{{\ttfamily
  1411.5226}}].

\bibitem{1993Gross_1}
D.J.~Gross, \emph{Two-dimensional qcd as a string theory},
  \href{https://doi.org/10.1016/0550-3213(93)90402-b}{\emph{Nuclear Physics B}
  {\bfseries 400} (1993) 161–180}.

\bibitem{GrossTaylor}
D.J.~Gross and W.~Taylor, \emph{Two-dimensional qcd is a string theory},
  \href{https://doi.org/https://doi.org/10.1016/0550-3213(93)90403-C}{\emph{Nuclear
  Physics B} {\bfseries 400} (1993) 181}.

\bibitem{Minahan1993}
J.A.~Minahan, \emph{Summing over inequivalent maps in the string theory
  interpretation of two-dimensional qcd},
  \href{https://doi.org/10.1103/physrevd.47.3430}{\emph{Physical Review D}
  {\bfseries 47} (1993) 3430–3436}.

\bibitem{SCHNITZER1993}
S.G.~Naculich, H.A.~Riggs and H.J.~Schnitzer, \emph{Two-dimensional yang-mills
  theories are string theories},
  \href{https://doi.org/10.1142/s0217732393001951}{\emph{Modern Physics Letters
  A} {\bfseries 08} (1993) 2223–2235}.

\bibitem{Gross1993a}
D.J.~Gross and W.~Taylor, \emph{{Twists and Wilson loops in the string theory
  of two-dimensional QCD}},
  \href{https://doi.org/10.1016/0550-3213(93)90042-N}{\emph{Nucl. Phys. B}
  {\bfseries 403} (1993) 395}
  [\href{https://arxiv.org/abs/hep-th/9303046}{{\ttfamily hep-th/9303046}}].

\bibitem{MP1993}
J.A.~Minahan and A.P.~Polychronakos, \emph{{Equivalence of two-dimensional QCD
  and the C = 1 matrix model}},
  \href{https://doi.org/10.1016/0370-2693(93)90504-B}{\emph{Phys. Lett. B}
  {\bfseries 312} (1993) 155}
  [\href{https://arxiv.org/abs/hep-th/9303153}{{\ttfamily hep-th/9303153}}].

\bibitem{Horava1996}
P.~Hořava, \emph{Topological rigid string theory and two-dimensional qcd},
  \href{https://doi.org/10.1016/0550-3213(96)00036-3}{\emph{Nuclear Physics B}
  {\bfseries 463} (1996) 238–286}.

\bibitem{Cordes1997}
S.~Cordes, G.W.~Moore and S.~Ramgoolam, \emph{Large {N} 2-{D} {Yang}-{Mills}
  theory and topological string theory},
  \href{https://doi.org/10.1007/s002200050102}{\emph{Commun. Math. Phys.}
  {\bfseries 185} (1997) 543}.

\bibitem{Kimura:2008gs}
Y.~Kimura and S.~Ramgoolam, \emph{{Holomorphic maps and the complete $1/N$
  expansion of $2D SU(N)$ Yang-Mills}},
  \href{https://doi.org/10.1088/1126-6708/2008/06/015}{\emph{JHEP} {\bfseries
  06} (2008) 015} [\href{https://arxiv.org/abs/0802.3662}{{\ttfamily
  0802.3662}}].

\bibitem{Malda}
J.M.~Maldacena, \emph{The {Large} {N} limit of superconformal field theories
  and supergravity}, \href{https://doi.org/10.1023/A:1026654312961}{\emph{Adv.
  Theor. Math. Phys.} {\bfseries 2} (1998) 231}.

\bibitem{Witten1998}
E.~Witten, \emph{Anti-de {Sitter} space and holography},
  \href{https://doi.org/10.4310/ATMP.1998.v2.n2.a2}{\emph{Adv. Theor. Math.
  Phys.} {\bfseries 2} (1998) 253}.

\bibitem{Gubser1998}
S.S.~Gubser, I.R.~Klebanov and A.M.~Polyakov, \emph{Gauge theory correlators
  from noncritical string theory},
  \href{https://doi.org/10.1016/S0370-2693(98)00377-3}{\emph{Phys. Lett. B}
  {\bfseries 428} (1998) 105}.

\bibitem{Douglas1993YM}
M.R.~Douglas, \emph{Conformal field theory techniques in large n yang-mills
  theory},  1993.
\newblock 10.48550/ARXIV.HEP-TH/9311130.

\bibitem{Murthy2022}
S.~Murthy, \emph{Unitary matrix models, free fermion ensembles, and the giant
  graviton expansion},  2022.
\newblock 10.48550/ARXIV.2202.06897.

\bibitem{Berenstein2006}
D.~Berenstein, \emph{{Large N BPS states and emergent quantum gravity}},
  \href{https://doi.org/10.1088/1126-6708/2006/01/125}{\emph{JHEP} {\bfseries
  01} (2006) 125} [\href{https://arxiv.org/abs/hep-th/0507203}{{\ttfamily
  hep-th/0507203}}].

\bibitem{Harmark2014}
T.~Harmark and M.~Orselli, \emph{{Spin Matrix Theory: A quantum mechanical
  model of the AdS/CFT correspondence}},
  \href{https://doi.org/10.1007/JHEP11(2014)134}{\emph{JHEP} {\bfseries 11}
  (2014) 134} [\href{https://arxiv.org/abs/1409.4417}{{\ttfamily 1409.4417}}].

\bibitem{Baiguera2022}
S.~Baiguera, T.~Harmark and Y.~Lei, \emph{{Spin Matrix Theory in near $
  \frac{1}{8} $-BPS corners of $ \mathcal{N} $ = 4 super-Yang-Mills}},
  \href{https://doi.org/10.1007/JHEP02(2022)191}{\emph{JHEP} {\bfseries 02}
  (2022) 191} [\href{https://arxiv.org/abs/2111.10149}{{\ttfamily
  2111.10149}}].

\bibitem{Jevicki1980}
A.~Jevicki and B.~Sakita, \emph{The {Quantum} {Collective} {Field} {Method} and
  {Its} {Application} to the {Planar} {Limit}},
  \href{https://doi.org/10.1016/0550-3213(80)90046-2}{\emph{Nucl. Phys. B}
  {\bfseries 165} (1980) 511}.

\bibitem{Yaffe1982}
L.G.~Yaffe, \emph{Large \${N}\$ limits as classical mechanics},
  \href{https://doi.org/10.1103/RevModPhys.54.407}{\emph{Reviews of Modern
  Physics} {\bfseries 54} (1982) 407}.

\bibitem{Das1990}
S.R.~Das and A.~Jevicki, \emph{String {Field} {Theory} and {Physical}
  {Interpretation} of \${D}=1\$ {Strings}},
  \href{https://doi.org/10.1142/S0217732390001888}{\emph{Mod. Phys. Lett. A}
  {\bfseries 5} (1990) 1639}.

\bibitem{MelloKoch2011a}
R.~de~Mello~Koch, A.~Jevicki, K.~Jin and J.P.~Rodrigues,
  \emph{\${AdS}\_4/{CFT}\_3\$ {Construction} from {Collective} {Fields}},
  \href{https://doi.org/10.1103/PhysRevD.83.025006}{\emph{Phys. Rev. D}
  {\bfseries 83} (2011) 025006}.

\bibitem{Witten1980}
E.~Witten, \emph{{THE} 1 / {N} {EXPANSION} {IN} {ATOMIC} {AND} {PARTICLE}
  {PHYSICS}}, \href{https://doi.org/10.1007/978-1-4684-7571-5_21}{\emph{NATO
  Sci. Ser. B} {\bfseries 59} (1980) 403}.

\bibitem{Sagan2013}
B.E.~Sagan, \emph{The {Symmetric} {Group}: {Representations}, {Combinatorial}
  {Algorithms}, and {Symmetric} {Functions}}, Springer Science \& Business
  Media (Mar., 2013).

\bibitem{Hamermesh1962}
Hamermesh, \emph{Group Theory and Its Application to Physical Problem},
  Addison-Wesley.

\bibitem{Serre1977}
J.-P.~Serre, \emph{Linear {Representations} of {Finite} {Groups}}, vol.~42 of
  \emph{Graduate {Texts} in {Mathematics}}, Springer New York, New York, NY
  (1977),
  \href{https://doi.org/10.1007/978-1-4684-9458-7}{10.1007/978-1-4684-9458-7}.

\bibitem{Brauer}
R.~Brauer, \emph{On algebras which are connected with the semisimple continuous
  groups}, {\emph{Annals of Mathematics} {\bfseries 38} (1937) 857}.

\bibitem{Etingof09}
P.~Etingof, O.~Golberg, S.~Hensel, T.~Liu, A.~Schwendner, D.~Vaintrob et~al.,
  \emph{Introduction to representation theory},  2009.
\newblock 10.48550/ARXIV.0901.0827.

\bibitem{AR90DissertCh1}
A.~Ram, \emph{Dissertation, chapter 1 representation theory},  1990.

\bibitem{Macdonald1998}
I.G.~Macdonald, \emph{Symmetric {Functions} and {Hall} {Polynomials}},
  Clarendon Press (1998).

\bibitem{strang2006linear}
G.~Strang, \emph{Linear algebra and its applications.}, Belmont, CA: Thomson,
  Brooks/Cole (2006).

\bibitem{Lionni2019}
L.~Lionni and N.~Sasakura, \emph{A random matrix model with non-pairwise
  contracted indices},
  \href{https://doi.org/10.1093/ptep/ptz057}{\emph{Progress of Theoretical and
  Experimental Physics} {\bfseries 2019} (2019) 073A01}.

\bibitem{deMelloKoch:2011uq}
R.~de~Mello~Koch and S.~Ramgoolam, \emph{{Strings from Feynman Graph counting :
  without large N}},
  \href{https://doi.org/10.1103/PhysRevD.85.026007}{\emph{Phys. Rev. D}
  {\bfseries 85} (2012) 026007}
  [\href{https://arxiv.org/abs/1110.4858}{{\ttfamily 1110.4858}}].

\bibitem{MelloKoch2012}
R.~de~Mello~Koch and S.~Ramgoolam, \emph{{A double coset ansatz for
  integrability in AdS/CFT}},
  \href{https://doi.org/10.1007/JHEP06(2012)083}{\emph{JHEP} {\bfseries 06}
  (2012) 083} [\href{https://arxiv.org/abs/1204.2153}{{\ttfamily 1204.2153}}].

\bibitem{Gabriel2015a}
F.~Gabriel, \emph{{Combinatorial theory of permutation-invariant random
  matrices I: partitions, geometry and renormalization}},
  \href{https://arxiv.org/abs/1503.02792}{{\ttfamily 1503.02792}}.

\bibitem{Lassalle}
M.~Lassalle, \emph{An explicit formula for the characters of the symmetric
  group}, .

\bibitem{Eden2000}
B.U.~Eden, P.S.~Howe, E.~Sokatchev and P.C.~West, \emph{{Extremal and
  next-to-extremal n point correlators in four-dimensional SCFT}},
  \href{https://doi.org/10.1016/S0370-2693(00)01181-3}{\emph{Phys. Lett. B}
  {\bfseries 494} (2000) 141}
  [\href{https://arxiv.org/abs/hep-th/0004102}{{\ttfamily hep-th/0004102}}].

\bibitem{bowman2013partition}
C.~Bowman, M.D.~Visscher and R.~Orellana, \emph{The partition algebra and the
  kronecker coefficients},  2013.

\bibitem{gabriel2016b}
F.~Gabriel, \emph{Combinatorial theory of permutation-invariant random matrices
  ii: cumulants, freeness and levy processes},
  \href{https://arxiv.org/abs/1507.02465}{{\ttfamily 1507.02465}}.

\bibitem{Gabriel2015}
F.~Gabriel, \emph{{A combinatorial theory of random matrices III: random walks
  on $\mathfrak{S}(N)$, ramified coverings and the $\mathfrak{S}(\infty)$
  Yang-Mills measure}},  \href{https://arxiv.org/abs/1510.01046}{{\ttfamily
  1510.01046}}.

\bibitem{osti1980}
D.J.~Gross and E.~Witten, \emph{Possible third-order phase transition in the
  large-n lattice gauge theory},
  \href{https://doi.org/10.1103/PhysRevD.21.446}{\emph{Phys. Rev., D; (United
  States)} {\bfseries 21} (1980) }.

\bibitem{Wadia1980N}
S.R.~Wadia, \emph{{$N = \infty$} phase transition in a class of exactly soluble
  model lattice gauge theories}, {\emph{Physics Letters B} {\bfseries 93}
  (1980) 403}.

\bibitem{Skag1984}
B.-S.~Skagerstam, \emph{On the large {$N_c$} limit of the {$SU( N_c )$} colour
  quark-gluon partition function},
  \href{https://doi.org/10.1007/BF01576294}{\emph{Zeitschrift für Physik C
  Particles and Fields} {\bfseries 24} (1984) 97}.

\bibitem{Douglas1993}
M.R.~Douglas and V.A.~Kazakov, \emph{Large {N} phase transition in continuum
  {QCD} in two-dimensions},
  \href{https://doi.org/10.1016/0370-2693(93)90806-S}{\emph{Phys. Lett. B}
  {\bfseries 319} (1993) 219}.

\bibitem{Sundborg2000}
B.~Sundborg, \emph{The hagedorn transition, deconfinement and sym theory},
  \href{https://doi.org/10.1016/s0550-3213(00)00044-4}{\emph{Nuclear Physics B}
  {\bfseries 573} (2000) 349–363}.

\bibitem{Aharony2004}
O.~Aharony, J.~Marsano, S.~Minwalla, K.~Papadodimas and M.~Van~Raamsdonk,
  \emph{The deconfinement and hagedorn phase transitions in weakly coupled
  large n gauge theories},
  \href{https://doi.org/10.1016/j.crhy.2004.09.012}{\emph{Comptes Rendus
  Physique} {\bfseries 5} (2004) 945–954}.

\bibitem{FHY2007}
B.~Feng, A.~Hanany and Y.-H.~He, \emph{{Counting gauge invariants: The
  Plethystic program}},
  \href{https://doi.org/10.1088/1126-6708/2007/03/090}{\emph{JHEP} {\bfseries
  03} (2007) 090} [\href{https://arxiv.org/abs/hep-th/0701063}{{\ttfamily
  hep-th/0701063}}].

\bibitem{Dutta2008}
S.~Dutta and R.~Gopakumar, \emph{{Free fermions and thermal AdS/CFT}},
  \href{https://doi.org/10.1088/1126-6708/2008/03/011}{\emph{JHEP} {\bfseries
  03} (2008) 011} [\href{https://arxiv.org/abs/0711.0133}{{\ttfamily
  0711.0133}}].

\bibitem{ramgoolam2019quiver}
S.~Ramgoolam, M.C.~Wilson and A.~Zahabi, \emph{Quiver asymptotics:
  {$\mathcal{N}=1$} free chiral ring},  2019.

\bibitem{2021Ali}
T.~Kimura and A.~Zahabi, \emph{Unitary matrix models and random partitions:
  Universality and multi-criticality},
  \href{https://doi.org/10.1007/jhep07(2021)100}{\emph{Journal of High Energy
  Physics} {\bfseries 2021} (2021) }.

\bibitem{Kazakov:2000pm}
V.~Kazakov, I.K.~Kostov and D.~Kutasov, \emph{{A Matrix model for the
  two-dimensional black hole}},
  \href{https://doi.org/10.1016/S0550-3213(01)00606-X}{\emph{Nucl. Phys. B}
  {\bfseries 622} (2002) 141}
  [\href{https://arxiv.org/abs/hep-th/0101011}{{\ttfamily hep-th/0101011}}].

\bibitem{SDTD2003}
S.~Dasgupta and T.~Dasgupta, \emph{{Nonsinglet sector of c=1 matrix model and
  2-D black hole}},  \href{https://arxiv.org/abs/hep-th/0311177}{{\ttfamily
  hep-th/0311177}}.

\bibitem{Maldacena2005}
J.~Maldacena, \emph{Long strings in two dimensional string theory and
  non-singlets in the matrix model},
  \href{https://doi.org/10.1088/1126-6708/2005/09/078}{\emph{Journal of High
  Energy Physics} {\bfseries 2005} (2005) 078}
  [\href{https://arxiv.org/abs/hep-th/0503112}{{\ttfamily hep-th/0503112}}].

\bibitem{VershikOkounkov}
A.M.~Vershik and A.Y.~Okounkov, \emph{A {New} {Approach} to the
  {Representation} {Thoery} of the {Symmetric} {Groups}. 2},
  {\emph{arXiv:math/0503040} (2005) }.

\bibitem{Jucys74}
A.-A.~Jucys, \emph{Symmetric polynomials and the center of the symmetric group
  ring},
  \href{https://doi.org/https://doi.org/10.1016/0034-4877(74)90019-6}{\emph{Reports
  on Mathematical Physics} {\bfseries 5} (1974) 107}.

\bibitem{BenGeloun2014}
J.~Ben~Geloun and S.~Ramgoolam, \emph{{Counting tensor model observables and
  branched covers of the 2-sphere}},
  \href{https://doi.org/10.4171/aihpd/4}{\emph{Ann. Inst. H. Poincare D Comb.
  Phys. Interact.} {\bfseries 1} (2014) 77}
  [\href{https://arxiv.org/abs/1307.6490}{{\ttfamily 1307.6490}}].

\bibitem{Polya}
G.~P{\'o}lya, \emph{{Kombinatorische Anzahlbestimmungen für Gruppen, Graphen
  und chemische Verbindungen}},
  \href{https://doi.org/10.1007/BF02546665}{\emph{Acta Mathematica} {\bfseries
  68} (1937) 145 }.

\bibitem{Constantine}
G.M.~Constantine, \emph{Combinatorial theory and statistical design}, vol.~205,
  Wiley New York (1987).

\bibitem{sagemath}
T.S.~Developers, \emph{{S}ageMath, the {S}age {M}athematics {S}oftware {S}ystem
  ({V}ersion 9.7)}, 2022.

\end{thebibliography}\endgroup

	% \printindex 		% to make the index
\end{document}